\documentclass[11pt]{article}

\usepackage[utf8]{inputenc}
\usepackage{amsfonts,amsmath,amsthm,amssymb}
\usepackage{graphicx}
\usepackage[caption=false]{subfig}
\usepackage{algorithm,algorithmic, tikz}
\usepackage{authblk, bm}
\usepackage[a4paper, total={6in, 9.5in}]{geometry}

\usepackage{multirow}
\usepackage{longtable}
\usepackage{tabularx}
\usepackage{array}
\usepackage{ragged2e}
\newcolumntype{L}[1]{>{\RaggedRight\arraybackslash}p{#1}}
\newcolumntype{Y}{>{\RaggedRight\arraybackslash}X}

\usepackage[authoryear]{natbib}
\usepackage[colorlinks=true, linkcolor=black, urlcolor=blue,citecolor=blue]{hyperref}

\usepackage{rotating,float,lscape}   

\usepackage{enumitem}		

\usepackage{setspace}    	

\usepackage{xr}   			

\theoremstyle{plain}

\newtheorem{theorem}{Theorem}[section]
\newtheorem{lemma}{Lemma}[theorem]

\theoremstyle{definition}

\theoremstyle{remark}
\newtheorem{remark}{\textbf{Remark}}[section]

\DeclareMathOperator*{\argmin}{arg\,min}

\begin{document}

\title{Robust and Scalable Sure Screening of Fixed effects\\ in Ultrahigh-dimensional Linear Mixed Models}

\author[1]{\Large{Abhik Ghosh}}
\author[2]{\Large{Magne Thoresen}}
\affil[1]{Indian Statistical Institute, Kolkata, India}
\affil[2]{University of Oslo, Oslo, Norway}

\maketitle

\begin{abstract}
In modern applications of linear mixed models, 
the number of candidate fixed-effects covariates can grow exponentially with the sample size, 
while dependence induced by random effects and possible data contamination pose substantial challenges 
for existing variable screening methods.
We propose a robust and computationally efficient sure screening procedure for identifying relevant fixed-effects covariates 
in ultrahigh-dimensional linear mixed models with known random effects. 
The proposed method leverages a proxy-based transformation to decouple dependence induced by random effects, 
enabling screening via marginal analysis in a transformed regression model.
Robustness is achieved by constructing marginal utilities based on minimum density power divergence, 
yielding stability under data contamination and model misspecification without sacrificing scalability.
The resulting procedure, termed DPD-SISP, is shown to retain all relevant covariates (sure screening property) 
with exponentially high probability under general conditions, 
allowing for non-Gaussian errors and nonpolynomial growth of dimensionality. 
In addition, DPD-SISP exhibits strong robustness properties supported by influence function and breakdown point analyses. 
The framework is further extended to incorporate prior information through conditional screening, 
mitigate correlation-induced masking via iterative refinement, and enable robust post-screening estimation of fixed effects.
Extensive simulation studies demonstrate competitive performance of DPD-SISP under ideal settings 
and substantial gains in stability under data contamination.
Its practical utility is illustrated through an application to high-dimensional longitudinal data from the ADNI2 study.
The proposed framework thus provides a unified, robust, and scalable approach for variable screening 
in ultrahigh-dimensional linear mixed models.
\end{abstract}



\newpage
\section{Introduction}\label{SEC:intro}

Linear mixed models (LMMs) constitute a flexible and widely used statistical framework for analyzing data with 
correlated, clustered, and hierarchical structures. By incorporating both fixed effects (FEs) and random effects (REs), 
LMMs simultaneously capture population-level relationship and subject- or cluster-specific variability 
within a coherent likelihood-based formulation. 
They arise naturally in clustered designs, where individuals are nested within higher-level units (e.g., hospitals, 
schools, or geographical regions), and longitudinal designs involving repeated measurements over time \citep{laird1982random}.
Owing to their versatility and interpretability, LMMs have become essential across diverse scientific disciplines,
including biomedical research, statistical genetics, social and behavioral sciences, economics, agriculture and ecology, 
environmental science, and engineering.
Comprehensive treatments of their theory and applications can be found in standard monographs, e.g.,  
\cite{pinheiro2000mixed}, \cite{demidenko2013mixed}, and \cite{west2022linear}.

Recent technological advances have led to the emergence of datasets with extremely high-dimensional covariate spaces, 
often containing a number of candidate fixed effects (FEs) that far exceeds the available sample size.
At the same time, these data often exhibit complex dependence structures necessitating the inclusion of REs.
For example, genome-wide association studies (GWAS) routinely employ LMMs to account for population structure 
and genetic relatedness while screening millions of genetic markers. 
Similarly, in 
longitudinal biomedical studies, high-dimensional time-varying biomarkers are collected repeatedly within subjects. 
In such settings, identifying the truly relevant FE covariates is crucial for interpretability, accurate prediction, 
and scientific discovery. 

A substantial body of work has extended high-dimensional variable selection techniques to mixed-model settings. 
Penalized likelihood approaches, incorporating both convex penalties (e.g., $\ell_1$) and nonconvex penalties (such as SCAD), 
have been proposed  for selecting FEs under mixed model dependence \citep{schelldorfer2011estimation, fan2012variable, 
rohart2013fixed, li2018doubly, ghosh2018non, yi2022variational, oliveira2023use, gorstein2025highdimmixedmodels}.
Bayesian shrinkage methods based on spike-and-slab, global-local and non-local priors have also been developed for LMMs 
\citep{scheipl2011spikeslabgam,yang2020/etc/bayesian, yang2020bayesian, li2023study, williams2023bgwas,zgodic2025sparse}. 
However, such direct penalized estimation is computationally intensive due to the presence of REs 
and become impractical in ultrahigh-dimensional regimes, where the number of covariates may grow exponentially with the sample size.
In high-dimensional linear regression without REs, this scalability challenge is effectively addressed by 
the sure independence screening (SIS) framework \citep{fan2008sure}, 
which reduces dimensionality by ranking covariates based on marginal measures prior to joint modeling.
Owing to its simplicity and computational efficiency, SIS has since been extended to various statistical models,
including longitudinal and mixed-effects settings. Existing screening procedures for LMMs typically rely on marginal correlations, 
partial correlations, or score-type statistics derived from least-squares (LS) or maximum-likelihood (ML) principles
\citep{chu2016feature,zhang2019variable,lai2020feature,alabiso2023high,bratsberg2024conditional}. 
While these approaches provide scalability and theoretical guarantees under standard model assumptions 
(e.g., Gaussian errors and REs), they inherit the well-known lack of robustness of LS and ML estimators.
As a result, they are highly sensitive to data contamination and model misspecification.

Such data irregularities are quite frequent in modern large-scale clustered and longitudinal studies.  
Outliers may arise at multiple levels, including measurement errors within subjects, 
transient shocks at specific time points, recording errors, or anomalous clusters. 
Under these departures from idealized assumptions, LS- and ML-based screening statistics can be severely distorted, 
leading to unstable variable rankings and potential failure of the sure screening property.
Although robust screening has been extensively studied in high-dimensional regression without REs, 
both in parametric and non-parametric, model free frameworks
\citep[see, e.g.,][among many others]{li2012robust,li2012feature, mu2014some,wang2017robust,he2019robust,ghosh2021robust, li2021robust, ghosh2023robust,yu2023general,li2024feature,roy2024exact, yan2025variable,	bai2026model,guo2026robust}  
relatively little attention has been paid to robust screening in mixed-model contexts.
Existing contributions are limited to specific subclasses of longitudinal models \citep{jiang2025robust,chen2025model}, 
and a general theoretical framework ensuring the sure screening property for robust procedures in ultrahigh-dimensional LMMs 
remains largely undeveloped. Moreover, current robust estimation methods for LMMs are primarily designed for low-dimensional settings 
\citep{garcia2021robust,fan2014robust,yang2024robust} and do not scale to modern high-dimensional applications.

A natural approach for achieving robustness is to construct marginal screening statistics based on robust low-dimensional LMM estimators
that explicitly account for the correlation induced by REs.
A rich literature exists on such estimators, including bounded-influence score equations \citep{richardson1995robust}, 
heavy-tailed likelihood formulations \citep{pinheiro2001efficient}, Huber-type M-estimators \citep{wang2005robust,koller2016robustlmm},
high-breakdown S-estimators \citep{copt2006high}, composite $\tau$-estimators \citep{agostinelli2016composite}, 
trimming-based methods \citep{zheng2021trimmed}, minimum density power divergence (DPD) estimators \citep{saraceno2024robust},
and minimum hierarchical $\gamma$-divergence (HGD) estimators \citep{sugasawa2025robust}; 
see \cite{jiang2019robust} for a comprehensive summary. 
However, directly incorporating these estimators into screening procedures poses both theoretical and computational challenges. 
Theoretically, the dependence structure induced by REs complicates concentration analysis and 
the characterization of marginal signal strength, making it difficult to establish the sure screening property. 
From a computational perspective, these estimators typically require iterative procedures for jointly estimating FEs 
and variance components, rendering them infeasible for screening problems involving millions of covariates. 
This computational burden is illustrated in Figure~\ref{FIG:runtime_comp}, 
which presents runtimes of several robust estimators in a simple LMM with two FEs and two REs across 1000 simulated datasets. 
These results further indicate that most robust estimators incur substantially higher computational costs than likelihood-based methods, 
which becomes even more significant as the number of random effects increases.
Even the most recent minimum HGD estimator (MHGDE) is, on average, 2.5 to 40 times slower than the REML,  
depending on the correlations between  FE and RE covariates!

\begin{figure}[!h]
\centering
	\subfloat[REs are two independent standard normal covariates generated independently of the FEs]{
	\includegraphics[width=0.48\textwidth]{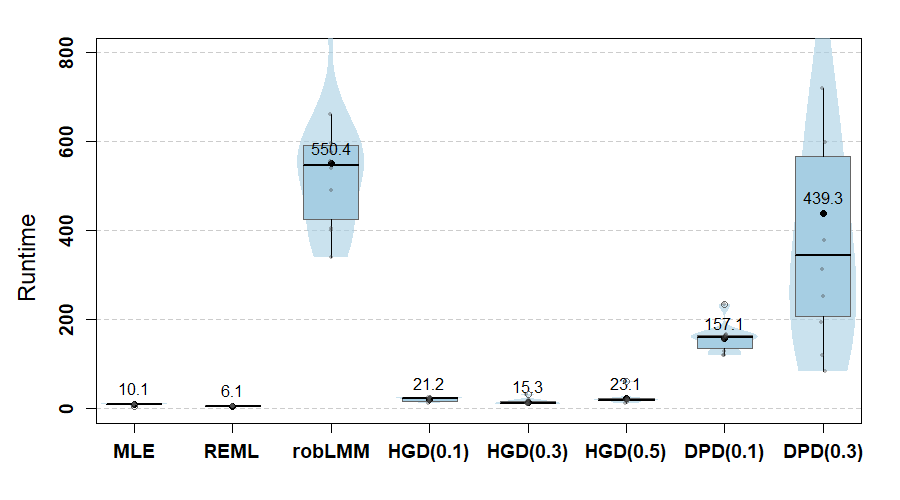}
	\label{FIG:runtime2}}
~~
\subfloat[REs are exactly the same as the FEs ~~~~~~~~~~~~~\\(an intercept and the same standard normal covariate)]{
	\includegraphics[width=0.48\textwidth]{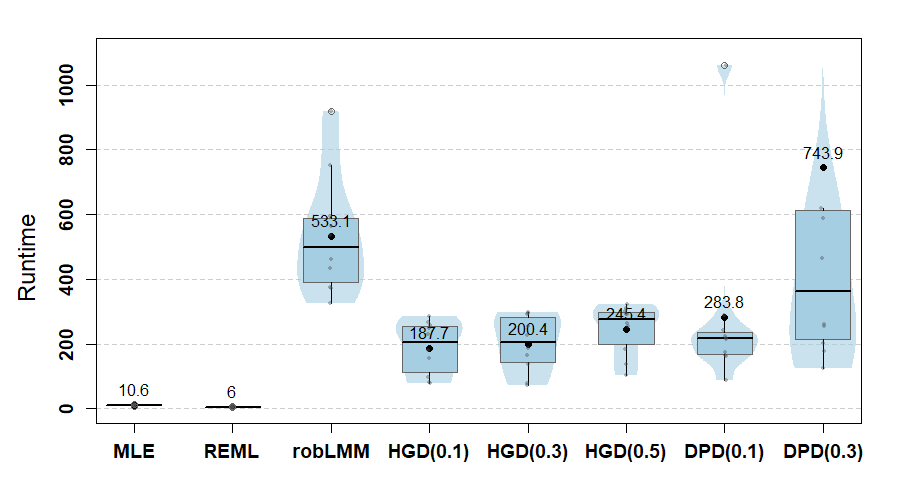}
	\label{FIG:runtime1}}
\\
	\subfloat[REs are two standard normal covariates, generated jointly with the FE covariate ensuring 
with all pairwise correlations to be 0.9]{
	\includegraphics[width=0.48\textwidth]{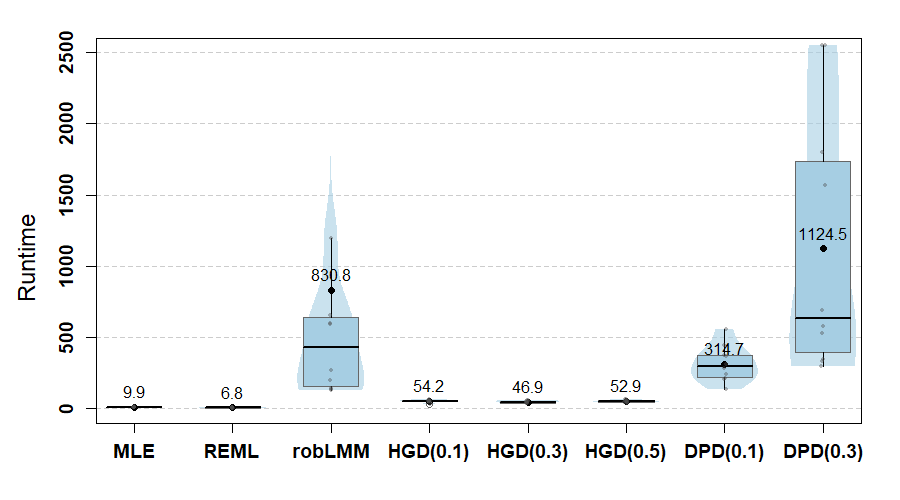}
	\label{FIG:runtime3}}
\caption{\small Comparison of computational times for different robust estimators across 1000 replications
under a simple low dimensional LMM with two fixed-effects (one intercept and one standard normal variable),
$q=2$ random-effects (in various specifications), true fixed-effect coefficient $(1,1)$,
random-effects covariance $0.7 \mathbb{I}_q + 0.3 \mathbb{J}_q$, and noise variance 1. 
Here robLMM refers the robust estimator of \cite{koller2016robustlmm}, 
HGD($\gamma$) and DPD($\alpha$) refer to the minimum HGD and minimum DPD estimators of \cite{sugasawa2025robust} 
and \cite{saraceno2024robust}, respectively, with tuning parameters $\gamma$ and $\alpha$. 
We used the \texttt{R}-function \texttt{rlmer} to compute robLMM, 
while the HGD and DPD based estimators are computed using the \texttt{R} codes supplied by \cite{sugasawa2025robust},
keeping all algorithmic hyperparameters to their default values. 
Sample means are also reported within all boxes.}
\label{FIG:runtime_comp}
\end{figure}

In this paper, we propose a computationally efficient and theoretically grounded robust screening procedure for ultrahigh-dimensional LMMs.
The key idea is to approximate the RE-induced covariance structure using a \textit{working covariance} (or \textit{proxy}) matrix,
which transforms the mixed model into a linear regression model (LRM) on suitably adjusted data, 
thereby decoupling the dependence structure from the screening step.
We then construct marginal screening statistics based on the minimum DPD estimators (MDPDEs) applied to the transformed data,;
they provide robustness against contamination while maintaining desired computational and statistical efficiency 
\citep{basu1998robust,basu2026statistical}.
This particular choice is motivated by recent evidence \citep{ghosh2021robust, ghosh2023robust}
demonstrating strong performance of the MDPDE in ultrahigh-dimensional screening under contaminated regression models.

The resulting procedure, termed DPD-SISP, achieves computational complexity comparable to classical likelihood-based screening methods
while offering substantially improved robustness under data contamination, 
thereby reconciling scalability and robustness within a unified framework. 
Under mild verifiable conditions on the proxy covariance matrix and signal strength, 
we establish that DPD-SISP consistently retains all truly active FEs while controlling the size of the selected model
with probability tending to one at an exponential rate.   
These results are first developed in a general distribution-free framework and subsequently specialized to Gaussian LMMs. 
We further demonstrate the robustness of the proposed method through bounded influence functions 
and high asymptotic breakdown properties of the underlying marginal estimators.

We also extend the proposed framework in several practically important directions. 
First, we develop a conditional screening version that accommodates prespecified important covariates. 
Second, we discuss an iterative screening strategy to mitigate masking effects due to strong correlations among covariates. 
Third, we consider a second-stage robust penalized estimation procedure based on penalized DPD criteria. 

Our theoretical derivations are complemented by extensive simulation studies comparing our proposal with benchmark screening procedures 
based on the classical ML/REML and the recent robust MHGDEs with relatively lower computational cost. 
The results demonstrate that the DPD-SISP performs comparably to classical screening approaches under clean data, 
while exhibiting substantially greater stability and reliability under contamination. 
Finally, we illustrate the practical utility of the method through the analysis of a real high-dimensional longitudinal dataset, 
demonstrating its effectiveness in identifying meaningful FE covariates in the presence of complex dependence structures and data irregularities.

\section{Model Framework and Notation}
\label{SEC:notation}

We consider data arising from $m$ independent clusters (e.g., subjects, centers, or groups), indexed by $i = 1, \dots, m$,
where the $i$-th cluster contains $n_i$ observations. The total sample size is then $n = \sum_{i=1}^m n_i$.
For each cluster $i\in\{1,\ldots, m\}$, we denote the response vector, and the fixed- and random-effects design matrices by 
$\mathbf{y}_i=(y_{i1}, \ldots, y_{in_i})^\top \in \mathbb{R}^{n_i}$, 
$\mathbf{X}_i = [\bm{x}_{i1} \cdots \bm{x}_{ip}] \in \mathbb{R}^{n_i \times p}$,
and $\mathbf{Z}_i = [\bm{z}_{i1} \cdots \bm{z}_{iq}] \in \mathbb{R}^{n_i \times q}$, respectively. 
Here, $p$ and $q$ represent the numbers of available fixed- and random-effects covariates, respectively. 
Then, the linear mixed model (LMM) describing the relationship between the response and covariates is given by \citep{pinheiro2000mixed}
\begin{equation}
	\mathbf{y}_i = \mathbf{X}_i \boldsymbol{\beta} + \mathbf{Z}_i \mathbf{b}_i + \boldsymbol{\varepsilon}_i,
	~~~~ i=1, \ldots, m, 
	\label{EQ:lmm}
\end{equation}
where $\boldsymbol{\beta}=(\beta_1, \ldots, \beta_p)^\top \in \mathbb{R}^{p}$ is the vector of FE coefficients,
$\mathbf{b}_i = (b_{i1}, \ldots, b_{iq})^\top \in \mathbb{R}^{q}$ is the cluster-specific REs, and 
$\boldsymbol{\varepsilon}_i=(\varepsilon_{i1}, \ldots, \varepsilon_{in_i})^\top \in \mathbb{R}^{n_i}$ is the within-cluster error vector.
We assume that $E[\mathbf{b}_i]= \bm{0}_q$, Cov$(\mathbf{b}_i) = \bm{\Psi}$, 
$E[\boldsymbol{\varepsilon}_i] = \bm{0}_{n_i}$, and Cov$(\boldsymbol{\varepsilon}_i) = \sigma^2 \mathbb{I}_{n_i}$, 
with $\bm{b}_i$ being independent of $\bm{\varepsilon}_i$  for each $i=1, \ldots, m$.
Note that clusters are also assumed mutually independent.

Since our primary objective here is to screen important FE variables, without loss of generality, 
we assume that $\mathbf{X}_i$ does not contain the intercept term.
However, the RE design matrix $\mathbf{Z}_i$ may include an intercept and may also share columns with $\mathbf{X}_i$, 
allowing for both correlated and exactly the same fixed and random effects covariates.

Under the additional assumption that $\mathbf{b}_i$ and $\boldsymbol{\varepsilon}_i$ follow Gaussian distributions, 
we have $\bm{b}_i\sim N_q(\bm{0},\bm{\Psi})$ and $\bm{\varepsilon}_i\sim N_{n_i}(\bm{0},\sigma^2 \mathbb{I}_{n_i})$,
independently across clusters $i=1, \ldots, m$.
Therefore, conditional on $\mathbf{X}_i$ and $\bm{Z}_i$, it follows that
\begin{eqnarray}\label{EQ:lmm-dist}
\bm{y}_i \sim \mathcal{N}_{n_i}(\mathbf{X}_i\bm{\beta}, \sigma^2 \bm{\Sigma}_i), ~~~~ \mbox{with}~~  
\boldsymbol{\Sigma}_i = \frac{1}{\sigma^2}\mathbf{Z}_i \bm{\Psi} \mathbf{Z}_i^\top +  \mathbb{I}_{n_i},
~~~~i=1, \ldots, m,
\end{eqnarray}
and the vectors $\bm{y}_i$'s are independent across clusters.  
We assume that $\bm{\Psi}$ is positive definite and $\sigma^2>0$, ensuring that $\bm{\Sigma}_i$ is positive definite for all $i$.

We focus on the ultrahigh-dimensional regime in which $p\gg n$, specifically  $\log p = O(n^\tau)$ for some $\tau\in(0,1)$. 
The true FE coefficient vector $\bm{\beta}_0=(\beta_{01}, \ldots, \beta_{0p})^\top\in\mathbb{R}^p$ is assumed to be sparse,
with support (active set), 
$\mathcal{S}_0 = \{j : \beta_{0j}\neq 0, ~j=1, \ldots, p\}$ having cardinality $|\mathcal{S}_0|=s \ll n$. 
Without loss of generality, we assume $\mathcal{S}_0=\{1, \ldots, s\}$, so that $\bm{\beta}_0 = (\bm{\beta}_{S0}^\top, \bm{0}^\top)^\top$.
In contrast, no sparsity structure is imposed on the REs. 
The number of random-effects covariates satisfies $q\ll n$, 
and all RE components are assumed to be known as the selection of REs is not considered in this work.

We conclude by introducing notation used throughout the paper.
For a matrix $\bm{A}\in\mathbb{R}^{k\times l}$ (which is a vector $\bm{a}$ if $l=1$), 
we denote its $\ell_p$-norm by $\|\bm{A}\|_p$, its operator norm by $\|\bm{A}\|_{op} = \|\bm{A}\|_2$,
and its minimum and maximum eigenvalues by $\lambda_{\min}(\bm{A})$ and $\lambda_{\max}(\bm{A})$, respectively. 
Further, we put $\mathcal{I}=\{1, \ldots, p\}$, and for any subset $\mathcal{C}\subset \mathcal{I}$ with cardinality $c_0=|\mathcal{C}|$, 
$\bm{A}_{\mathcal{C}}$ denote the $k\times c_0$ submatrix of $\bm{A}$ formed by selecting columns indexed by $\mathcal{C}$.
When $\bm{A}=\bm{A}_i$ (has a subscript), 	we write $\bm{A}_{\mathcal{C}} = \bm{A}_{i,\mathcal{C}}$, 
and denote its $j$-th row by $\bm{a}_{i,\mathcal{C},j}$.

\section{Proposed Robust Screening of Fixed-Effects covariates }
\label{SEC:method}

\subsection{Whitening Transformation and Proxy matrix}

We begin by considering the ideal setting of Gaussian LMM in which 
the error variance $\sigma^2$ and RE covariance matrix $\bm{\Psi}$ are known. 
In this case, the dependence induced by the REs can be removed via a suitable whitening transformation based on (\ref{EQ:lmm-dist}).
Specifically, we adopt the zero-phase component analysis (ZCA) whitening transformation \citep{bell1997independent}, 
which premultiplies the data of the $i$-th cluster by $\bm{\Sigma}_i^{-1/2}$ for each $i=1, \ldots, m$.
As shown in \cite{Kessy02102018}, this simple ZCA whitening transformation uniquely maximizes the average cross-covariance
between the original and transformed (whitened) variables among all whitening transforms.
Let us denote the resulting \textit{whitened variables} by $\bm{y}_i^* = \bm{\Sigma}_i^{-1/2}\bm{y}_i$ 
and $\mathbf{X}_i^* = \bm{\Sigma}_i^{-1/2}\mathbf{X}_i$ for each $i$.
It then follows from (\ref{EQ:lmm-dist}) that  
$\bm{y}_i^* \sim \mathcal{N}_{n_i}(\mathbf{X}_i^*\bm{\beta}, \sigma^2 \mathbb{I}_{n_i})$,
so that the original LMM (\ref{EQ:lmm}) is transformed into an LRM with independent and identically distributed (IID) Gaussian errors:
\begin{eqnarray}
	\bm{y}_i^* = \mathbf{X}_i^* \bm{\beta} + \bm{\varepsilon}_i^*, ~~~~
	\bm{\varepsilon}_i^* \sim \mathcal{N}_{n_i}(\bm{0}_{n_i}, \sigma^2 \mathbb{I}_{n_i}), 
	~~~~ i=1, \ldots, m. 
\label{EQ:lm}
\end{eqnarray}
This transformation $(\bm{y}_i, \mathbf{X}_i) \mapsto (\bm{y}_i^*, \mathbf{X}_i^*)$ is linear and bijective 
due to the positive definiteness (and hence invertibility) of each $\bm{\Sigma}_i$.
As a result, it defines a one-to-one reparameterization of the data that preserves the Gaussian likelihood.
In particular, the parameter $\bm{\beta}$ remains unchanged,  its support (i.e., the active set) is invariant, 
and hypothesis testing for FEs is unaffected. In other words, this transformation simply re-express the model 
in a different coordinate system with spherical errors, without altering the estimable structure.
Consequently, all inference on the FEs, including selection procedures, in the original LMM 
can equivalently be carried out using the transformed LRM.

In practice, however, the covariance components $\boldsymbol{\Psi}$ and $\sigma^2$ are unknown, 
and so the exact whitening transformation is unavailable. 
To address this issue, we adopt a proxy-based approach that approximates the covariance structure suitably.
This idea has previously been used for constructing valid statistical inference under high- and ultrahigh-dimensional LMMs; 
see, e.g., \cite{fan2012variable,bradic2020test,li2022inference,rugamer2022selective}, among others.
Specifically, we introduce a working covariance (\textit{proxy}) matrix $\bm{P}$ to approximate $\sigma^{-2}\bm{\Psi}$,
and define the plug-in covariance estimators 
$\widehat{\boldsymbol{\Sigma}}_i = \mathbf{Z}_i \bm{P} \mathbf{Z}_i^\top + \mathbb{I}_{n_i}$ for each $i=1, \ldots, m$.
Using this proxy, we construct the \textit{feasible transformation}
\begin{eqnarray}
\widetilde{\mathbf{y}}_i = (\widetilde{y}_{i1}, \ldots, \widetilde{y}_{in_i})^\top = \widehat{\boldsymbol{\Sigma}}_i^{-1/2} \mathbf{y}_i,
\qquad
\widetilde{\mathbf{X}}_i = [\widetilde{\bm{x}}_{i1}, \cdots, \widetilde{\bm{x}}_{ip}] 
= \widehat{\boldsymbol{\Sigma}}_i^{-1/2} \mathbf{X}_i, ~~~ i=1, \ldots, m,
\label{EQ:proxy}
\end{eqnarray}
and the resulting \textit{feasible} LRM then takes the form
\begin{eqnarray}
	\widetilde{\bm{y}}_i = \widetilde{\mathbf{X}}_i \bm{\beta} + \widetilde{\bm{\varepsilon}}_i, ~~~~
	\widetilde{\bm{\varepsilon}}_i \sim \mathcal{N}_{n_i}(\bm{0}_{n_i}, \sigma^2 \mathbf{V}_i), 
	\label{EQ:lm-feasible}
\end{eqnarray}
with $\bm{V}_i = \widehat{\boldsymbol{\Sigma}}_i^{-1/2}\bm{\Sigma}_i \widehat{\boldsymbol{\Sigma}}_i^{-1/2}$ 
for each $i=1, \ldots, m$.
The errors in the feasible LRM (\ref{EQ:lm-feasible}) are exactly spherical if $\widehat{\boldsymbol{\Sigma}}_i = \bm{\Sigma}_i$ 
for all $i\geq1$; otherwise, the transformation induces a covariance distortion captured by $\bm{V}_i$s.
If the proxy $\bm{P}$ is chosen such that $\bm{V}_i$ is uniformly close to $\mathbb{I}_{n_i}$  in operator norm with high probability, 
the feasible model (\ref{EQ:lm-feasible}) behaves approximately as an LRM with spherical errors, 
and the impact of covariance misspecification becomes asymptotically negligible. 
This proxy approximation thus underlies the validity of applying high-dimensional regression techniques in the transformed space. 
A detailed discussion on the choice of the proxy matrix is deferred to Section  \ref{SEC:proxy}.


\subsection{The DPD-SISP Algorithm }
\label{SEC:DPD-SIS}

To develop a robust screening procedure for identifying important FEs in the ultrahigh-dimensional LMM (\ref{EQ:lmm}),
we propose to rank covariates based on the marginal MDPDEs computed from the feasible LRM (\ref{EQ:lm-feasible}),
after suitably approximating $\bm{V}_i$'s to the identity matrices for $i=1, \ldots, m$.
The DPD framework, introduced originally by \cite{basu1998robust} for robust parametric estimation under IID data, 
provides a continuum of estimators indexed by a tuning parameter $\alpha \geq 0$, 
where $\alpha=0$ corresponds to ML estimation and $\alpha>0$ yields robust alternatives. 
These estimators have been successfully applied in several complex statistical models including linear and generalized linear models 
\citep{basu2026statistical}. Building on this framework, specifically the work of \cite{ghosh2013robust},   
a DPD-based robust SIS procedure for ultrahigh-dimensional LRMs  with spherical errors is developed by \cite{ghosh2021robust}, 
who demonstrated its improved robustness-efficiency trade-offs compared to other robust SIS approaches. 
Its sure screening property was subsequently proved under generalized linear models by \cite{ghosh2023robust}.
Motivated by its computational simplicity and strong robustness properties, we adopt this marginal MDPDE-based SIS approach 
applied to the feasible model (\ref{EQ:lm-feasible}) to construct our robust screening procedure described below;
its sure screening properties are rigorously established in Section \ref{SEC:theory_strong}.

For each FE covariate $j\in\mathcal{I}$, we consider a marginal feasible LRM based on the transformed data, 
obtained from (\ref{EQ:lm-feasible}) with the spherical error approximation, as given by 
\begin{eqnarray}
\widetilde{\bm{y}}_i = \beta_{j0} + \widetilde{\mathbf{x}}_{ij} \beta_{j1} + \bm{\varepsilon}_{ij}, ~~~~
\bm{\varepsilon}_{ij} \sim \mathcal{N}_{n_i}(\bm{0}_{n_i}, \sigma_{j}^2 \mathbb{I}_{n_i}), ~~i=1, \ldots, m. 
\label{EQ:lm-feas-marginal}
\end{eqnarray}   
Although the full model excludes the FE intercept, here an intercept term is included to account for marginal shifts 
induced by projection onto a single covariate. Following \cite{ghosh2013robust} and \cite{ghosh2021robust}, 
we define the marginal MDPDE $\widehat{\bm\theta}_{j, \alpha}$, with tuning parameter $\alpha\geq0$,  
of the parameters $\bm{\theta}_j = (\beta_{j0}, \beta_{j1}, \sigma_{j}^2)^\top$ under (\ref{EQ:lm-feas-marginal}) 
as a minimizer of $\mathcal{L}_{n, \alpha}(\bm{\theta}_j)$ with respect to $\bm{\theta}_j$,  
where the empirical DPD loss function is given by  
\begin{eqnarray}
\mathcal{L}_{n, \alpha}(\bm{\theta}) = \frac{1}{n}\sum_{i=1}^{m}\sum_{k=1}^{n_i} 
V_\alpha(\widetilde{y}_{ik}, \beta_0+ \widetilde{x}_{ij,k}\beta_1, \sigma^2),
~~\bm{\theta} = (\beta_0, \beta_1, \sigma^2)^\top,
\label{EQ:dpd-loss}
\end{eqnarray}
with  $\widetilde{x}_{ij,k}$ being the $k$-th element of $\widetilde{\mathbf{x}}_{ij}$ for each $k, i,j$, and 
\begin{eqnarray}
V_\alpha(y, \mu, \sigma^2) = \left\{\begin{array}{ll}
\frac{1}{\sigma^{\alpha}(2\pi)^{\alpha/2}}
\left(\frac{1}{\sqrt{1+\alpha}} - \frac{1+\alpha}{\alpha} e^{-\frac{\alpha(y - \mu)^2}{2\sigma^2}}\right) 
+ \frac{1}{\alpha},  & \mbox{ if } \alpha>0,  
\\[10pt]
\log(\sigma\sqrt{2\pi}) + \frac{1}{2\sigma^2}(y - \mu)^2 & \mbox{ if } \alpha=0.
\end{array}
\right..
\label{EQ:dpd-obj-norm}
\end{eqnarray}
For $\alpha=0$, this loss function reduces to the negative log-likelihood under Gaussian errors,
whereas it yields a robust loss that downweights outliers at $\alpha>0$.
Throughout this paper we restrict to $\alpha>0$ to ensure robustness.

Let $\widehat{\beta}_{j1, \alpha}$ denote the second component of $\widehat{\bm\theta}_{j, \alpha}$, 
which is the marginal MDPDE of $\beta_{j1}$ with tuning parameter $\alpha\geq 0$.
The proposed screening procedure then ranks all covariates according to $|\widehat{\beta}_{j1, \alpha}|$.
Given a target model size $d$, we define the selected set $\widehat{\mathcal{S}}_\alpha(d)$ 
as the indices corresponding to the top $d$ ranked FE covariates, i.e.,
\begin{eqnarray}
\widehat{\mathcal{S}}_\alpha(d) 
= \left\{ j : |\widehat{\beta}_{j1, \alpha}| \text{ ranks among the top } d \mbox{ for } j\in \mathcal{I} \right\}.
\label{EQ:DPD-SIS0}
\end{eqnarray}
Following classical SIS theory, a common suggestion for $d$ is $d=[n/\log n]$. 

Alternatively, for a given threshold $\gamma_n>0$, we may define the selected set as those FE covariates 
whose estimated coefficients exceed the threshold in absolute value, i.e., 
\begin{eqnarray}
	\widehat{\mathcal{S}}_\alpha(\gamma_n) 
	= \left\{ j\in \mathcal{I} : |\widehat{\beta}_{j1, \alpha}| \geq \gamma_n \right\}.
	\label{EQ:DPD-SIS1}
\end{eqnarray}
With appropriate calibration, these two definitions of the selected set are equivalent in practice. 
We refer to the resulting screening procedure as DPD-SISP, an abbreviation for `\textit{DPD based SIS with Proxy matrix}', 
with tuning parameter $\alpha$; the full procedure is presented in Algorithm \ref{ALG:DPD-SISP} below.

\begin{minipage}{.95\textwidth}
	\begin{algorithm}[H]
		\caption{\small DPD-SISP($\alpha$) for Gaussian LMMs}
		\label{ALG:DPD-SISP}
		\small
		\begin{algorithmic}[1]
			
			\STATE \textbf{Input:} Clustered data 
			$\{(\bm y_i,\mathbf X_i,\mathbf Z_i)\}_{i=1}^m$ from the LMM \eqref{EQ:lmm}; 
			tuning parameter $\alpha>0$; \\~~~~~~~~~~
			proxy matrix $\bm P$; 
			screening parameter $d$ (or threshold $\gamma_n$).
			
			\STATE Set total sample size $n=\sum_{i=1}^m n_i$.\\~~~~~~~~~
						\COMMENT{\textit{Proxy-based covariance approximation}}
			\FOR{$i=1$ to $m$} 
			\STATE Compute proxy covariance $\widehat{\boldsymbol\Sigma}_i =\mathbf Z_i \bm P \mathbf Z_i^\top+\mathbb{I}_{n_i}$.
			\STATE Obtain transformed data $(\widetilde{\bm y}_i,\widetilde{\mathbf X}_i)$ using \eqref{EQ:proxy}.
			\ENDFOR\\~~~~~~~~~			\COMMENT{\textit{Marginal DPD estimation}}
			\FOR{$j\in \mathcal{I}$} 			
			\STATE Define the DPD loss $V_\alpha(\cdot)$ as in \eqref{EQ:dpd-obj-norm}.
			\STATE Compute the marginal MDPDE
			\[
		\widehat{\bm\theta}_{j,\alpha} = \left(\widehat\beta_{j0,\alpha}, \widehat\beta_{j1,\alpha}, \widehat\sigma_{j,\alpha}^2\right)^\top
			=
			\argmin_{(\beta_0, \beta_1, \sigma^2)^\top}
			\frac{1}{n}
			\sum_{i=1}^{m}\sum_{k=1}^{n_i}
			V_\alpha\!\left(
			\widetilde y_{ik},
			\beta_0+\widetilde x_{ij,k}\beta_1,
			\sigma^2
			\right)
			\]
			\STATE Extract the marginal slope estimate $\widehat\beta_{j1,\alpha}$ from $\widehat{\bm\theta}_{j,\alpha}$.
			\ENDFOR\\~~~~~~~~~
			\COMMENT{\textit{Screening}}
			\STATE Rank $\left\{|\widehat\beta_{j1,\alpha}|: j \in \mathcal{I}\right\}$ in decreasing order.
			
			\STATE Define 
			$\widehat{\mathcal S}_\alpha(d)= \left\{ j\in \mathcal{I}: |\widehat\beta_{j1,\alpha}| \text{ is among the top } d \right\}$,
			or  $\widehat{\mathcal S}_\alpha(\gamma_n)=\left\{ j\in \mathcal{I}: |\widehat\beta_{j1,\alpha}| \ge \gamma_n \right\}$.
			
			\STATE \textbf{Output:} Screened active set $\widehat{\mathcal S}_\alpha$.
			
		\end{algorithmic}
	\end{algorithm}
\end{minipage}

\bigskip
We may note that both versions of DPD-SISP, given in (\ref{EQ:DPD-SIS0}) and (\ref{EQ:DPD-SIS1}), 
correspond	to hard thresholding rules, since either $d$ or $\gamma_n$ is pre-specified. 

\subsection{Extension beyond Normality}
\label{SEC:DPD_SISP-gen}

An important feature of the proposed DPD-SISP procedure is that it is fundamentally based on a working marginal  model, 
rather than the full distributional specification of the original LMM (\ref{EQ:lmm}). 
In particular, its construction relies on the marginal feasible LRM (\ref{EQ:lm-feas-marginal}) 
together with an assumed working distribution for its error terms. 
Consequently, the procedure does not depend critically on Gaussian assumptions for either the REs or the within-cluster errors.
Even in the Gaussian setting, the spherical error structure imposed in the marginal feasible model is itself only an approximation, 
arising from the use of a proxy covariance matrix. Thus, DPD-SISP is inherently a working-model-based approach, 
whose validity depends primarily on the adequacy of this approximation rather than strict distributional correctness. 
This observation allows a natural extension of the proposed DPD-SISP to LMMs with non-Gaussian REs and/or errors,
provided that a suitable parametric working density is specified for the marginal errors in (\ref{EQ:lm-feas-marginal}) 
to reasonably approximate the true marginal error distribution.

To formalize this extension, consider the general LMM introduced in Section \ref{SEC:notation}, 
but without imposing any specific distributional assumption on the REs or error terms. 
In this case, the original model for $\bm{y}_i$, conditional on  $\mathbf{X}_i$ and $\bm{Z}_i$, 
is given by the following modified form of \eqref{EQ:lmm-dist}: 
\begin{eqnarray}\label{EQ:lmm-dist_gen}
	\bm{y}_i = \mathbf{X}_i\bm{\beta} + \bm{v}_i, 
		~~ \mbox{with}~  E[\bm{v}_i]= \bm{0}_{n_i}, ~
	\mbox{Cov}(\bm{v}_i) = \boldsymbol{\Sigma}_i = \frac{1}{\sigma^2}\mathbf{Z}_i \bm{\Psi} \mathbf{Z}_i^\top +  \mathbb{I}_{n_i},
	~~i=1, \ldots, m.
\end{eqnarray}
Then, applying the same whitening transformation as before, we get a distribution-free generalization of \eqref{EQ:lm} as given by 
\begin{eqnarray}
	\bm{y}_i^* = \mathbf{X}_i^* \bm{\beta} + \bm{\varepsilon}_i^*, ~~~~
	E[\bm{\varepsilon}_i^*] =\bm{0}_{n_i}, ~ \mbox{Cov}(\bm{\varepsilon}_i^*) = \sigma^2 \mathbb{I}_{n_i}, 
	~~~~ i=1, \ldots, m. 
	\label{EQ:lm_gen}
\end{eqnarray}
Similarly, under the proxy-based transformation, the feasible LRM \eqref{EQ:lm-feasible} becomes
\begin{eqnarray}
	\widetilde{\bm{y}}_i = \widetilde{\mathbf{X}}_i \bm{\beta} + \widetilde{\bm{\varepsilon}}_i, ~~~~
E[\widetilde{\bm{\varepsilon}}_i ] =\bm{0}_{n_i}, ~ \mbox{Cov}(\widetilde{\bm{\varepsilon}}_i ) = \sigma^2 \mathbf{V}_i, 
		~~~~ i=1, \ldots, m. 
	\label{EQ:lm-feasible_gen}
\end{eqnarray}

Now, with an appropriate choice of the proxy matrix, the matrices $\mathbf{V}_i$ can be made close to $\mathbb{I}_{n_i}$ in operator norm, 
so that the components of the transformed errors 
$\widetilde{\bm{\varepsilon}}_i = \widetilde{\bm{y}}_i - \widetilde{\mathbf{X}}_i \bm{\beta}$ 
are approximately uncorrelated and homoscedastic, although not necessarily independent.
Stacking the transformed data across clusters, let 
$\widetilde{\bm{y}} = (\widetilde{\bm{y}}_1^\top, \ldots, \widetilde{\bm{y}}_m^\top)^\top$ denote the combined response vector and
$\widetilde{\mathbf{x}}_{j} = (\widetilde{\mathbf{x}}_{1j}, \ldots, \widetilde{\mathbf{x}}_{mj})^\top$ 
the corresponding transformed covariate for each $j\in\mathcal{I}$.
For each $j\in \mathcal{I}$, we then consider the following marginal working model as a generalization of (\ref{EQ:lm-feas-marginal}):
\begin{eqnarray}
	\widetilde{\bm{y}} = [\bm{1}_n  ~ \widetilde{\mathbf{x}}_{j}] \bm{\beta}_{j} + \bm{\varepsilon}_{j}, ~~~ ~
	\bm{\beta}_j = (\beta_{j0}, \beta_{j1})^\top, ~~ E[\bm{\varepsilon}_{j}] = \bm{0}_n, 
	~~\mbox{Cov}(\bm{\varepsilon}_{j}) = \sigma_{j}^2\mathbb{I}_n.
	\label{EQ:lm-feas-marginal_gen}
\end{eqnarray}   
The distribution of $\bm{\varepsilon}_{j}$ is generally unknown; we therefore introduce a parametric working model for its components. 
Specifically, we assume that each component of $\bm{\varepsilon}_{j}$ follows density $f_{\bm{\eta}_j}$, 
indexed by parameters $\bm{\eta}_j\in\Theta_\eta$, such that the distribution has mean zero and variance $\sigma_{j}^2$, 
for each $j\in \mathcal{I}$.  The choice of $f_{\bm{\eta}_j}$ can be guided by the empirical characteristics of the response.
For example, if the data exhibit skewness, a skew-normal density may be used as $f_{\bm{\eta}_j}$ 
with $\bm{\eta}_j = (\sigma_{j}^2, \gamma_j)^\top$, where $\gamma_j$ controls the degree of asymmetry.

Under this general formulation, the parameter vector for the $j$-th marginal model is 
$\bm{\theta}_j = (\bm{\beta}_{j}, \bm{\eta}_{j})^\top$.
Following the general theory of DPD from \cite{ghosh2013robust}, we can define the MDPDE of $\bm{\theta}_j$ as 
a minimizer of the same empirical objective function as in \eqref{EQ:dpd-loss} but with the inner Gaussian loss function 
$V_\alpha(\widetilde{y}_{ik}, \beta_0+ \widetilde{x}_{ij,k}\beta_1, \sigma^2)$ being replaced by 
$V_\alpha(\widetilde{y}_{ik}, \beta_0+ \widetilde{x}_{ij,k}\beta_1, \bm{\eta}_j)$, where  
\begin{eqnarray}
	V_\alpha(y, \mu, \bm{\eta}) = \left\{\begin{array}{ll}
		\int f_{\bm{\eta}}^{1+\alpha}(s)ds - \frac{1+\alpha}{\alpha} f_{\bm{\eta}}^\alpha(y - \mu) + \frac{1}{\alpha},  
		& \mbox{ if } \alpha>0,  
		\\[10pt]
		-\log f_{\bm{\eta}}(y - \mu) & \mbox{ if } \alpha=0.
	\end{array}
	\right..
	\label{EQ:dpd-obj-gen}
\end{eqnarray}
The DPD-SISP procedure is then defined exactly as in the Gaussian case: 
for each covariate $j$, we compute the marginal MDPDE $\widehat\beta_{j1,\alpha}$ of the slope parameter $\beta_{j1}$,
rank the covariates based on the magnitude of these estimates, and select variables using either a fixed model size or a thresholding rule. 
The only modification is that the Gaussian loss function in Algorithm \ref{ALG:DPD-SISP} (Steps 8--9) 
is replaced by the general DPD loss defined above in \eqref{EQ:dpd-obj-gen}.
This formulation demonstrates that DPD-SISP extends naturally to a broad class of non-Gaussian settings, 
without requiring explicit distributional assumptions on the underlying REs or error terms. 
Moreover, the intrinsic robustness of the DPD framework provides protection against misspecification of the working error distribution, 
making the procedure particularly well suited for complex and heterogeneous data.

\begin{remark}
When the working error density $f_{\bm{\eta}_j}$ is chosen to be Gaussian with mean zero and variance $\sigma_{j}^2$,
the parameter reduces to $\bm{\eta}_j = \sigma_{j}^2$ for each $j\geq 1$, 
and the general DPD loss  $V_\alpha(\cdot)$ in (\ref{EQ:dpd-obj-gen}) simplifies exactly to the Gaussian form 
given in (\ref{EQ:dpd-obj-norm}).
\end{remark}

\subsection{Second-Stage Selection and Estimation }
\label{SEC:second_stage}

The proposed DPD-SISP procedure reduces the dimensionality of the original  LMM from an ultrahigh-dimensional regime ($p \gg n$) 
to a moderate scale by retaining a subset $\widehat{\mathcal{S}}_\alpha$ of size $d = o(n)$. 
Although this screening step satisfies the sure screening property, 
the selected set may still contain irrelevant variables and does not, by itself, provide final estimates of the FE coefficients.
Therefore, a second-stage refinement is necessary to achieve both consistent variable selection and efficient estimation.

In traditional SIS framework, this refinement is typically performed via penalized likelihood methods 
applied the the reduced model \citep{fan2008sure}. However, such approaches inherit the lack of robustness of likelihood-based procedures 
and can perform poorly in the presence of contamination or model misspecification. 
To address this issue, complementing our DPD-SISP, we propose a robust second-stage procedure based on minimum penalized DPD estimation 
\citep{ghosh2020ultrahigh} applied on the reduced feasible LRM of the proxy-transformed data.

Formally, consider the general LMM setting described in Section \ref{SEC:DPD_SISP-gen}, without assuming Gaussianity,
and let $\widehat{\mathcal{S}}_\alpha \subset \mathcal{I}$ denote the set of FE covariates retained after applying the DPD-SISP($\alpha$),
with cardinality $d = |\widehat{\mathcal{S}}_\alpha|=o(n)$. 
Following \eqref{EQ:lm-feasible_gen}, we consider the reduced feasible model for the proxy-transformed data as given by 
\begin{equation}
	\widetilde{\bm{y}}_i = \widetilde{\mathbf{X}}_{i,\widehat{\mathcal{S}}_\alpha} \bm{\beta}_{\widehat{\mathcal{S}}_\alpha} + \widetilde{\bm{\varepsilon}}_i,
	~~~ i=, \ldots, m,
	\label{EQ:reduced_model_final}
\end{equation}
where the transformed errors are assumed to be approximately uncorrelated with mean zero and common marginal density $f_{\bm{\eta}}$, 
as discussed in the previous subsection (Section \ref{SEC:DPD_SISP-gen}).
Then, to perform simultaneous estimation and variable selection, we define the minimum penalized DPD estimator (MPDPDE) 
of  $(\bm{\beta}_{\widehat{\mathcal{S}}_\alpha}, \bm{\eta})$ as the minimizer of 
\begin{equation}
	Q_{n,\alpha}(\bm{\beta}_{\widehat{\mathcal{S}}_\alpha}, \bm{\eta}) 
	= \frac{1}{n}\sum_{i=1}^{m}\sum_{k=1}^{n_i} V_\alpha(\widetilde{y}_{ik}, \widetilde{\bm{x}}_{i,\widehat{\mathcal{S}}_\alpha,k}^\top \bm{\beta}_{\widehat{\mathcal{S}}_\alpha}, \bm{\eta})
	+ \sum_{j \in \widehat{\mathcal{S}}_\alpha} p_\lambda(|\beta_j|),
	\label{EQ:penalized_dpd_final}
\end{equation}
where $p_\lambda(\cdot)$ is a suitable penalty function indexed by the regularization parameter $\lambda > 0$,
and $V_\alpha$ is the DPD loss function \eqref{EQ:dpd-obj-gen} corresponding to the chosen working density $f_{\bm{\eta}}$
which simplifies to \eqref{EQ:dpd-obj-norm} under Gaussian LMMs. 
Common examples of penalty include the $\ell_1$ penalty (Lasso), adaptive lasso penalty, and nonconvex penalties such as SCAD 
that enjoys desirable properties such as unbiasedness, sparsity, and continuity \citep{fan2001variable}.
For coherence between the two stages, it is natural to use the same tuning parameter $\alpha$ 
in both the screening (via DPD-SISP) and estimation steps.

The theoretical properties of the MPDPDE under LRMs with general error distributions have been established 
by \cite{ghosh2020ultrahigh} for general non-concave penalties (covering $\ell_1$ and SCAD) and
by \cite{ghosh2024robust} for adaptive lasso penalty. Both proved consistency, oracle properties, and robustness of the MPDPDEs
under suitable assumptions and devised efficient algorithms for their computations.
To facilitate valid inference in the present setting, we adopt a sample-splitting strategy as in \cite{fan2008sure}. 
The data are randomly divided into two parts, one used for screening via DPD-SISP and the other for penalized estimation. 
This ensures approximate independence between the two stages, simplifying the theoretical analysis.
Then, under regularity conditions on the transformed design matrices, the penalty function and the working model, 
and provided that the screening step satisfies the sure screening property (as shown in Section \ref{SEC:theory_strong}), 
the second-stage estimator inherits strong asymptotic guarantees of the MPDPDE identically from \cite{ghosh2020ultrahigh,ghosh2024robust}. 
In particular, letting $(\widehat{\bm{\beta}}_{\widehat{\mathcal{S}}_\alpha}, \widehat{\bm{\eta}}_\alpha)$ denote 
a minimizer of the penalized objective in \eqref{EQ:penalized_dpd_final}, the following properties hold as $n\rightarrow\infty$: 
\begin{itemize}
	\item[(a)] (Selection consistency) The final selected model recovers the true active set with probability tending to one, i.e., 
	$P\left( \widehat{\widehat{\mathcal{S}}}_\alpha= \mathcal{S}_0\right) \to 1$, where 
$\widehat{\widehat{\mathcal{S}}}_\alpha = \left\{ j\in\mathcal{I} : \widehat{\bm{\beta}}_{\widehat{\mathcal{S}}_\alpha, j}\neq 0 \right\}$
	with $\widehat{\bm{\beta}}_{\widehat{\mathcal{S}}_\alpha, j}$ denoting the $j$-th element of $\widehat{\bm{\beta}}_{\widehat{\mathcal{S}}_\alpha}$ for $j\in\mathcal{I}$.
	\item[(b)] (Estimation consistency) $P\left(\|\widehat{\bm{\beta}}_{\widehat{\mathcal{S}}_\alpha} - \bm{\beta}_{0}\|_2 \leq O(\sqrt{s/n})\right) \to 1$.
	\item[(c)] (Asymptotic normality) The non-zero elements of 
	$\sqrt{n}(\widehat{\bm{\beta}}_{\widehat{\mathcal{S}}_\alpha} - \bm{\beta}_{S0})$, 
	together with  $\sqrt{n}(\widehat{\bm{\eta}}_\alpha-\bm{\eta}_0)$,
	have joint asymptotic normal distribution with mean vector zero and 
	an appropriate covariance matrix depending on the derivatives of $V_\alpha$ and $p_\lambda$.

	\item[(d)] (Robustness) The influence function of $\widehat{\bm{\beta}}_{\widehat{\mathcal{S}}_\alpha}$ is bounded for all $\alpha > 0$,
		ensuring stability under small departures from the assumed model.
\end{itemize}

These results justify the proposed two-stage procedure as a unified framework for robust high-dimensional inference in LMMs. 
The screening step effectively reduces dimensionality while retaining all relevant variables, 
and the second-stage penalized DPD estimation achieves consistent variable selection and efficient parameter estimation. 
The use of SCAD or adaptive lasso penalties further enables oracle properties under relatively mild conditions, 
as provided in \cite{ghosh2020ultrahigh} and \cite{ghosh2024robust}, respectively, 
making the overall procedure both theoretically sound and practically effective.

\section{Theoretical Guarantees for the DPD-SISP}
\label{SEC:theory_strong}

%

\subsection{Population-level Justification under General LMMs}
\label{SEC:SISP_ssp_pop}

Let us consider the general LMM framework described in Section \ref{SEC:DPD_SISP-gen}, 
but without assuming Gaussianity  on the REs and error terms, let the working error density $f_{\bm{\eta}}$ as specified therein.
We impose the following standard regularity conditions.

\begin{itemize}
	\item[(A0)] There exists constants $C_Y, C_X>0$ such that  \\
	$\max\limits_{1\leq i\leq m} E\|\bm{y}_{i}\|_2^2 \leq C_Y$ and
	$\max\limits_{1\leq i\leq m}E\|\bm{x}_{ij}\|_2^2 \leq C_X$ for all $j\in \mathcal{I}$.
	
	\item[(A1)] There exists constants $C_m, C_M>0$ such that  \\
	$0<C_m
	\le
	\lambda_{\min}(\bm{\Sigma}_i) 
	\le
	\lambda_{\max}(\bm{\Sigma}_i)
	\le C_M<\infty$ for all $i\geq 1$.
	
	\item[(A2)] For each $\bm{\eta}\in\Theta_\eta$, the working density $f_{\bm{\eta}}$ is strictly log-concave and 
	continuously differentiable  on its support $\mathbb{R}$.  
	Further, the family $\mathcal{F} = \left\{f_{\bm{\eta}} : \bm{\eta}\in\Theta_\eta \right\}$ is identifiable in $\bm{\eta}$. 	

\end{itemize} 
These assumptions are quite flexible, accommodating most common cases of practical LMMs. 
Particularly, (A0)–(A1) are mild moment conditions ensuring that the covariates and responses are well-behaved 
and that the covariance matrices are uniformly well-conditioned. These hold, for example, under sub-Gaussian tails. 
A sufficient condition implying (A1) is that the REs are neither linearly dependent nor distributed with a singular covariance.  
Assumption (A2) accommodates a wide class of error distributions, including normal, skew-normal, and logistic families.

Recall that, under this general LMM formulation the whitening transformation yields the LRM \eqref{EQ:lm_gen}.
Stacking across clusters, as in \eqref{EQ:lm-feas-marginal_gen}, we obtain
\begin{eqnarray}
\bm{y}^* = \mathbf{X}^* \bm{\beta} + \bm{\varepsilon}^*, ~~~~
E[\bm{\varepsilon}^*] =\bm{0}_{n}, ~ \mbox{Cov}(\bm{\varepsilon}^*) = \sigma^2 \mathbb{I}_{n},  
\label{EQ:lm_genV}
\end{eqnarray}
where $\bm{y}^* = (\bm{y}_1^{*\top}, \ldots, \bm{y}_m^{*\top})^\top$,
$\mathbf{X}^* = [\mathbf{X}_1^{*\top} \cdots \mathbf{X}_m^{*\top}]^\top = [\bm{x}_{1}^* \cdots \bm{x}_{p}^*]$
and  $\bm{\varepsilon}^* = (\bm{\varepsilon}_{1}^{*\top}, \ldots, \bm{\varepsilon}_{m}^{*\top})^\top$.
Although independence is not guaranteed without Gaussian assumptions, from \eqref{EQ:lm_genV} and \eqref{EQ:lm-feas-marginal_gen}, 
the components of $\bm{y}^*$, $\widetilde{\bm{y}}$, $\bm{x}_{j}^*$, $\widetilde{\mathbf{x}}_{j}$, for $j\in\mathcal{I}$, 
are uncorrelated and identically distributed. Let the scalar random variables $Y^*$, $\widetilde{Y}$, $X_{j}^*$, $\widetilde{X}_{j}$
denote generic components of the corresponding vectors $\bm{y}^*$, $\widetilde{\bm{y}}$, $\bm{x}_{j}^*$, 
$\widetilde{\mathbf{x}}_{j}$, respectively, for $j\in \mathcal{I}$. 
%
Then, for any $j\geq 1$, Assumptions (A0)--(A1) implies that
\begin{eqnarray}\label{EQ:Ass1}
E[X_j^{*2}] = \frac{1}{n}\sum\limits_{i=1}^m E\left[\bm{x}_{ij}^\top\bm{\Sigma}_i^{-1}\bm{x}_{ij}\right]
\leq \frac{1}{C_m} \frac{1}{n}\sum\limits_{i=1}^m E\left[\bm{x}_{ij}^\top\bm{x}_{ij}\right] \leq \frac{C_X}{C_m}
<\infty.
\end{eqnarray}

Now, for each covariate $j\in\mathcal{I}$, define the population-level marginal parameter as 
\begin{eqnarray}
\bm{\theta}_{j,\alpha}^M =  (\beta_{j0,\alpha}^M, \beta_{j1,\alpha}^M, \bm{\eta}_{j,\alpha}^M) = \argmin_{(\beta_0, \beta_1, \bm{\eta})} 
E\left[V_\alpha(Y^*, \beta_0 + X_{j}^*\beta_1, \bm{\eta})\right], 
~~~ j=1, \ldots, p.
\label{EQ:mdpde-true}
\end{eqnarray}
This represents the population target of the marginal MDPDE when the true whitened data are available,
and can be interpreted as the best-fitting marginal regression parameter under the DPD criterion.
Assumption (A2) ensures that the loss function in \eqref{EQ:mdpde-true} is well-behaved, guaranteeing that 
the minimizers $\bm{\theta}_{j,\alpha}^M$ exists, is unique and is well-defined. 
Let us define  $\psi_\alpha(y, \mu, \bm{\eta}) = -\frac{\partial}{\partial\mu}V_\alpha(y, \mu, \bm{\eta})$ for any $\alpha\geq 0$. 
Then, $\bm{\theta}_{j,\alpha}^M$ satisfies 
\begin{equation*}
E\left[\psi_\alpha(Y^*,  \beta_{j0,\alpha}^M + X_{j}^* \beta_{j1,\alpha}^M, \bm{\eta}_{j,\alpha}^M)\right]=0, ~~ \mbox{ and }
E\left[\psi_\alpha(Y^*,  \beta_{j0,\alpha}^M + X_{j}^* \beta_{j1,\alpha}^M, \bm{\eta}_{j,\alpha}^M)X_j^*\right]=0.
\end{equation*}

The second equation in the above plays a central role in characterizing the marginal relationship between the response and each covariate.
Motivated by this estimating equation, we define a  robust association measure
\begin{eqnarray}
S_{j,\alpha}(\mu, \bm{\eta}) = E[\psi_\alpha(Y^*,  \mu, \bm{\eta})X_j^*], 
~~ \mbox{ for } j\in\mathcal{I}, ~~\bm{\eta}\in\Theta_\eta, 
\label{EQ:signal_robust}
\end{eqnarray}
where $\mu\in\mathbb{R}$ may or may not depend on the covariates.
When $E(X_j^*)=0$, which typically holds under mean-centering and independence between FE and RE covariates, 
the quantity $S_{j,\alpha}(\mu, \bm{\eta}) $ can be interpreted as a generalized covariance-type dependence measure.
At $\alpha=0$ under a Gaussian working model, $\psi_\alpha$ reduces to a linear function of $(Y^* - \mu)$, 
and hence $S_{j,\alpha}(\mu, \bm{\eta})$ reduces (up to scaling) to the classical covariance between $Y^*$ and $X_j^*$; 
see \eqref{EQ:signal_robust_normal} in Section \ref{SEC:SISP-normal}.
More generally, $S_{j,\alpha}$ provides a robust dependence measure that downweights extreme observations at $\alpha>0$.

To formalize the relationship between this robust association measure and the population-level marginal slope parameters, 
we further impose the following smoothness condition from the theory of MDPDEs \citep{basu1998robust,basu2026statistical}.

\begin{itemize}
	\item[(A3)] The function $V_\alpha(y, \mu, \bm{\eta})$ is twice continuously differentiable in its arguments,
	and all second order partial derivatives are uniformly bounded, in absolute value, away from both zero and infinity 
	for all $y, \mu\in\mathbb{R}$ and $\bm{\eta}\in\Theta_\eta$. 
\end{itemize}

This assumption ensures that the DPD loss function has uniformly controlled curvature, 
which is essential for establishing stability of the estimating equations and deriving quantitative bounds.
Then, the following theorem establishes a precise connection between $\beta_{j1,\alpha}^M$  and $S_{j,\alpha}$, 
providing the basis for the sure screening property of the DPD-SISP procedure.

\begin{theorem}\label{THM:SISP-population}
Under a general LMM with working error density satisfying (A2) and non-degenerate FE covariates,  
the following results hold:
\begin{itemize}
	\item[a)] For any $\alpha\geq 0$ and $j\in\mathcal{I}$,  
	$\beta_{j1,\alpha}^M=0$ if and only if $S_{j,\alpha}(\beta_{j0,\alpha}^M, \bm{\eta}_{j,\alpha}^M)=0$.
	
	\item[b)] If additionally (A0)-(A1) and (A3) hold for a given $\alpha\geq 0$, 
	and $\left|S_{j,\alpha}(\beta_{j0,\alpha}^M, \bm{\eta}_{j,\alpha}^M)\right| \geq c_1n^{-\kappa}$
	for all $j\in \mathcal{S}_0$, with constants $c_1>0$ and $\kappa\in(0,1/2)$,
	then there exists a constant $c_2>0$ such that  
	$ \min_{j\in\mathcal S_0} |\beta_{j1,\alpha}^M|\geq c_2n^{-\kappa}.$
\end{itemize}
\end{theorem}

The proof of the above theorem is given in Appendix \ref{APP}. 
It provides the key population-level justification for the proposed DPD-SISP procedure. 
Part (a) ensures that inactive covariates, those not belonging to the true active set $\mathcal{S}_0$,  
do not exhibit spurious marginal signals under the DPD-based criterion.
Part (b) further shows that, for active covariates, a sufficiently strong robust marginal association, 
as quantified through $S_{j,\alpha}$, translates into a non-negligible lower bound on the magnitude of 
the corresponding marginal slope parameter.
Together, these results establishes a clear separation between active and inactive variables at the population level, 
which is fundamental for achieving consistent variable screening in ultrahigh-dimensional settings.

\subsection{Sure Screening Property of the DPD-SISP: General results}
\label{SEC:SISP_ssp}

We now establish the sample-level sure screening property of the proposed DPD-SISP procedure under general LMMs. 
The population-level results in Theorem \ref{THM:SISP-population} indicate that a screening rule based on suitable estimators of 
$\beta_{j1,\alpha}^M$, together with a threshold $\gamma_n \in(0, c_2n^{-\kappa})$, can recover the true active set 
with high probability. 
%
However, in practice, the whitened data $(\bm{y}^*, \mathbf{X}^*)$  are not observable, 
and hence the oracle parameter $\beta_{j1}^M$ cannot be estimated directly. 
Instead, the DPD-SISP utilizes the MDPDE computed from proxy-transformed data 
$(\widetilde{\bm{y}}, \widetilde{\mathbf{x}}_{1}, \ldots, \widetilde{\mathbf{x}}_{p})$,
whose corresponding population target is defined as 
\begin{eqnarray}
	\bm{\theta}_{j,\alpha}^{P} =  (\beta_{j0,\alpha}^P, \beta_{j1,\alpha}^P, \bm{\eta}_{j,\alpha}^P) 
	= \argmin_{(\beta_0, \beta_1, \bm{\eta})} E\left[V_\alpha(\widetilde{Y}, \beta_0 + \widetilde{X}_{j}\beta_1, \bm{\eta})\right], 
	~~~ j\in\mathcal{I}, ~\alpha\geq 0.
	\label{EQ:mdpde-proxy}
\end{eqnarray}
Therefore, the sure screening property at the sample level can be expected if the proxy-induced target $\beta_{j1,\alpha}^P$ 
is sufficiently close to the oracle target $\beta_{j1,\alpha}^M$ and the sample MDPDE $\widehat\beta_{j1,\alpha}$ is concentrated 
sharply around $\beta_{j1,\alpha}^P$, uniformly over $j\in\mathcal{I}$, both with exponentially high probability. 
These requirements are formalized through the following assumptions for some $\kappa\in(0,1/2)$ and fixed $\alpha\geq 0$.

\begin{itemize}
	\item[(B1)] (Uniform exponential consistency of the MDPDE)
	For any $c>0$, there exists a sequence $\{R_n\}_{n\geq 1}$  such that
	$R_n\rightarrow 0$ and $n^{-\tau}\log R_n \rightarrow -\infty$ as $n\rightarrow\infty$, and   
	$$
	P\left(\left|\widehat\beta_{j1,\alpha} - \beta_{j1,\alpha}^P\right| > c n^{-\kappa}  \right) =O(R_n),
	$$
	uniformly over $j\in\mathcal{I}$, for all sufficiently large	$n$.
	 
	\item[(B2)] (Accuracy of the proxy approximation)
	For any $c>0$, there exists a sequence $\{\widetilde{R}_n\}_{n\geq 1}$  such that
	$\widetilde{R}_n\rightarrow 0$ and $n^{-\tau}\log \widetilde{R}_n \rightarrow -\infty$ as $n\rightarrow\infty$, and   
	$$
	P\left(\left|\beta_{j1,\alpha}^P - \beta_{j1,\alpha}^M\right| > c n^{-\kappa}  \right) = O(\widetilde{R}_n),
	$$
	uniformly over $j\in\mathcal{I}$, for all sufficiently large	$n$.
	
	\item[(B3)] (Minimal signal strength) 
	There exists $c_2>0$ such that $\min_{j\in\mathcal S_0} |\beta_{j1,\alpha}^M|\geq c_2n^{-\kappa}.$
	
	\item[(B4)] (Global signal strength control) 
	$||\bm{\beta}_{1,\alpha}^M||_2^2 = O\left( \lambda_{\max}(\bm{\Sigma}^*)\right)$, 
with $\bm{\beta}_{1,\alpha}^M = (\beta_{11,\alpha}^M, \ldots, \beta_{p1,\alpha}^M)^\top$ 
and $\bm{\Sigma}^*= \frac{1}{n}\sum\limits_{i=1}^m E\left[\bm{X}_i^\top\bm{\Sigma}_i^{-1}\bm{X}_{i}\right].$ 
\end{itemize}

These Assumptions (B1)--(B4) are fully generic requirements for establishing the desired properties of the DPD-SISP.
Assumption (B1) ensures exponential concentration of the MDPDE (a special M-estimator) based on transformed data 
$(\widetilde{\bm{y}}, \widetilde{\mathbf{x}}_{1}, \ldots, \widetilde{\mathbf{x}}_{p})$.
Such results are well-studied for both specific and general classes of M-estimators in the literature,
and can be easily verified under suitable moment and tail conditions on distributions of transformed variables. 
Such conditions may be found, for example, in \cite{ghosh2023robust} for the cases 
where  $f_{\bm{\eta}}$ belongs to an exponential family of distributions with known $\bm{\eta}$, 
or in \cite{bratsberg2025exponential} for more general cases with unknown $\bm{\eta}$. 
Assumption (B2) controls the approximation error induced by the proxy matrix;
it is further analyzed in Section \ref{SEC:proxy} to derive simpler sufficient conditions on $\bm{P}$. 
Assumption (B3) guarantees sufficient separation between active and inactive covariates in terms of the marginal minimum DPD functionals; 
it holds, e.g, under the conditions of Theorem \ref{THM:SISP-population} including (A0)--(A3). 
Finally, Assumption (B4) restricts the overall signal strength to control false positives, 
as a straightforward generalization of the usual condition required in the literature of SIS; 
see, e.g., \cite{fan2008sure,ghosh2023robust} among others. 
Under these conditions, we can now prove the following theoretical guarantees for the proposed DPD-SISP;
the proof is given in Appendix \ref{APP}.

\begin{theorem}\label{THM:SISP-ssp}
Under a general LMM satisfying Assumptions (B1)--(B3), for a given $\alpha\geq 0$ and $\kappa\in(0,1/2)$,  
the following results hold.
\begin{itemize}
	\item[a)]  For the threshold $\gamma_n=C n^{-\kappa}$, with $0<C\leq c_2/2$, 
	$$
	P\left(\mathcal{S}_0\subseteq \widehat{\mathcal{S}}_\alpha(\gamma_n)\right)\geq 1 - s  (R_n+\widetilde{R}_n)C_1, 
	$$
	for some constant $C_1>0$ and all sufficiently large $n$. 
	
	\item[b)]  If additionally (B4) holds, then for the same threshold $\gamma_n$ as in Part (a), we have
	$$
	P\left(|\widehat{\mathcal{S}}_\alpha(\gamma_n)|\leq O(n^{2\kappa}\lambda_{\max}(\bm{\Sigma}^*))\right)\geq 1 - p (R_n+\widetilde{R}_n)C_2,
	$$
	for some constant  $C_2>0$ and all sufficiently large $n$. 
\end{itemize}
\end{theorem}

The above theorem establishes that the proposed DPD-SISP procedure achieves the sure screening property 
with probability tending to one at an exponential rate. 
This is because the required assumptions ensure that both $R_n$ and $\widetilde{R}_n$ decay super-exponentially relative to $\log p$, 
and hence $p R_n\rightarrow 0$ and $p \widetilde{R}_n\rightarrow 0$ as $n\rightarrow\infty$ even when $\log p=O(n^\tau))$. 
Consequently, since $s\ll p$,  the final probabilities in both parts of Theorem \ref{THM:SISP-ssp} tend to one
exponentially fast, leading to the desired sure screening and size control guarantees.

We may further note the crucial role of the signal rate parameter $\kappa \in (0, 1/2)$. 
The minimal signal condition in (B3) matches the classical ``beta-min" condition from 
the theory of high-dimensional statistical methods including SIS procedures. 
The restriction $\kappa <1/2$ ensures that the signal dominates the stochastic fluctuation of the estimator, 
which typically occurs at rate $n^{-1/2}$, enabling consistent ranking of variables. 
This matches the optimal detection boundaries of existing SIS approaches.
The threshold $\gamma_n=C n^{-\kappa}$ is therefore rate-optimal, balancing false negatives and false positives. 
Furthermore, this cut-off also ensures that the number of selected variables is polynomial in $n$, even when $p$ is exponential. 
Particularly, assuming $\lambda_{\max}(\bm{\Sigma}^*)=O(1)$, which holds under (A0)--(A1),
we get $|\widehat{\mathcal{S}}_\alpha(\gamma_n)|\leq O(n^{2\kappa})=o(n)$. 
This reduction makes the model amenable to subsequent refinements, like penalized estimation, 
placing DPD-SISP on the same theoretical footing as state-of-the-art SIS methods 
while providing superior robustness and dependence adjustment.

\subsection{On the choice of the Proxy matrix }
\label{SEC:proxy}

Unlike classical SIS methods that operate directly on independent observations,
the proposed DPD-SISP introduces an additional approximation layer through the proxy matrix $\bm{P}$.
Assumption (B2) requires that the bias induced by this proxy must be asymptotically negligible relative to the minimal signal strength. 
This can be ensured by imposing the following spectral consistency condition on $\bm{P}$, 
analogous to the ``\textit{working correlation consistency}" conditions used in high-dimensional inference for dependent data 
\cite[see, e.g.,][]{fan2012variable,bradic2020test}. 

\begin{itemize}
	\item[(P0)] (Proxy Spectral Consistency) The proxy matrix $\bm{P}$ satisfies, 	with probability at least  $1 - O(\widetilde{R}_n)$,  
	$$
	\max_{1 \le i \le m} \|\widehat{\bm{\Sigma}}_i - \bm{\Sigma}_i\|_{op} = O(n^{-\kappa}),
	$$
where the sequence  $\{\widetilde{R}_n\}_{n\geq 1}$  satisfies $\widetilde{R}_n\rightarrow 0$ 
and $n^{-\tau}\log \widetilde{R}_n \rightarrow -\infty$ as $n\rightarrow\infty$. 
In addition, there exist constants $0<\widetilde{C}_m \le \widetilde{C}_M<\infty$ such that\\
$\widetilde{C}_m
	\le
	\lambda_{\min}(\widehat{\bm{\Sigma}}_i) 
	\le
	\lambda_{\max}(\widehat{\bm{\Sigma}}_i)
	\le \widetilde{C}_M$ for all $i\geq 1$.
\end{itemize}

\noindent
The sufficiency of (P0) is formally stated in the following theorem; 
see Appendix \ref{APP} for its detailed proof. 

\begin{theorem}\label{THM:Proxy1}	
Under a general LMM satisfying (A0), (A1) and (A3), 
Assumption (P0) implies Assumption (B2) for any given $\alpha\geq 0$ and $\kappa\in(0,1/2)$.  
\end{theorem}

One may quickly note that, together with (A1), Assumption (P0) implies that 
$$
\max_{1 \le i \le m} \|\bm{V}_i - \mathbb{I}_{n_i}\|_{op} \leq \max_{1 \le i \le m}\| \widehat{\bm{\Sigma}}_i^{-1/2}\|_{op}^{2} \|\widehat{\bm{\Sigma}}_i - \bm{\Sigma}_i\|_{op}
= O(n^{-\kappa}),
$$
with probability at least  $1 - O(\widetilde{R}_n)$. 
The rate $O(n^{-\kappa})$ matches the minimal signal level, 
ensuring that the proxy-induced approximation error does not dominate the minimal signal strength with high probability.
This, in turn, justifies the proxy-based construction of the DPD-SISP for sure screening of FEs.

The following lemma provides a convenient sufficient condition for verifying (P0) in practice;
its proof follows simply from the properties of operator norms  and is hence omitted.

\begin{lemma}\label{LEM:proxy}
Suppose that the proxy matrix $\bm{P}$ satisfies the following under a general LMM: 
$\|\bm{P} - \sigma^{-2}\bm{\Psi}\|_{\mathrm{op}} =O(\delta_n)$,
and $\max\limits_{1\leq i\leq m} \|\mathbf{Z}_i\|_{\mathrm{op}} = O(1)$,
for some sequence $\delta_n \to 0$. Then, we have
\[
\max\limits_{1\leq i\leq m} \|\widehat{\bm{\Sigma}}_i - \bm{\Sigma}_i\|_{\mathrm{op}} = O(\delta_n).
\]
\end{lemma}

Motivated by existing work on high-dimensional LMMs, we consider two practical choices of the proxy matrix $\bm{P}$:
\begin{itemize}
	\item[(I0-P)] $\bm{P}_I =  \widehat{\sigma}^{-2}\widehat{\bm{\Psi}}$, 
	where $\widehat{\sigma}$ and $\widehat{\bm{\Psi}}$ are the REML estimates of $\sigma$ and $\bm{\Psi}$, respectively, 
	obtained from the intercept-only sub-model of (\ref{EQ:lmm}) given by 
	$$
	\bm{y}_i = \beta_0 + \mathbf{Z}_i \mathbf{b}_i + \boldsymbol{\varepsilon}_i, ~~~ i=1, \ldots, m.
	$$ 
	This choice is computationally efficient, as explored previously in \cite{rugamer2022selective}.
	In view of Lemma \ref{LEM:proxy}, $\bm{P}_I$ satisfies the required condition (P0)
	whenever  $\|\mathbf{Z}_i\|_{\mathrm{op}}$ is uniformly bounded for $i\geq1$.
	This is because the REML estimates $\widehat{\sigma}$ and $\widehat{\bm{\Psi}}$
	are $\sqrt{n}$-consistent under standard regularity assumptions on the LMM 
	and we have taken $\kappa<1/2$.

	\item[(cv-P)] $\bm{P}_{cv} = a_{opt}\mathbb{I}_{q}$, where $a_{opt}$ is selected via cross-validation 
	to minimize prediction error in the whitened model \citep{li2022inference}. 
	This data-driven proxy improves upon fixed choices such as $(\log n)\mathbb{I}_q$, initially proposed in \cite{fan2012variable}, 
	by better adapting to the underlying covariance structure, 
	particularly when $\bm{\Psi}$ is well approximated by a scalar multiple of the identity (e.g., bounded condition number).
	This is demonstrated, both theoretically and empirically,  in \cite{li2022inference}, 
	In this case, Assumption (P1) holds based on the rate of approximation for $\bm{\Psi}$
	and consistency of the estimated $a_{opt}$. 
\end{itemize}

\begin{remark}[Can we simply ignore the random-effects structure in practice?]
The proxy matrix $\bm{P} = \bm{O}_q$ corresponds to ignoring the RE structure entirely, leading to the usual LRM based SIS.
While computationally attractive, this choice is valid only when intra-cluster dependence is weak.
Formally, it satisfies (P0) only when  $\|\bm{\Sigma}_i - \mathbb{I}_{n_i}\|_{op} = O(n^{-\kappa})$, 
since $\widehat{\bm{\Sigma}}_i = \mathbb{I}_{n_i}$; this is equivalent to the condition  
$\|\bm{\Sigma}_i^{-1} - \mathbb{I}_{n_i}\|_{op} = O(n^{-\kappa})$ under (A1). 
Since the elements of  $\bm{\Sigma}_i^{-1} $ is characterized by the intra-cluster correlations, 
they have to be smaller compared to the (rate of) signal strength to ensure (P0), and hence the sure screening of the resulting DPD-SISP.
In the presence of strong correlations, however, this simplification can lead to substantial bias and failure of the screening procedure.
\end{remark}

\subsection{Sure Screening under the Gaussian LMM}
\label{SEC:SISP-normal}


We now specialize the general theory to Gaussian LMMs, where both the REs and the errors are normally distributed. 
In this case, the working density in the DPD framework is Gaussian with variance parameter $\eta_j=\sigma_j^2$.
Under this setup, Assumptions (A2) and (A3) hold automatically due to the smoothness and regularity of the normal density. 
Moreover, Gaussianity implies bounded second moments of the responses, and sub-Gaussian covariates ensure that (A0) is satisfied. 
Therefore, whenever (A1) and (P0) hold, the DPD-SISP enjoys both population- and sample-level sure screening guarantees.

In particular, for this specific case, 
we can easily obtain from (\ref{EQ:dpd-obj-norm}) that
\begin{eqnarray}
\psi_\alpha(y, \mu, \sigma) = 	 \omega_\alpha(y-\mu|\sigma) (y-\mu),
~~~~~\mbox{ with } \omega_\alpha(r|\sigma) = \frac{1+\alpha}{\sigma^{\alpha+2}(2\pi)^{\alpha/2}}
			e^{-\frac{\alpha r^2}{2\sigma^2}}, 
	\label{EQ:dpd-EstEq-norm}
\end{eqnarray}
for any $\alpha\geq 0$. 
Thus, each marginal MDPDE  solves a weighted score equation, 
where observations with large residuals receive exponentially smaller weights through $\omega_\alpha$. 
This highlights the robustness of the DPD approach relative to classical likelihood-based methods.
Next, our robust association measure $S_{j,\alpha}(\mu, \bm{\eta})$ simplifies to 
\begin{eqnarray}
	S_{j,\alpha}(\mu, \sigma^2) = E[\omega_\alpha(Y^*-\mu|\sigma) (Y^*-\mu)X_j^*], 
	~~ \mbox{ for } j=1, \ldots, p, 
	\label{EQ:signal_robust_normal}
\end{eqnarray}
which can be interpreted as a weighted covariance between $Y^*$ and $X_j^*$ when $E[X_j^*]=0$.
Under these conditions, the conclusions of Theorem \ref{THM:SISP-population} remain valid for all $\alpha \geq 0$, 
particularly implying the minimal signal condition (B3), whenever $S_{j,\alpha}(\mu, \sigma^2)$ remains bounded away 
from zero for all $j\in \mathcal{S}_0$. 
Further, as shown in Theorem \ref{THM:Proxy1}, Assumptions (A1) and (P0) ensures (B2). 
Finally, following the theory of \cite{bratsberg2025exponential}, one can show that (A1) and (P0) also implies Assumptions (B1) and (B4).
The following theorem summarizes all these results, thereby establishing theoretical guarantees for the DPD-SISP under Gaussian LMMs.

\begin{theorem}\label{THM:SISP-normal}
Consider a Gaussian LMM with non-degenerate FE covariates, and assume the following conditions for a given $\alpha\geq 0$: 
\begin{itemize}
	\item[ i)] There exists $C_X>0$ such that $\max\limits_{1\leq i\leq m}E\|\bm{x}_{ij}\|_2^2 \leq C_X$ for all $j\in\mathcal{I}$
	(e.g., sub-Gaussian covariates).
	\item[ ii)] There exists $\kappa\in(0, 1/2)$ and $c_1>0$ such that  \\
	$\left|E[\omega_\alpha\left(Y^*-\beta_{j0,\alpha}^M|\sigma_{j,\alpha}^M\right) \left(Y^*-\beta_{j0,\alpha}^M\right)X_j^*] \right| \geq c_1n^{-\kappa}$ for all $j\in \mathcal{S}_0$.
	\item[ iii)] Assumptions (A1) and (P0) hold.  
\end{itemize}	
Then, we have the following results:
\begin{itemize}
	\item[a)] $\beta_{j1,\alpha}^M=0$ if and only if $S_{j,\alpha}(\beta_{j0,\alpha}^M, \bm{\eta}_{j,\alpha}^M)=0$ for any $j\in\mathcal{I}$.
		Also, $ \min\limits_{j\in\mathcal S_0} |\beta_{j1,\alpha}^M|\geq c_2n^{-\kappa}$ for some constant $c_2>0$. 
			
	\item[ b)] Taking $\gamma_n=C n^{-\kappa}$ for some $0<C\leq c_2/2$ (with $c_2$ from Part (a)), we get 
	$$
	P\left(\mathcal{S}_0\subseteq \widehat{\mathcal{S}}_\alpha(\gamma_n)\right)\rightarrow 1,  
~~~\mbox{ and }~~	
P\left(|\widehat{\mathcal{S}}_\alpha(\gamma_n)|\leq O(n^{2\kappa}\lambda_{\max}(\bm{\Sigma}^*))\right)\rightarrow 1,
~~\mbox{ as }n\rightarrow\infty.
	$$
\end{itemize}	
\end{theorem}

The Gaussian LMM setting provides a particularly transparent interpretation of the DPD-SISP procedure. 
The robust association measure reduces to a weighted covariance, where the exponential weights attenuate the influence of outliers, 
thereby stabilizing variable ranking in contaminated or heavy-tailed settings. At the same time, 
the proxy-based transformation effectively removes dependence induced by REs, 
allowing the procedure to mimic classical SIS in an approximately independent framework. 
Consequently, DPD-SISP achieves the same optimal screening rates as standard SIS 
while offering enhanced robustness and adaptability to complex dependence structures.

\subsection{Robustness Guarantees}
\label{SEC:SISP-robustness}

A central advantage of the proposed DPD-SISP procedure lies in its inherent robustness against 
data contamination and mild model specification.  This robustness is induced by the use of  the DPD loss at the marginal screening stage, 
which yields bounded influence and stable estimation behavior even in the presence of outliers.
Importantly, this property persists regardless of how the proxy matrix is constructed.
Even when $\bm{P}$ is obtained through non-robust procedures such as REML or LS, 
e.g., the choices I0-P and cv-P in Section \ref{SEC:proxy}, the marginal MDPDEs
$\widehat{\bm{\theta}}_{j,\alpha}$, $j \in \mathcal{I}$, effectively mitigate the impact of contamination in the transformed data.
Consequently, the resulting screening rule inherits robustness at the level of variable ranking.
To formalize this, we examine both local and global robustness properties of the marginal MDPDEs
through the classical notions of influence function (IF) and asymptotic breakdown point (BP).

Let us fix any $j\in\mathcal{I}$. It follows from the general theory of M-estimators or the MDPDEs 
\citep{basu1998robust,basu2011statistical, basu2026statistical} that the IF of the $j$-th marginal MDPDE 
$\widehat{\bm{\theta}}{j,\alpha}$ at a contamination point $(y_t, x_{jt}) \in \mathbb{R}^2$ is proportional to the gradient of the DPD loss:
\begin{eqnarray}\label{EQ:IF_MDPDE}
IF_\alpha(y_t, {x}_{jt}|\bm{\theta}_j) \propto \nabla V_\alpha(y_t, \beta_{j0}+ {x}_{jt}\beta_{j1}, \sigma)_j,
\end{eqnarray}
where $\nabla$ denotes the full gradient with respect to $\bm{\theta}_j =(\beta_{j0}, \beta_{j1}, \sigma_j)^\top$.
Thus, boundedness of the derivative $\nabla V_\alpha$ with respect to the data ensures bounded influence of individual observations.
For $\alpha > 0$, this boundedness holds for a wide class of parametric models, implying that the marginal MDPDE,
and hence the induced association measure $S_{j,\alpha}$, is locally robust against infinitesimal contamination.
As a result, the screening statistic $\widehat{\beta}_{j1,\alpha}$ is insensitive to extreme observations, 
yielding a stable ranking of covariates.

In the particular case of Gaussian LMMs, we can further simplify the IF of $\widehat{\beta}_{j1, \alpha}$ to a more explicit form as 
\begin{eqnarray}\label{EQ:IF_MDPDE-normal}
	IF_\alpha(y_t, {x}_{jt}|\bm{\theta}_j) \propto (y_t- \beta_{j0}+ {x}_{jt}\beta_{j1})e^{-\frac{\alpha (y_t- \beta_{j0}+ \widetilde{x}_{jt}\beta_{j1})^2}{2\sigma_j^2}}.
\end{eqnarray}
This function is bounded for all $\alpha > 0$, while it becomes unbounded at $\alpha = 0$ (corresponding to the classical likelihood case).
Moreover, the exponential downweighting becomes stronger as $\alpha$ increases, 
leading to progressively smaller influence of extreme residuals.
This illustrates the robustness–efficiency trade-off governed by $\alpha$.

Beyond local behavior, global robustness of the marginal MDPDEs can be characterized through their asymptotic BP.
For LRMs, it is well established that when the error variance is known, the MDPDEs of regression coefficients achieve 
the maximal asymptotic BP of $1/2$ under mild conditions \citep{ghosh2013robust}.
This implies stability under up to $50\%$ contamination.
More recent results extend this property to the case of unknown variance \citep{jana2025asymptotic},
showing that the BP remains strictly positive and increases with $\alpha$,
and in several important models (e.g., the Gaussian case) attains the optimal value $1/2$.
These results apply directly to each marginal regression in our framework,
implying that the estimators $\widehat{\beta}_{j1,\alpha}$ inherit strong global robustness properties.

Since the DPD-SISP ranks variables based on the magnitudes of $\widehat{\beta}_{j1,\alpha}$,
the bounded influence and high breakdown point of these estimators ensure that the ranking is resistant to contamination.
In particular, both the magnitude and ordering of marginal signals remain stable under a substantial fraction of corrupted observations.
This provides a rigorous justification for the robustness of the proposed screening procedure,
distinguishing it from classical SIS methods that rely on non-robust estimators and are highly sensitive to outliers.

\section{Robust Conditional Screening of Fixed-effects}
\label{SEC:dpdCSIS}

The DPD-SISP procedure screen through all available FE covariates. However, in many applications, 
prior scientific knowledge suggests that certain FE covariates must be included in the model, 
such as confounders, baseline characteristics, or variables of intrinsic interest.
To accommodate such scenarios, \cite{barut2016conditional} developed a conditional SIS (CSIS) for LRMs
which evaluates the incremental contribution of each candidate variable given a pre-specified conditioning set 
through a low-dimensional joint regression model. 
Since classical CSIS relies on LS- or ML-based estimators, it inherits their sensitivity to contamination and model misspecification. 
To address this limitation, we extend the DPD-SISP framework by developing the \textit{DPD-based conditional SIS with Proxy matrix}, 
termed DPD-CSISP, for robust conditional screening under ultrahigh-dimensional LMMs.

\subsection{The  DPD-CSISP Algorithm}

Let $\mathcal C \subset \mathcal{I}$ be a given conditioning set of FE covariates  with cardinality $|\mathcal C|=p_c$, where $p_c \ll n$. 
Using the proxy-whitened data $\{(\widetilde{\bm{y}}_i,\widetilde{\mathbf{X}}_i)\}_{i=1}^m$ defined in (\ref{EQ:proxy}), 
for each $j \notin \mathcal C$, we consider the following low-dimensional feasible LRM as given by 
\begin{eqnarray}
	\widetilde{\bm{y}}_i = \beta_{j0}^M + 	\widetilde{\mathbf{X}}_{i,\mathcal C}^{\top}\bm{\beta}_{j, \mathcal C}^M + \widetilde{\mathbf{x}}_{ij} \beta_{j1}^M + \bm{\varepsilon}_{ij}, 
	\label{EQ:lm-feas-marginalC}
\end{eqnarray}   
where $\widetilde{\mathbf{X}}_{i,\mathcal{C}}$ denote the $n_i\times p_c$ submatrix of $\widetilde{\mathbf{X}}_i$ consisting of 
the columns with index in $\mathcal{C}$, for each $i=1, \ldots, m$.
This coincides with the  marginal feasible LRM (\ref{EQ:lm-feas-marginal}) used in DPD-SISP when $\mathcal{C} = \emptyset$,
the empty set, representing no pre-specified conditioning variable. 

In the presence of some conditioning variables in $\mathcal{C}$, 
we compute the marginal MDPDE $\widehat{\beta}_{j1,\alpha}^{(\mathcal{C})}$ of $\beta_{j1}^M$ from (\ref{EQ:lm-feas-marginalC}) 
under suitable assumption on its error components. 
In particular,  the associated DPD loss function in the present case is given by 
\begin{eqnarray}
\mathcal{L}_{n, \alpha}^{(C)}(\bm{\theta}) = \frac{1}{n}\sum_{i=1}^{m}\sum_{k=1}^{n_i} V_\alpha(\widetilde{y}_{ik}, 
\beta_0^M+ \widetilde{\mathbf{x}}_{i,\mathcal C,k}^{\top}\bm{\beta}_{\mathcal C}^M+ \widetilde{x}_{ij,k}\beta_1^M, \sigma_M),
	~~\bm{\theta} = (\beta_0^M, \bm{\beta}_{\mathcal C}^M \beta_1^M, \sigma_M^2)^\top,
	\label{EQ:dpd-loss-C}
\end{eqnarray}
where $\bm{x}_{i, \mathcal{C}, k}^\top$ denotes the $k$-th row of $\widetilde{\mathbf{X}}_{i,\mathcal{C}}$ for $k=1, \ldots, n_i$, 
$i=1, \ldots, m$, and $V_\alpha(y, \eta, \sigma)$ is as given in (\ref{EQ:dpd-obj-norm}) or (\ref{EQ:dpd-obj-gen})
depending on whether the errors are assumed to have Gaussian or any general distribution.
The conditional screening statistic is defined as the marginal slope estimate $\widehat{\beta}_{j1,\alpha}^{(\mathcal C)}$.
We  rank the available (non-conditioning) FE covariates by the absolute value of the associated marginal MDPDE 
$\widehat{\beta}_{j1,\alpha}^{(\mathcal{C})}$ for $j\notin\mathcal{C}$, and select the top ranked ones as before. 
The resulting DPD-CSISP screened set is then given by 
$\widehat{\mathcal S}_\alpha^{\mathrm{CSIS}} = \mathcal{C} \cup \widehat{\mathcal S}_\alpha^{(\mathcal C)}$,
where $\widehat{\mathcal S}_\alpha^{(\mathcal C)}$ is the set of variables selected based on $\widehat{\beta}_{j1,\alpha}^{(\mathcal{C})}$ 
for a given cardinality $d\in\mathbb{N}$ or a pre-specified threshold $\gamma_n$.
The full procedure of DPD-CSISP is presented in Algorithm \ref{ALG:DPD-CSISP};
note that it reduces to the DPD-SISP given in Algorithm \ref{ALG:DPD-SISP} for $\mathcal{C} = \emptyset$, the empty set.

Conceptually, DPD-CSISP removes the contribution of $\mathcal{C}$ before evaluating each candidate variable. 
This conditional formulation reduces spurious correlations induced by shared predictors 
and mitigates distortion due to proxy-based transformations.

\begin{minipage}{.97\textwidth}
	\begin{algorithm}[H]
		\caption{\small DPD-CSISP($\alpha$) for general LMMs}
		\label{ALG:DPD-CSISP}
		\small
		\begin{algorithmic}[1]
			
			\STATE \textbf{Input:} Clustered data 
			$\{(\bm y_i,\mathbf X_i,\mathbf Z_i)\}_{i=1}^m$ from the LMM \eqref{EQ:lmm}; 
			tuning parameter $\alpha>0$; \\~~~~~~~~~~
			proxy matrix $\bm P$; conditioning set $\mathcal C$;
			screening parameter $d$ (or threshold $\gamma_n$).
			
			\STATE Set total sample size $n=\sum_{i=1}^m n_i$.\\~~~~~~~~~
			\COMMENT{\textit{Proxy-based covariance approximation}}
			\FOR{$i=1$ to $m$} 
			\STATE Compute proxy covariance $\widehat{\boldsymbol\Sigma}_i =\mathbf Z_i \bm P \mathbf Z_i^\top+\mathbb{I}_{n_i}$.
			\STATE Obtain transformed data $(\widetilde{\bm y}_i,\widetilde{\mathbf X}_i)$ using \eqref{EQ:proxy}.
			\ENDFOR\\~~~~~~~~~			\COMMENT{\textit{Marginal DPD estimation}}
			\FOR{each $j\notin\mathcal{C}$} 			
			\STATE Define the DPD loss $V_\alpha(\cdot)$ as in \eqref{EQ:dpd-obj-norm} for  Gaussian error distribution, 
			or as in (\ref{EQ:dpd-obj-gen}) otherwise.
			\STATE Compute the marginal MDPDE
			$\widehat{\bm\theta}_{j,\alpha}^{(\mathcal{C})}	= \left(\widehat\beta_{j0,\alpha}^{(\mathcal{C})}, 
			\widehat{\bm{\beta}}_{j\mathcal{C},\alpha}^{(\mathcal{C})},
			\widehat\beta_{j1,\alpha}^{(\mathcal{C})}, \widehat{\bm{\eta}}_{j,\alpha}^{(\mathcal{C})}\right)^\top$ as
			\[
\widehat{\bm\theta}_{j,\alpha}^{(\mathcal{C})}	= \argmin_{\left(\beta_0, \bm{\beta}_{\mathcal C}, \beta_1, \bm{\eta}\right)^\top} \frac{1}{n}\sum_{i=1}^{m}\sum_{k=1}^{n_i} 
V_\alpha\!\left(\widetilde y_{ik},
\beta_0+ \widetilde{\mathbf{x}}_{i,\mathcal C,k}^{\top}\bm{\beta}_{\mathcal C}+ \widetilde{x}_{ij,k}\beta_1, \bm{\eta}\right).
			\]
			\STATE Extract the marginal slope estimate $\widehat\beta_{j1,\alpha}^{(\mathcal{C})}$ from 
			$\widehat{\bm\theta}_{j,\alpha}^{(\mathcal{C})}$.
			\ENDFOR\\~~~~~~~~~
			\COMMENT{\textit{Screening}}
			\STATE Rank $\left\{|\widehat\beta_{j1,\alpha}^{(\mathcal{C})}|: j\notin\mathcal{C} \right\}$ in decreasing order.
			
			\STATE Define 
$\widehat{\mathcal S}_\alpha^{(\mathcal{C})}(d)= \left\{ j\notin\mathcal{C}: |\widehat\beta_{j1,\alpha}^{(\mathcal{C})}| 
\text{ is among the top } d \right\}$, or  $\widehat{\mathcal S}_\alpha^{(\mathcal{C})}(\gamma_n)
=\left\{ j\notin\mathcal{C}: |\widehat\beta_{j1,\alpha}^{(\mathcal{C})}| \ge \gamma_n \right\}$.
			
			\STATE \textbf{Output:} Screened active set 
			$\widehat{\mathcal S}_\alpha^{\rm CSIS} = \mathcal{C} \cup \widehat{\mathcal S}_\alpha^{(\mathcal{C})}$.
			
		\end{algorithmic}
	\end{algorithm}
\end{minipage}

\subsection{Sure Screening Property}

We now establish the theoretical guarantees for DPD-CSISP, starting with the population level justification 
using the robust association measure $S_{j,\alpha}(\mu, \bm{\eta})$ defined in \eqref{EQ:signal_robust}.
However, now the results should depend on the marginal association in the presence of conditioning FE covariates.
Formally, let us define the population-level conditional (target) parameter, given $\mathcal{C}$, as 
\begin{eqnarray}
	\bm{\theta}_{j,\alpha}^{(\mathcal{C})M} =  (\beta_{j0,\alpha}^{(\mathcal{C})M}, \beta_{j\mathcal{C},\alpha}^{(\mathcal{C})M}, 
	\beta_{j1,\alpha}^{(\mathcal{C})M}, \bm{\eta}_{j,\alpha}^{(\mathcal{C})M}) 
	= \argmin_{(\beta_0, \beta_{\mathcal{C}}, \beta_1, \bm{\eta})} 
	E\left[V_\alpha(Y^*, \beta_0 + \bm{X}_{\mathcal{C}}^*\beta_{\mathcal{C}} + X_{j}^*\beta_1, \bm{\eta})\right], 
	\nonumber\\
	~~~~\mbox{ for each }~ j\notin\mathcal{C}.~~~~~~~
	\label{EQ:mdpde-trueC}
\end{eqnarray}
Then, we have the following extension of Theorem \ref{THM:SISP-population},
providing the desired connection between $\beta_{j1,\alpha}^{(\mathcal{C})M}$  and $S_{j,\alpha}$ with suitable arguments 
which, in turn, forms the basis for the proposed DPD-CSISP.
The proof is exactly similar to that of Theorem \ref{THM:SISP-population} and is thus omitted. 

\begin{theorem}\label{THM:CSISP-population}
Under a general LMM with working error density satisfying (A2) and non-degenerate FEs,  
we have the following results:
	\begin{itemize}
		\item[a)] For any $\alpha\geq 0$ and $j\notin\mathcal{C}$,  
		$\beta_{j1,\alpha}^{(\mathcal{C})M}=0$ if and only if 
		$S_{j,\alpha}(\beta_{j0,\alpha}^M+ \bm{X}_\mathcal{C}^*\beta_{j\mathcal{C},\alpha}^{(\mathcal{C})M}, \bm{\eta}_{j,\alpha}^M)=0$.
		
		\item[b)] If additionally (A0)-(A1) and (A3) hold for a given $\alpha\geq 0$, 
		and\\ $\left|S_{j,\alpha}\left(\beta_{j0,\alpha}^M+\bm{X}_\mathcal{C}^*\beta_{j\mathcal{C},\alpha}^{(\mathcal{C})M}, \bm{\eta}_{j,\alpha}^M\right)\right| \geq c_1n^{-\kappa}$
		for all $j\in \mathcal{S}_0\cap \mathcal{C}^c$, with some constants $c_1>0$ and $\kappa\in(0,1/2)$,
		then there exists a constant $c_2>0$ such that  
		$ \min\limits_{j\in\mathcal S_0\cap \mathcal{C}^c} |\beta_{j1,\alpha}^{(\mathcal{C})M}|\geq c_2n^{-\kappa}.$
	\end{itemize}
\end{theorem}

Next, we establish the sample-level sure screening property of the DPD-CSISP by extending the arguments presented for DPD-SISP. 
In particular, the following theorem presents the extension of Theorem \ref{THM:SISP-ssp} under general LMMs,
where we have now defined  
$\bm{\beta}_{1,\alpha}^{\mathcal{C}M} = \left(  \beta_{j1,\alpha}^{(\mathcal{C})M} : j \notin \mathcal{C}\right)^\top$ 
based on \eqref{EQ:mdpde-trueC} and 
\begin{eqnarray}
	\bm{\theta}_{j,\alpha}^{(\mathcal{C})P} 
	=  \left(\beta_{j0,\alpha}^{(\mathcal{C})P}, \bm{\beta}_{j\mathcal{C},\alpha}^{(\mathcal{C})P^\top}, 
	\beta_{j1,\alpha}^{(\mathcal{C})P}, \bm{\eta}_{j,\alpha}^{(\mathcal{C})P}\right)^\top 
	= \argmin_{(\beta_0, \bm{\beta}_{\mathcal{C}}^\top, \beta_1, \bm{\eta})^\top} 
	E\left[V_\alpha(Y^*, \beta_0 + \bm{X}_{\mathcal{C}}^*\bm{\beta}_{\mathcal{C}} + X_{j}^*\beta_1, \bm{\eta})\right], 
	\nonumber\\
	~~~~\mbox{ for each }~ j\notin\mathcal{C}, ~~\alpha\geq0.~~~~~~~
	\label{EQ:mdpde-proxyC}
\end{eqnarray}

\begin{theorem}\label{THM:CSISP-ssp}
Consider a general LMM which satisfies Assumptions (B0)--(B3) for the conditional parameters  
$\beta_{j1,\alpha}^{M}$ and $\beta_{j1,\alpha}^{P}$, replacing $\beta_{j1,\alpha}^{(\mathcal{C})M}$ and 
$\beta_{j1,\alpha}^{(\mathcal{C})P}$, respectively, for all $j\notin\mathcal{C}$ (instead of all $j\in\mathcal{I}$),
and some $\kappa\in(0,1/2)$. Let $\widehat{\mathcal S}_\alpha^{\rm CSIS}$ denote 
the screened active set obtained by DPD-CSISP (Algorithm \ref{ALG:DPD-CSISP}) with $\gamma_n=C n^{-\kappa}$ for some $0<C\leq c_2/2$.
Then, we have the following results.
	\begin{itemize}
		\item[a)]  $P\left(\mathcal{S}_0\subseteq \widehat{\mathcal S}_\alpha^{\rm CSIS}\right)\geq 1 - s  (R_n+\widetilde{R}_n)C_1,$
		for some constant $C_1>0$ and all sufficiently large $n$. 
		
		\item[b)]  If additionally 	$||\bm{\beta}_{1,\alpha}^{\mathcal{C}M}||_2^2 
		= O\left( \lambda_{\max}(\bm{\Sigma}_\mathcal{C}^*)\right)$ for  some suitable matrix $\bm{\Sigma}_\mathcal{C}^*$,  
		then we have
		$$
		P\left(|\widehat{\mathcal S}_\alpha^{\rm CSIS}|\leq O(n^{2\kappa}\lambda_{\max}(\bm{\Sigma}^*))\right)\geq 1 - p (R_n+\widetilde{R}_n)C_2,
		$$
		for some constant  $C_2>0$ and all sufficiently large $n$. 
	\end{itemize}
\end{theorem}

The proof of the above theorem follows along the lines of that of Theorem \ref{THM:SISP-ssp}, and hence it is not added here for brevity. 
It establishes that DPD-CSISP achieves the same sure screening and dimensionality control guarantees as DPD-SISP, 
up to conditioning on $\mathcal{C}$. Thus, all theoretical insights developed for the unconditional procedure 
extend naturally to the conditional setting.
Importantly, conditioning can substantially improve screening accuracy in the presence of correlated predictors. 
By removing variation explained by $\mathcal{C}$, DPD-CSISP reduces spurious marginal associations 
and enhances detection of variables with genuine incremental effects.
From a robustness perspective, the advantages of DPD-SISP are fully retained. 
Since each conditional marginal fit is based on the DPD loss, the procedure remains resistant to outliers,
heavy-tailed errors, and proxy misspecification. In fact, conditioning often stabilizes estimation further 
by reducing variance inflation caused by omitted relevant covariates.

Finally, the matrix $\bm{\Sigma}_{\mathcal C}^*$ can typically be expressed in a form analogous to that in \cite{ghosh2023robust}, 
involving conditional covariance operators after projection onto $\mathcal{C}$. 
Under standard regularity conditions (e.g., bounded eigenvalues and weak dependence), 
one obtains $\lambda_{\max}(\bm{\Sigma}_{\mathcal C}^*) = O(1)$, yielding
$$
\left|\widehat{\mathcal S}_\alpha^{\rm CSIS}\right| = O(n^{2\kappa})=o(n),
$$
which ensures feasibility for downstream high-dimensional inference procedures.
Therefore, the DPD-CSISP provides a principled, robust, and theoretically grounded extension of conditional screening 
to dependent and contaminated data settings, maintaining the optimal screening rates while improving stability and interpretability.

\section{The DPD-ISISP: Iterative DPD-SISP for Correlated FE Covariates}
\label{SEC:dpdISIS}

Although DPD-SISP enjoys the sure screening property under suitable regularity conditions, 
it remains a purely marginal procedure and therefore inherits the fundamental limitations of independence screening.
In the presence of correlated FE covariates, marginal utilities (here the MDPDEs of marginal slope parameters) 
may not faithfully reflect joint relevance.
In particular, an active FE can be masked when its marginal signal is weakened due to dependence with other variables, 
while an inactive variable may appear important through its correlation with truly active predictors. 
This phenomenon is also observed in our numerical studies (Section \ref{SEC:simulations}).
DPD-CSISP partially addresses this issue by evaluating marginal utilities conditional on a pre-specified subset of covariates. 
However, its effectiveness depends critically on the choice of the conditioning set. 
Because the active set is unknown, separating covariates into conditioning and non-conditioning groups 
according to their true importance is infeasible in practice. 
Moreover, DPD-CSISP is inherently static and does not adapt to information revealed during the screening process.
To overcome these limitations, we develop an adaptive and iterative extension of DPD-SISP that progressively refines the conditioning set. 
The resulting procedure, termed DPD-ISISP (Iterative DPD-based SIS with Proxy matrix), 
is inspired by the iterative SIS (ISIS) of \cite{fan2008sure}, but tailored to the present setting involving proxy-based whitening 
and robust MDPDE estimation under the feasible LRM \eqref{EQ:lm-feasible}.

We start with an initial conditioning set $\mathcal{C}$, 
which can be empty or determined based on prior scientific knowledge or domain considerations.
Using DPD-CSISP (or DPD-SISP if $\mathcal{C} = \emptyset$), we obtain an initial screening set $\widehat{\mathcal S}_\alpha^{(1)}$ 
with a target model size $d_1$ (or based on a threshold $\gamma_n$). 
Then, at iteration $t \ge 2$, we proceed as follows: 
we first fit the feasible LRM (\ref{EQ:lm-feasible}) restricted to the previously selected variables $\widehat{\mathcal S}_\alpha^{(t-1)}$ 
by computing the MDPDE of the regression coefficients  as
\begin{eqnarray}
\widehat{\bm{\beta}}_{\alpha, \widehat{\mathcal S}_\alpha^{(t-1)}}
= \arg\limits_{\bm{\beta}}\min_{(\bm{\beta}, \sigma^2)} \frac{1}{n}\sum_{i=1}^m \sum_{k=1}^{n_i} V_\alpha\!\left(
\widetilde{y}_{ik},\widetilde{\mathbf{x}}_{i,k, \widehat{\mathcal S}_\alpha^{(t-1)}}^\top \bm{\beta}, \sigma^2\right).
\label{EQ:ISIS-mdpde0}
\end{eqnarray}
We then compute the working residuals as 
\begin{eqnarray}
\widetilde{\bm{r}}_{i}^{(t)}= \widetilde{y}_{ik} -
\widetilde{\mathbf{x}}_{i,\widehat{\mathcal S}_\alpha^{(t-1)}}^\top \widehat{\bm{\beta}}_{\widehat{\mathcal S}_\alpha^{(t-1)},\alpha}.
\label{EQ:ISIS-residual}
\end{eqnarray}
%
Next, for each $j \notin \widehat{\mathcal S}_\alpha^{(t-1)}$, 
we fit the marginal feasible regression model
\[
\widetilde{\bm{r}}_{i}^{(t)}
=
\beta_{j0}^{(t)} + \widetilde{\mathbf{x}}_{ij} \beta_{j1}^{(t)} + \bm{\varepsilon}_{ij}^{(t)},
\]
under the same spherical working error assumption as in (\ref{EQ:lm-feas-marginal}), 
and compute the corresponding marginal MDPDE $\widehat{\beta}_{j1,\alpha}^{(t)}$.
A new batch of variables $\widehat{\mathcal A}_\alpha^{(t)}$ is selected based on the highest values of 
$|\widehat{\beta}_{j1,\alpha}^{(t)}|$, given a fixed target model size $d_t$ or a threshold $\gamma_n$.
The screened set is then updated as

\[
\widehat{\mathcal S}_\alpha^{(t)}
=
\widehat{\mathcal S}_\alpha^{(t-1)}
\cup
\widehat{\mathcal A}_\alpha^{(t)}.
\]
The procedure is iterated until one of the following stopping criteria is met:
(i) $|\widehat{\mathcal S}\alpha^{(t)}| \ge d_{\max}$,
(ii) no new variables are selected at the $t$-th iteration, i.e., $\widehat{\mathcal A}_\alpha^{(t)}=\emptyset$, or
(iii) a maximum number of iterations $T$ is reached.
The complete procedure is presented in Algorithm \ref{ALG:DPD-ISISP}.

\begin{minipage}{.95\textwidth}
	\begin{algorithm}[H]
		\caption{\small DPD-ISISP($\alpha$) for general LMMs}
		\label{ALG:DPD-ISISP}
		\small
		\begin{algorithmic}[1]
			
			\STATE \textbf{Input:} Clustered data $\{(\bm y_i,\mathbf X_i,\mathbf Z_i)\}_{i=1}^m$ from the LMM \eqref{EQ:lmm}; 
			tuning parameter $\alpha>0$; \\~~~~~~~~~~
			proxy matrix $\bm P$; conditioning set $\mathcal C$ (may be empty); maximum model size $d_{\max}$; \\~~~~~~~~~~
			batch sizes $\{d_t\}$ (or threshold $\gamma_n$); maximum iterations $T$.

			\STATE Set total sample size $n=\sum_{i=1}^m n_i$.\\~~~~~~~~~
			\COMMENT{\textit{Proxy-based covariance approximation}}
			
			\FOR{$i=1$ to $m$} 
			\STATE Compute proxy covariance $\widehat{\boldsymbol\Sigma}_i =\mathbf Z_i \bm P \mathbf Z_i^\top+\mathbb{I}_{n_i}$.
			\STATE Obtain transformed data $(\widetilde{\bm y}_i,\widetilde{\mathbf X}_i)$ using \eqref{EQ:proxy}.
			\ENDFOR\\~~~~~~~~~			\COMMENT{Initialization}
			
			\STATE Obtain initial screened set $\widehat{\mathcal S}_\alpha^{(1)}$ via DPD-CSISP (Algorith \ref{ALG:DPD-CSISP}) 
			with conditioning set $\mathcal{C}$, and target model size $d_1$ (or threshold $\gamma_n$)
			\\~~~~~~~~~			\COMMENT{Iterative screening}
						
			\STATE Set $t=2$. 
			\WHILE{$t \le T$ and $|\widehat{\mathcal S}_\alpha^{(t-1)}| < d_{\max}$}
			\STATE Compute the MDPDE $\widehat{\bm\beta}_{\widehat{\mathcal S}_\alpha^{(t-1)},\alpha}$ 
			using the feasible LRM restricted to  $\widehat{\mathcal S}_\alpha^{(t-1)}$ via (\ref{EQ:ISIS-mdpde0}).
			
			\FOR{$i=1$ to $m$} 
			\STATE Compute working residuals 
			$\widetilde{\bm{r}}_{i}^{(t)}= \left(\widetilde{r}_{i1}^{(t)}, \ldots, \widetilde{r}_{in_i}^{(t)}\right)^\top$ via (\ref{EQ:ISIS-residual}) for each $i$.
			\ENDFOR~\\~~~~~~~~~			\COMMENT{Conditional marginal screening}
			
			\FOR{each $j\notin\widehat{\mathcal S}_\alpha^{(t-1)}$} 			
			\STATE Define the DPD loss $V_\alpha(\cdot)$ as in \eqref{EQ:dpd-obj-norm} if the errors are Gaussian, 
			or as in (\ref{EQ:dpd-obj-gen}) otherwise.
			\STATE Compute the marginal MDPDE
			\[
			\widehat{\bm\theta}_{j,\alpha}^{(t)}	= \left(\widehat\beta_{j0,\alpha}^{(t)}, 
			\widehat\beta_{j1,\alpha}^{(t)}, \widehat{\bm{\eta}}_{j,\alpha}^{(t)}\right)^\top 
			= \argmin_{(\beta_0^M, \beta_1, \bm{\eta})^\top} \frac{1}{n}\sum_{i=1}^{m}\sum_{k=1}^{n_i} 
			V_\alpha\!\left(\widetilde{r}_{ik}^{(t)}, \beta_0 + \widetilde{x}_{ij,k}\beta_1, \bm{\eta}\right).
			\]
			\STATE Extract the marginal slope estimate $\widehat\beta_{j1,\alpha}^{(t)}$ from  $\widehat{\bm\theta}_{j,\alpha}^{(t)}$.
			\ENDFOR\\~~~~~~~~~

			\STATE Rank $\left\{|\widehat\beta_{j1,\alpha}^{(t)}|: j\notin\widehat{\mathcal S}_\alpha^{(t-1)} \right\}$ in decreasing order.
			
			\STATE Define 
			$\widehat{\mathcal A}_\alpha^{(t)}= \left\{ j\notin\widehat{\mathcal S}_\alpha^{(t-1)}: |\widehat\beta_{j1,\alpha}^{(t)}| 
			\text{ is among the top } d_t \right\}$ or  $\left\{ j\notin\widehat{\mathcal S}_\alpha^{(t-1)}: |\widehat\beta_{j1,\alpha}^{(t)}| \ge \gamma_n \right\}$.
			\\~~~~~~~~~			
			\COMMENT{Screening update}
			
			\STATE Update the screened set as   $\widehat{\mathcal S}_\alpha^{(t)} = \widehat{\mathcal S}_\alpha^{(t-1)}
			\cup \widehat{\mathcal A}_\alpha^{(t)}$.
			
			\STATE IF $\widehat{\mathcal A}_\alpha^{(t)}=\emptyset$, BREAK.

			\STATE Set $t \leftarrow t+1$.
			\ENDWHILE
			
			\STATE \textbf{Output:} Screened active set 
			$\widehat{\mathcal S}_\alpha^{\rm ISIS} = \widehat{\mathcal S}_\alpha^{(t)}$, the final screened set.
			
		\end{algorithmic}
	\end{algorithm}
\end{minipage}

\bigskip
By iteratively adjusting for previously selected variables, 
DPD-ISISP effectively approximates partial regression screening in a sequential manner. 
This substantially mitigates masking effects arising from correlated covariates 
and enhances the recovery of variables that are jointly important but marginally weak.
At the same time, robustness is preserved throughout the procedure: 
each stage relies on MDPDE-based estimation, ensuring bounded influence and stability under contamination or model misspecification. 
The proxy-based whitening continues to account for within-cluster dependence, 
while the iterative refinement reduces sensitivity to proxy-induced distortions in marginal rankings.
All these makes DPD-ISISP a natural and practically powerful extension of DPD-SISP 
in practical settings with strong correlation among FE covariates.

\section{Monte-Carlo  Illustrations: Comparative Performance of DPD-SISP}
\label{SEC:simulations}

\subsection{Simulation Setups}

We conduct a comprehensive simulation study to evaluate the finite-sample performance of the robust DPD-SISP procedures.
We generate 100 datasets from the LMM (\ref{EQ:lmm}), with $m=10$ clusters of sizes 10, 15, 20, 25, 30, 10, 15, 20, 25, 30, respectively, 
yielding a total sample size $n=200$; the errors and REs are generated from Gaussian distributions.
The model includes $p=1000$ FE covariates -- an intercept and $999$ additional FE covariates 
simulated from a multivariate normal distribution  with mean zero and covariance matrix $\bm{\Sigma}_X$. 
We consider two covariance structures for $\bm{\Sigma}_X$, 
namely the compound symmetry (CS) and AR(1) structure (denoted by T, for Toeplitz) with the same correlation parameter $\rho_x$;
all variances are set to 1, and both structures reduce to the identity matrix when $\rho_x=0$. 
For the RE design, we again consider two scenarios: 
\begin{itemize}
	\item[(R1)] The first $q=4$ FE covariates are used as RE covariates, which include the intercept at the first position.
\item[(R2)] Besides the RE intercept, $q-1=3$ RE covariates are generated from a standard multivariate normal distribution, 
independent of the FE covariates.
\end{itemize}
The covariate values are held fixed across simulation replications, 
while the REs and errors are independently generated in each dataset. 
The RE covariance matrix is specified as $\bm{\Psi} = (1-\rho_b)\mathbb{I}_q + \rho_b \mathbb{J}_q$ with $\rho_b=0.3$, 
and error variance $\sigma^2=1$.  Here, $\mathbb{J}_q$ denotes the $q \times q$ matrix whose entries are all equal to 1.
The true regression coefficient vector $\bm{\beta}_0\in\mathbb{R}^p$ is sparse with 6 nonzero entries (including intercept term),
so that we need to screen $s=5$ truly active non-constant FE covariates out of 999 such variables.  
All nonzero entries of $\bm{\beta}_0$ are set to 1, while their locations are varied as per the following two scenarios.
\begin{itemize}
	\item[(S1)] Dense signals: The first six coefficients of $\bm{\beta}_0$  are nonzero. 
	Thus, the first five non-constant FE covariates (excluding the intercept) are truly active, 
	and all REs correspond to active FEs under (R1).

	\item[(S2)] Sparse signals: The nonzero coefficients occur at positions 1 (intercept), 5, 10, 50, 100, and 200. 
	In this case, the active covariates are more sparsely distributed, and none of the REs (except the intercept) 
	are associated with active FEs even under (R1).
\end{itemize}

To assess the claimed robustness performance,  we contaminate a fixed proportion (5\%, 10\% and 20\%) of randomly selected observations 
in each simulated dataset  under four different schemes as described below:
\begin{itemize}
	\item[(C1)] Contaminate responses by adding independent noise from a $\mathcal{N}(20,1)$ distribution.

	\item[(C2)] Contaminate the first non-constant FE covariate (column 2 of the FE design matrix) 
	by adding independent noise from $\mathcal{N}(-2,1)$. 
	This contaminated covariate is active in (S1) but inactive in (S2),  and it is also associated with an (uncontaminated) RE under (R1).
	
	\item[(C3)] Replace the fifth non-constant FE covariate (column 6 of the FE design matrix)  with randomly generated leverage points, 
	which also changes the associated response values.  
	This contaminating covariate is inactive in both (S1) and (S2), and has no associated REs even under (R1).
	
	\item[(C4)] Replace the first non-constant RE covariate (column 2 of the RE design matrix)  with leverage points, 
	inducing contamination in the response. But, all FE covariates are left uncontaminated which yields a hypothetical scenarios under (R1). 
\end{itemize}

We compare the proposed DPD based robust screening approaches with benchmark methods constructed directly 
under the mixed model structure using classical ML and REML-based marginal utilities, and
a relatively fast robust alternative based on MHGDE. 

Their performances are evaluated using the median true positive rate (TPR), 
defined as the median proportion of true signals correctly identified, 
and the empirical sure screening probability (EmpSSP), defined as the proportion of replications 
in which all true signals are successfully selected (i.e., TPR equals 1). 
These metrics are computed using a fixed target model size $d=[n/\log n]=86$ across 100 replications.
To further assess screening efficiency, we examine the minimum model size (MinMS), 
defined as the smallest value of target model size $d$ required to achieve perfect recovery (TPR = 1). 
This metric provides a finer comparison among competing methods that may exhibit similar median TPR but differ in their efficiency.
In addition, computational scalability is evaluated by recording the runtime of each procedure. 
The distributions of both MinMS and runtime are summarized through boxplots, with corresponding mean values reported for completeness.

\subsection{Simulation Results}
\label{SEC:sim_DPDSISP}

For all simulation configurations, arising from various combinations of (R1)--(R2), (S1)--(S2), and contamination schemes (C0)--(C4), 
where (C0) denotes the uncontaminated case, we report the median TPR and EmpSSP in Tables \ref{TAB:C0}--\ref{TAB:C1} and Supplementary Tables S1--S3. 
Representative boxplots of MinMS are presented in Figures \ref{FIG:MMS_C0}--\ref{FIG:MMS_Csp3} for selected scenarios, 
with the remaining ones are shown in Supplementary Figures S1--S5.

\begin{table}[!h]
\centering
	\caption{Median TPR and EmpSSp (in parenthesis) under no contamination (C0)}
	\resizebox{\textwidth}{!}{
\begin{tabular}{l|cc|ccc|ccc|ccc}\hline
Settings	&		\multicolumn{5}{|c|}{Benchmark SIS}	&	\multicolumn{3}{c|}{DPD-SISP with cv-P}	&	\multicolumn{3}{c}{DPD-SISP with I0-P}\\	
$\bm\Sigma_x$	&	MLE	&	REML	&	HGD(0.1) 	&	HGD(0.3)	&	HGD(0.5)	&	($0.1$)	&	($0.3$)	&	($0.5$)	&	($0.1$)	&	($0.3$)	&	($0.5$)	\\	\hline
\hline
(S1)$\times$(R0)	&		1	&	1	&	1	&	1	&	1	&	1	&	1	&	1	&	1	&	1	&	1	\\	
Identity	&	(1)	&	(1)	&	(1)	&	(1)	&	(1)	&	(0.98)	&	(0.97)	&	(0.96)	&	(0.97)	&	(0.97)	&	(0.97)	\\	\hline
(S1)$\times$(R1)	&		1	&	1	&	1	&	1	&	1	&	1	&	1	&	1	&	0.6	&	0.8	&	0.8	\\	
Identity		&	(0.56)	&	(0.55)	&	(0.56)	&	(0.59)	&	(0.6)	&	(0.87)	&	(0.79)	&	(0.71)	&	(0.03)	&	(0.16)	&	(0.29)	\\	\hline
(S2)$\times$(R1)	&		1	&	1	&	1	&	1	&	1	&	1	&	1	&	1	&	1	&	1	&	1	\\	
Identity	&	(1)	&	(1)	&	(1)	&	(1)	&	(1)	&	(0.99)	&	(0.97)	&	(0.97)	&	(0.99)	&	(0.97)	&	(0.95)	\\	\hline
\hline
(S1)$\times$(R1)	&	0.8	&	0.8	&	0.8	&	0.8	&	0.8	&	1	&	0.8	&	0.8	&	0.4	&	0.6	&	0.6	\\	
CS(0.3)		&	(0.4)	&	(0.4)	&	(0.42)	&	(0.42)	&	(0.43)	&	(0.52)	&	(0.48)	&	(0.4)	&	(0)	&	(0.06)	&	(0.14)	\\	\hline
(S2)$\times$(R0)	&		1	&	1	&	1	&	1	&	1	&	1	&	1	&	1	&	1	&	1	&	1	\\	
CS(0.3)	&	(1)	&	(1)	&	(1)	&	(0.99)	&	(1)	&	(0.89)	&	(0.88)	&	(0.86)	&	(0.88)	&	(0.89)	&	(0.88)	\\	\hline
(S2)$\times$(R1)	&		1	&	1	&	1	&	1	&	1	&	1	&	1	&	1	&	1	&	1	&	1	\\	
CS(0.3)		&	(1)	&	(1)	&	(1)	&	(1)	&	(1)	&	(0.84)	&	(0.87)	&	(0.84)	&	(0.86)	&	(0.86)	&	(0.84)	\\	\hline\hline
(S1)$\times$(R1)	&		1	&	1	&	1	&	1	&	1	&	1	&	1	&	1	&	0.6	&	0.6	&	0.8	\\	
T(0.3)	&	(0.61)	&	(0.62)	&	(0.58)	&	(0.6)	&	(0.6)	&	(0.99)	&	(0.98)	&	(0.97)	&	(0.03)	&	(0.18)	&	(0.22)	\\	\hline
(S2)$\times$(R0)	&		1	&	1	&	1	&	1	&	1	&	1	&	1	&	1	&	1	&	1	&	1	\\	
T(0.3)	&	(1)	&	(1)	&	(1)	&	(1)	&	(1)	&	(0.99)	&	(0.98)	&	(0.95)	&	(0.99)	&	(0.99)	&	(0.96)	\\	\hline
(S2)$\times$(R1)	&		1	&	1	&	1	&	1	&	1	&	1	&	1	&	1	&	1	&	1	&	1	\\	
T(0.3)	&	(1)	&	(1)	&	(1)	&	(1)	&	(1)	&	(0.99)	&	(0.97)	&	(0.95)	&	(0.99)	&	(0.99)	&	(0.95)	\\	\hline
\end{tabular}
}
\label{TAB:C0}
\end{table}

\begin{figure}[!h]
	\centering
	\subfloat[\tiny (S1)$\times$(R0) with $\bm\Sigma_x$ as Identity]{
		\includegraphics[page=2, width=0.3\textwidth]{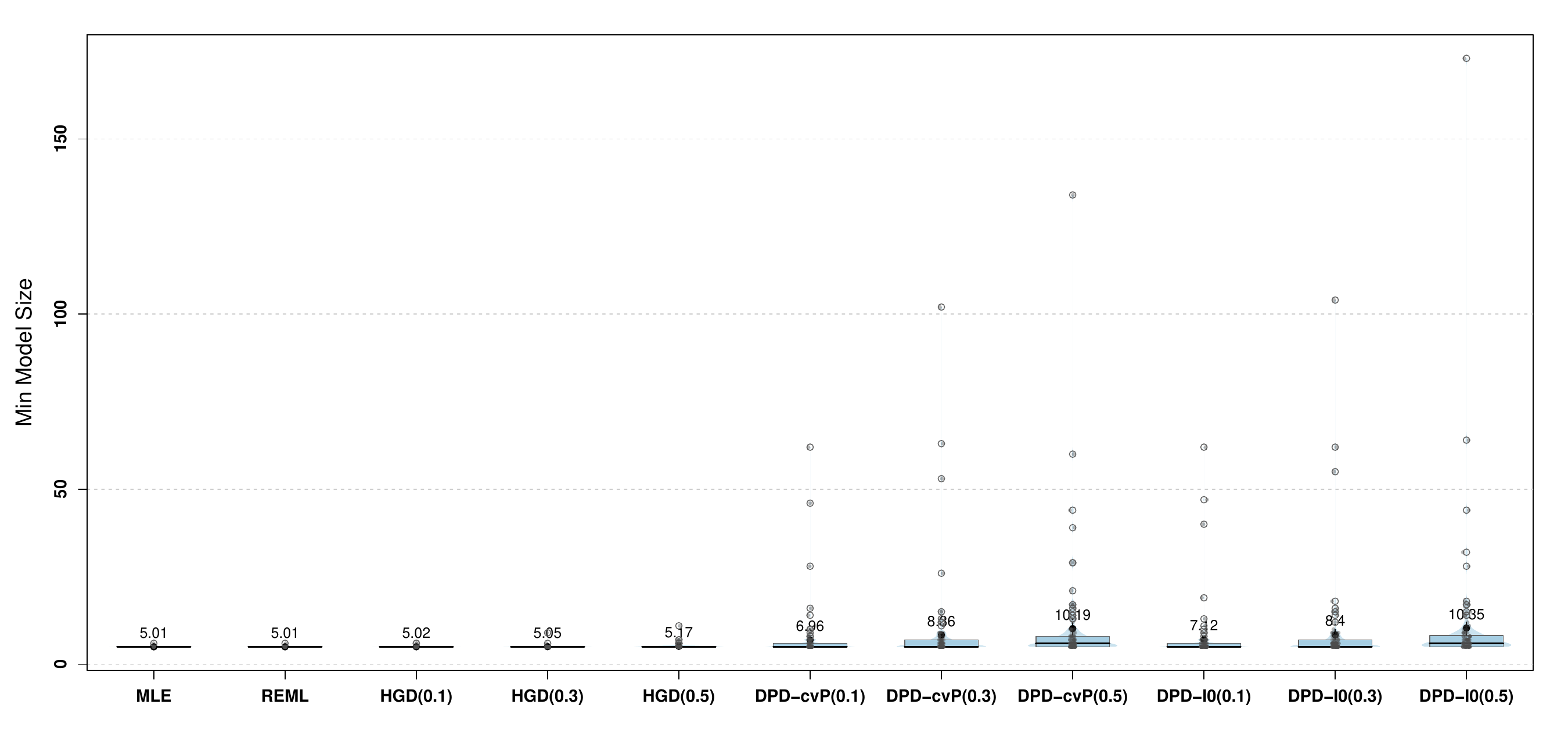}
		\label{FIG:boxplot_Y}}
	~	
	\subfloat[\tiny (S1)$\times$(R1) with $\bm\Sigma_x$ as Identity]{
		\includegraphics[page=2, width=0.3\textwidth]{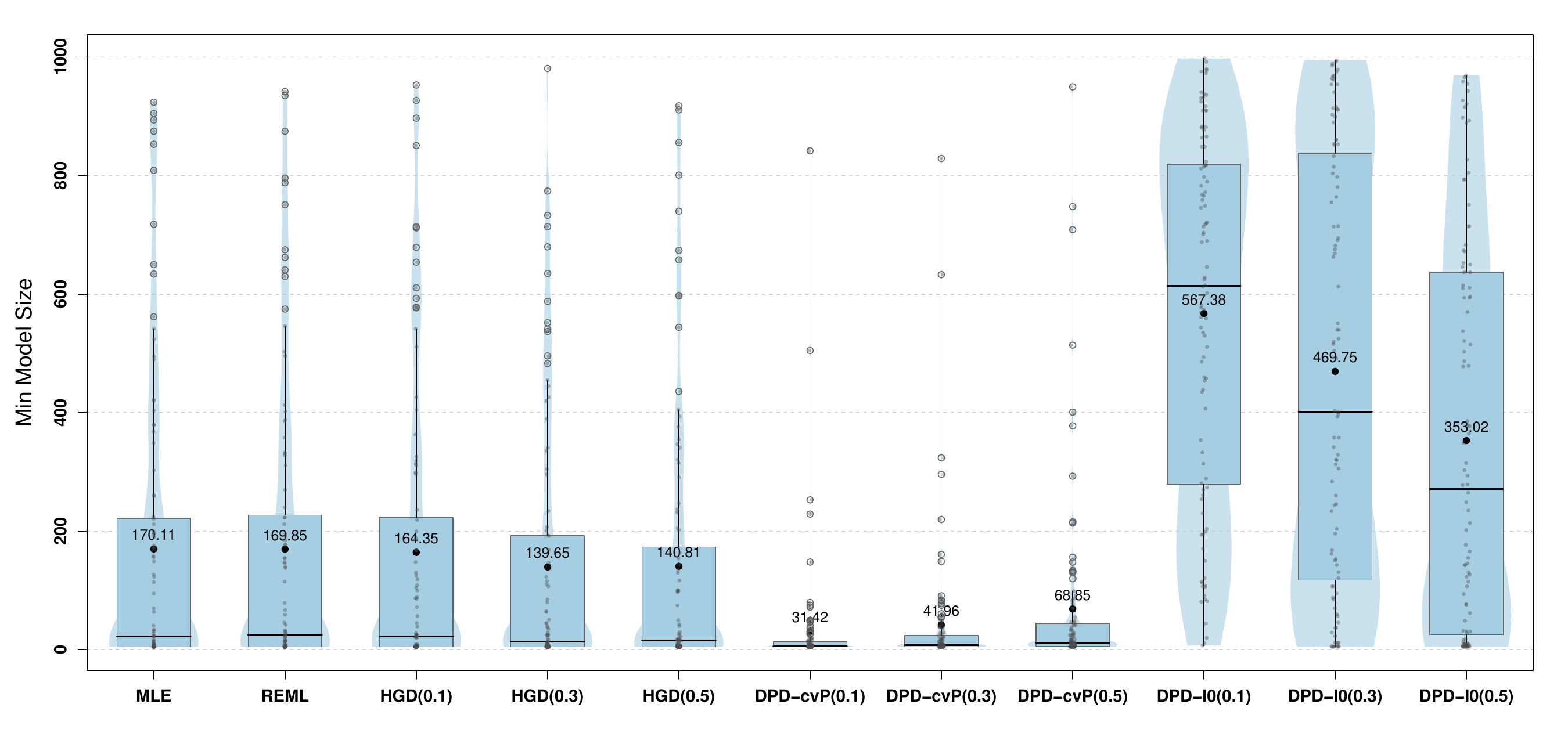}
		\label{FIG:boxplot_Y}}
	~	
	\subfloat[\tiny (S2)$\times$(R1) with $\bm\Sigma_x$ as Identity]{
		\includegraphics[page=2, width=0.3\textwidth]{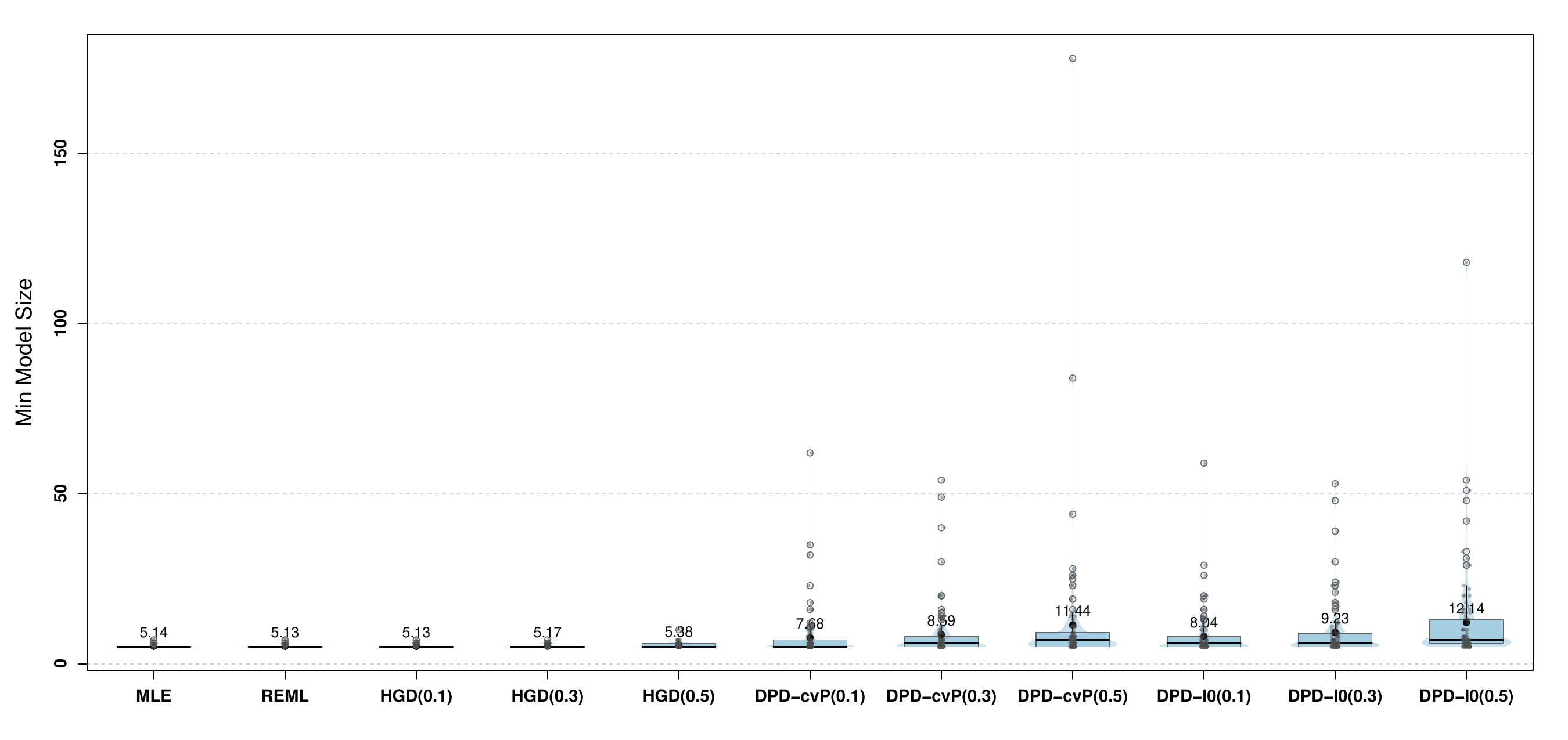}
		\label{FIG:boxplot_Y}}
	\\	
	\subfloat[\tiny (S1)$\times$(R1) with $\bm\Sigma_x$ as CS(0.3)]{
		\includegraphics[page=2, width=0.3\textwidth]{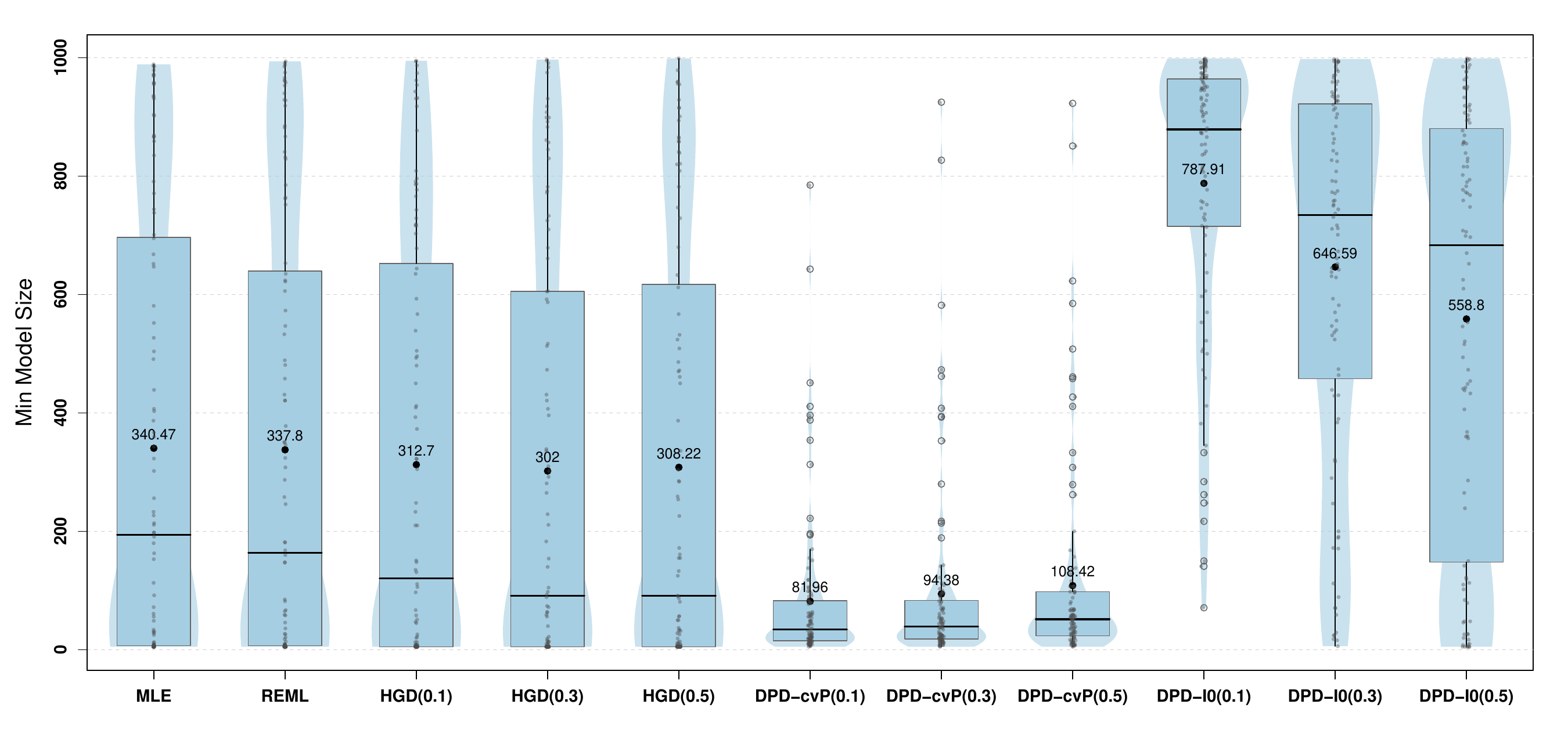}
		\label{FIG:boxplot_Y}}
	~	
	\subfloat[\tiny (S2)$\times$(R0) with $\bm\Sigma_x$ as CS(0.3)]{
		\includegraphics[page=2, width=0.3\textwidth]{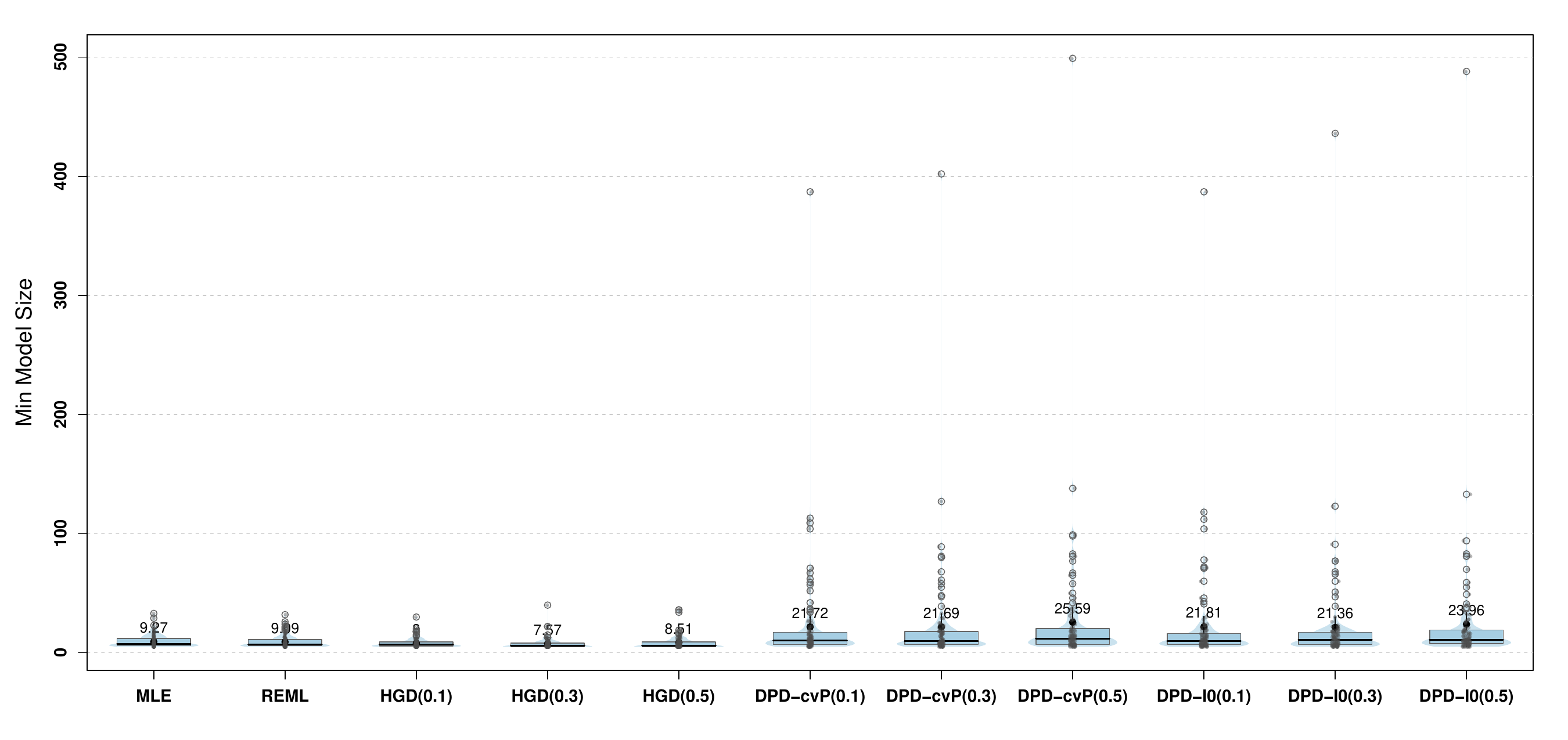}
		\label{FIG:boxplot_Y}}
	~	
	\subfloat[\tiny (S2)$\times$(R1) with $\bm\Sigma_x$ as CS(0.3)]{
		\includegraphics[page=2, width=0.3\textwidth]{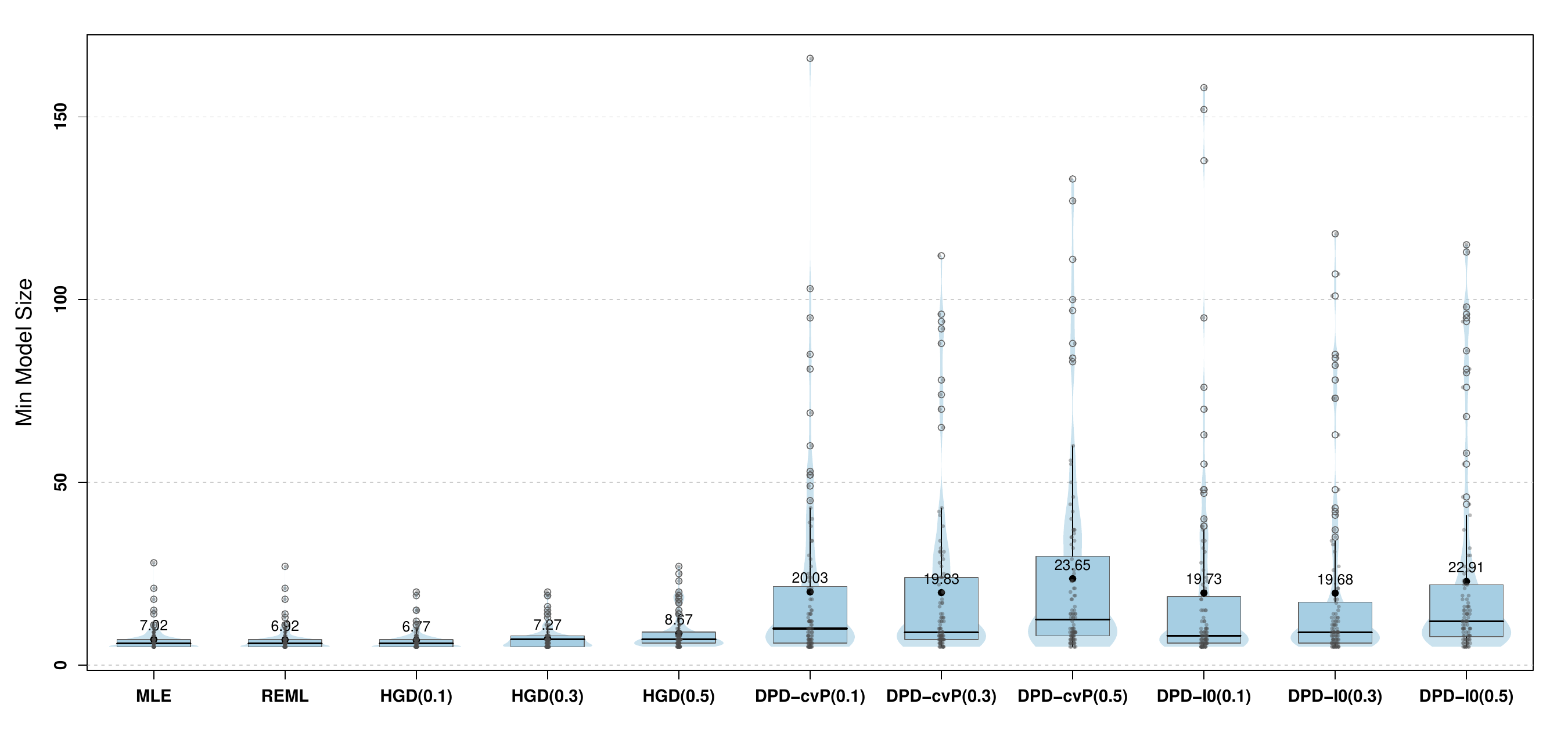}
		\label{FIG:boxplot_Y}}
	\\	
	\subfloat[\tiny (S1)$\times$(R1) with $\bm\Sigma_x$ as T(0.3)]{
		\includegraphics[page=2,width=0.3\textwidth]{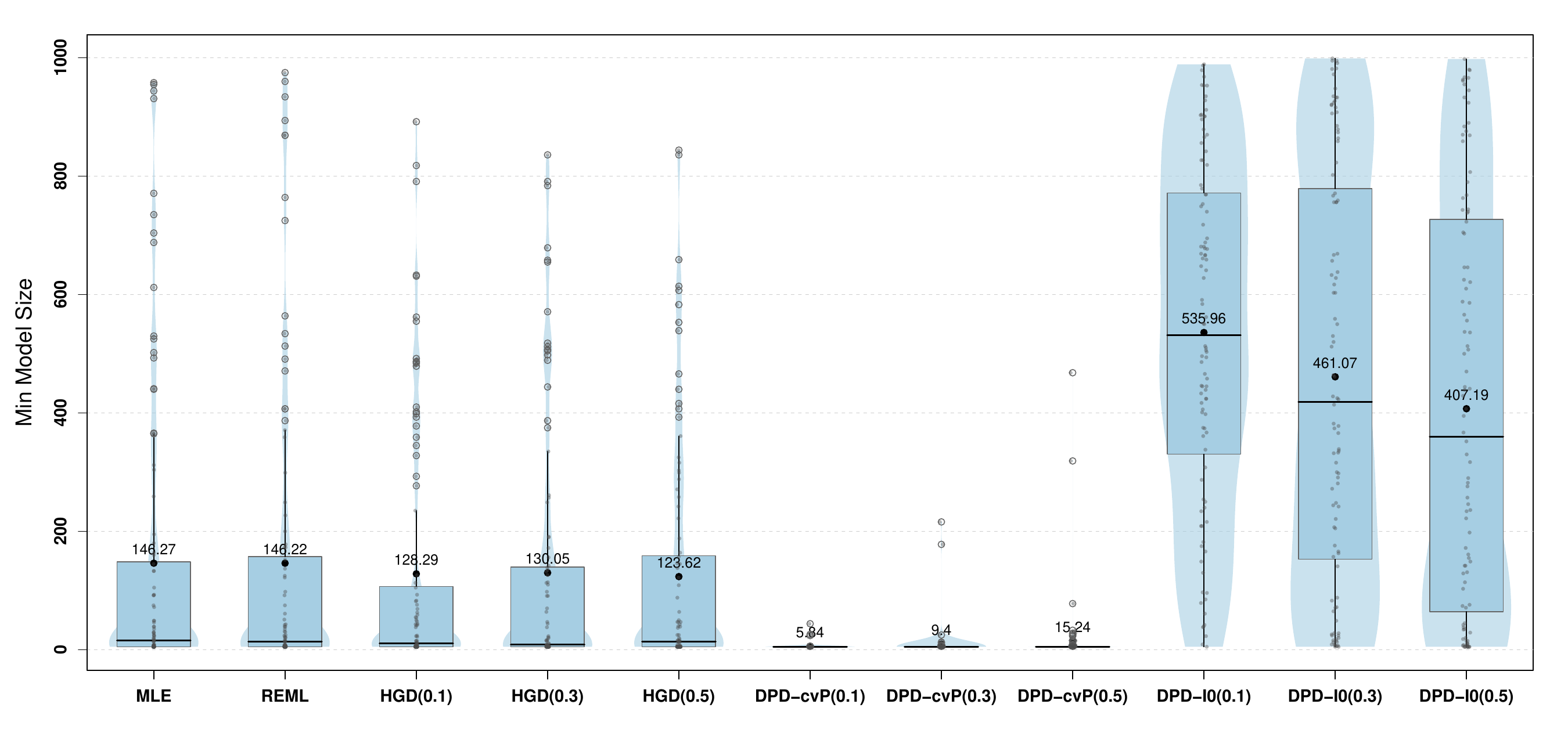}
		\label{FIG:boxplot_Y}}
	~	
	\subfloat[\tiny (S2)$\times$(R0) with $\bm\Sigma_x$ as T(0.3)]{
		\includegraphics[page=2, width=0.3\textwidth]{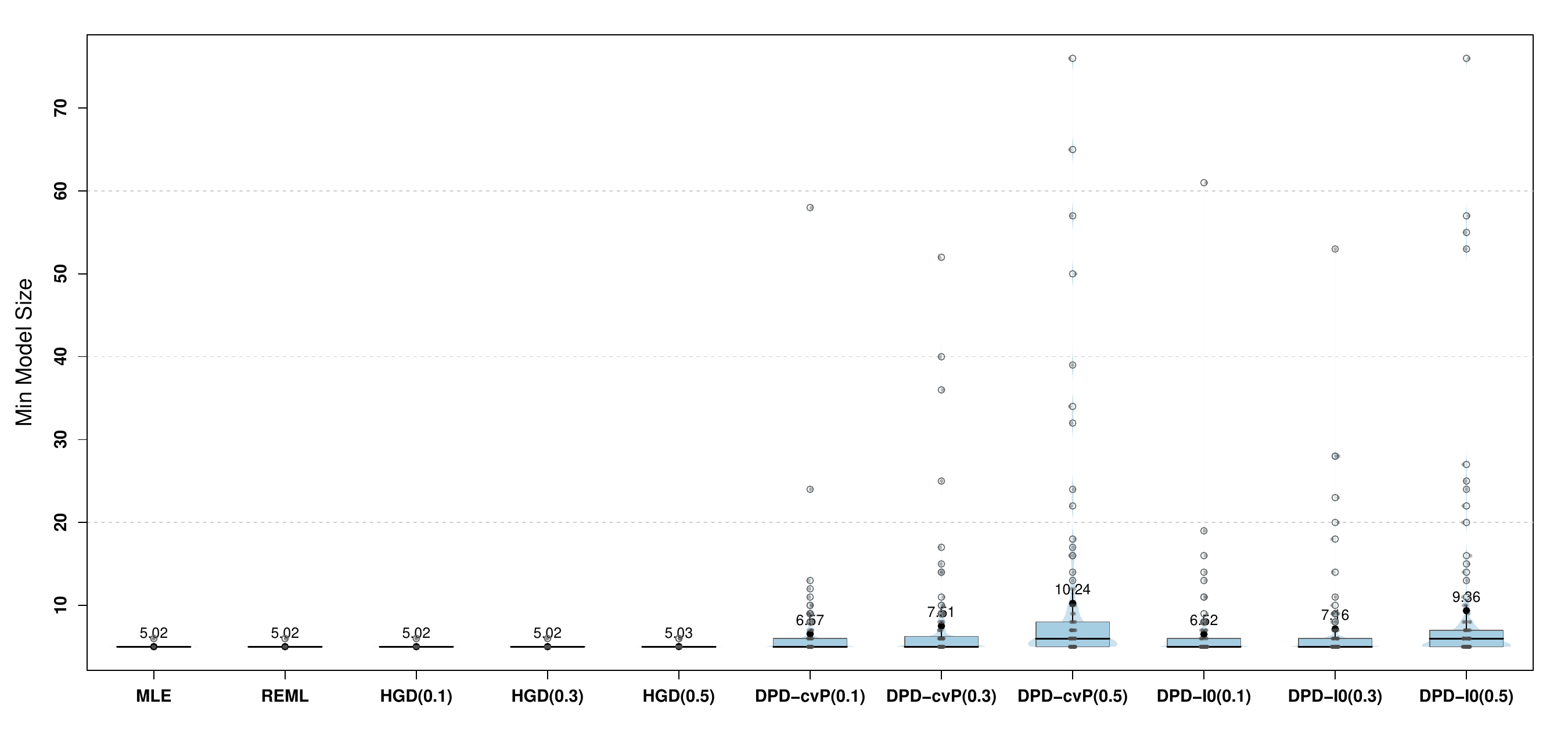}
		\label{FIG:boxplot_Y}}
	~	
	\subfloat[\tiny (S2)$\times$(R1) with $\bm\Sigma_x$ as T(0.3)]{
		\includegraphics[page=2, width=0.3\textwidth]{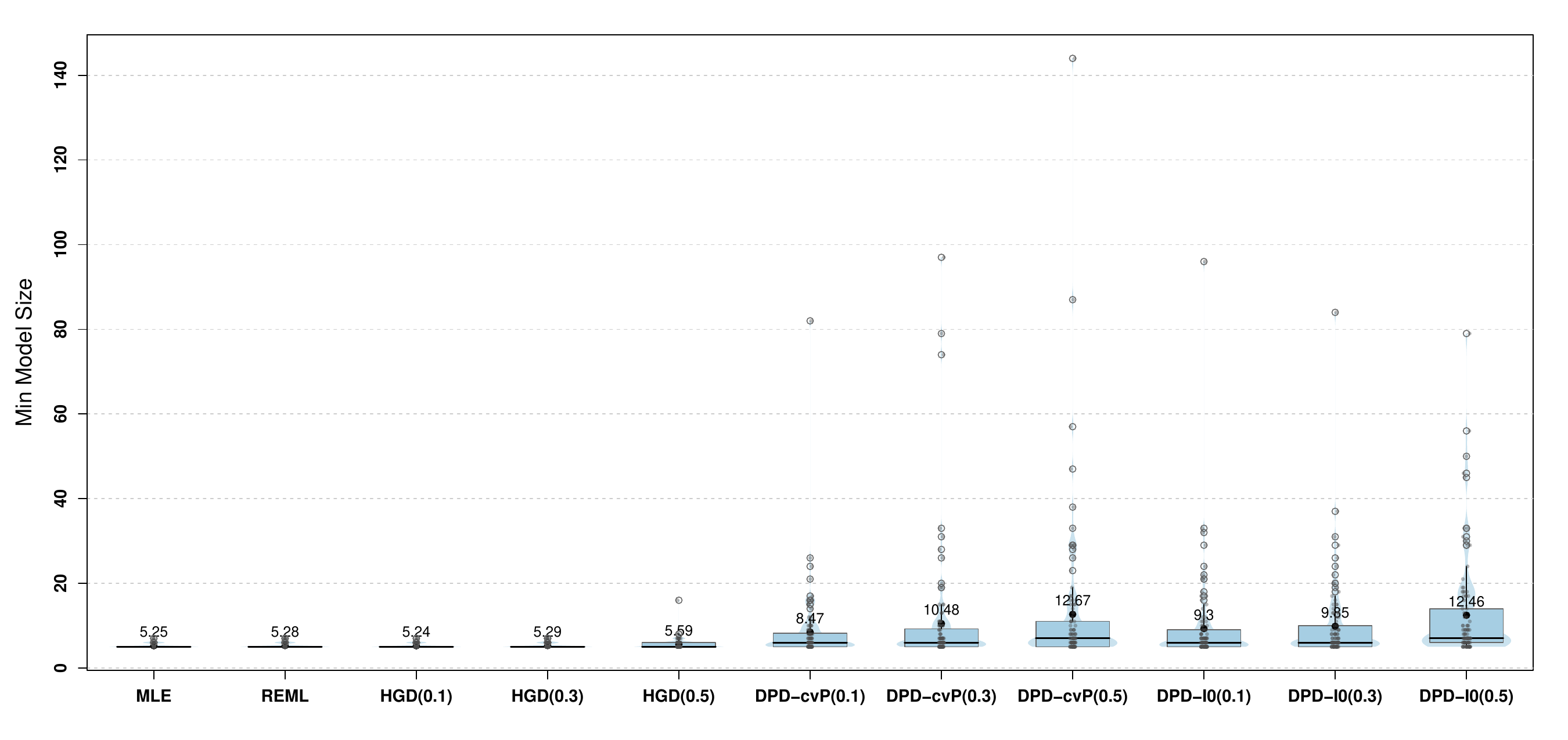}
		\label{FIG:boxplot_Y}}
	\caption{Boxplots (with overlaid sample means) of MinMS required for sure screening under no contamination (C0).
	The boxes represent (from left to right) the benchmark SIS based on MLE, REML, and HGD with $\gamma=0.1, 0.3, 0,5$, 
	the DPD-SISP with cv-P with $\alpha=0.1, 0.3, 0.5$, and the DPD-SISP with I0-P with $\alpha=0.1, 0.3, 0.5$, respectively. }
	\label{FIG:MMS_C0}
\end{figure}

\bigskip
\noindent\textbf{Performance under Uncontaminated data:}
\\
Under no contamination (C0), Table \ref{TAB:C0} shows that DPD-SISP is highly competitive with benchmark SIS procedures. 
Across nearly all configurations, the median TPR is equal to 1 for both the proposed and benchmark methods, 
with only a few minor deviations (e.g., under correlated designs with overlapping FE and RE covariates). 
Thus, detection power is broadly comparable across methods.
More informative differences emerge in EmpSSP, which captures the proportion of exact recovery. 
In relatively simple settings, such as independent covariates or weak interaction between FE and RE covariates,
all methods achieve near-perfect EmpSSP. 
DPD-SISP performs similarly here, with only minor reductions due to the additional robustness weighting.
This fact is also evident in Figure \ref{FIG:MMS_C0} which shows only a slight increase in minimum model sizes (MinMS)
required to achieve sure screening for DPD-based procedures compared to classical screening methods. 

In  more challenging scenarios involving correlated covariates or overlap between FE and RE covariates 
(e.g., (S1)$\times$(R1) under Identity, CS(0.3), and T(0.3) covariance structures), 
benchmark SIS procedures show noticeable reductions in EmpSSP, 
reflecting increased variability in recovering the full active set. 
In contrast, DPD-SISP with the cv-P proxy generally attains higher EmpSSP in these scenarios, 
often substantially improving over the benchmark methods, although some degradation persists under stronger dependence. 
The advantage is particularly pronounced under T(0.3), and remains moderate but consistent under CS(0.3).
The minMS required for DPD-SISP with cv-P proxy is also significantly smaller than the requirement of 
the benchmark ML, REML or MHGD based screening procedures in all these cases (Figure \ref{FIG:MMS_C0}).
Thus,  DPD-based marginal utilities yield more reliable rankings in structurally complex settings.

\begin{table}[!h]
	\centering
	\caption{Median TPR and EmpSSp (in parenthesis) under contaminated data with (C1)}
	\resizebox{\textwidth}{!}{
		\begin{tabular}{l|cc|ccc|ccc|ccc}\hline
			Settings	&		\multicolumn{5}{|c|}{Benchmark SIS}	&	\multicolumn{3}{c|}{DPD-SISP with cv-P}	&	\multicolumn{3}{c}{DPD-SISP with I0-P}\\	
			$\bm\Sigma_x$	&	MLE	&	REML	&	HGD(0.1) 	&	HGD(0.3)	&	HGD(0.5)	&	($0.1$)	&	($0.3$)	&	($0.5$)	&	($0.1$)	&	($0.3$)	&	($0.5$)	\\	\hline
			\hline
\multicolumn{12}{c}{\textbf{5\% contamination}} \\																										
(S1)$\times$(R1)			&	0.6	&	0.4	&	1	&	1	&	1	&	1	&	1	&	1	&	0.8	&	0.8	&	0.8	\\	
Identity	&	(0.01)	&	(0)	&	(0.56)	&	(0.59)	&	(0.58)	&	(0.68)	&	(0.8)	&	(0.78)	&	(0.31)	&	(0.41)	&	(0.38)	\\	\hline
(S2)$\times$(R1)			&	0.4	&	0.4	&	1	&	1	&	1	&	0.8	&	1	&	1	&	0.8	&	1	&	1	\\	
Identity	&	(0)	&	(0)	&	(0.96)	&	(1)	&	(1)	&	(0.21)	&	(0.92)	&	(0.92)	&	(0.21)	&	(0.94)	&	(0.89)	\\	\hline\hline
(S1)$\times$(R1)			&	0.6	&	0.6	&	0.8	&	1	&	1	&	0.6	&	0.8	&	0.8	&	0.6	&	0.8	&	0.8	\\	
CS(0.3)	&	(0)	&	(0)	&	(0.47)	&	(0.54)	&	(0.53)	&	(0.05)	&	(0.38)	&	(0.35)	&	(0.13)	&	(0.27)	&	(0.22)	\\	\hline
(S2)$\times$(R0)			&	0.8	&	0.8	&	0.8	&	1	&	1	&	0.8	&	1	&	1	&	0.8	&	1	&	1	\\	
CS(0.3)	&	(0)	&	(0)	&	(0.45)	&	(0.96)	&	(0.98)	&	(0.33)	&	(0.83)	&	(0.81)	&	(0.27)	&	(0.76)	&	(0.77)	\\	\hline
(S2)$\times$(R1)			&	0.4	&	0.4	&	0.8	&	1	&	1	&	0.6	&	1	&	1	&	0.6	&	1	&	1	\\	
CS(0.3)	&	(0)	&	(0)	&	(0.13)	&	(0.92)	&	(0.88)	&	(0.03)	&	(0.67)	&	(0.74)	&	(0.04)	&	(0.78)	&	(0.72)	\\	\hline\hline
(S1)$\times$(R1)			&	0.8	&	0.6	&	1	&	1	&	1	&	1	&	1	&	1	&	0.6	&	0.8	&	0.8	\\	
T(0.3)	&	(0.09)	&	(0.08)	&	(0.62)	&	(0.64)	&	(0.68)	&	(0.93)	&	(0.97)	&	(0.98)	&	(0.21)	&	(0.26)	&	(0.24)
\\\hline	
(S2)$\times$(R0)			&	0.8	&	0.8	&	1	&	1	&	1	&	1	&	1	&	1	&	1	&	1	&	1	\\	
T(0.3)	&	(0)	&	(0.01)	&	(1)	&	(1)	&	(1)	&	(0.9)	&	(0.99)	&	(0.98)	&	(0.87)	&	(0.99)	&	(0.99)	\\	\hline
(S2)$\times$(R1)			&	0.4	&	0.4	&	1	&	1	&	1	&	0.8	&	1	&	1	&	0.8	&	1	&	1	\\	
T(0.3)	&	(0)	&	(0)	&	(0.98)	&	(1)	&	(1)	&	(0.18)	&	(0.89)	&	(0.84)	&	(0.22)	&	(0.87)	&	(0.8)	\\	\hline\hline
\multicolumn{12}{c}{\textbf{10\% contamination}} \\																								
(S1)$\times$(R1)			&	0.4	&	0.4	&	0.6	&	1	&	1	&	0.6	&	1	&	1	&	0.4	&	1	&	1	\\	
Identity	&	(0)	&	(0)	&	(0)	&	(0.59)	&	(0.62)	&	(0)	&	(0.75)	&	(0.71)	&	(0)	&	(0.56)	&	(0.58)	\\	\hline
(S2)$\times$(R1)			&	0.2	&	0.2	&	0.4	&	1	&	1	&	0.4	&	1	&	1	&	0.4	&	1	&	1	\\	
Identity	&	(0)	&	(0)	&	(0)	&	(1)	&	(1)	&	(0)	&	(0.9)	&	(0.92)	&	(0)	&	(0.83)	&	(0.87)	\\	\hline\hline
(S1)$\times$(R1)			&	0.4	&	0.4	&	0.6	&	1	&	1	&	0.4	&	0.8	&	0.8	&	0.2	&	0.8	&	0.8	\\
CS(0.3)	&	(0)	&	(0)	&	(0)	&	(0.59)	&	(0.64)	&	(0)	&	(0.24)	&	(0.29)	&	(0)	&	(0.49)	&	(0.46)	\\	\hline
(S2)$\times$(R0)			&	0.6	&	0.6	&	0.6	&	1	&	1	&	0.7	&	1	&	1	&	0.8	&	1	&	1	\\	
CS(0.3)	&	(0)	&	(0)	&	(0)	&	(0.99)	&	(0.98)	&	(0)	&	(0.85)	&	(0.84)	&	(0)	&	(0.65)	&	(0.73)	\\	\hline
(S2)$\times$(R1)			&	0.2	&	0.2	&	0.4	&	1	&	1	&	0.4	&	1	&	1	&	0.4	&	1	&	1	\\	
CS(0.3)	&	(0)	&	(0)	&	(0)	&	(0.91)	&	(0.85)	&	(0)	&	(0.59)	&	(0.69)	&	(0)	&	(0.62)	&	(0.67)	\\	\hline\hline
(S1)$\times$(R1)			&	0.8	&	0.8	&	0.8	&	1	&	1	&	0.8	&	1	&	1	&	0.4	&	1	&	1	\\	
T(0.3)	&	(0)	&	(0)	&	(0.29)	&	(0.69)	&	(0.74)	&	(0.18)	&	(0.96)	&	(0.95)	&	(0)	&	(0.64)	&	(0.6)	\\	\hline
(S2)$\times$(R0)			&	0.6	&	0.6	&	0.6	&	1	&	1	&	0.7	&	1	&	1	&	0.8	&	1	&	1	\\	
T(0.3)	&	(0)	&	(0)	&	(0)	&	(1)	&	(1)	&	(0.01)	&	(0.97)	&	(0.96)	&	(0.01)	&	(0.93)	&	(0.9)	\\	\hline
(S2)$\times$(R1)			&	0.2	&	0.2	&	0.4	&	1	&	1	&	0.4	&	1	&	1	&	0.4	&	1	&	1	\\	
T(0.3)	&	(0)	&	(0)	&	(0)	&	(1)	&	(1)	&	(0)	&	(0.9)	&	(0.89)	&	(0)	&	(0.82)	&	(0.84)	\\	\hline\hline
\multicolumn{12}{c}{\textbf{20\% contamination}} \\																									
(S1)$\times$(R1)			&	0.2	&	0.2	&	0.2	&	0.6	&	1	&	0	&	0.8	&	1	&	0.2	&	0.4	&	1	\\	
Identity	&	(0)	&	(0)	&	(0)	&	(0.06)	&	(0.61)	&	(0)	&	(0.34)	&	(0.71)	&	(0)	&	(0.12)	&	(0.54)	\\	\hline
(S2)$\times$(R1)			&	0	&	0	&	0	&	0.2	&	1	&	0	&	1	&	1	&	0	&	0.2	&	0.8	\\	
Identity	&	(0)	&	(0)	&	(0)	&	(0.01)	&	(0.99)	&	(0)	&	(0.51)	&	(0.91)	&	(0)	&	(0.01)	&	(0.44)	\\	\hline\hline
(S2)$\times$(R0)			&	0.4	&	0.4	&	0.4	&	0.4	&	1	&	0.4	&	0.6	&	1	&	0.4	&	0.6	&	0.8	\\	
CS(0.3)	&	(0)	&	(0)	&	(0)	&	(0)	&	(0.97)	&	(0)	&	(0)	&	(0.69)	&	(0)	&	(0)	&	(0.32)	\\	\hline
(S1)$\times$(R1)			&	0.4	&	0.4	&	0.4	&	0.5	&	1	&	0	&	0.2	&	0.8	&	0.2	&	0.4	&	0.8	\\	
CS(0.3)	&	(0)	&	(0)	&	(0)	&	(0.03)	&	(0.72)	&	(0)	&	(0)	&	(0.14)	&	(0)	&	(0.02)	&	(0.32)	\\	\hline
(S2)$\times$(R1)			&	0	&	0	&	0	&	0.2	&	0.8	&	0	&	0.2	&	0.8	&	0	&	0.2	&	0.6	\\	
CS(0.3)	&	(0)	&	(0)	&	(0)	&	(0)	&	(0.43)	&	(0)	&	(0.01)	&	(0.35)	&	(0)	&	(0)	&	(0.2)	\\	\hline\hline
(S1)$\times$(R1)			&	0.4	&	0.4	&	0.4	&	0.8	&	1	&	0.4	&	0.8	&	1	&	0.2	&	0.6	&	1	\\	
T(0.3)	&	(0)	&	(0.02)	&	(0.01)	&	(0.34)	&	(0.67)	&	(0.02)	&	(0.4)	&	(0.95)	&	(0)	&	(0.12)	&	(0.62)	\\	\hline
(S2)$\times$(R0)			&	0.4	&	0.4	&	0.4	&	0.8	&	1	&	0.4	&	1	&	1	&	0.4	&	0.6	&	0.9	\\	
T(0.3)	&	(0)	&	(0)	&	(0)	&	(0.31)	&	(1)	&	(0)	&	(0.61)	&	(0.82)	&	(0)	&	(0.08)	&	(0.5)	\\	\hline
(S2)$\times$(R1)			&	0.2	&	0.2	&	0.2	&	0.2	&	1	&	0.2	&	1	&	1	&	0.2	&	0.2	&	0.8	\\	
T(0.3)	&	(0)	&	(0)	&	(0)	&	(0.04)	&	(0.96)	&	(0)	&	(0.53)	&	(0.84)	&	(0)	&	(0.02)	&	(0.3)	\\	\hline
\hline
		\end{tabular}
	}
	\label{TAB:C1}
\end{table}

The faster I0-P proxy exhibits a different pattern.
While it performs comparably to cv-P in settings with weak dependence or minimal interaction between FE and RE covariates, 
its EmpSSP deteriorates substantially in the presence of strong correlation or overlap among variables. 
For example, in (S1)$\times$(R1) configurations, I0-P maintains reasonable TPR but yields very low EmpSSP
and high MinMS, indicating unstable rankings. 
This highlights the importance of an appropriate proxy choice when dependence structures are nontrivial.


\begin{figure}[!h]
	\centering
	\subfloat[\tiny 5\% (C1) Contamination]{
		\includegraphics[page=2, width=0.3\textwidth]{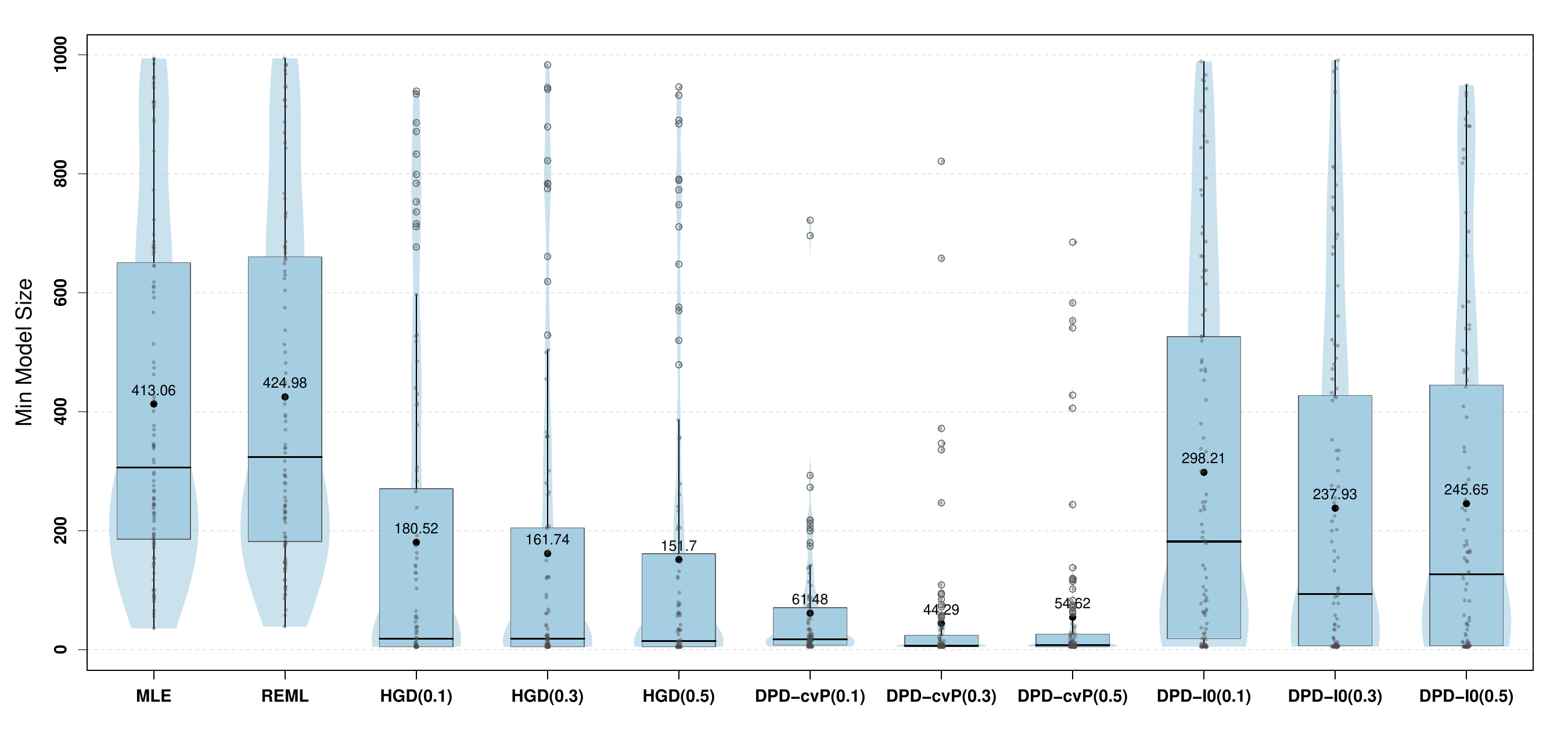}
		\label{FIG:boxplot_Y}}
	~	
	\subfloat[\tiny 10\% (C1) Contamination]{
		\includegraphics[page=2, width=0.3\textwidth]{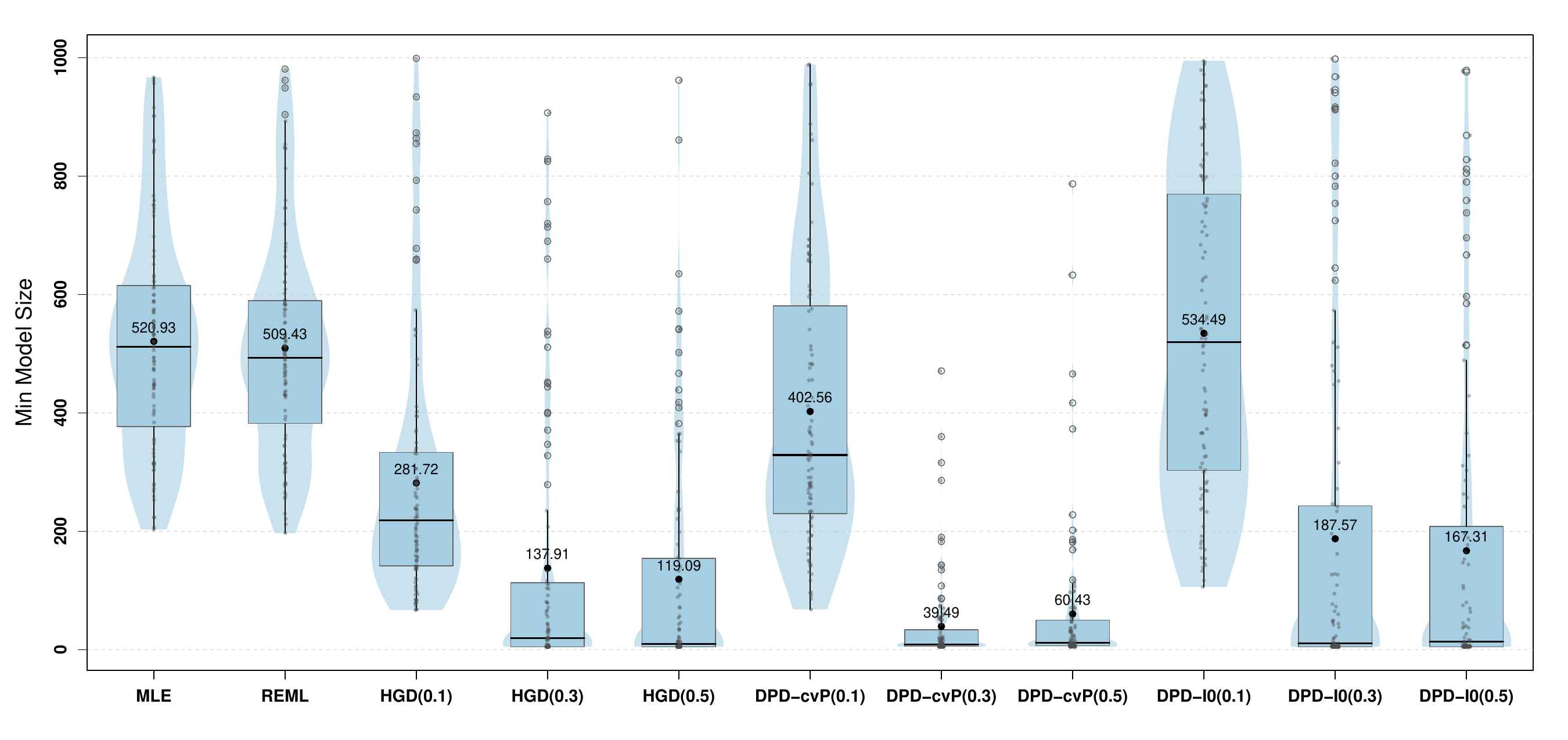}
		\label{FIG:boxplot_Y}}
	~	
	\subfloat[\tiny 20\% (C1) Contamination]{
		\includegraphics[page=2, width=0.3\textwidth]{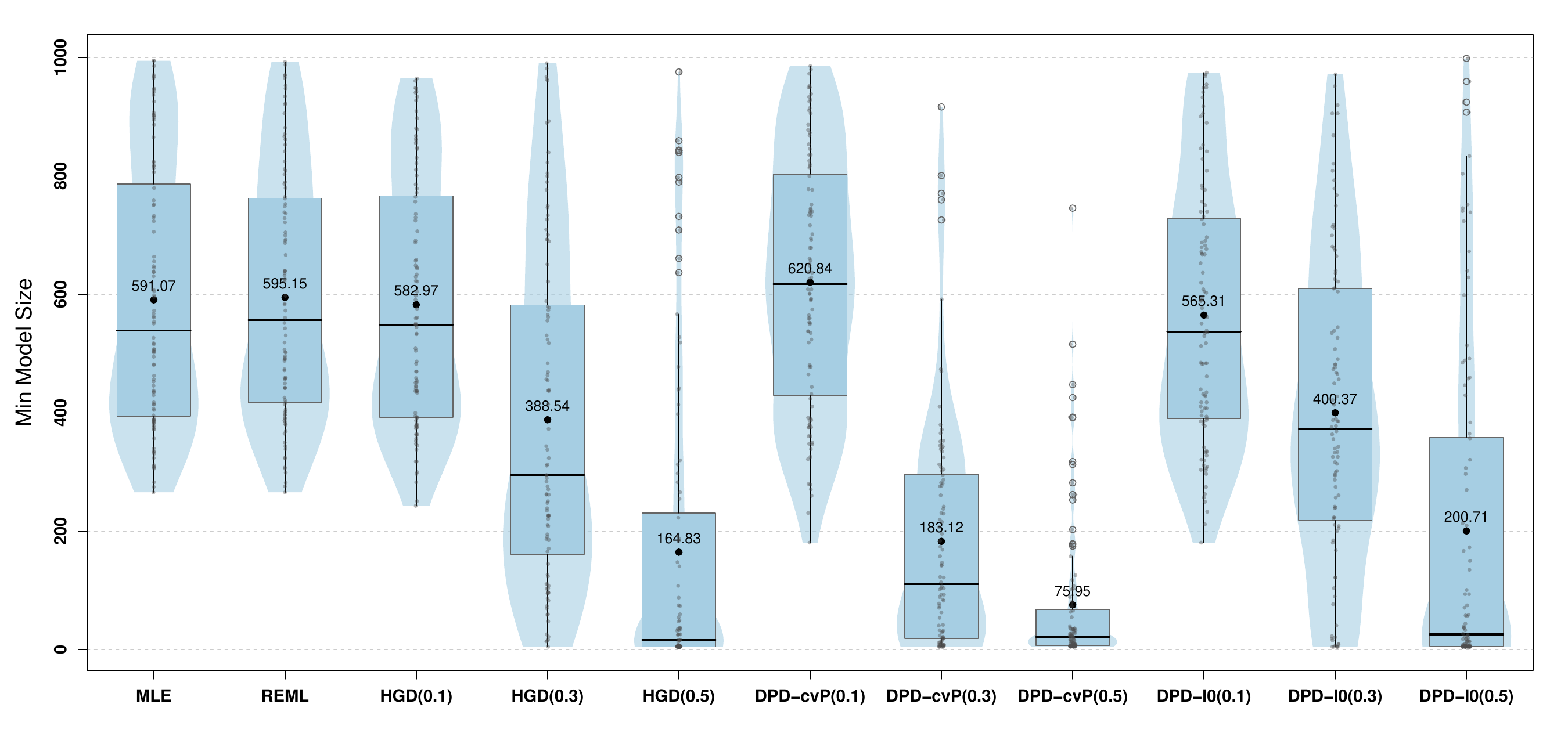}
		\label{FIG:boxplot_Y}}
	\\	
	\subfloat[\tiny 5\% (C2) Contamination]{
	\includegraphics[page=2, width=0.3\textwidth]{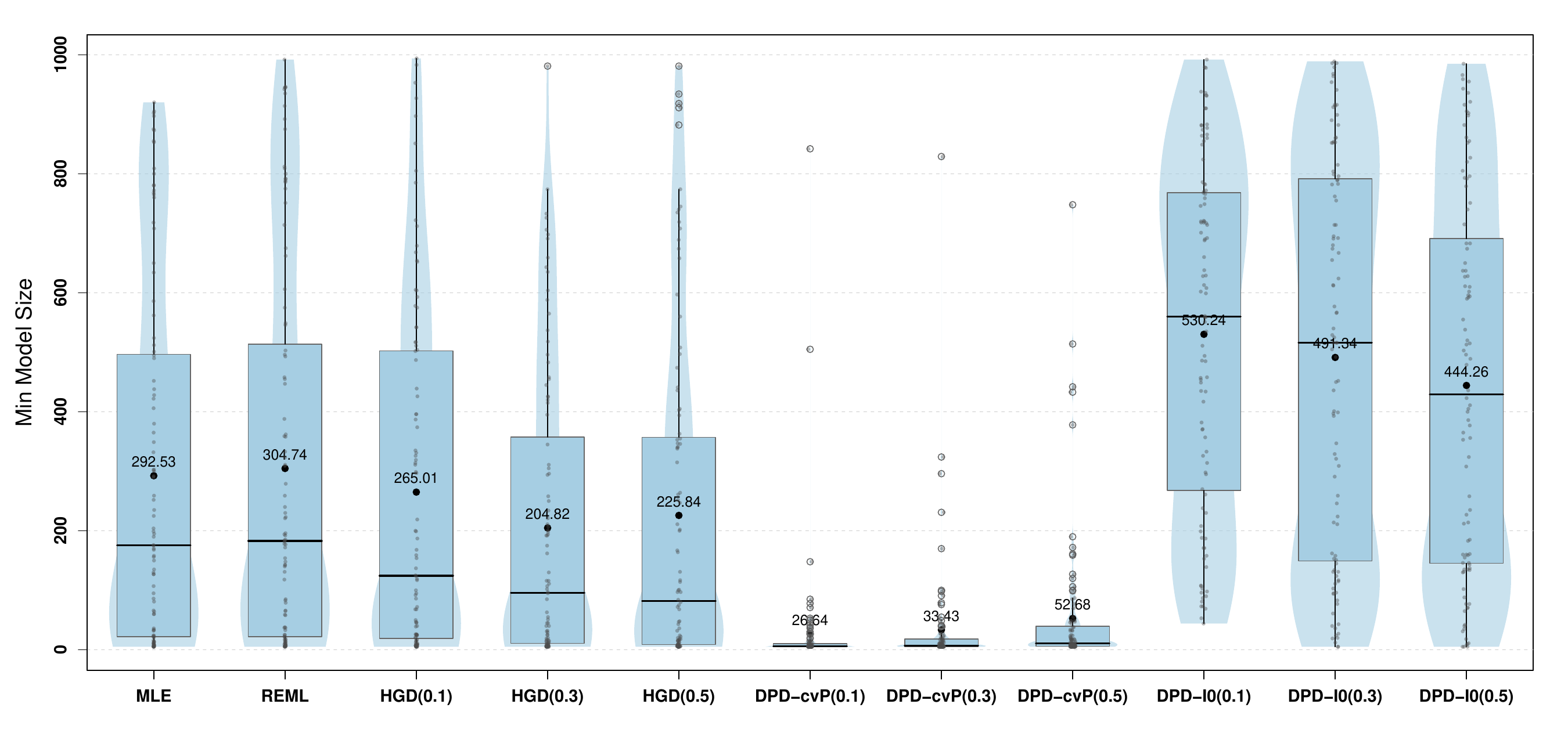}
	\label{FIG:boxplot_Y}}
~	
\subfloat[\tiny 10\% (C2) Contamination]{
	\includegraphics[page=2, width=0.3\textwidth]{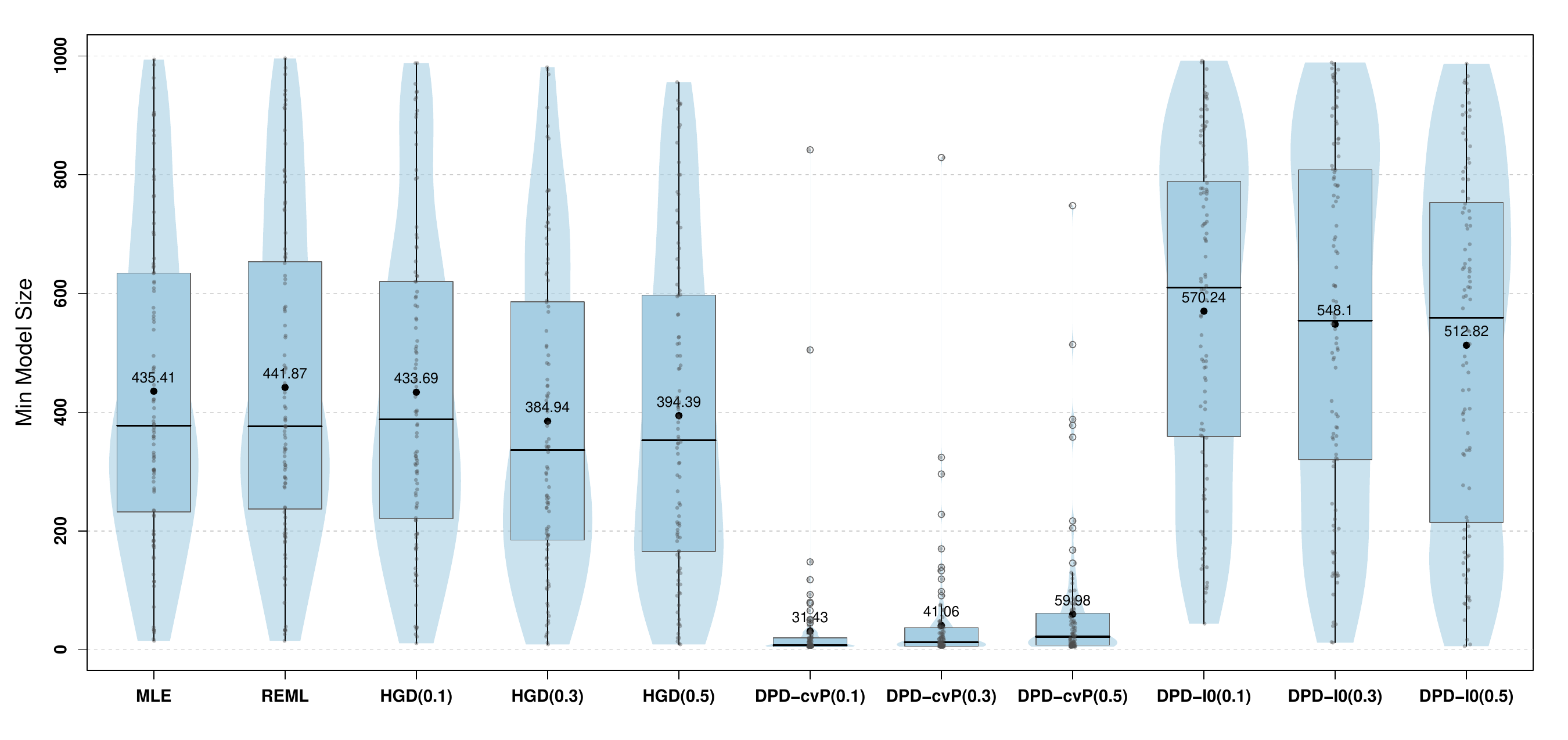}
	\label{FIG:boxplot_Y}}
~	
\subfloat[\tiny 20\% (C2) Contamination]{
	\includegraphics[page=2, width=0.3\textwidth]{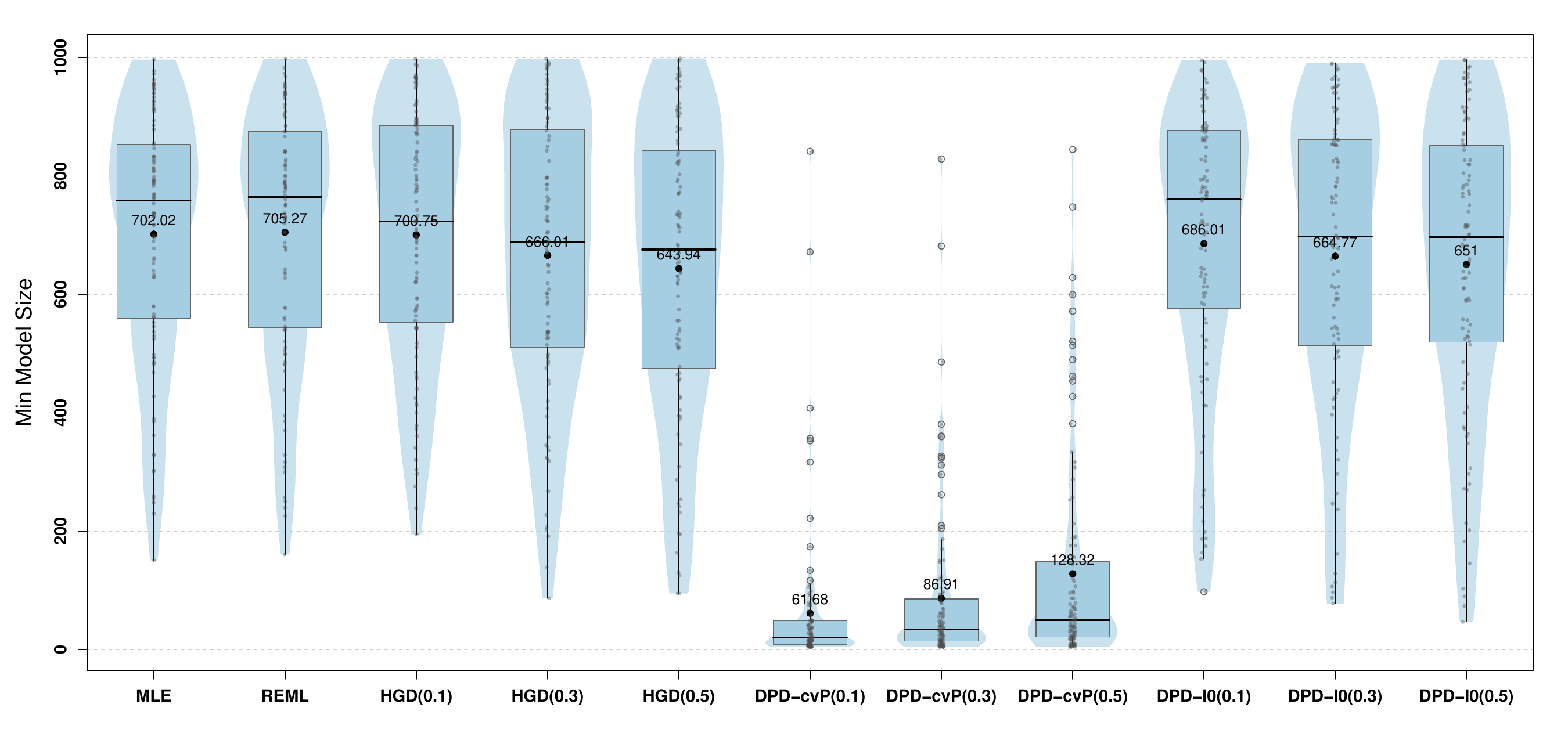}
	\label{FIG:boxplot_Y}}
\\	
	\subfloat[\tiny 5\% (C3) Contamination]{
	\includegraphics[page=2, width=0.3\textwidth]{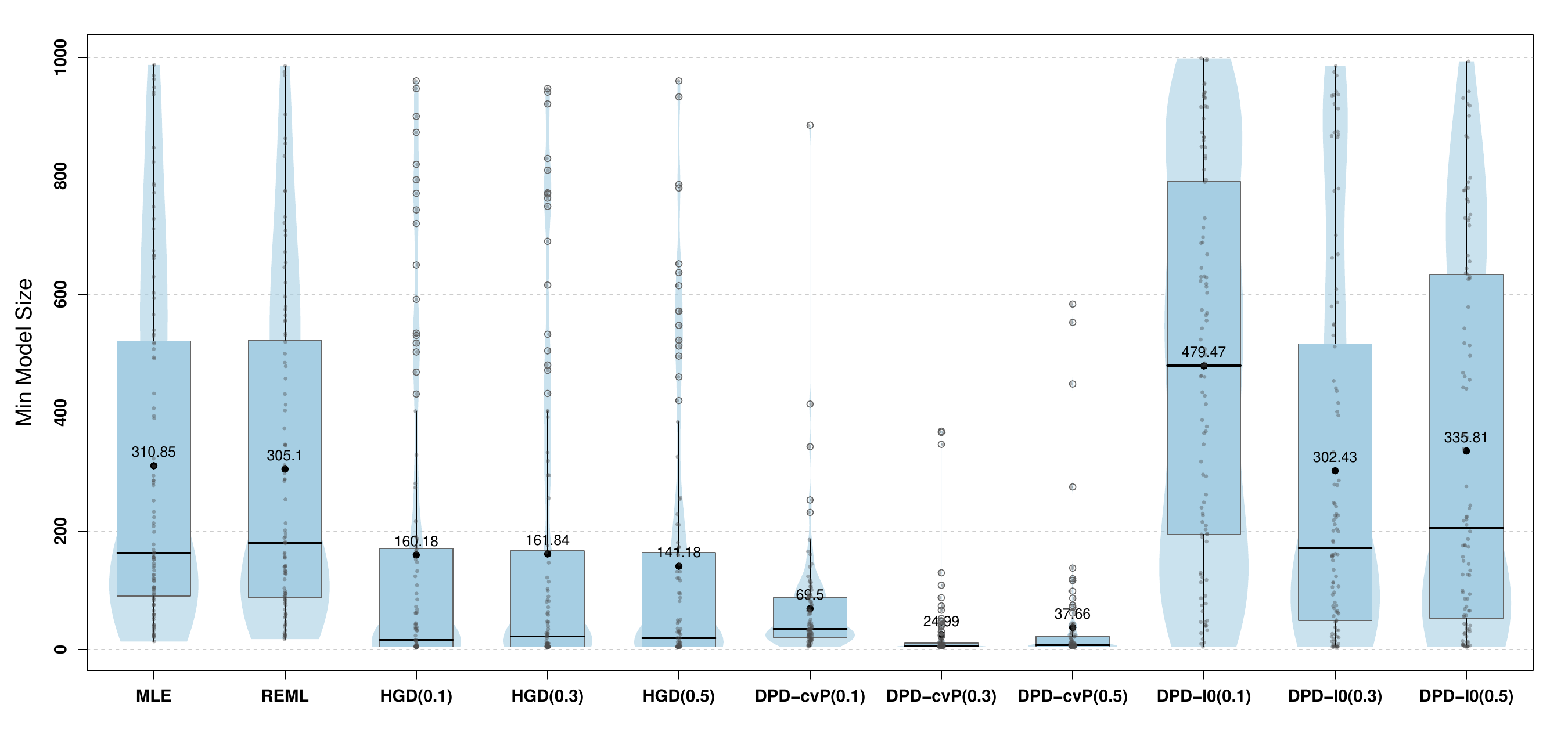}
	\label{FIG:boxplot_Y}}
~	
\subfloat[\tiny 10\% (C3) Contamination]{
	\includegraphics[page=2, width=0.3\textwidth]{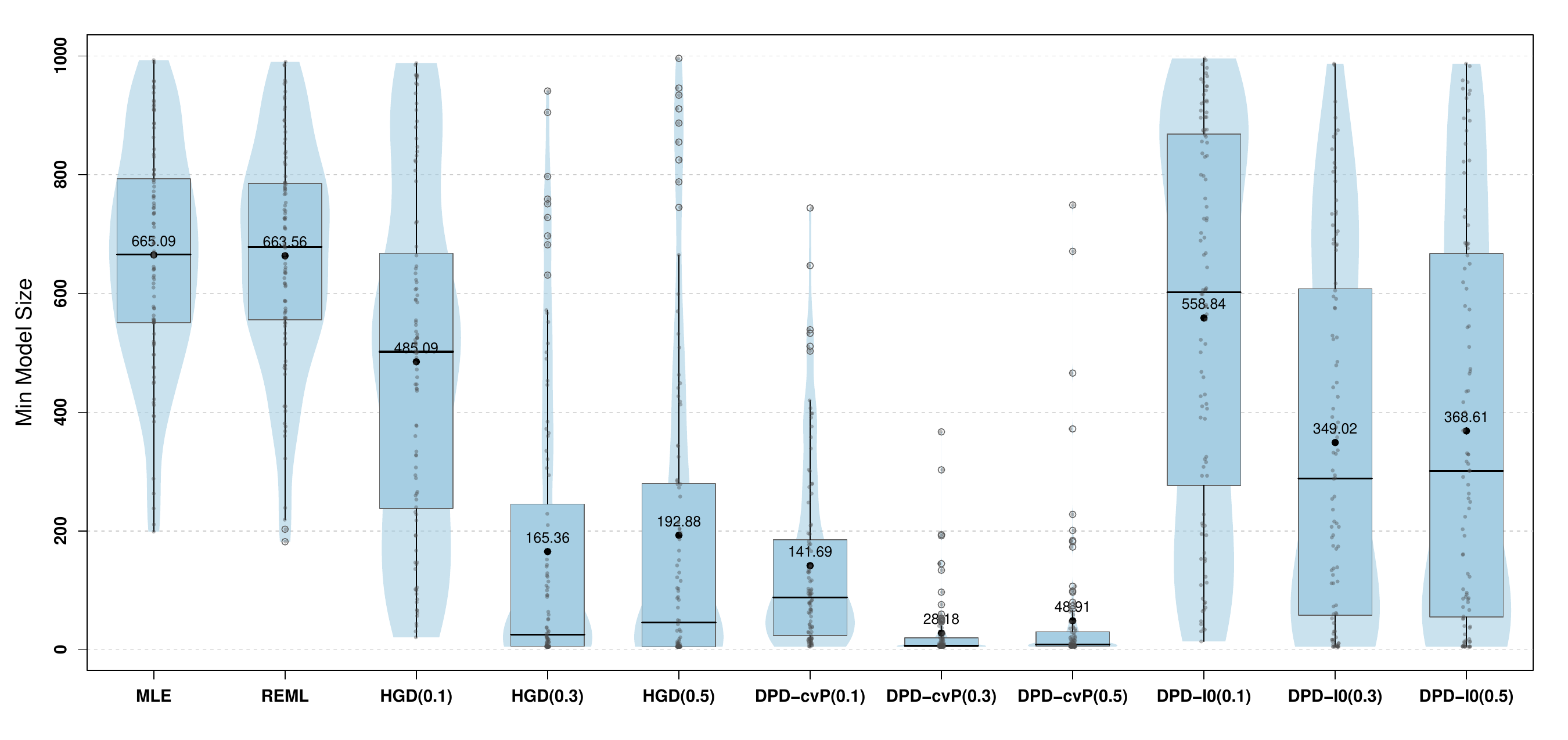}
	\label{FIG:boxplot_Y}}
~	
\subfloat[\tiny 20\% (C3) Contamination]{
	\includegraphics[page=2, width=0.3\textwidth]{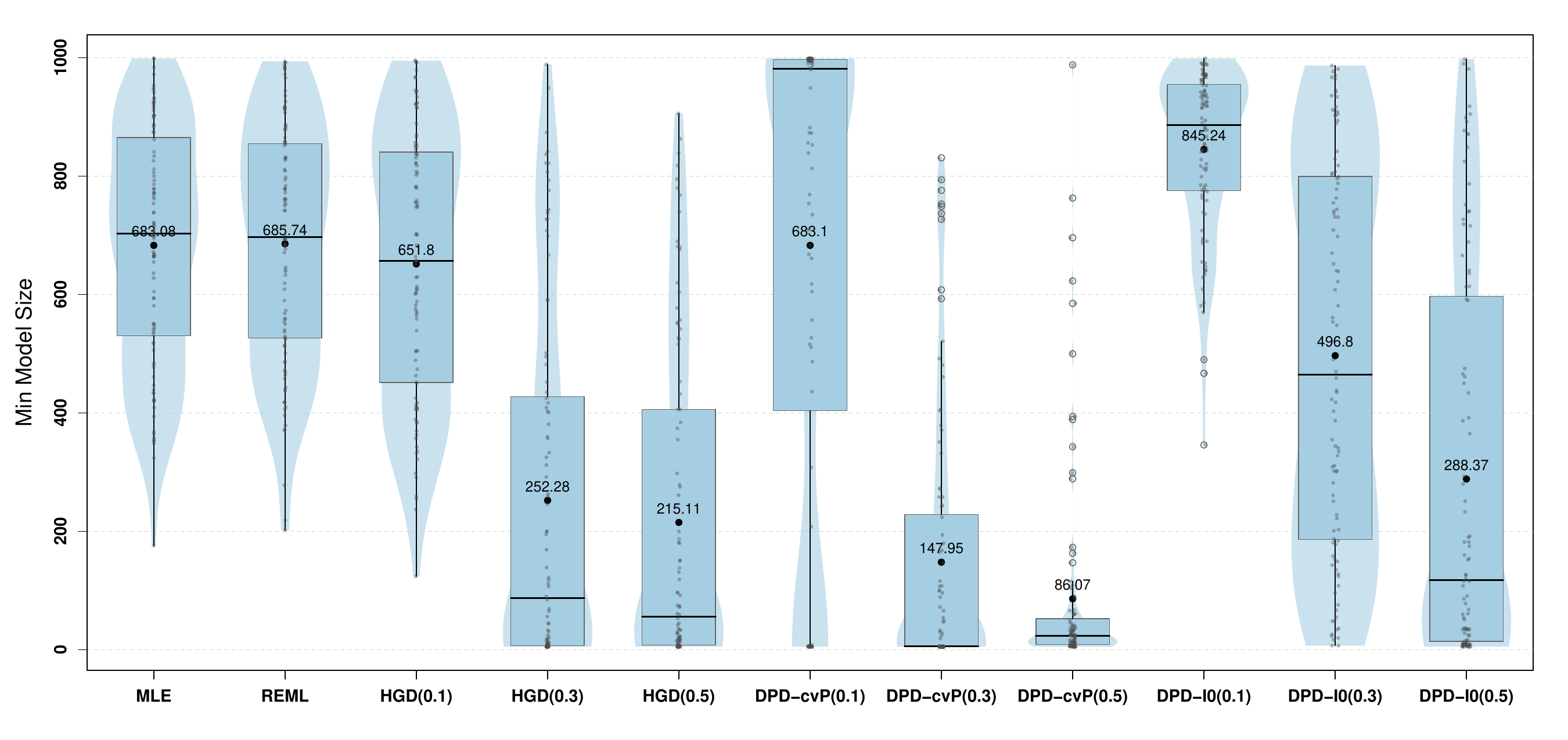}
	\label{FIG:boxplot_Y}}
\\	
	\subfloat[\tiny 5\% (C4) Contamination]{
	\includegraphics[page=2, width=0.3\textwidth]{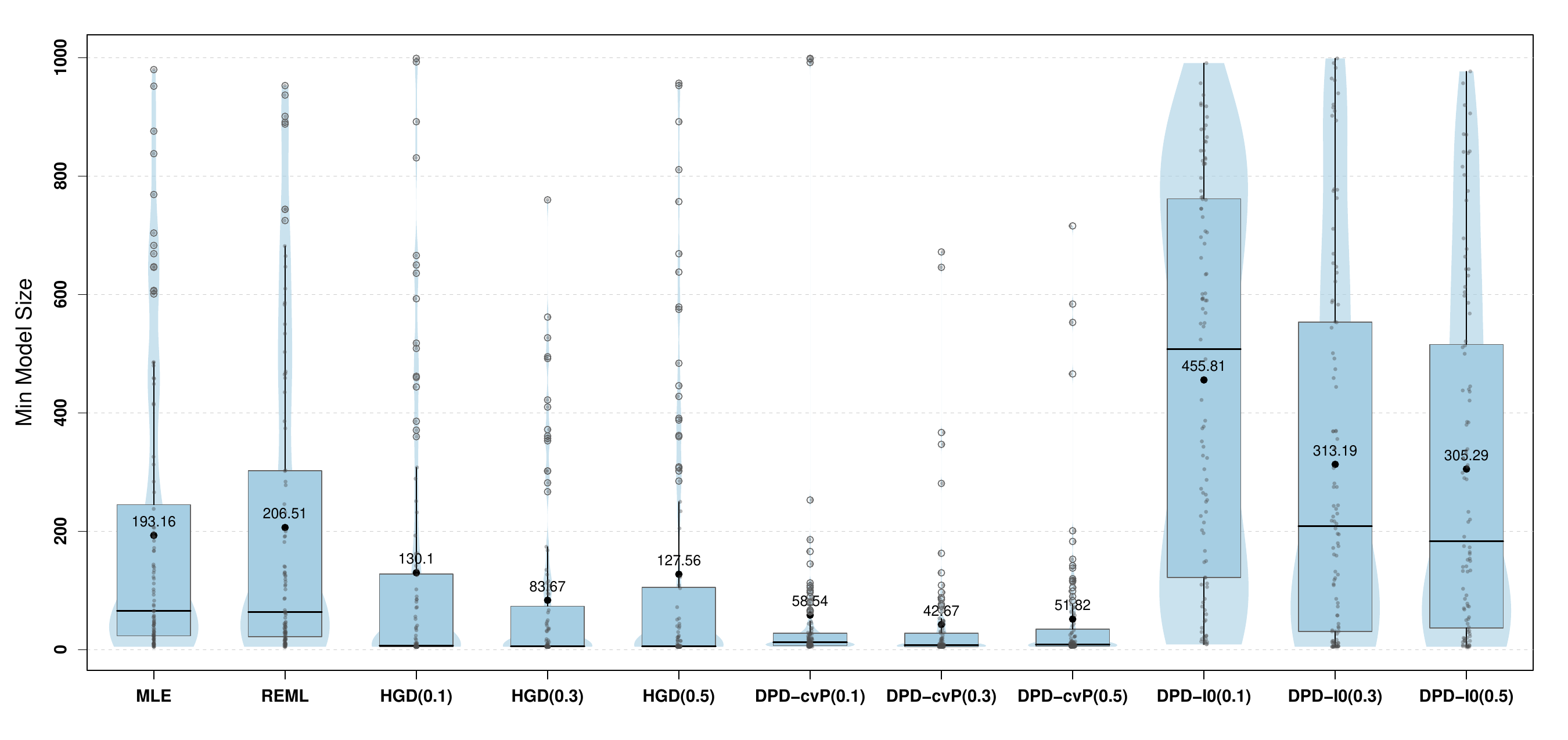}
	\label{FIG:boxplot_Y}}
~	
\subfloat[\tiny 10\% (C4) Contamination]{
	\includegraphics[page=2, width=0.3\textwidth]{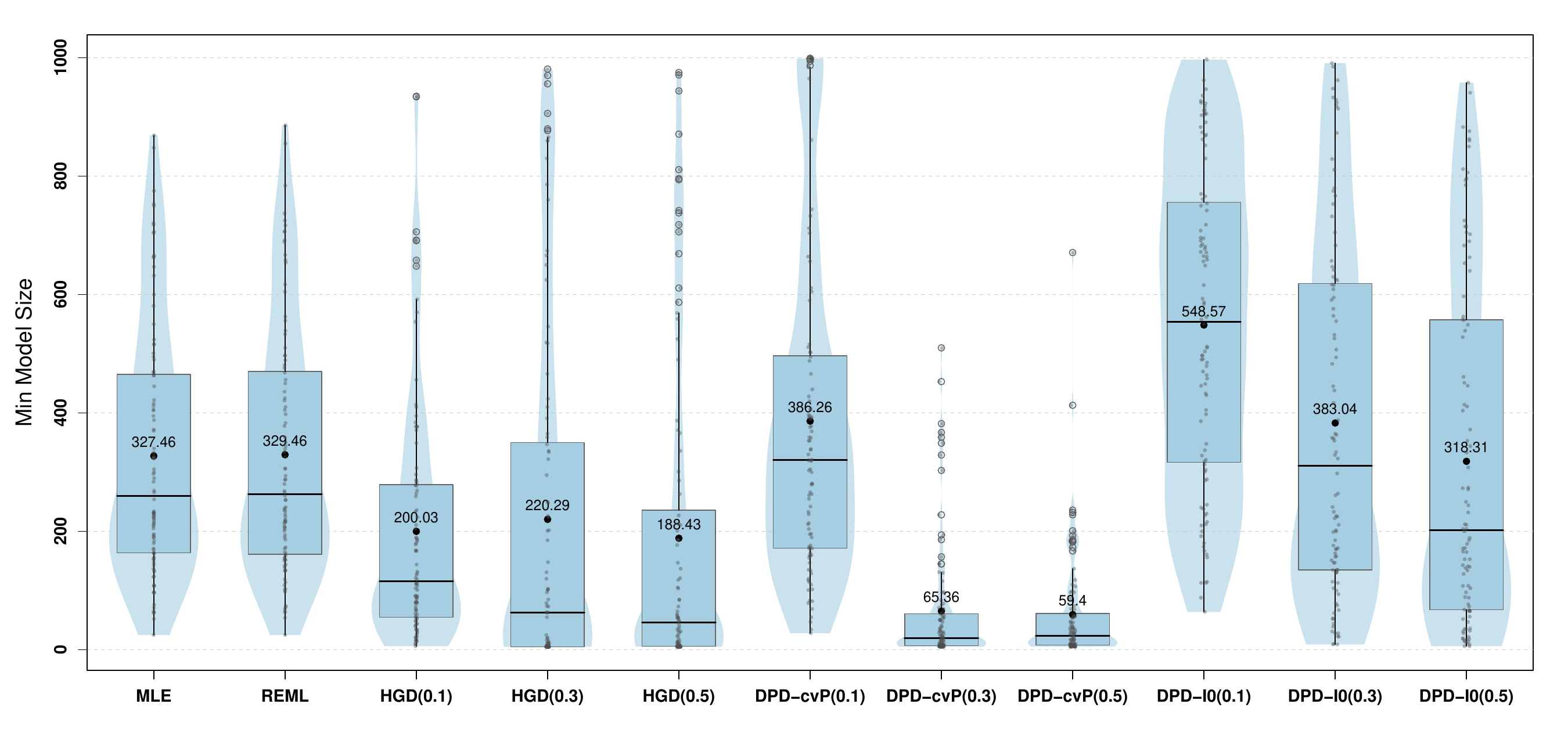}
	\label{FIG:boxplot_Y}}
~	
\subfloat[\tiny 20\% (C4) Contamination]{
	\includegraphics[page=2, width=0.3\textwidth]{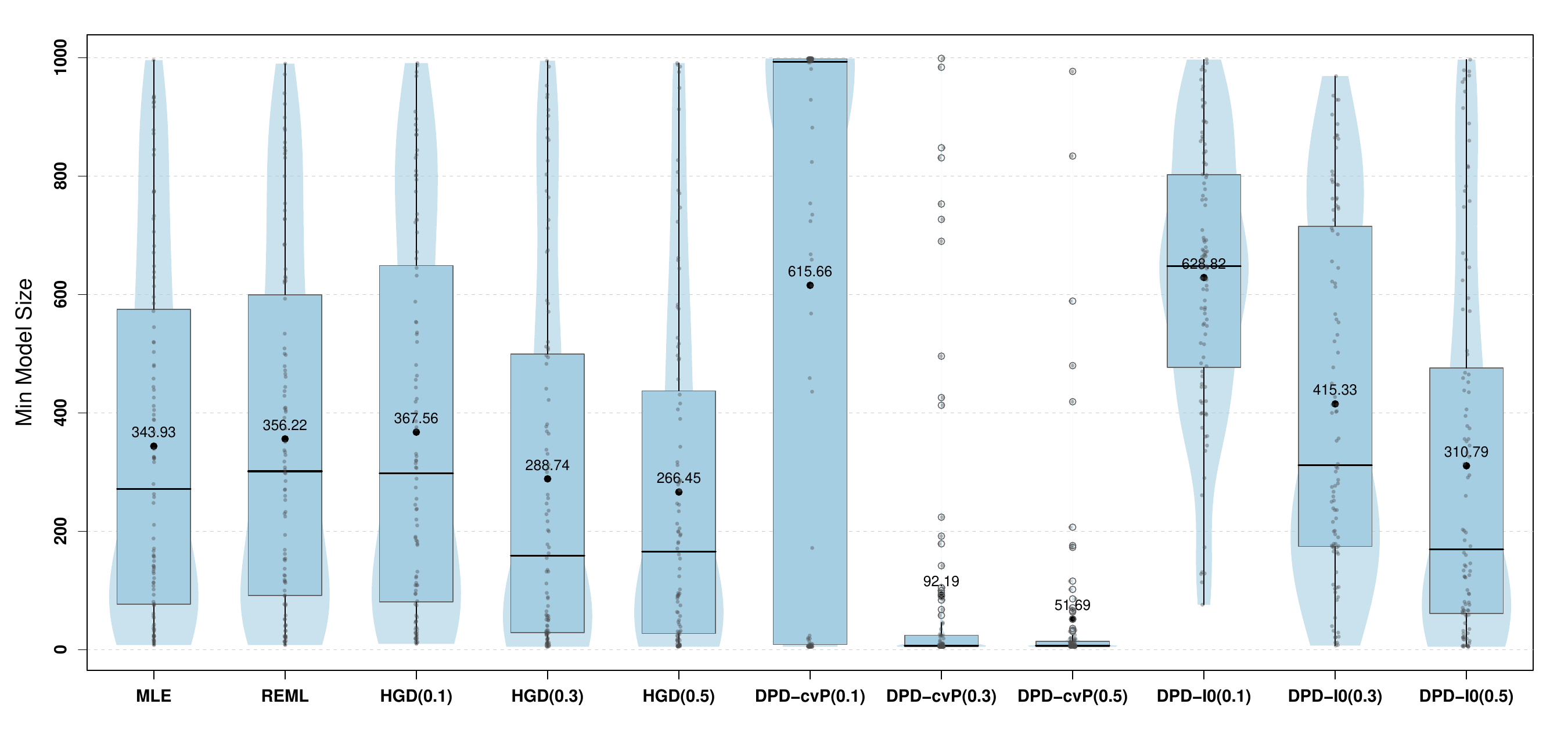}
	\label{FIG:boxplot_Y}}
	\caption{Boxplots (with overlaid sample means) of MinMS required for sure screening under scenario (S1)$\times$(R1) with $\bm\Sigma_x=\mathbb{I}$ and different types of contamination.
		The boxes represent (from left to right) the benchmark SIS based on MLE, REML, and HGD with $\gamma=0.1, 0.3, 0,5$, 
		he DPD-SISP with cv-P with $\alpha=0.1, 0.3, 0.5$, and the DPD-SISP with I0-P with $\alpha=0.1, 0.3, 0.5$, respectively.}
	\label{FIG:MMS_Csp1}
\end{figure}

\begin{figure}[!h]
	\centering
	\subfloat[\tiny 5\% (C1) Contamination]{
		\includegraphics[page=2, width=0.3\textwidth]{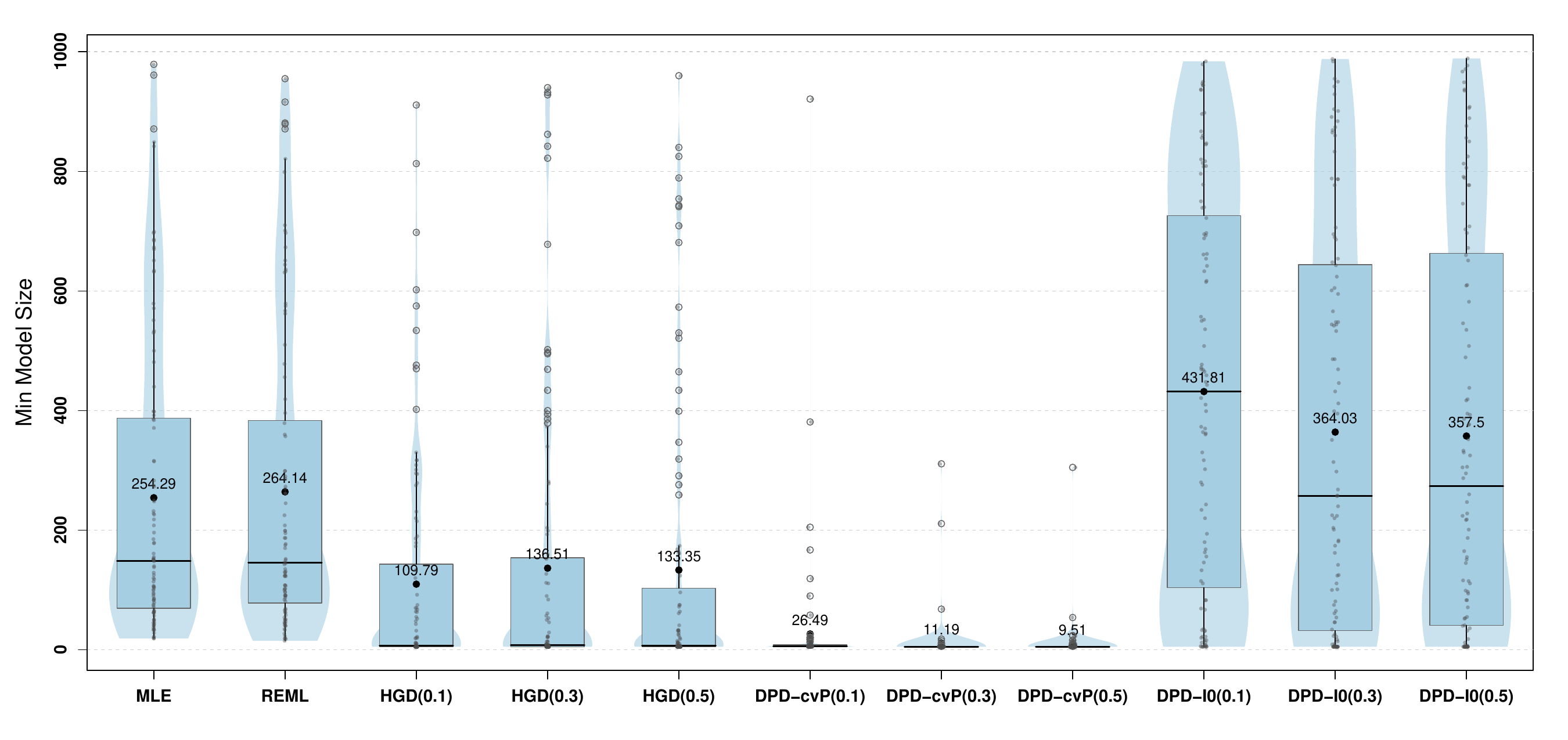}
		\label{FIG:boxplot_Y}}
	~	
	\subfloat[\tiny 10\% (C1) Contamination]{
		\includegraphics[page=2, width=0.3\textwidth]{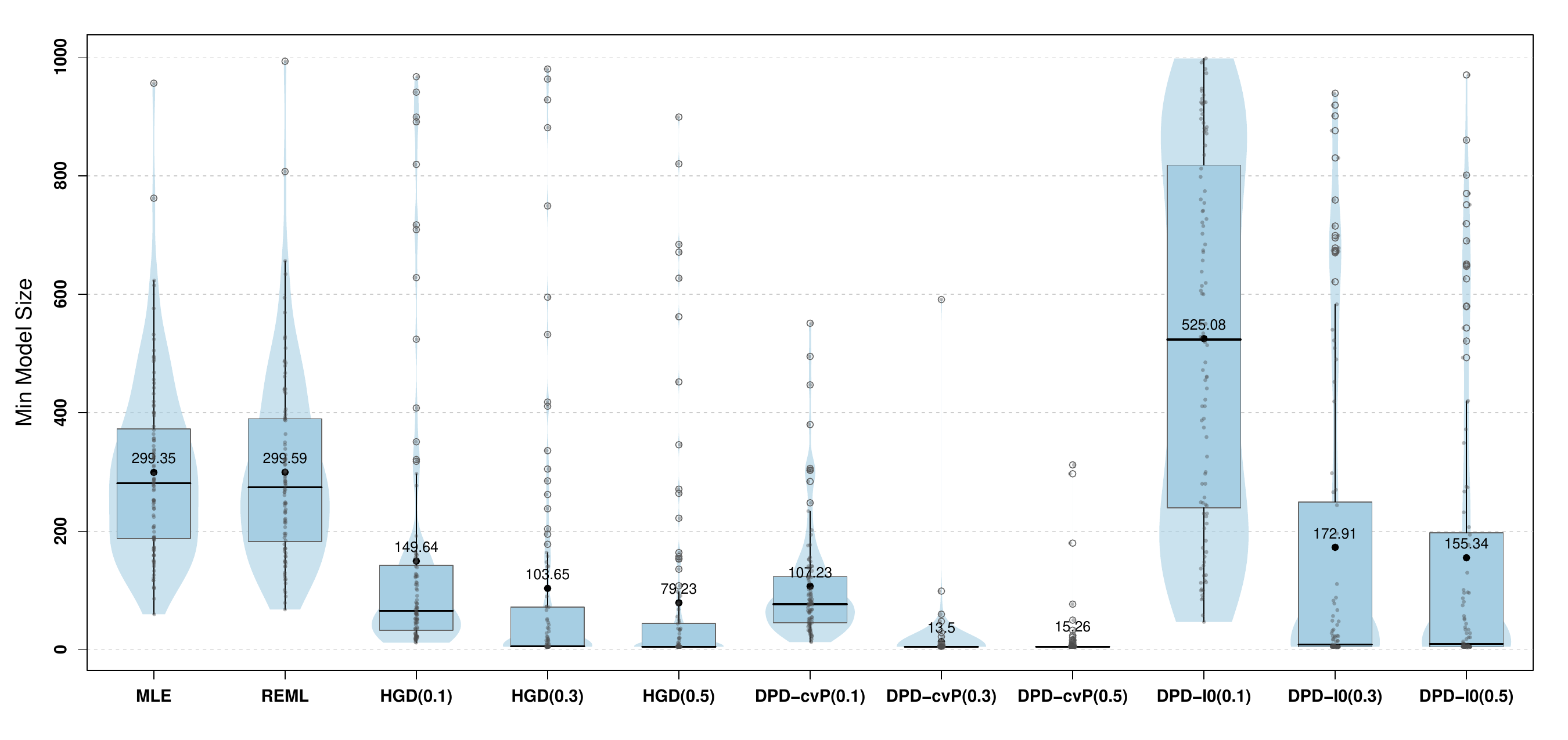}
		\label{FIG:boxplot_Y}}
	~	
	\subfloat[\tiny 20\% (C1) Contamination]{
		\includegraphics[page=2, width=0.3\textwidth]{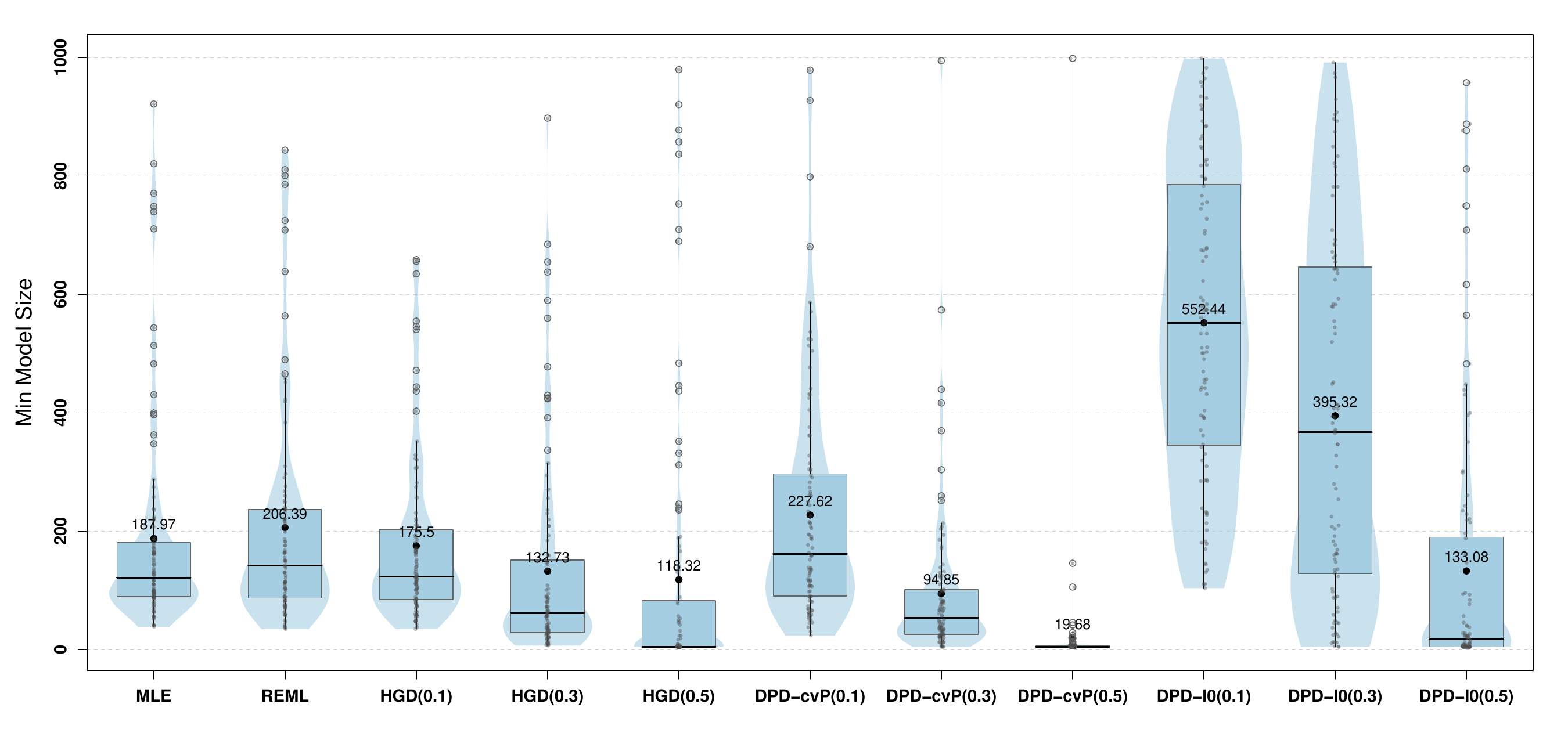}
		\label{FIG:boxplot_Y}}
	\\	
	\subfloat[\tiny 5\% (C2) Contamination]{
		\includegraphics[page=2, width=0.3\textwidth]{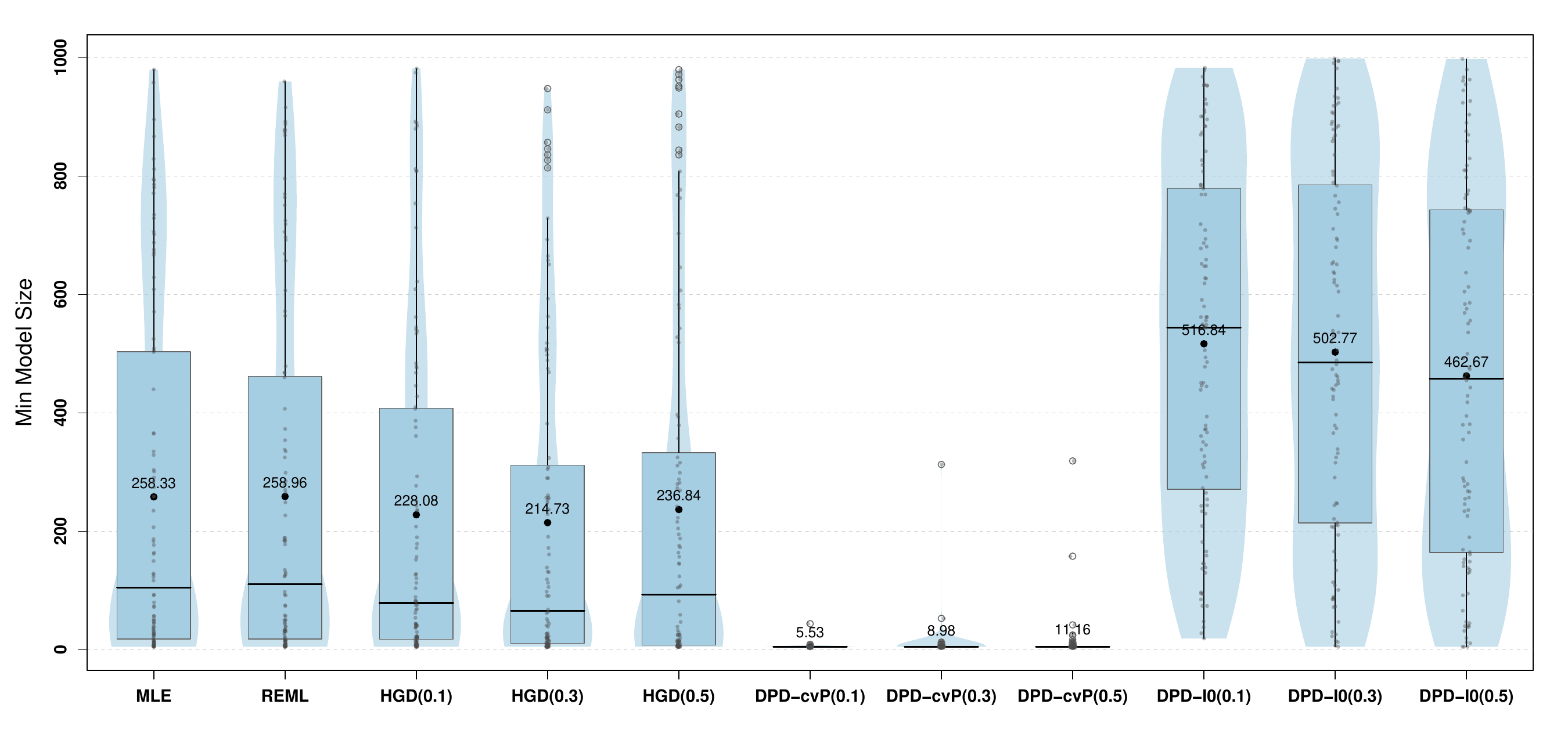}
		\label{FIG:boxplot_Y}}
	~	
	\subfloat[\tiny 10\% (C2) Contamination]{
		\includegraphics[page=2, width=0.3\textwidth]{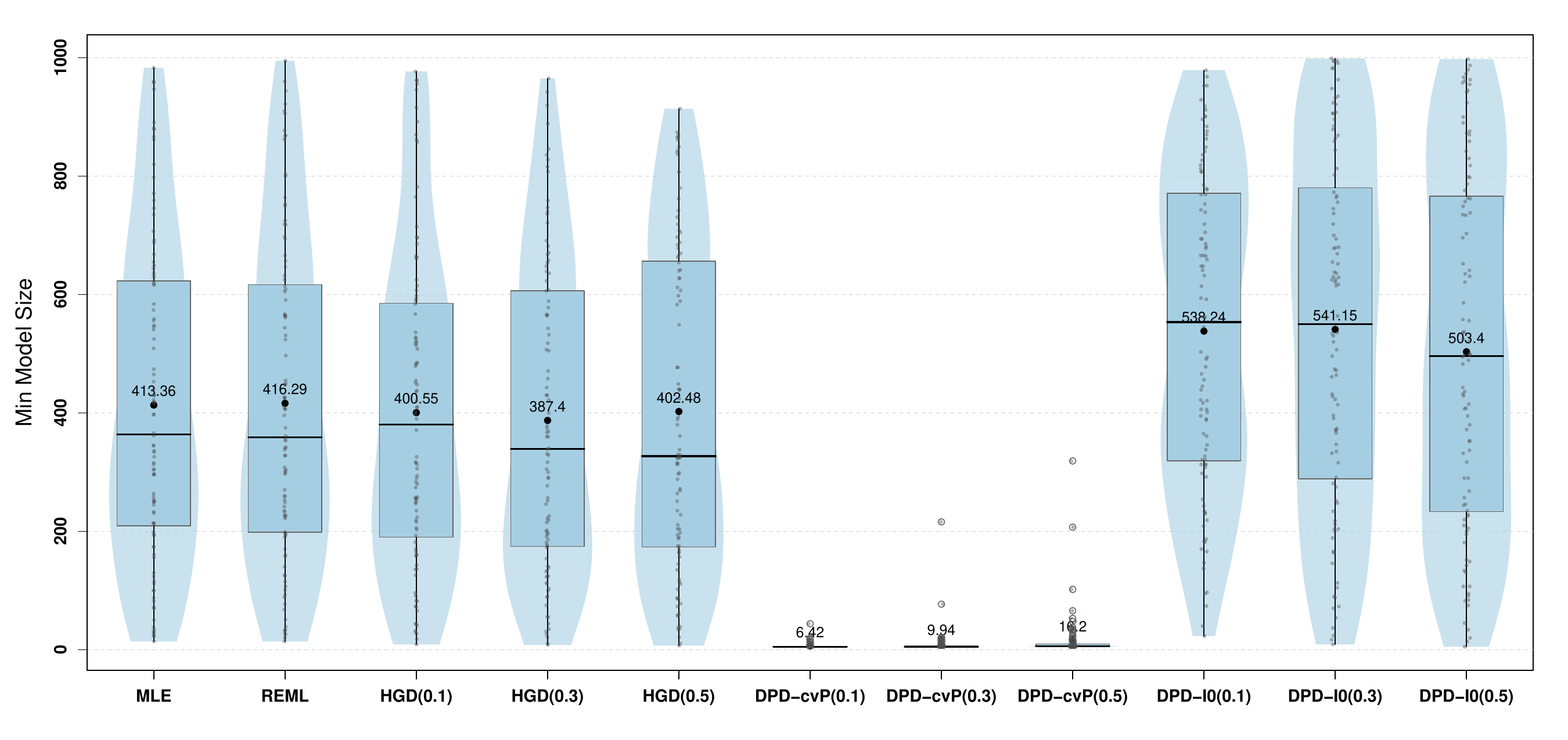}
		\label{FIG:boxplot_Y}}
	~	
	\subfloat[\tiny 20\% (C2) Contamination]{
		\includegraphics[page=2, width=0.3\textwidth]{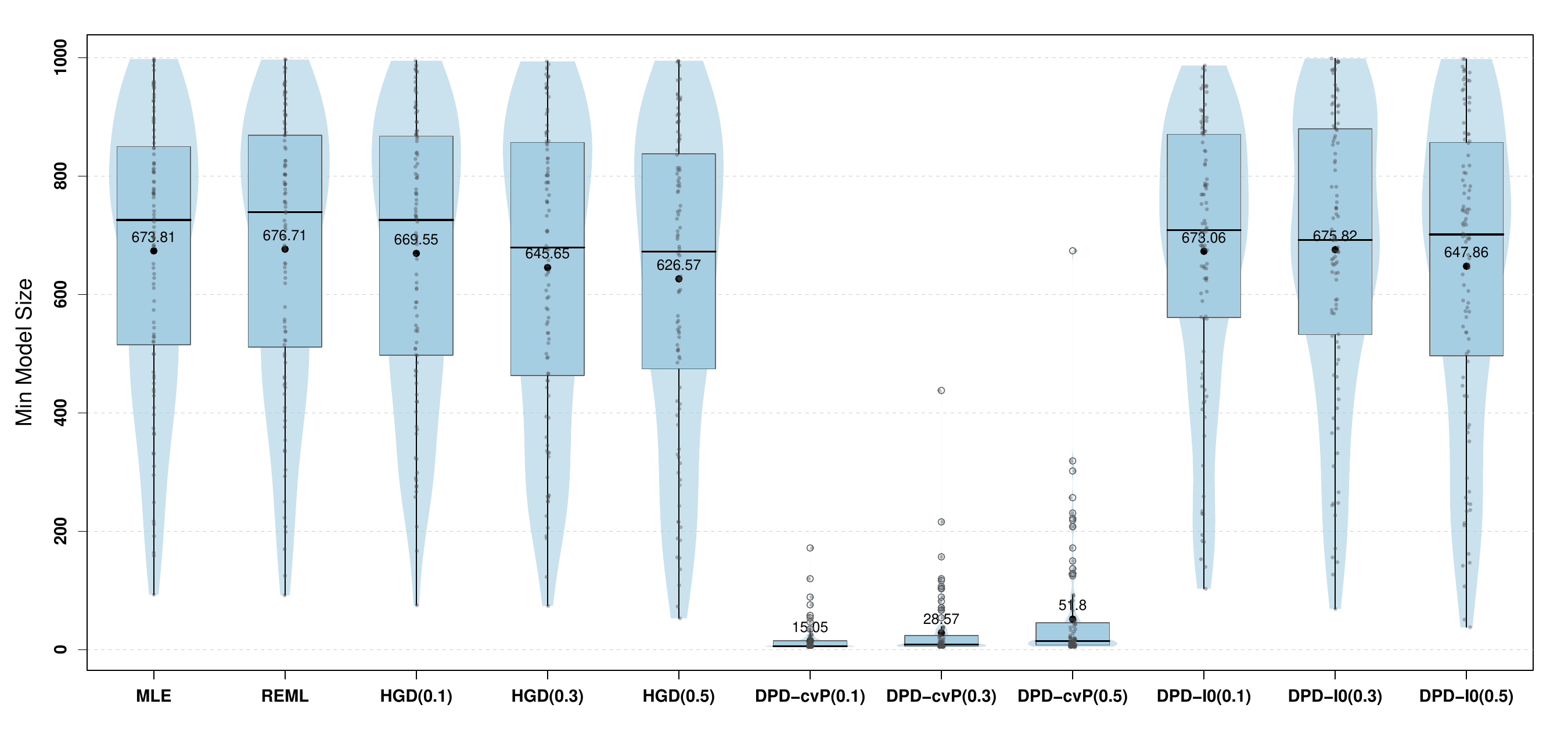}
		\label{FIG:boxplot_Y}}
	\\	
	\subfloat[\tiny 5\% (C3) Contamination]{
		\includegraphics[page=2, width=0.3\textwidth]{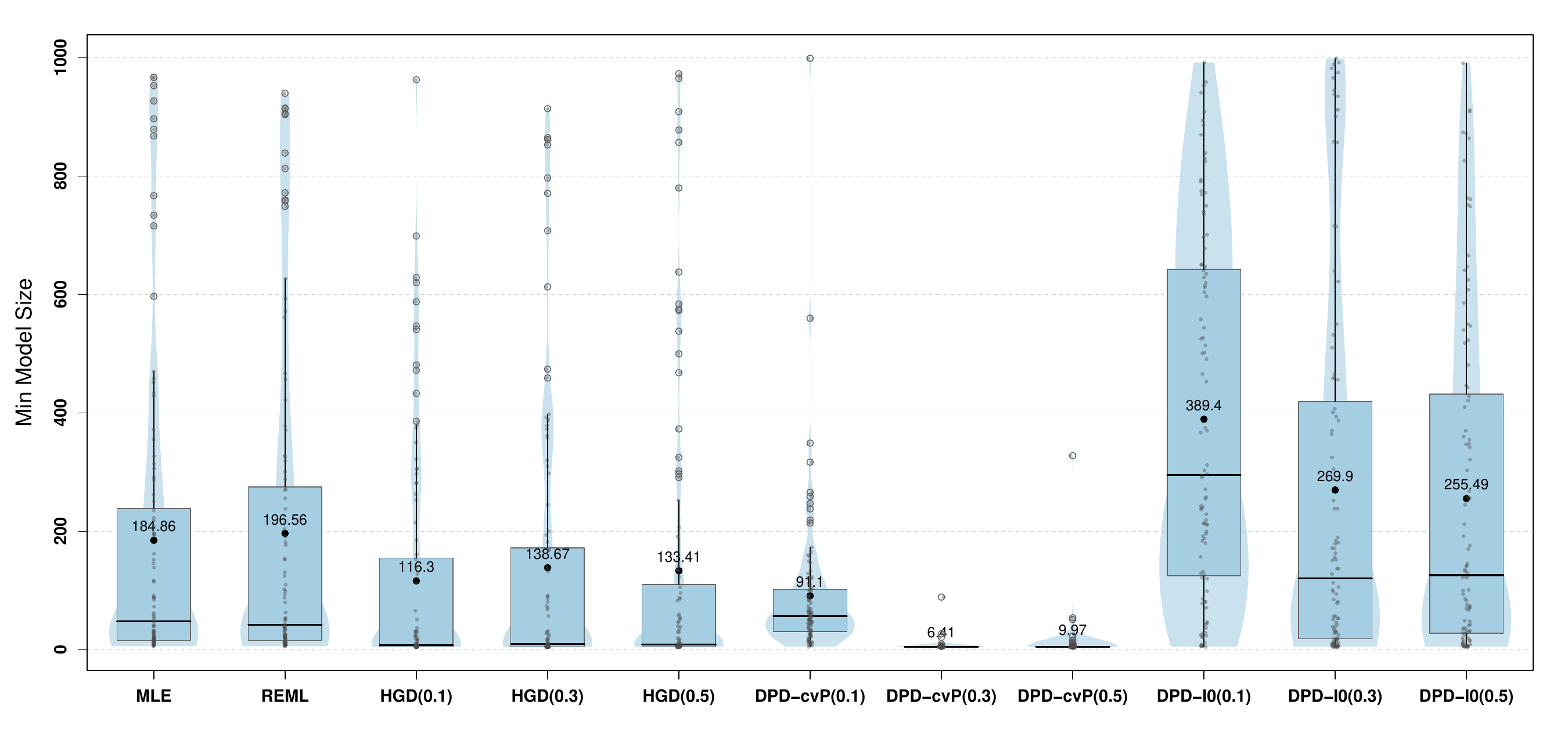}
		\label{FIG:boxplot_Y}}
	~	
	\subfloat[\tiny 10\% (C3) Contamination]{
		\includegraphics[page=2, width=0.3\textwidth]{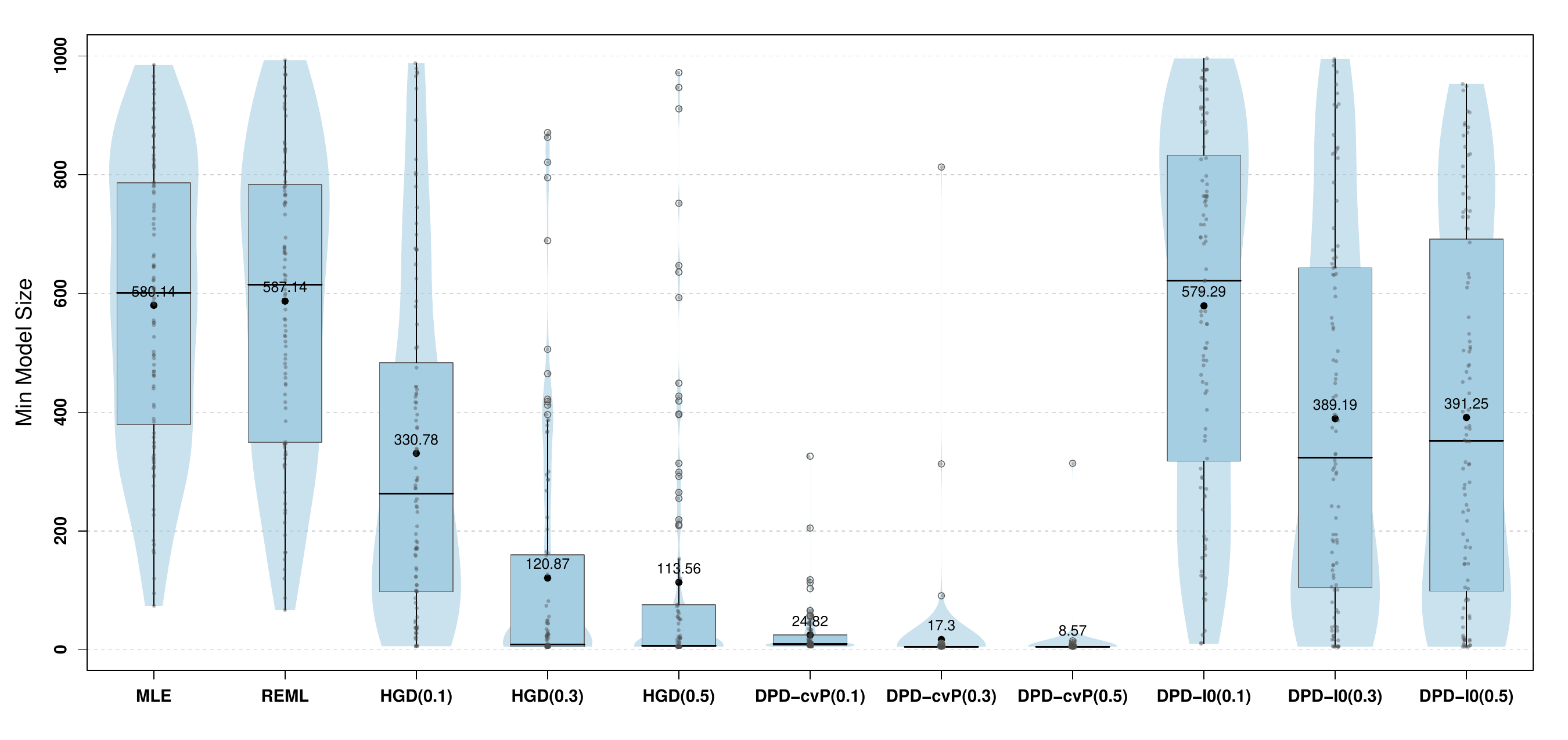}
		\label{FIG:boxplot_Y}}
	~	
	\subfloat[\tiny 20\% (C3) Contamination]{
		\includegraphics[page=2, width=0.3\textwidth]{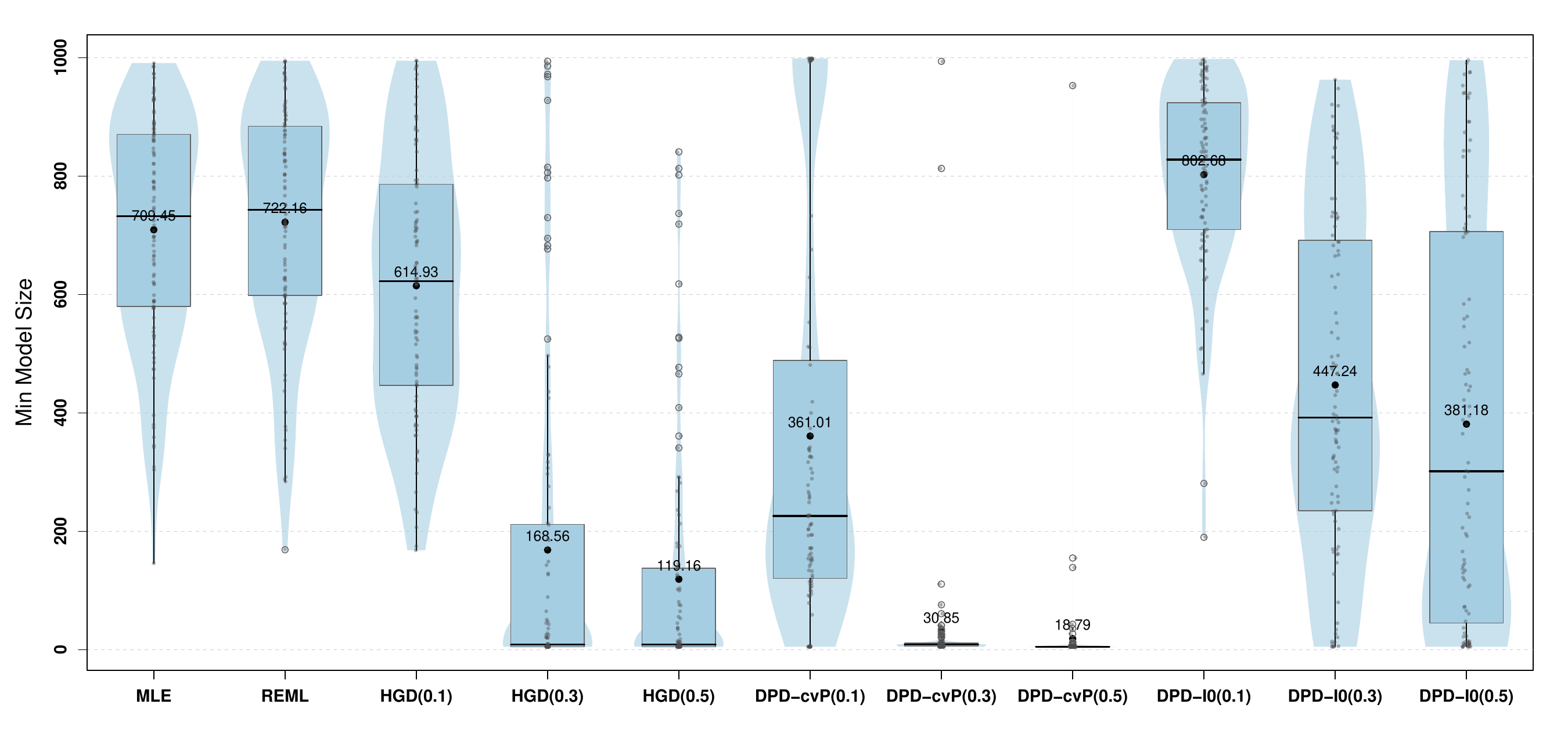}
		\label{FIG:boxplot_Y}}
	\\	
	\subfloat[\tiny 5\% (C4) Contamination]{
		\includegraphics[page=2, width=0.3\textwidth]{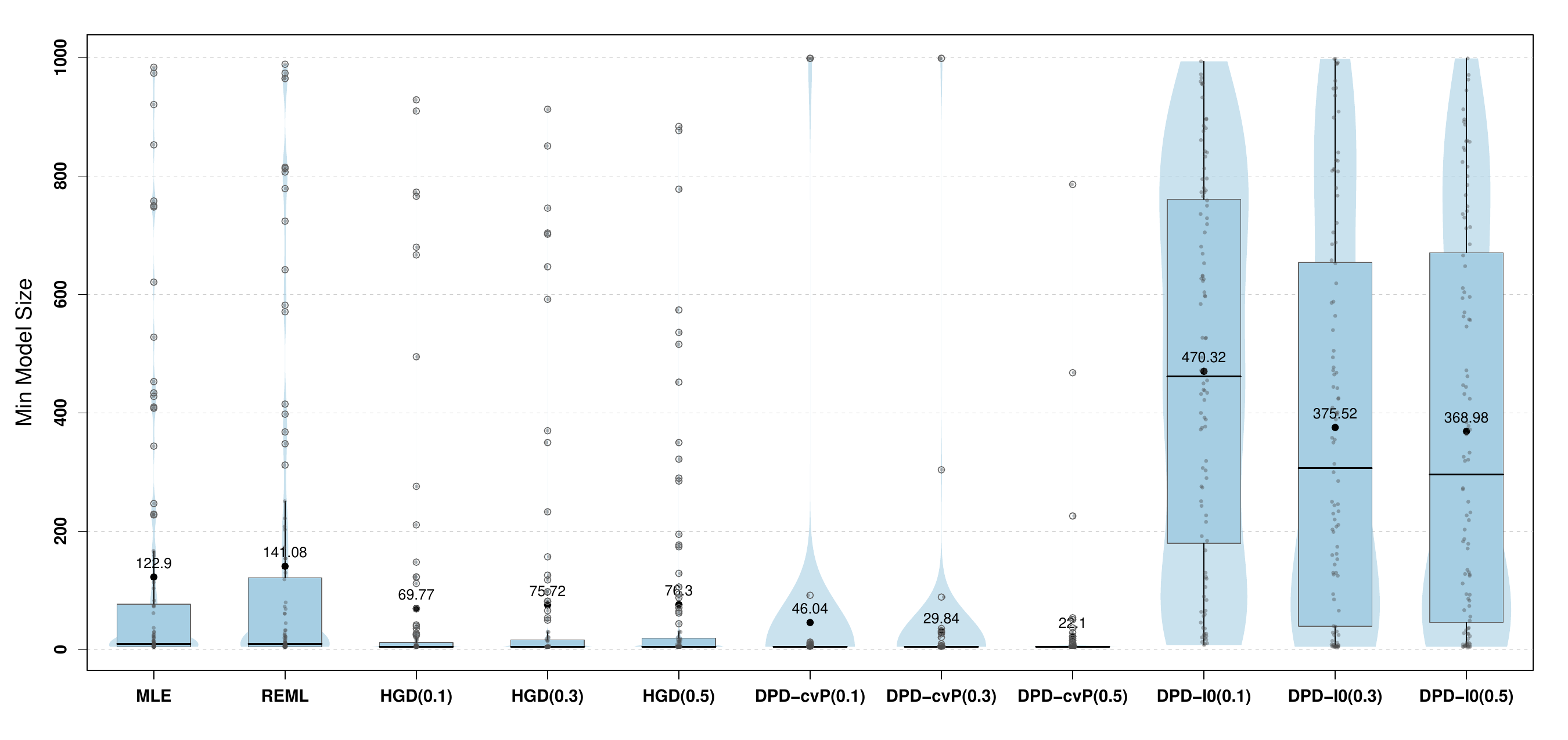}
		\label{FIG:boxplot_Y}}
	~	
	\subfloat[\tiny 10\% (C4) Contamination]{
		\includegraphics[page=2, width=0.3\textwidth]{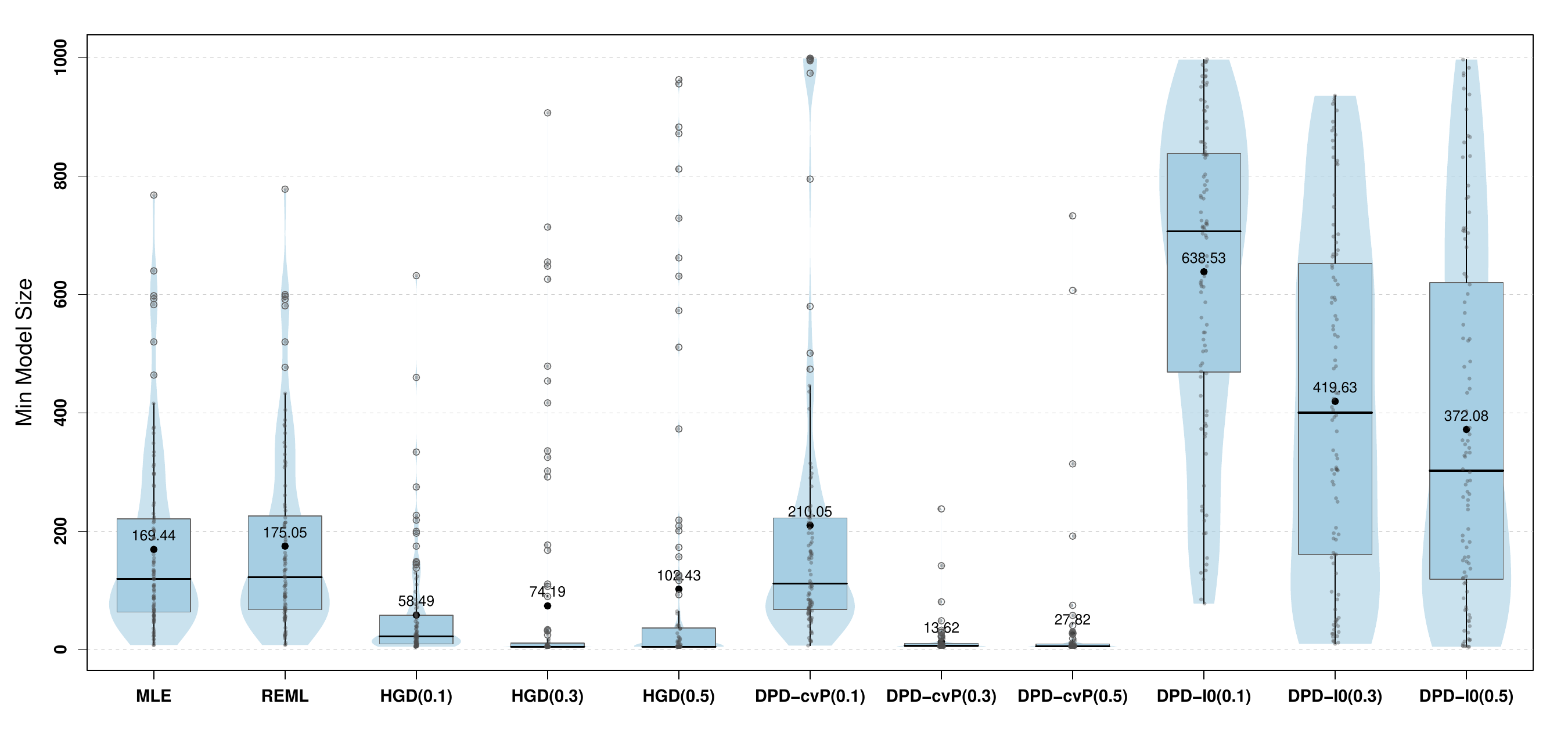}
		\label{FIG:boxplot_Y}}
	~	
	\subfloat[\tiny 20\% (C4) Contamination]{
		\includegraphics[page=2, width=0.3\textwidth]{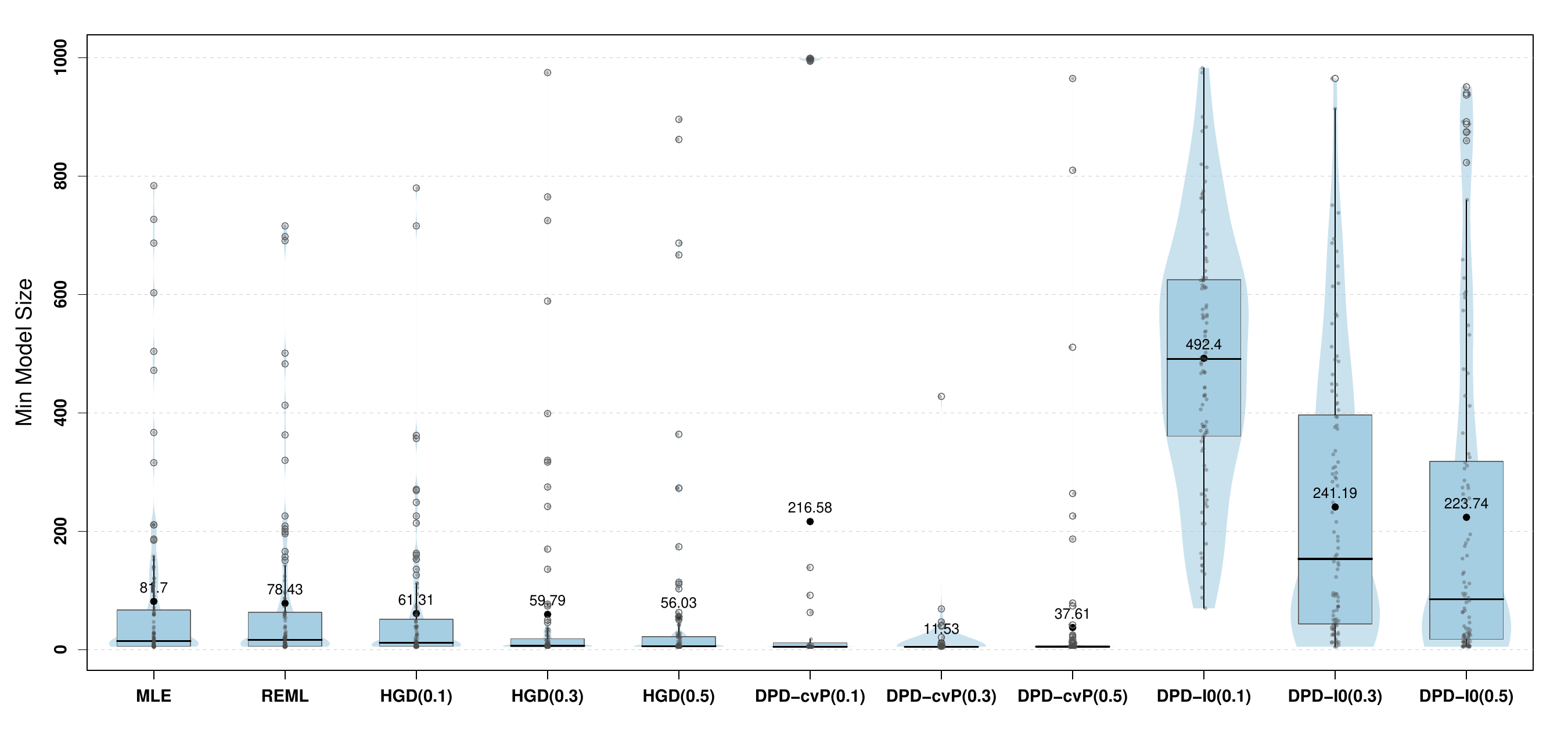}
		\label{FIG:boxplot_Y}}
	\caption{Boxplots (with overlaid sample means) of MinMS required for sure screening under scenario (S1)$\times$(R1) with $\bm\Sigma_x$ as T(0.3) and different types of contamination.
		The boxes represent (from left to right) the benchmark SIS based on MLE, REML, and HGD with $\gamma=0.1, 0.3, 0,5$, 
		the DPD-SISP with cv-P with $\alpha=0.1, 0.3, 0.5$, and the DPD-SISP with I0-P with $\alpha=0.1, 0.3, 0.5$, respectively.}
	\label{FIG:MMS_Csp2}
\end{figure}

\begin{figure}[!h]
	\centering
	\subfloat[\tiny 5\% (C1) Contamination]{
		\includegraphics[page=2, width=0.3\textwidth]{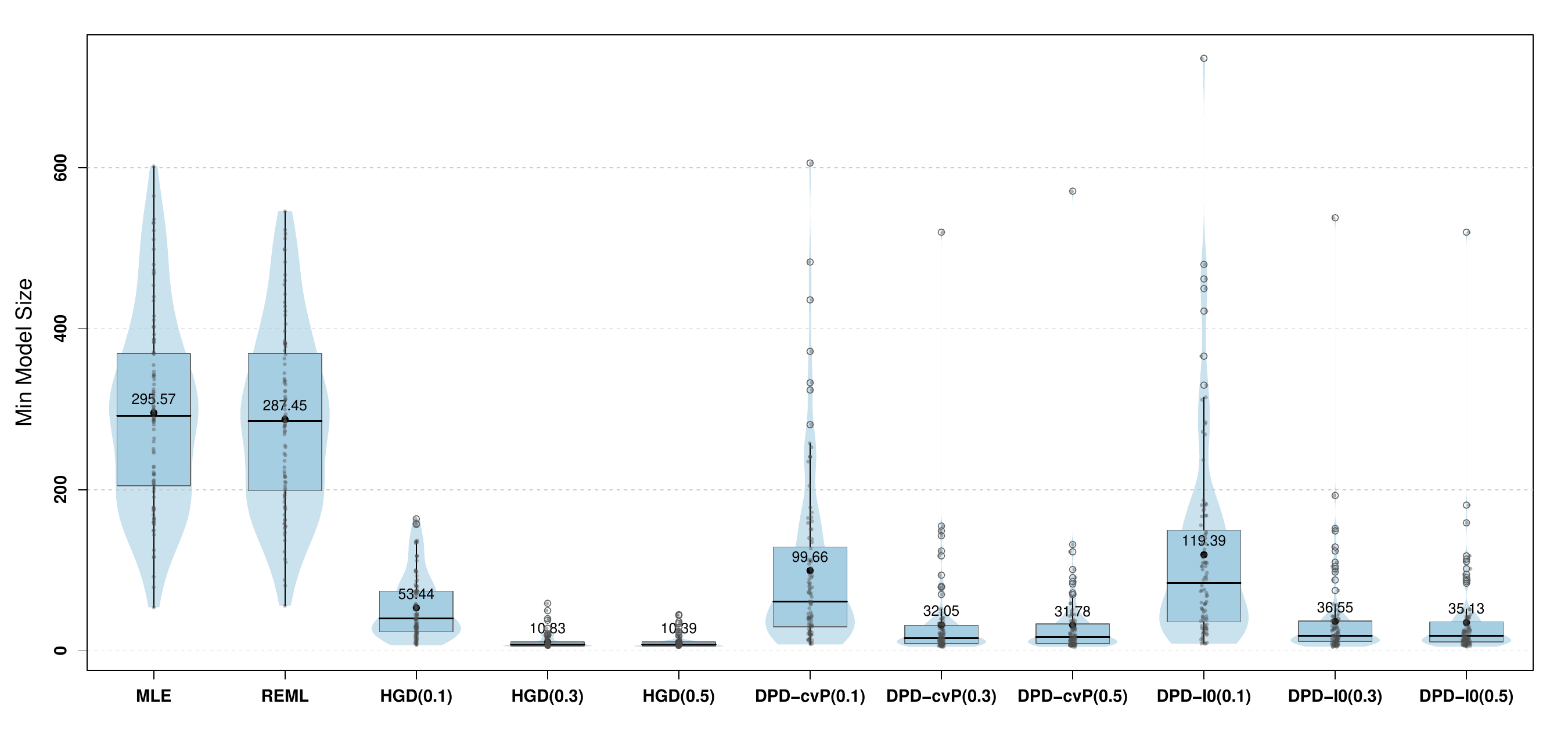}
		\label{FIG:boxplot_Y}}
	~	
	\subfloat[\tiny 10\% (C1) Contamination]{
		\includegraphics[page=2, width=0.3\textwidth]{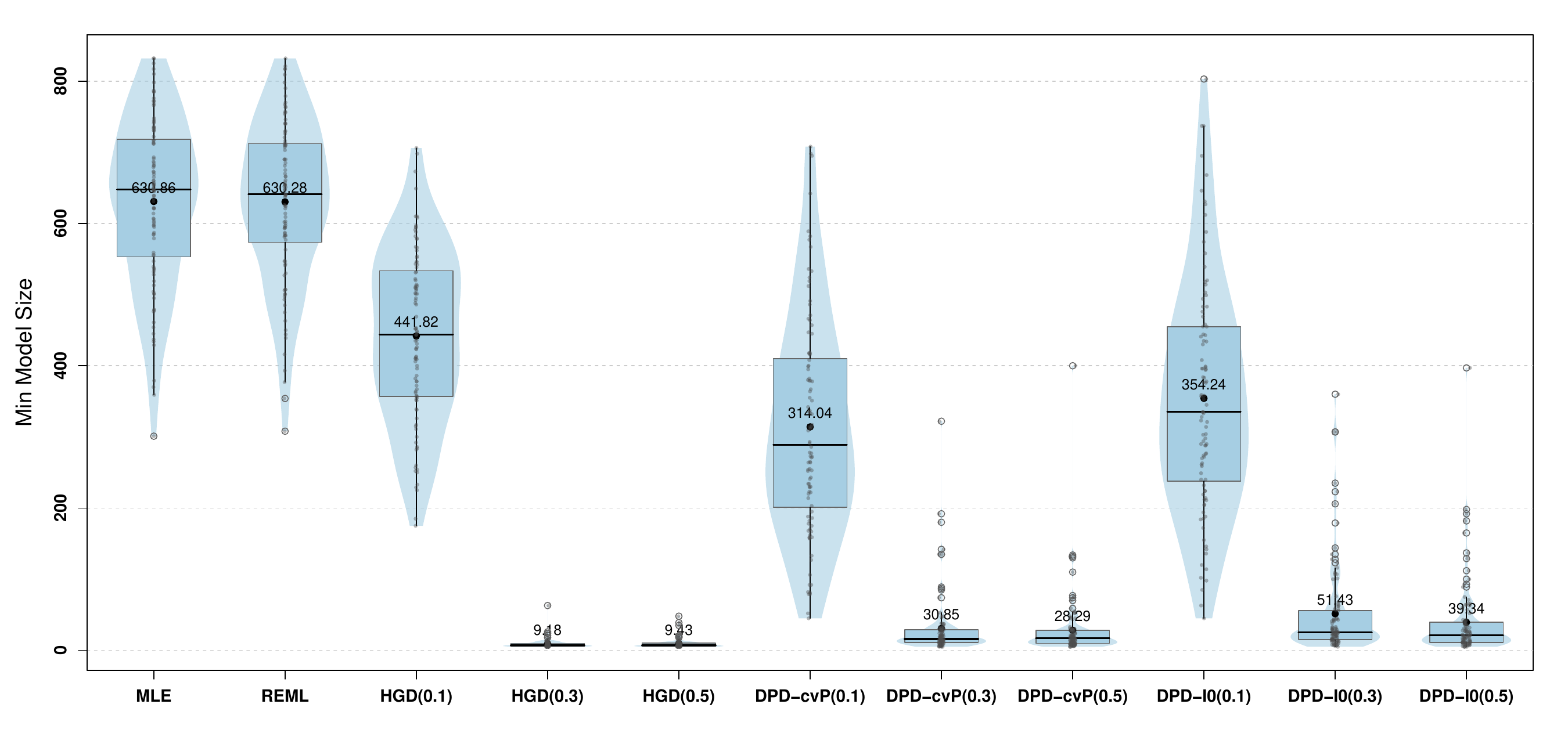}
		\label{FIG:boxplot_Y}}
	~	
	\subfloat[\tiny 20\% (C1) Contamination]{
		\includegraphics[page=2, width=0.3\textwidth]{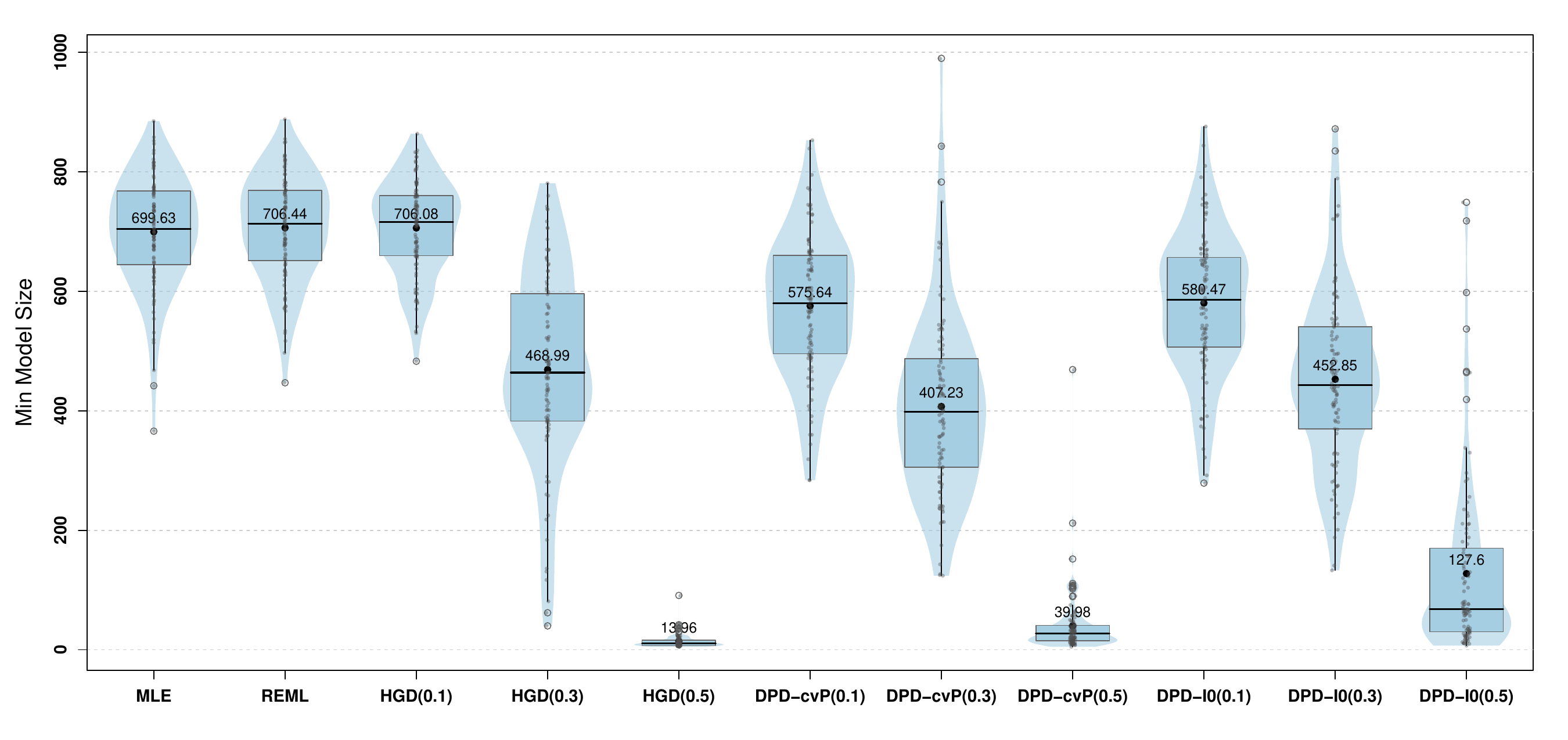}
		\label{FIG:boxplot_Y}}
	\\	
	\subfloat[\tiny 5\% (C2) Contamination]{
		\includegraphics[page=2, width=0.3\textwidth]{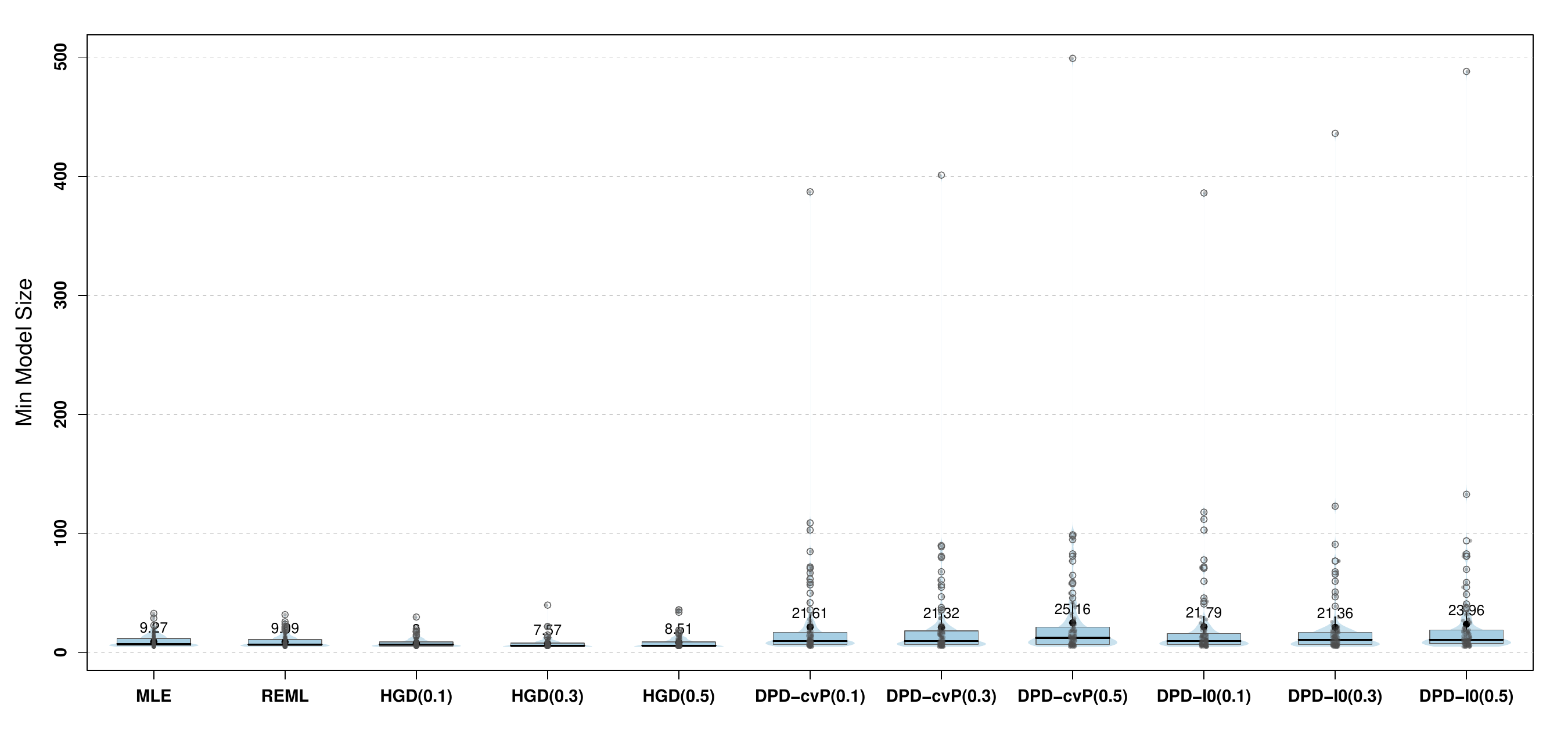}
		\label{FIG:boxplot_Y}}
	~	
	\subfloat[\tiny 10\% (C2) Contamination]{
		\includegraphics[page=2, width=0.3\textwidth]{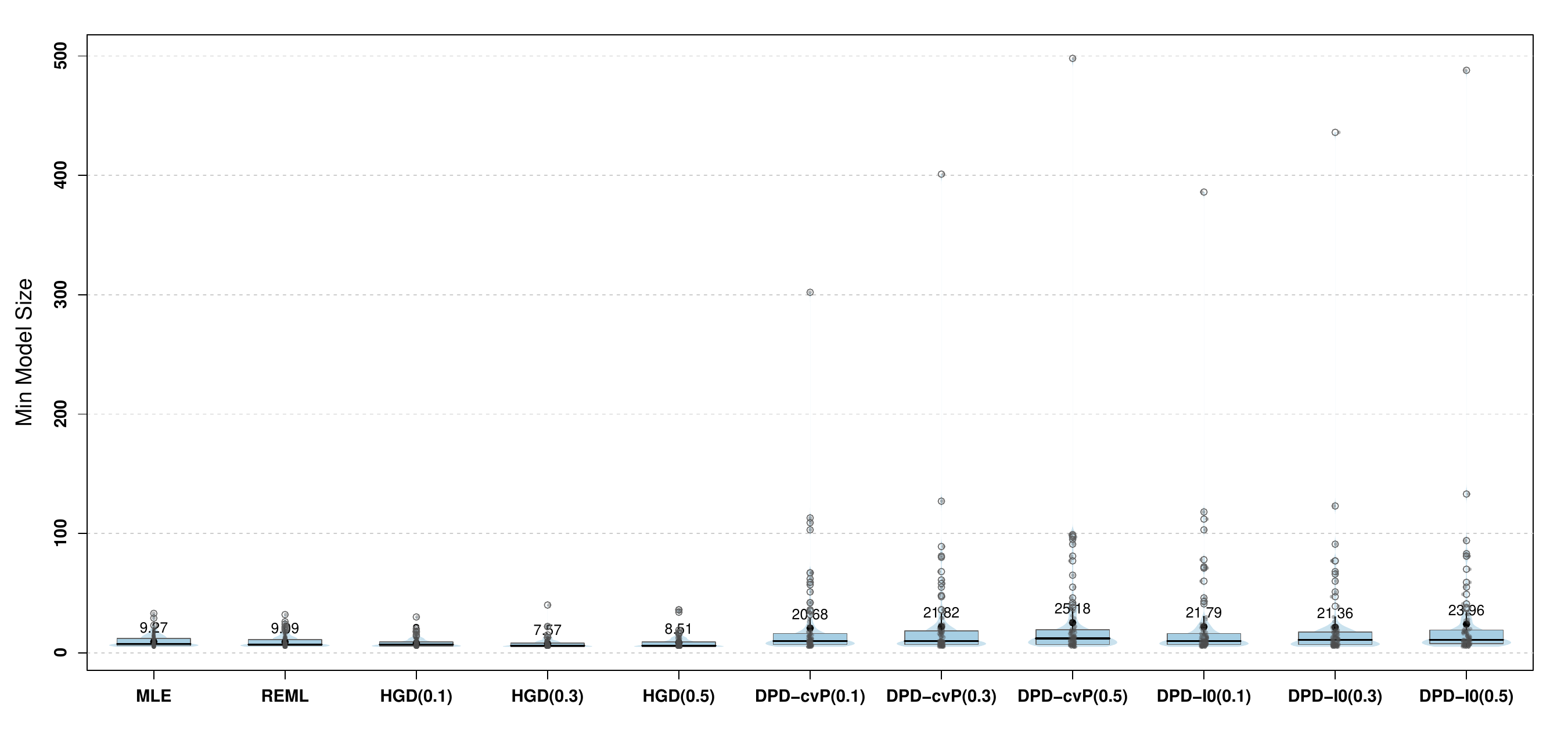}
		\label{FIG:boxplot_Y}}
	~	
	\subfloat[\tiny 20\% (C2) Contamination]{
		\includegraphics[page=2, width=0.3\textwidth]{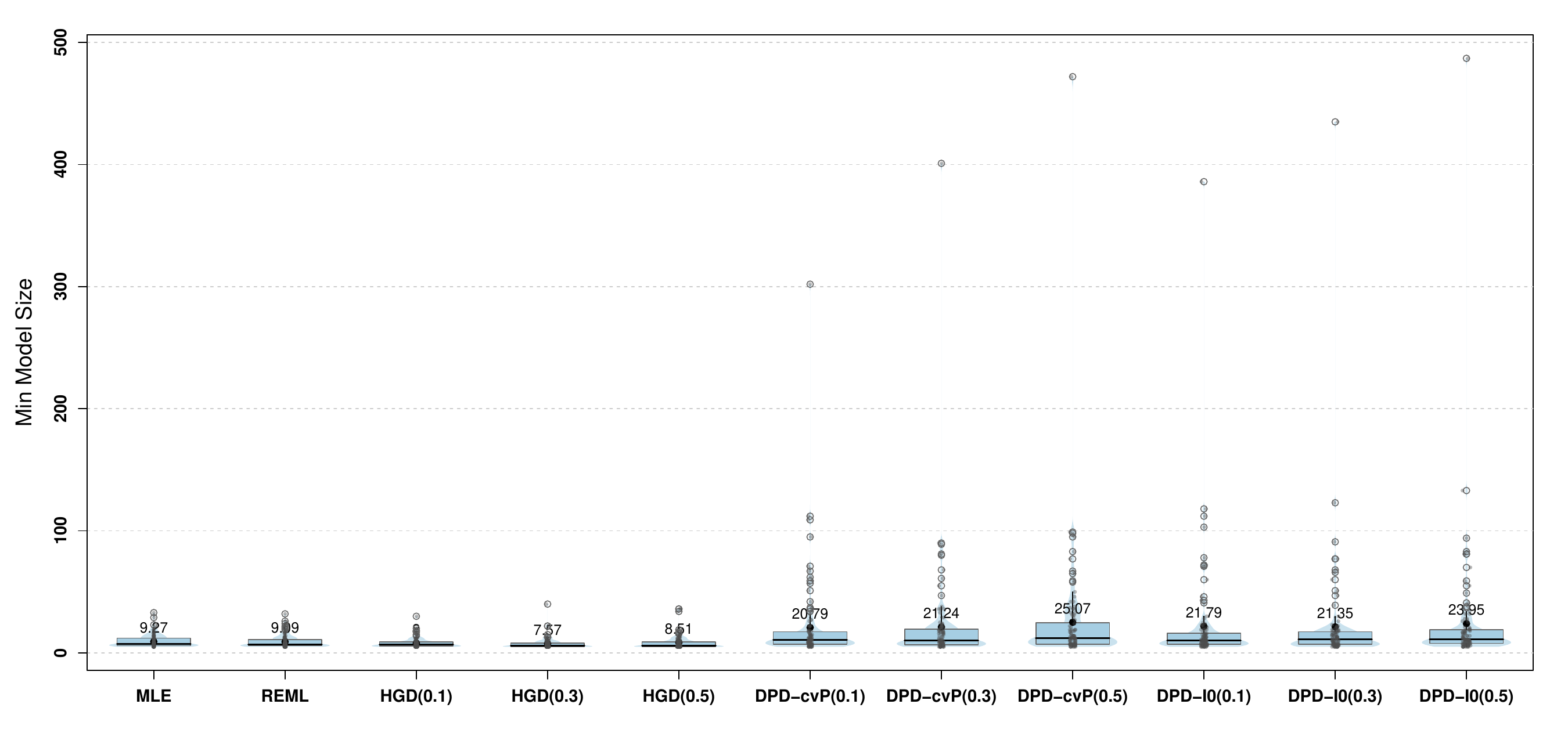}
		\label{FIG:boxplot_Y}}
	\\	
	\subfloat[\tiny 5\% (C3) Contamination]{
		\includegraphics[page=2, width=0.3\textwidth]{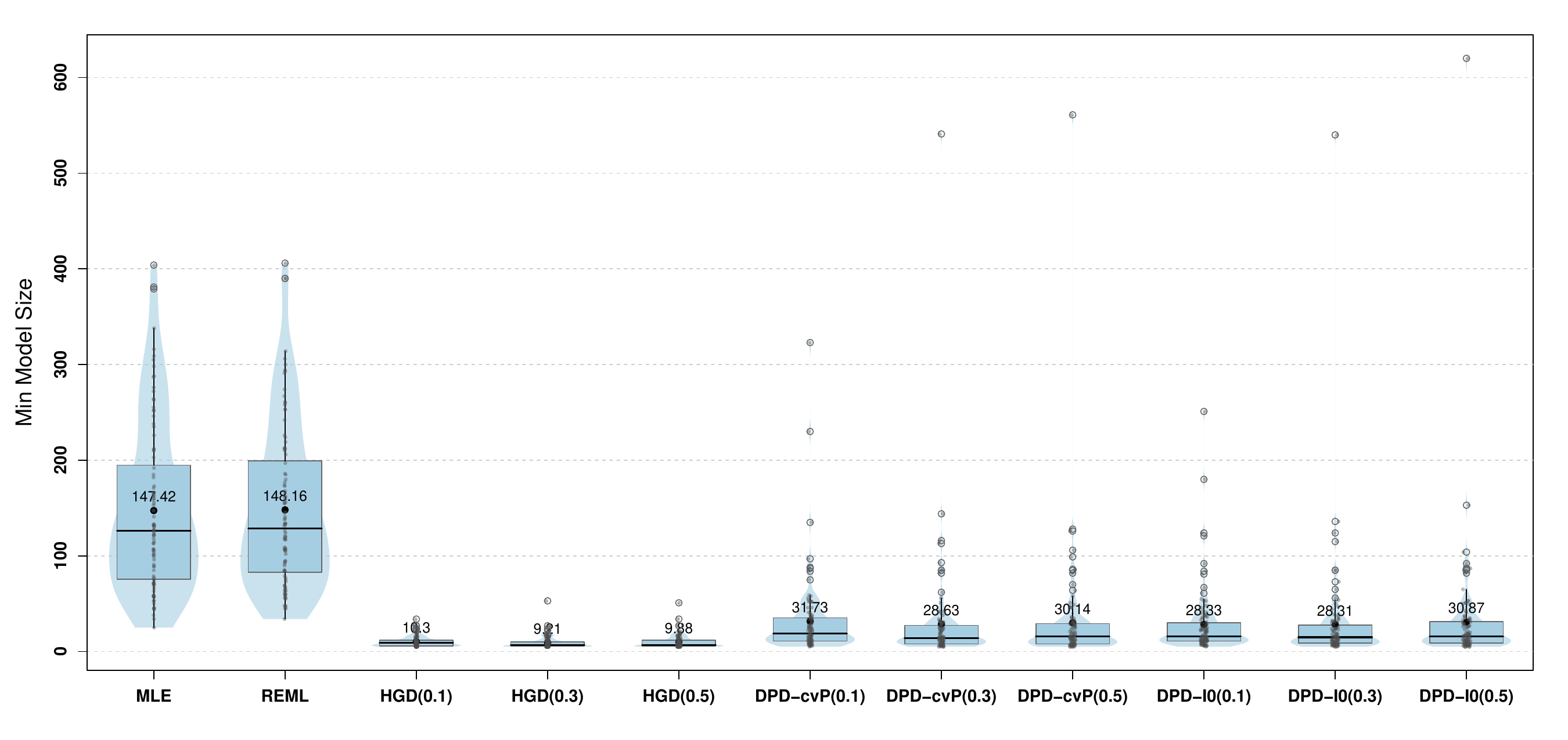}
		\label{FIG:boxplot_Y}}
	~	
	\subfloat[\tiny 10\% (C3) Contamination]{
		\includegraphics[page=2, width=0.3\textwidth]{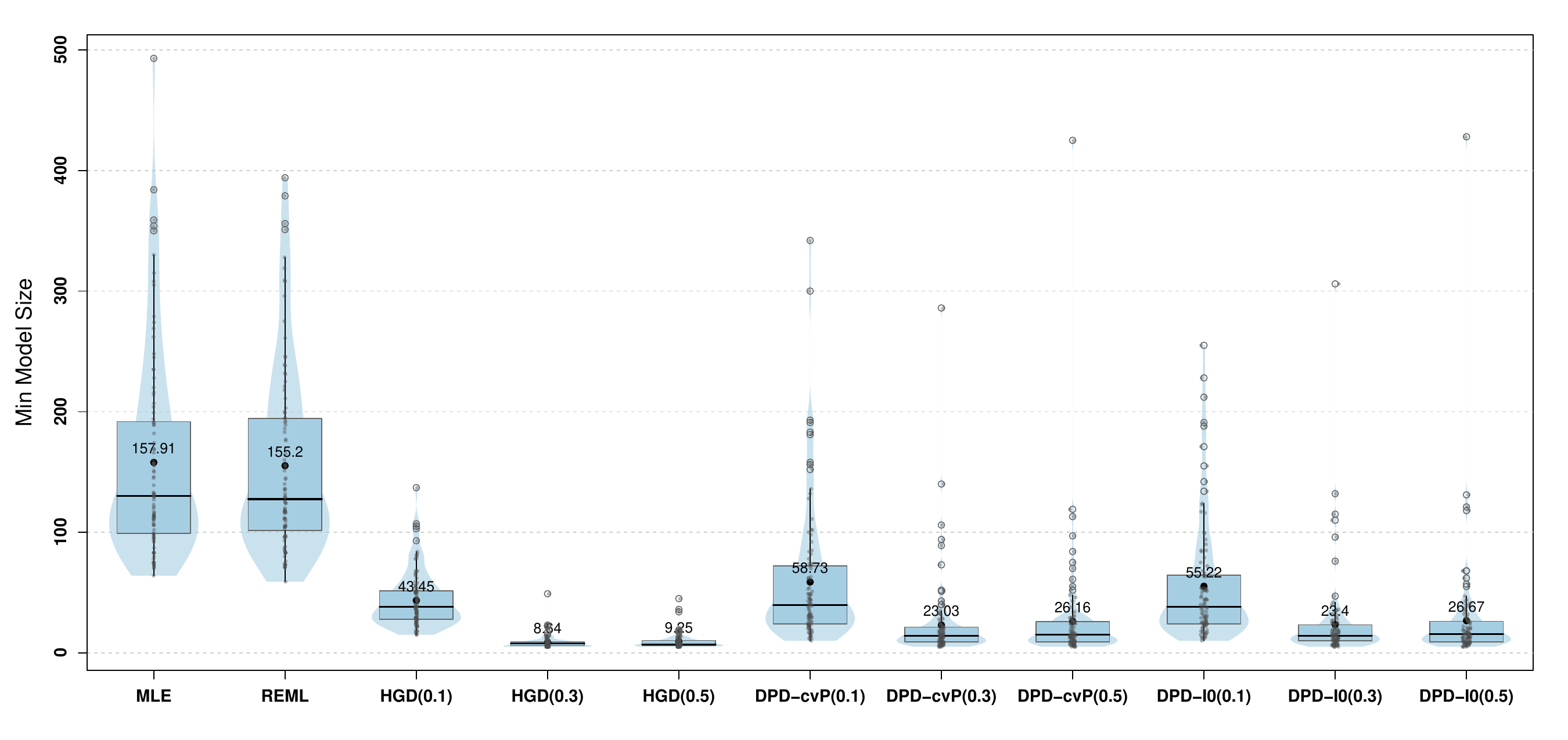}
		\label{FIG:boxplot_Y}}
	~	
	\subfloat[\tiny 20\% (C3) Contamination]{
		\includegraphics[page=2, width=0.3\textwidth]{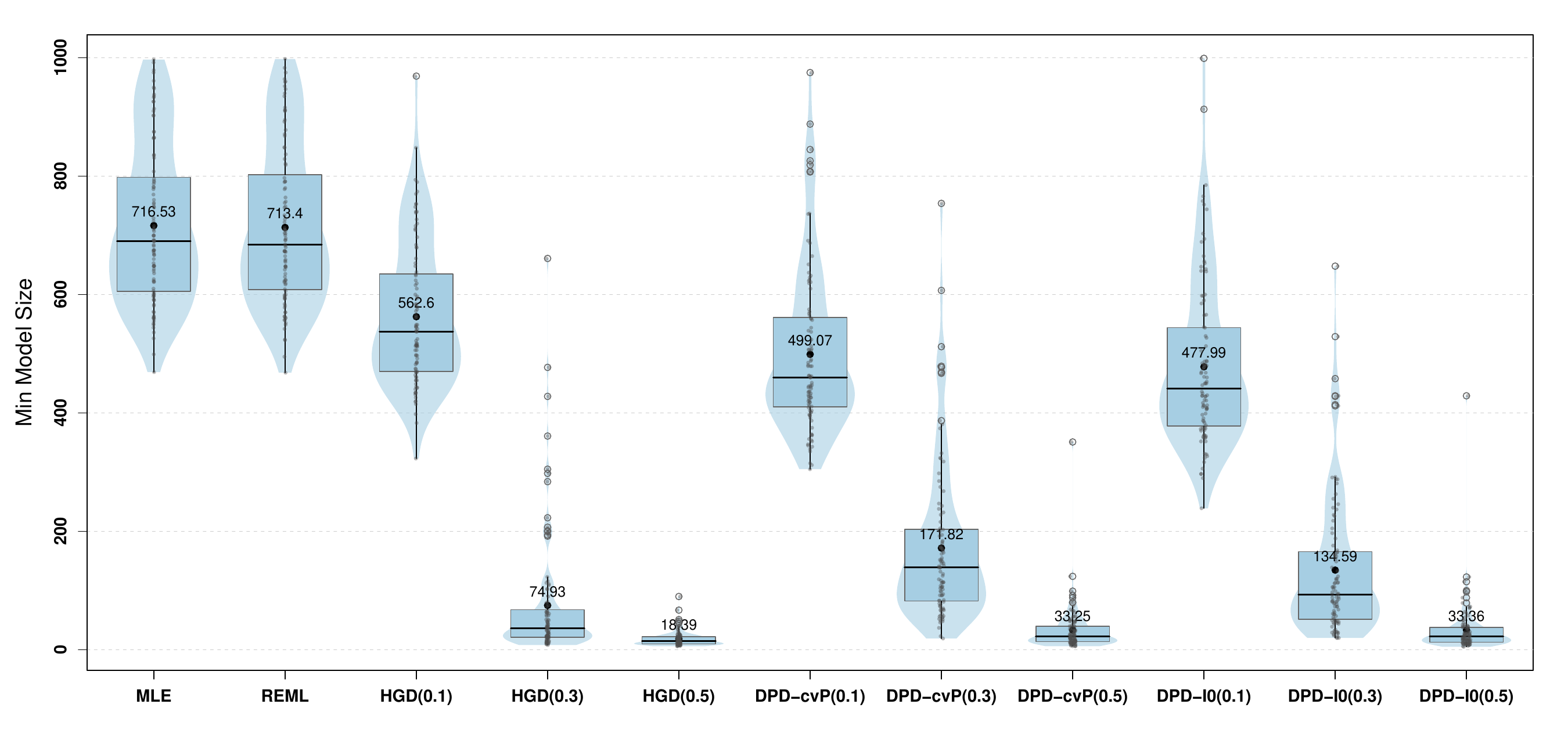}
		\label{FIG:boxplot_Y}}
	\\	
	\subfloat[\tiny 5\% (C4) Contamination]{
		\includegraphics[page=2, width=0.3\textwidth]{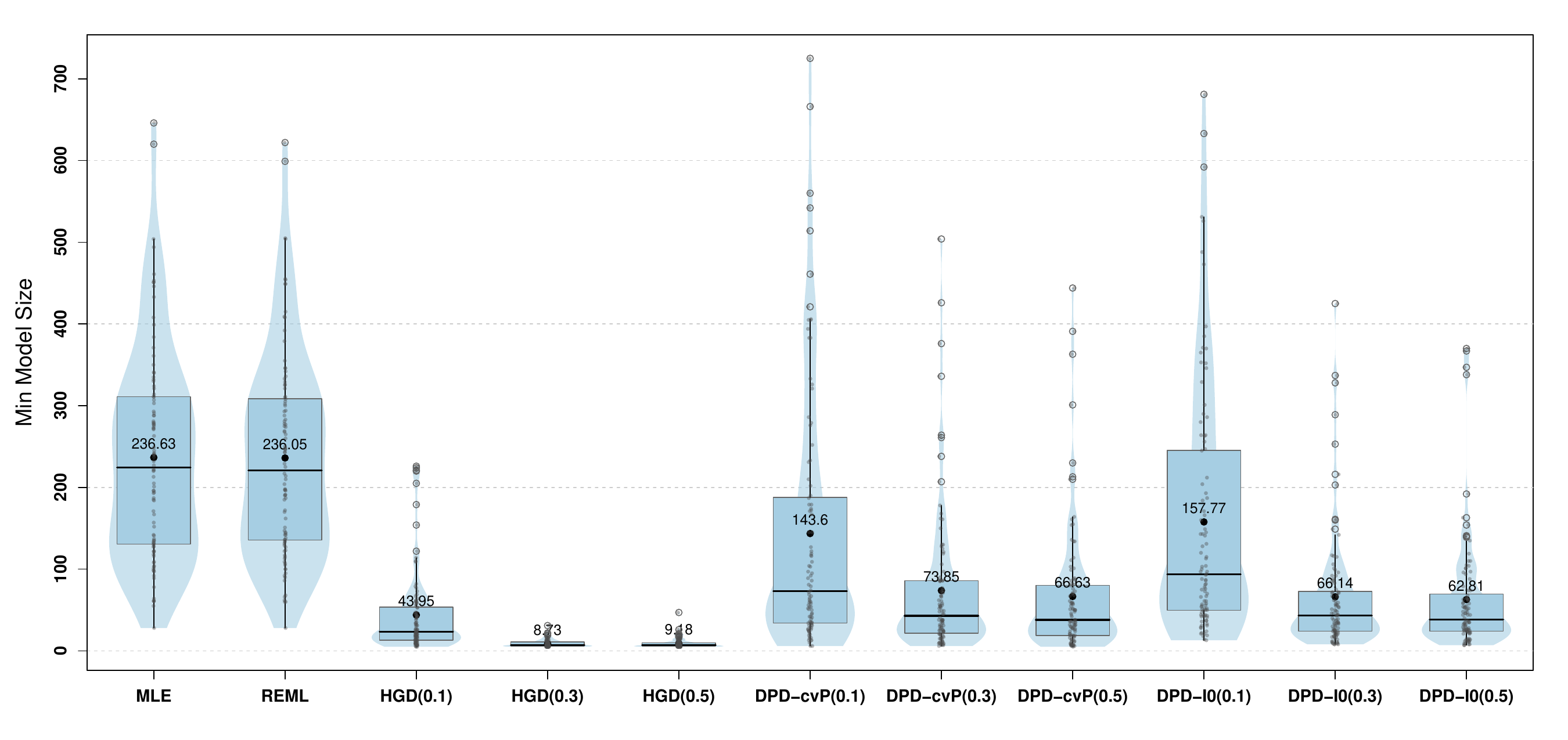}
		\label{FIG:boxplot_Y}}
	~	
	\subfloat[\tiny 10\% (C4) Contamination]{
		\includegraphics[page=2, width=0.3\textwidth]{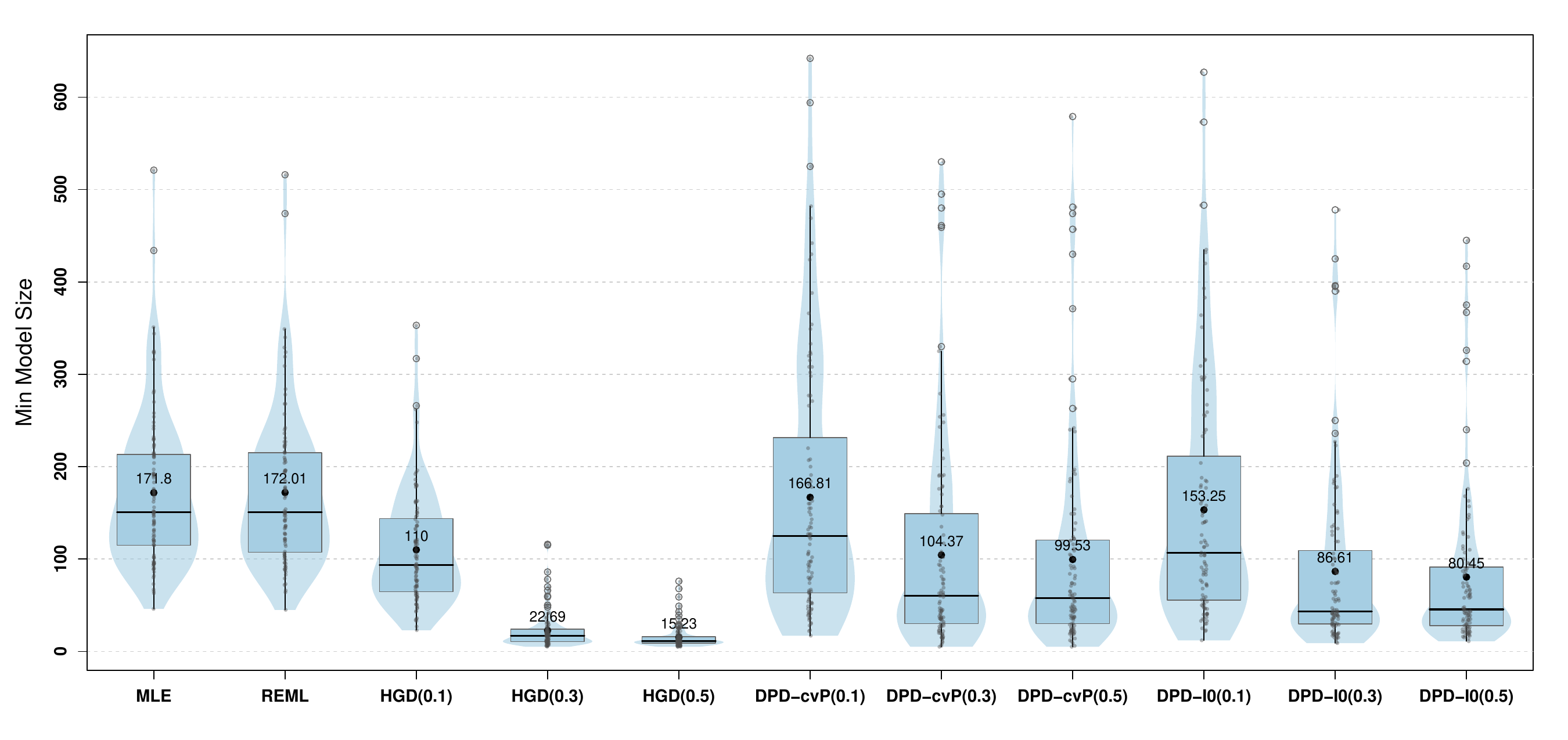}
		\label{FIG:boxplot_Y}}
	~	
	\subfloat[\tiny 20\% (C4) Contamination]{
		\includegraphics[page=2, width=0.3\textwidth]{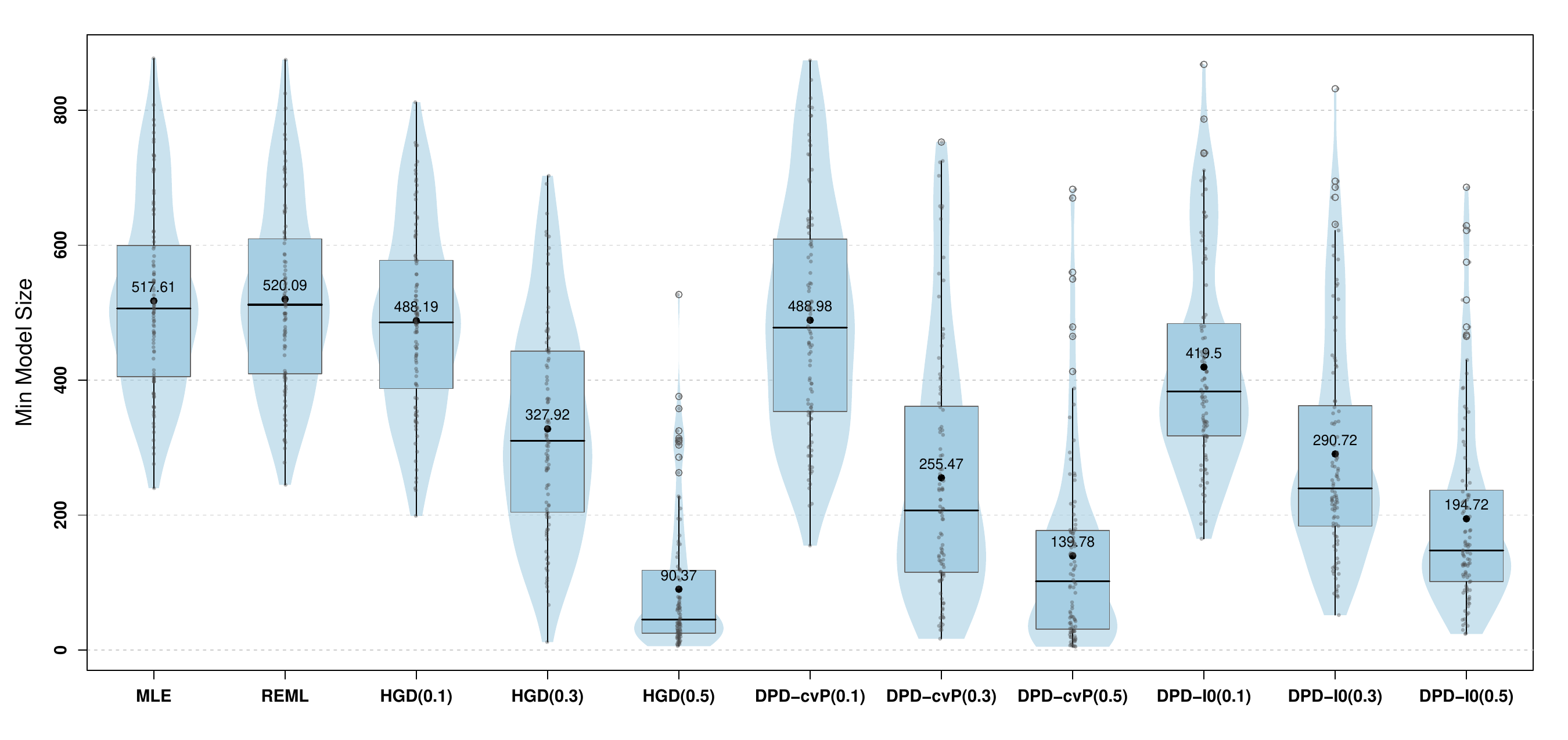}
		\label{FIG:boxplot_Y}}
	\caption{Boxplots  (with overlaid sample means) of MinMS required for sure screening under scenario (S2)$\times$(R0) with $\bm\Sigma_x$ as CS(0.3) and different types of contamination.
		The boxes represent (from left to right) the benchmark SIS based on MLE, REML, and HGD with $\gamma=0.1, 0.3, 0,5$, 
		the DPD-SISP with cv-P with $\alpha=0.1, 0.3, 0.5$, and the DPD-SISP with I0-P with $\alpha=0.1, 0.3, 0.5$, respectively.}
	\label{FIG:MMS_Csp3}
\end{figure}

\bigskip
\noindent\textbf{Performance under Contaminated data:}
\\
Under data contamination (Table \ref{TAB:C1} and Supplementary Tables S1--S3), 
the advantages of the proposed DPD-SISP become substantially more evident. 
Even at 5\% contamination, classical ML- and REML-based SIS procedures exhibit sharp declines in both TPR and EmpSSP, 
particularly under more challenging settings such as (C3) and (C4). 
In many cases, EmpSSP drops to nearly zero despite moderate TPR, indicating highly unstable recovery;
these are further prominent under overlapping effects and dependent covariates (e.g., (S1)$\times$(R1)). 
For instance, under (C3) with 10\% contamination and Identity covariance in (S2)$\times$(R1), 
MLE yields TPR around 0.2 with near-zero EmpSSP, reflecting near-complete failure.

In contrast, DPD-SISP demonstrates strong robustness across contamination levels and types. 
For tuning parameters $\alpha \geq 0.3$, it typically maintains high TPR while achieving substantially improved EmpSSP 
relative to likelihood-based SIS methods. This advantage is consistent across (C1)–(C4), 
though variability increases under the most severe schemes. As contamination intensifies (from 5\% to 20\%), 
classical SIS methods often collapse in both TPR and EmpSSP, whereas DPD-SISP degrades more gradually, 
retaining moderate-to-high TPR and non-negligible EmpSSP even in challenging scenarios involving strong dependence and overlap.
%
Relative to the robust MHGDE benchmark, DPD-SISP achieves comparable performance overall and often higher EmpSSP, 
particularly in complex or heavily contaminated settings. While MHGDE may occasionally attain slightly higher TPR, 
DPD-SISP provides a more stable balance between detection and recovery.
The same is also observed in terms of the MinMS distributions presented in Figure \ref{FIG:MMS_Csp1}--\ref{FIG:MMS_Csp3} 
(and Figures S1--S5 in the Supplementary material);
its values  for DPD-SISP with cv-P proxy are always smaller or competitive  with the requirements of MHGDE based screening 
while being significantly smaller compared to the likelihood based screening procedures.

A comparison between the two proxy choices within DPD-SISP further highlights 
the importance of accurate dependence adjustment under contamination. Across all simulation settings (C1)--(C4), 
the cv-P proxy generally delivers more stable performance than I0-P, particularly in terms of EmpSSP. 
While both proxies benefit from the robustness of the DPD loss and often achieve comparable TPR in moderate settings, 
the advantage of cv-P becomes increasingly pronounced as contamination severity and dependence complexity increase. 
For instance, under (C3) and (C4) with 10\%--20\% contamination, cv-P consistently attains substantially higher EmpSSP than I0-P, 
particularly in settings involving overlapping effects or correlated covariates, 
where the identity proxy fails to adequately capture dependence. 
This gap reflects the ability of cv-P to provide a more accurate whitening transformation, 
thereby stabilizing marginal rankings. In contrast, I0-P, although computationally attractive, 
exhibits greater variability and tends to underperform in challenging scenarios, 
occasionally yielding near-zero EmpSSP despite moderate TPR. 
The MinMS required for DPD-SISP with cv-P proxy are also significantly lower than the same with I0-P proxy 
across most contamination scenarios (Figure \ref{FIG:MMS_Csp1}--\ref{FIG:MMS_Csp3}, and S1--S5 in the Supplementary material). 
These results suggest that while I0-P may serve as a fast baseline, 
the cv-P proxy is preferable in practice when robustness and reliable variable recovery are of primary concern.

\bigskip
\noindent\textbf{Computational efficiency:}
\\
The boxplots of the computational times required for each screening methods is presented in Figure \ref{FIG:Runtime_C0} 
for the uncontaminated cases; the same for contaminated  scenarios are presented in Figures S6--S13 in the Supplementary material. 
We can clearly see that classical ME- and REML-based SIS procedures remain among the fastest. 
However, DPD-SISP with the I0-P proxy achieves comparable runtime, 
while the cv-P version incurs only a moderate additional cost due to proxy estimation. 
Importantly, in several scenarios, both variants of DPD-SISP exhibit runtime comparable to, or even lower than, likelihood-based SIS, 
reflecting the computational burden of repeated mixed-model fitting. 
The MHGDE-based robust benchmark, by contrast, is substantially more computationally intensive, 
limiting its scalability in ultrahigh-dimensional settings.

\begin{figure}[!h]
	\centering
	\subfloat[\tiny (S1)$\times$(R0) with $\bm\Sigma_x$ as Identity]{
		\includegraphics[page=3, width=0.3\textwidth]{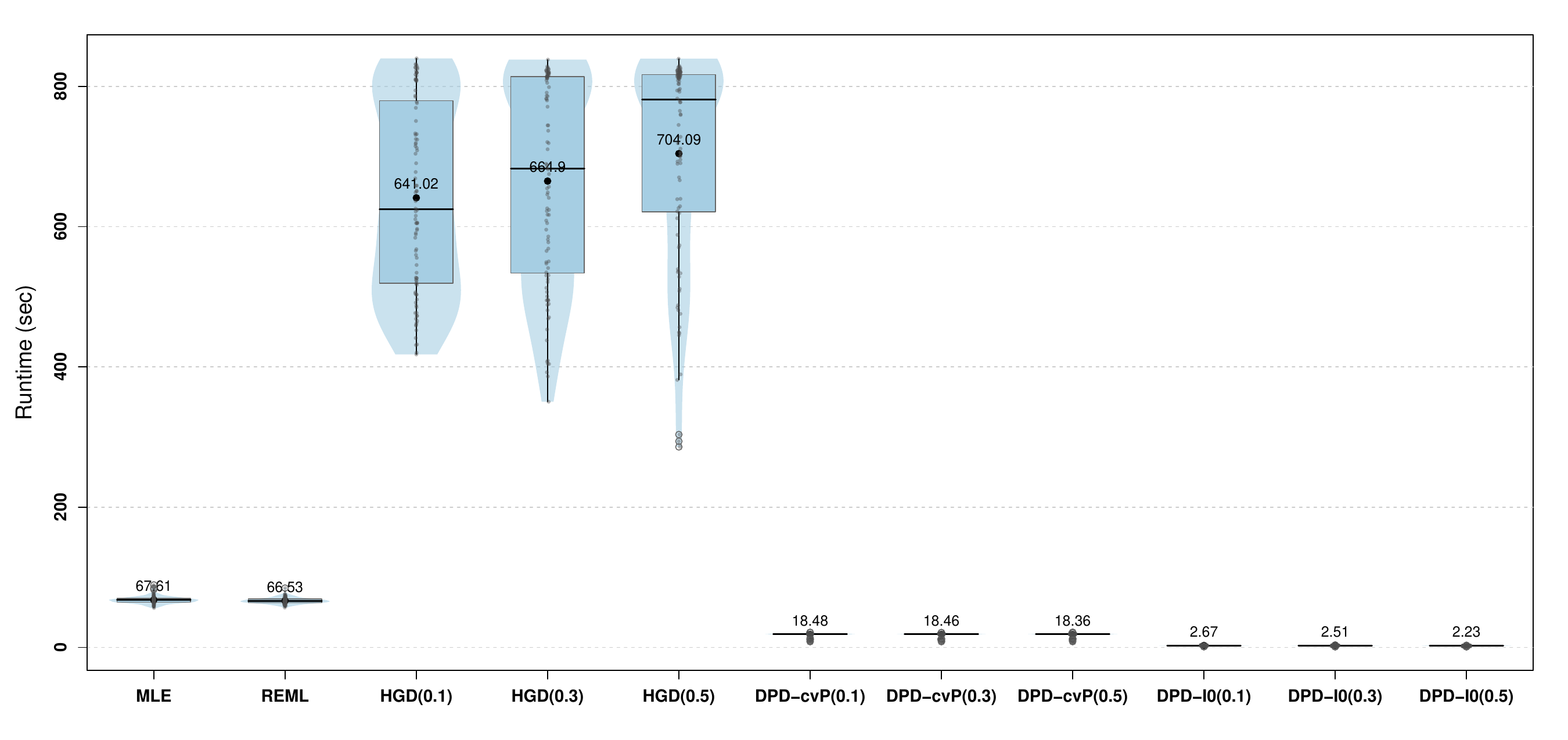}
		\label{FIG:boxplot_Y}}
	~	
	\subfloat[\tiny (S1)$\times$(R1) with $\bm\Sigma_x$ as Identity]{
		\includegraphics[page=3, width=0.3\textwidth]{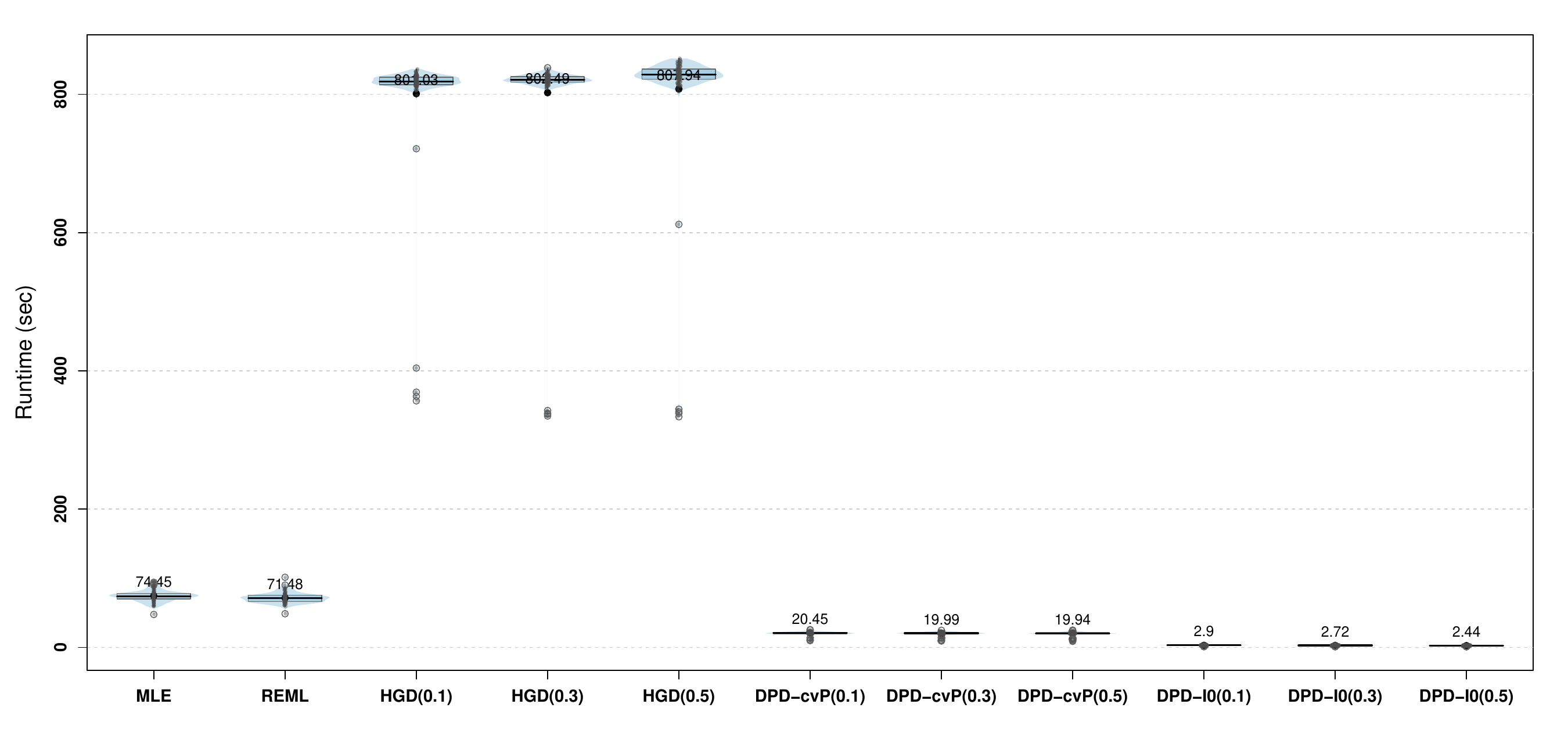}
		\label{FIG:boxplot_Y}}
	~	
	\subfloat[\tiny (S2)$\times$(R1) with $\bm\Sigma_x$ as Identity]{
		\includegraphics[page=3, width=0.3\textwidth]{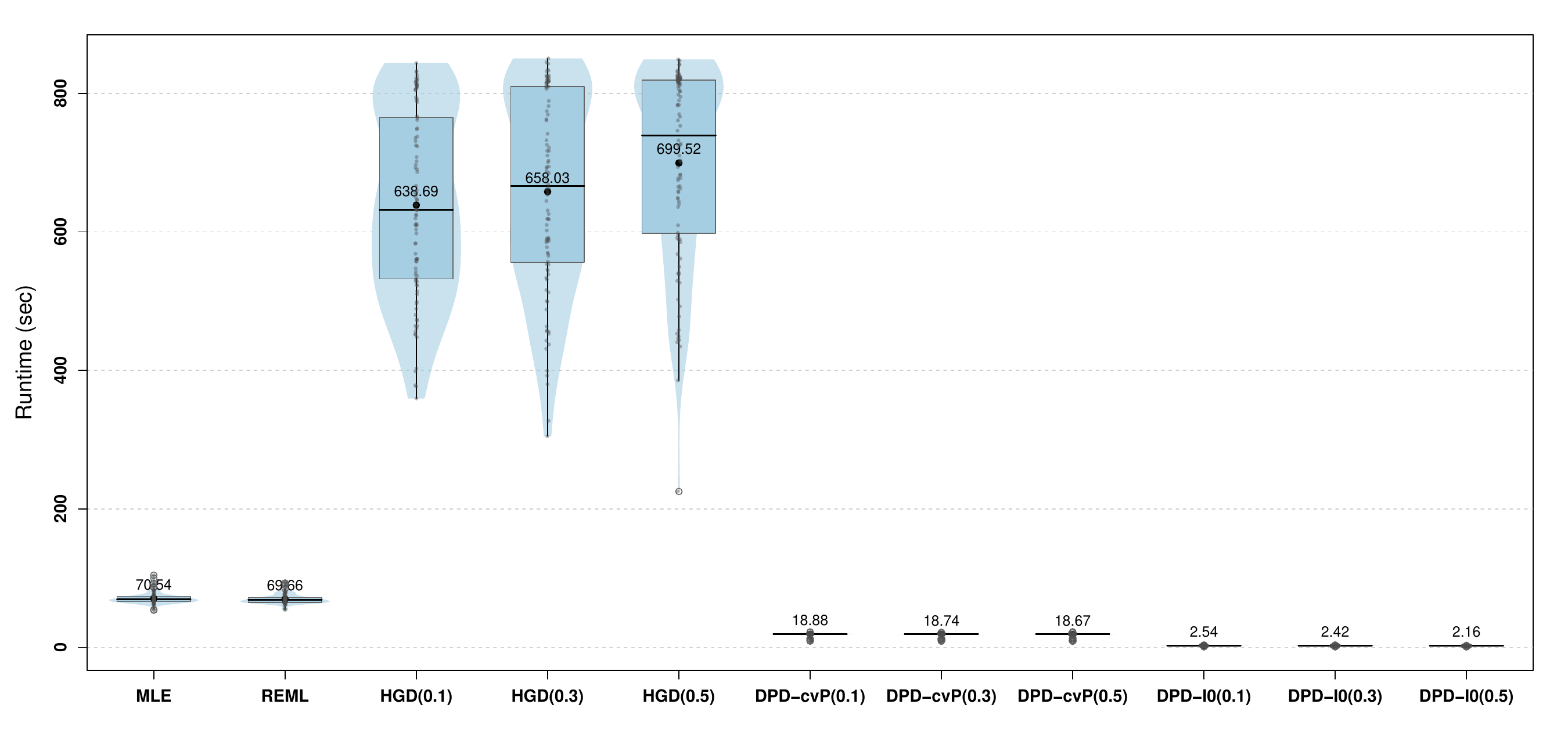}
		\label{FIG:boxplot_Y}}
	\\	
	\subfloat[\tiny (S1)$\times$(R1) with $\bm\Sigma_x$ as CS(0.3)]{
		\includegraphics[page=3, width=0.3\textwidth]{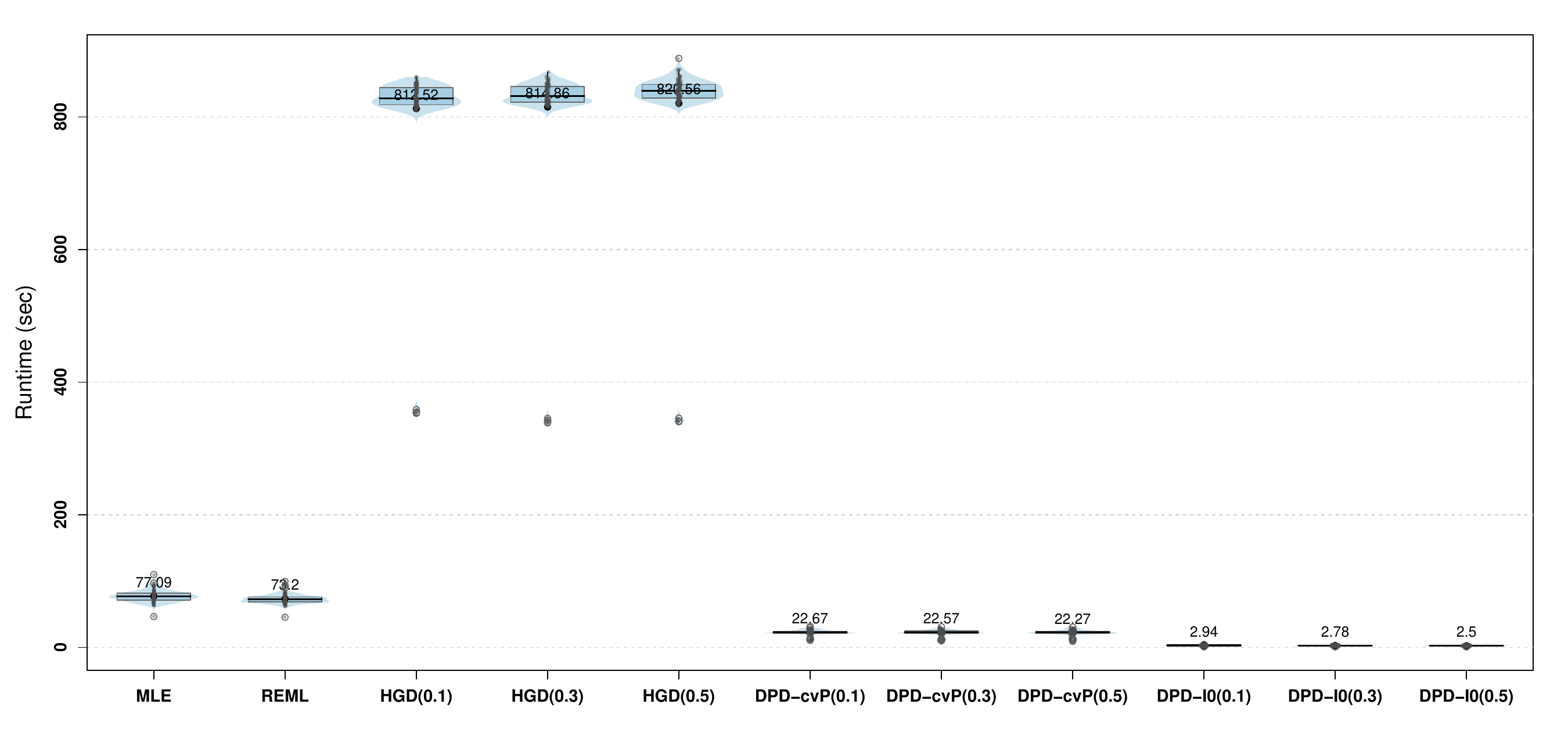}
		\label{FIG:boxplot_Y}}
	~	
	\subfloat[\tiny (S2)$\times$(R0) with $\bm\Sigma_x$ as CS(0.3)]{
		\includegraphics[page=3, width=0.3\textwidth]{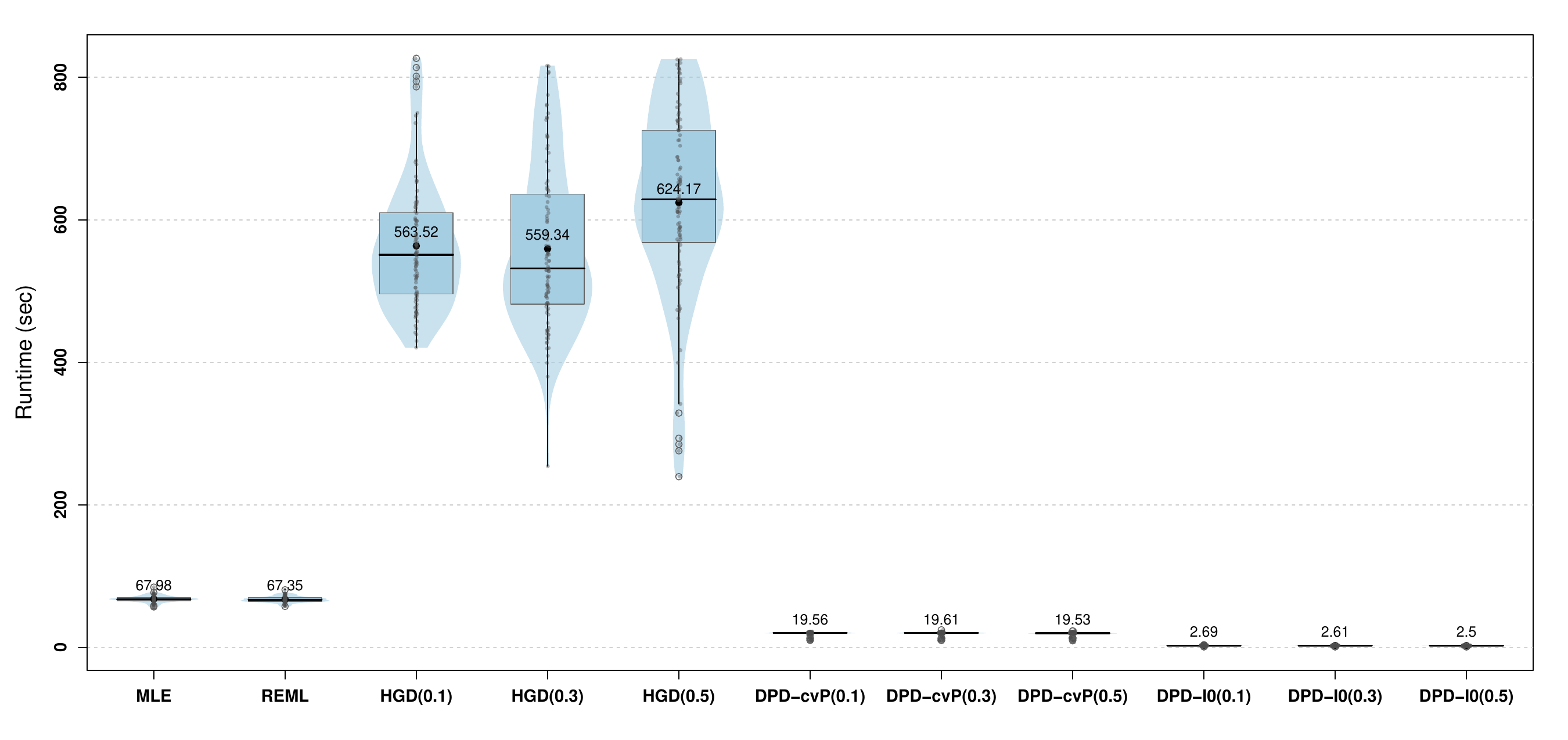}
		\label{FIG:boxplot_Y}}
	~	
	\subfloat[\tiny (S2)$\times$(R1) with $\bm\Sigma_x$ as CS(0.3)]{
		\includegraphics[page=3, width=0.3\textwidth]{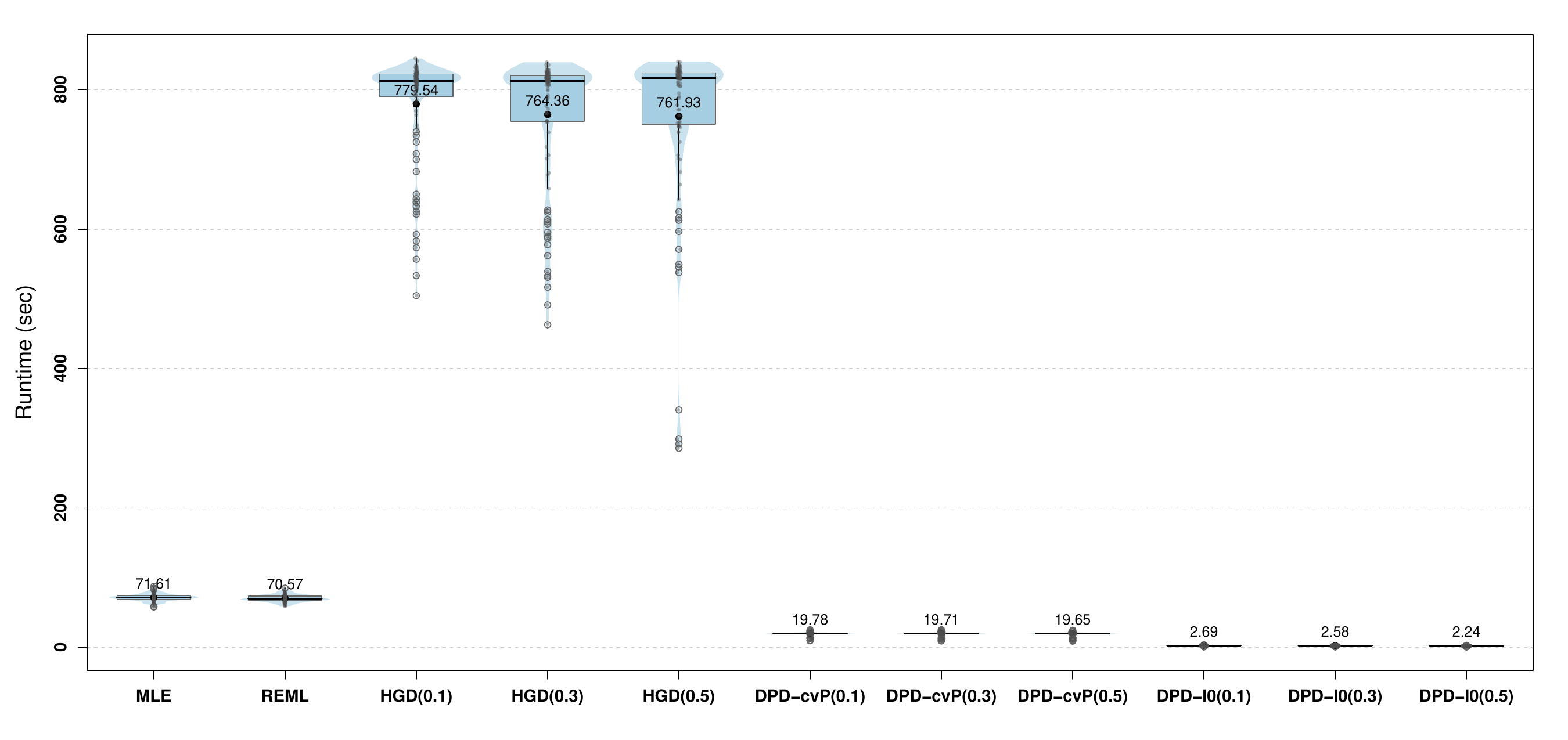}
		\label{FIG:boxplot_Y}}
	\\	
	\subfloat[\tiny (S1)$\times$(R1) with $\bm\Sigma_x$ as T(0.3)]{
		\includegraphics[page=3,width=0.3\textwidth]{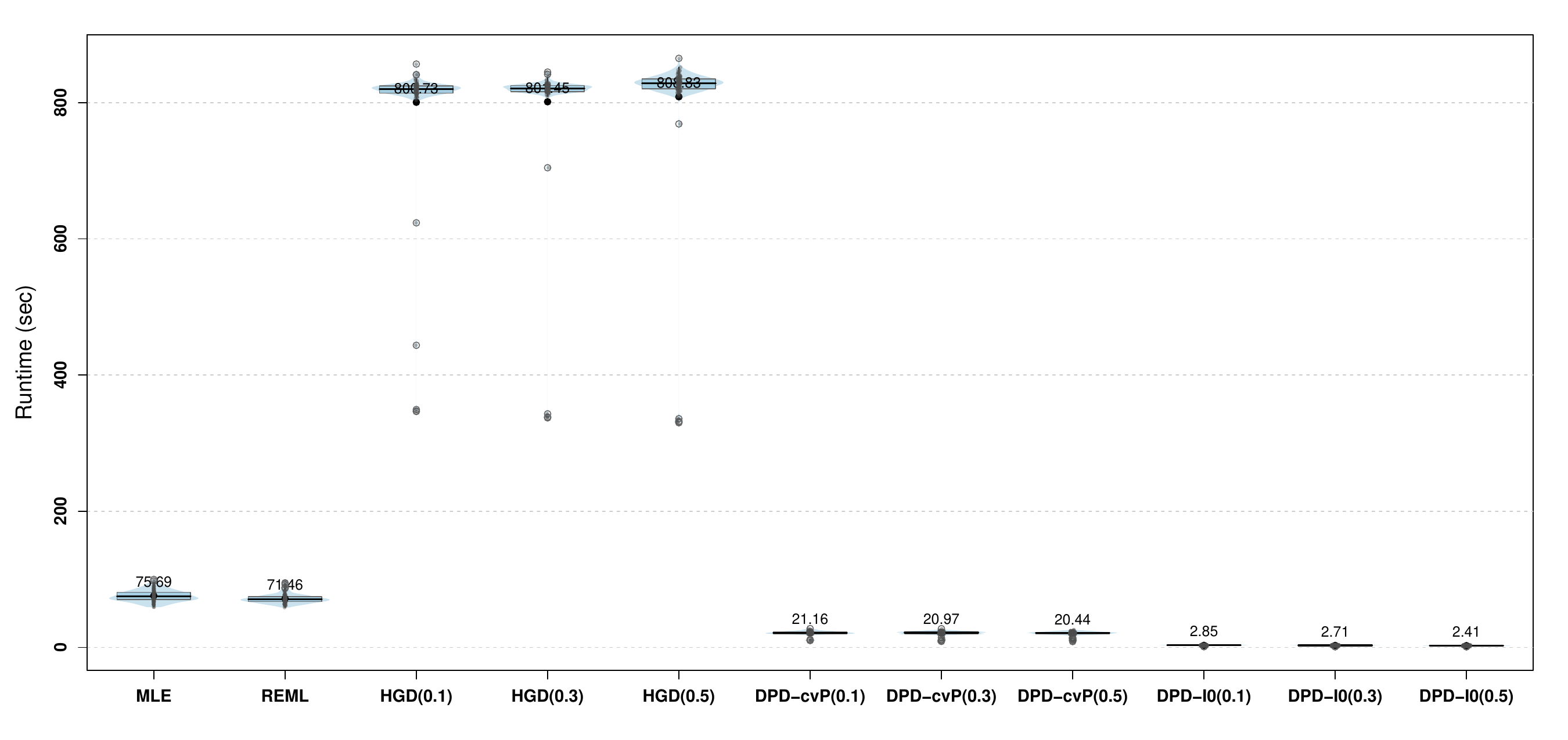}
		\label{FIG:boxplot_Y}}
	~	
	\subfloat[\tiny (S2)$\times$(R0) with $\bm\Sigma_x$ as T(0.3)]{
		\includegraphics[page=3, width=0.3\textwidth]{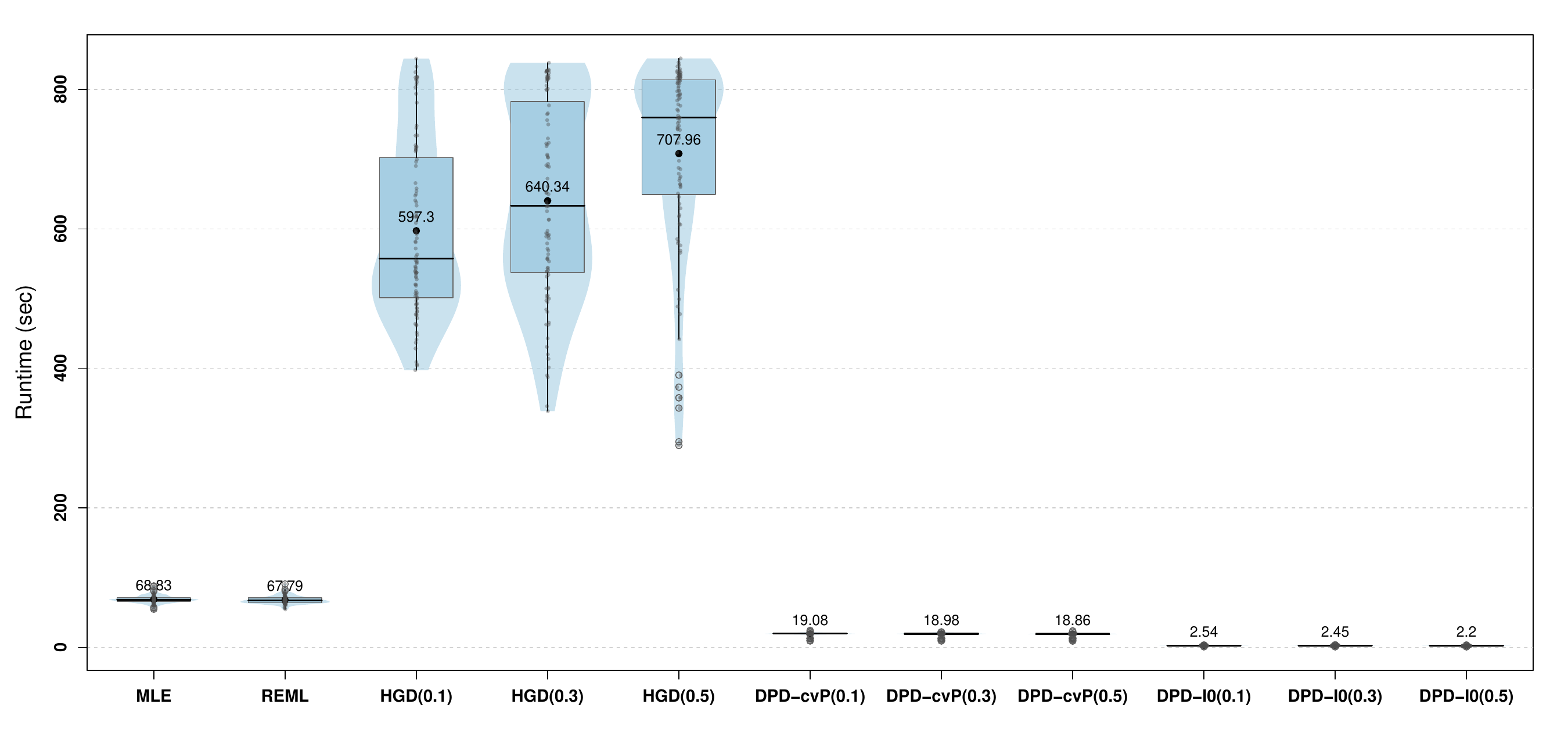}
		\label{FIG:boxplot_Y}}
	~	
	\subfloat[\tiny (S2)$\times$(R1) with $\bm\Sigma_x$ as T(0.3)]{
		\includegraphics[page=3, width=0.3\textwidth]{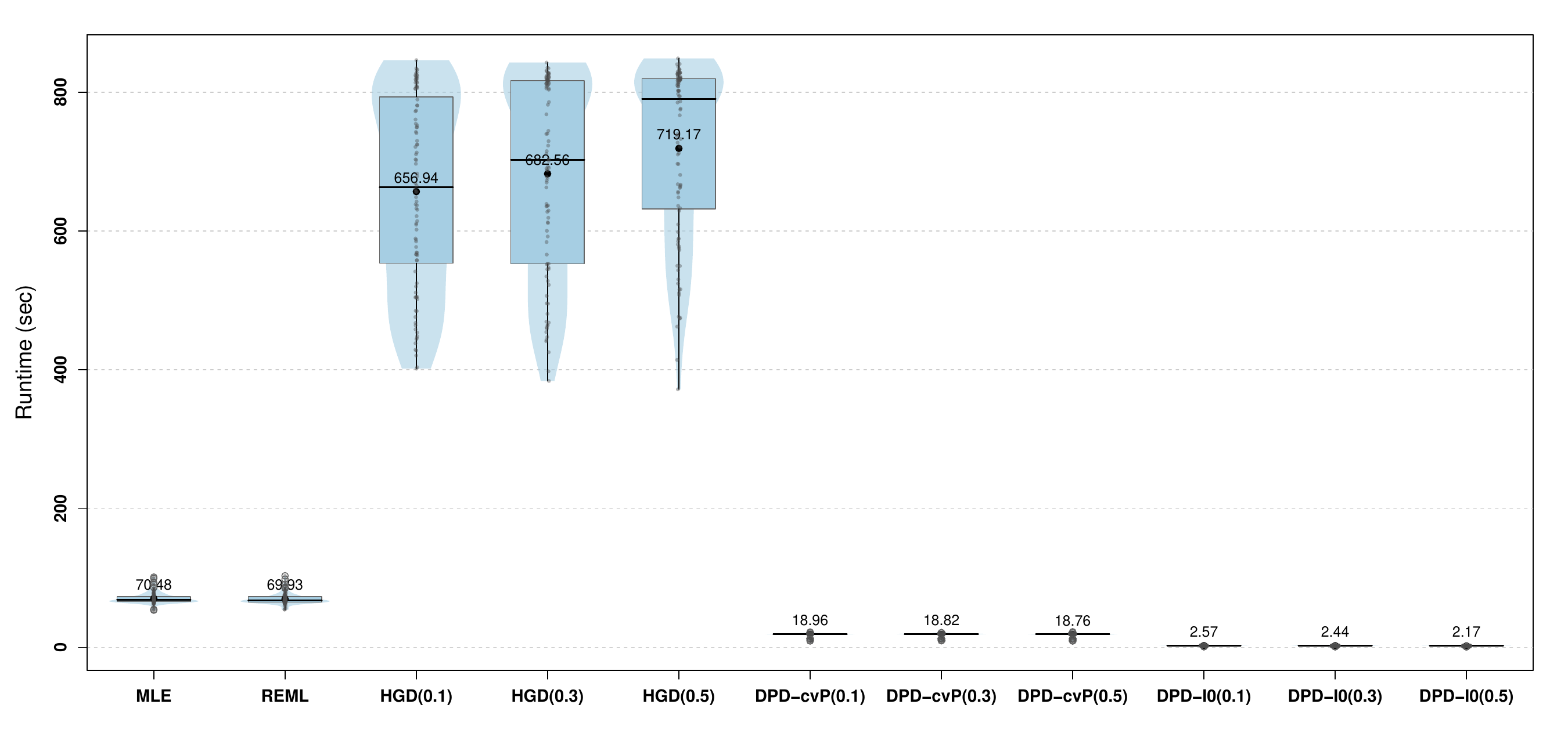}
		\label{FIG:boxplot_Y}}
	\caption{Boxplots  (with overlaid sample means) of computational times (in seconds) of the screening methods under no contamination.
		The boxes represent (from left to right) the benchmark SIS based on MLE, REML, and HGD with $\gamma=0.1, 0.3, 0,5$, 
		 the DPD-SISP with cv-P with $\alpha=0.1, 0.3, 0.5$, and the DPD-SISP with I0-P with $\alpha=0.1, 0.3, 0.5$, respectively.}
	\label{FIG:Runtime_C0}
\end{figure}

Overall, the results demonstrate that DPD-SISP, particularly with the cv-P proxy, 
provides a favorable balance between robustness, accuracy, and computational efficiency across both clean and contaminated settings. 
Although performance degrades under severe dependence and heavy contamination, 
it remains consistently more stable than classical SIS methods and competitive with more computationally demanding robust alternatives.
These findings support DPD-SISP as a reliable screening framework for ultrahigh-dimensional LMMs in the presence of data contamination.


\section{A Real Data Application}
\label{SEC:data}

\subsection{The ADNI-2 Data}

We illustrate the proposed DPD-SISP procedure using a subset of data from Phase 2 of the Alzheimer's Disease Neuroimaging Initiative (ADNI-2), 
a multicenter longitudinal study designed to identify biomarkers associated with the onset and progression of Alzheimer's disease \citep{mueller2005ways,mueller2005alzheimer,weiner2013alzheimer}. 
ADNI-2 extends the original ADNI cohort through additional recruitment, longitudinal follow-up assessments, and a collection of multimodal biomarkers, 
including neuroimaging, genomic, transcriptomic, metabolomic, and clinical measurements (\url{https://adni.loni.usc.edu/data-samples/adni-data/}). 

For the present analysis, longitudinal cognitive assessments were integrated with high-dimensional gene-expression measurements (at baseline)
obtained through the \texttt{R} package \texttt{ADNIMERGE} \citep{ADNIMERGE}. 
The response variable is taken to be the repeated Mini-Mental State Examination (MMSE) score \citep{folstein1975mini}, a widely used measure of global cognitive function,
while the FE covariates consist of 49,386 gene-expression probes together with time. 
Subject-specific dependence induced by repeated observations is modeled through a LMM with random intercept and random slope for time. 
Participants were identified using unique ADNI research identifiers (RID), 
and follow-up times were determined from the ADNI visit schedule (e.g., 6-, 12-, and 24-month assessments).

To ensure a reasonably balanced longitudinal structure and avoid excessive sparsity in the repeated measurements, 
we restricted attention to participants having either three or four visits. 
The resulting sample comprised 343 individuals, of whom 41 contributed three observations and 
302 contributed four observations, yielding a total of $n=1331$ observations.
The resulting dataset is ultrahigh-dimensional longitudinal, with $p=49,387$ candidate FE variables (probes plus time),  
substantially exceeding the sample size $n$. 
Direct model fitting in such settings is computationally challenging,
necessitating the use of an efficient screening procedure as a preliminary dimension-reduction step.

\subsection{Screening Analysis}

We applied the proposed DPD-SISP procedure using both proxy matrices, namely the cv-P and I0-P, with tuning parameter $\alpha=0.3$,
to screen the top $d=\lfloor n/\log(n)\rfloor=185$ FE variables. 
For comparison, we also applied standard SIS based on ML and REML estimation under marginal LMMs;
both approaches produced identical screened feature sets and are therefore reported jointly as the benchmark result. 
The HGD-based SIS procedure was computationally infeasible for this application due to the extremely large number of candidate covariates 
and could therefore not be performed.

The screened probes were subsequently mapped to known genes.
The resulting gene sets obtained from the benchmark ML/REML-based SIS procedure 
and the DPD-SISP(0.3) with cv-P and I0-P proxies are denoted by \texttt{Gene\_Set\_0}, 
\texttt{Gene\_Set\_1(cv)}, and \texttt{Gene\_Set\_1(I0)}, respectively. 
These sets contained 162, 165, and 165 unique genes, respectively.

The two implementations of DPD-SISP produced highly similar gene lists \texttt{Gene\_Set\_1(**)}, 
whereas both differed substantially from the benchmark screening results (\texttt{Gene\_Set\_0});
the complete lists are reported in Supplementary Table S4.
Such discrepancies are not unexpected in high-dimensional transcriptomic studies, 
where strong association among genes often allow different screening procedures to identify 
distinct sets of predictors with comparable explanatory power while representing the same underlying biological processes.  
Consistent with this, stability selection analyses \citep{meinshausen2010stability} 
yielded relatively modest stability proportions across all three screening approaches, 
with the DPD-SISP with I0-P proxy exhibiting the highest stability.
These factors make it inherently challenging to directly compare the screened gene sets, 
thereby motivating external biological validation analyses to assess the relevance of the identified candidates.


\subsection{External Biological Validation}

To assess the biological relevance of the screened genes, 
we conducted an external phenotype-driven validation using \textit{VarElect} \citep{stelzer2016varelect}, 
a gene-prioritization framework within the \textit{GeneCards} Suite \citep{stelzer2016genecards}. 
VarElect integrates information from multiple curated resources, 
including GeneCards, MalaCards, PathCards, and molecular interaction databases, 
to rank genes according to their relevance to a specified phenotype or disease term \citep{stelzer2016varelect}. 
For a queried phenotype, the platform assigns a phenotype relevance (PR) score 
and classifies evidence as either direct or indirect. 
Direct associations correspond to explicit links between a gene and the queried phenotype, 
whereas indirect associations are inferred through intermediate genes, shared pathways, functional annotations, or network-based relationships.
%
For genes with direct associations, VarElect additionally reports a (pseudo) $p$-value based on 
the empirical distribution of PR scores across genes with sufficient annotation coverage (GeneCards GIFtS $\ge 10$). 
Although the platform relies on existing biological knowledge and cannot establish novel causal relationships, 
it provides a useful external benchmark for evaluating whether screened genes are enriched for previously documented biological evidences.

In this study, using ``Alzheimer'' as the target phenotype term, 
we categorized screened genes into six groups according to the strength and nature of their VarElect associations:
\begin{itemize}
	\item G1: Highly significant direct association, having p-values < 0.01.
	\item G2: Significant direct association, having p-values between 0.01 to 0.05.
	\item G3: Moderate direct association, having p-values between 0.05 and 0.1.
	\item G4: Weak direct association, having p-values > 0.1.
	\item G5: Indirect association	
	\item G6: No documented direct or indirect association.
\end{itemize}
The resulting distributions of screened genes across these categories are summarized in Table \ref{TAB:data_no}, 
together with the average PR scores. Complete gene lists and individual PR scores are provided in Supplementary Table S4.

\begin{table}[!h]
	\centering\small
	\caption{Numbers of genes and average PR scores (in parenthesis) across different groups of association, 
		obtained by the three screening methods}
		\begin{tabular}{ll|c|c|c}\hline
Assoc. type & Group	&	\texttt{Gene\_Set\_0}	&	\texttt{Gene\_Set\_1(cv)}	&	\texttt{Gene\_Set\_1(I0)}	
\\\hline\hline
\multirow{4}{*}{Direct} &  G1	&	1 (10.78)	&	3 (13.67)	&	4 (19.8)	\\
&G2	&	17 (5.32)	&	16 (4.57)	&	17 (4.69)	\\
&G3	&	11 (2.22)	&	15 (2.41)	&	14 (2.4)	\\
&G4	&	73 (0.65)	&	65 (0.76)	&	67 (0.78)	\\\hline
Indirect &G5	&	43 (1.14)	&	53 (1.16)	&	53 (1.28)	\\\hline
No &G6	&	17	&	13	&	10	\\\hline
& Total	&	162	&	165	&	165	\\ \hline
		\end{tabular}
	\label{TAB:data_no}
\end{table}

To further assess the biological plausibility of the screened genes, 
we examined pathway annotations and disease associations available through the \textit{GeneCards} Suite. 
Specifically, for genes exhibiting at least a moderate direct association (G1--G3), 
we recorded pathway annotations and publications containing explicit references to Alzheimer's disease (AD).
The corresponding results are summarized in Table \ref{TAB:data_GO}, 
with genes grouped according to their disease associations reported in the \textit{MalaCards} database \citep{rappaport2017malacards}.
Notably, many of the prioritized genes show direct associations with AD or epilepsy (a common comorbidity of AD), 
while their inferred associations extend to a broader range of neurological and neurodegenerative disorders, 
including dementia, Down syndrome, nervous system diseases, and familial forms of AD.


It should be noted that the phenotype-driven validation is intended only as an external assessment of disease relevance 
and not as a definitive measure of screening accuracy. Because MMSE decline is studied within an AD-focused cohort, 
validation based on the phenotype term ``Alzheimer'' is not entirely independent of the underlying scientific objective. 
Nevertheless, the use of external knowledge bases provides a useful benchmark for comparing the disease relevance of competing screened gene sets.

\setlength{\LTcapwidth}{\textwidth}
\begin{small}
\begin{longtable}{l|l|lll|ll}
	

\caption{Pathway and disease association analyses of prioritized genes (group G1--G3) from the three screened gene sets,
along with their PR scores and numbers of publications (No. Pub) with '`Alzheimer' annotations.}  \\
\hline
Gene	&	Biological 	&	\multicolumn{3}{c|}{\texttt{Gene\_Set\_}}	&	No.	&	PR \\
symbol	&	Pathway	Type &	\texttt{0}	&	\texttt{1(cv)}	&	\texttt{1(I0)}&	Pub	&	Score	\\
\hline
\endfirsthead

\multicolumn{2}{c}{\tablename\ \thetable{} -- Continued from previous page} \\
\hline
Gene	&	Biological 	&	\multicolumn{3}{c|}{\texttt{Gene\_Set\_}}	&	No.	&	PR	\\
symbol	&	Pathway	Type &	\texttt{0}	&	\texttt{1(cv)}	&	\texttt{1(I0)}&	Pub	&	Score	\\
\hline
\endhead

\hline
\multicolumn{2}{r}{Continued on next page...} \\
\endfoot

\hline
\endlastfoot

\hline
\multicolumn{7}{l}{\textbf{Direct association with AD}}\\
VCP	&	Regulatory	&	--	&	--	&	\checkmark	&	39	&	33.29	\\
LMNB1	&	Disease	&	--	&	\checkmark	&	\checkmark	&	9	&	22.08	\\
SNAP25	&	Neuronal System	&	--	&	--	&	\checkmark	&	107	&	13.88	\\
CBS	&	Metabolism	&	--	&	\checkmark	&	\checkmark	&	45	&	9.96	\\
DNM1	&	Signaling, Regulatory 	&	--	&	\checkmark	&	\checkmark	&	11	&	6.94	\\
INSR	&	Disease, Signaling	&	\checkmark	&	--	&	--	&	33	&	10.78	\\
AKT2	&	Disease, Signaling	&	\checkmark	&	--	&	--	&	3	&	6.07	\\
\hline
\multicolumn{7}{l}{\textbf{Inferred association with AD}}\\
HMOX1	&	Regulatory	&	\checkmark	&	--	&	--	&	216	&	8.25	\\
BCL2	&	Signaling, Regulatory 	&	--	&	\checkmark	&	--	&	86	&	8.97	\\
DRD5	&	Signaling, Signal Transduction	&	\checkmark	&	--	&	--	&	16	&	4.18	\\
MSR1	&	Signaling	&	\checkmark	&	--	&	--	&	13	&	6.71	\\
NR4A1	&	Signaling	&	\checkmark	&	--	&	--	&	12	&	4.11	\\
CHGB	&	Metabolism	&	\checkmark	&	--	&	--	&	13	&	6.29	\\
PCMT1	&	Metabolism	&	--	&	\checkmark	&	\checkmark	&	12	&	5.61	\\
CTNNA1	&	Cell-Cell communication	&	--	&	\checkmark	&	\checkmark	&	10	&	4.47	\\
HNRNPDL	&	Immune System	&	--	&	\checkmark	&	\checkmark	&	7	&	3.87	\\
C1R	&	Immune System	&	--	&	\checkmark	&	\checkmark	&	5	&	4.57	\\
S100A12	&	Immune System	&	--	&	\checkmark	&	\checkmark	&	4	&	4.09	\\
\hline
\multicolumn{7}{l}{\textbf{No association with AD, but direct association with Epilepsy}} \\
KCNMA1	&	Neuronal System	&	\checkmark	&	--	&	--	&	5	&	6.07	\\
NLGN4X	&	Neuronal System	&	\checkmark	&	--	&	--	&	1	&	3.45	\\
AFF3	&	Neurological system process	&	--	&	\checkmark	&	\checkmark	&	2	&	2.51	\\
RREB1	&	Regulatory	&	\checkmark	&	\checkmark	&	--	&	1	&	2.51	\\
NFIX	&	Regulatory 	&	\checkmark	&	--	&	--	&	1	&	2.35	\\
TGIF1	&	Signaling, Regulatory 	&	--	&	\checkmark	&	\checkmark	&	4	&	6.11	\\
GRIA3	&	Signaling	&	\checkmark	&	--	&	--	&	6	&	6.15	\\
GRM1	&	Signaling	&	\checkmark	&	--	&	--	&	5	&	6.16	\\
PDE10A	&	Signaling, Signal Transduction	&	\checkmark	&	--	&	--	&	5	&	3.64	\\
ARL13B	&	Signal Transduction	&	\checkmark	&	--	&	--	&	2	&	2.96	\\
PDYN	&	Signal Transduction	&	--	&	\checkmark	&	\checkmark	&	3	&	6.29	\\
TAT	&	Metabolism	&	\checkmark	&	--	&	--	&	2	&	5.22	\\
LYRM7	&	Metabolism	&	--	&	\checkmark	&	\checkmark	&	2	&	2.71	\\
COX14	&	Metabolism	&	--	&	\checkmark	&	\checkmark	&	0	&	2.44	\\
EDC3	&	Metabolism of RNA	&	--	&	\checkmark	&	\checkmark	&	0	&	1.97	\\
SLC13A5	&	SLC-mediated transport	&	\checkmark	&	--	&	--	&	0	&	1.97	\\
SLC12A1	&	SLC-mediated transport	&	\checkmark	&	--	&	--	&	1	&	1.97	\\
MEGF10	&	Cell adhesion	&	\checkmark	&	\checkmark	&	\checkmark	&	3	&	3.71	\\
BCL11A	&	Chromatin organization	&	--	&	--	&	\checkmark	&	3	&	2.71	\\
LGI4	&	Developmental Biology	&	--	&	\checkmark	&	\checkmark	&	0	&	1.97	\\
ANK2	&	Developmental Biology	&	--	&	\checkmark	&	\checkmark	&	6	&	3.49	\\
PMM2	&	Disease	&	--	&	\checkmark	&	\checkmark	&	3	&	2.33	\\
ALG13	&	Disease	&	--	&	\checkmark	&	\checkmark	&	0	&	1.97	\\
RAB27A	&	Disease	&	--	&	\checkmark	&	\checkmark	&	5	&	3.16	\\
ASL		&	Disease	&	--	&	\checkmark	&	\checkmark	&	2	&	2.6	\\
WDR45B	&	Disease	&	--	&	\checkmark	&	\checkmark	&	2	&	2.33	\\
TANGO2	&	Disease	&	\checkmark	&	--	&	--	&	1	&	3.06	\\
\hline
\multicolumn{7}{l}{\textbf{No association with AD, but inferred association with similar diseases\footnote{Those having `Alzheimer' annotation}}}\\
ABCG2	&	Neurological agents	&	--	&	\checkmark	&	\checkmark	&	24	&	3.49	\\
CASK	&	Neuronal System	&	--	&	\checkmark	&	\checkmark	&	3	&	3.35	\\
CTBP2	&	Regulatory	&	--	&	\checkmark	&	\checkmark	&	3	&	2.21	\\
CLIC1	&	Signaling, Regulatory	&	--	&	\checkmark	&	\checkmark	&	5	&	4.38	\\
HMGN1	&	Signaling, Regulatory	&	--	&	\checkmark	&	\checkmark	&	2	&	2.2	\\
GRK5	&	Signaling, Signal Transduction	&	\checkmark	&	--	&	--	&	19	&	4.96	\\
CHEK1	&	Signaling, Signal Transduction	&	\checkmark	&	--	&	--	&	11	&	1.96	\\
USF2	&	Signal Transduction	&	\checkmark	&	--	&	--	&	2	&	1.96	\\
FMNL1	&	Signal Transduction	&	--	&	\checkmark	&	\checkmark	&	6	&	3.29	\\
CLTCL1	&	Vesicle-mediated transport	&	\checkmark	&	--	&	--	&	3	&	1.75	\\
PCSK9	&	Metabolism	&	\checkmark	&	--	&	--	&	37	&	6.59	\\
PIAS2	&	Metabolism	&	\checkmark	&	--	&	--	&	1	&	3.58	\\
CYP1A1	&	Metabolism,	&	--	&	\checkmark	&	\checkmark	&	6	&	3.96	\\
IL16	&	Immune System	&	\checkmark	&	--	&	--	&	3	&	5.24	\\
FTL	&	Disease	&	\checkmark	&	--	&	--	&	10	&	1.96	\\
RDX	&	Developmental Biology	&	--	&	\checkmark	&	\checkmark	&	3	&	1.97	\\
\hline
\multicolumn{7}{l}{\textbf{No association }}\\
RET	&	Signaling	&	\checkmark	&	--	&	--	&	8	&	1.97	\\
TCN2	&	Metabolism	&	--	&	\checkmark	&	\checkmark	&	5	&	4.76	\\
TGFBR2	&	Signal Transduction	&	--	&	\checkmark	&	\checkmark	&	5	&	3.97	\\
F5	&	Disease	&	--	&	--	&	\checkmark	&	4	&	3.77	\\
EFS	&	Regulatory	&	--	&	--	&	\checkmark	&	4	&	3.59	\\
RAB8B	&	Vesicle-mediated transport	&	--	&	--	&	\checkmark	&	5	&	2.02	\\
\label{TAB:data_GO}
\end{longtable}
\end{small}

\subsection{Results: Comparative Biological Relevance}

The external biological validation analyses reveal substantial differences in the disease relevance of 
the genes selected by the competing screening procedures. 
Although all three methods identify genes with some degree of association to neurological disorders, 
the proposed DPD-SISP procedures consistently prioritize genes with stronger documented evidence 
for involvement in AD and related neurodegenerative processes.

A first indication of this improvement is provided by the VarElect PR scores summarized in Table \ref{TAB:data_no}. 
The benchmark ML/REML-based SIS procedure identifies  only a single gene in the highest-confidence category (G1), 
namely \textit{INSR} (PR score = 10.78). In contrast, \texttt{Gene\_Set\_1(cv)} and \texttt{Gene\_Set\_1(I0)} 
contain three and four G1 genes, respectively.  The latter includes \textit{VCP}, \textit{LMNB1}, \textit{SNAP25}, and \textit{CBS}, 
with PR scores ranging from 9.96 to 33.29. 
The highest score among all screened genes is obtained by \textit{VCP} (33.29), exceeding the best benchmark score by more than threefold.
Moreover, the average PR score among G1 genes increases from 10.78 for benchmark \texttt{Gene\_Set\_0} to 13.67 
for \texttt{Gene\_Set\_1(cv)} and 19.80 for \texttt{Gene\_Set\_1(I0)}. 
Similar improvements are also observed in the G3 and G5 categories. 
In addition, the DPD-based procedures identify fewer genes without documented phenotype associations (G6), 
indicating greater concordance with existing AD-related knowledge bases.


The disease-association analysis in Table  \ref{TAB:data_GO} further highlights qualitative differences between the selected gene sets.
The benchmark method primarily identifies genes associated with broad neurological, metabolic, and signaling processes, 
including insulin signaling (\textit{INSR}, \textit{AKT2}), oxidative stress response (\textit{HMOX1}), lipid metabolism (\textit{PCSK9}), 
neurotransmitter signaling (\textit{GRM1}, \textit{DRD5}), and microglial scavenger activity (\textit{MSR1}). 
In contrast, several genes uniquely prioritized by the DPD-based procedures have more direct links to neurodegenerative phenotypes. 
Examples include \textit{VCP}, which has been implicated in dementia and proteinopathy-related disorders \citep{meyer2012vcp}; 
\textit{LMNB1}, which has been associated with neurodegeneration and dementia-related nuclear abnormalities \citep{koufi2023lamin}; 
\textit{SNAP25}, a recognized biomarker of synaptic degeneration in AD \citep{brinkmalm2014snap25}; 
and \textit{DNM1}, a key regulator of synaptic vesicle recycling and neuronal communication \citep{zolochevska2018postsynaptic}.
The recovery of these genes suggests that the proposed DPD-SISP procedures preferentially capture molecular mechanisms 
more directly related to neuronal dysfunction and neurodegeneration.

The pathway annotations further reinforce this interpretation. 
While \texttt{Gene\_Set\_0} contains a larger proportion of genes involved in broad metabolic, signaling, 
and stress-response processes, \texttt{Gene\_Set\_1(cv)} and particularly \texttt{Gene\_Set\_1(I0)} show greater enrichment for 
genes participating in biological systems widely implicated in AD pathogenesis.
For example, \textit{SNAP25}, \textit{DNM1}, \textit{ANK2}, and \textit{PDYN} are involved in synaptic transmission, 
vesicle trafficking, and neuronal communication, processes that are essential for cognitive function 
and are disrupted during AD progression. 
Similarly, \textit{VCP}, \textit{WDR45B}, and \textit{RAB27A} participate in intracellular protein quality-control 
and trafficking pathways, including ubiquitin-mediated degradation, autophagy, and vesicular transport. 
In addition, \textit{C1R}, \textit{S100A12}, \textit{HNRNPDL}, \textit{CLIC1}, and \textit{ABCG2} are involved 
in immune-system and inflammatory pathways that have increasingly been recognized as important contributors to AD pathogenesis. 
Collectively, these pathways are closely linked to neuronal dysfunction, impaired protein homeostasis, and chronic neuroinflammation, 
all of which are established drivers of AD progression \citep{heneka2015neuroinflammation,destrooper2016ad,long2019ad}.
The identification of multiple genes within these pathways therefore provides additional support for 
the biological relevance of the DPD-based screening results.


Among the two robust implementations, \texttt{Gene\_Set\_1(I0)} exhibits the strongest overall validation profile. 
It contains the largest number of highly significant AD-associated genes (G1), achieves the highest average PR scores, 
and yields the smallest number of genes lacking documented phenotype associations (G6). 
The I0-P proxy-based DPD-SISP also demonstrated the greatest stability under repeated subsampling analyses. 
Although the differences between the two DPD-SISP are modest relative to their shared improvement over the benchmark method, 
the available evidence suggests a slight advantage for the I0-P proxy in this application.

All these validation results derived from external knowledge bases consistently indicate that 
robust screening improves the prioritization of genes with established relevance to AD and related neurodegenerative mechanisms. 
These findings support the practical utility of the proposed DPD-SISP procedure for such ultrahigh-dimensional genomic studies.

\section{Concluding Remarks}

This paper develops a robust and computationally scalable framework for fixed-effects screening 
in ultrahigh-dimensional LMMs, addressing the dual challenges of complex dependence structures and data contamination. 
By integrating a proxy-based whitening strategy with marginal minimum DPD estimation, 
the proposed DPD-SISP procedure effectively separates dependence handling from robustness considerations within a unified framework. 
As a result, it achieves computational efficiency comparable to classical screening methods 
while substantially improving stability under data contamination and model misspecification. 
The simplicity of the marginal screening step further ensures scalability to ultrahigh-dimensional regimes 
encountered in modern applications.

We establish the sure screening property of DPD-SISP under general model settings, 
including non-Gaussian REs and error distributions, and allowing non-polynomial dimensionality. 
The analysis highlights the role of proxy matrix approximation in controlling the impact of dependence, 
and demonstrates that robust marginal estimation can reliably recover relevant variables under minimal signal strength conditions. 
The proposed method further integrates naturally with second-stage penalized estimation, 
enabling consistent variable selection and efficient parameter estimation in a robust manner.
%
%
We also explore several extensions to enhance practical applicability. 
These include conditional screening to account for pre-specified variable inclusion, and 
iterative refinement to alleviate correlation-induced masking. 
These extensions collectively broaden the applicability of the framework in real-world high-dimensional scenarios.

Despite these strengths, the performance of the proposed procedures depends on the quality of the proxy matrix approximation 
and careful tuning of the divergence parameter. Data-driven choices of both the proxy and tuning parameters, 
such as cross-validation-based strategies, are likely to improve performance across different dependence structures 
and contamination levels. Although the cross-validated proxy (cv-P) demonstrates clear advantages over simpler choices such as I0-P, 
particularly under strong dependence and contamination, further improvements may be achieved by incorporating structural information 
about the REs or by leveraging sparsity and low-rank approximations in complex settings, which we hope to explore in future research.
Similarly, the development of data-adaptive strategies for selecting the divergence tuning parameter, 
which governs the trade-off between robustness and efficiency, is another important avenue for future works. 
While moderate choices ($\alpha\approx 0.3$ or 0.5) perform well across a wide range of scenarios, as evidenced in the simulation studies, 
a principled and theoretically grounded selection mechanism could further enhance finite-sample performance. 

There are also several other directions for future research.
A promising direction is the extension of the proposed methodology to generalized and nonlinear mixed models, 
where non-Gaussian responses and nonlinear link functions introduce additional layers of complexity. 
Furthermore, investigating theoretical guarantees under weaker assumptions, 
including more general forms of dependence and contamination, as well as studying finite-sample behavior and optimal screening thresholds, 
would provide a deeper understanding of the method's practical performance.
Finally, integrating the proposed screening procedure with advanced post-screening inference techniques, 
including uncertainty quantification and reproducibility assessments, 
represents an important step towards fully robust high-dimensional modeling pipelines. 
Collectively, these directions would further strengthen the flexibility, reliability, 
and real-world applicability of robust screening methodologies for ultrahigh-dimensional mixed models.

\subsection*{Acknowledgments}

Data used in preparation of this article were obtained from the Alzheimer’s Disease Neuroimaging Initiative (ADNI) database 
(\url{https://adni.loni.usc.edu/data-samples/adni-data/}). 
The investigators within ADNI contributed to the design and implementation of ADNI and/or provided data 
but did not participate in analysis or writing of this report. 
A complete listing of ADNI investigators can be found at the ADNI acknowledgement page 
(\url{https://adni-lde.loni.usc.edu/wp-content/uploads/how_to_apply/ADNI_Acknowledgement_List.pdf}).

\appendix
\section{Proofs of the Results}
\label{APP}

\subsection{Proof of Theorem \ref{THM:SISP-population}}
Part (a) follows directly from the identifiability conditions in (A2) 
	and the fact that DPD defines a valid statistical divergence.
	This ensures Fisher consistency of the minimum DPD functional $\bm{\theta}_j^M$, 
	implying that the marginal slope is zero if and only if the corresponding estimating equation yields zero association.

	For Part (b), note that $\psi_\alpha$ is continuously differentiable. 
	So, applying the mean value theorem together with Assumption (A3), we obtain
	$$
	\left| \left\{\psi_\alpha(Y^*, \beta_{j0,\alpha}^M+ X_{j}^* \beta_{j1,\alpha}^M, \bm{\eta}_{j,\alpha}^M) 
	- \psi_\alpha(Y^*,\beta_{j0,\alpha}^M, \bm{\eta}_{j,\alpha}^M)\right\}\right|
	\leq C_D 	|\beta_{j1,\alpha}^M X_j^{*}|, ~~j\in\mathcal{I},
	$$
	where $0<C_D<\infty$ is the uniform upper bound on $\frac{\partial^2}{\partial\mu^2}V_\alpha(y, \mu, \bm{\eta})$. 
	Multiplying both sides by $|X_j^*|$ and taking expectations, we get
	$$
	C_D 	|\beta_{j1,\alpha}^M|E[X_j^{*2}] \geq \left| S_{j,\alpha}(\beta_{j0,\alpha}^M, \bm{\eta}_{j,\alpha}^M)\right|,
	$$
	since, by definition of the target parameter, 
	$$
	S_{j,\alpha}(\beta_{j0,\alpha}^M+ X_{j}^* \beta_{j1,\alpha}^M, \bm{\eta}_{j,\alpha}^M)  
	=E\left[\psi_\alpha(Y^*, \beta_{j0,\alpha}^M+ X_{j}^* \beta_{j1,\alpha}^M, \bm{\eta}_{j,\alpha}^M) X_j^*\right] =0.
	$$ 
	Further, Assumptions (A0)--(A1) implies (\ref{EQ:Ass1}).
	Combining this with the assumed lower bound on $S_{j,\alpha}$, we obtain 
	$$
	\left|\beta_{j1,\alpha}^M\right| \geq \frac{C_m}{C_XC_D}c_1 n^{-\kappa}, ~~\mbox{ for all }j\in\mathcal{S}_0.
	$$
	Thus, the desired result follows with the constant $c_2 = \frac{C_m c_1}{C_XC_D}>0$. 
\qed

\subsection{Proof of Theorem \ref{THM:SISP-ssp}}
\noindent\textbf{Proof of Part (a):}\\
	Let us define the event $\mathcal{E}_n = \left\{ \max_{j\in\mathcal{S}_0} \left|\widehat\beta_{j1,\alpha} - \beta_{j1}^M\right|
	\leq \frac{c_2}{2}n^{-\kappa}  \right\}$, where $c_2$ is as in (B3). 
	On $\mathcal{E}_n$, Assumption (B3) implies 
	$\min\limits_{j\in\mathcal{S}_0} \left|\widehat\beta_{j1,\alpha}\right| \geq  \frac{c_2}{2}n^{-\kappa} \geq  \gamma_n$,
	so that  $\mathcal{S}_0\subseteq \widehat{\mathcal{S}}_\alpha(\gamma_n)$. 
	Hence, we have
	\begin{eqnarray}
		P\left(\mathcal{S}_0\subseteq \widehat{\mathcal{S}}_\alpha(\gamma_n)\right)&=& 
		P\left(\min\limits_{j\in\mathcal{S}_0} \left|\widehat\beta_{j1,\alpha}\right| \geq  \gamma_n\right)
		\geq P(\mathcal{E}_n)
		\nonumber\\
		&=& 1 - P\left( \max_{j\in\mathcal{S}_0}\left|\widehat\beta_{j1,\alpha} - \beta_{j1}^P\right| > \frac{c_2}{4} n^{-\kappa}  \right) 
		- P\left(\max_{j\in\mathcal{S}_0}\left|\beta_{j1}^P - \beta_{j1}^M\right| > \frac{c_2}{4} n^{-\kappa}  \right) 
		\nonumber\\
		&\geq& 1 -  s  (R_n+\widetilde{R}_n)C_1,
	\end{eqnarray}
	for some $C_1>0$ and all sufficiently large $n$, by Assumptions (B1)--(B2) and a union bound of probability.\\
	
	\noindent\textbf{Proof of Part (b):}\\
	Under Assumption (B4), for any $\delta>0$, we obtain
	$$
	\left| \left\{ j \in\mathcal{I} : \left|\beta_{j1}^M\right| > \delta n^{-\kappa}\right\}\right|
	\leq O(n^{2\kappa}\lambda_{\max}(\bm{\Sigma}^*)).
	$$ 
	Now, on the event $\mathcal{E}_n^* = 
	\left\{ \max\limits_{j \in\mathcal{I}} \left|\widehat\beta_{j1,\alpha} - \beta_{j1}^M\right|\leq \frac{C}{2}n^{-\kappa}  \right\}$,
	we have 
	$$
	\left| \left\{ j\in\mathcal{I} : \left|\widehat\beta_{j1,\alpha}\right| > C n^{-\kappa} = \gamma_n\right\}\right|
	\leq \left| \left\{ j\in\mathcal{I} : \left|\beta_{j1,\alpha}^M\right| > \frac{C}{2} n^{-\kappa}\right\}\right|
	\leq O\left(n^{2\kappa}\lambda_{\max}(\bm{\Sigma}^*)\right).
	$$
	Therefore, by using Assumption (B1)--(B2) as before, we get 
	$$
	P\left(|\widehat{\mathcal{S}}_\alpha(\gamma_n)|\leq O(n^{2\kappa}\lambda_{\max}(\bm{\Sigma}^*))\right)
	\geq P(\mathcal{E}_n^*) \geq 1 - p (R_n+\widetilde{R}_n)C_2,
	$$
	for some $C_2>0$ and all sufficiently large $n$.
\qed

\subsection{Proof of Theorem \ref{THM:Proxy1}}	
	Under Assumption (A3), a Taylor expansion around the oracle parameter $\bm{\theta}_{j,\alpha}^{M}$ yields
	$$
	||\bm{\theta}_{j,\alpha}^{P} - \bm{\theta}_{j,\alpha}^{M}||_1 \leq c 
	E\left|\psi_\alpha(Y^*, \beta_{j0,\alpha}^{P} + {X}^*_{j}\beta_{j1,\alpha}^{P}, \bm{\eta}_{j,\alpha}^{P})\right|, 
	$$ 	
	for some constant $c>0$, since $E\left[\psi_\alpha(Y^*, \beta_{j0,\alpha}^{M} + \widetilde{X}^*_{j}\beta_{j1,\alpha}^{M}, \bm{\eta}_{j,\alpha}^{M})\right]=0$. 
	But we have\\ 
	$$
	E\left[\psi_\alpha(\widetilde{Y}, \beta_{j0,\alpha}^{P} + \widetilde{X}_{j}\beta_{j1,\alpha}^{P}, \bm{\eta}_{j,\alpha}^{P})\right]=0.
	$$
	So, another Taylor series expansion and using (A3), we further get 
	$$
	E\left|\psi_\alpha(Y^*, \beta_{j0,\alpha}^{P} + {X}^*_{j}\beta_{j1,\alpha}^{P}, \bm{\eta}_{j,\alpha}^{P})\right|
	\leq c_1 E|Y^* - \widetilde{Y}| + c_2 E|X_j^* - \widetilde{X}_j|,
	$$
	for some $c_1, C_2 >0$, independent of the choice of $j\in\mathcal{I}$. 
	Now, let us define the event $\mathcal{E}_n$ under which Assumption (P0) hold, 
	so that $P(\mathcal{E}_n) \geq 1 - O(\widetilde{R}_n)$. 
	However, on $\mathcal{E}_n$, we have 
	\begin{eqnarray*}
		E|Y^* - \widetilde{Y}| &=& \frac{1}{n} \sum_{i=1}^m ||\bm{y}_i^* - \widetilde{\bm{y}}_i||_1 
		= \frac{1}{n} \sum_{i=1}^m ||\left(\bm{\Sigma}_i^{-1/2} - \widehat{\bm{\Sigma}}_i^{-1/2}\right)\bm{y}_i||_1 
		\\
		&\leq& \frac{1}{n} \sum_{i=1}^m ||\bm{\Sigma}_i^{-1/2} - \widehat{\bm{\Sigma}}_i^{-1/2}||_1||\bm{y}_i||_1 
		= \frac{1}{n} \sum_{i=1}^m n_i ||\widehat{\bm{\Sigma}}_i^{-1/2}- \bm{\Sigma}_i^{-1/2}||_{op}||\bm{y}_i||_2
		\\
		&\leq & \frac{C_Y}{n} \sum_{i=1}^m n_i ||\widehat{\bm{\Sigma}}_i^{-1/2}- \bm{\Sigma}_i^{-1/2}||_{op}  ~~~~~~ \mbox{[by (A0)]}
		\\
		\\
		&\leq & \frac{C_Yc}{n} \sum_{i=1}^m n_i ||\widehat{\bm{\Sigma}}_i- \bm{\Sigma}_i||_{op}  ~~~~~~~~~~~~~ \mbox{[by (A1) and (P0)]}
		\\
		&\leq& O\left(n^{-\kappa}\right),~~~~~\mbox{[by (P0)]}
	\end{eqnarray*}
	and similarly 
	\begin{eqnarray*}
		E|X_j^* - \widetilde{X}_j| &\leq &  O\left(n^{-\kappa}\right),
		~~~\mbox{ uniformly in }j\in\mathcal{I}.
	\end{eqnarray*}
	Combining everything, we get 
	$$
	|\bm{\theta}_{j,\alpha}^{P} - \bm{\theta}_{j,\alpha}^{M}| \leq  \|\bm{\theta}_{j,\alpha}^{P} - \bm{\theta}_{j,\alpha}^{M}\|_1 \leq O(n^{-\kappa}), 
	$$ 	
	with probability at least $1 - O(\widetilde{R})$, uniformly over $j=1, \ldots, p$.
	This completes the proof.
\qed

\bibliographystyle{apalike}
\bibliography{Reference.bib}

\newpage

	\begin{center}
	{\Huge Supplementary Figures and Table }
\end{center}
\setcounter{figure}{0} 
\renewcommand{\thefigure}{S\arabic{figure}} 
\setcounter{table}{0} 
\renewcommand{\thetable}{S\arabic{table}} 

\begin{figure}[!h]
	\centering
	\subfloat[\tiny 5\% (C1) Contamination]{
		\includegraphics[page=2, width=0.3\textwidth]{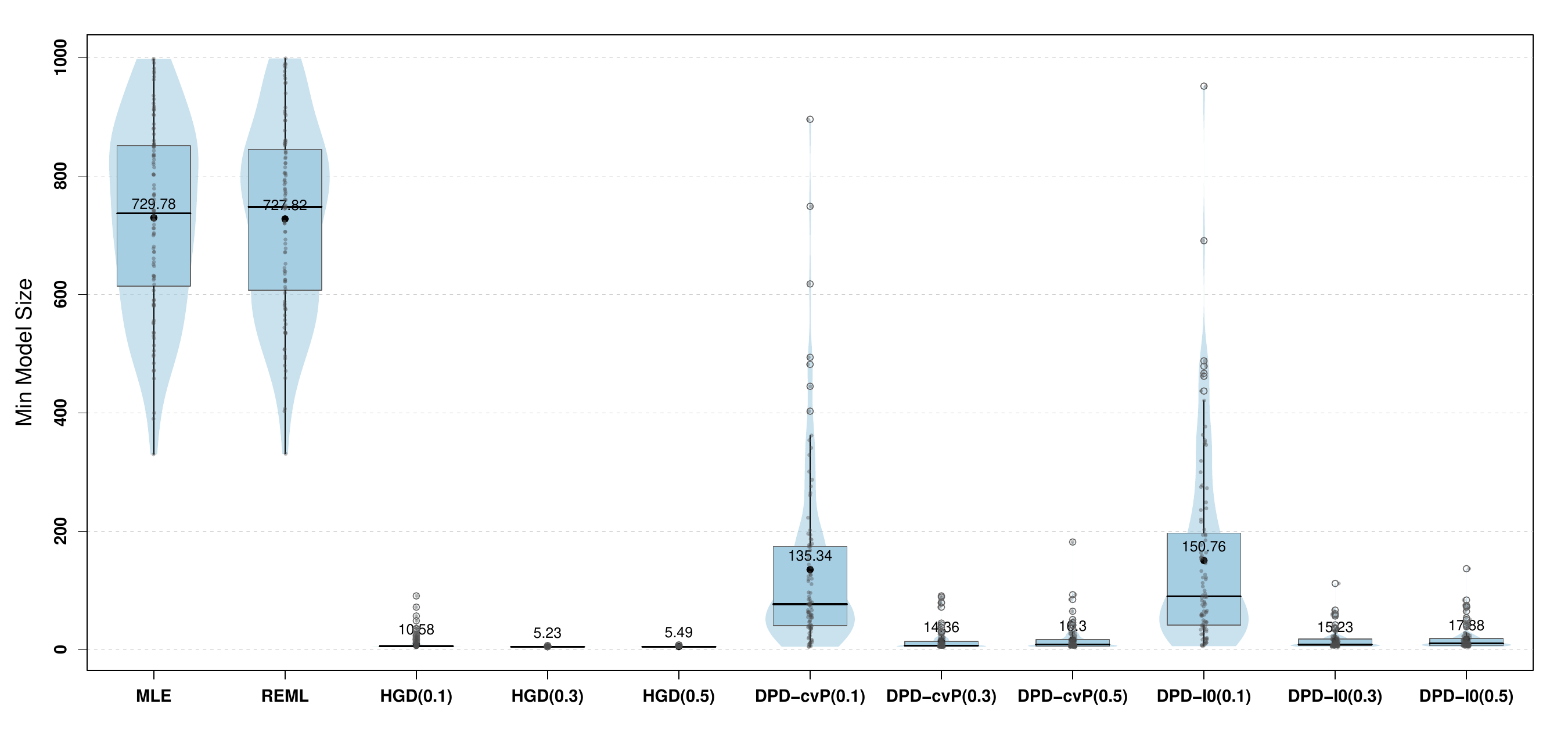}
		\label{FIG:boxplot_Y}}
	~	
	\subfloat[\tiny 10\% (C1) Contamination]{
		\includegraphics[page=2, width=0.3\textwidth]{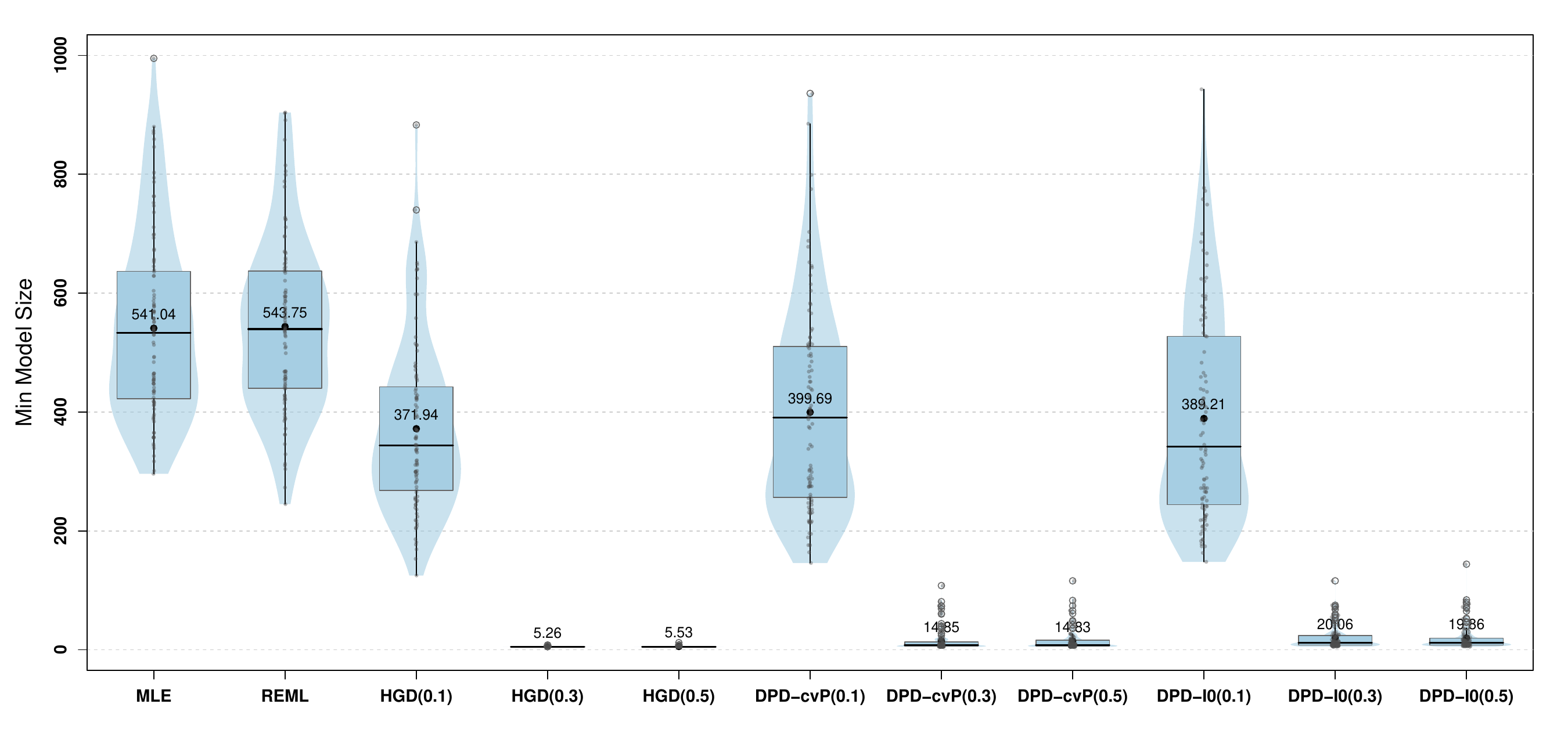}
		\label{FIG:boxplot_Y}}
	~	
	\subfloat[\tiny 20\% (C1) Contamination]{
		\includegraphics[page=2, width=0.3\textwidth]{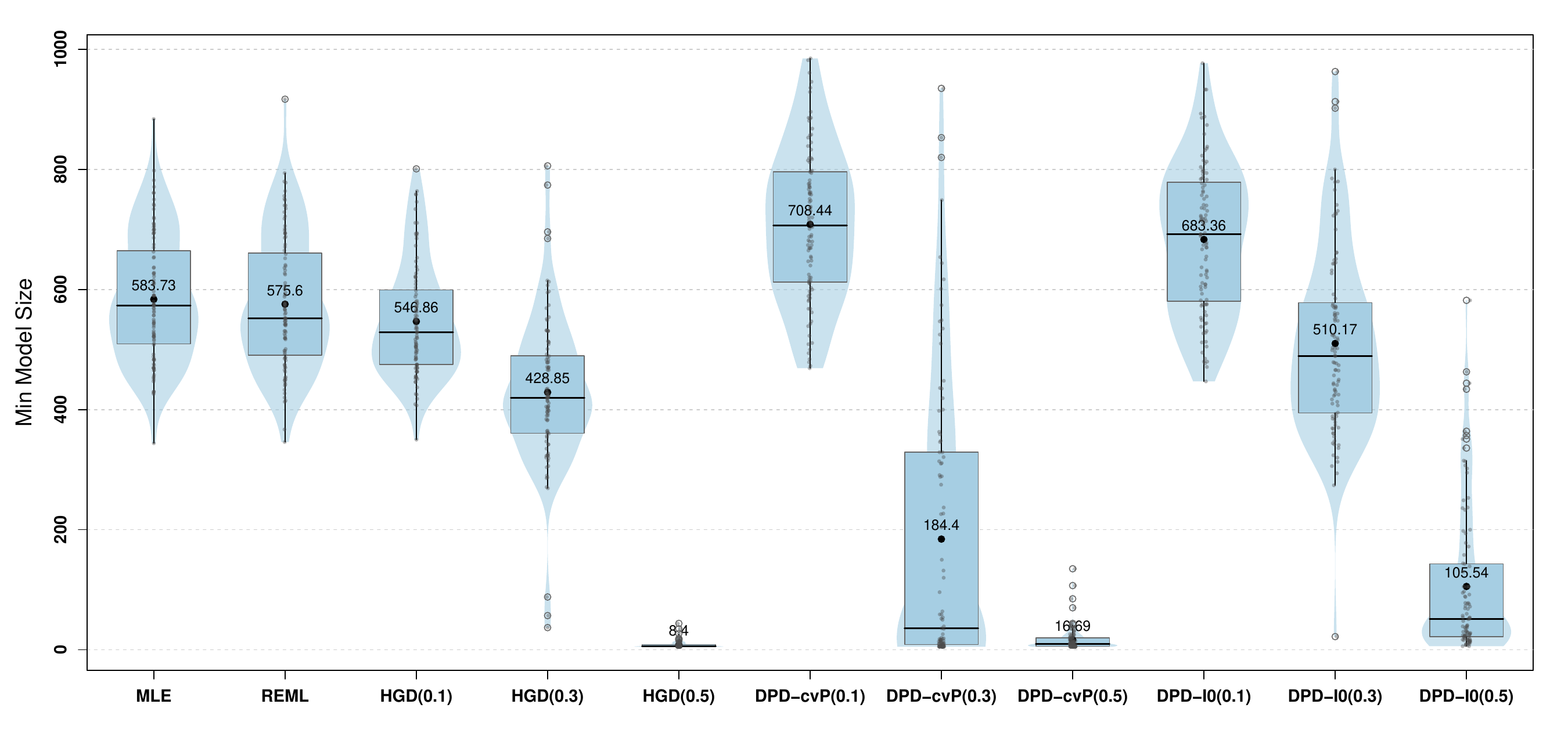}
		\label{FIG:boxplot_Y}}
	\\	
	\subfloat[\tiny 5\% (C2) Contamination]{
		\includegraphics[page=2, width=0.3\textwidth]{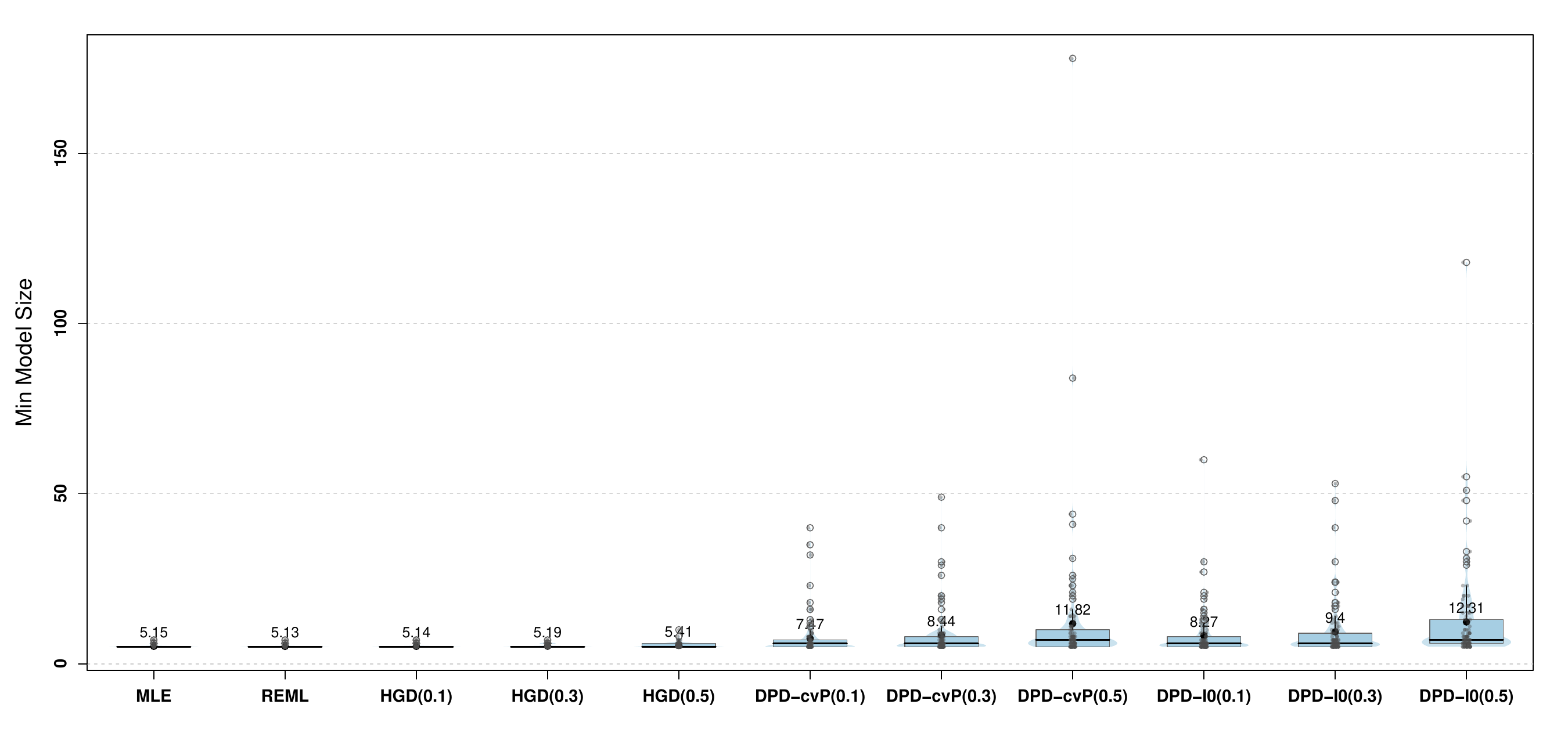}
		\label{FIG:boxplot_Y}}
	~	
	\subfloat[\tiny 10\% (C2) Contamination]{
		\includegraphics[page=2, width=0.3\textwidth]{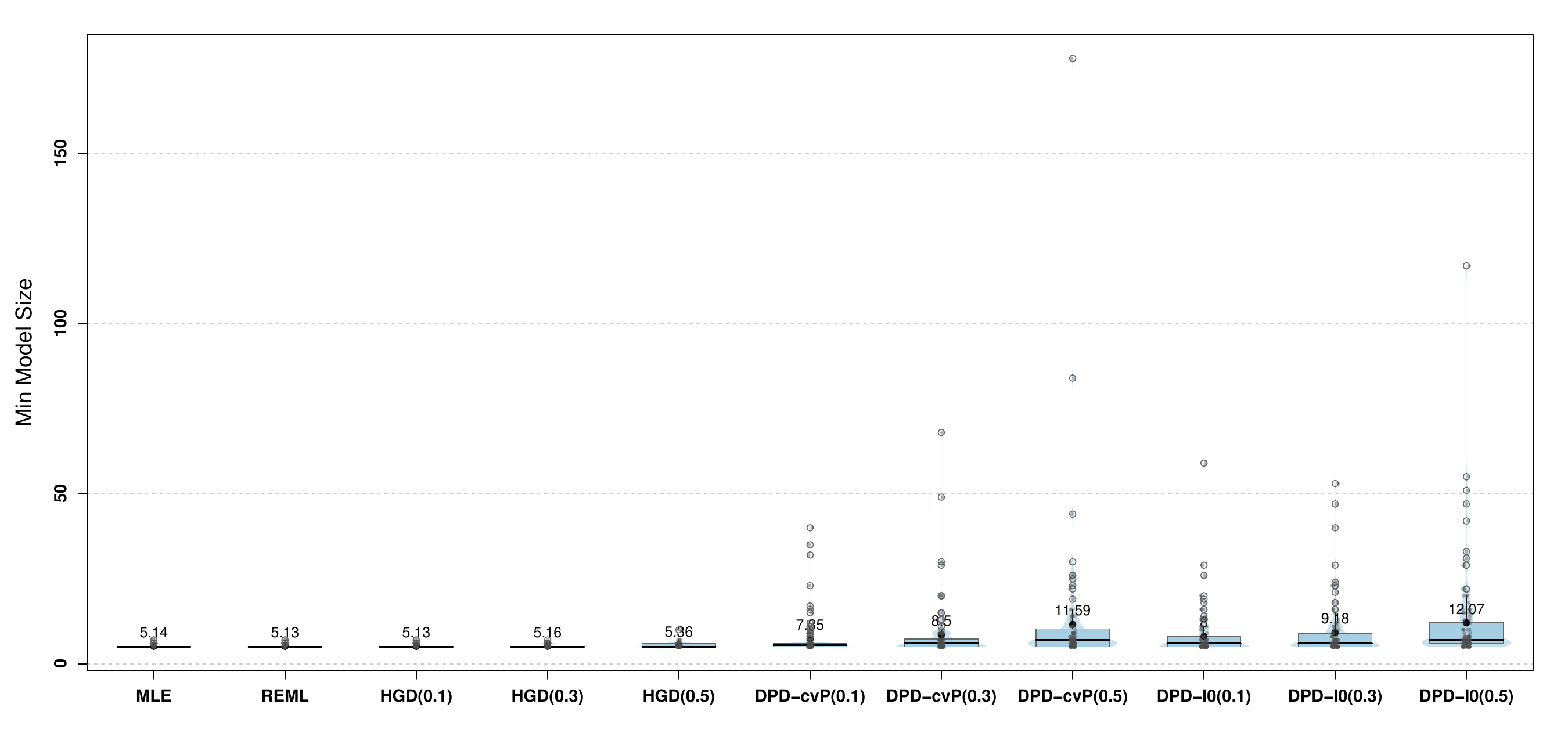}
		\label{FIG:boxplot_Y}}
	~	
	\subfloat[\tiny 20\% (C2) Contamination]{
		\includegraphics[page=2, width=0.3\textwidth]{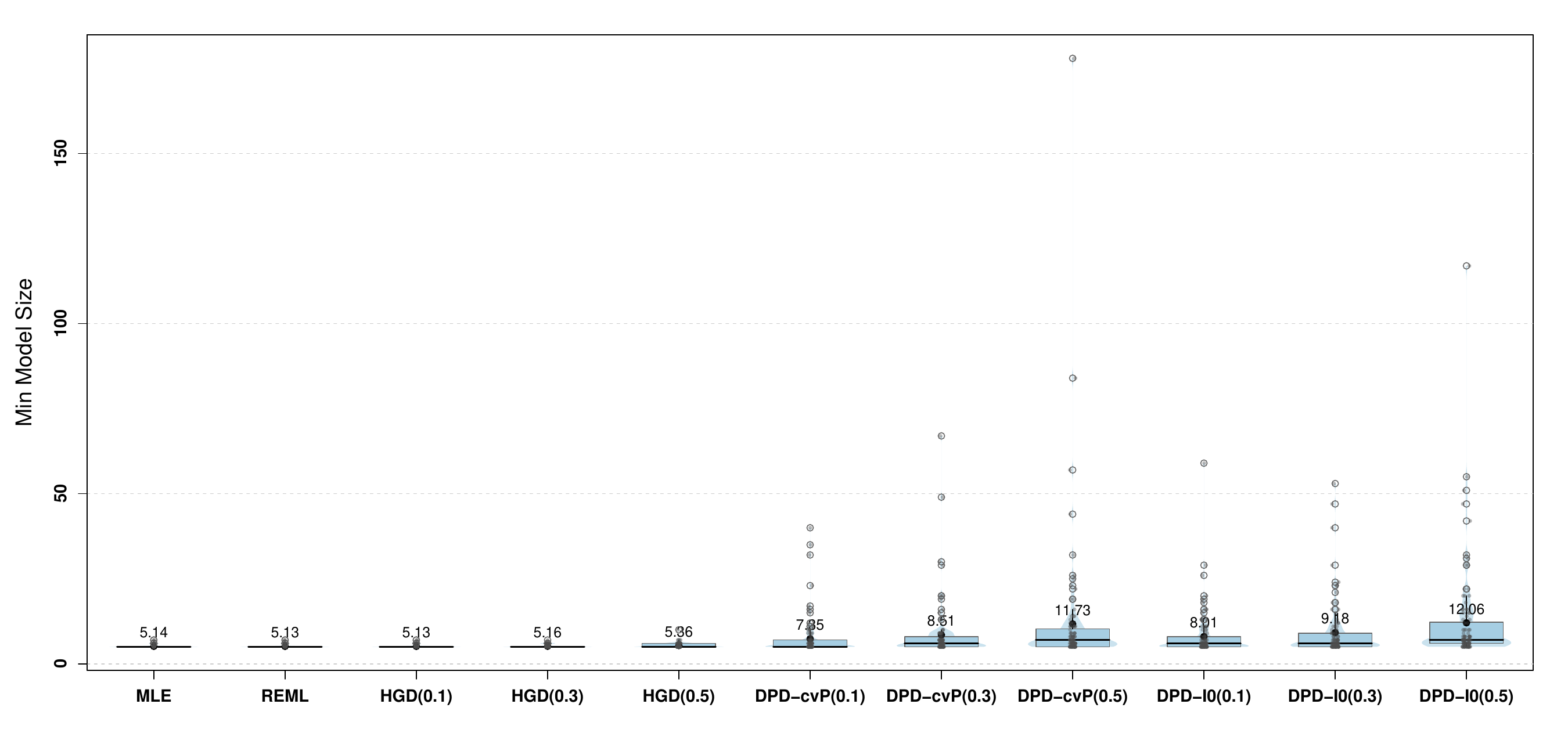}
		\label{FIG:boxplot_Y}}
	\\	
	\subfloat[\tiny 5\% (C3) Contamination]{
		\includegraphics[page=2, width=0.3\textwidth]{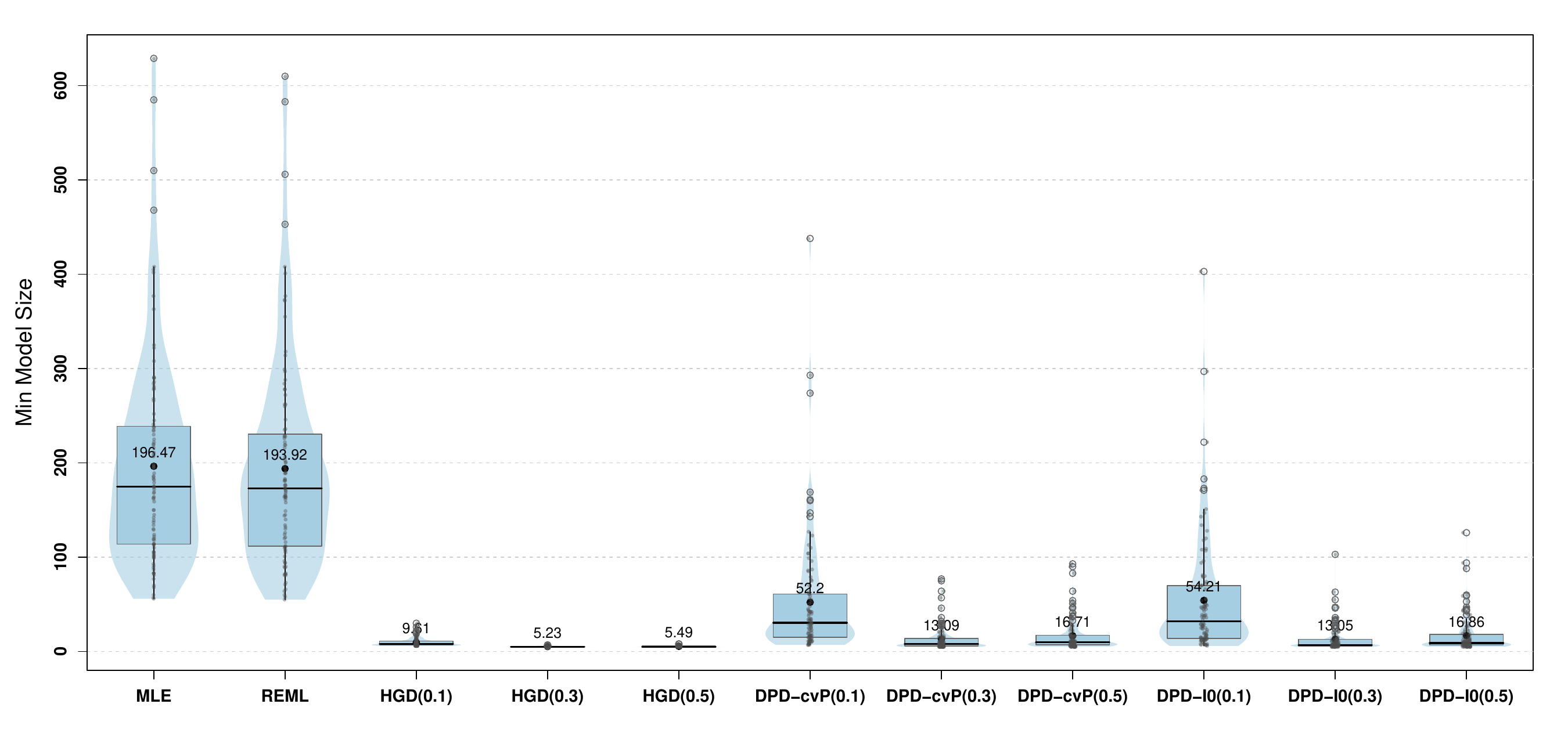}
		\label{FIG:boxplot_Y}}
	~	
	\subfloat[\tiny 10\% (C3) Contamination]{
		\includegraphics[page=2, width=0.3\textwidth]{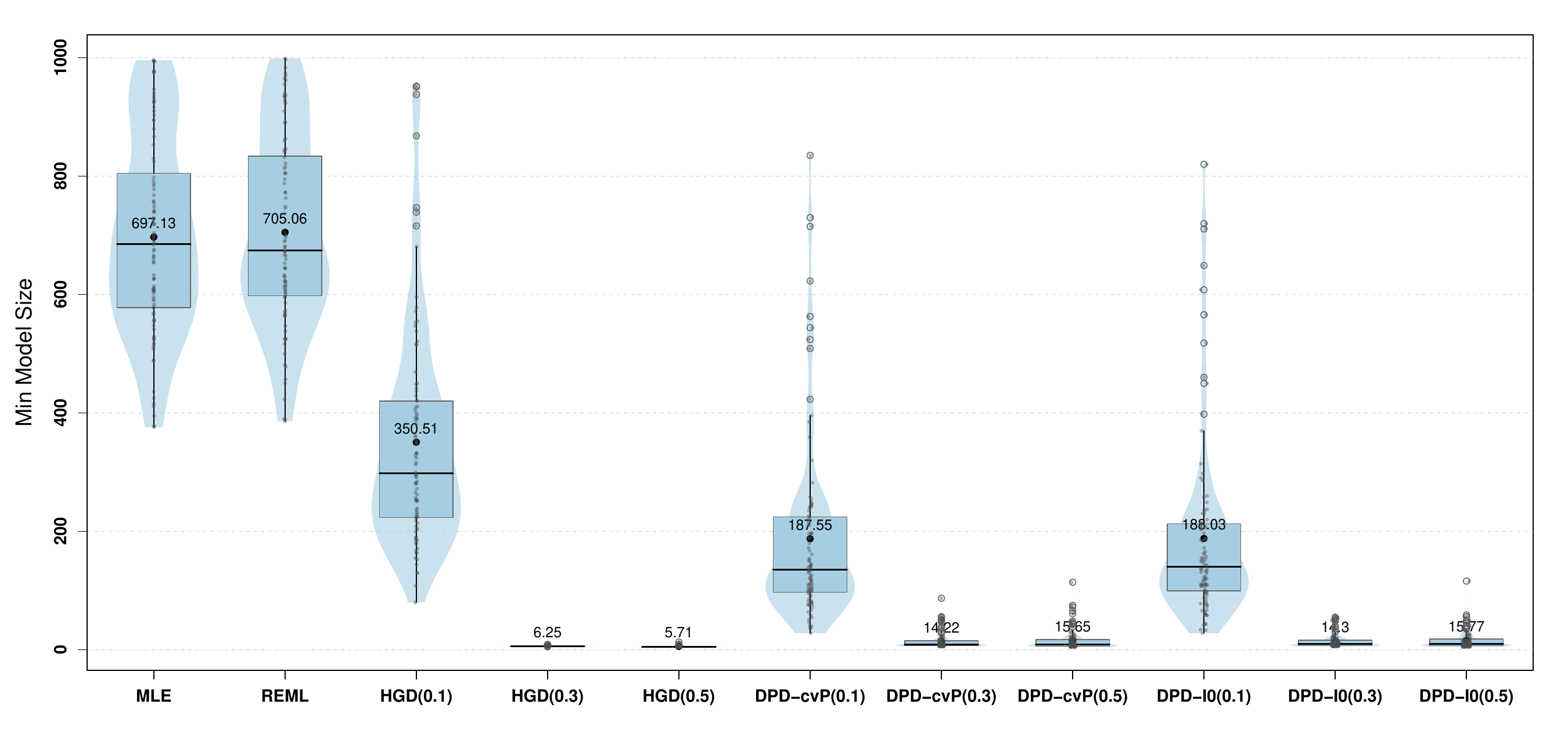}
		\label{FIG:boxplot_Y}}
	~	
	\subfloat[\tiny 20\% (C3) Contamination]{
		\includegraphics[page=2, width=0.3\textwidth]{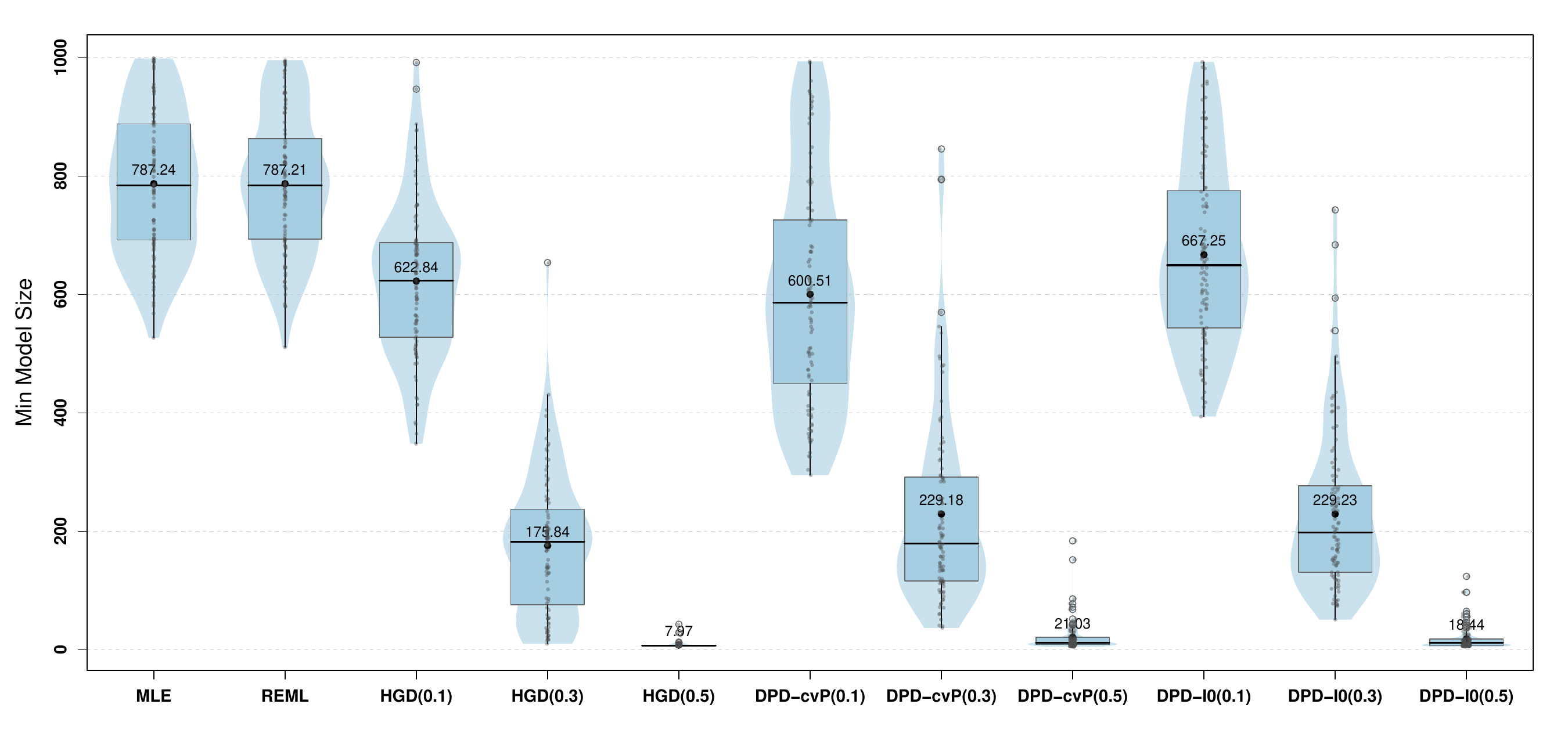}
		\label{FIG:boxplot_Y}}
	\\	
	\subfloat[\tiny 5\% (C4) Contamination]{
		\includegraphics[page=2, width=0.3\textwidth]{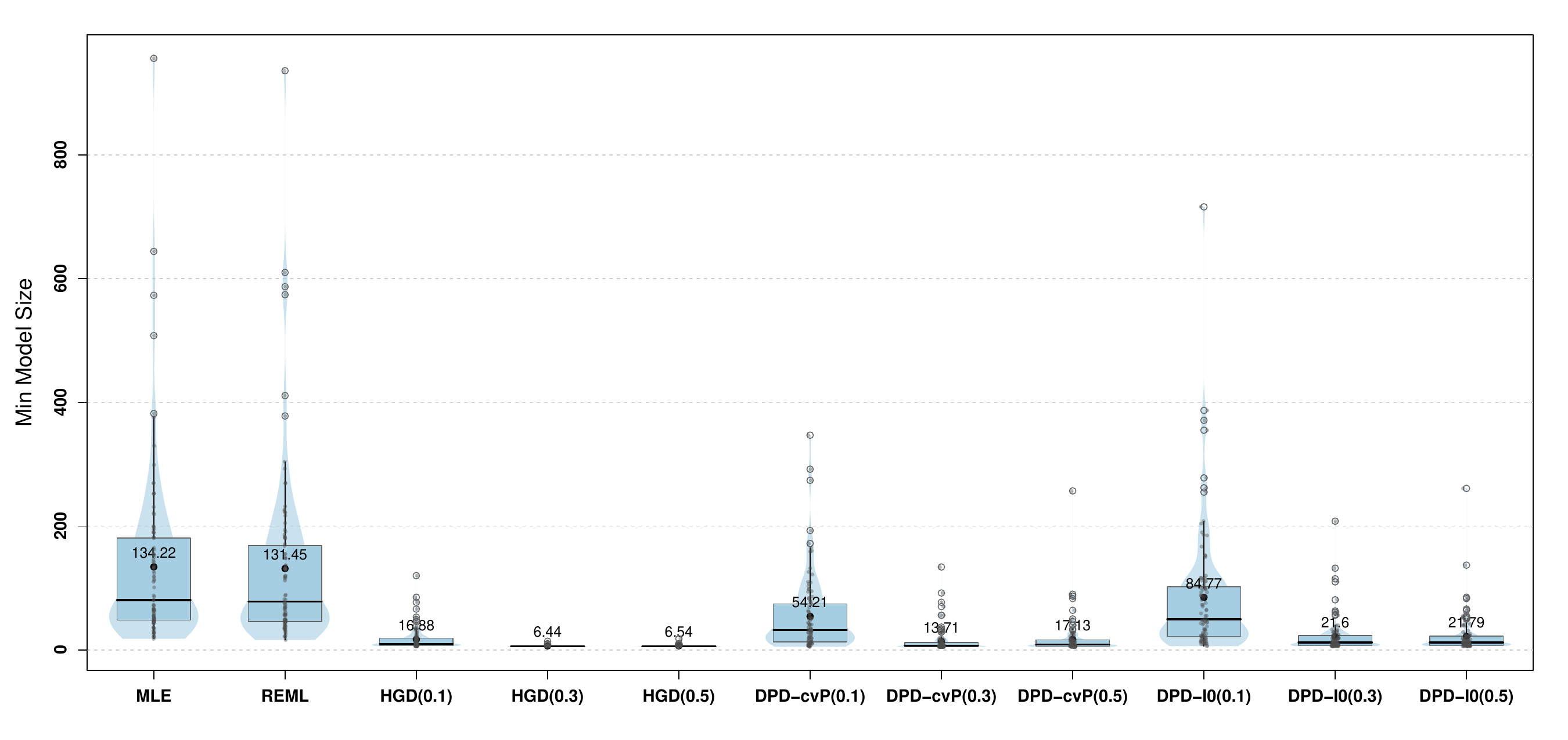}
		\label{FIG:boxplot_Y}}
	~	
	\subfloat[\tiny 10\% (C4) Contamination]{
		\includegraphics[page=2, width=0.3\textwidth]{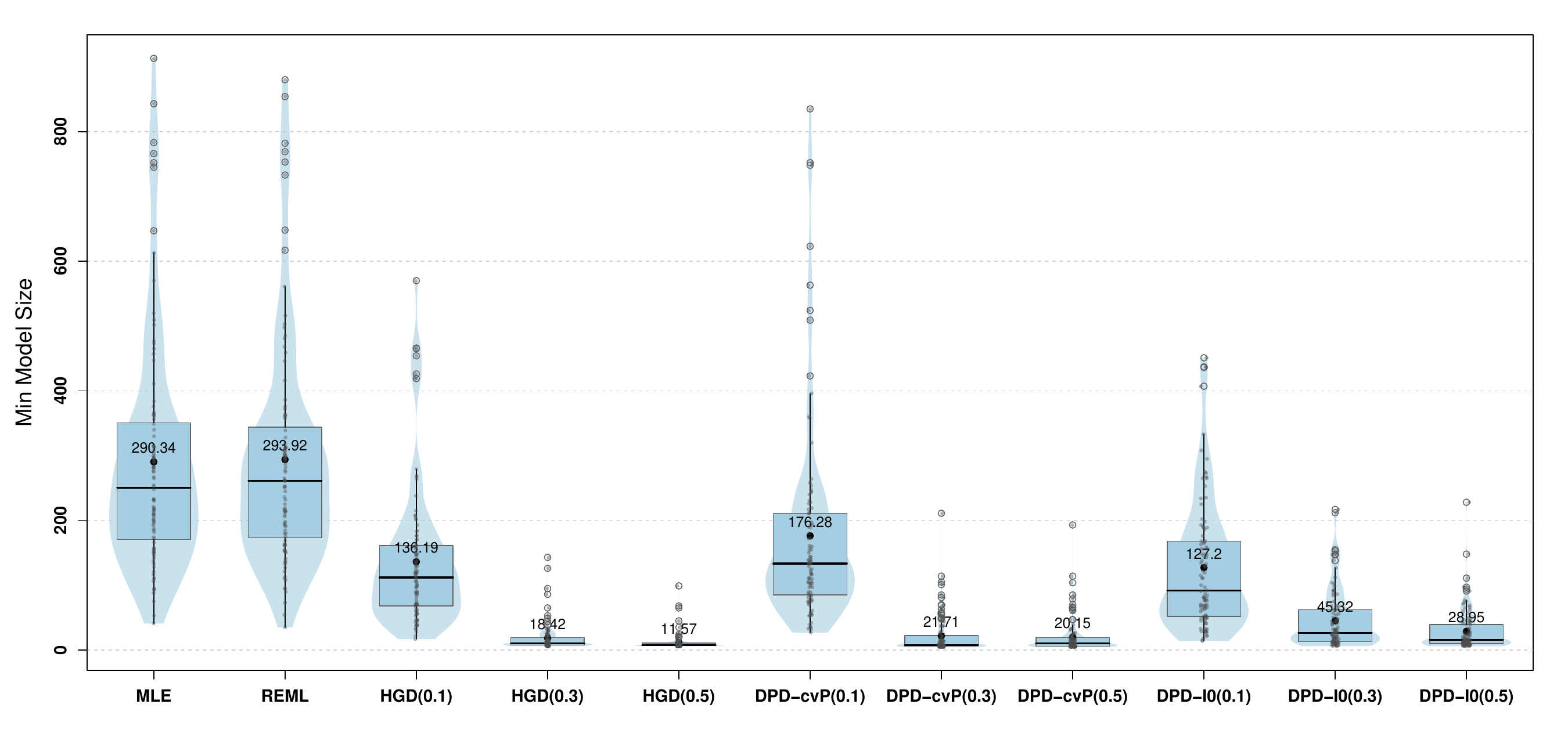}
		\label{FIG:boxplot_Y}}
	~	
	\subfloat[\tiny 20\% (C4) Contamination]{
		\includegraphics[page=2, width=0.3\textwidth]{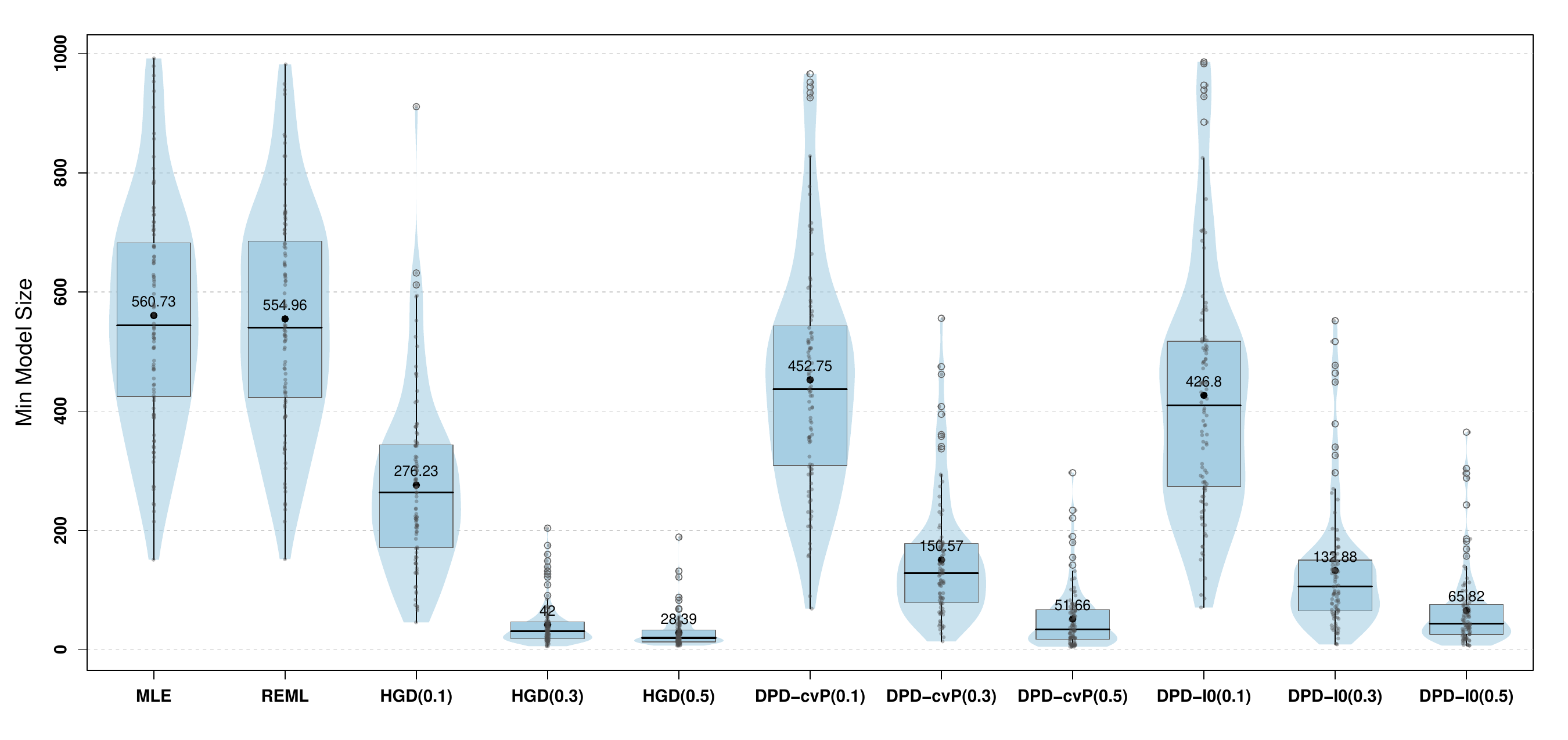}
		\label{FIG:boxplot_Y}}
	\caption{Boxplots  (with overlaid sample means) of MinMS required under scenario (S2)$\times$(R1) 
		with $\bm\Sigma_x=\mathbb{I}$ and different types of contamination}
	\label{FIG:MMS_S2R1Id}
\end{figure}

\begin{figure}[!h]
	\centering
	\subfloat[\tiny 5\% (C1) Contamination]{
		\includegraphics[page=2, width=0.3\textwidth]{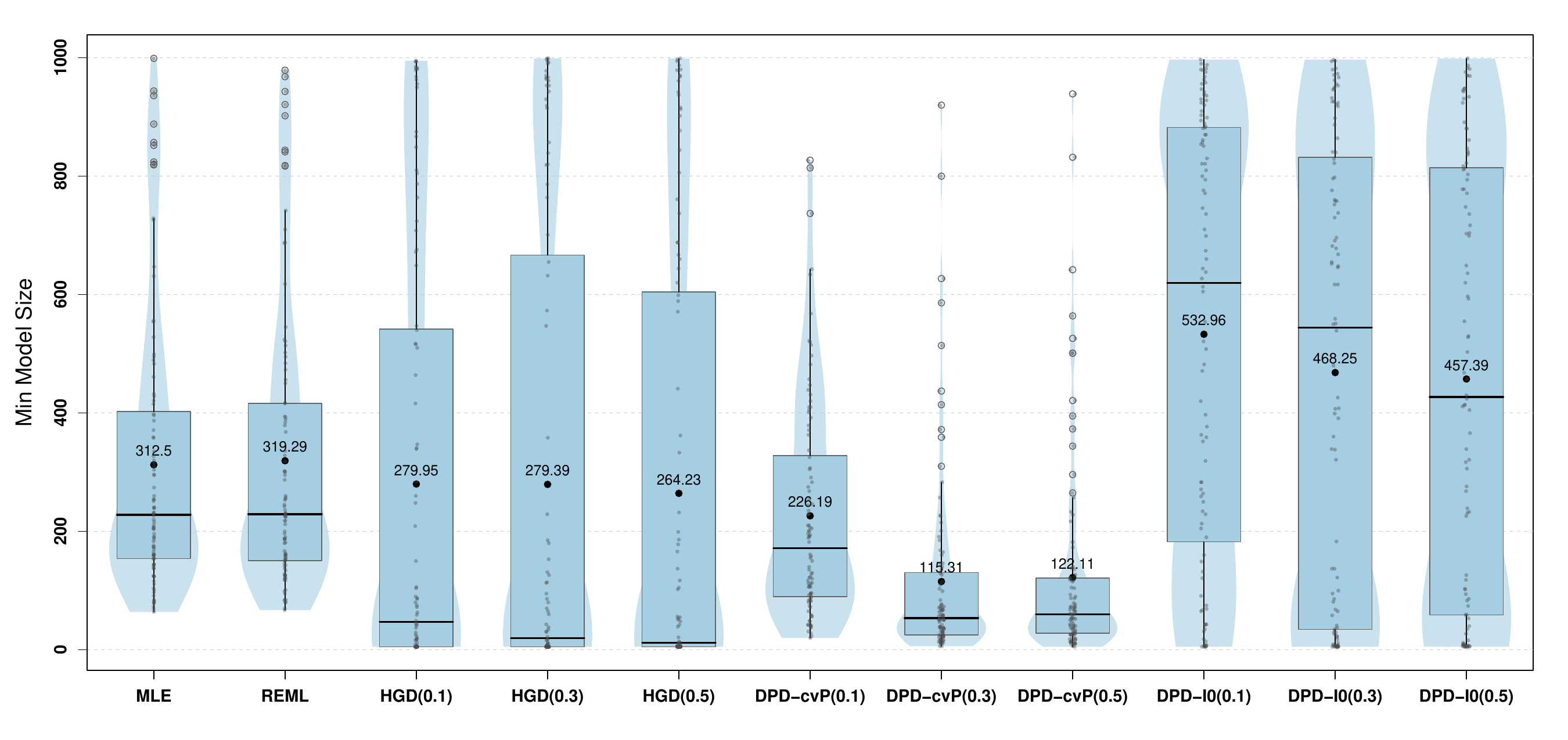}
		\label{FIG:boxplot_Y}}
	~	
	\subfloat[\tiny 10\% (C1) Contamination]{
		\includegraphics[page=2, width=0.3\textwidth]{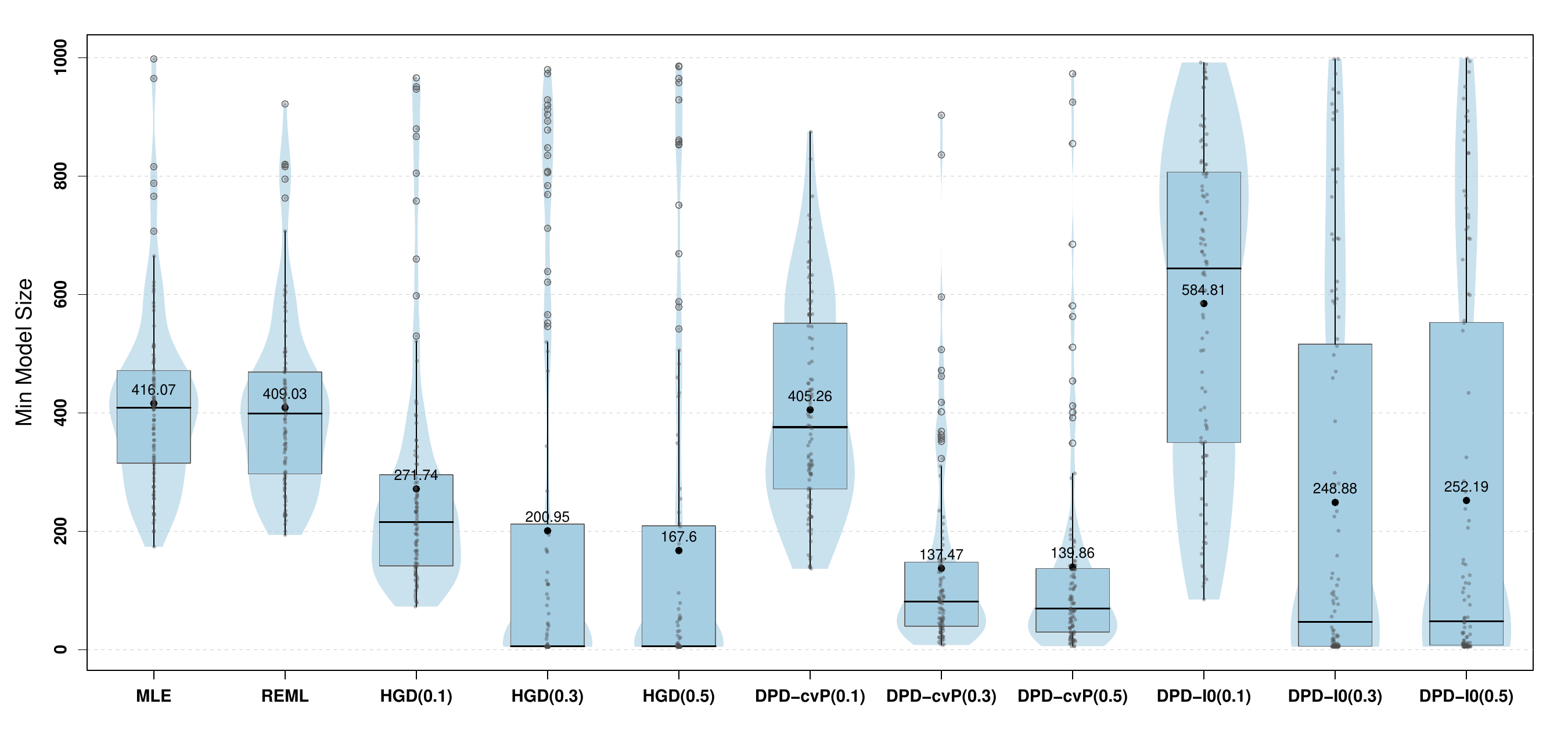}
		\label{FIG:boxplot_Y}}
	~	
	\subfloat[\tiny 20\% (C1) Contamination]{
		\includegraphics[page=2, width=0.3\textwidth]{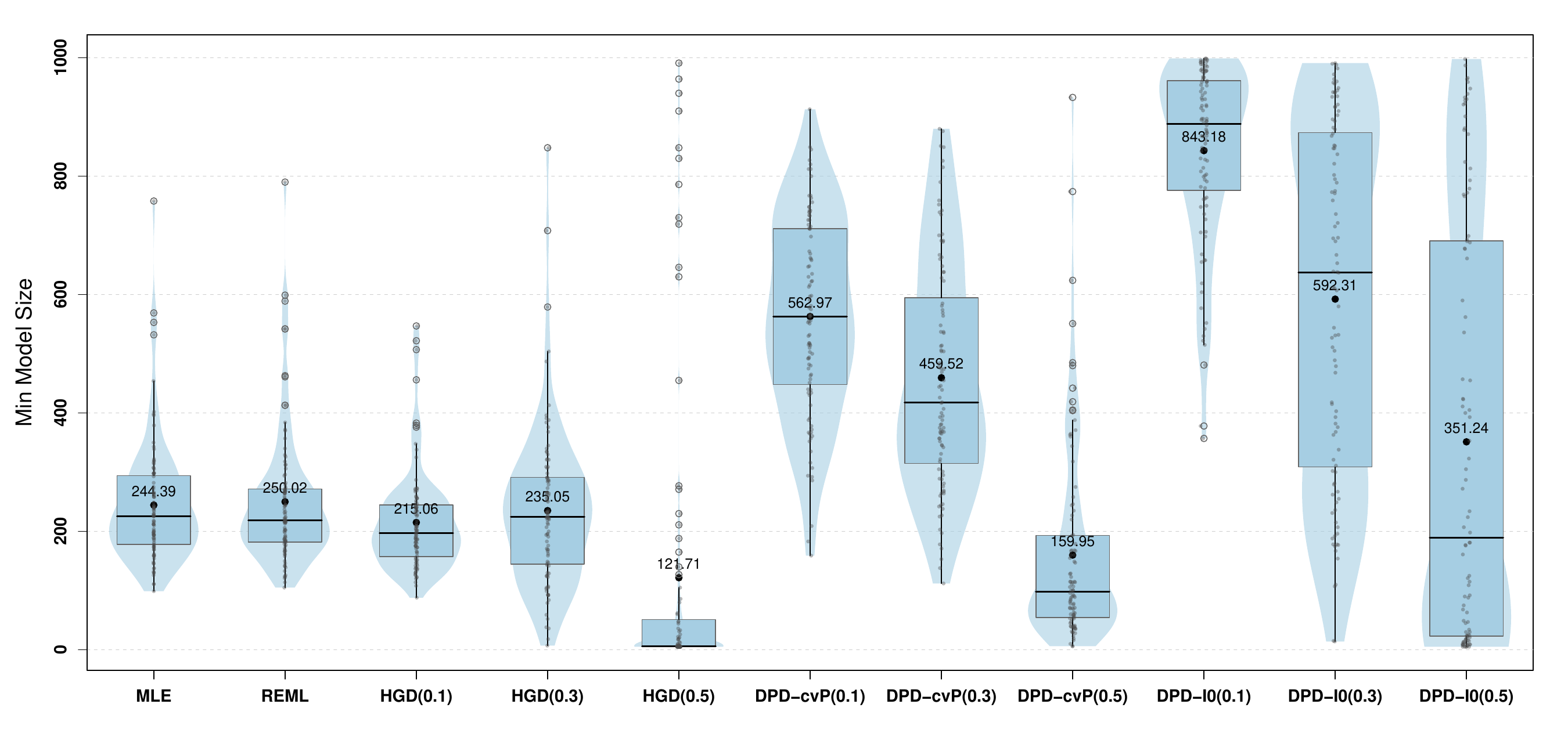}
		\label{FIG:boxplot_Y}}
	\\	
	\subfloat[\tiny 5\% (C2) Contamination]{
		\includegraphics[page=2, width=0.3\textwidth]{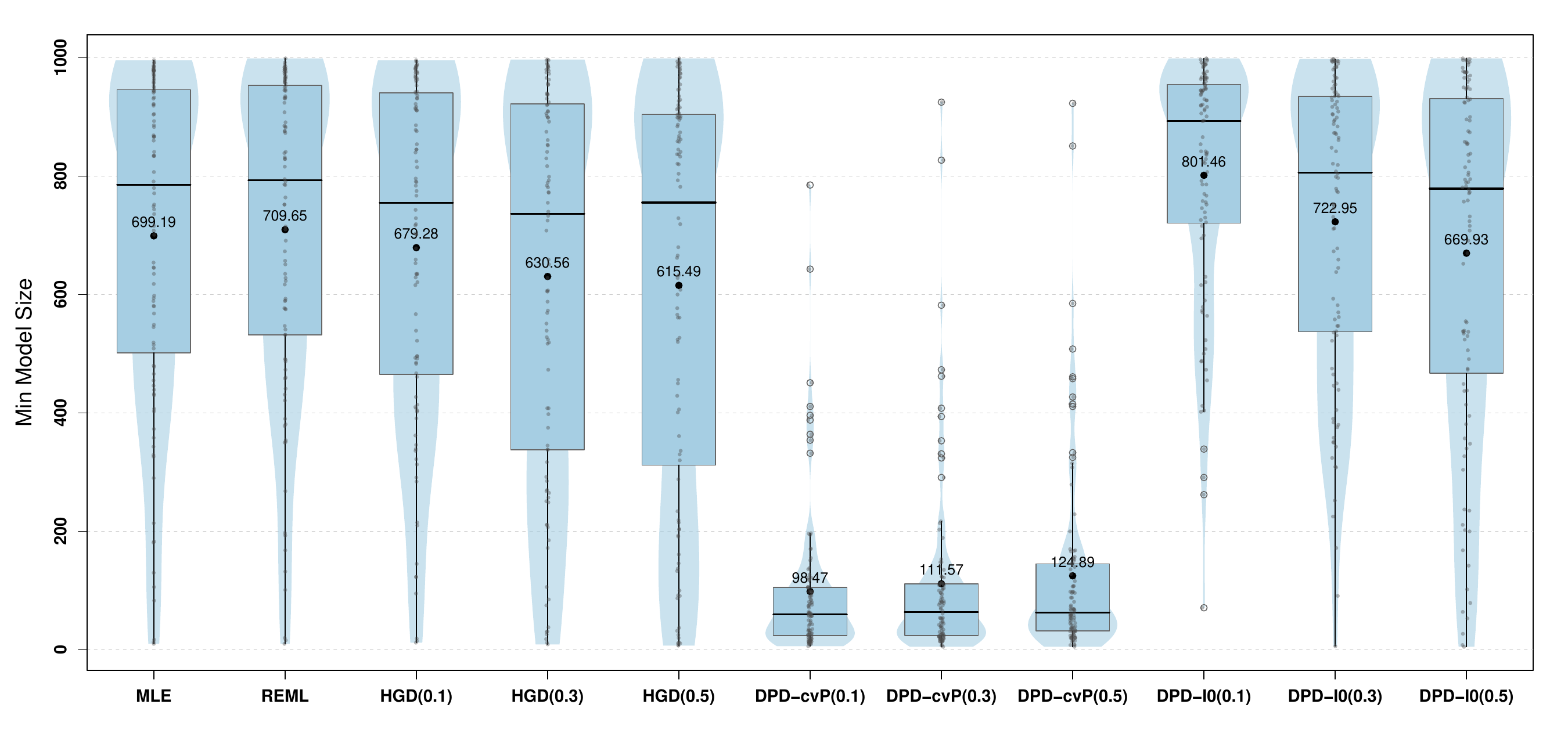}
		\label{FIG:boxplot_Y}}
	~	
	\subfloat[\tiny 10\% (C2) Contamination]{
		\includegraphics[page=2, width=0.3\textwidth]{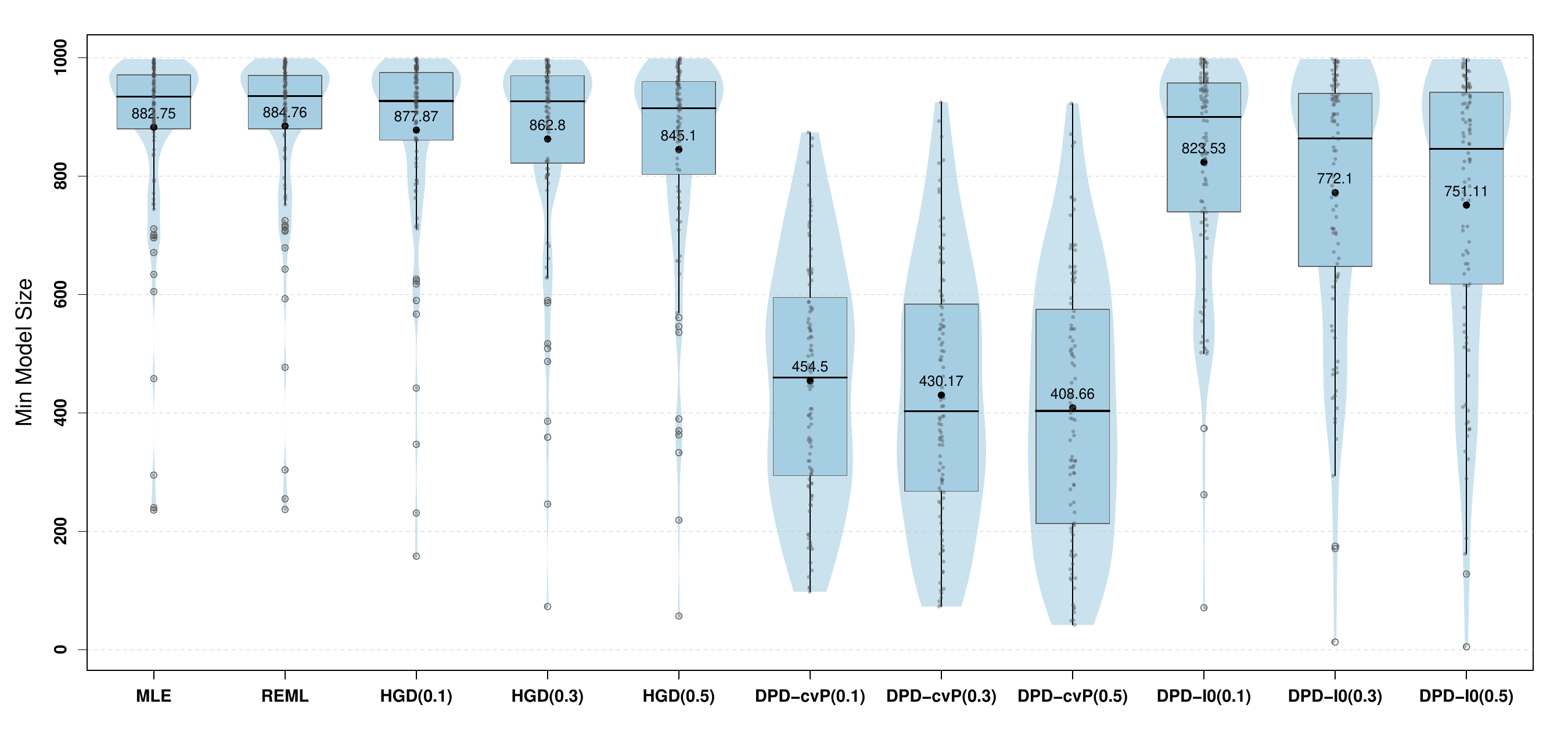}
		\label{FIG:boxplot_Y}}
	~	
	\subfloat[\tiny 20\% (C2) Contamination]{
		\includegraphics[page=2, width=0.3\textwidth]{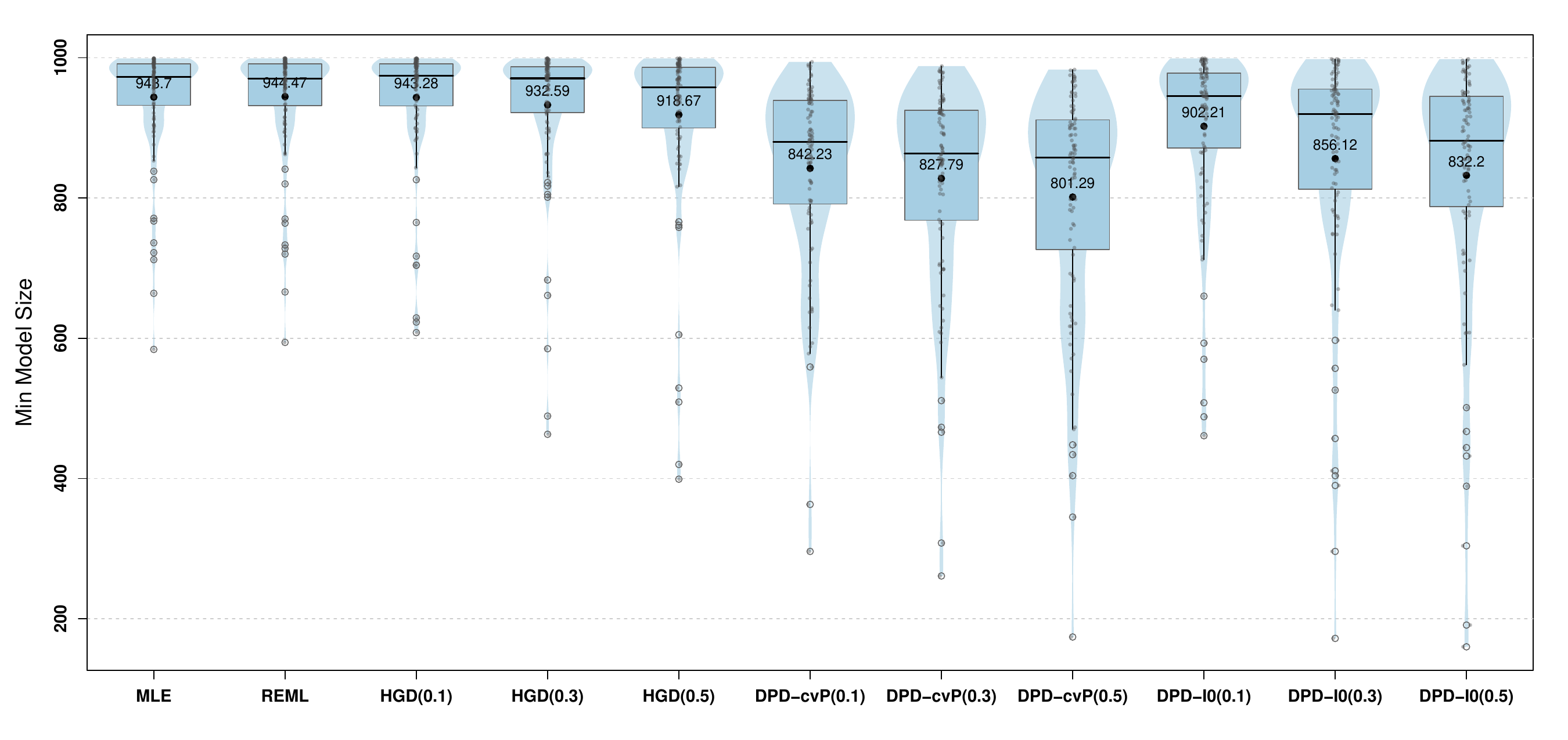}
		\label{FIG:boxplot_Y}}
	\\	
	\subfloat[\tiny 5\% (C3) Contamination]{
		\includegraphics[page=2, width=0.3\textwidth]{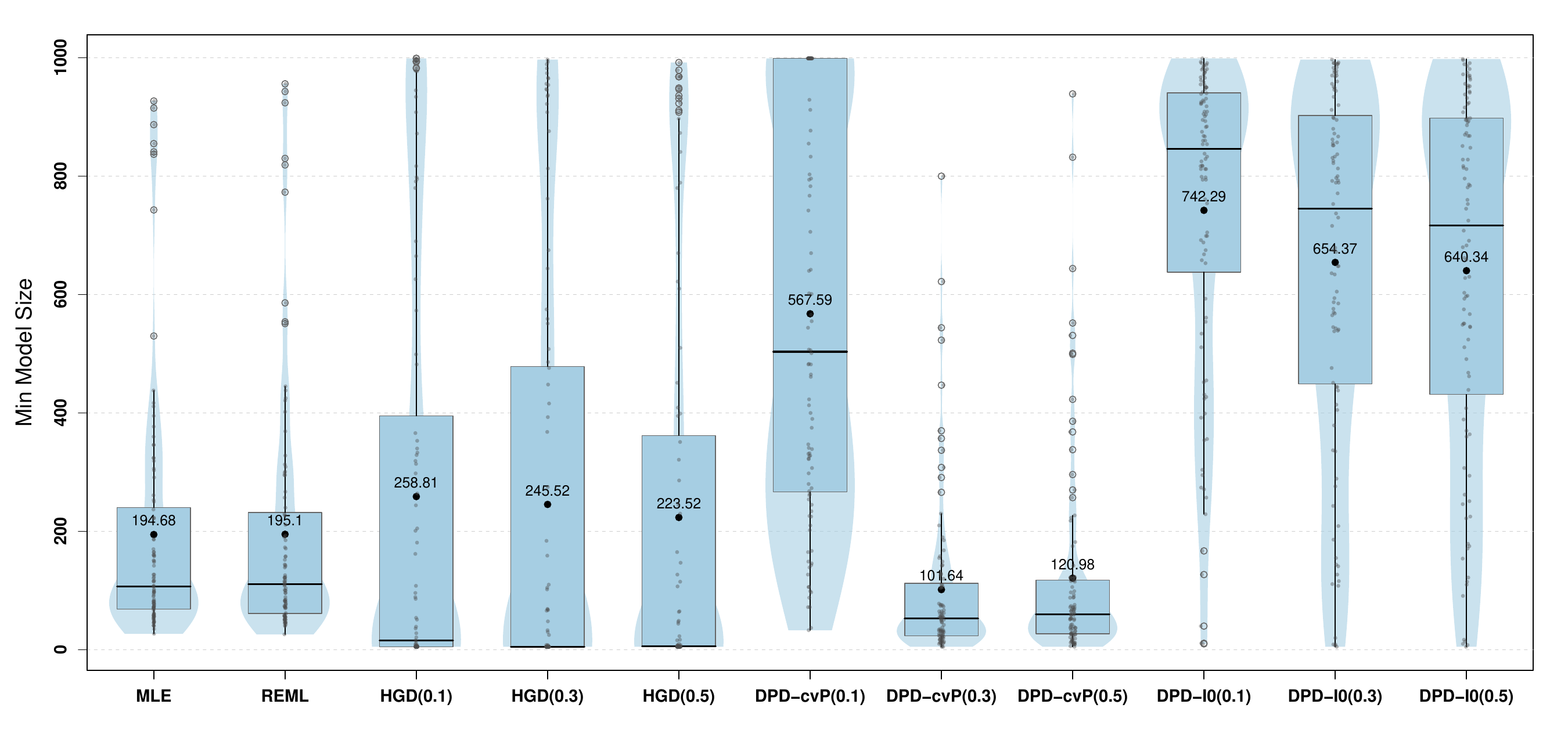}
		\label{FIG:boxplot_Y}}
	~	
	\subfloat[\tiny 10\% (C3) Contamination]{
		\includegraphics[page=2, width=0.3\textwidth]{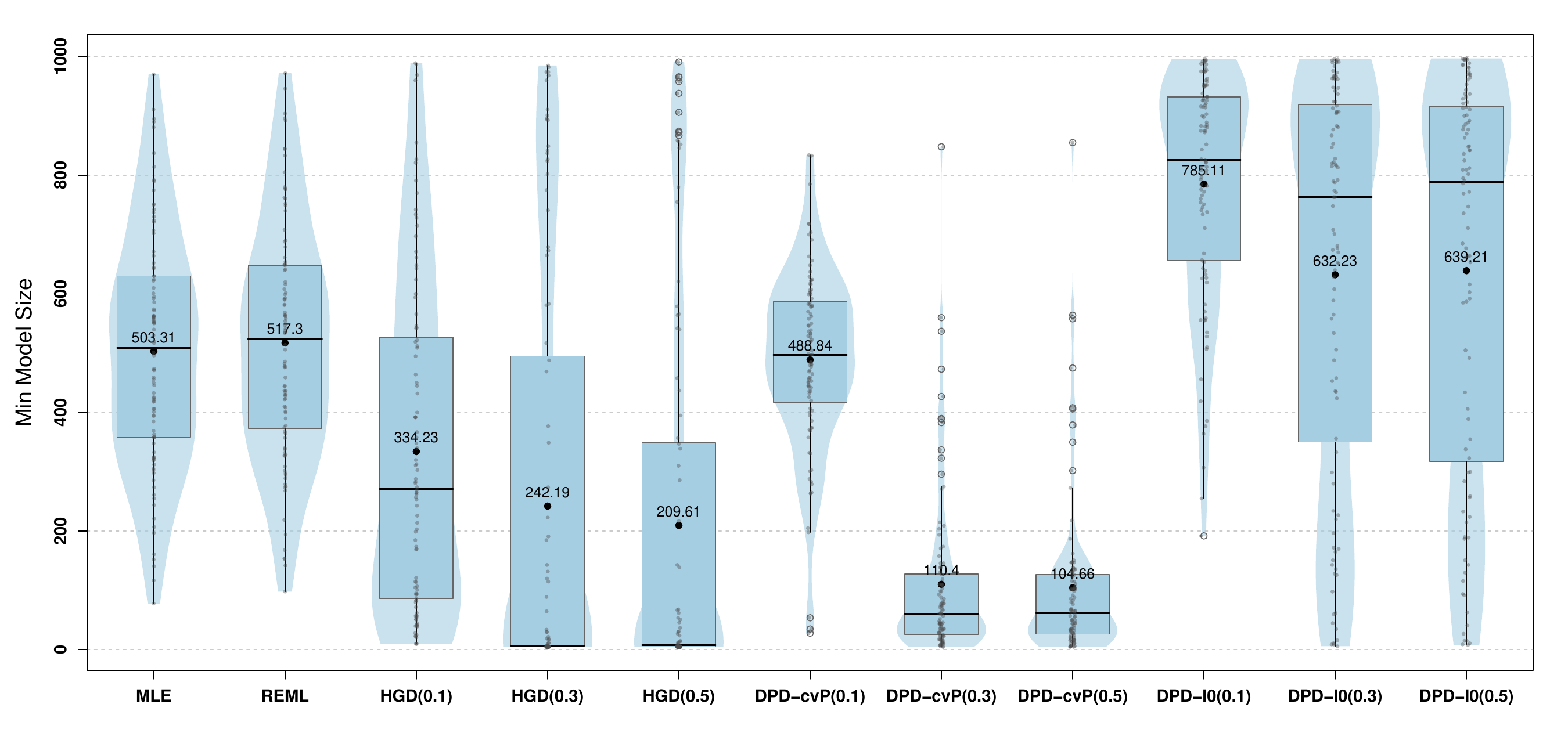}
		\label{FIG:boxplot_Y}}
	~	
	\subfloat[\tiny 20\% (C3) Contamination]{
		\includegraphics[page=2, width=0.3\textwidth]{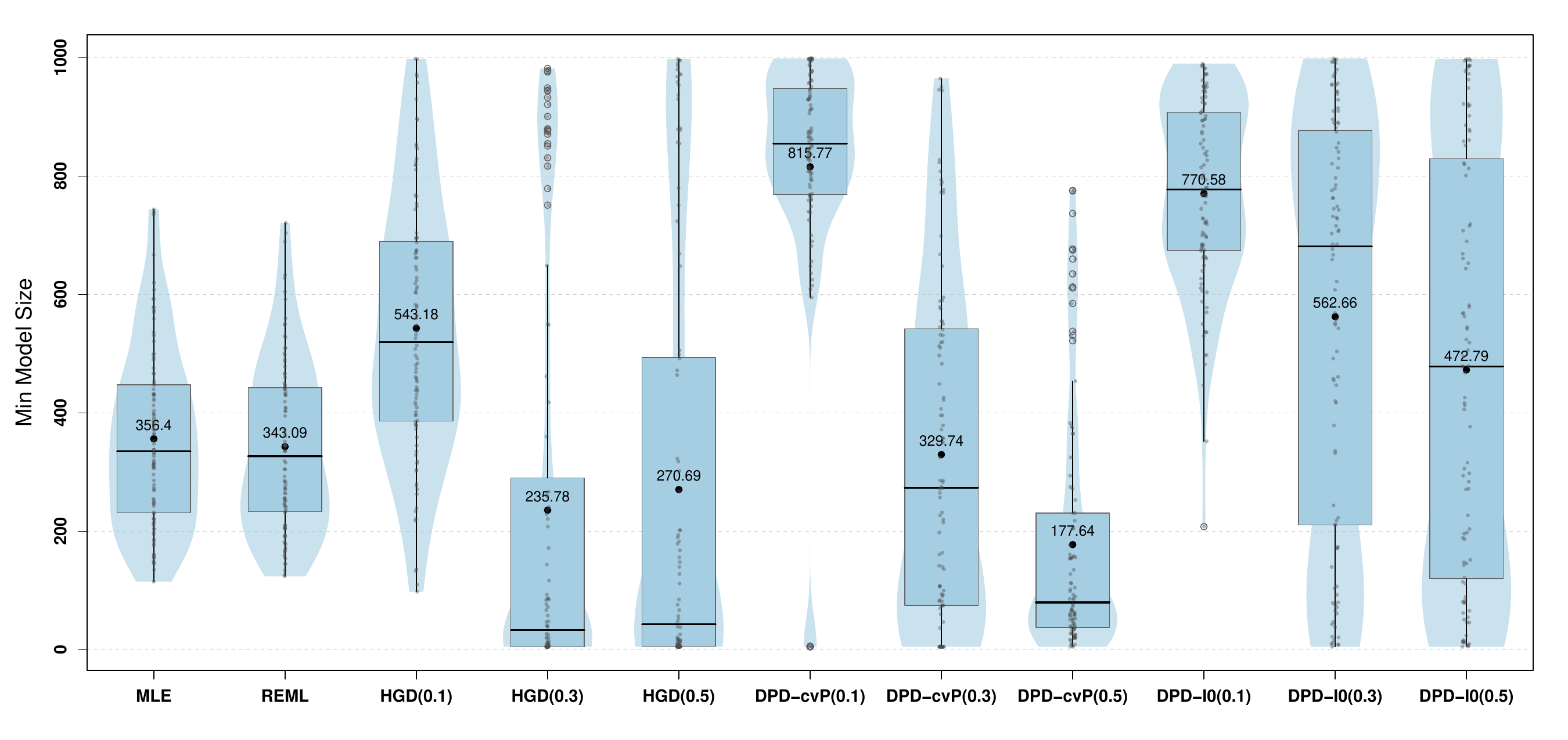}
		\label{FIG:boxplot_Y}}
	\\	
	\subfloat[\tiny 5\% (C4) Contamination]{
		\includegraphics[page=2, width=0.3\textwidth]{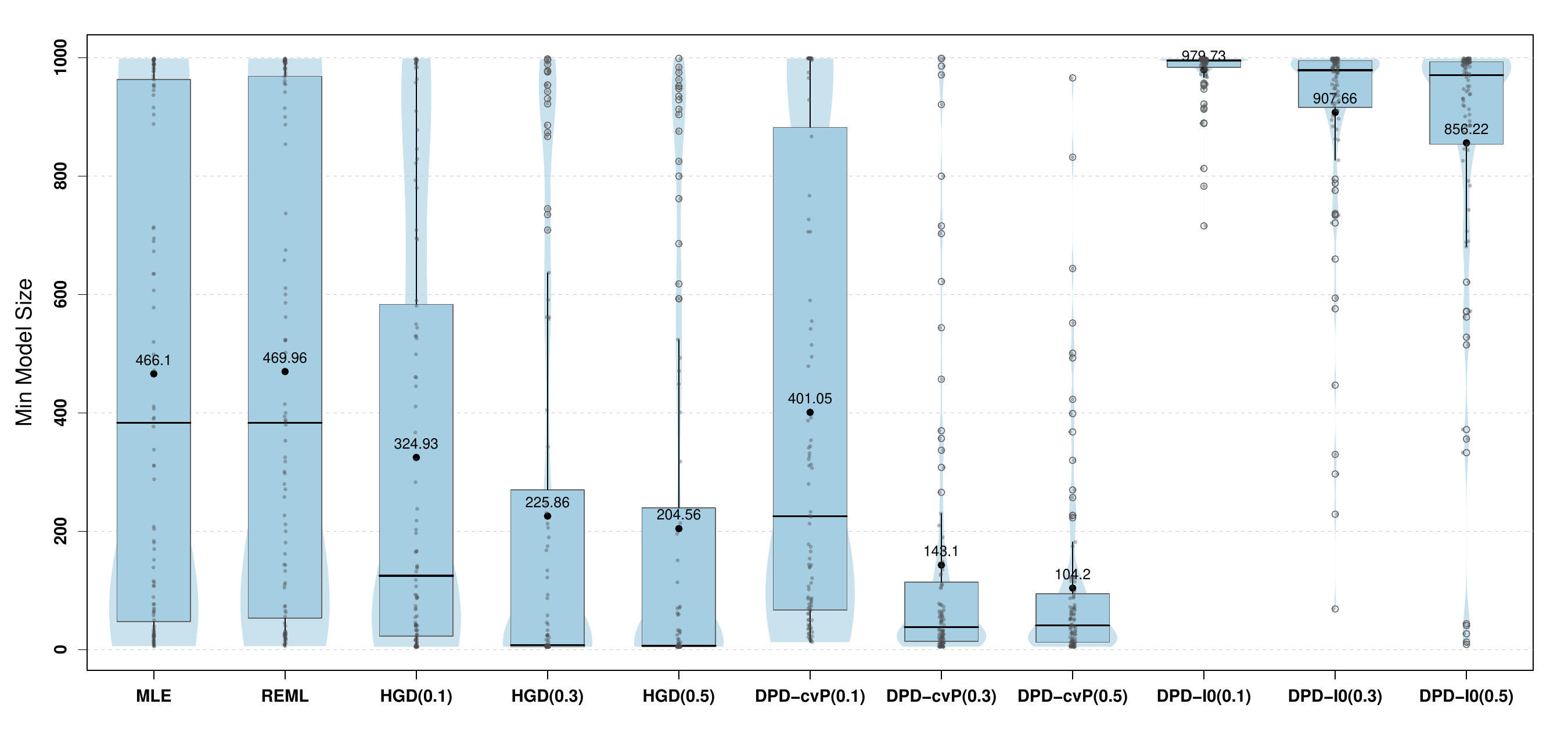}
		\label{FIG:boxplot_Y}}
	~	
	\subfloat[\tiny 10\% (C4) Contamination]{
		\includegraphics[page=2, width=0.3\textwidth]{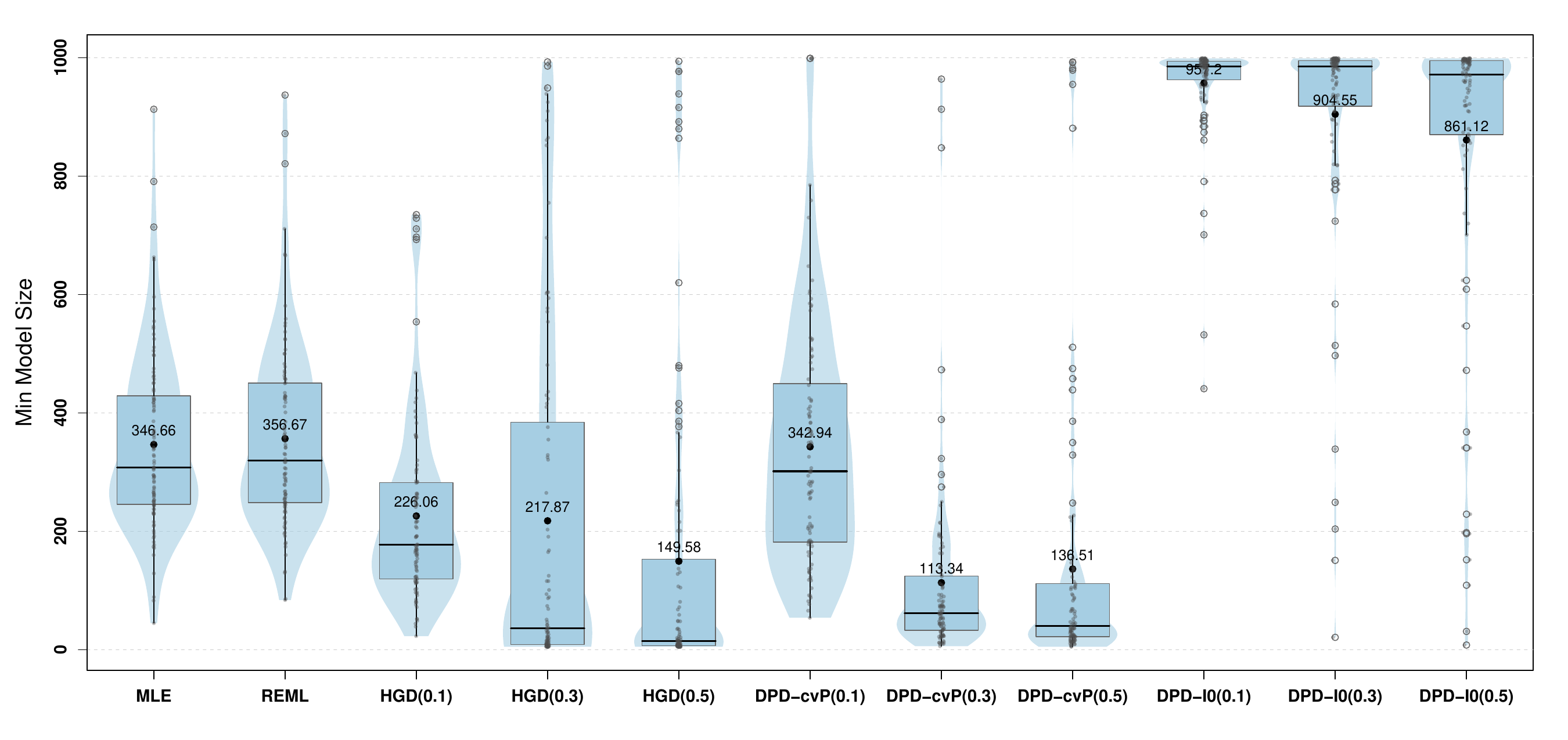}
		\label{FIG:boxplot_Y}}
	~	
	\subfloat[\tiny 20\% (C4) Contamination]{
		\includegraphics[page=2, width=0.3\textwidth]{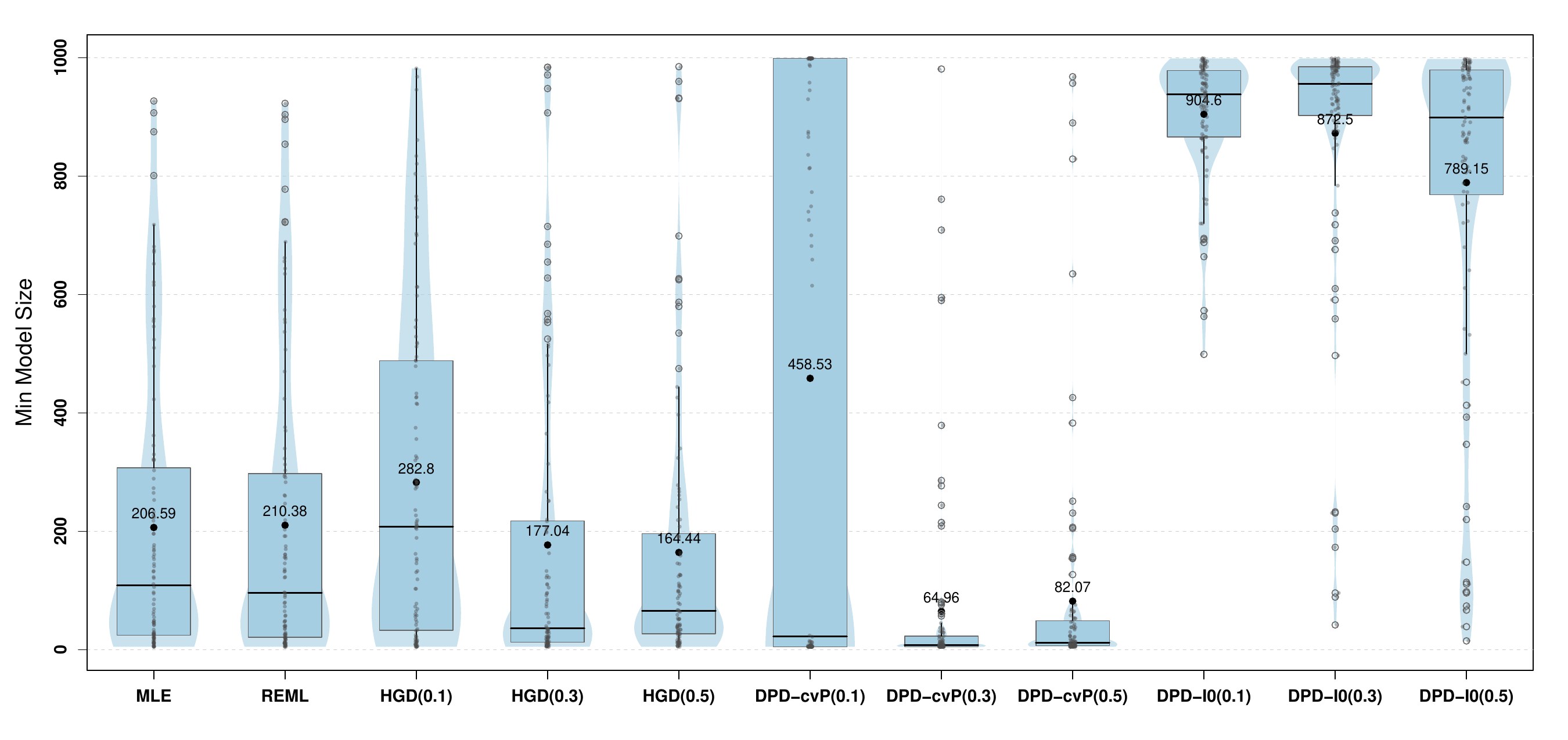}
		\label{FIG:boxplot_Y}}
	\caption{Boxplots  (with overlaid sample means) of MinMS required for sure screening under scenario (S1)$\times$(R1) with $\bm\Sigma_x$ as CS(0.3) and different types of contamination}
	\label{FIG:MMS_S1R1CS}
\end{figure}

\begin{figure}[!h]
	\centering
	\subfloat[\tiny 5\% (C1) Contamination]{
		\includegraphics[page=2, width=0.3\textwidth]{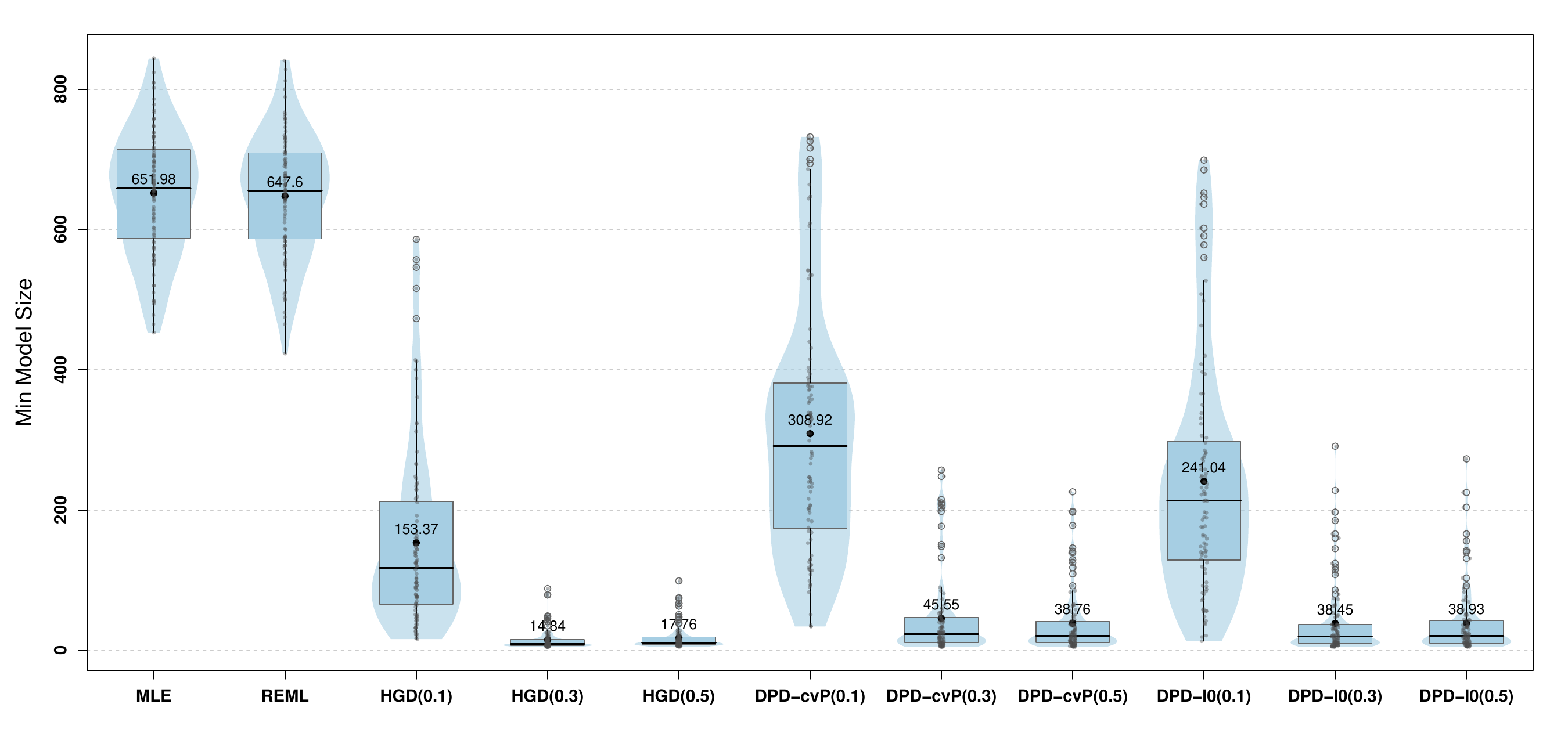}
		\label{FIG:boxplot_Y}}
	~	
	\subfloat[\tiny 10\% (C1) Contamination]{
		\includegraphics[page=2, width=0.3\textwidth]{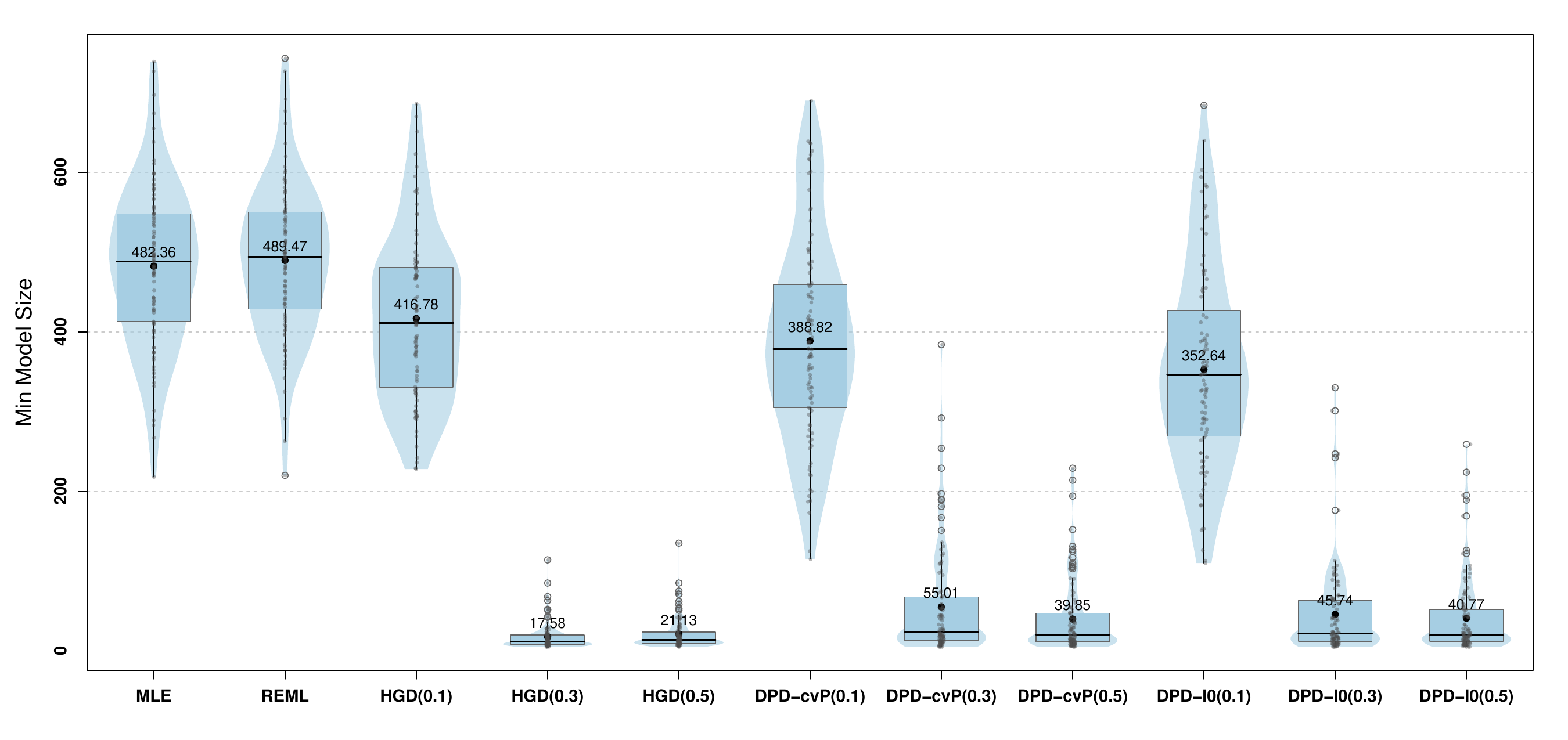}
		\label{FIG:boxplot_Y}}
	~	
	\subfloat[\tiny 20\% (C1) Contamination]{
		\includegraphics[page=2, width=0.3\textwidth]{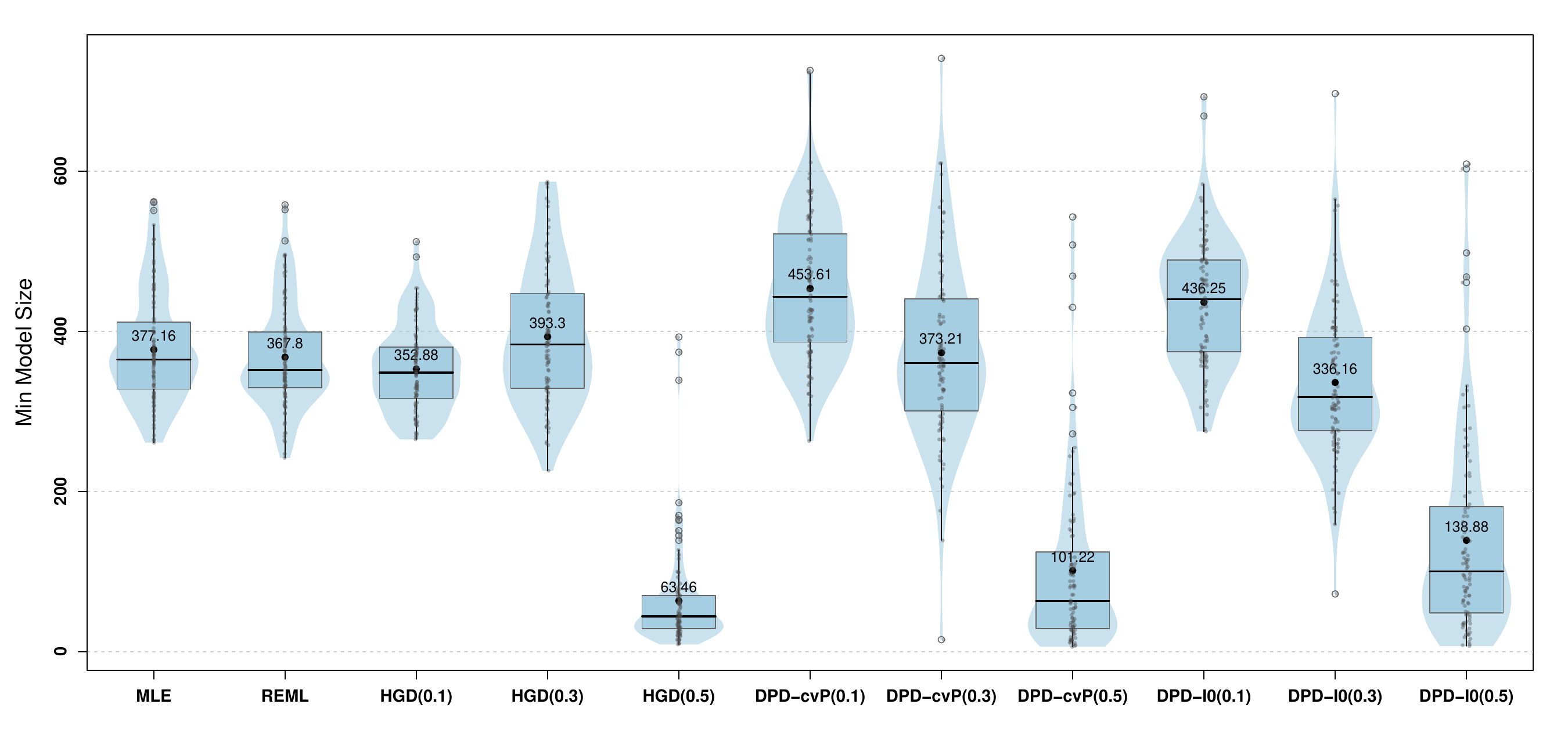}
		\label{FIG:boxplot_Y}}
	\\	
	\subfloat[\tiny 5\% (C2) Contamination]{
		\includegraphics[page=2, width=0.3\textwidth]{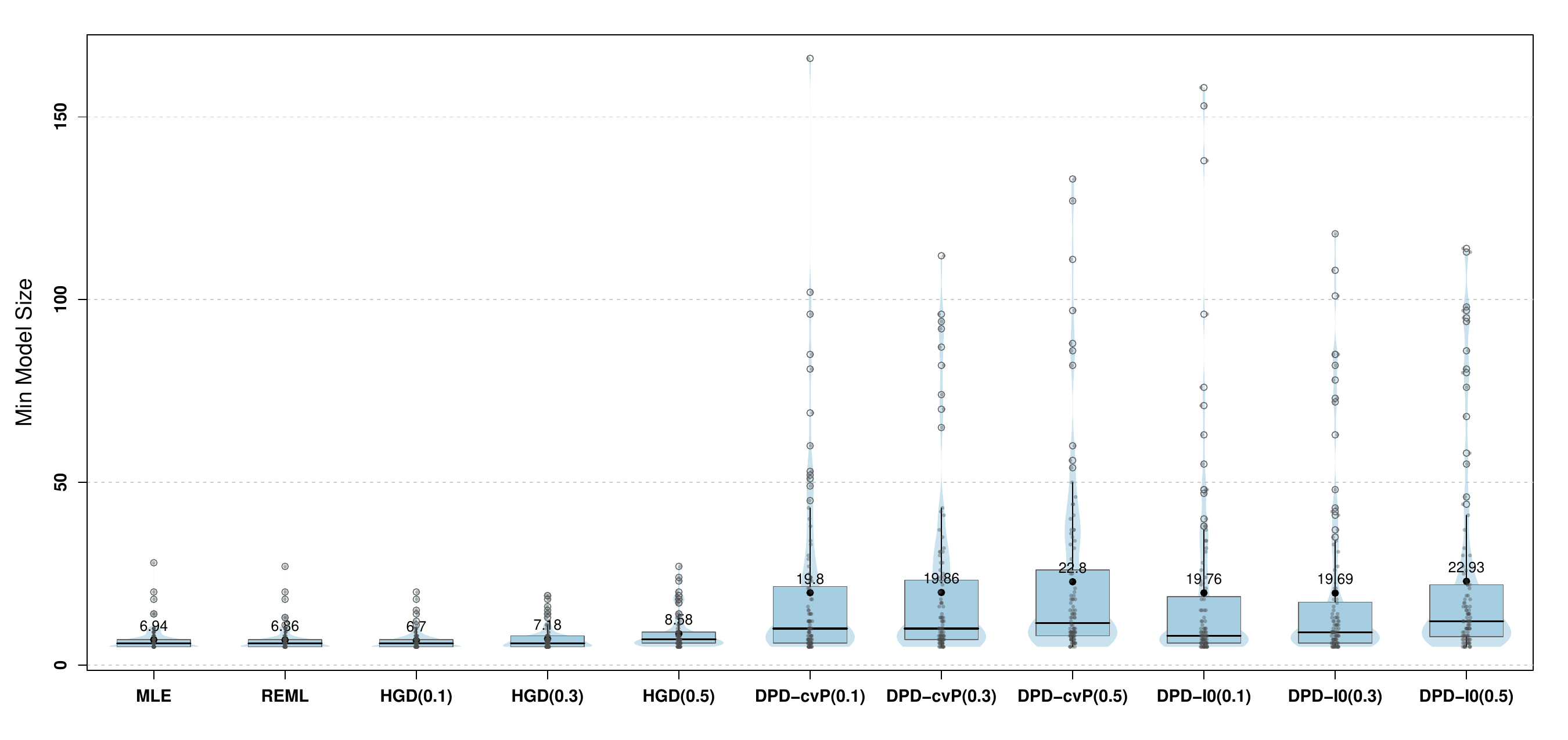}
		\label{FIG:boxplot_Y}}
	~	
	\subfloat[\tiny 10\% (C2) Contamination]{
		\includegraphics[page=2, width=0.3\textwidth]{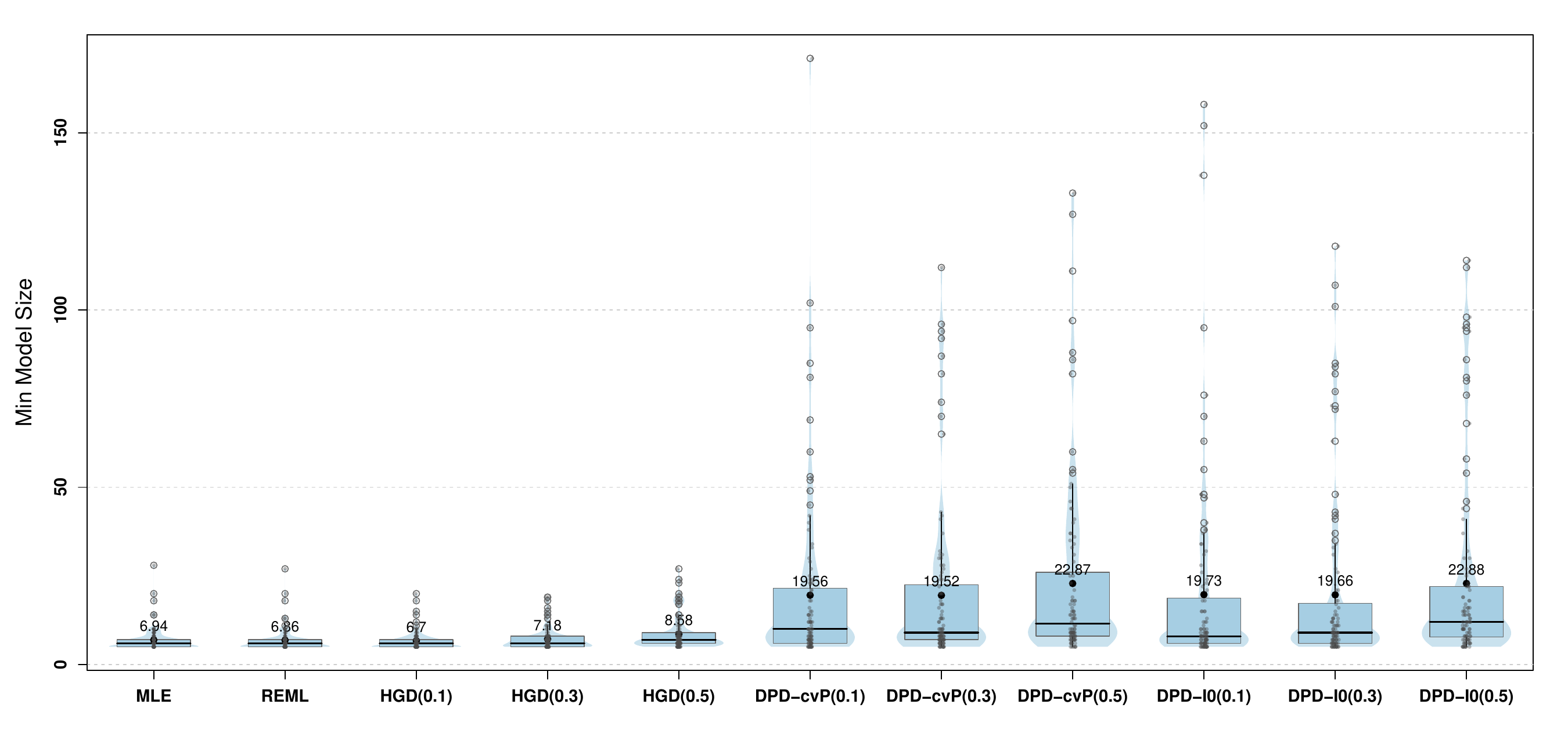}
		\label{FIG:boxplot_Y}}
	~	
	\subfloat[\tiny 20\% (C2) Contamination]{
		\includegraphics[page=2, width=0.3\textwidth]{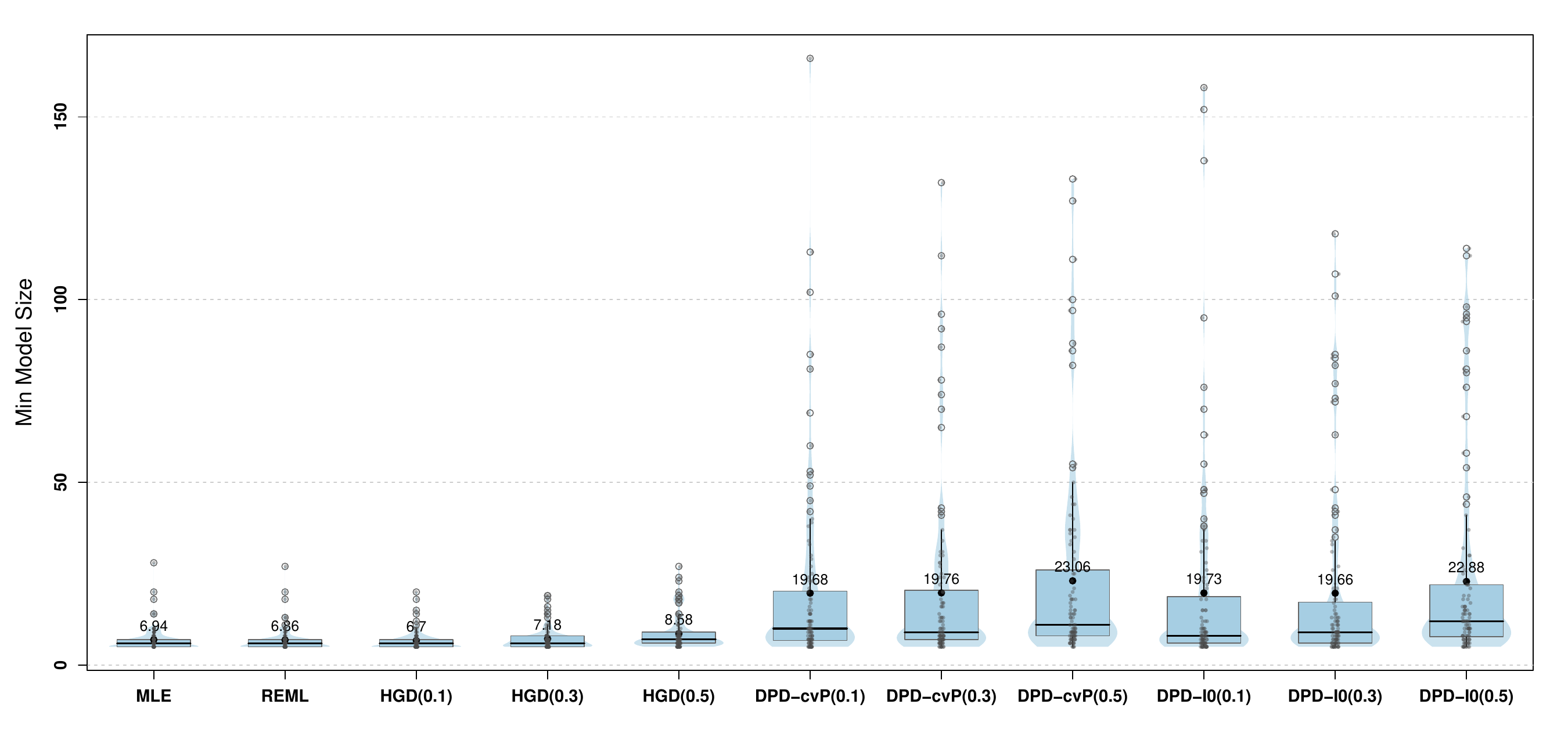}
		\label{FIG:boxplot_Y}}
	\\	
	\subfloat[\tiny 5\% (C3) Contamination]{
		\includegraphics[page=2, width=0.3\textwidth]{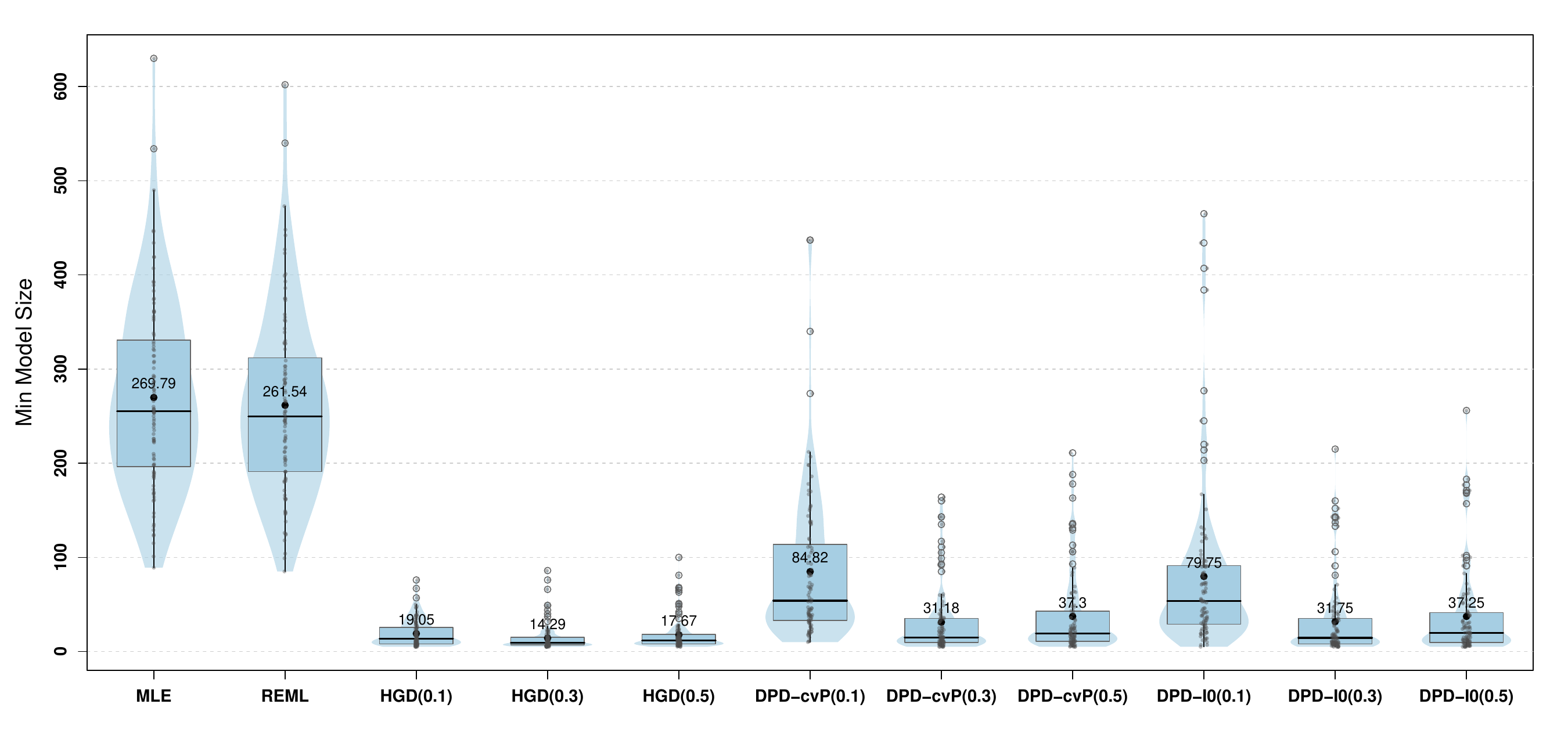}
		\label{FIG:boxplot_Y}}
	~	
	\subfloat[\tiny 10\% (C3) Contamination]{
		\includegraphics[page=2, width=0.3\textwidth]{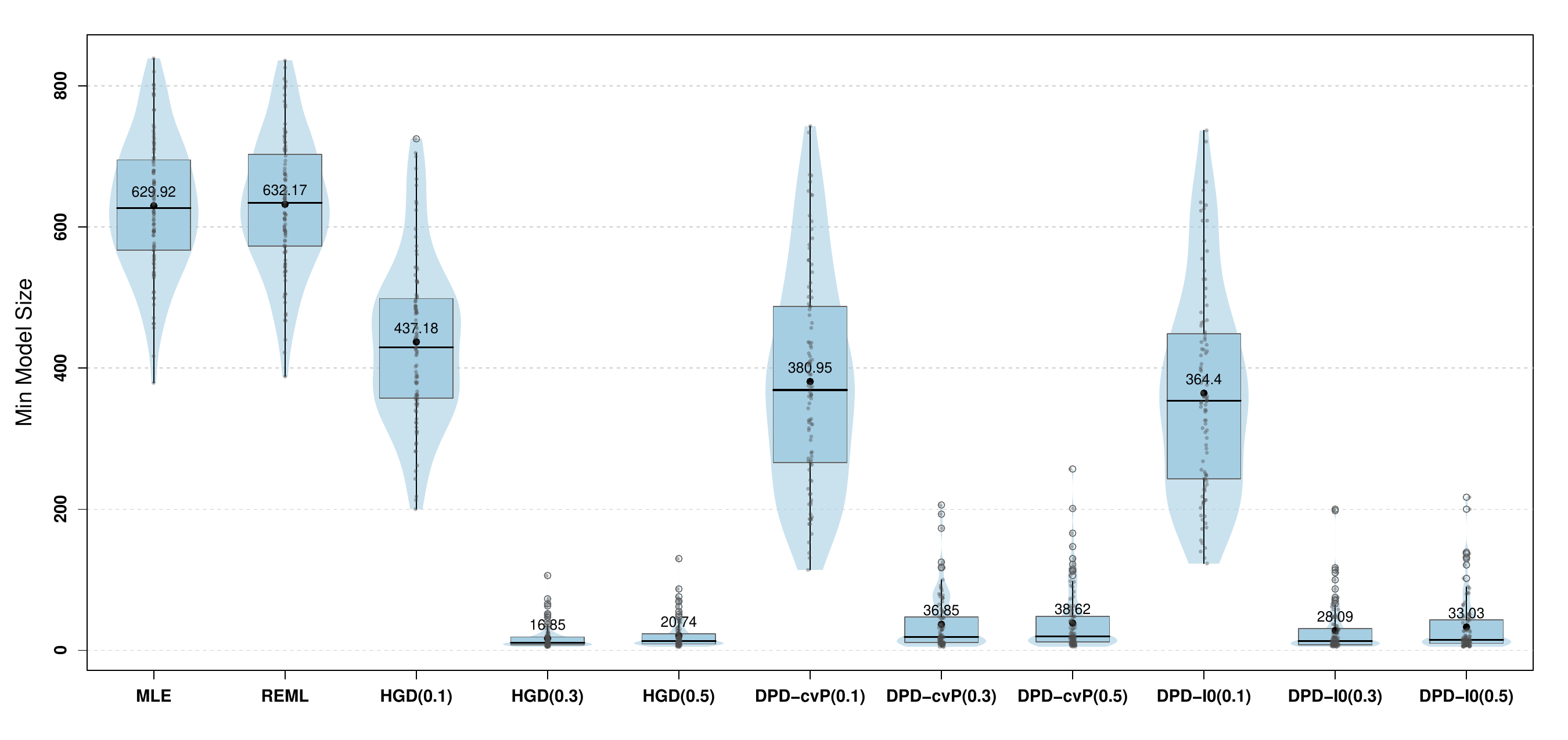}
		\label{FIG:boxplot_Y}}
	~	
	\subfloat[\tiny 20\% (C3) Contamination]{
		\includegraphics[page=2, width=0.3\textwidth]{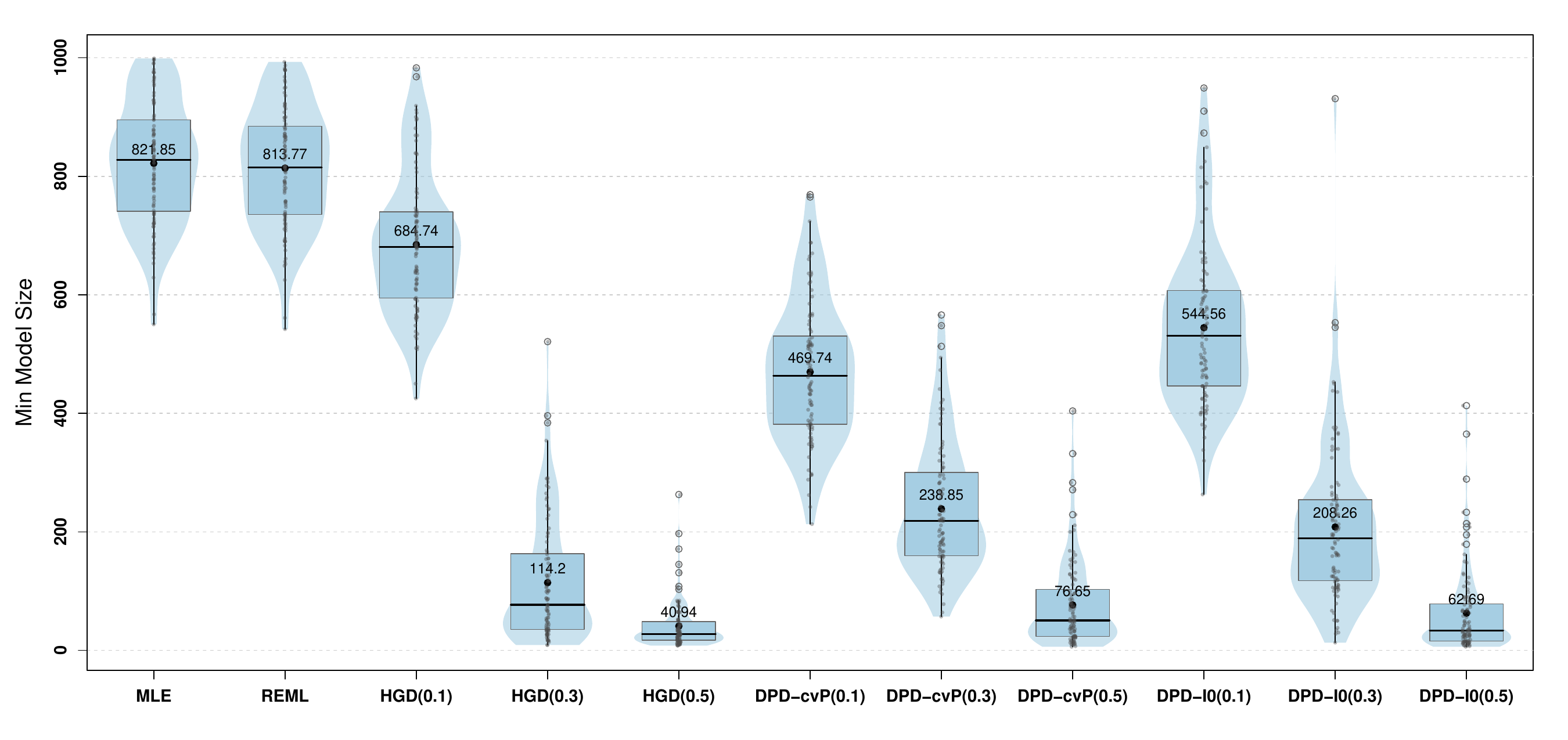}
		\label{FIG:boxplot_Y}}
	\\	
	\subfloat[\tiny 5\% (C4) Contamination]{
		\includegraphics[page=2, width=0.3\textwidth]{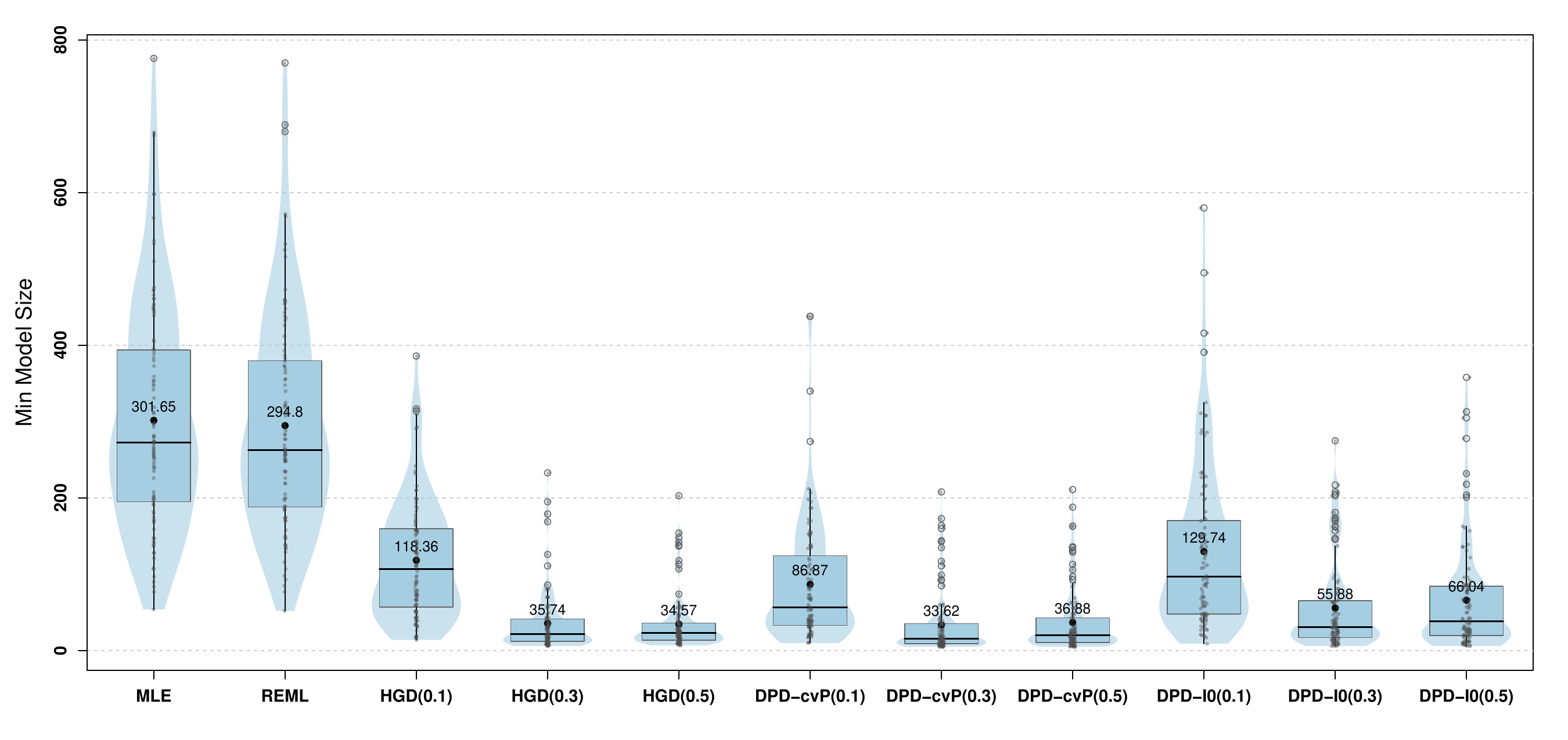}
		\label{FIG:boxplot_Y}}
	~	
	\subfloat[\tiny 10\% (C4) Contamination]{
		\includegraphics[page=2, width=0.3\textwidth]{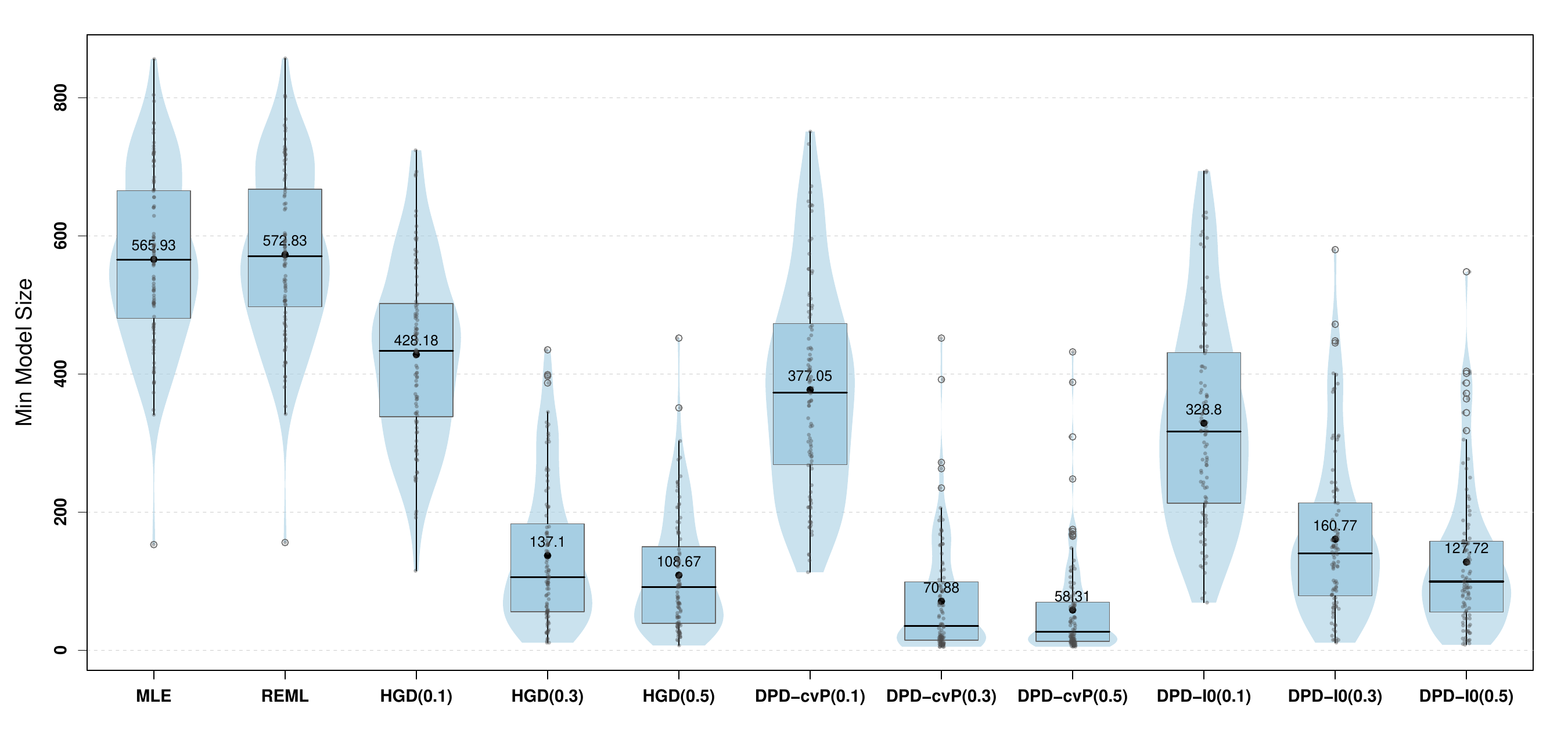}
		\label{FIG:boxplot_Y}}
	~	
	\subfloat[\tiny 20\% (C4) Contamination]{
		\includegraphics[page=2, width=0.3\textwidth]{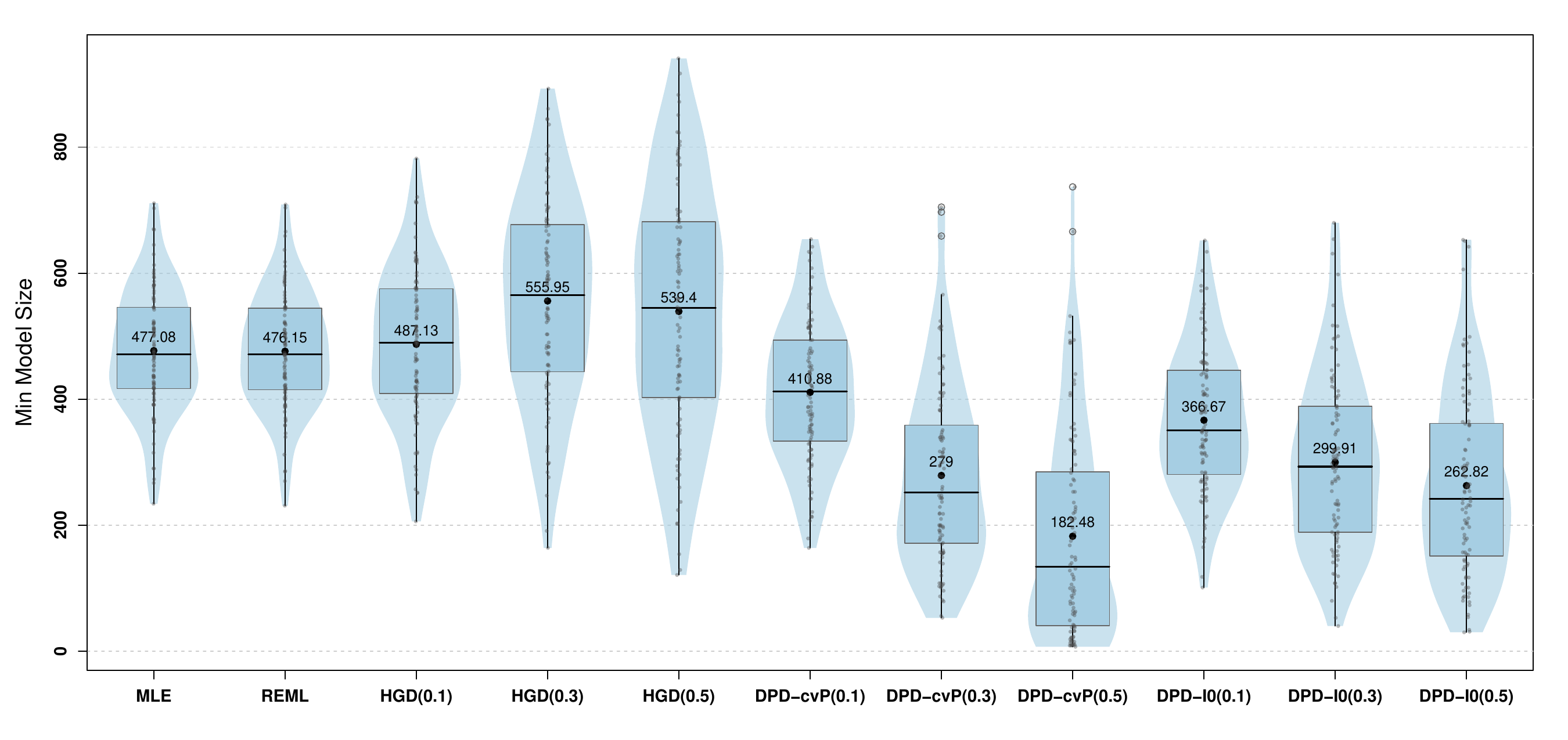}
		\label{FIG:boxplot_Y}}
	\caption{Boxplots  (with overlaid sample means) of MinMS required for sure screening under scenario (S2)$\times$(R1) with $\bm\Sigma_x$ as CS(0.3) and different types of contamination}
	\label{FIG:MMS_S2R1CS}
\end{figure}

\begin{figure}[!h]
	\centering
	\subfloat[\tiny 5\% (C1) Contamination]{
		\includegraphics[page=2, width=0.3\textwidth]{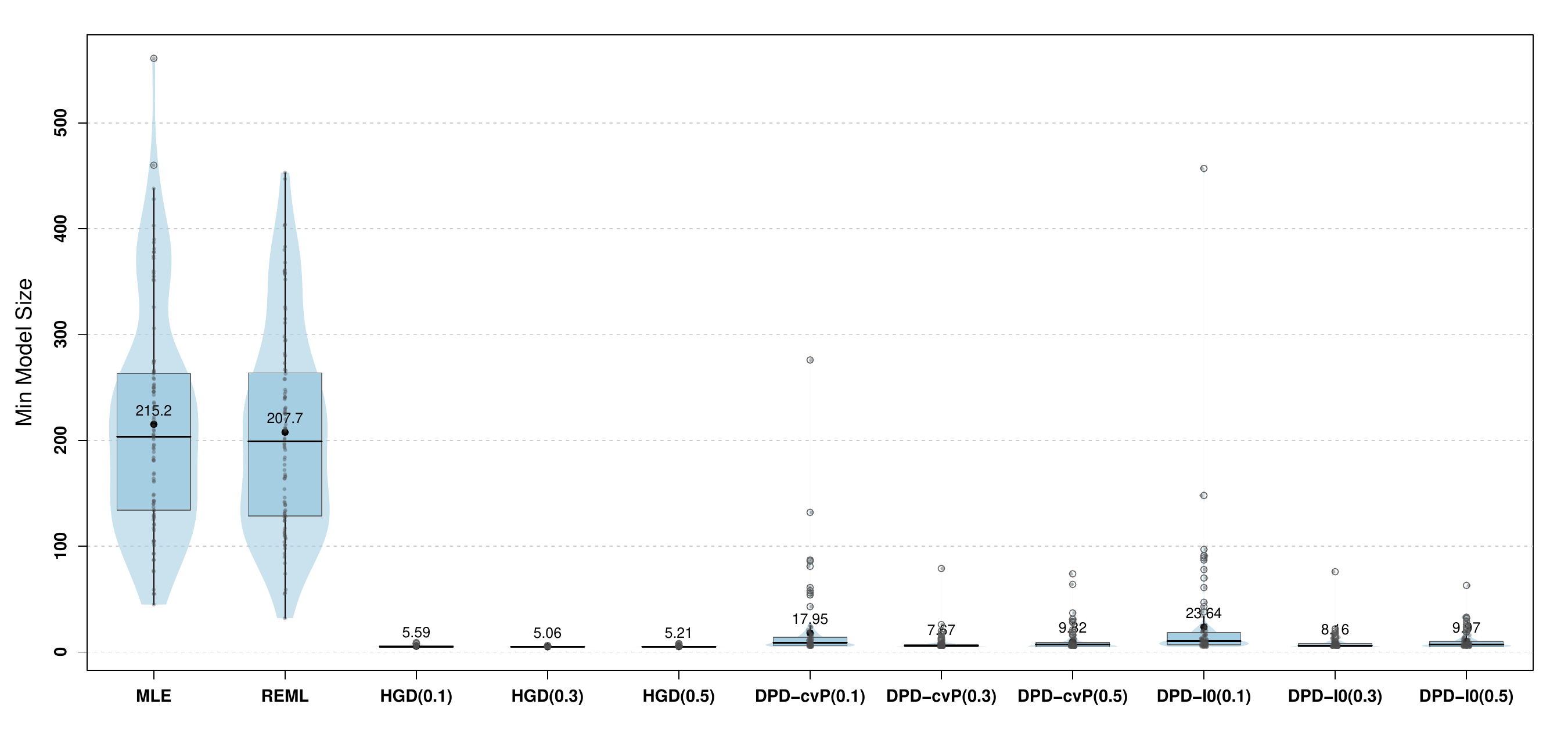}
		\label{FIG:boxplot_Y}}
	~	
	\subfloat[\tiny 10\% (C1) Contamination]{
		\includegraphics[page=2, width=0.3\textwidth]{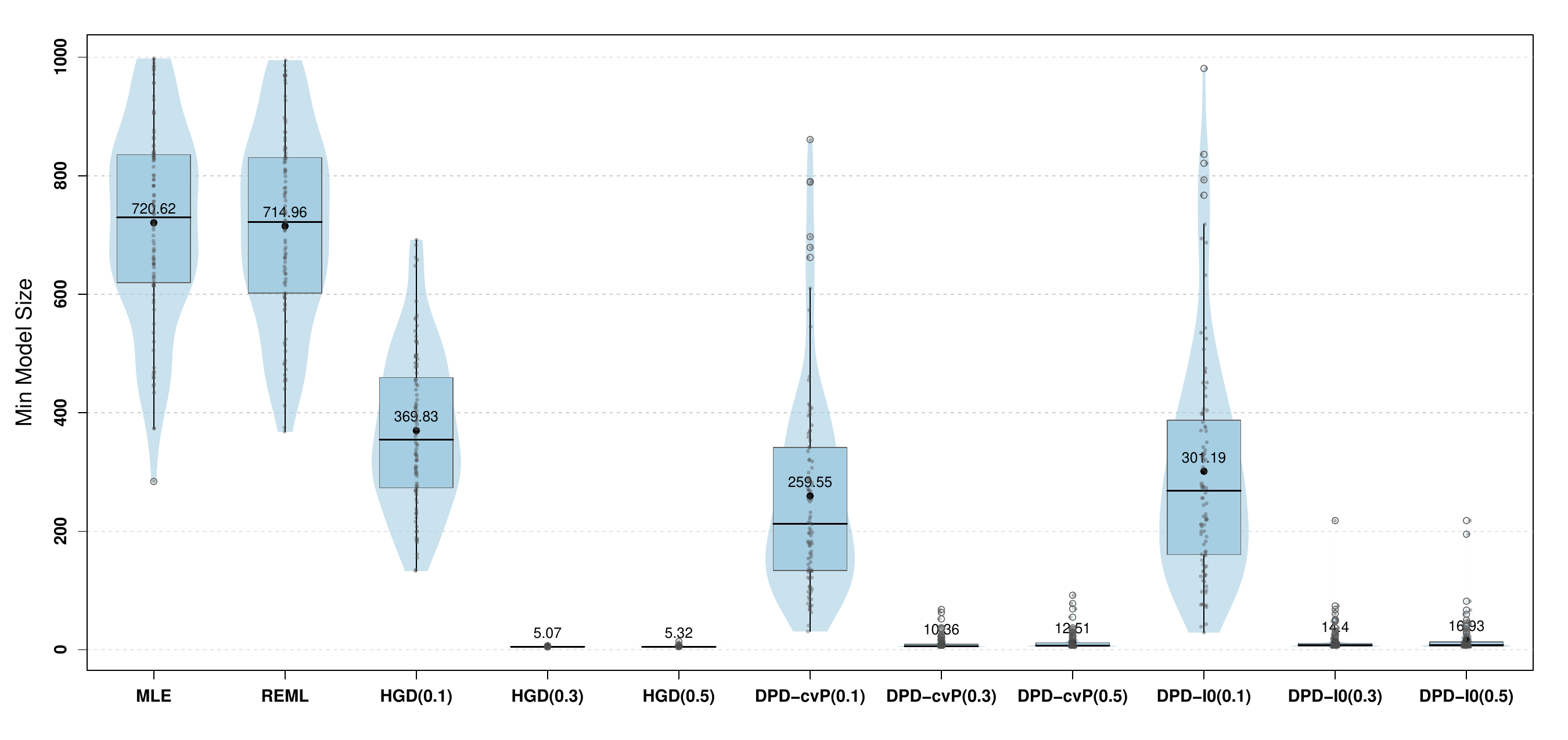}
		\label{FIG:boxplot_Y}}
	~	
	\subfloat[\tiny 20\% (C1) Contamination]{
		\includegraphics[page=2, width=0.3\textwidth]{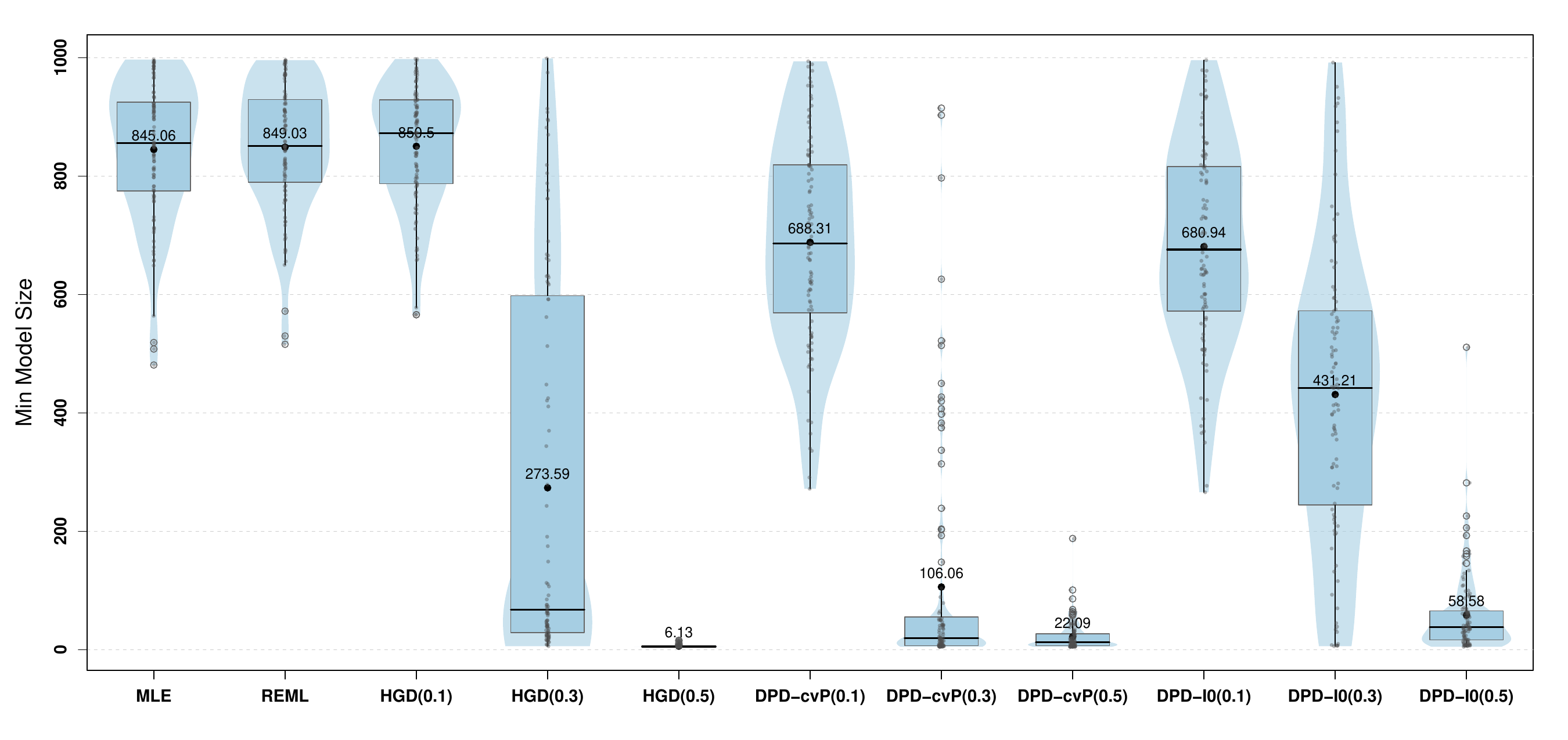}
		\label{FIG:boxplot_Y}}
	\\	
	\subfloat[\tiny 5\% (C2) Contamination]{
		\includegraphics[page=2, width=0.3\textwidth]{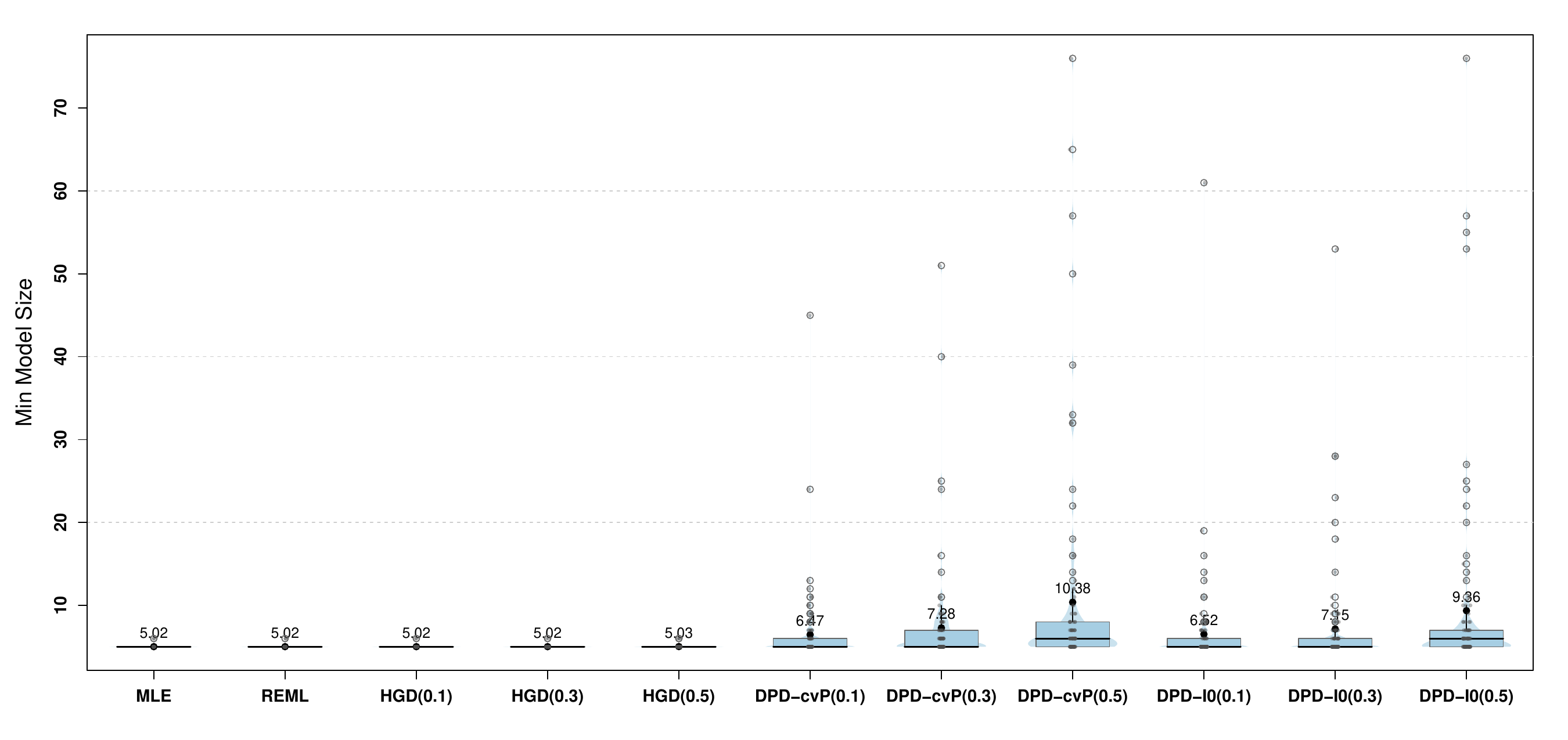}
		\label{FIG:boxplot_Y}}
	~	
	\subfloat[\tiny 10\% (C2) Contamination]{
		\includegraphics[page=2, width=0.3\textwidth]{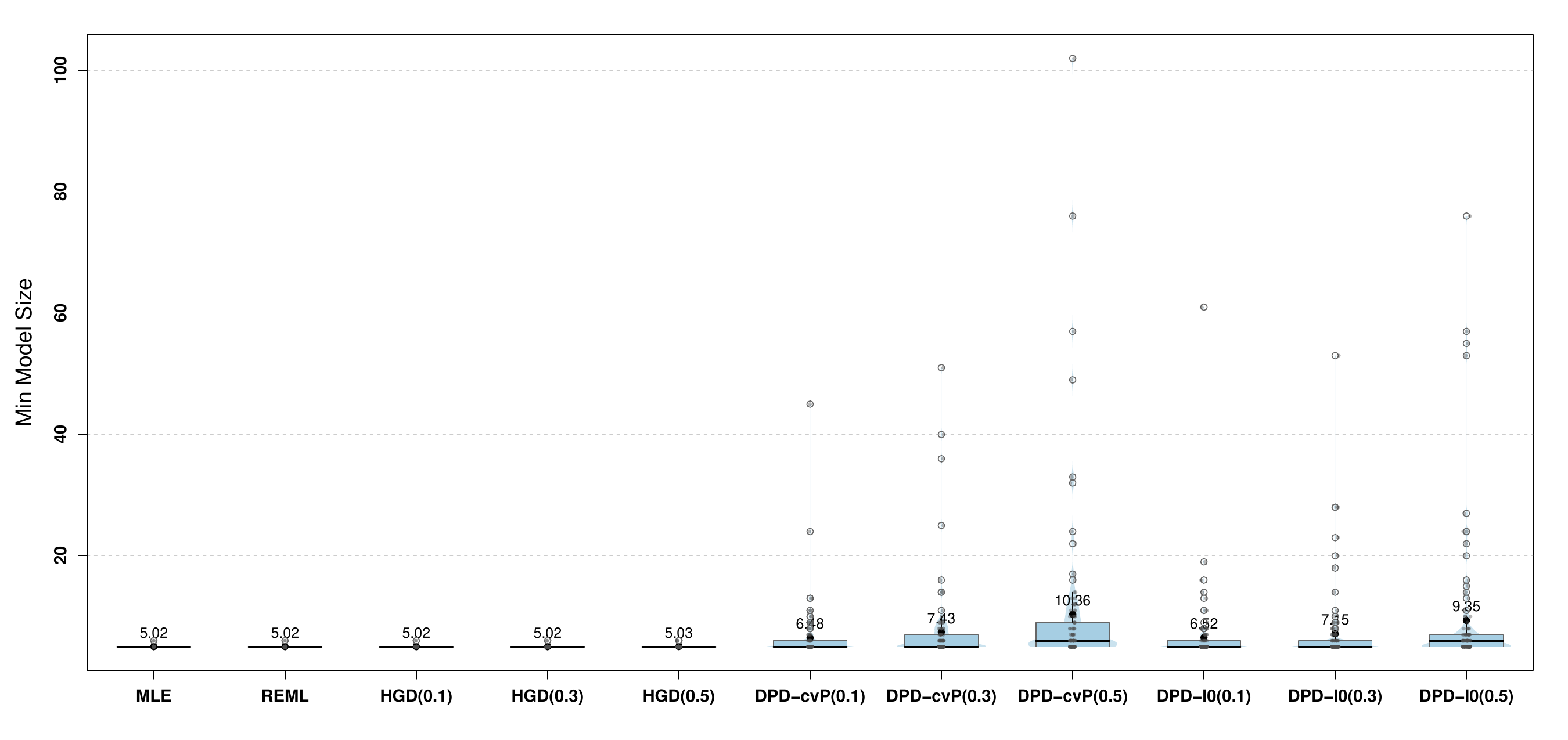}
		\label{FIG:boxplot_Y}}
	~	
	\subfloat[\tiny 20\% (C2) Contamination]{
		\includegraphics[page=2, width=0.3\textwidth]{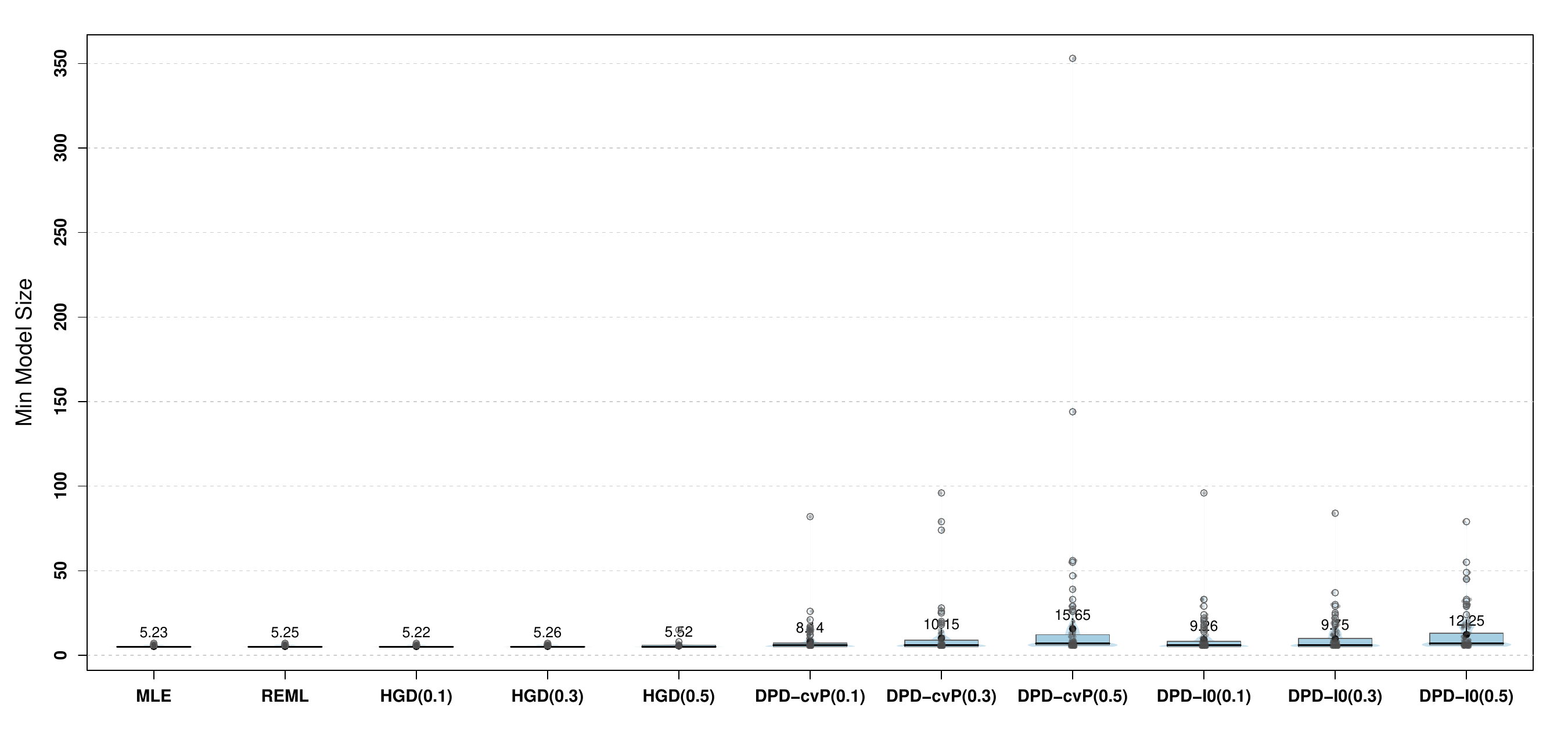}
		\label{FIG:boxplot_Y}}
	\\	
	\subfloat[\tiny 5\% (C3) Contamination]{
		\includegraphics[page=2, width=0.3\textwidth]{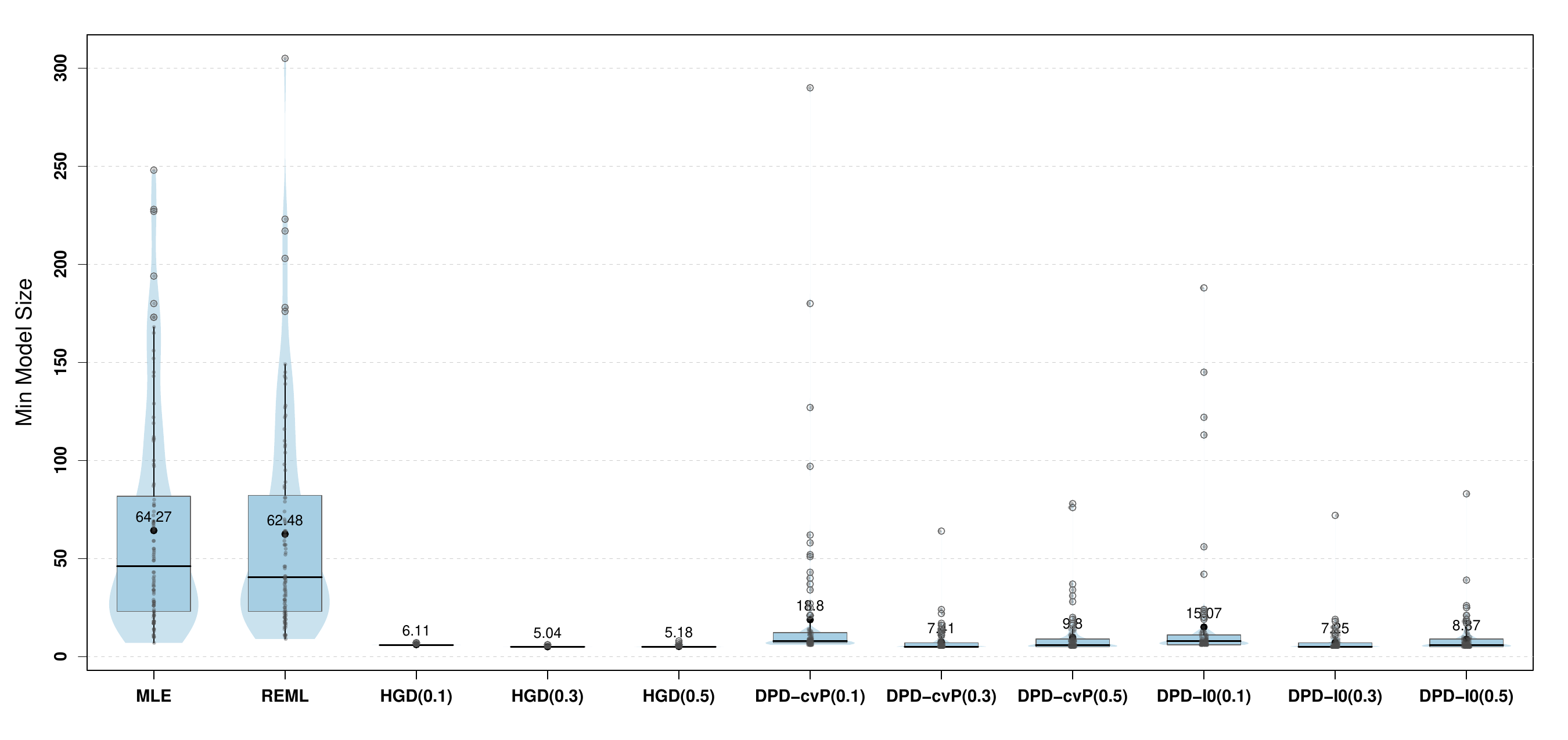}
		\label{FIG:boxplot_Y}}
	~	
	\subfloat[\tiny 10\% (C3) Contamination]{
		\includegraphics[page=2, width=0.3\textwidth]{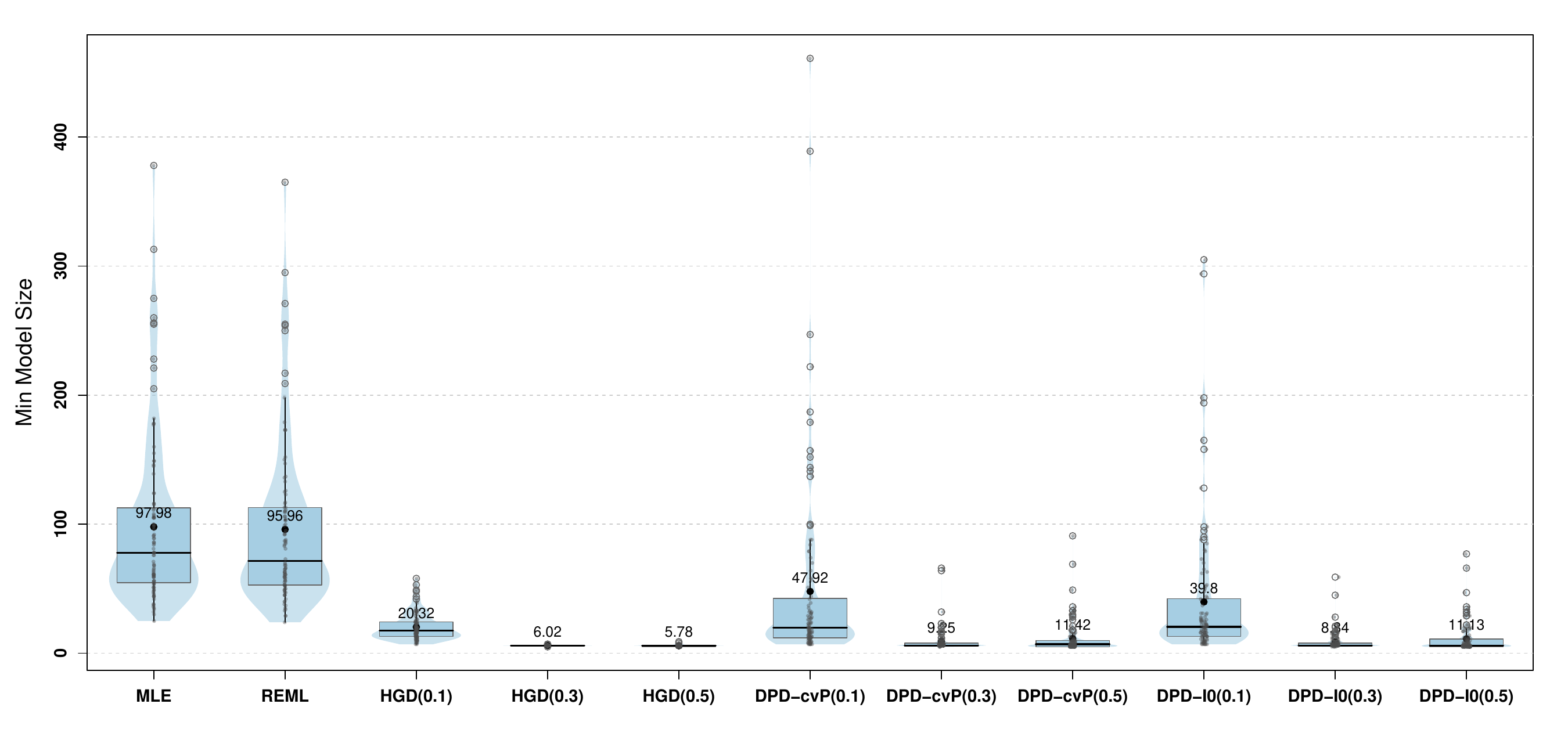}
		\label{FIG:boxplot_Y}}
	~	
	\subfloat[\tiny 20\% (C3) Contamination]{
		\includegraphics[page=2, width=0.3\textwidth]{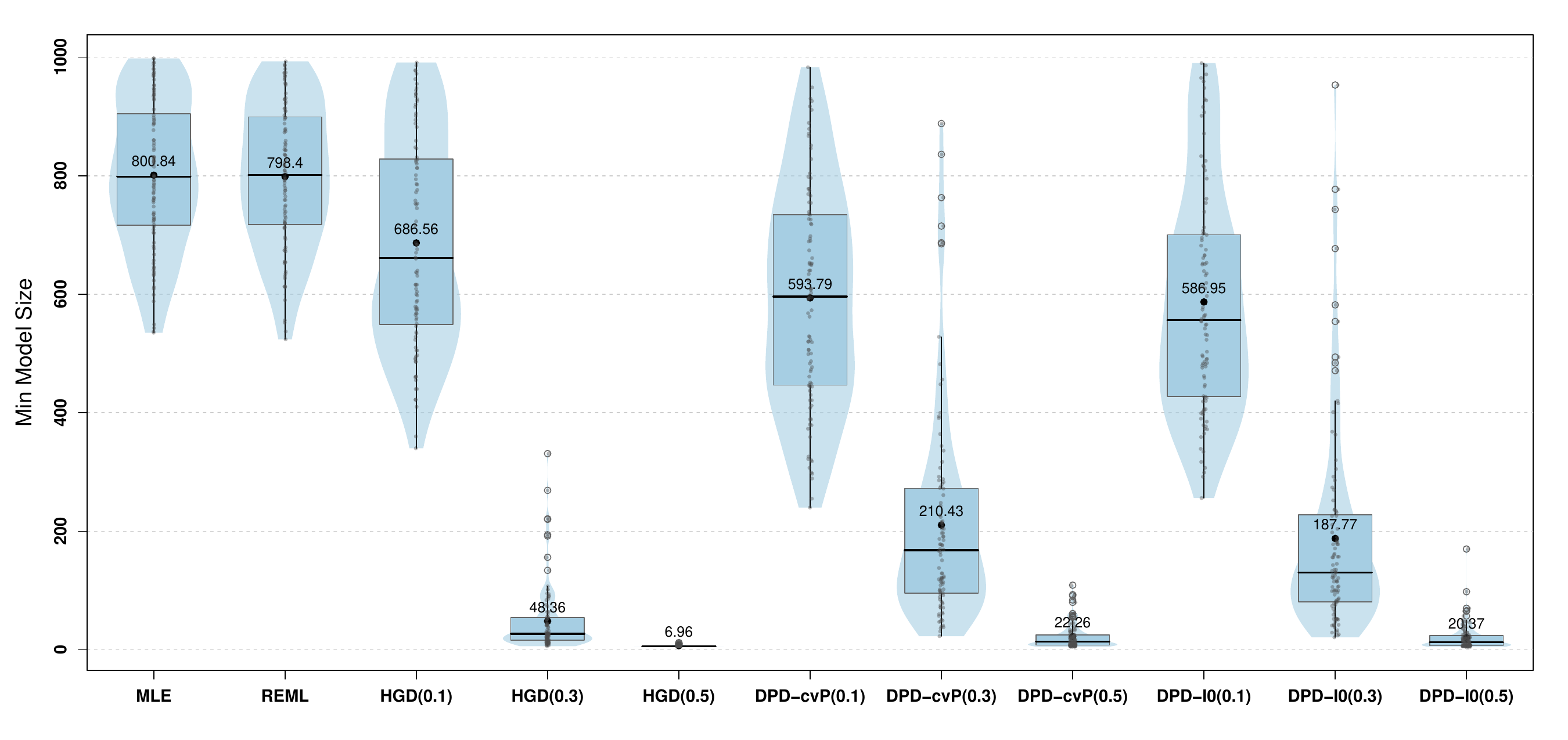}
		\label{FIG:boxplot_Y}}
	\\	
	\subfloat[\tiny 5\% (C4) Contamination]{
		\includegraphics[page=2, width=0.3\textwidth]{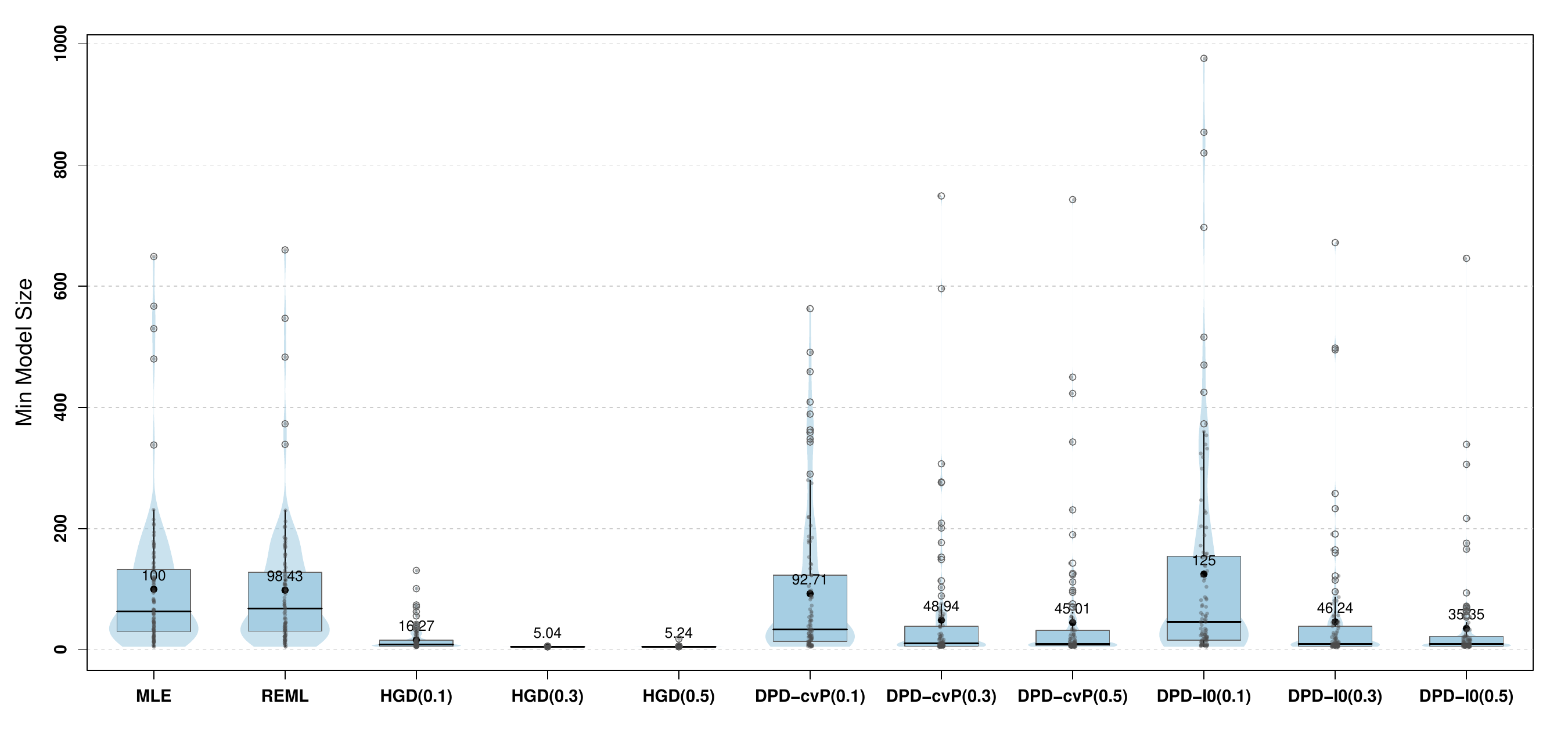}
		\label{FIG:boxplot_Y}}
	~	
	\subfloat[\tiny 10\% (C4) Contamination]{
		\includegraphics[page=2, width=0.3\textwidth]{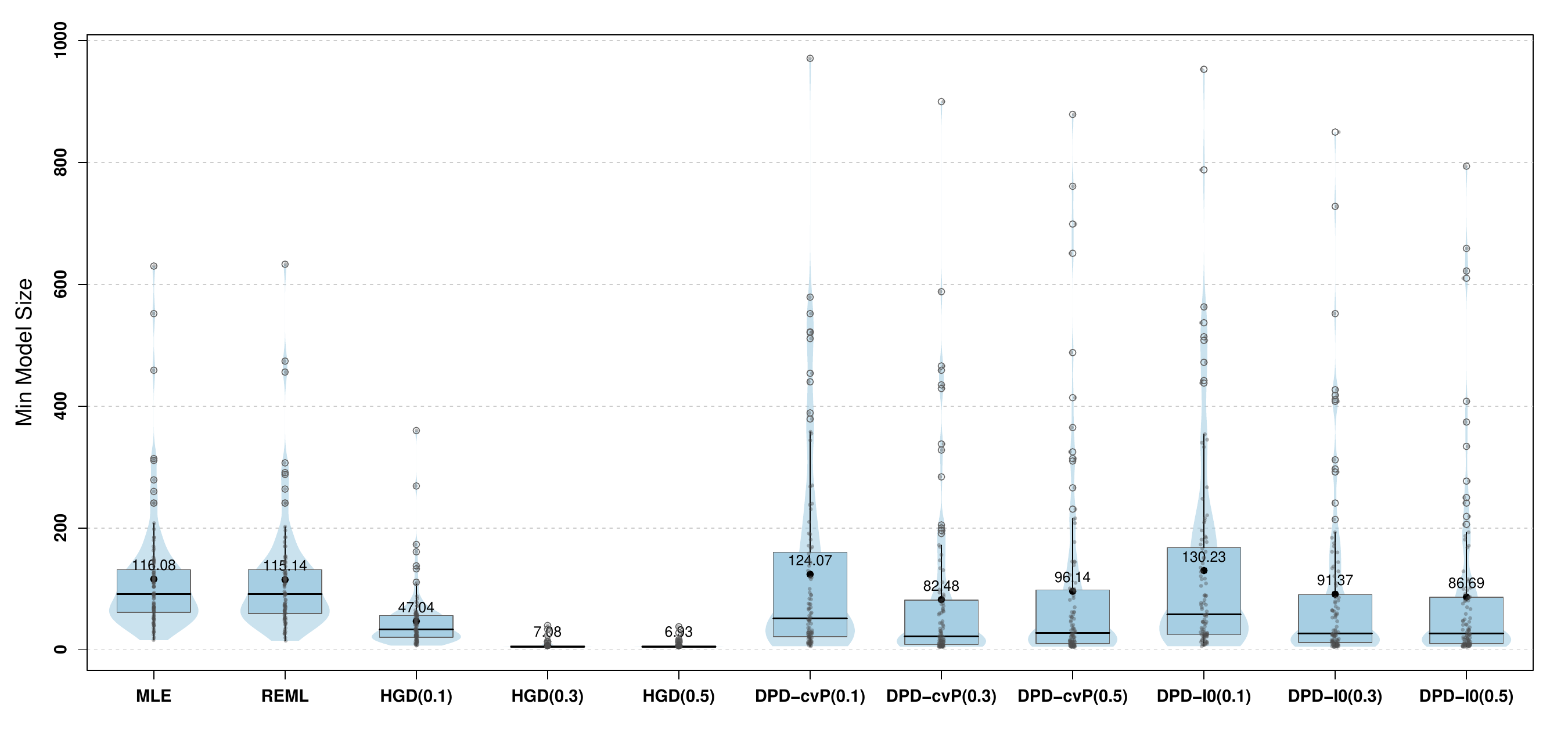}
		\label{FIG:boxplot_Y}}
	~	
	\subfloat[\tiny 20\% (C4) Contamination]{
		\includegraphics[page=2, width=0.3\textwidth]{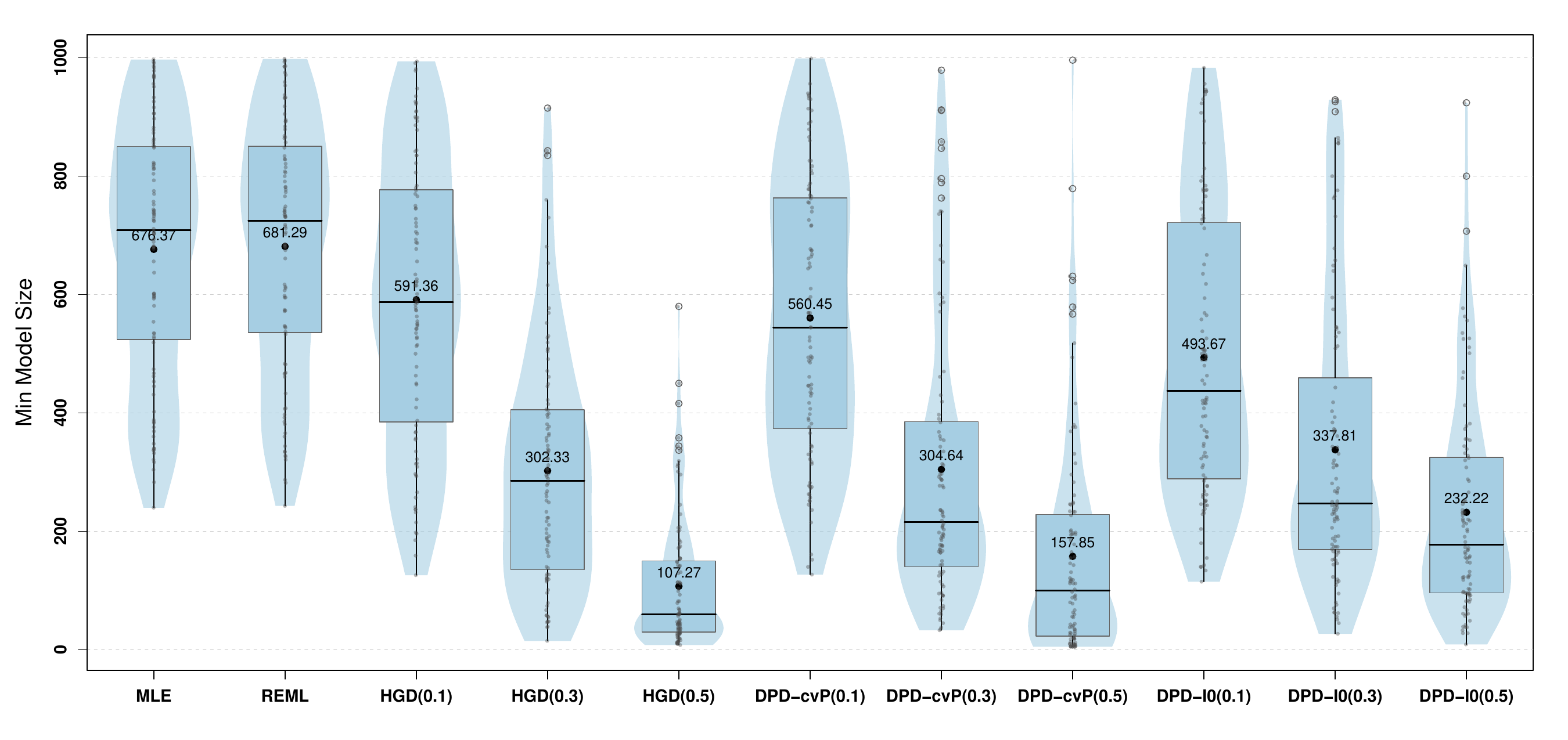}
		\label{FIG:boxplot_Y}}
	\caption{Boxplots  (with overlaid sample means) of MinMS required for sure screening under scenario (S2)$\times$(R0) with $\bm\Sigma_x$ as T(0.3) and different types of contamination}
	\label{FIG:MMS_S2R0T}
\end{figure}

\begin{figure}[!h]
	\centering
	\subfloat[\tiny 5\% (C1) Contamination]{
		\includegraphics[page=2, width=0.3\textwidth]{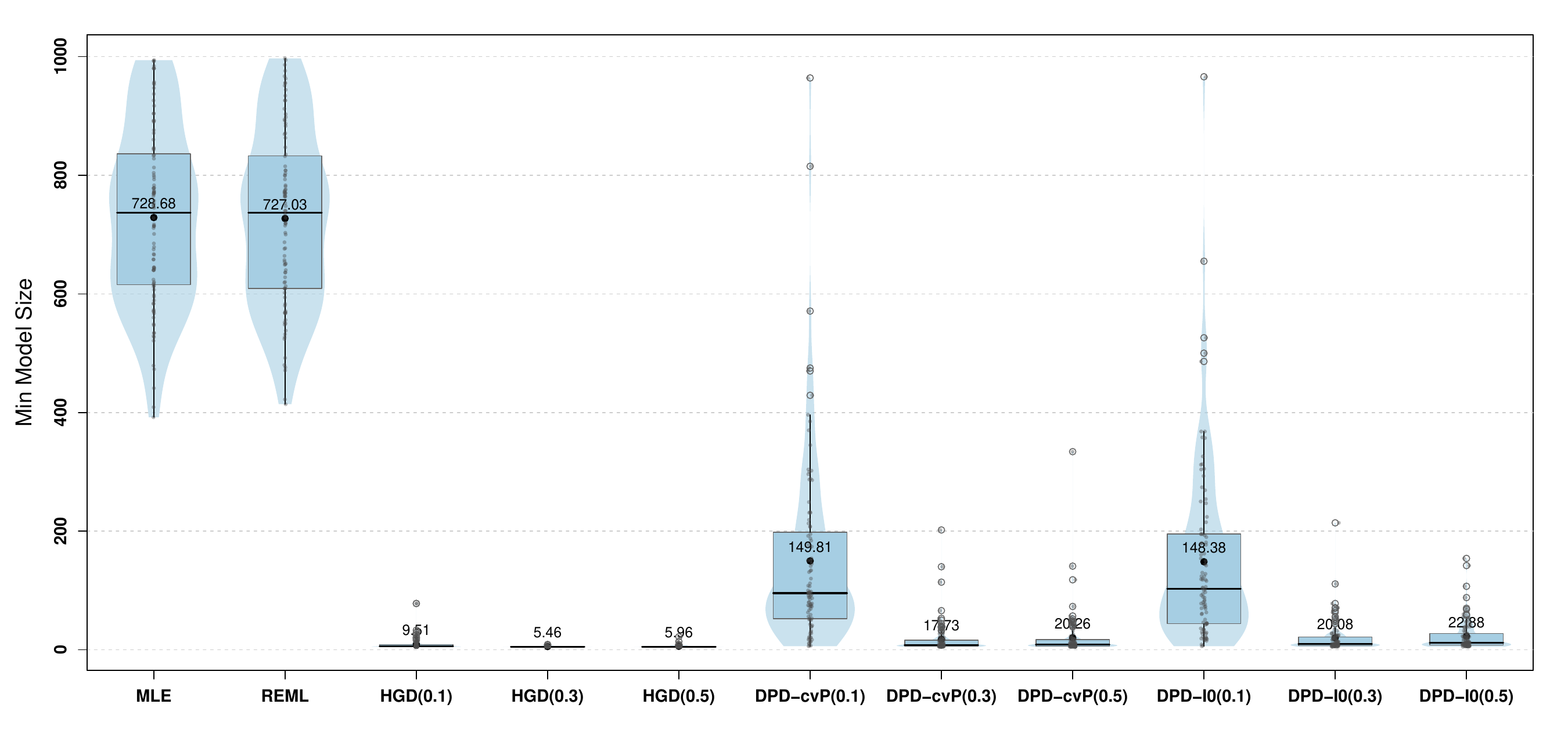}
		\label{FIG:boxplot_Y}}
	~	
	\subfloat[\tiny 10\% (C1) Contamination]{
		\includegraphics[page=2, width=0.3\textwidth]{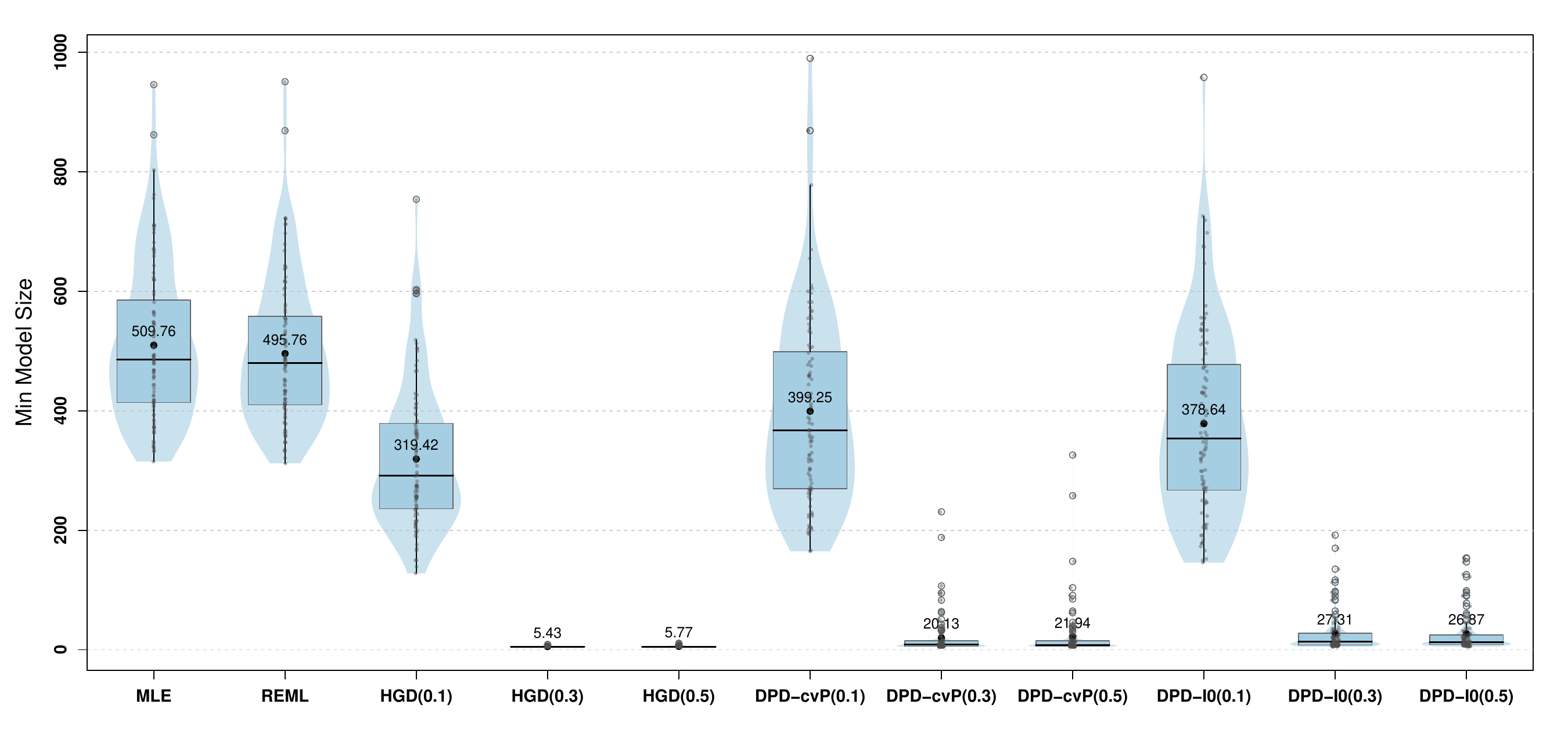}
		\label{FIG:boxplot_Y}}
	~	
	\subfloat[\tiny 20\% (C1) Contamination]{
		\includegraphics[page=2, width=0.3\textwidth]{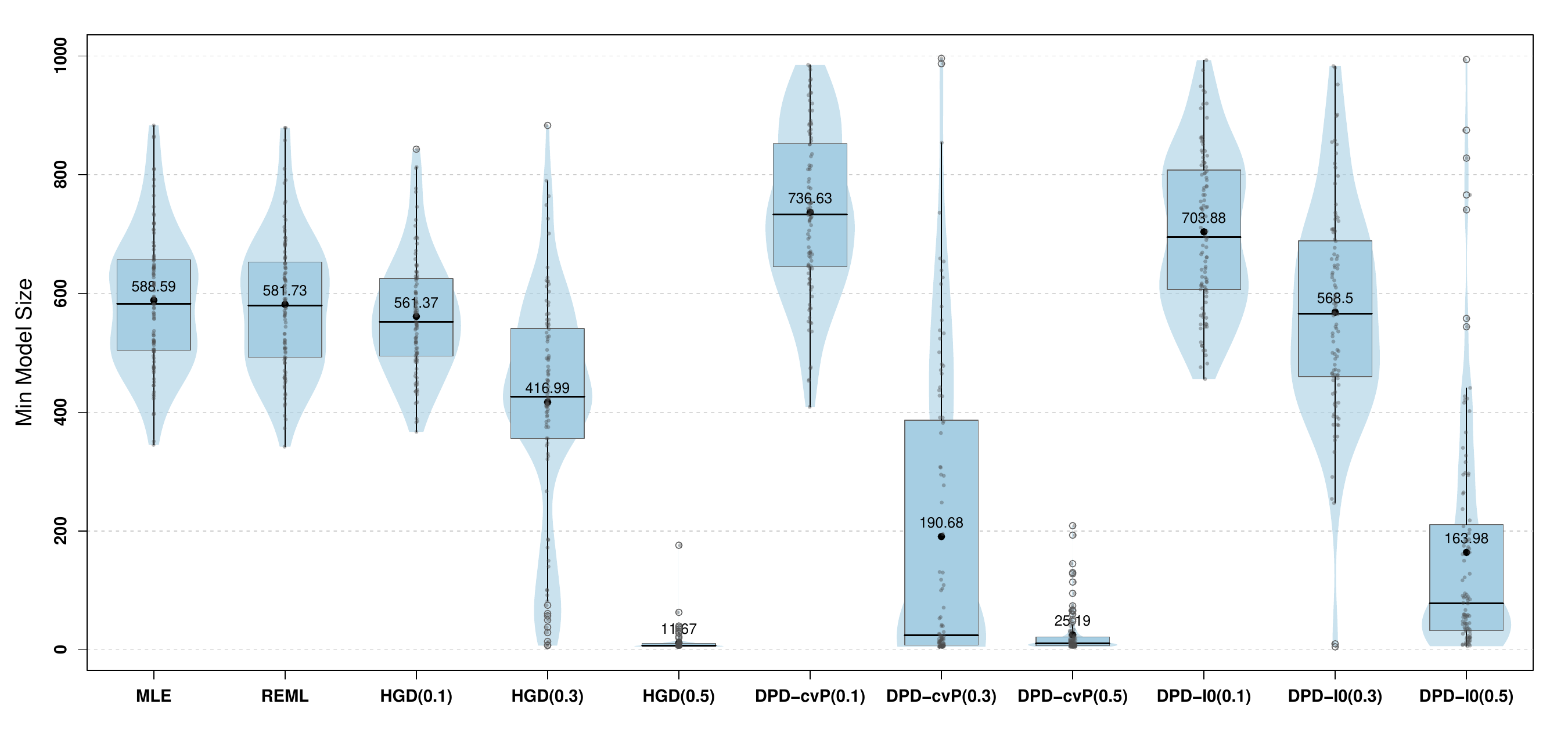}
		\label{FIG:boxplot_Y}}
	\\	
	\subfloat[\tiny 5\% (C2) Contamination]{
		\includegraphics[page=2, width=0.3\textwidth]{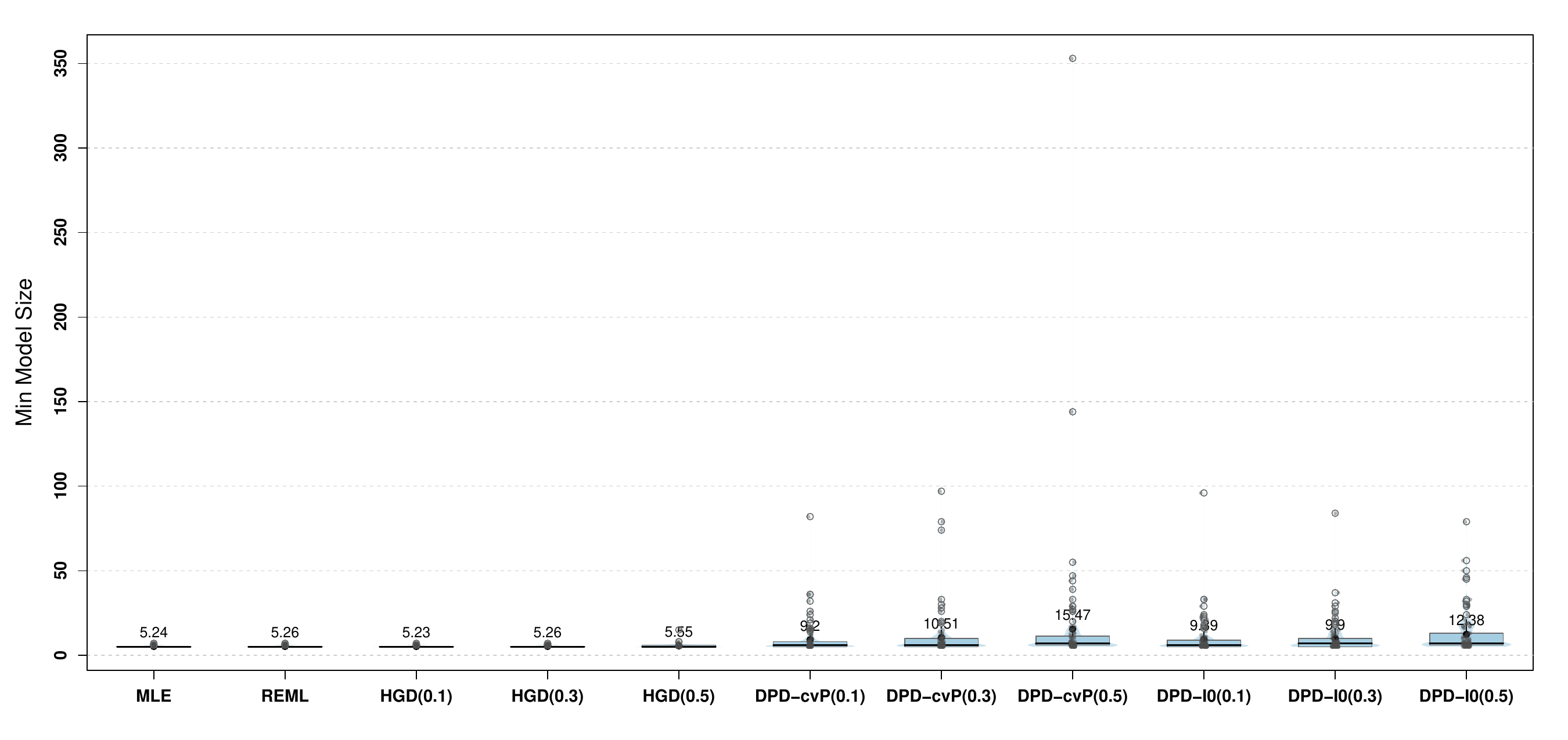}
		\label{FIG:boxplot_Y}}
	~	
	\subfloat[\tiny 10\% (C2) Contamination]{
		\includegraphics[page=2, width=0.3\textwidth]{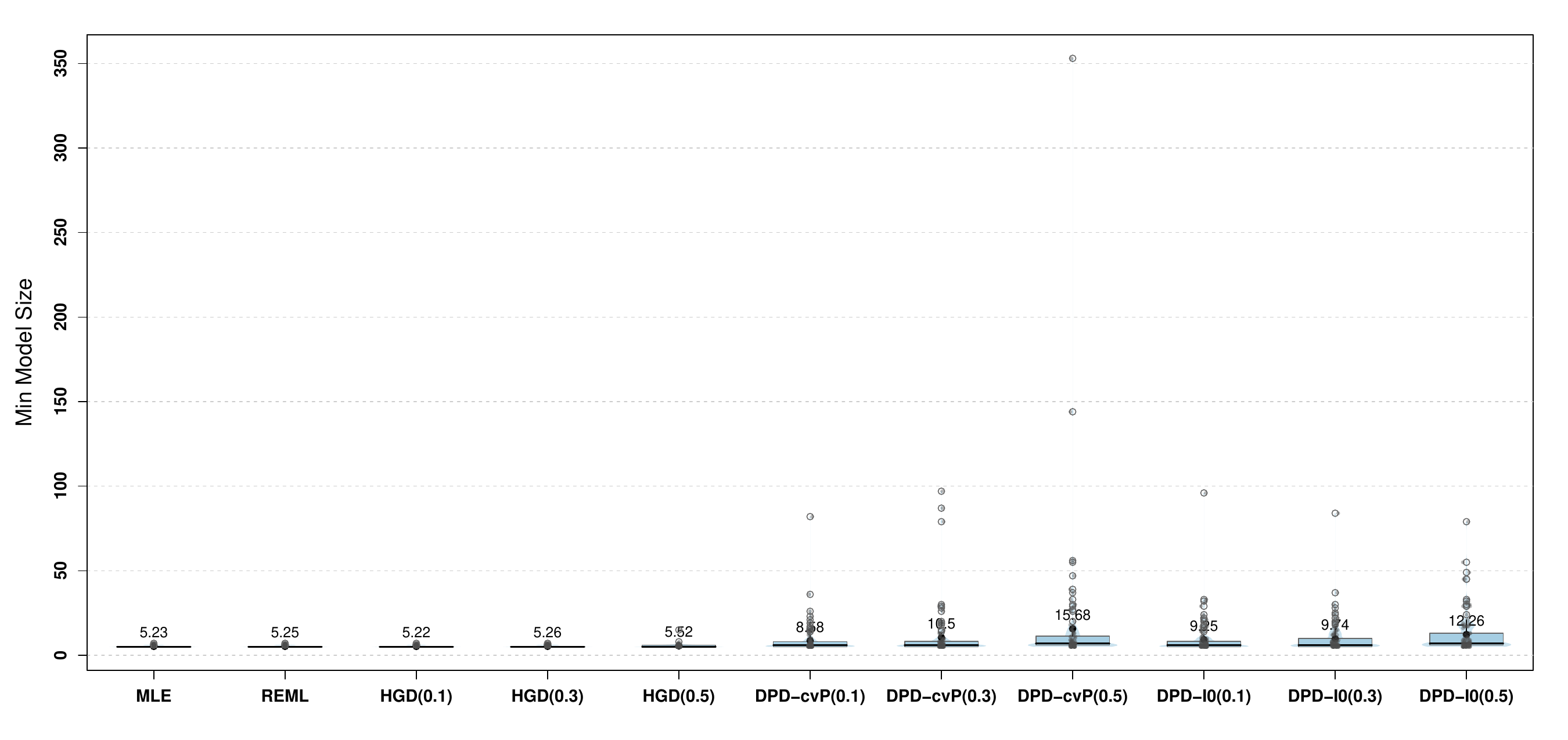}
		\label{FIG:boxplot_Y}}
	~	
	\subfloat[\tiny 20\% (C2) Contamination]{
		\includegraphics[page=2, width=0.3\textwidth]{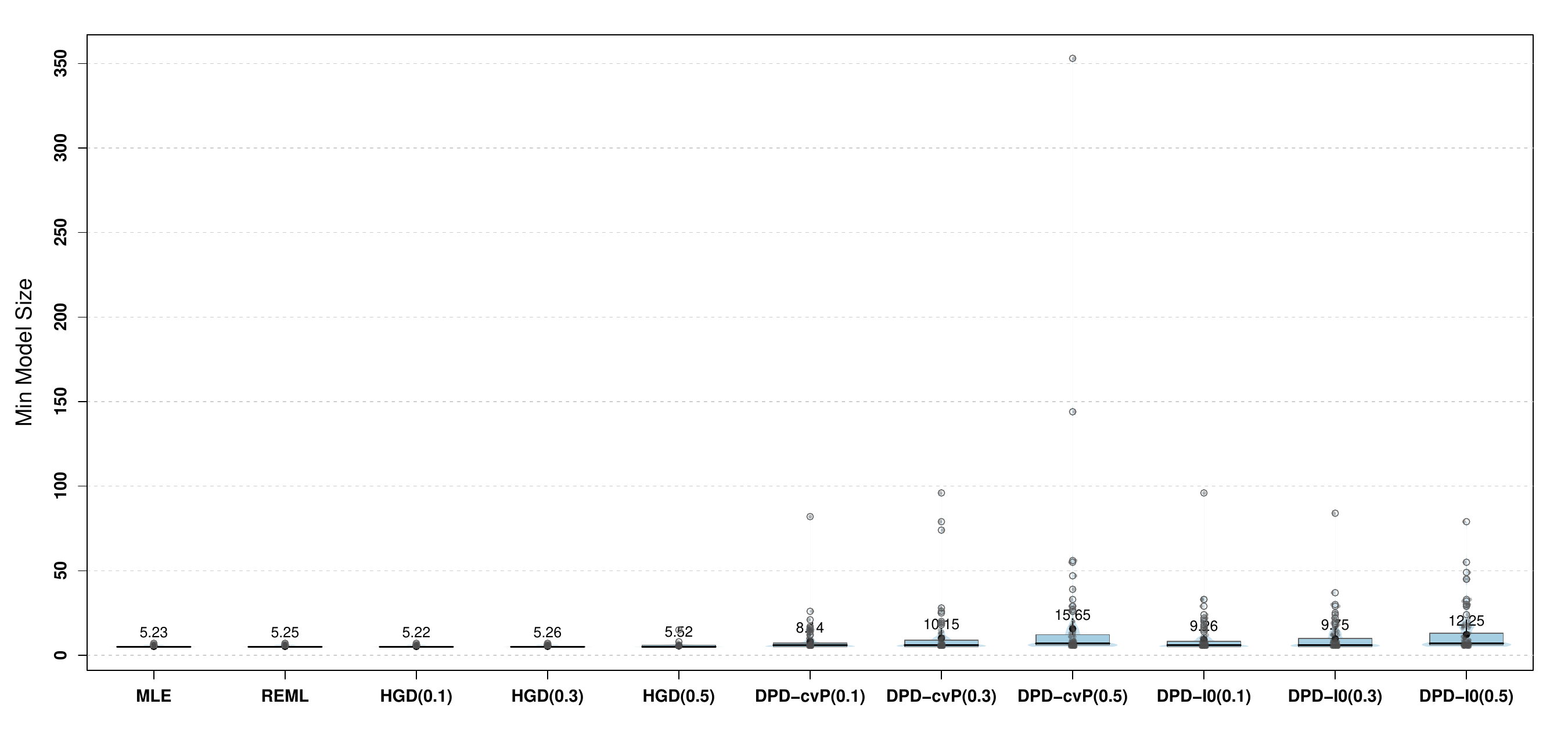}
		\label{FIG:boxplot_Y}}
	\\	
	\subfloat[\tiny 5\% (C3) Contamination]{
		\includegraphics[page=2, width=0.3\textwidth]{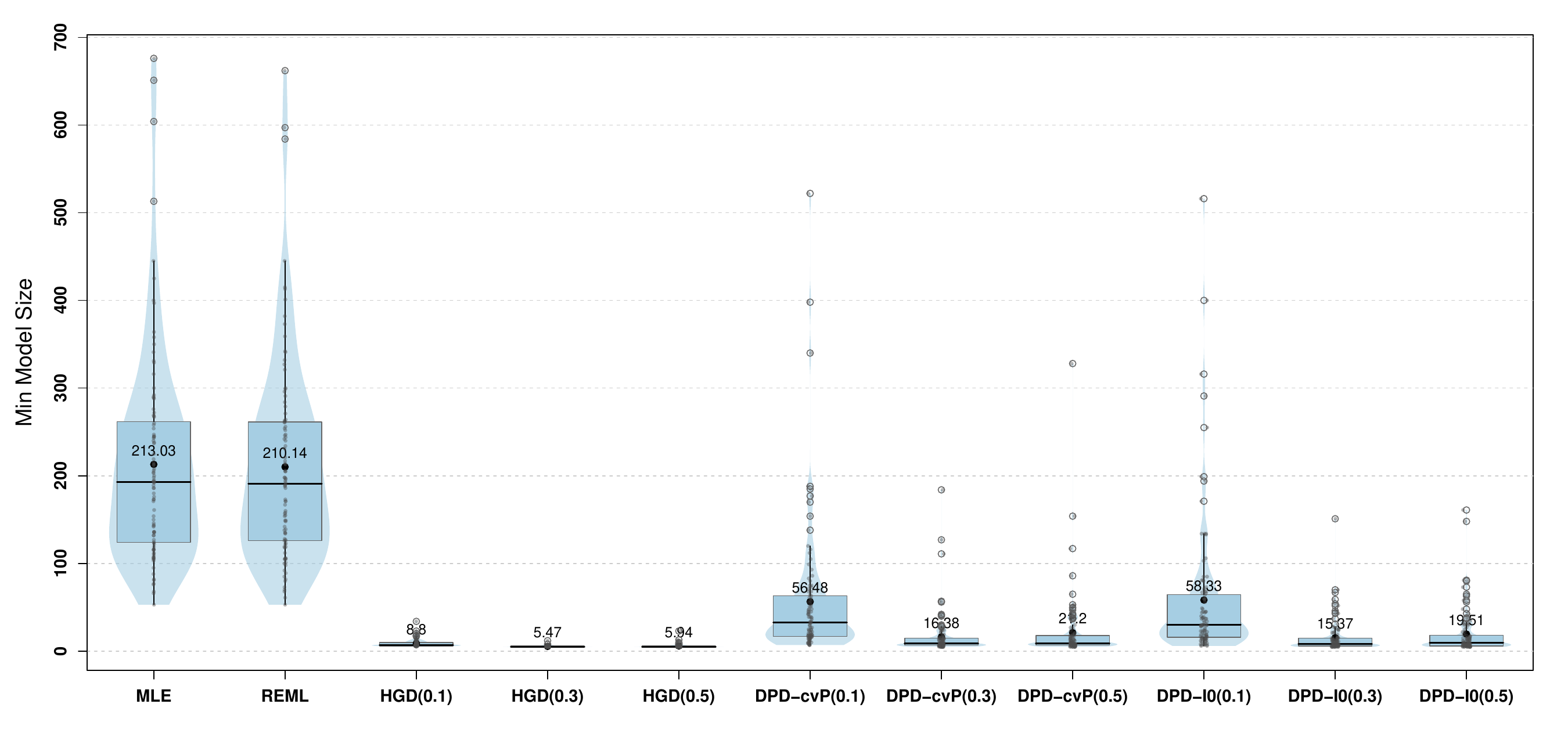}
		\label{FIG:boxplot_Y}}
	~	
	\subfloat[\tiny 10\% (C3) Contamination]{
		\includegraphics[page=2, width=0.3\textwidth]{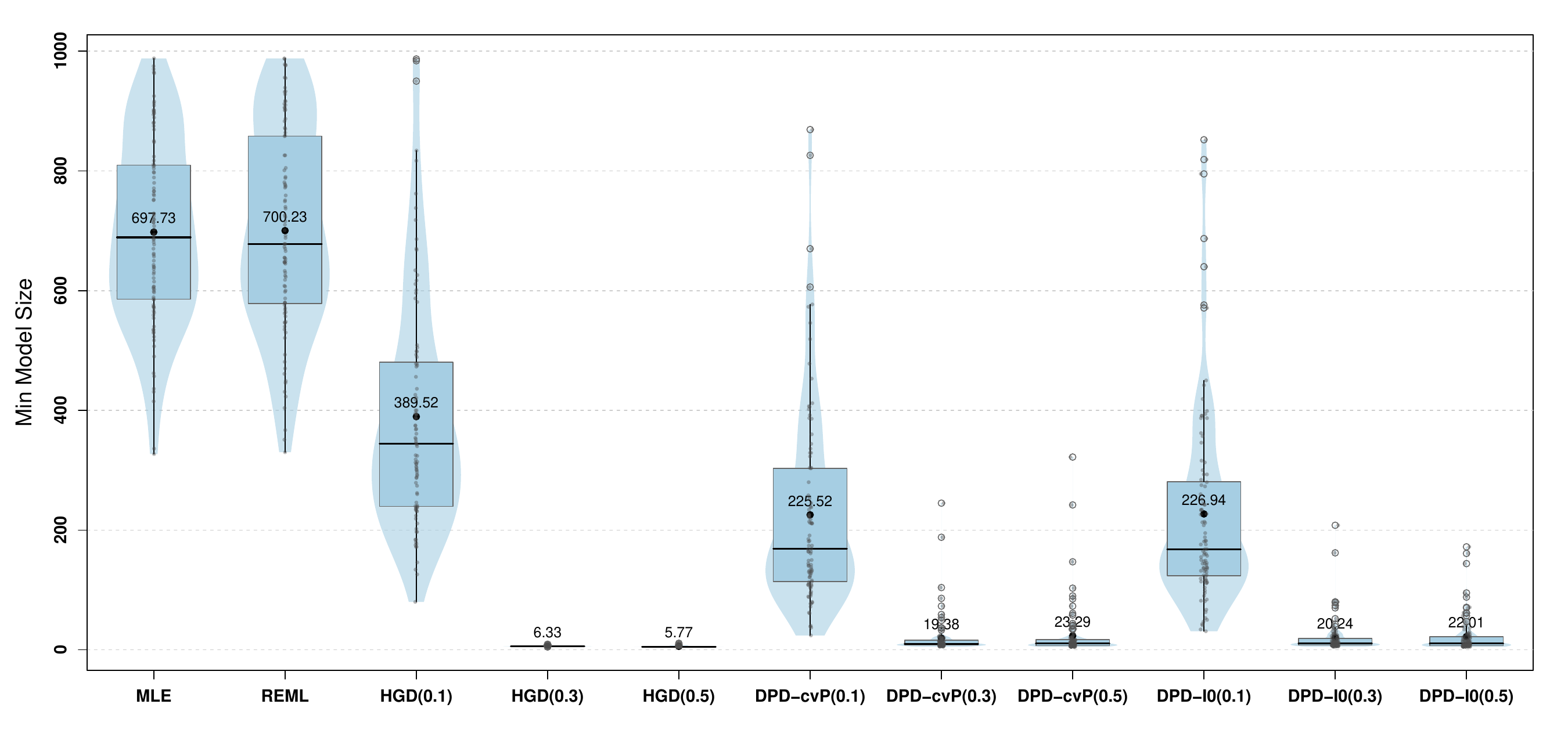}
		\label{FIG:boxplot_Y}}
	~	
	\subfloat[\tiny 20\% (C3) Contamination]{
		\includegraphics[page=2, width=0.3\textwidth]{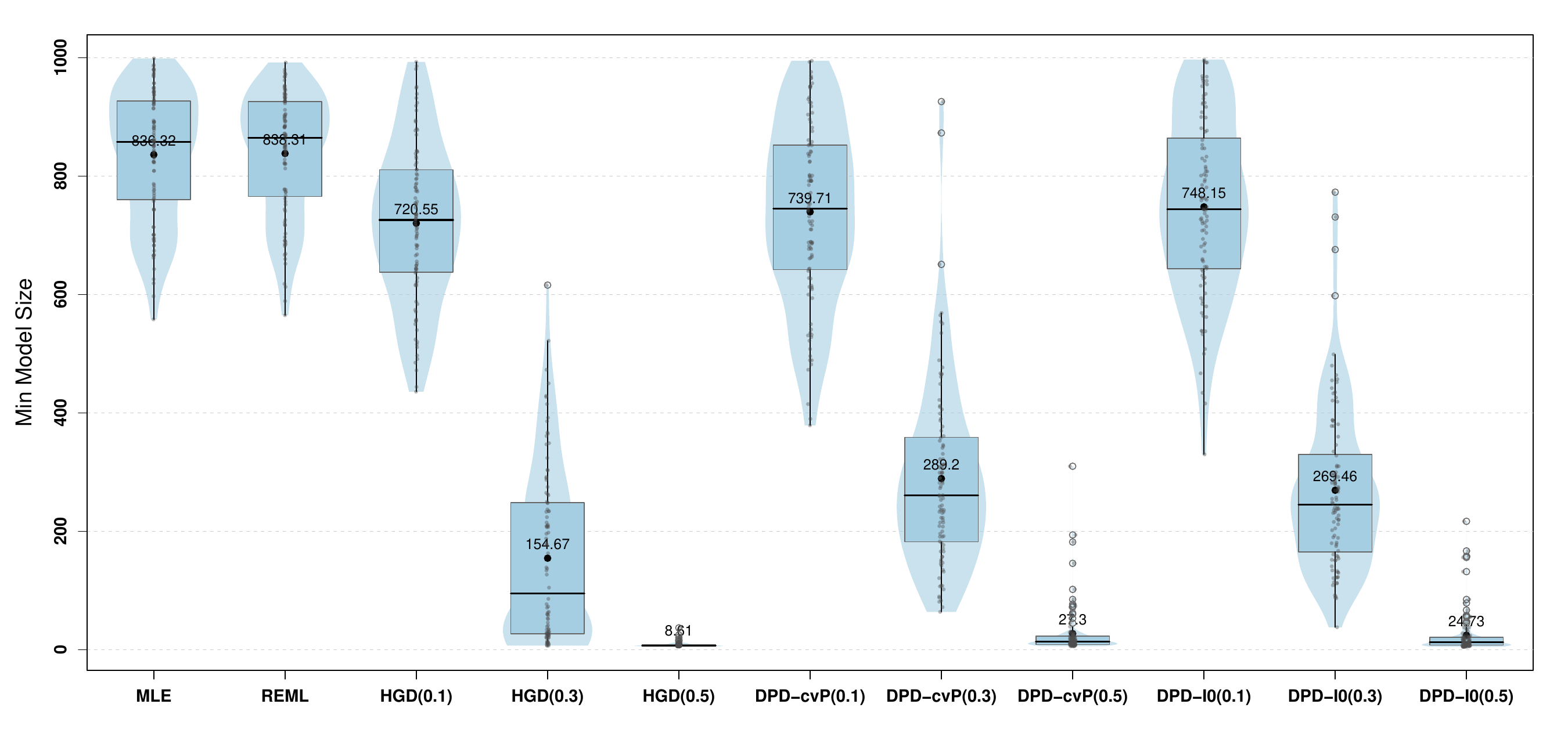}
		\label{FIG:boxplot_Y}}
	\\	
	\subfloat[\tiny 5\% (C4) Contamination]{
		\includegraphics[page=2, width=0.3\textwidth]{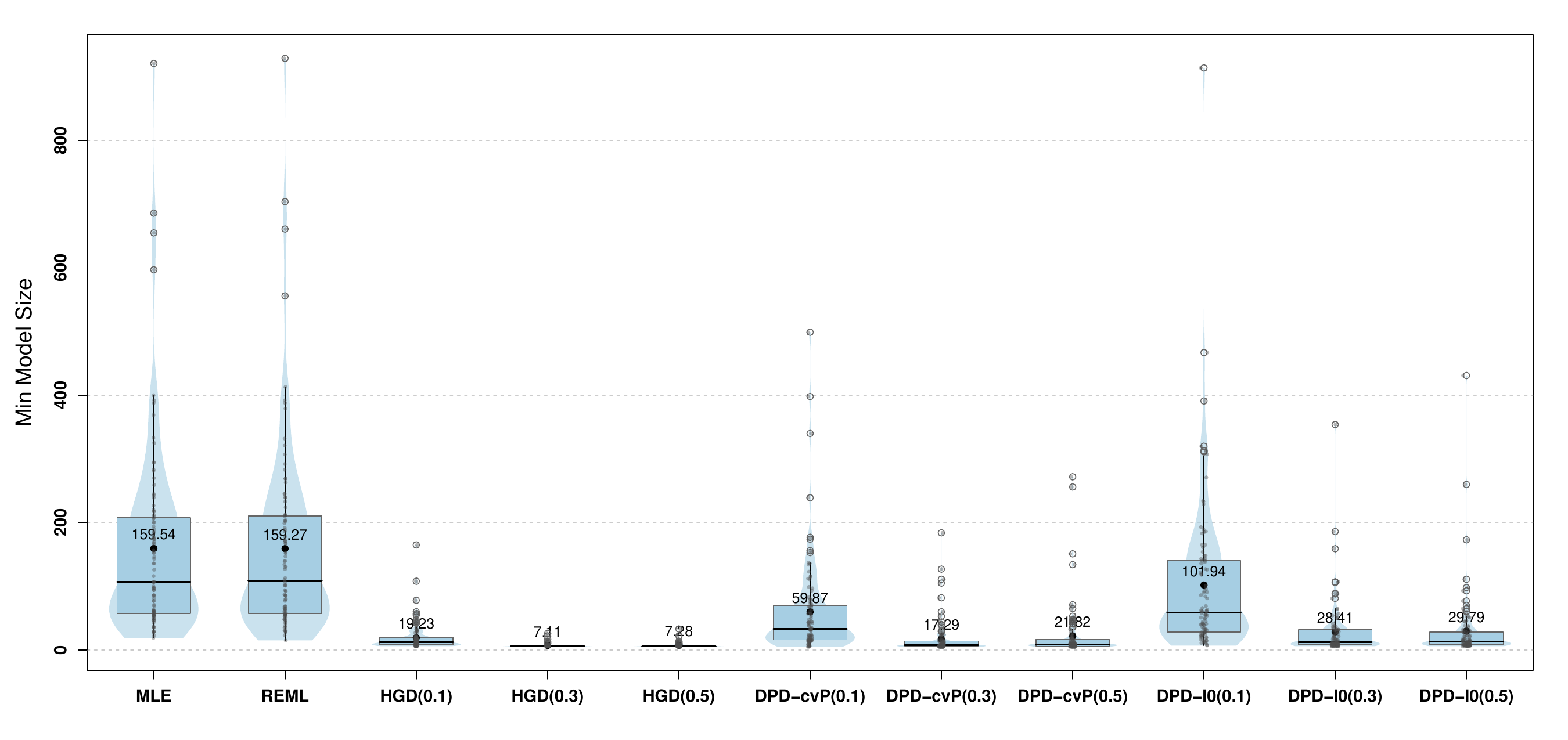}
		\label{FIG:boxplot_Y}}
	~	
	\subfloat[\tiny 10\% (C4) Contamination]{
		\includegraphics[page=2, width=0.3\textwidth]{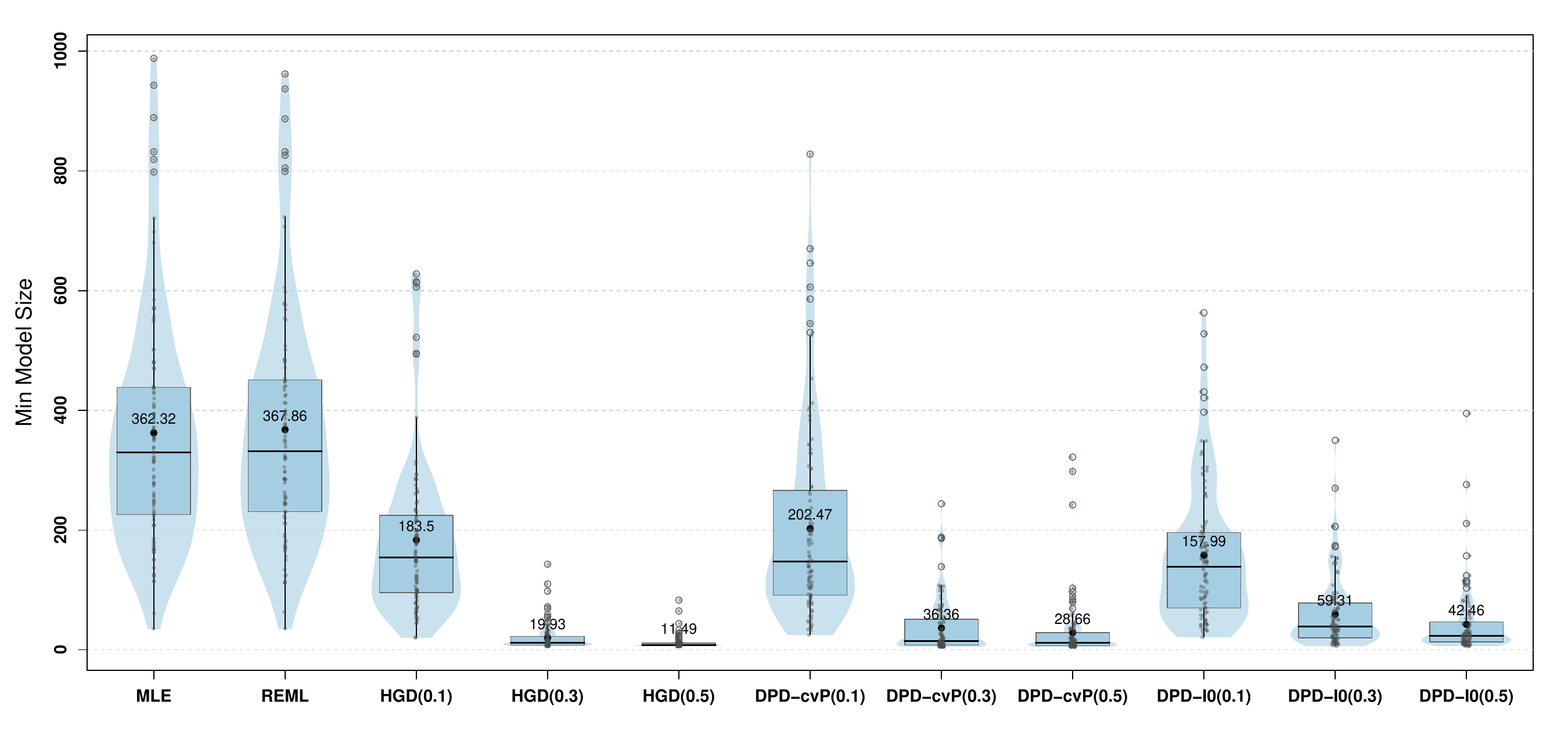}
		\label{FIG:boxplot_Y}}
	~	
	\subfloat[\tiny 20\% (C4) Contamination]{
		\includegraphics[page=2, width=0.3\textwidth]{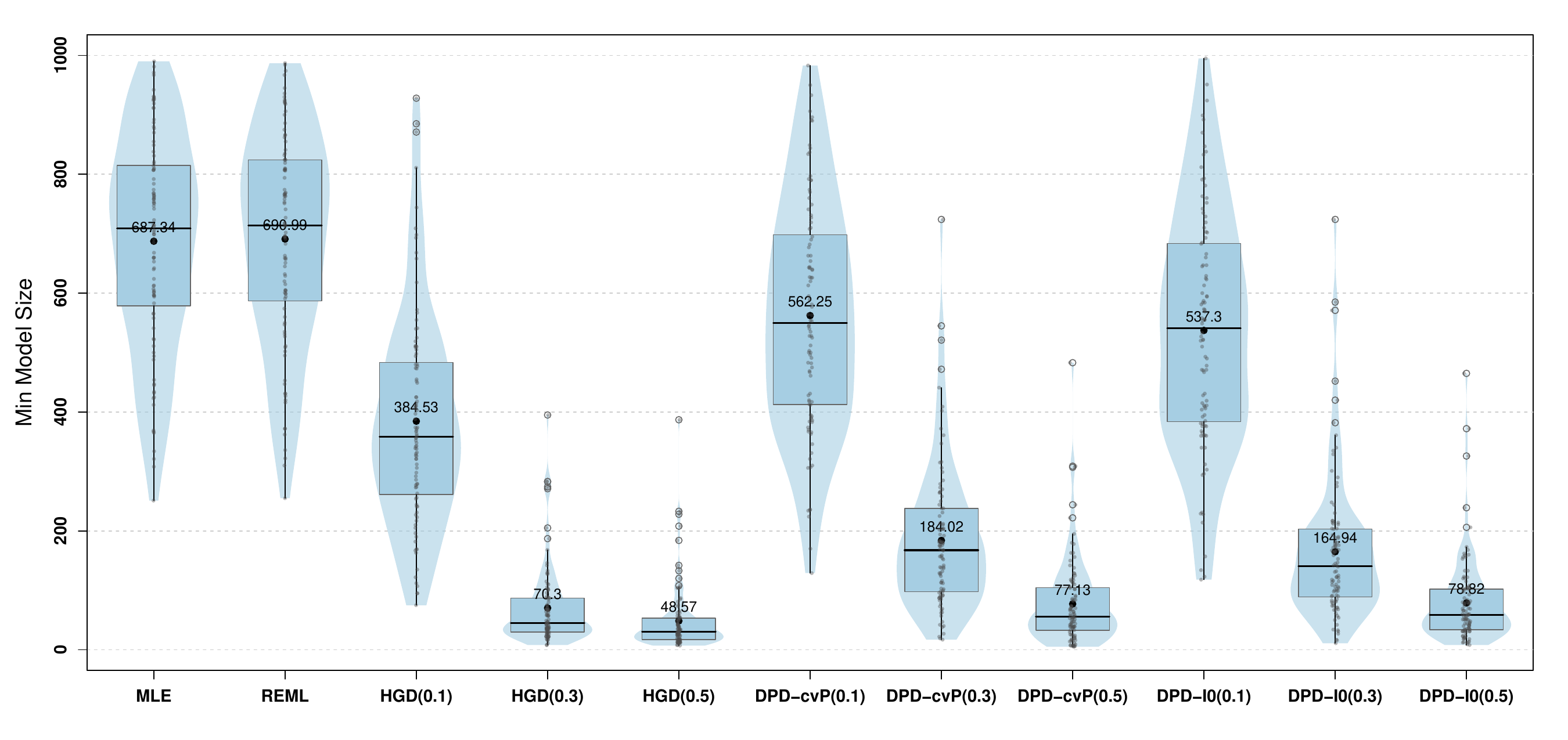}
		\label{FIG:boxplot_Y}}
	\caption{Boxplots  (with overlaid sample means) of MinMS required for sure screening under scenario (S2)$\times$(R1) with $\bm\Sigma_x$ as T(0.3) and different types of contamination}
	\label{FIG:MMS_S2R1T}
\end{figure}


\begin{figure}[!h]
	\centering
	\subfloat[\tiny 5\% (C1) Contamination]{
		\includegraphics[page=3, width=0.3\textwidth]{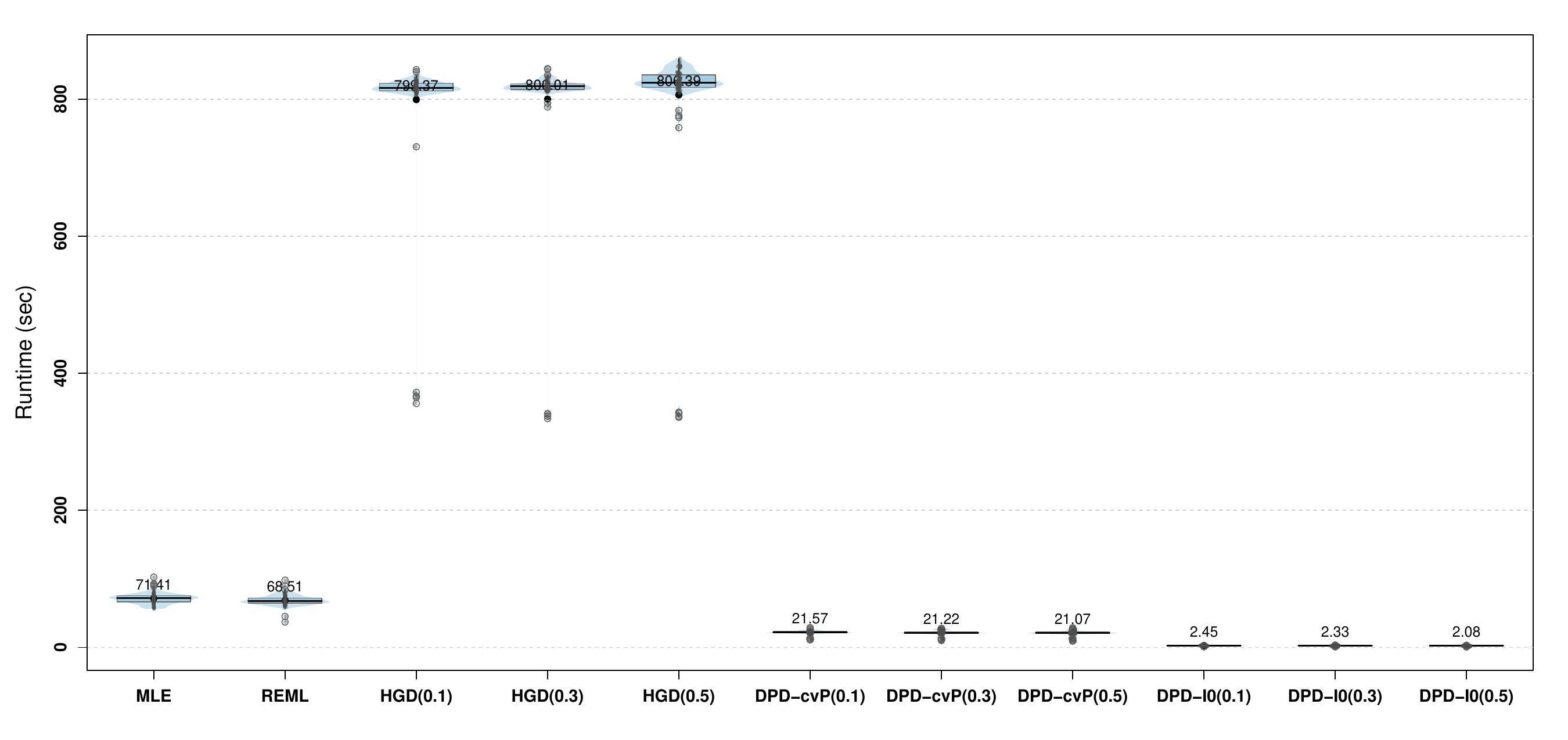}
		\label{FIG:boxplot_Y}}
	~	
	\subfloat[\tiny 10\% (C1) Contamination]{
		\includegraphics[page=3, width=0.3\textwidth]{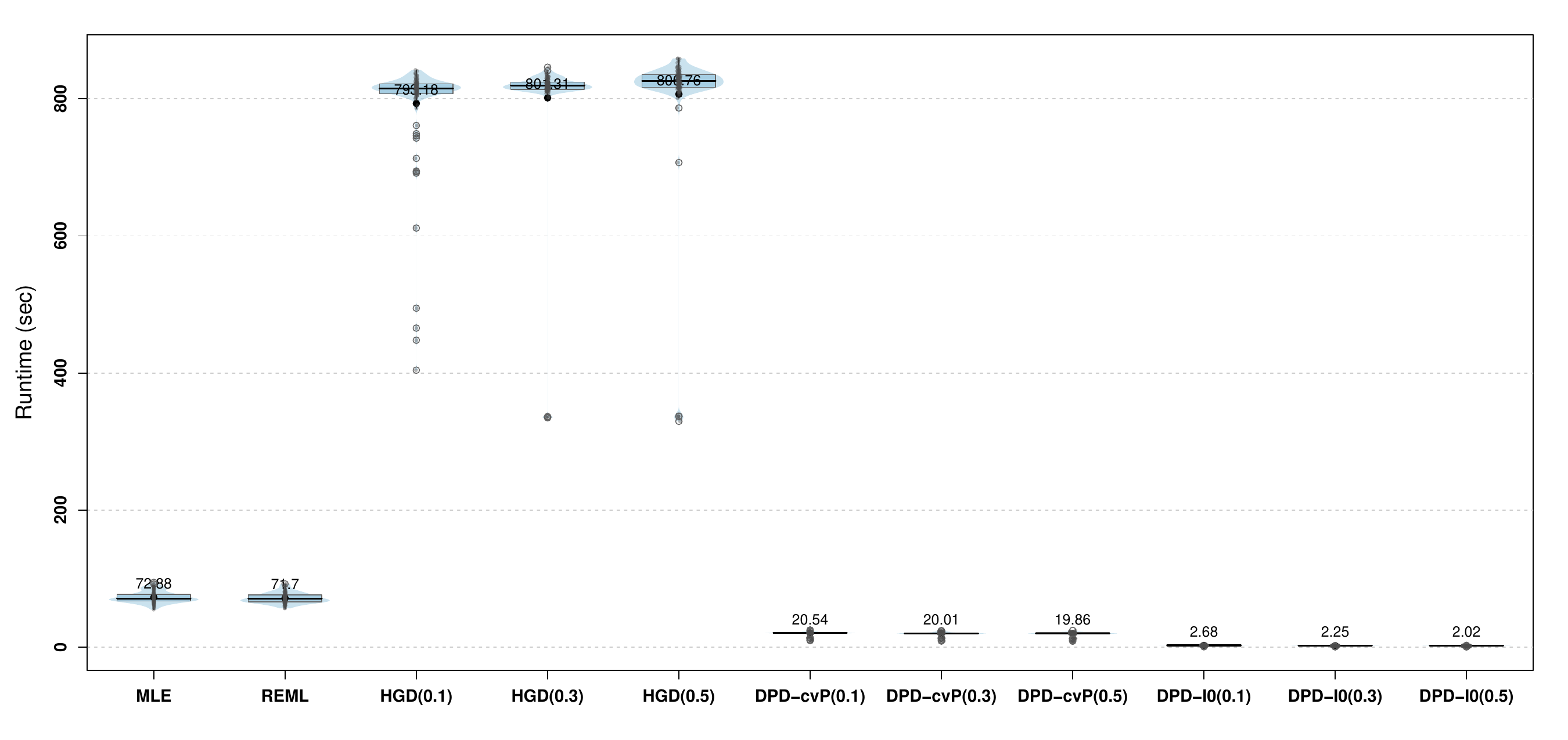}
		\label{FIG:boxplot_Y}}
	~	
	\subfloat[\tiny 20\% (C1) Contamination]{
		\includegraphics[page=3, width=0.3\textwidth]{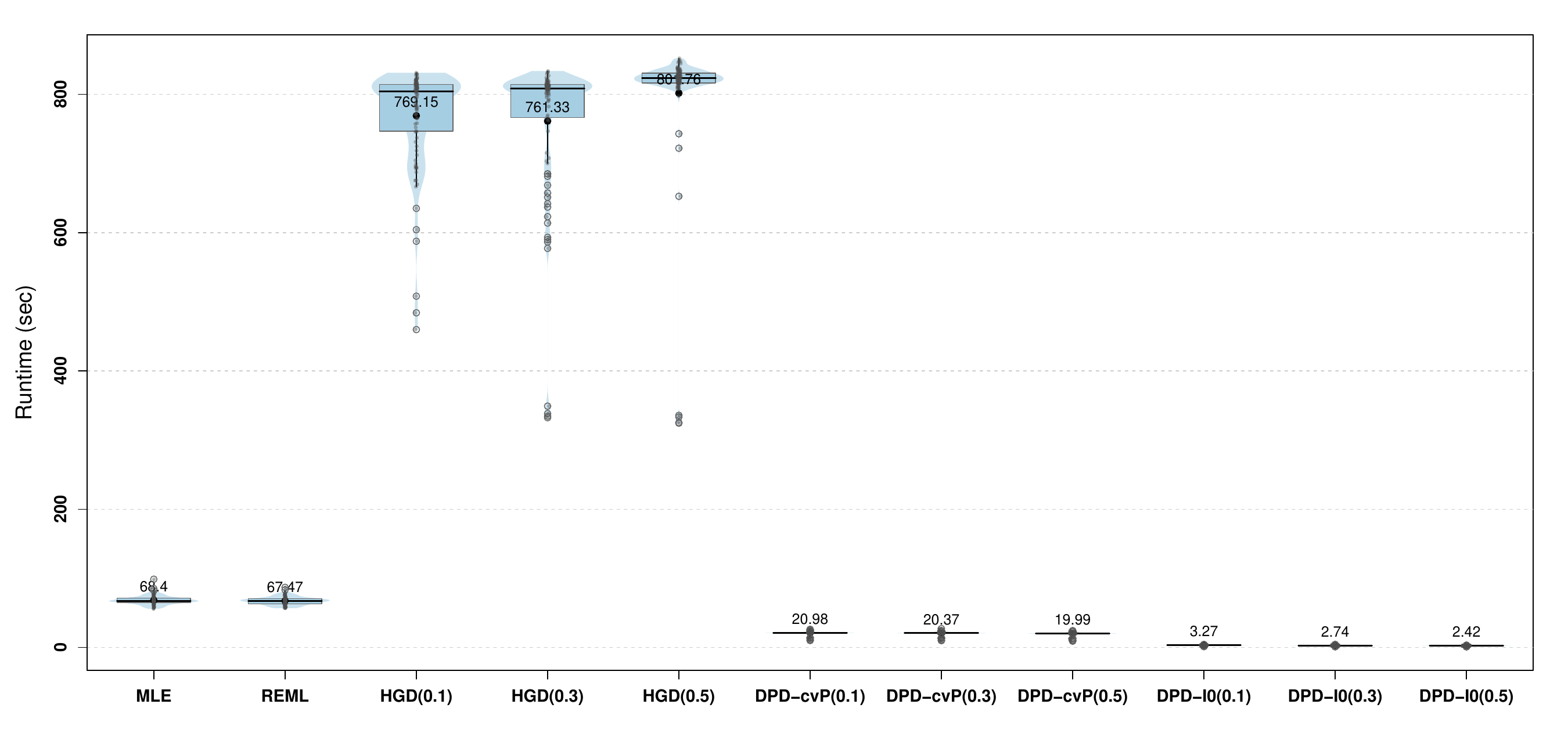}
		\label{FIG:boxplot_Y}}
	\\	
	\subfloat[\tiny 5\% (C2) Contamination]{
		\includegraphics[page=3, width=0.3\textwidth]{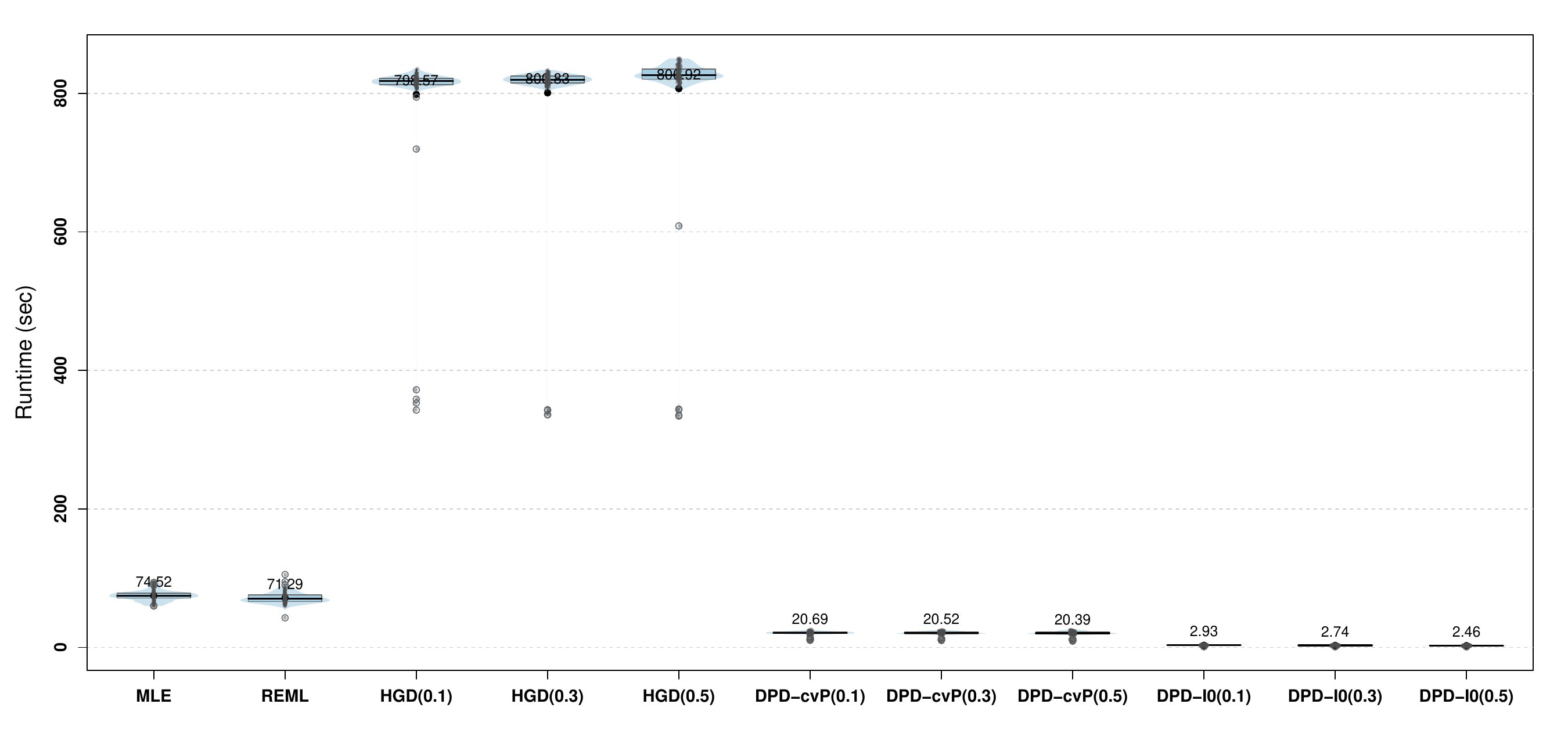}
		\label{FIG:boxplot_Y}}
	~	
	\subfloat[\tiny 10\% (C2) Contamination]{
		\includegraphics[page=3, width=0.3\textwidth]{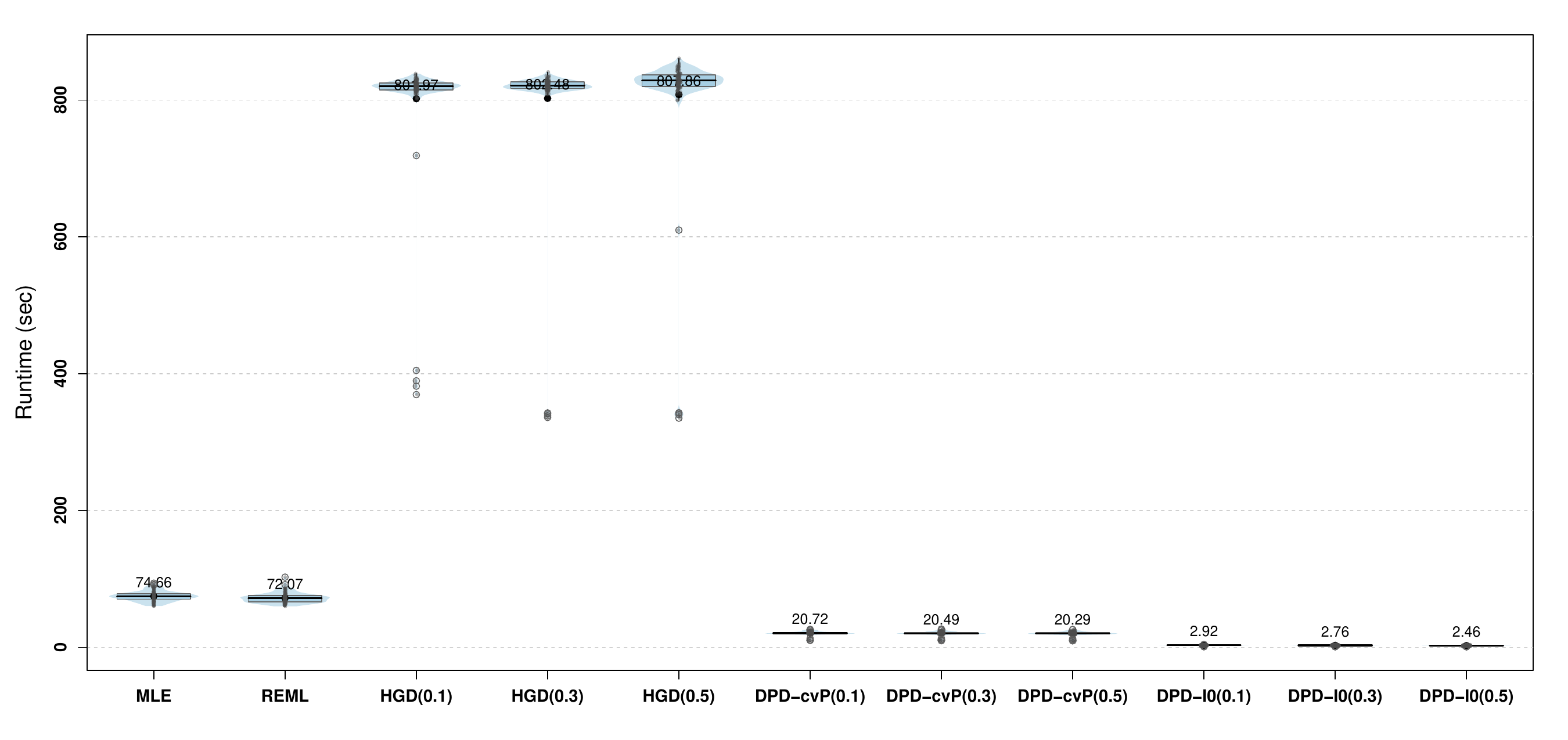}
		\label{FIG:boxplot_Y}}
	~	
	\subfloat[\tiny 20\% (C2) Contamination]{
		\includegraphics[page=3, width=0.3\textwidth]{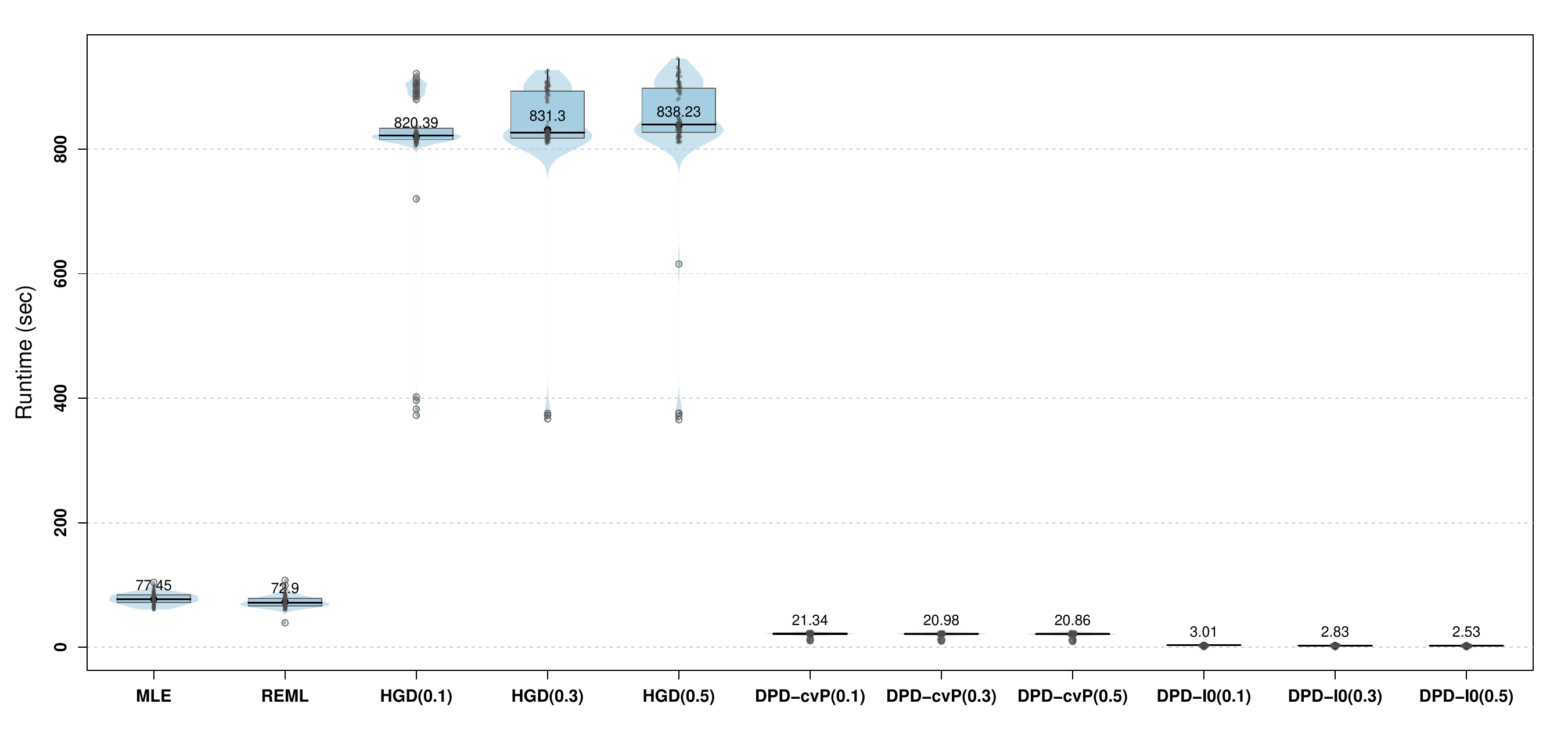}
		\label{FIG:boxplot_Y}}
	\\	
	\subfloat[\tiny 5\% (C3) Contamination]{
		\includegraphics[page=3, width=0.3\textwidth]{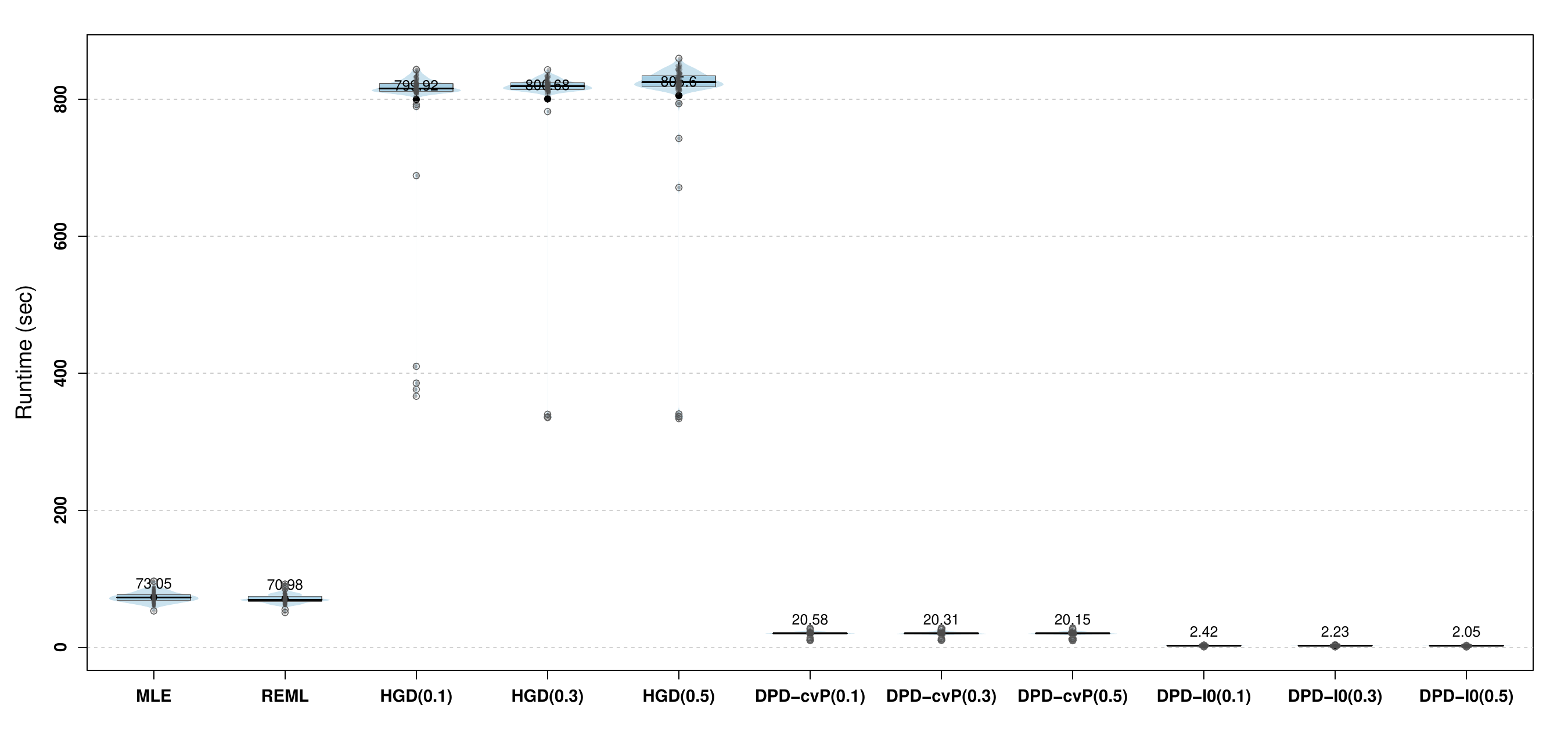}
		\label{FIG:boxplot_Y}}
	~	
	\subfloat[\tiny 10\% (C3) Contamination]{
		\includegraphics[page=3, width=0.3\textwidth]{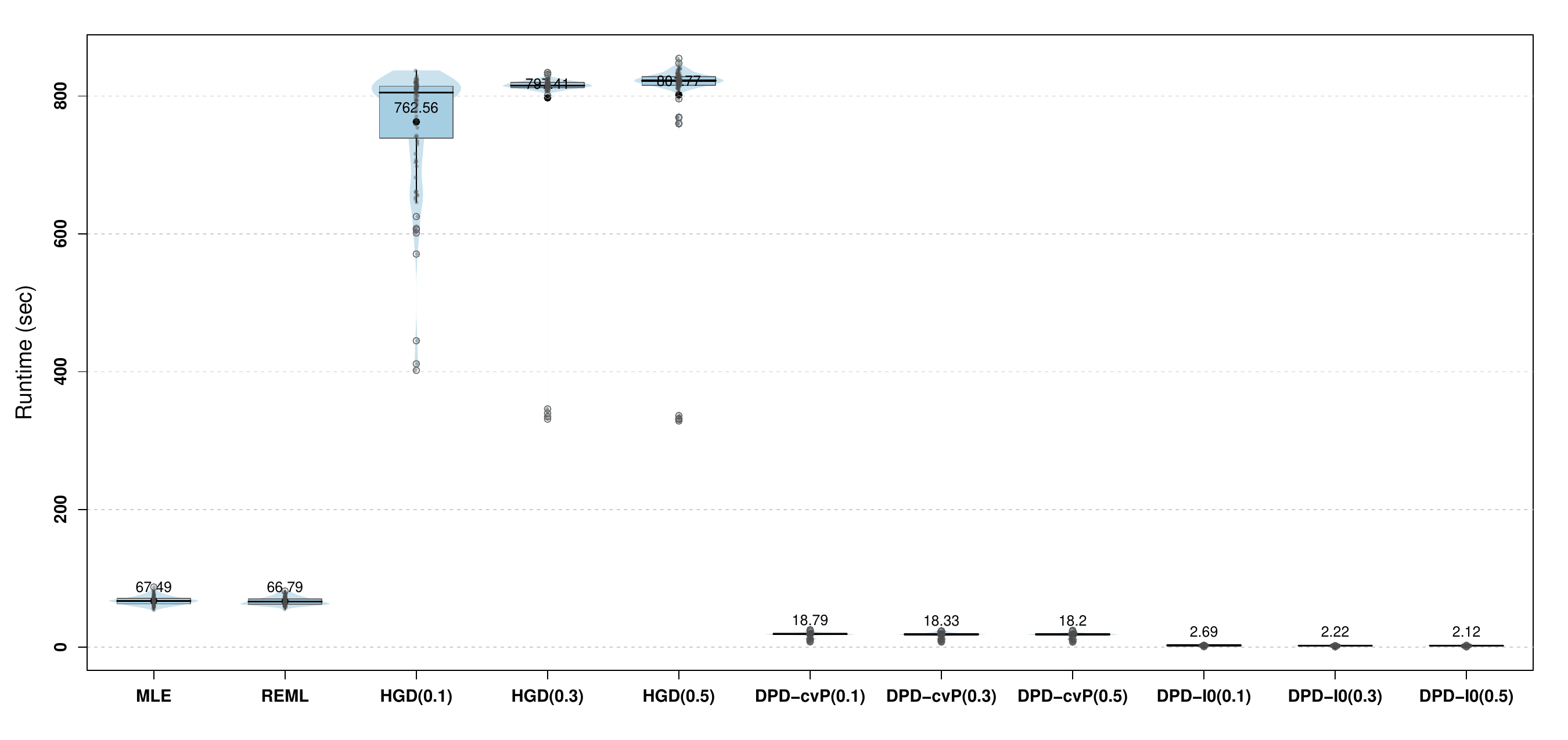}
		\label{FIG:boxplot_Y}}
	~	
	\subfloat[\tiny 20\% (C3) Contamination]{
		\includegraphics[page=3, width=0.3\textwidth]{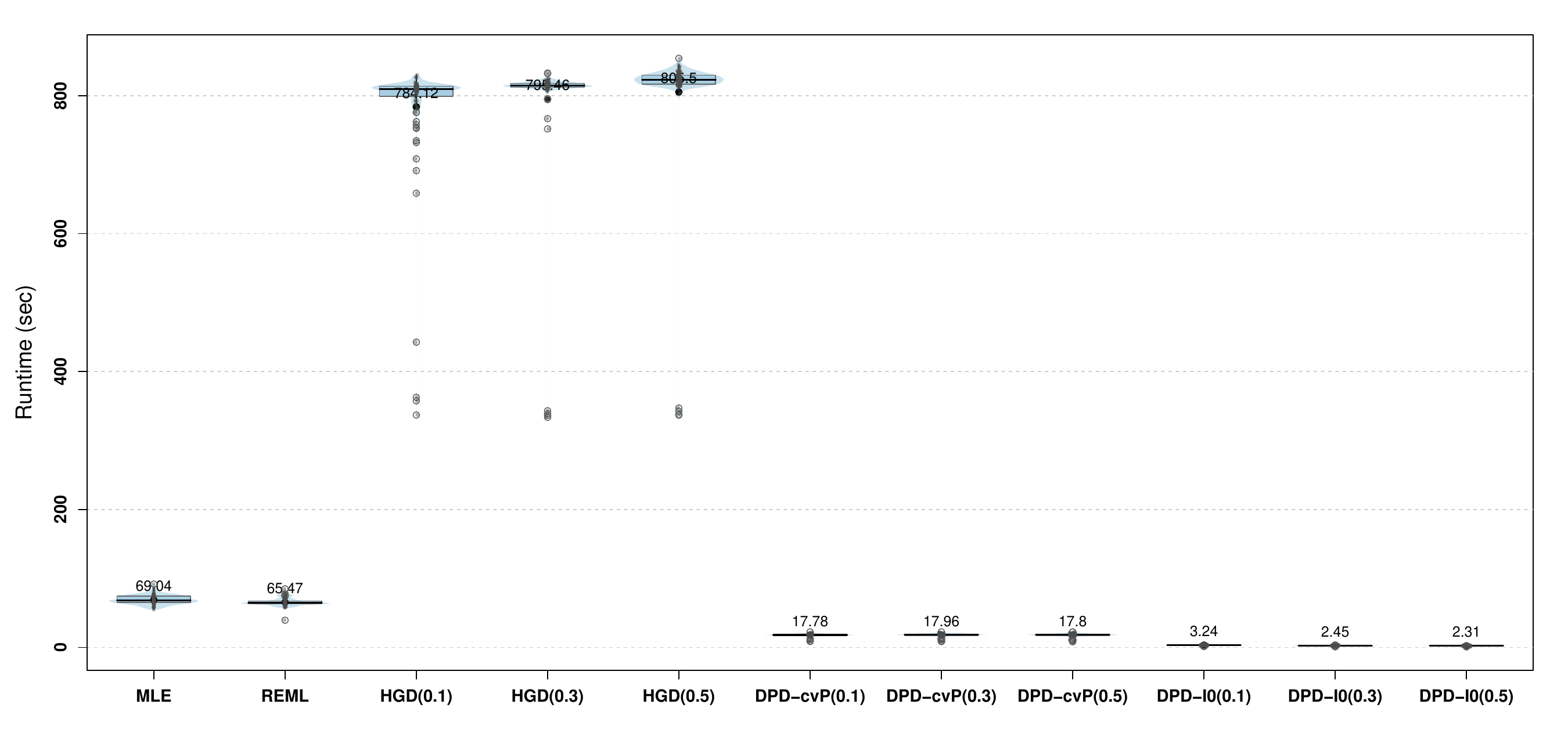}
		\label{FIG:boxplot_Y}}
	\\	
	\subfloat[\tiny 5\% (C4) Contamination]{
		\includegraphics[page=3, width=0.3\textwidth]{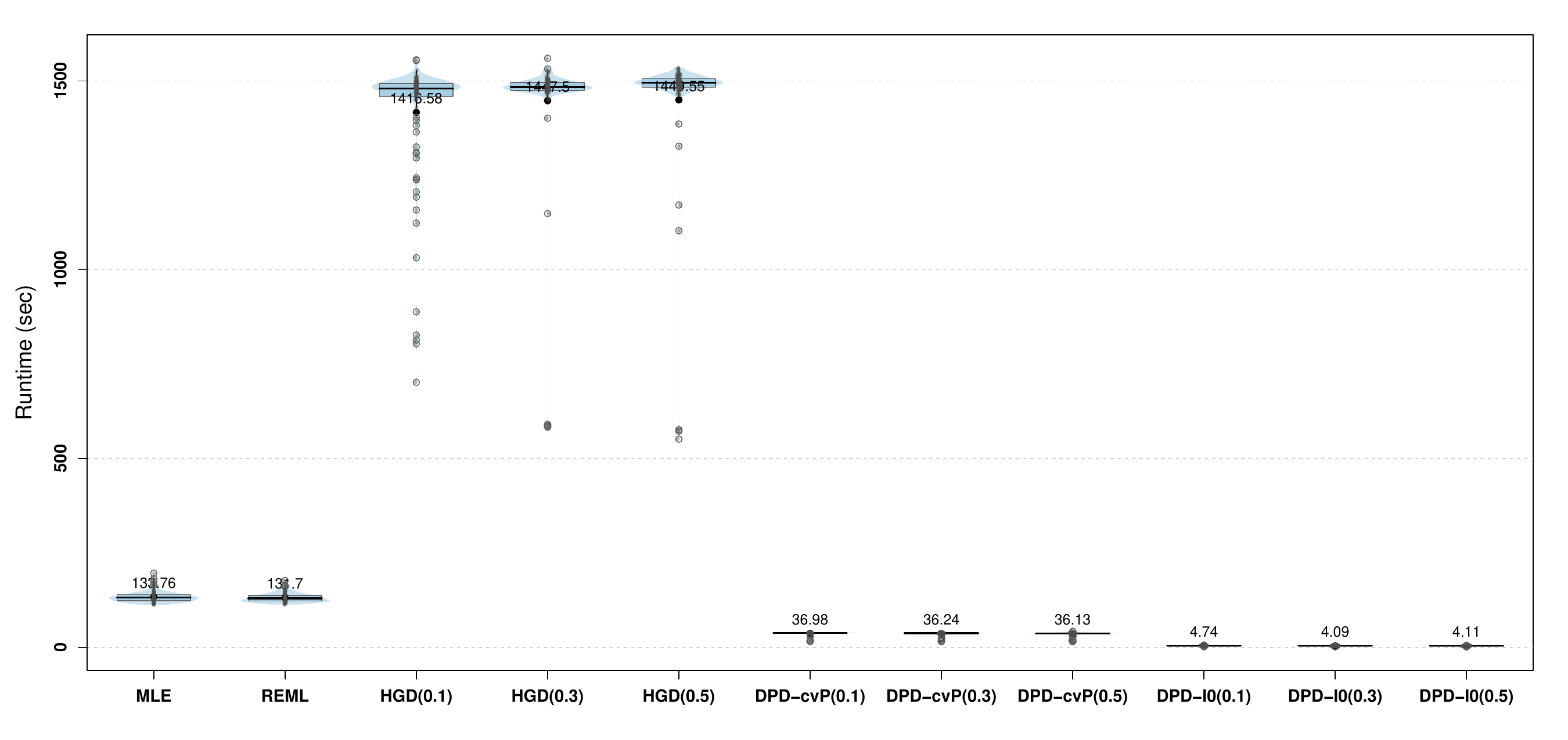}
		\label{FIG:boxplot_Y}}
	~	
	\subfloat[\tiny 10\% (C4) Contamination]{
		\includegraphics[page=3, width=0.3\textwidth]{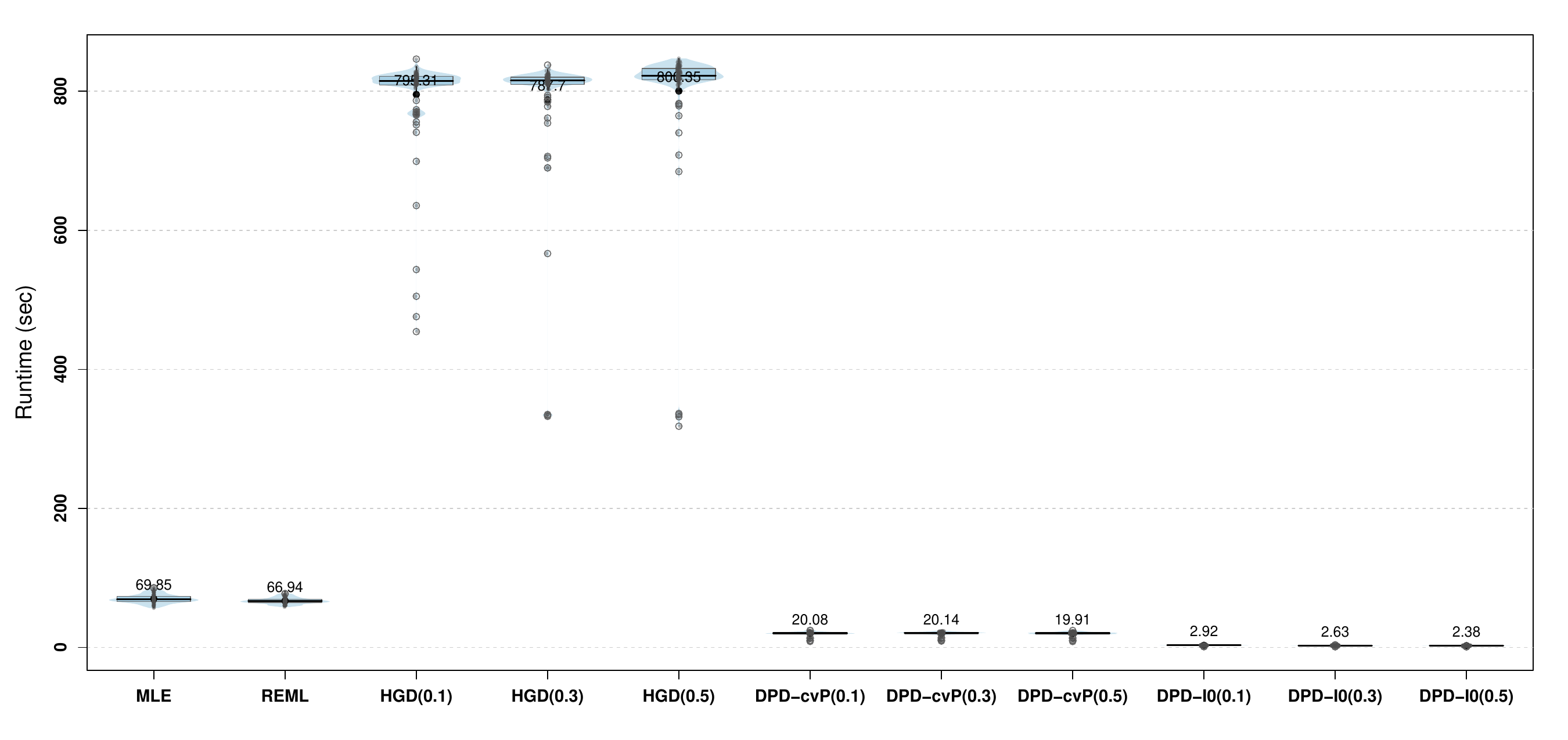}
		\label{FIG:boxplot_Y}}
	~	
	\subfloat[\tiny 20\% (C4) Contamination]{
		\includegraphics[page=3, width=0.3\textwidth]{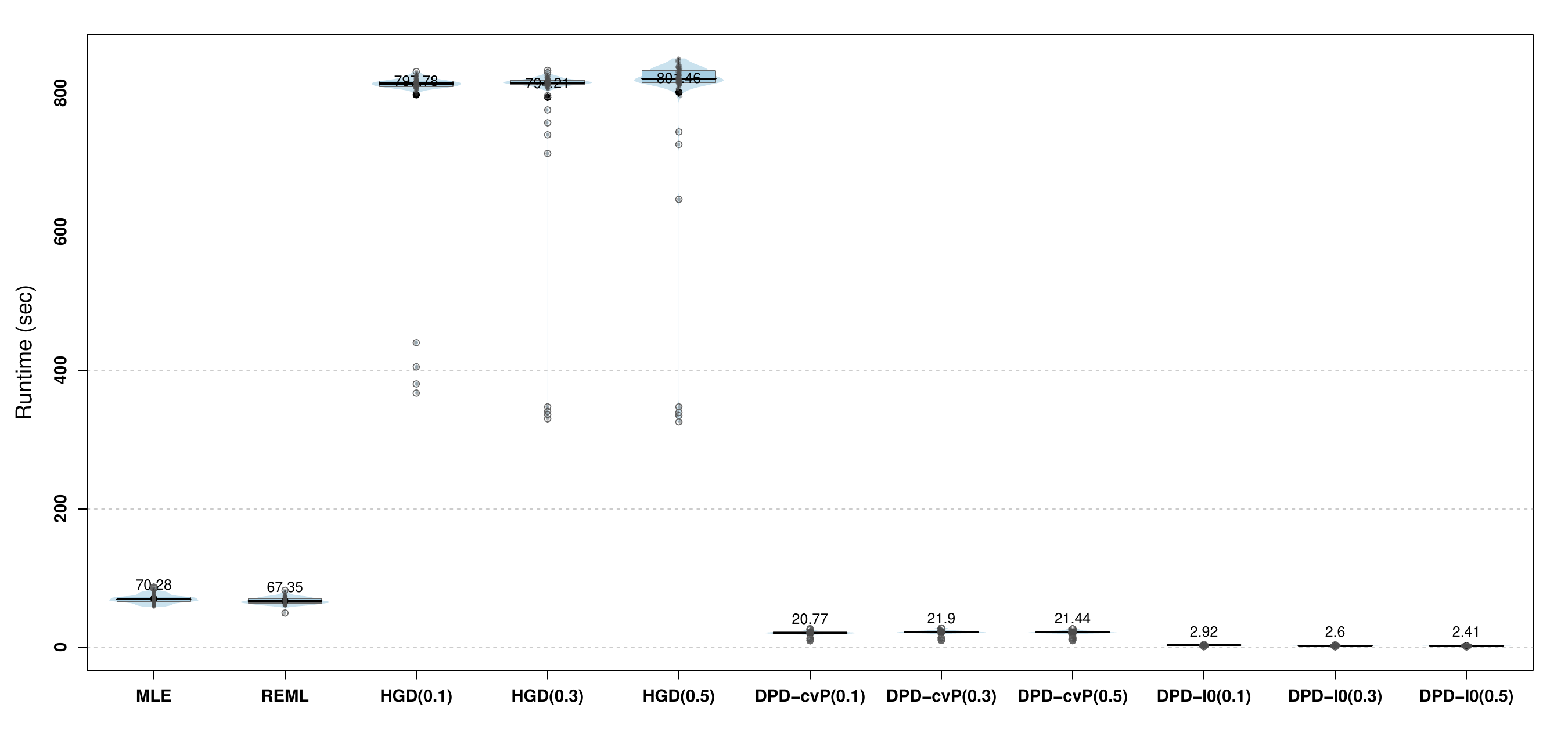}
		\label{FIG:boxplot_Y}}
	\caption{Boxplots  (with overlaid sample means) of computational times (in seconds) of the screening methods  under scenario (S1)$\times$(R1) with $\bm\Sigma_x=\mathbb{I}$ and different types of contamination}
	\label{FIG:time_Csp1}
\end{figure}

\begin{figure}[!h]
	\centering
	\subfloat[\tiny 5\% (C1) Contamination]{
		\includegraphics[page=3, width=0.3\textwidth]{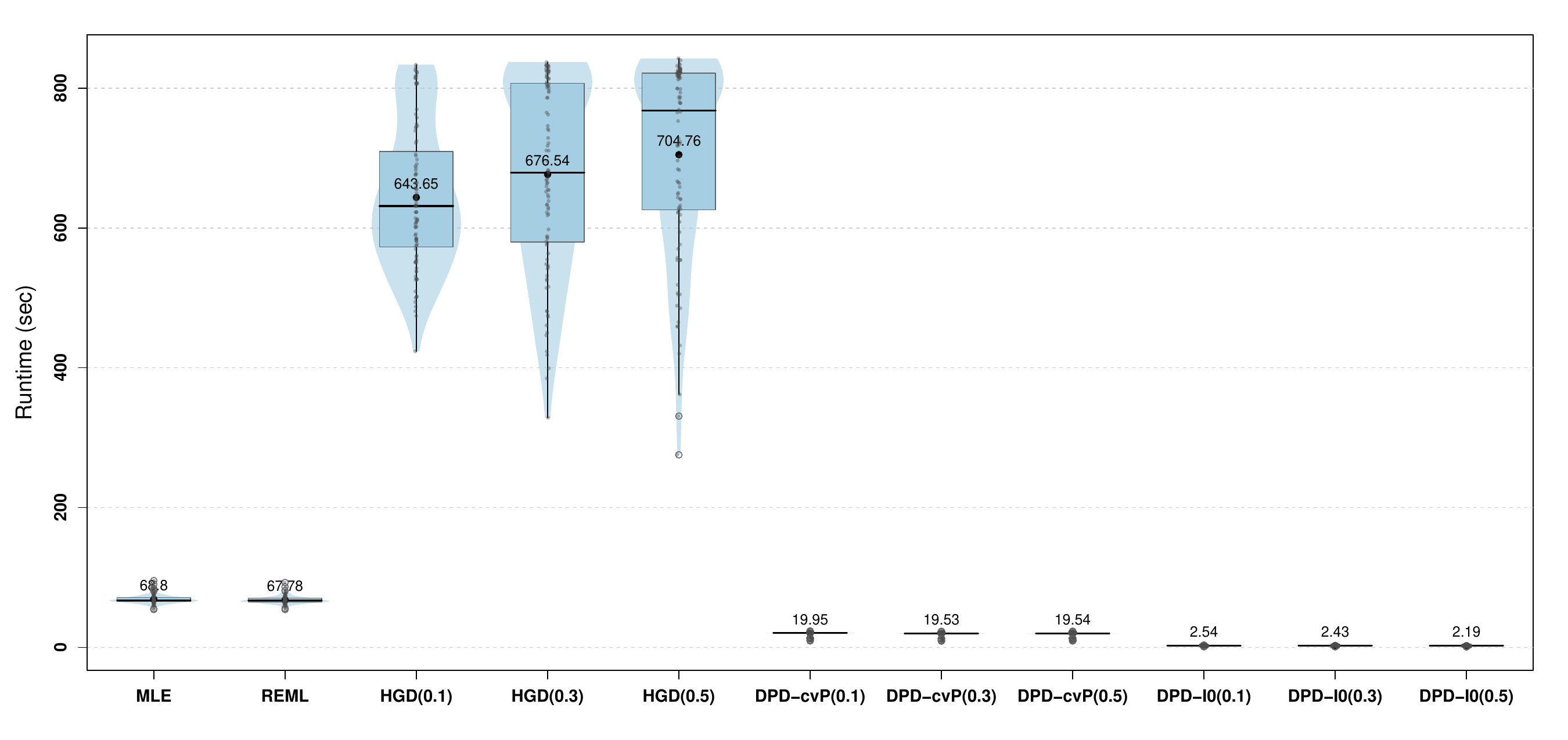}
		\label{FIG:boxplot_Y}}
	~	
	\subfloat[\tiny 10\% (C1) Contamination]{
		\includegraphics[page=3, width=0.3\textwidth]{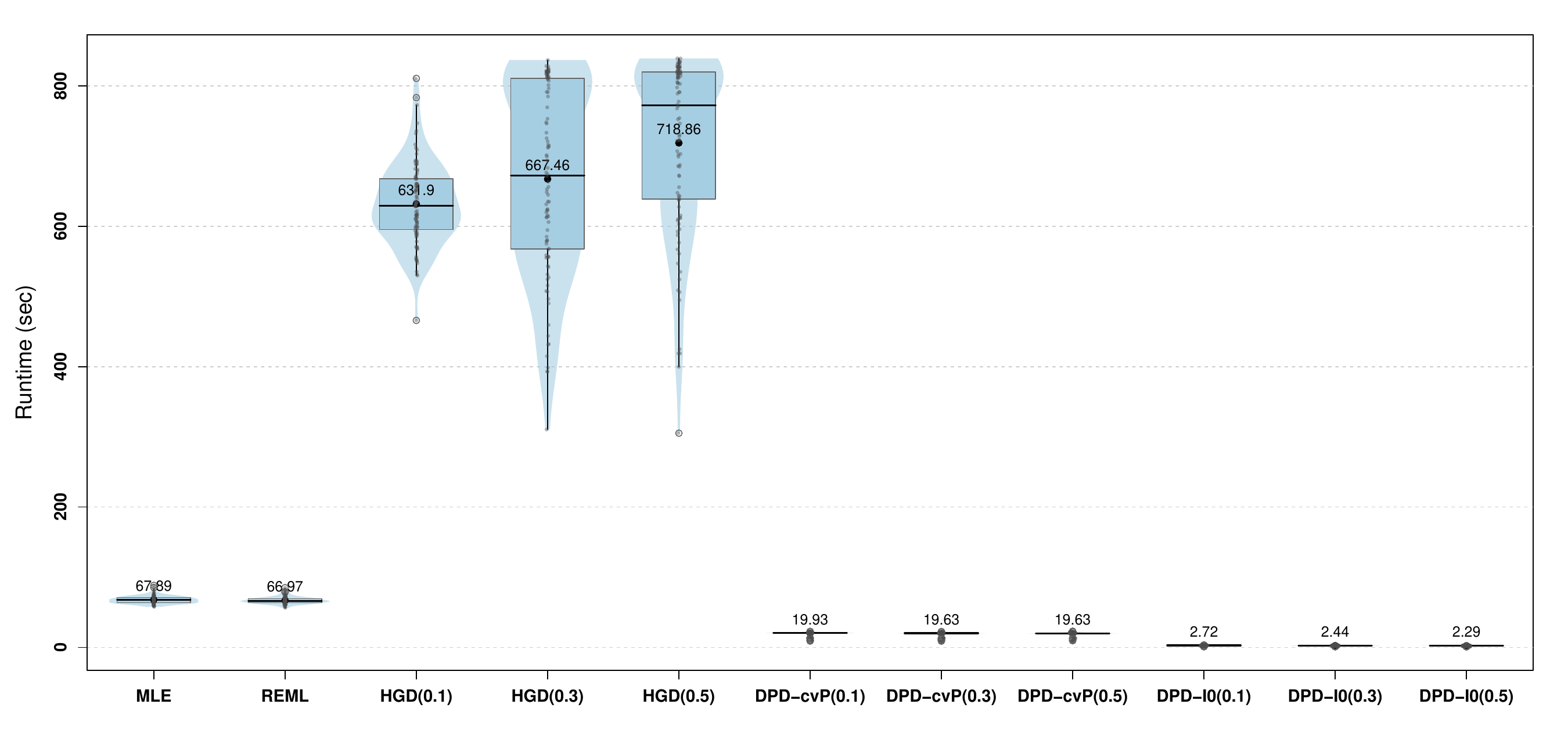}
		\label{FIG:boxplot_Y}}
	~	
	\subfloat[\tiny 20\% (C1) Contamination]{
		\includegraphics[page=3, width=0.3\textwidth]{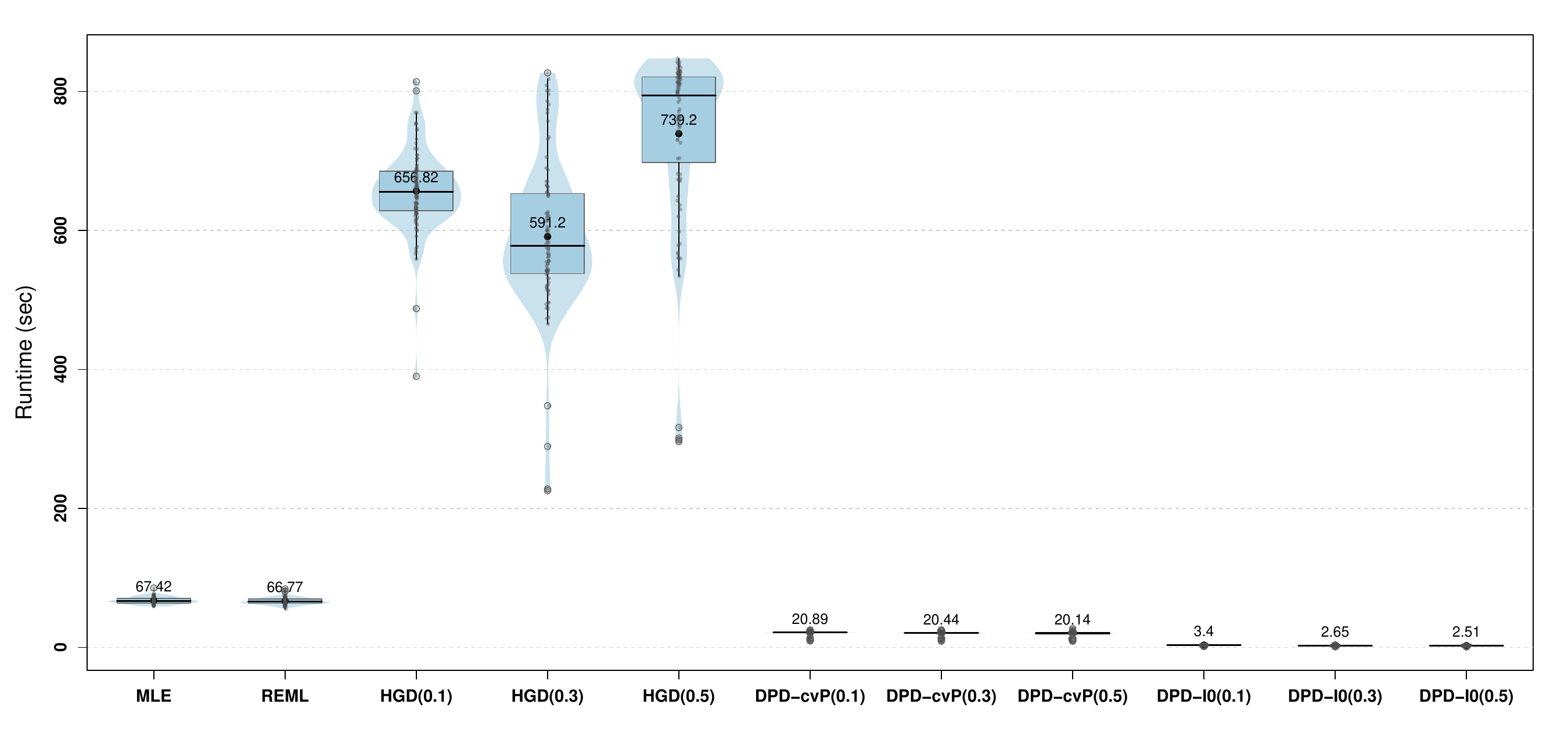}
		\label{FIG:boxplot_Y}}
	\\	
	\subfloat[\tiny 5\% (C2) Contamination]{
		\includegraphics[page=3, width=0.3\textwidth]{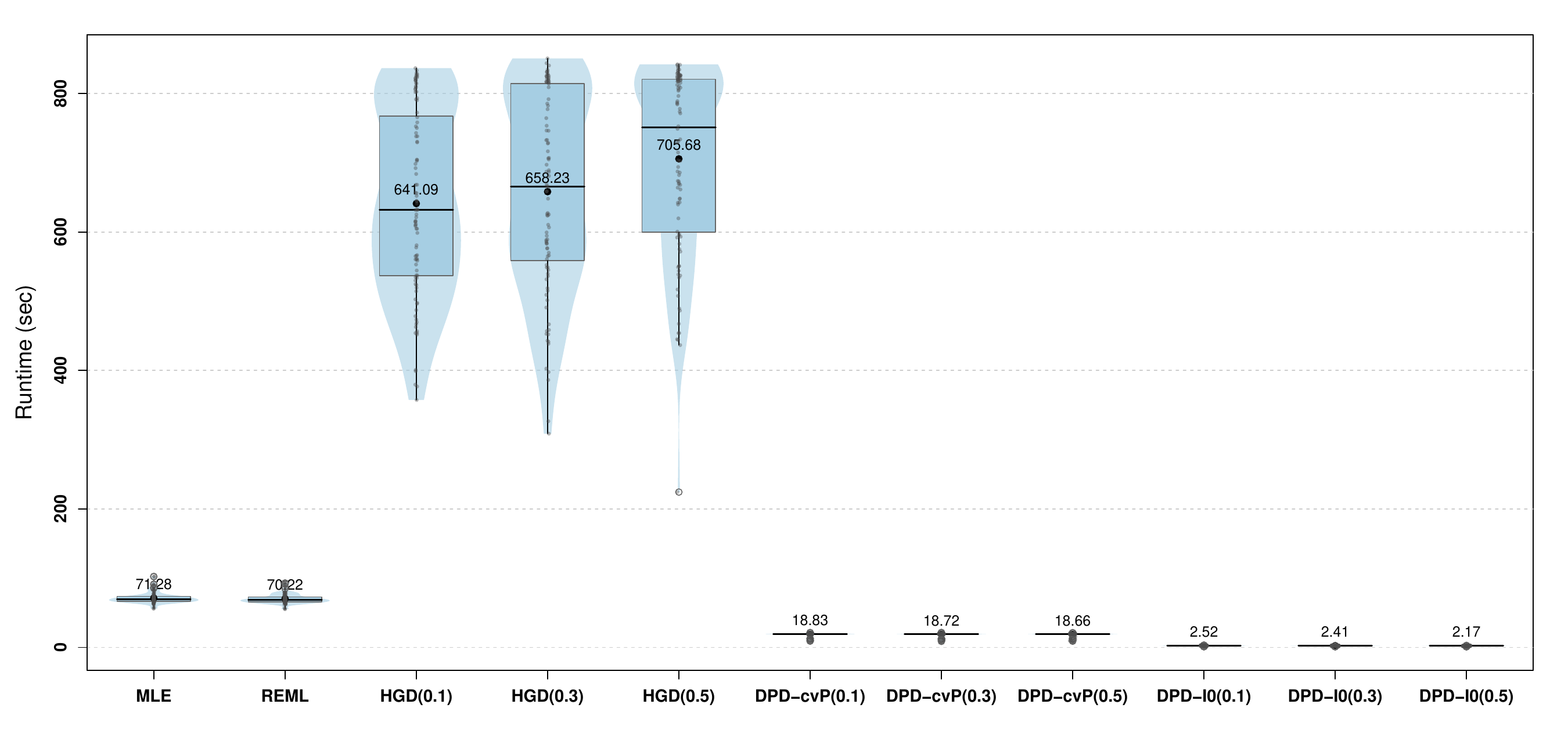}
		\label{FIG:boxplot_Y}}
	~	
	\subfloat[\tiny 10\% (C2) Contamination]{
		\includegraphics[page=3, width=0.3\textwidth]{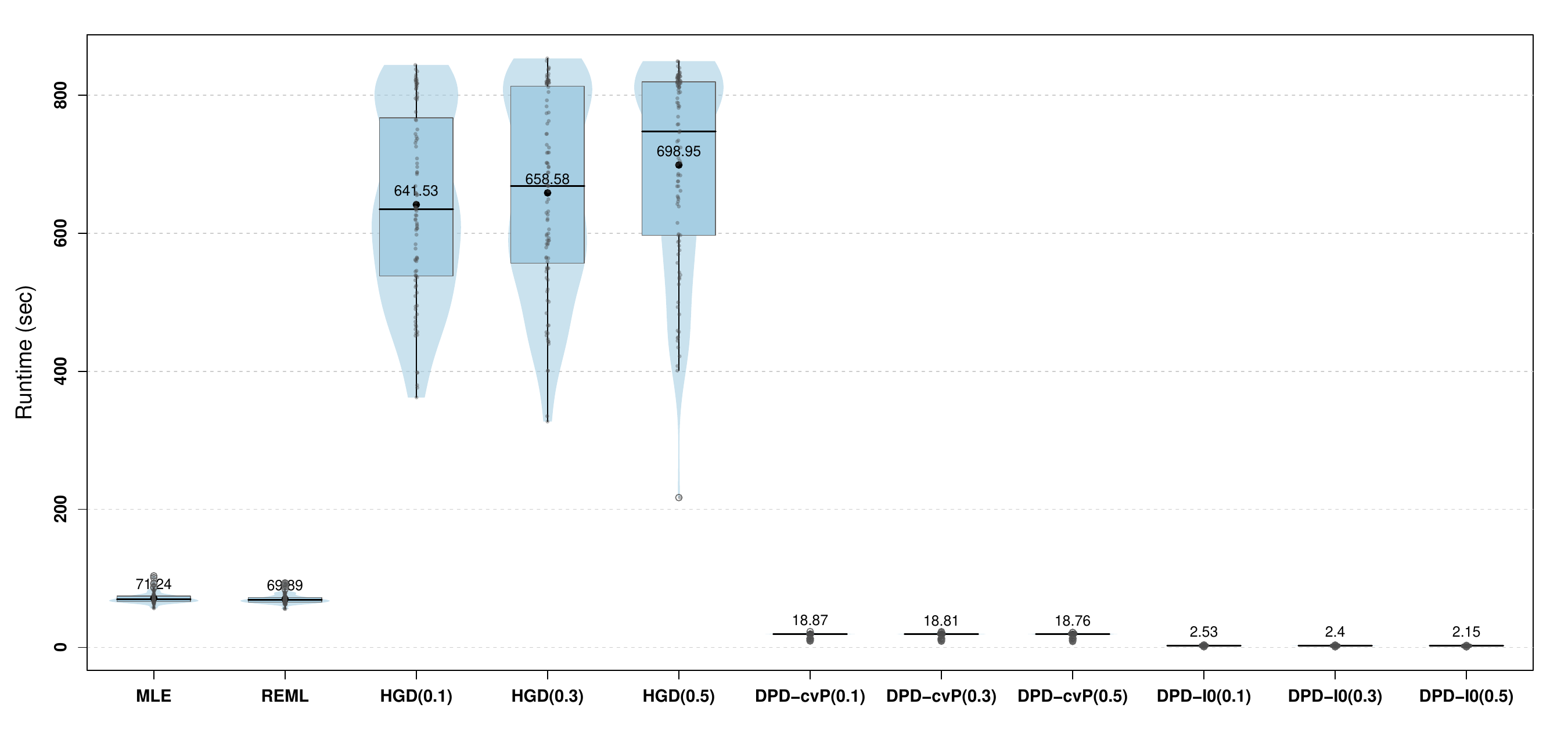}
		\label{FIG:boxplot_Y}}
	~	
	\subfloat[\tiny 20\% (C2) Contamination]{
		\includegraphics[page=3, width=0.3\textwidth]{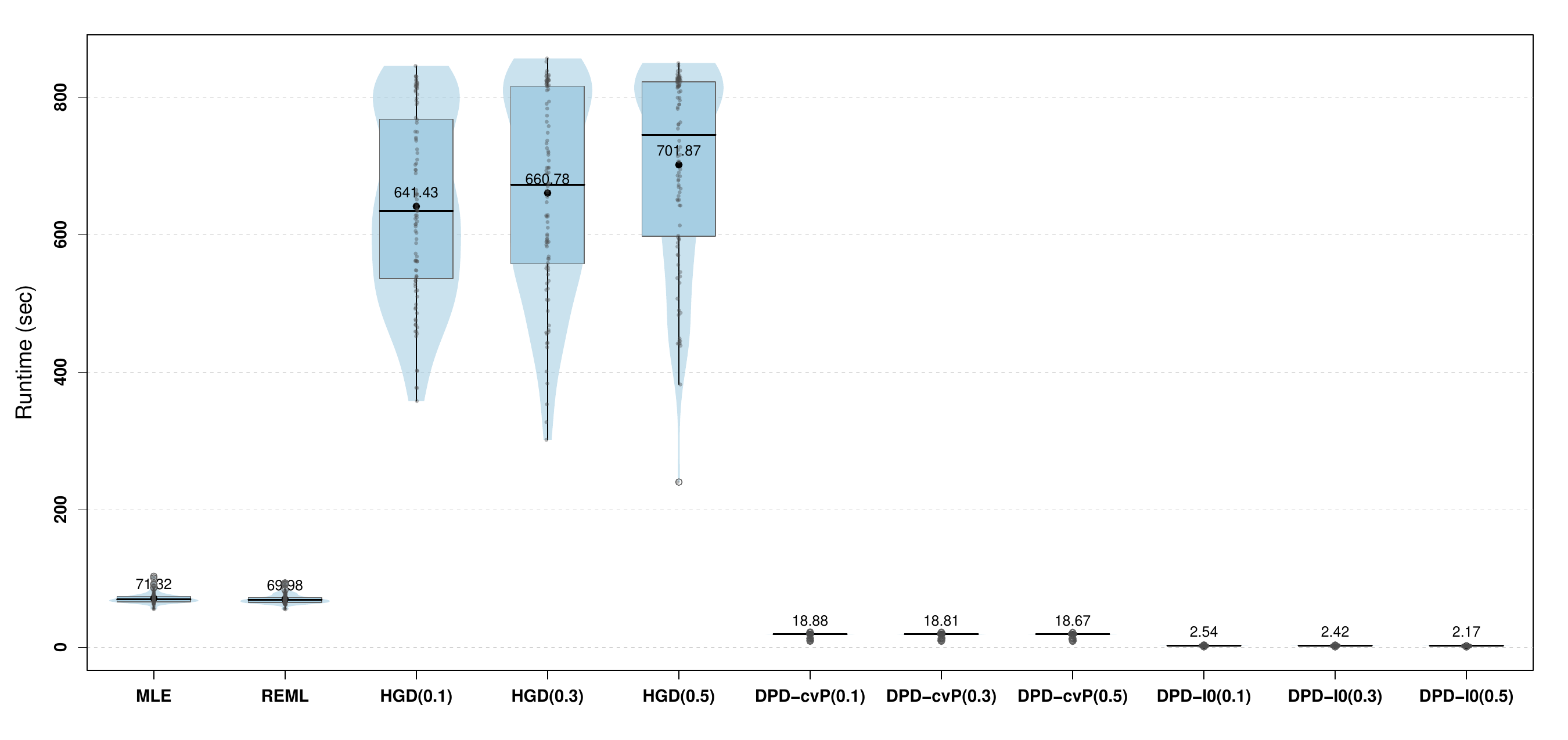}
		\label{FIG:boxplot_Y}}
	\\	
	\subfloat[\tiny 5\% (C3) Contamination]{
		\includegraphics[page=3, width=0.3\textwidth]{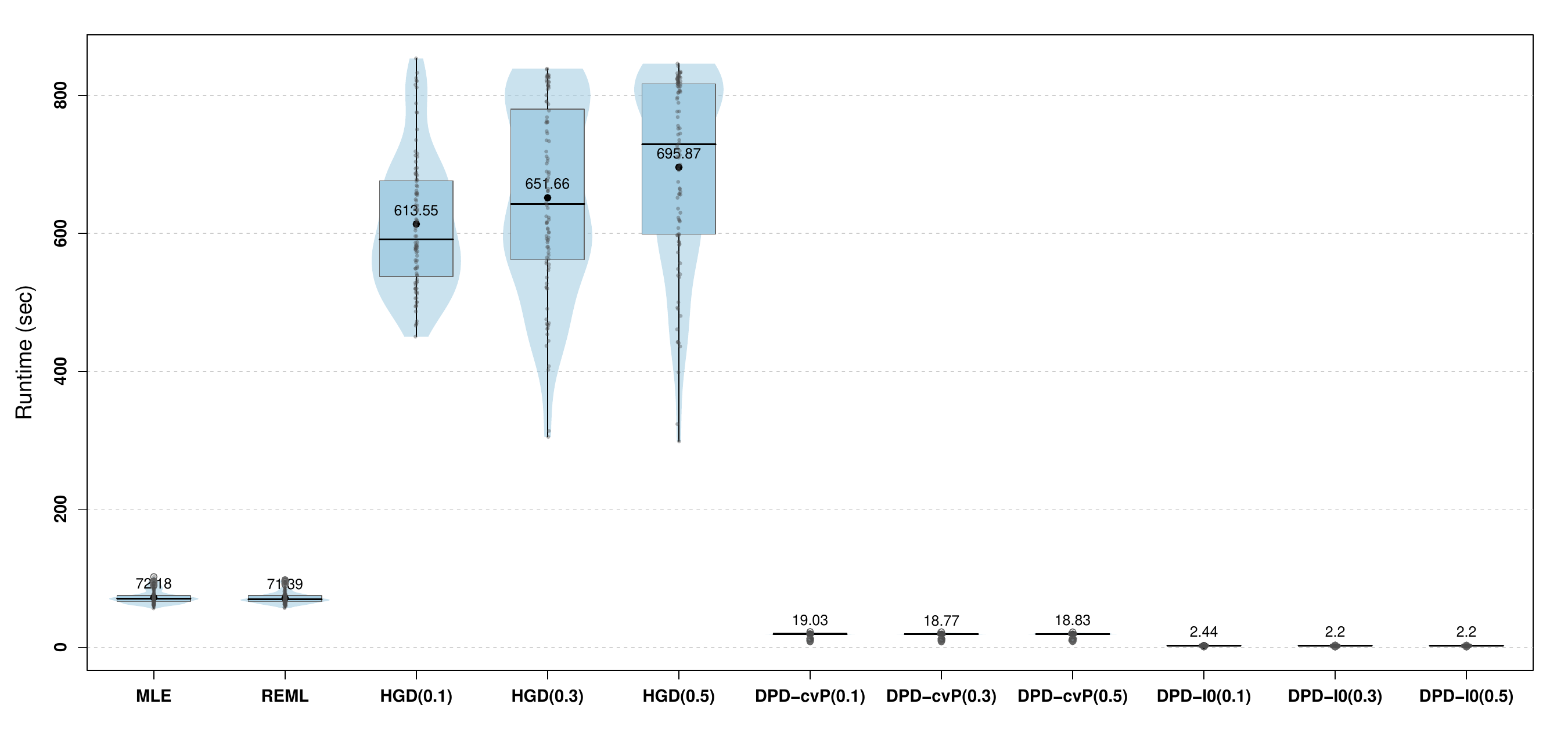}
		\label{FIG:boxplot_Y}}
	~	
	\subfloat[\tiny 10\% (C3) Contamination]{
		\includegraphics[page=3, width=0.3\textwidth]{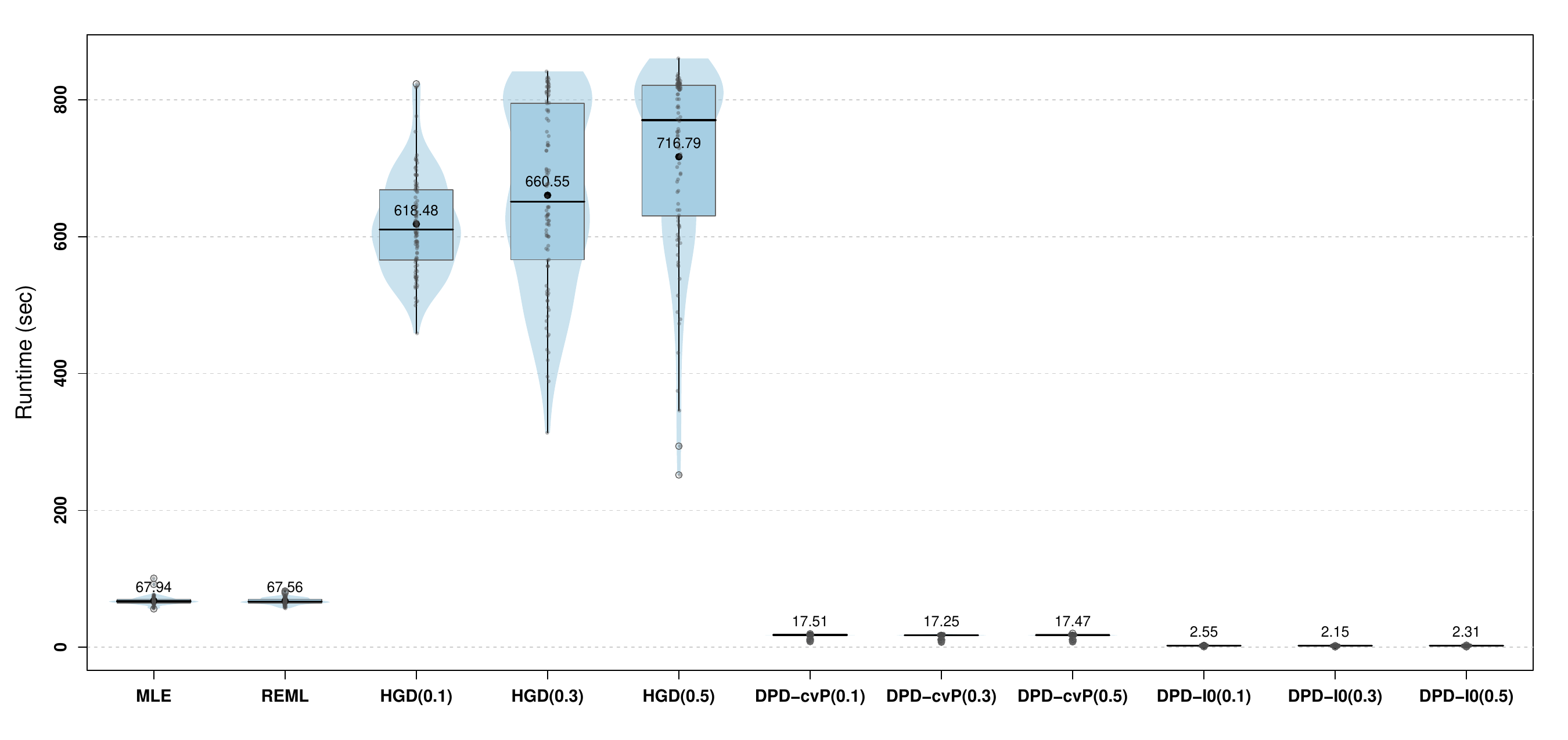}
		\label{FIG:boxplot_Y}}
	~	
	\subfloat[\tiny 20\% (C3) Contamination]{
		\includegraphics[page=3, width=0.3\textwidth]{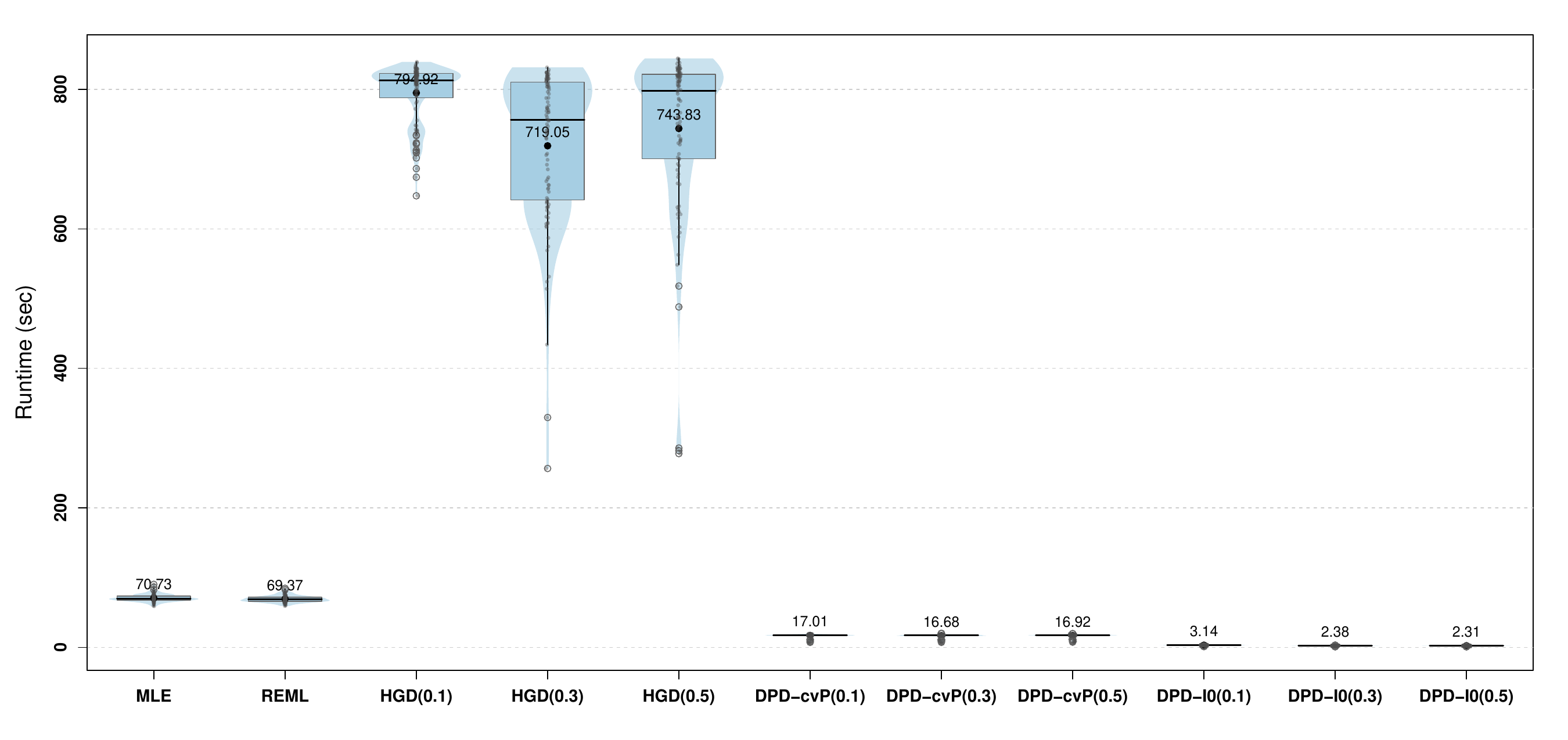}
		\label{FIG:boxplot_Y}}
	\\	
	\subfloat[\tiny 5\% (C4) Contamination]{
		\includegraphics[page=3, width=0.3\textwidth]{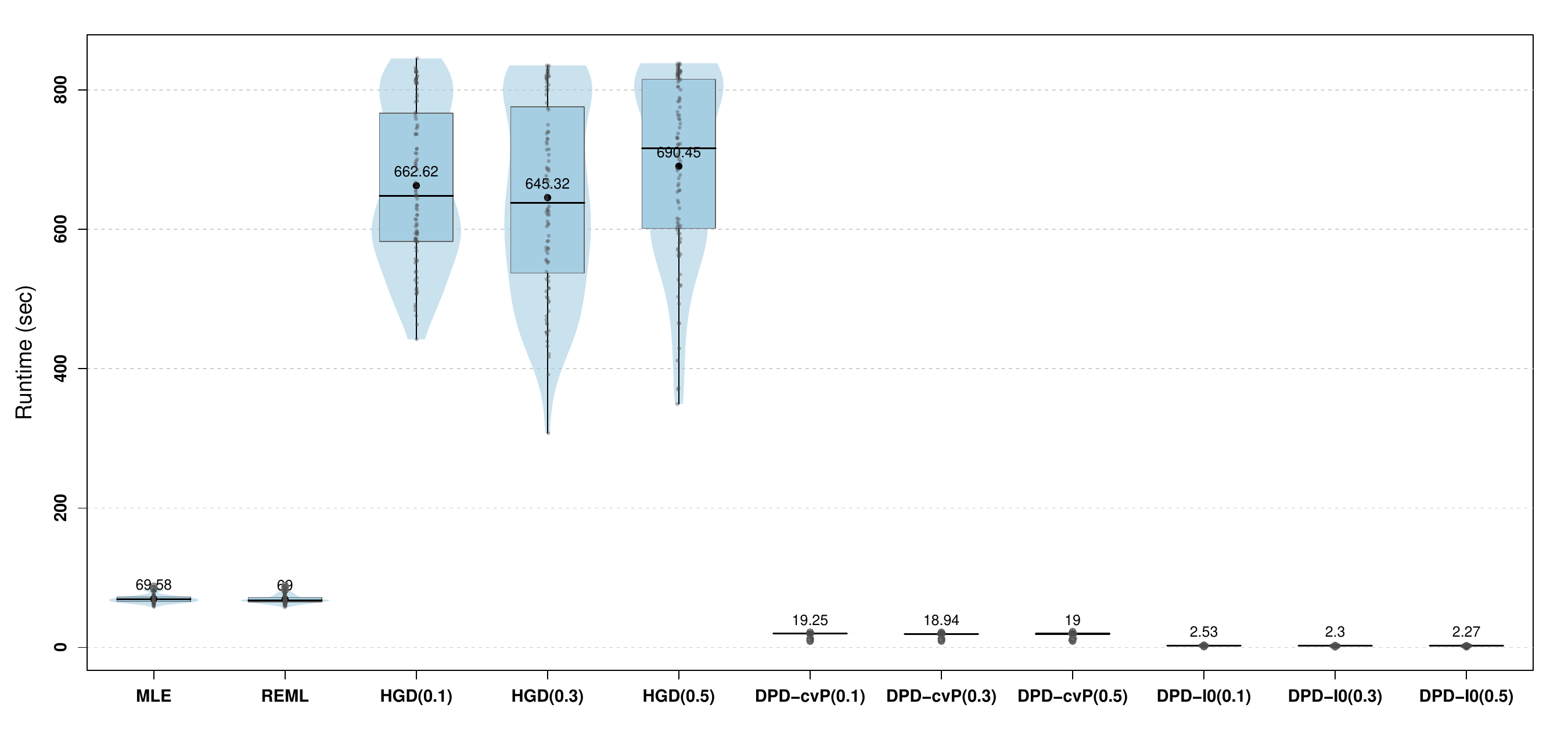}
		\label{FIG:boxplot_Y}}
	~	
	\subfloat[\tiny 10\% (C4) Contamination]{
		\includegraphics[page=3, width=0.3\textwidth]{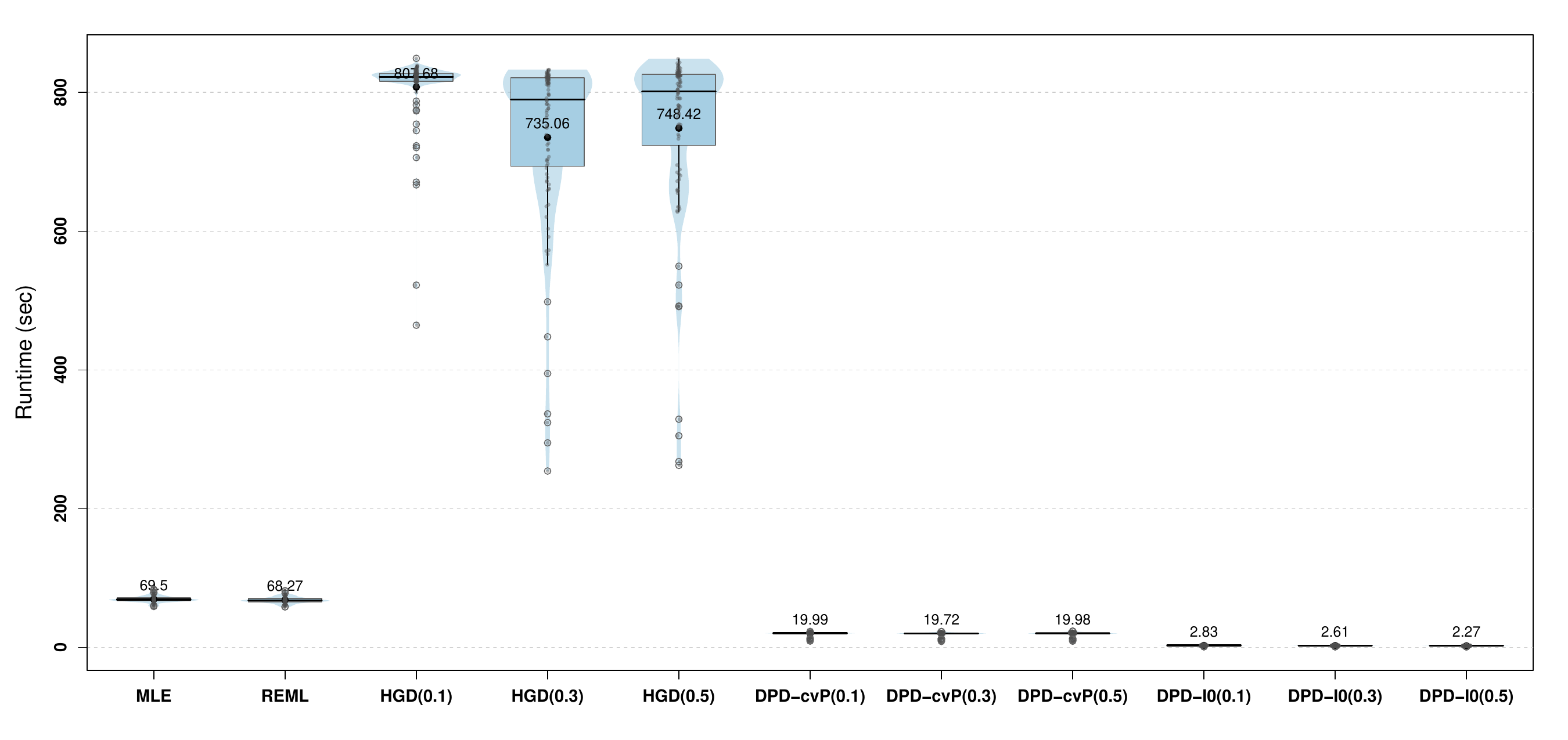}
		\label{FIG:boxplot_Y}}
	~	
	\subfloat[\tiny 20\% (C4) Contamination]{
		\includegraphics[page=3, width=0.3\textwidth]{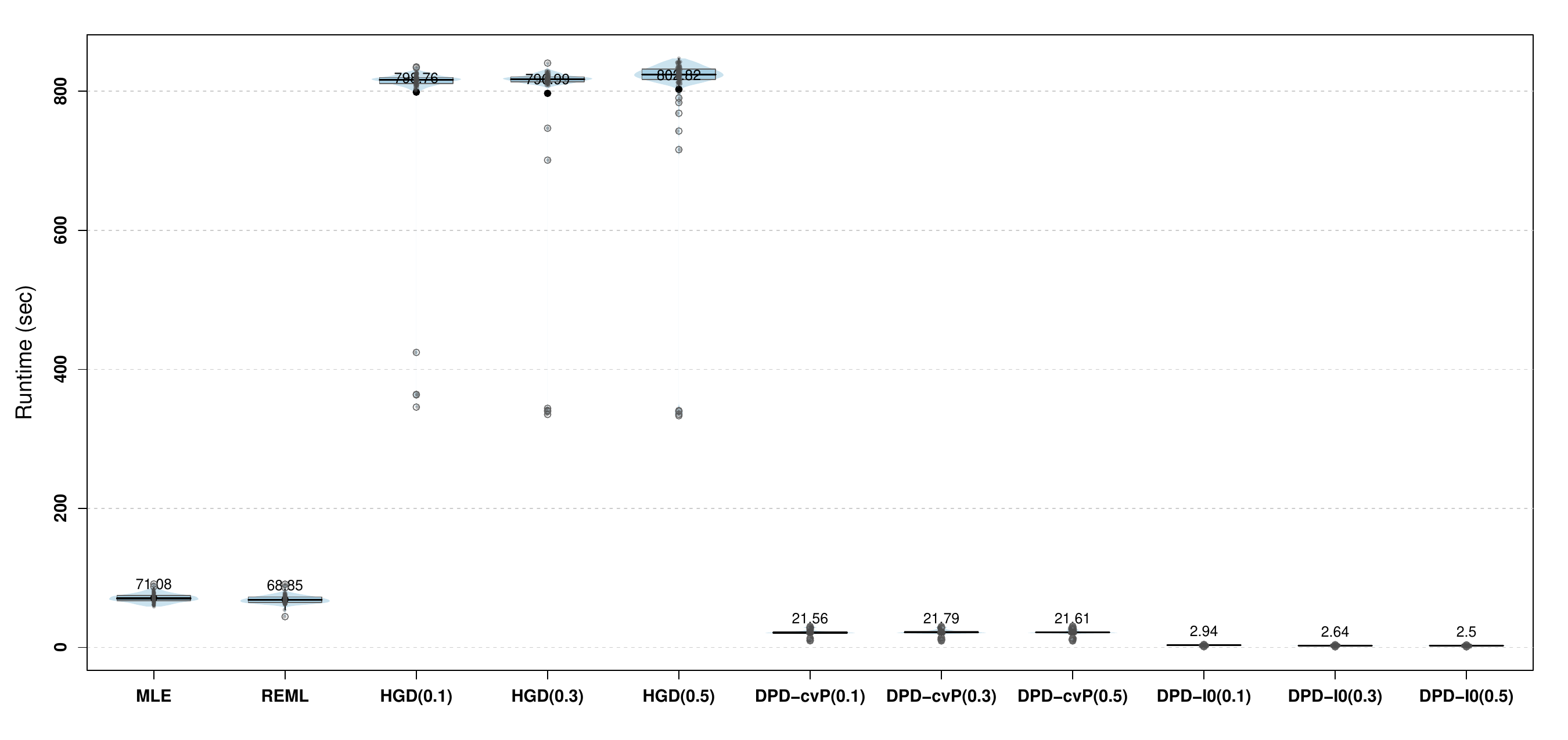}
		\label{FIG:boxplot_Y}}
	\caption{Boxplots  (with overlaid sample means) of computational times (in seconds) of the screening methods  under scenario (S2)$\times$(R1) 
		with $\bm\Sigma_x=\mathbb{I}$ and different types of contamination}
	\label{FIG:time_S2R1Id}
\end{figure}

\begin{figure}[!h]
	\centering
	\subfloat[\tiny 5\% (C1) Contamination]{
		\includegraphics[page=3, width=0.3\textwidth]{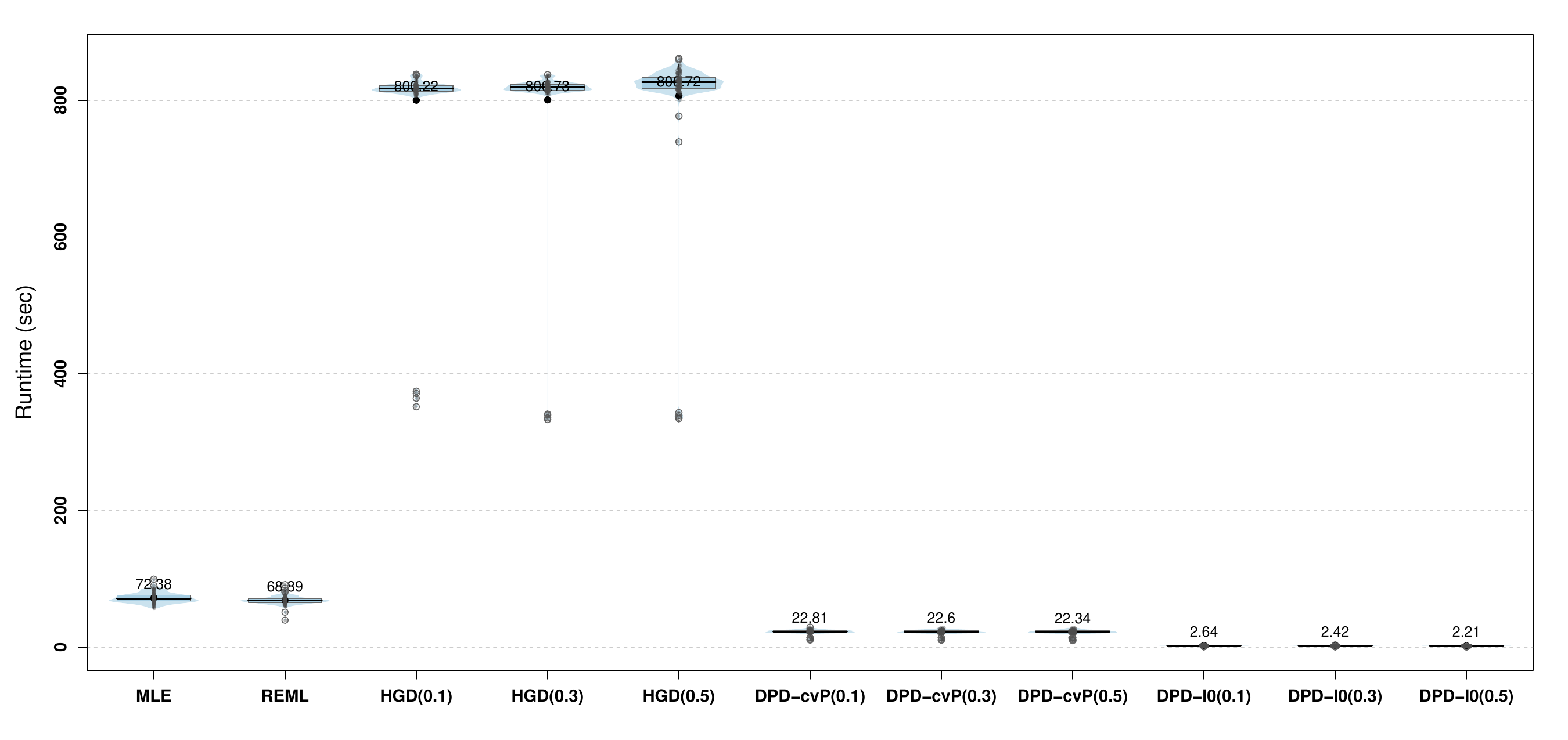}
		\label{FIG:boxplot_Y}}
	~	
	\subfloat[\tiny 10\% (C1) Contamination]{
		\includegraphics[page=3, width=0.3\textwidth]{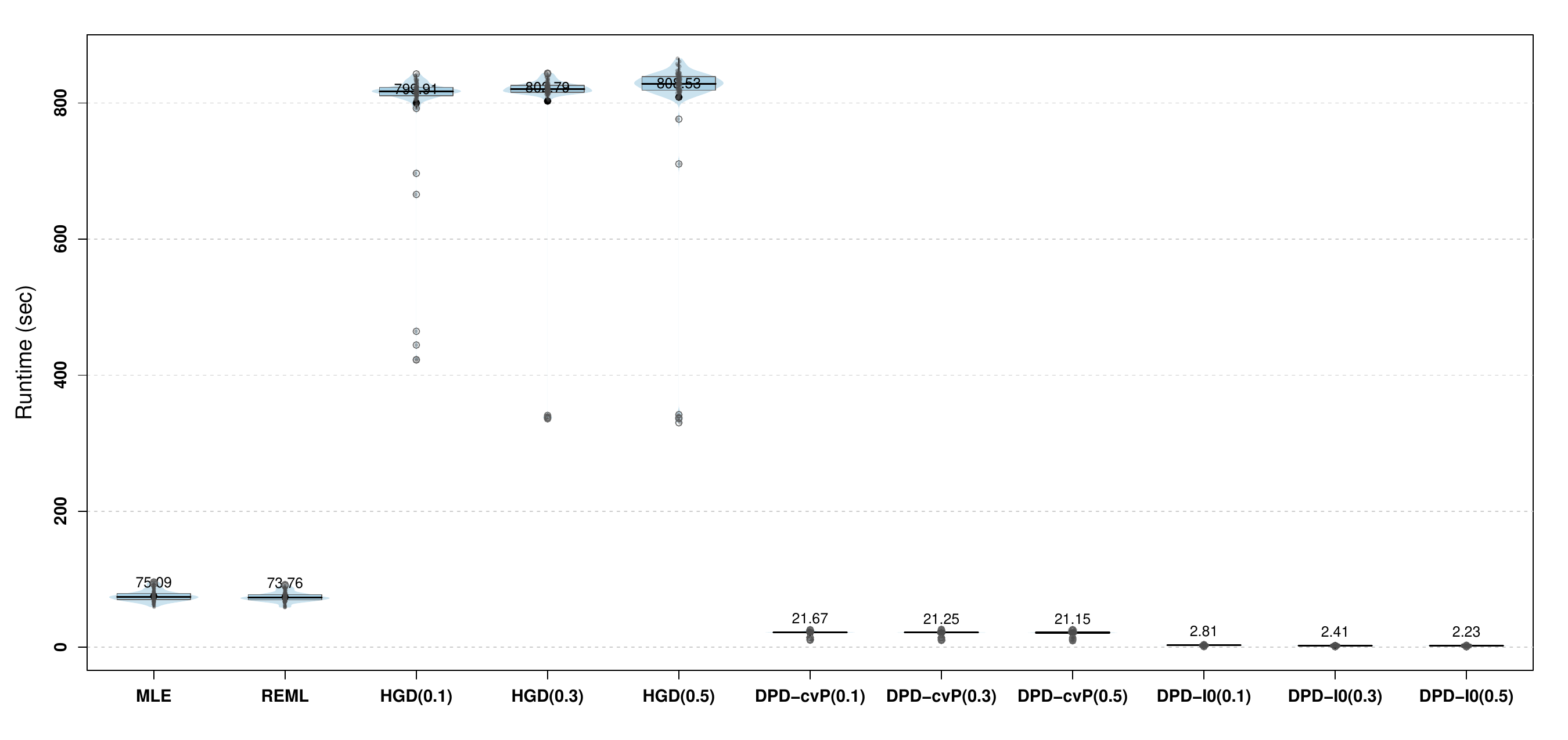}
		\label{FIG:boxplot_Y}}
	~	
	\subfloat[\tiny 20\% (C1) Contamination]{
		\includegraphics[page=3, width=0.3\textwidth]{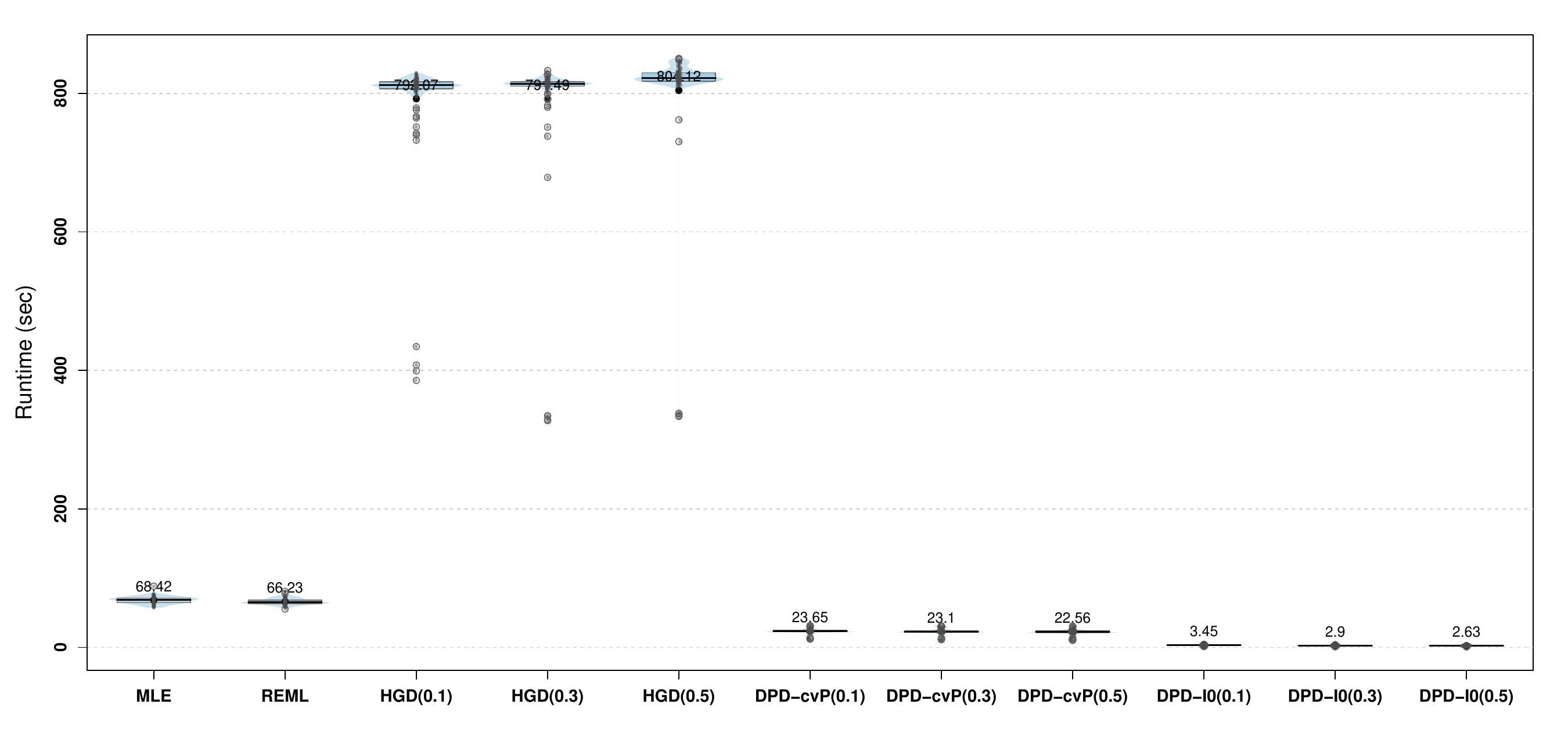}
		\label{FIG:boxplot_Y}}
	\\	
	\subfloat[\tiny 5\% (C2) Contamination]{
		\includegraphics[page=3, width=0.3\textwidth]{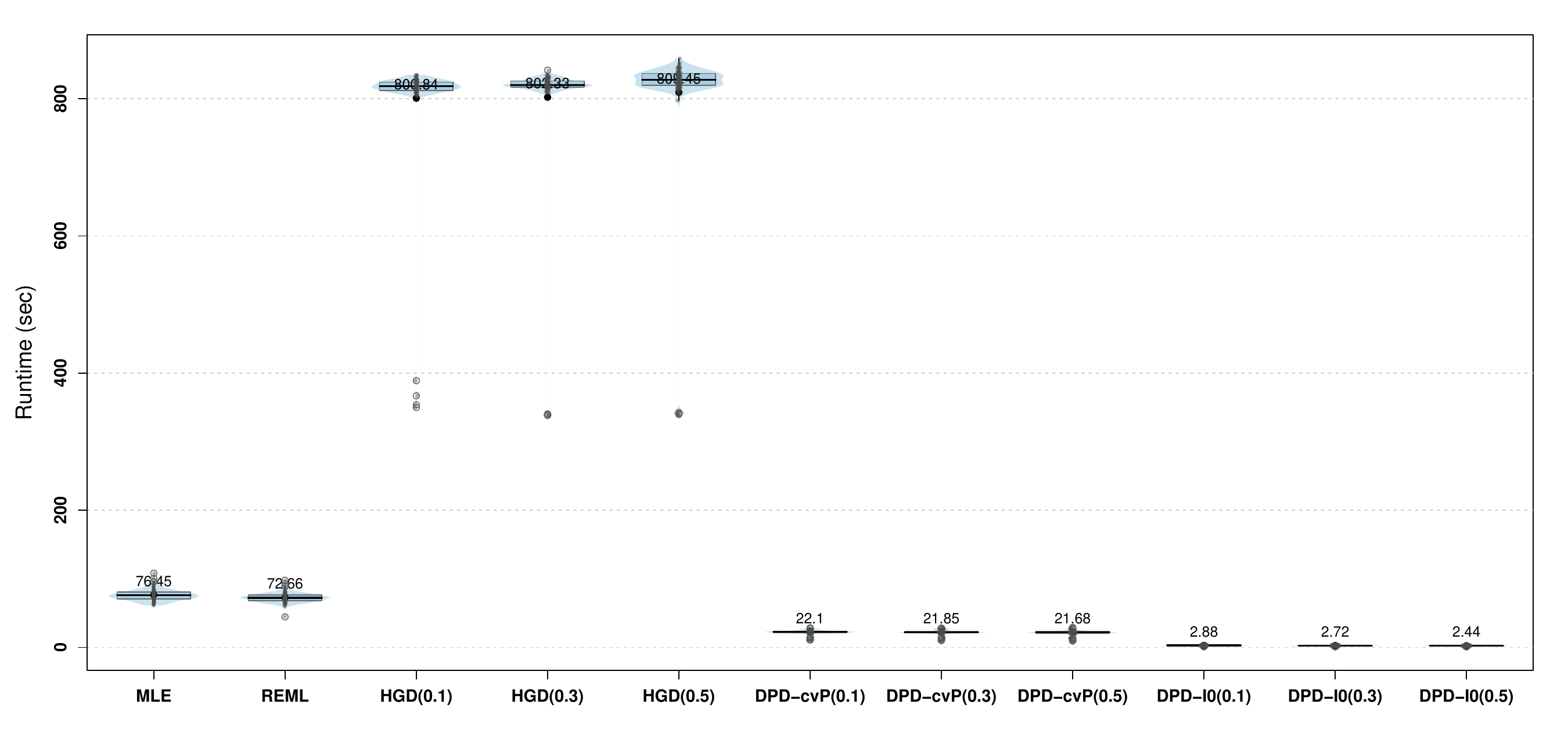}
		\label{FIG:boxplot_Y}}
	~	
	\subfloat[\tiny 10\% (C2) Contamination]{
		\includegraphics[page=3, width=0.3\textwidth]{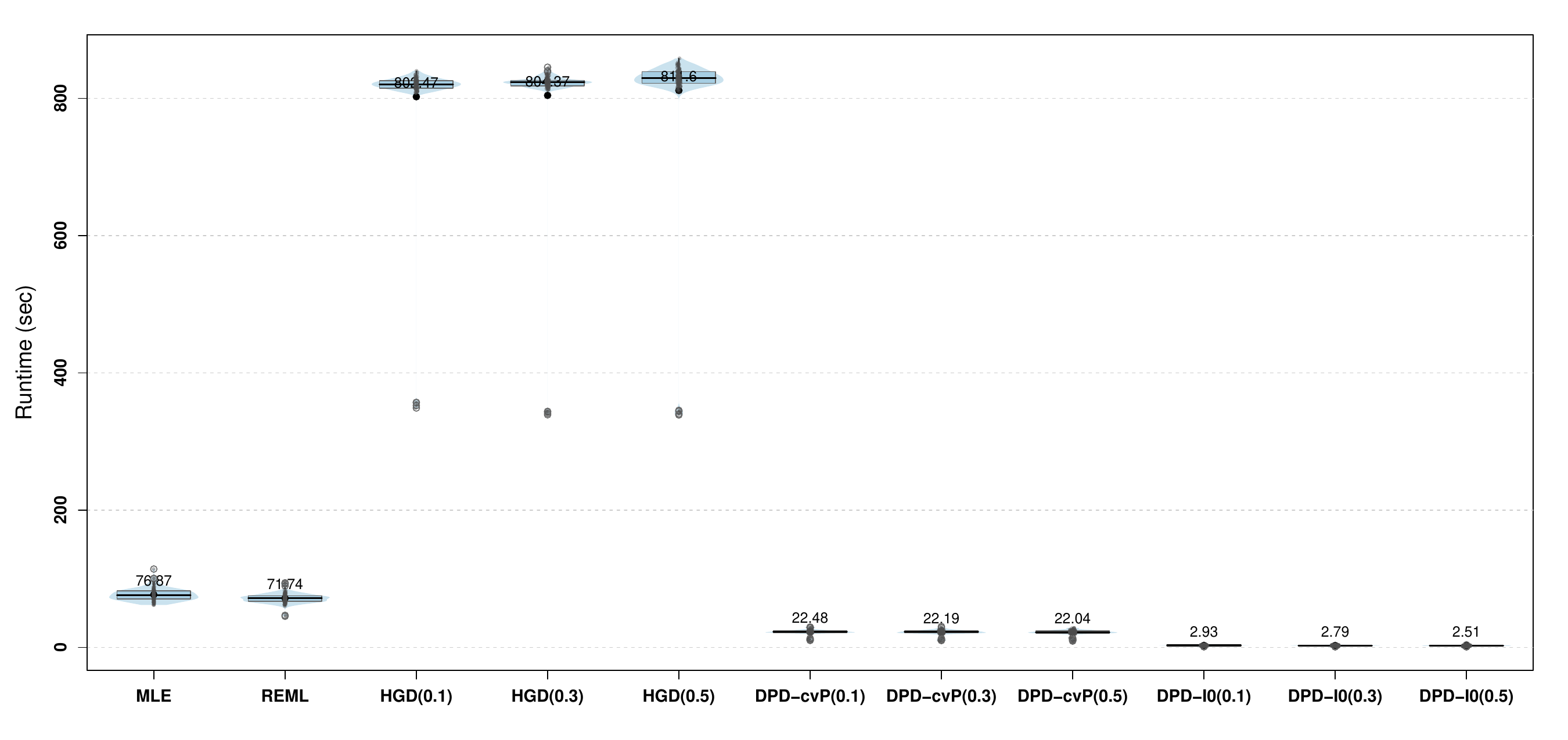}
		\label{FIG:boxplot_Y}}
	~	
	\subfloat[\tiny 20\% (C2) Contamination]{
		\includegraphics[page=3, width=0.3\textwidth]{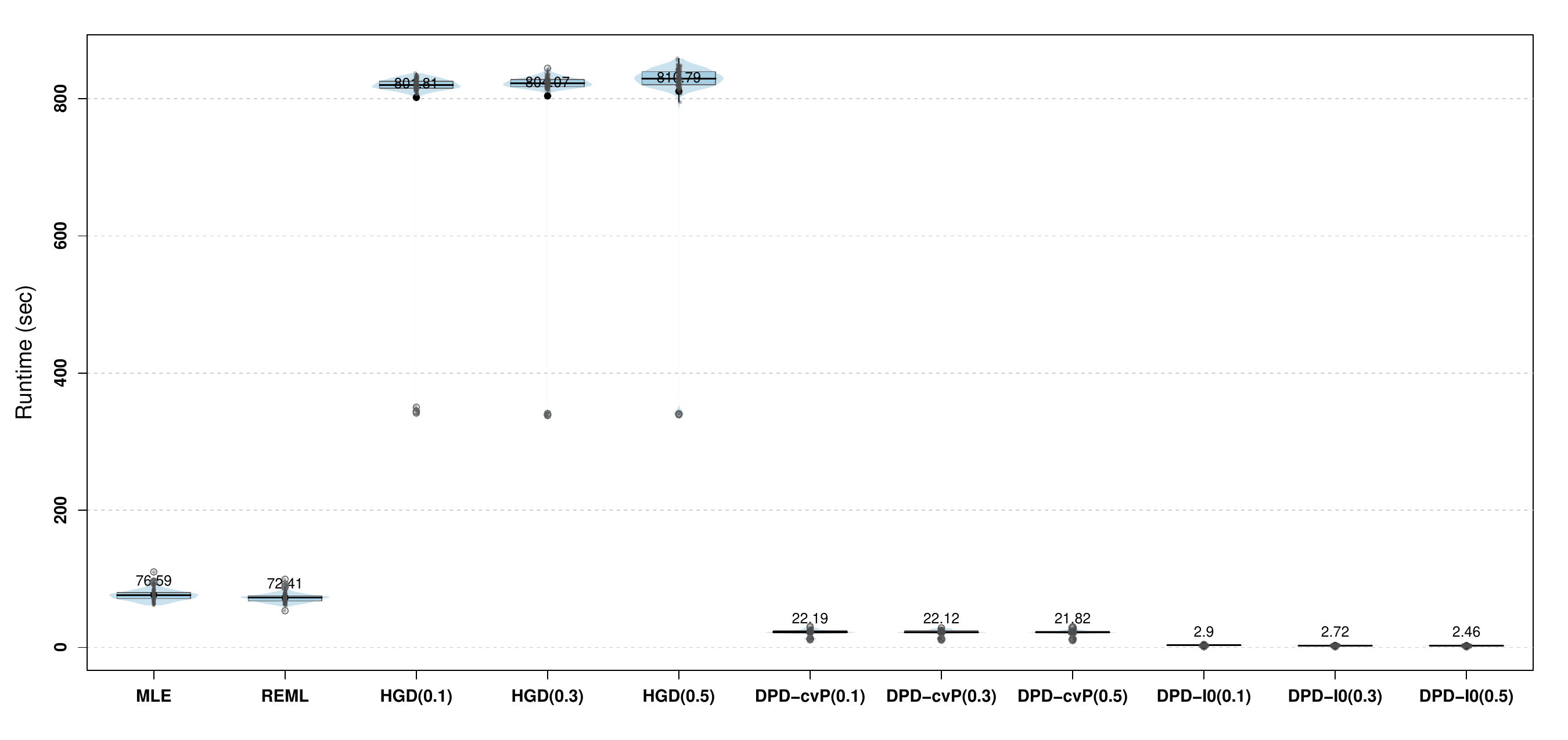}
		\label{FIG:boxplot_Y}}
	\\	
	\subfloat[\tiny 5\% (C3) Contamination]{
		\includegraphics[page=3, width=0.3\textwidth]{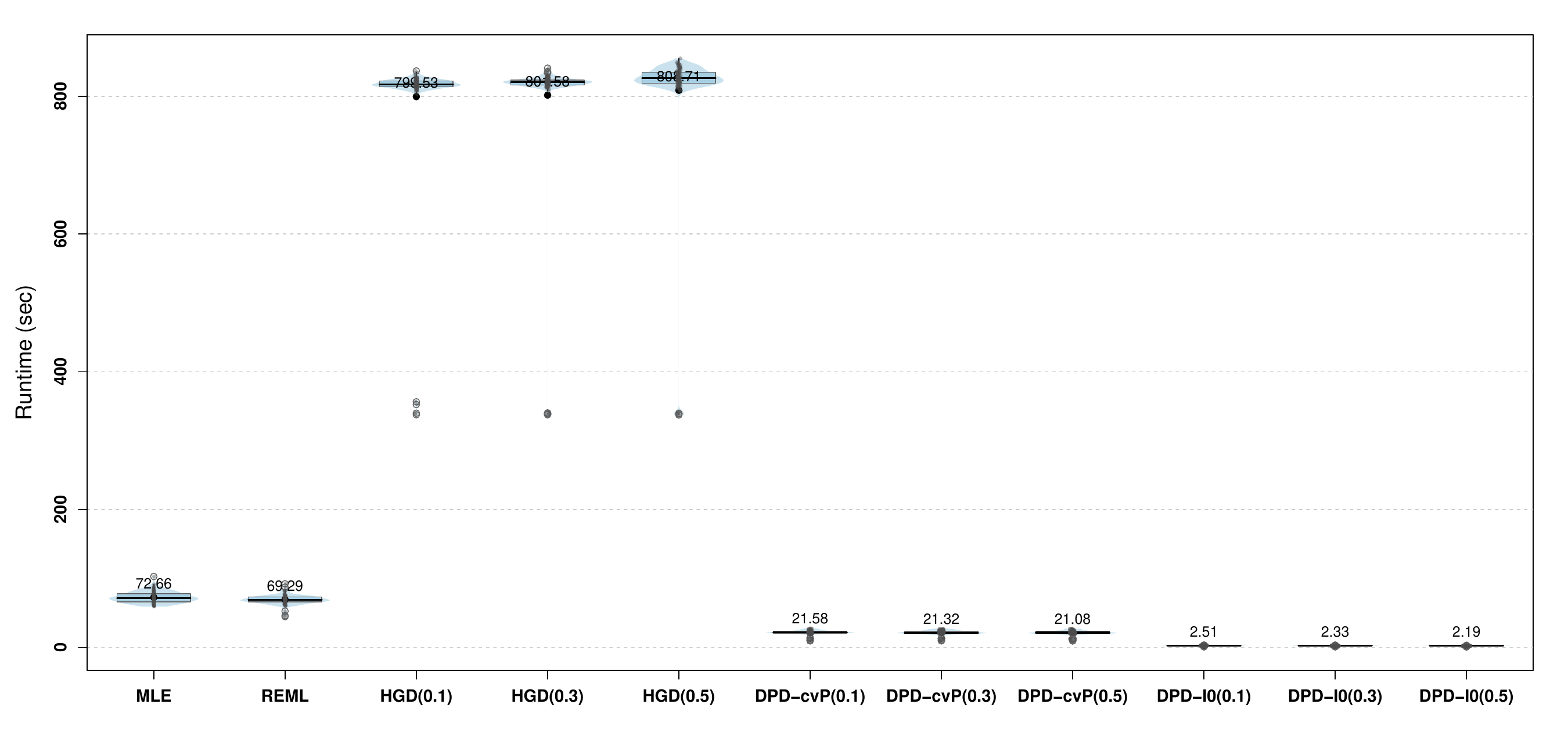}
		\label{FIG:boxplot_Y}}
	~	
	\subfloat[\tiny 10\% (C3) Contamination]{
		\includegraphics[page=3, width=0.3\textwidth]{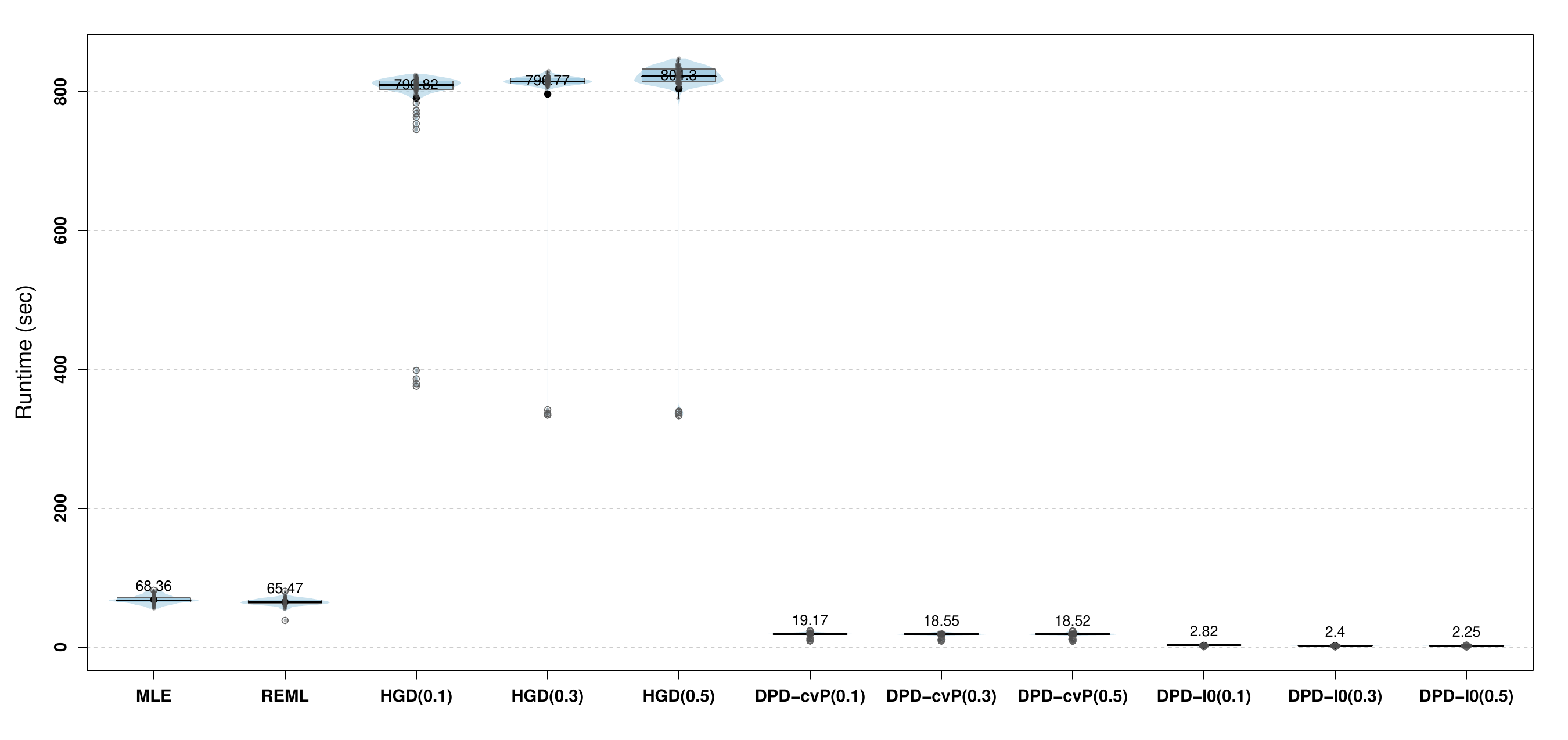}
		\label{FIG:boxplot_Y}}
	~	
	\subfloat[\tiny 20\% (C3) Contamination]{
		\includegraphics[page=3, width=0.3\textwidth]{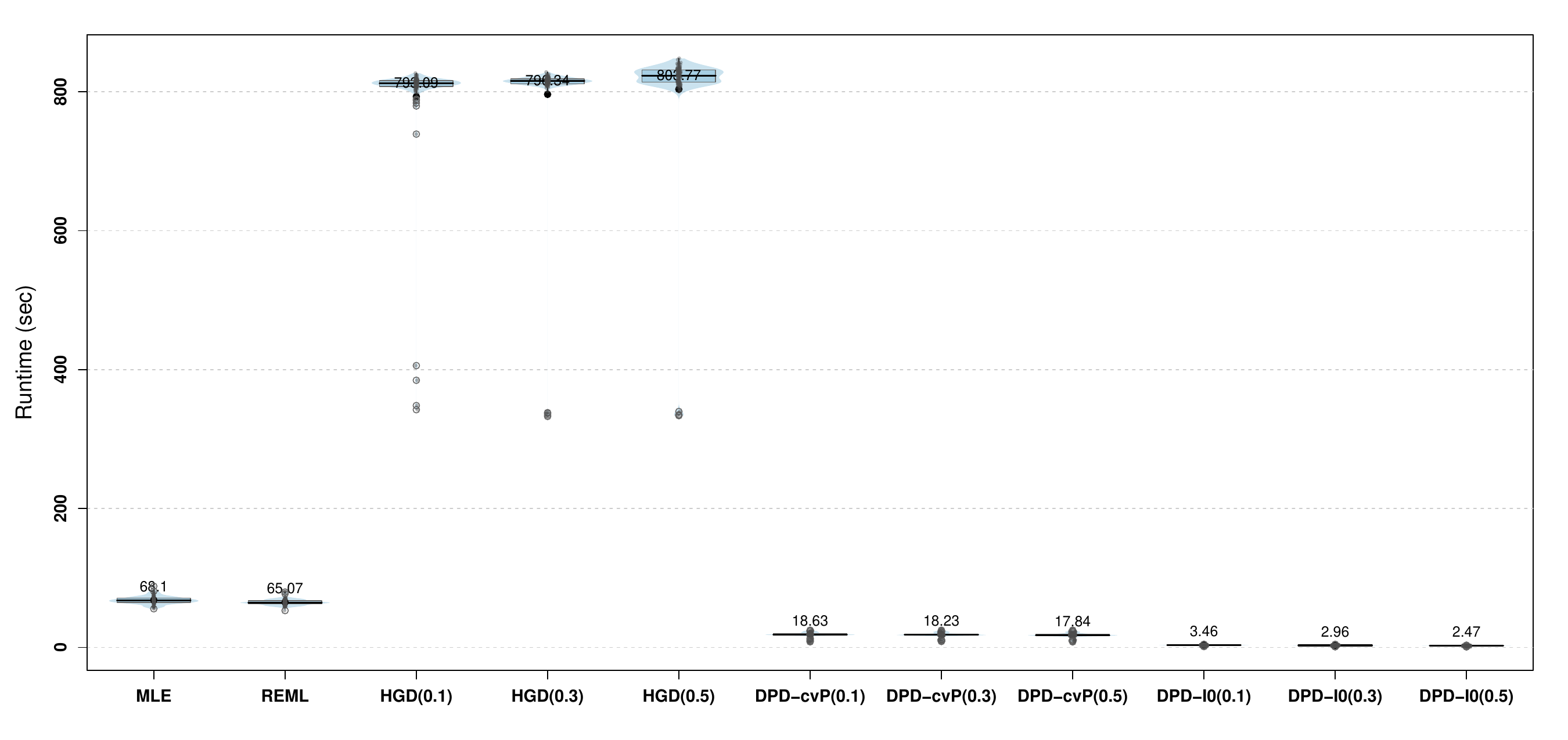}
		\label{FIG:boxplot_Y}}
	\\	
	\subfloat[\tiny 5\% (C4) Contamination]{
		\includegraphics[page=3, width=0.3\textwidth]{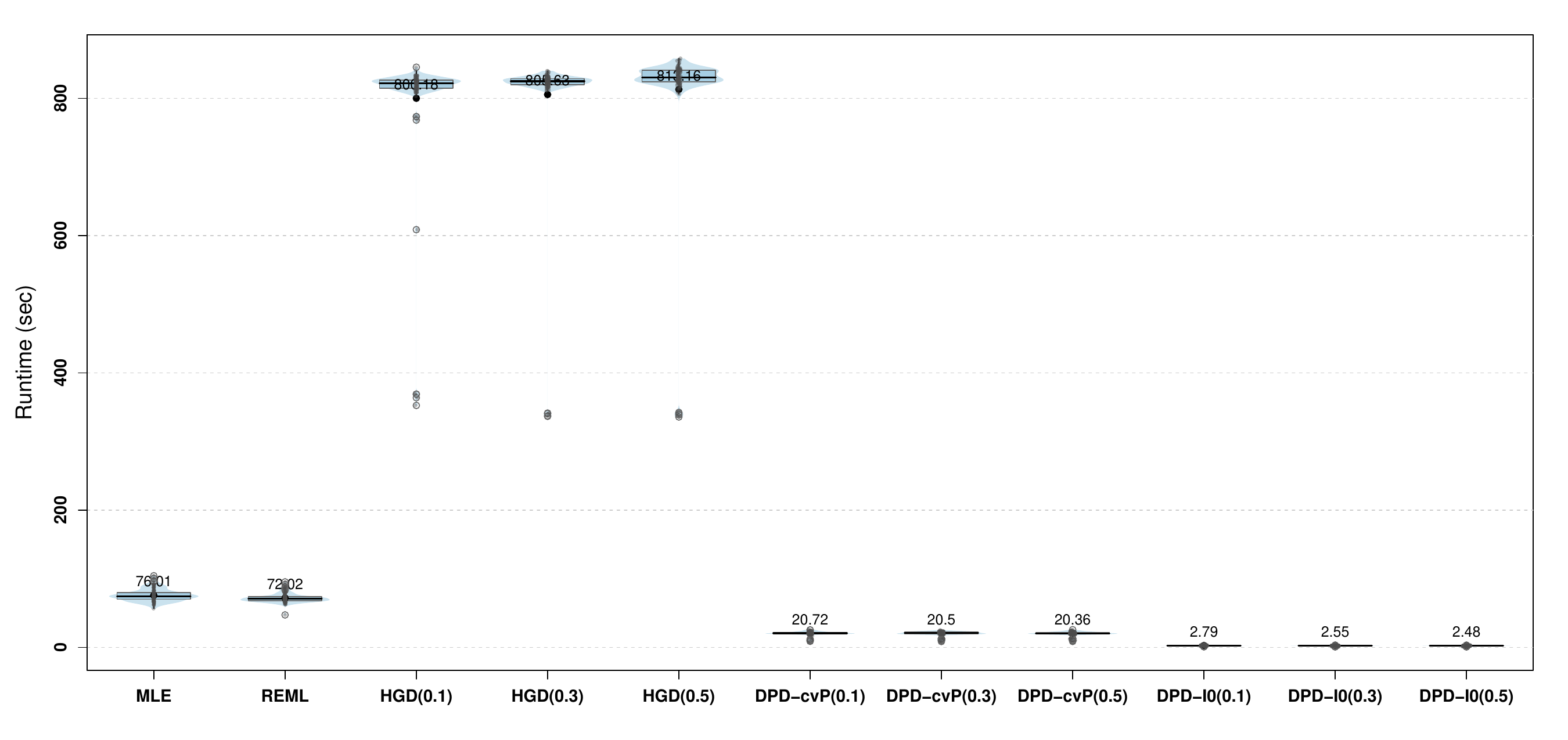}
		\label{FIG:boxplot_Y}}
	~	
	\subfloat[\tiny 10\% (C4) Contamination]{
		\includegraphics[page=3, width=0.3\textwidth]{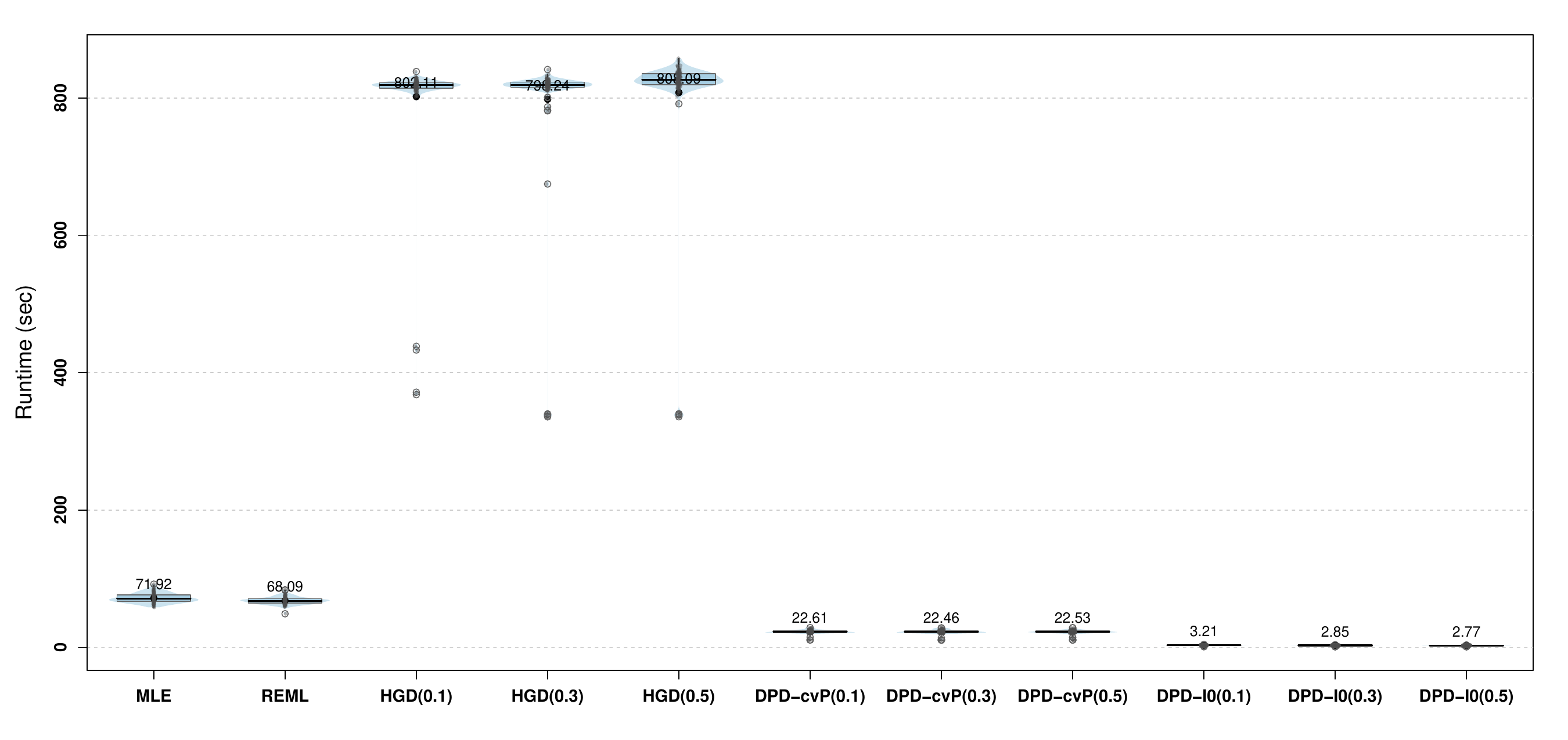}
		\label{FIG:boxplot_Y}}
	~	
	\subfloat[\tiny 20\% (C4) Contamination]{
		\includegraphics[page=3, width=0.3\textwidth]{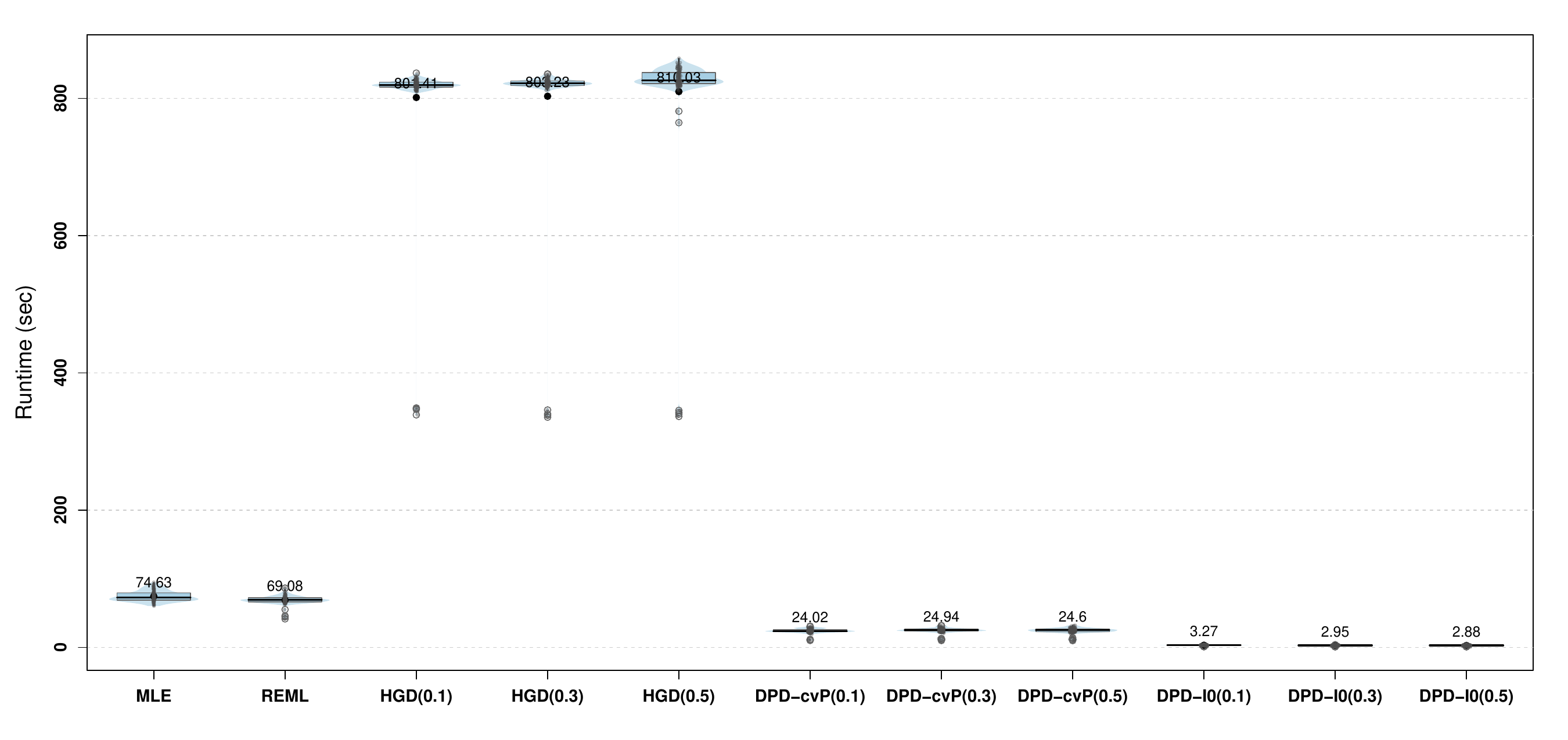}
		\label{FIG:boxplot_Y}}
	\caption{Boxplots  (with overlaid sample means) of computational times (in seconds) of the screening methods  under scenario (S1)$\times$(R1) with $\bm\Sigma_x$ as CS(0.3) and different types of contamination}
	\label{FIG:time_S1R1CS}
\end{figure}

\begin{figure}[!h]
	\centering
	\subfloat[\tiny 5\% (C1) Contamination]{
		\includegraphics[page=3, width=0.3\textwidth]{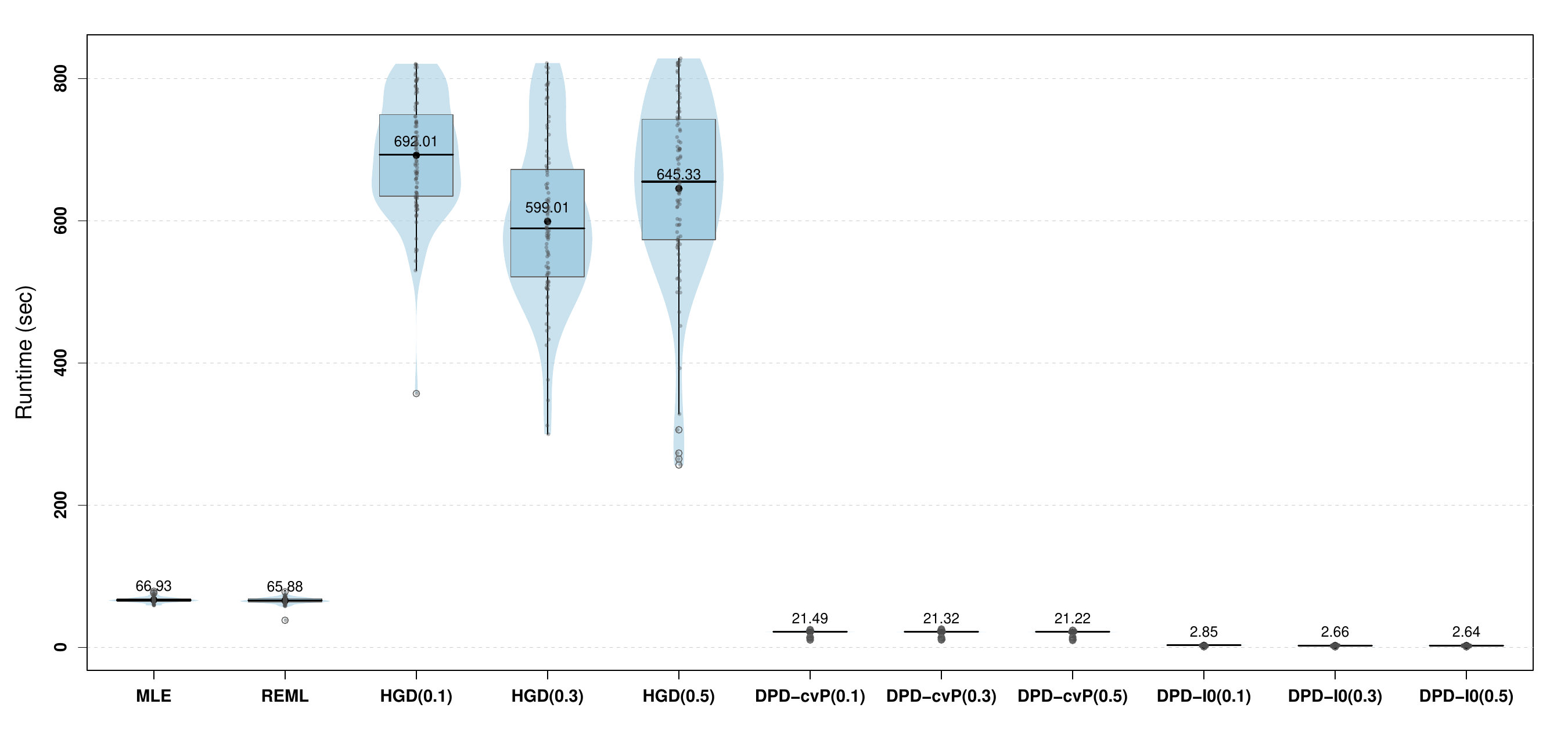}
		\label{FIG:boxplot_Y}}
	~	
	\subfloat[\tiny 10\% (C1) Contamination]{
		\includegraphics[page=3, width=0.3\textwidth]{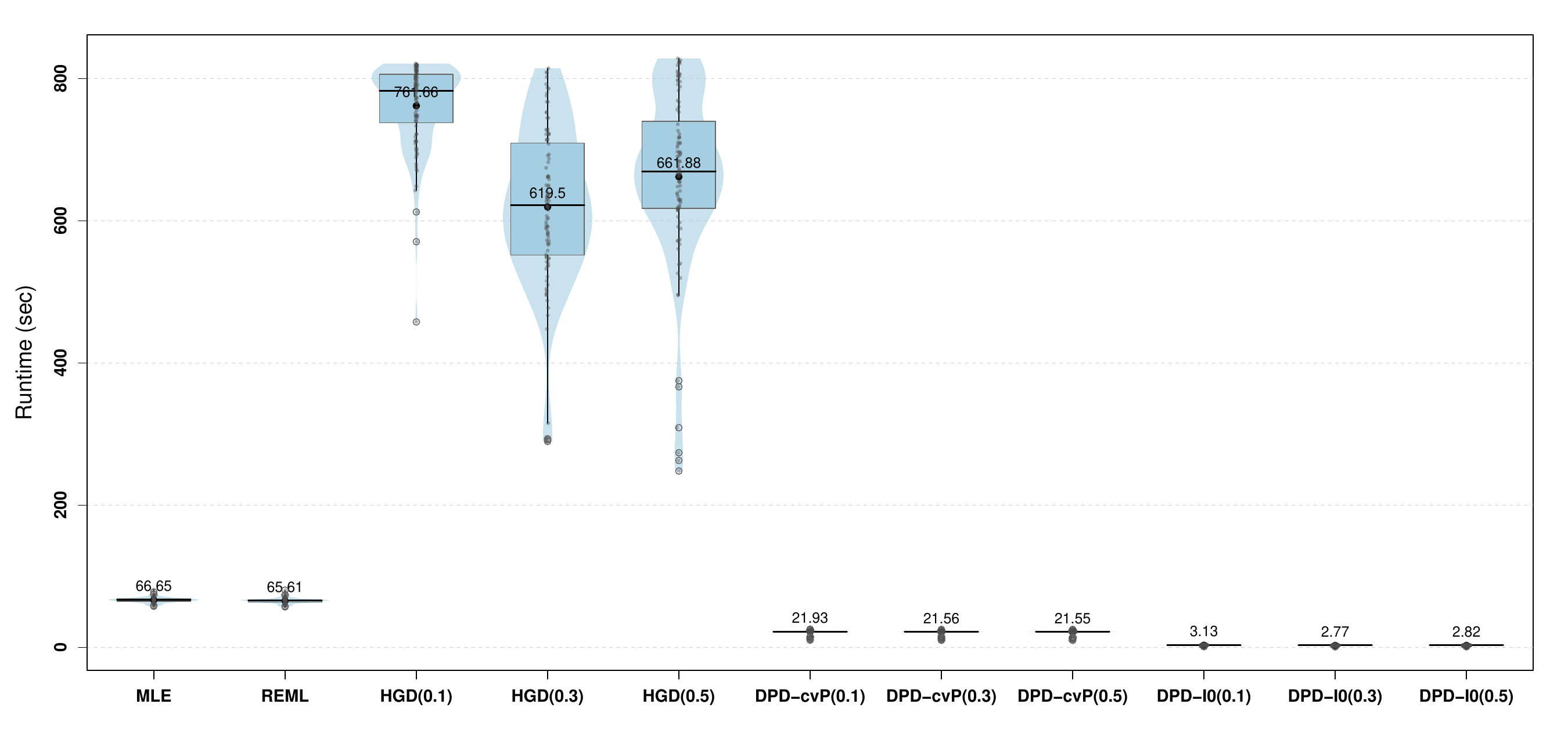}
		\label{FIG:boxplot_Y}}
	~	
	\subfloat[\tiny 20\% (C1) Contamination]{
		\includegraphics[page=3, width=0.3\textwidth]{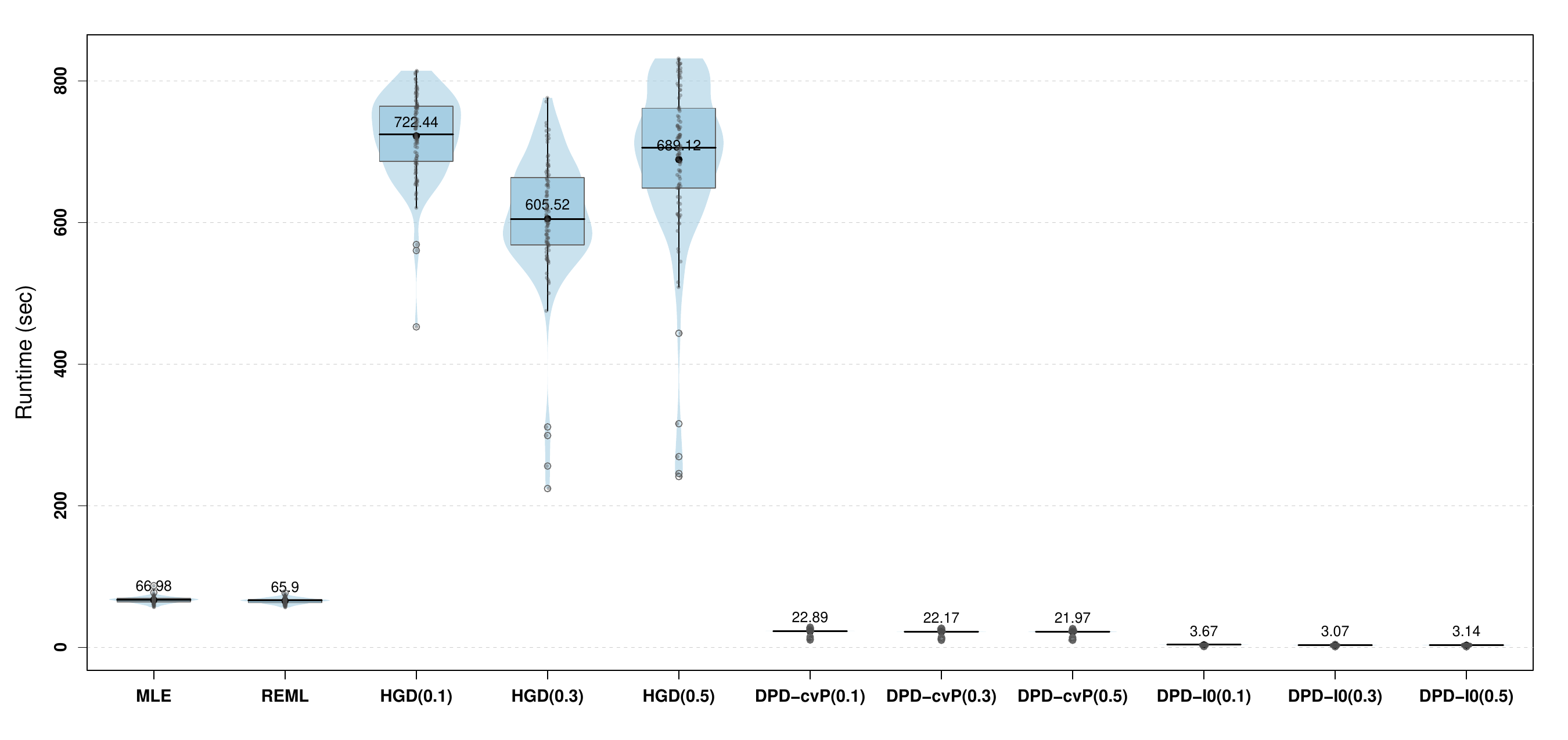}
		\label{FIG:boxplot_Y}}
	\\	
	\subfloat[\tiny 5\% (C2) Contamination]{
		\includegraphics[page=3, width=0.3\textwidth]{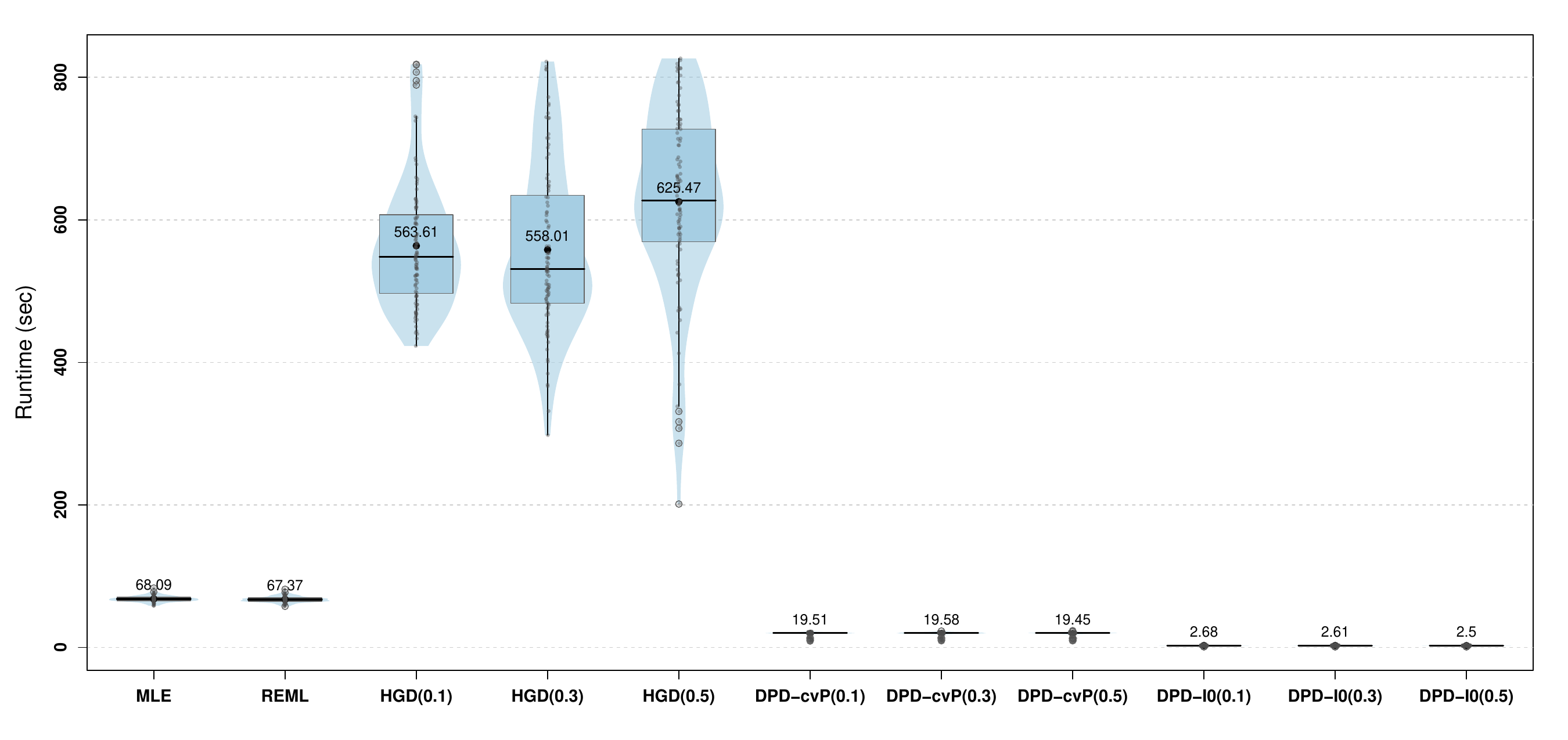}
		\label{FIG:boxplot_Y}}
	~	
	\subfloat[\tiny 10\% (C2) Contamination]{
		\includegraphics[page=3, width=0.3\textwidth]{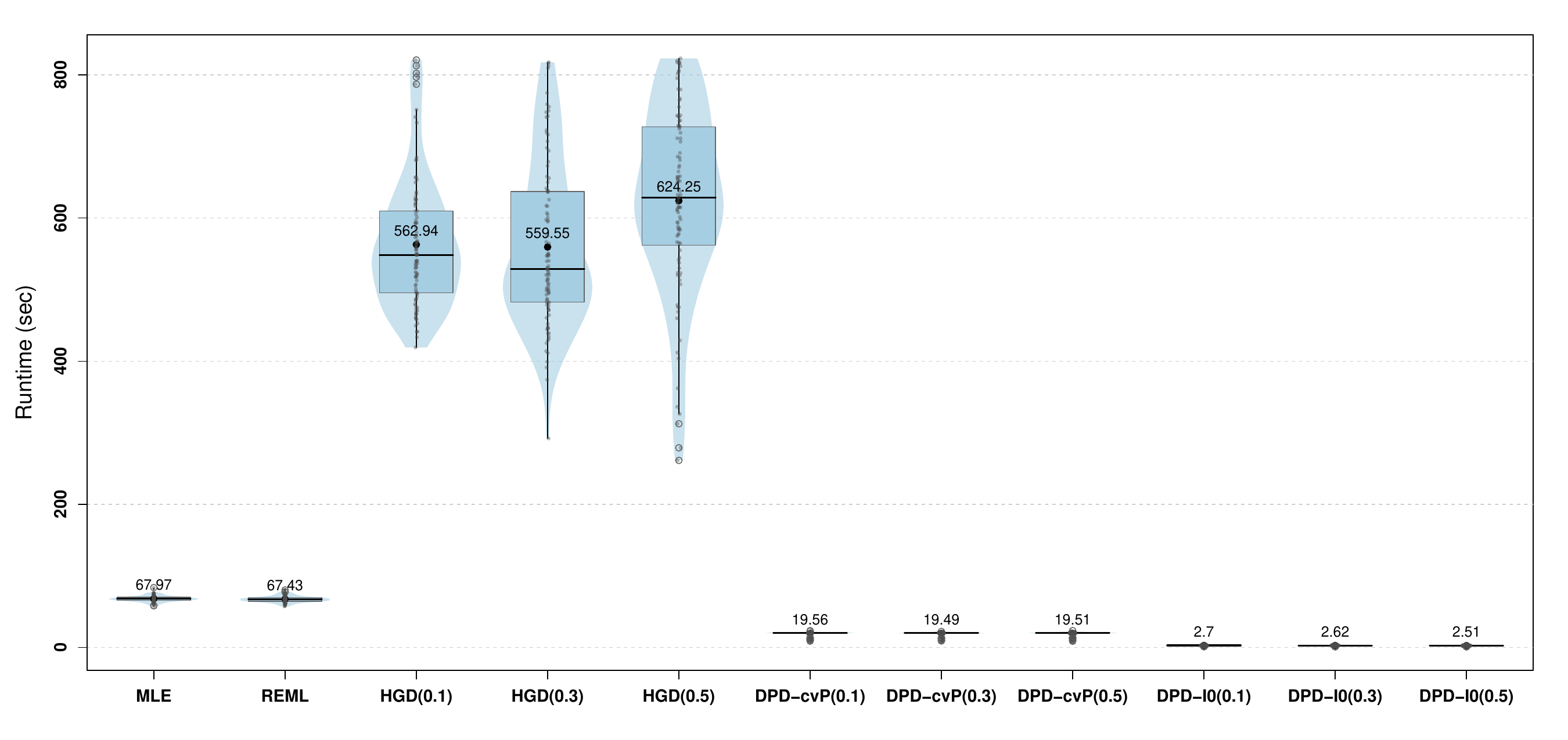}
		\label{FIG:boxplot_Y}}
	~	
	\subfloat[\tiny 20\% (C2) Contamination]{
		\includegraphics[page=3, width=0.3\textwidth]{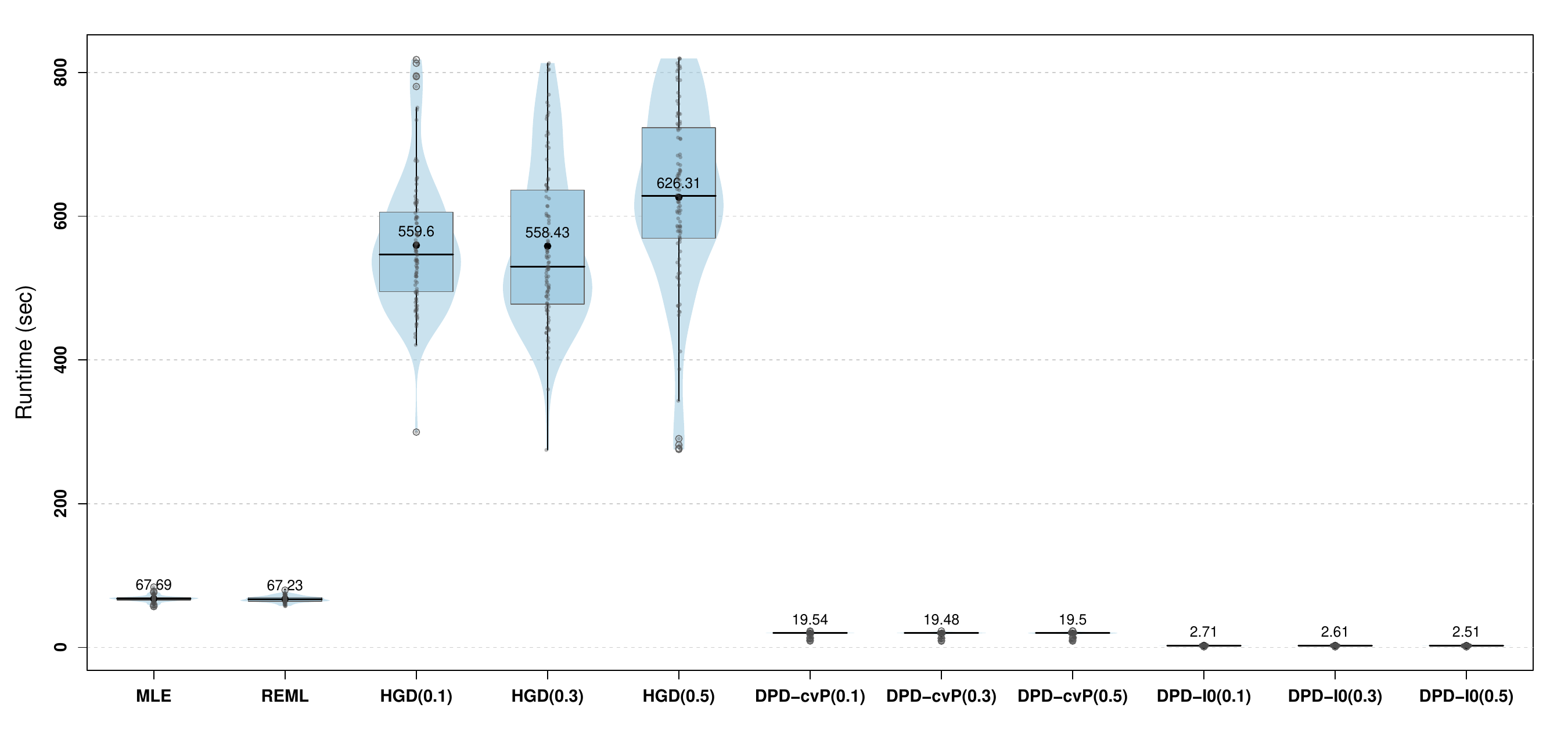}
		\label{FIG:boxplot_Y}}
	\\	
	\subfloat[\tiny 5\% (C3) Contamination]{
		\includegraphics[page=3, width=0.3\textwidth]{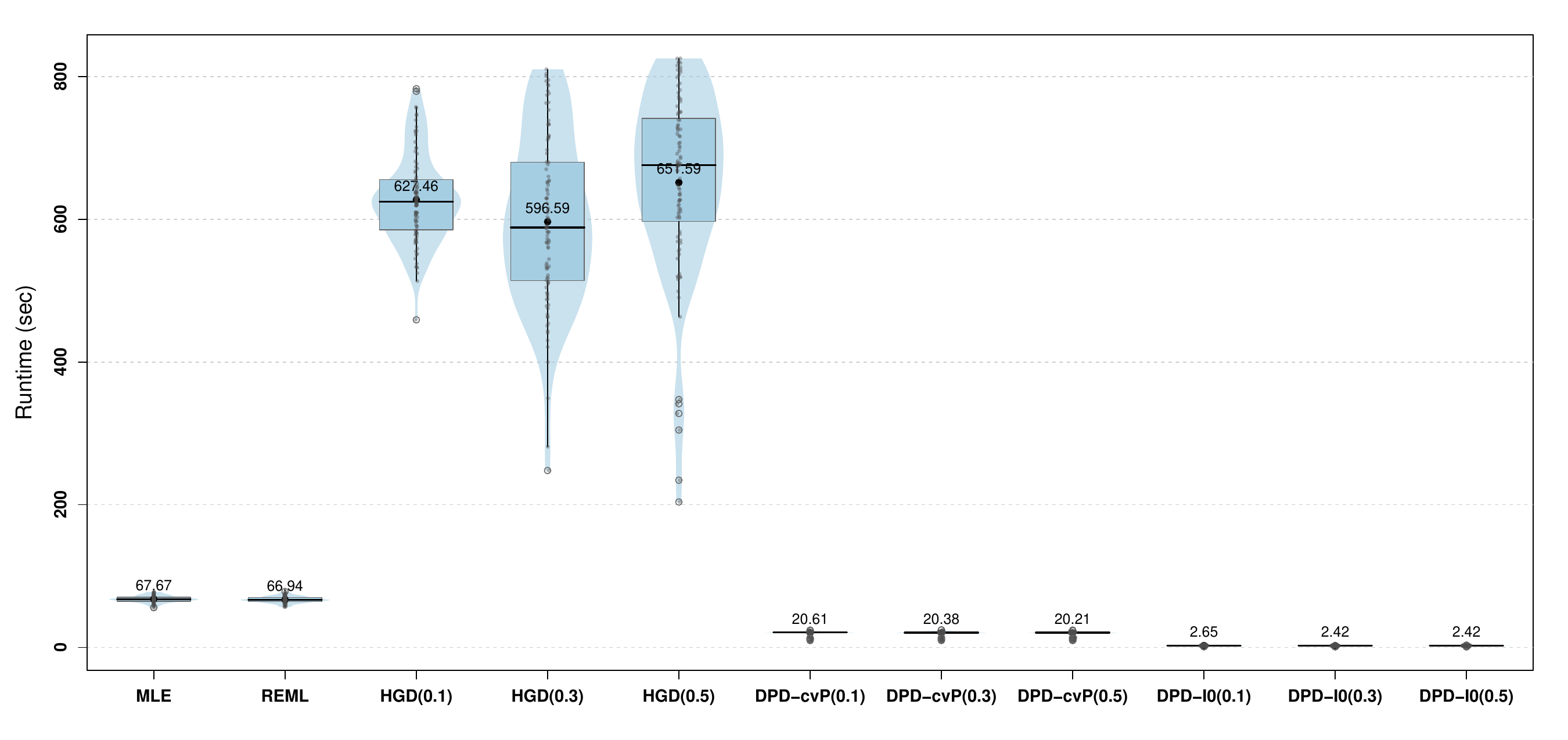}
		\label{FIG:boxplot_Y}}
	~	
	\subfloat[\tiny 10\% (C3) Contamination]{
		\includegraphics[page=3, width=0.3\textwidth]{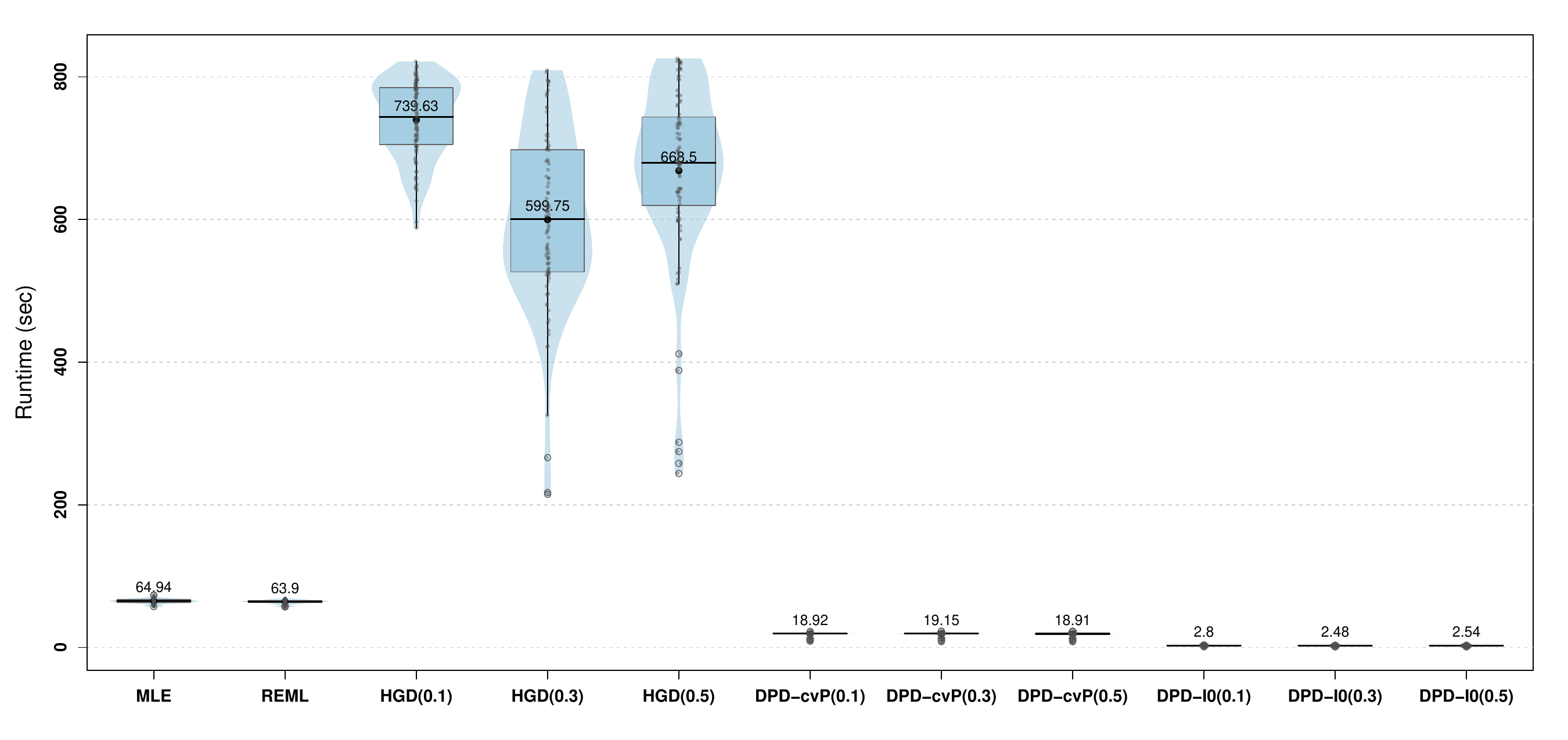}
		\label{FIG:boxplot_Y}}
	~	
	\subfloat[\tiny 20\% (C3) Contamination]{
		\includegraphics[page=3, width=0.3\textwidth]{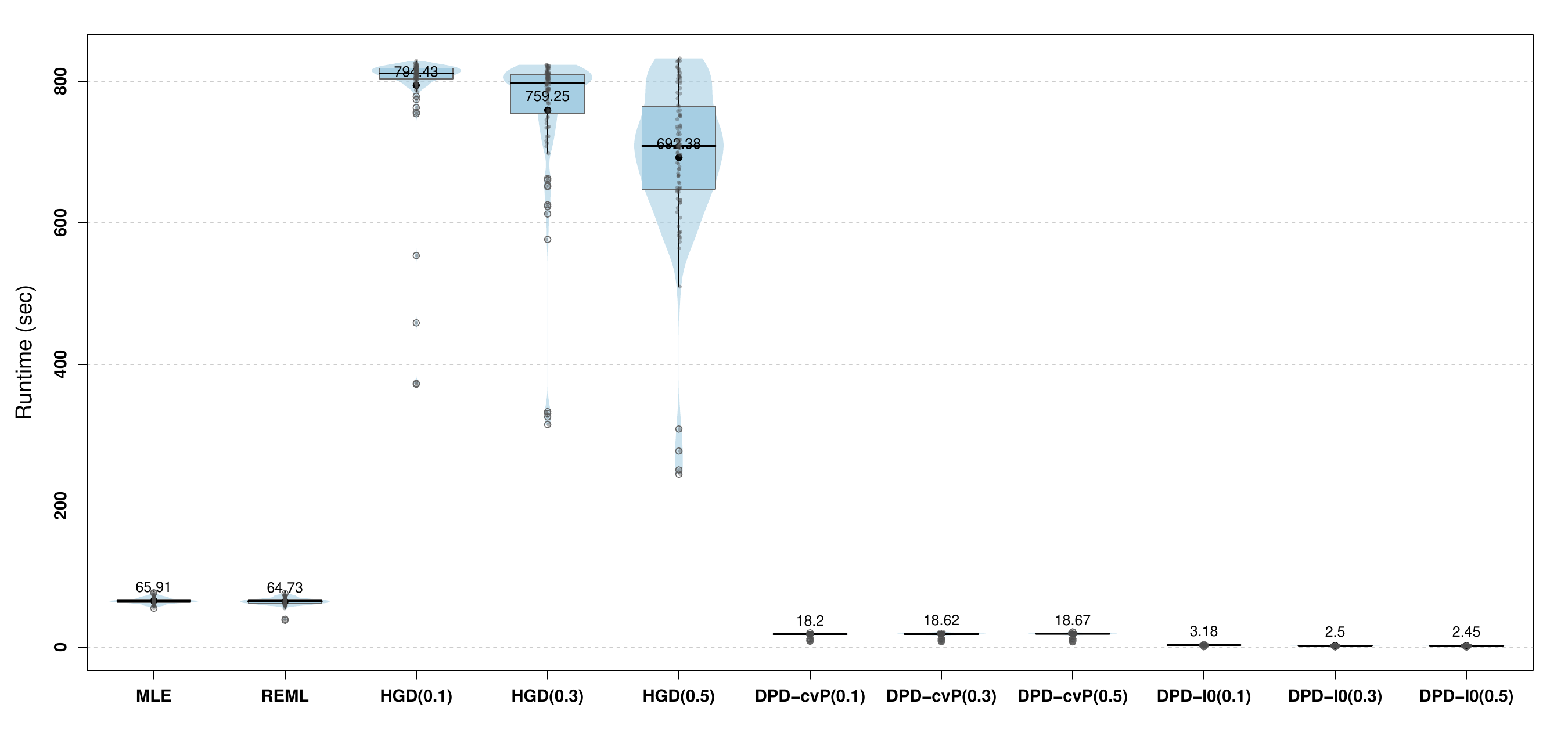}
		\label{FIG:boxplot_Y}}
	\\	
	\subfloat[\tiny 5\% (C4) Contamination]{
		\includegraphics[page=3, width=0.3\textwidth]{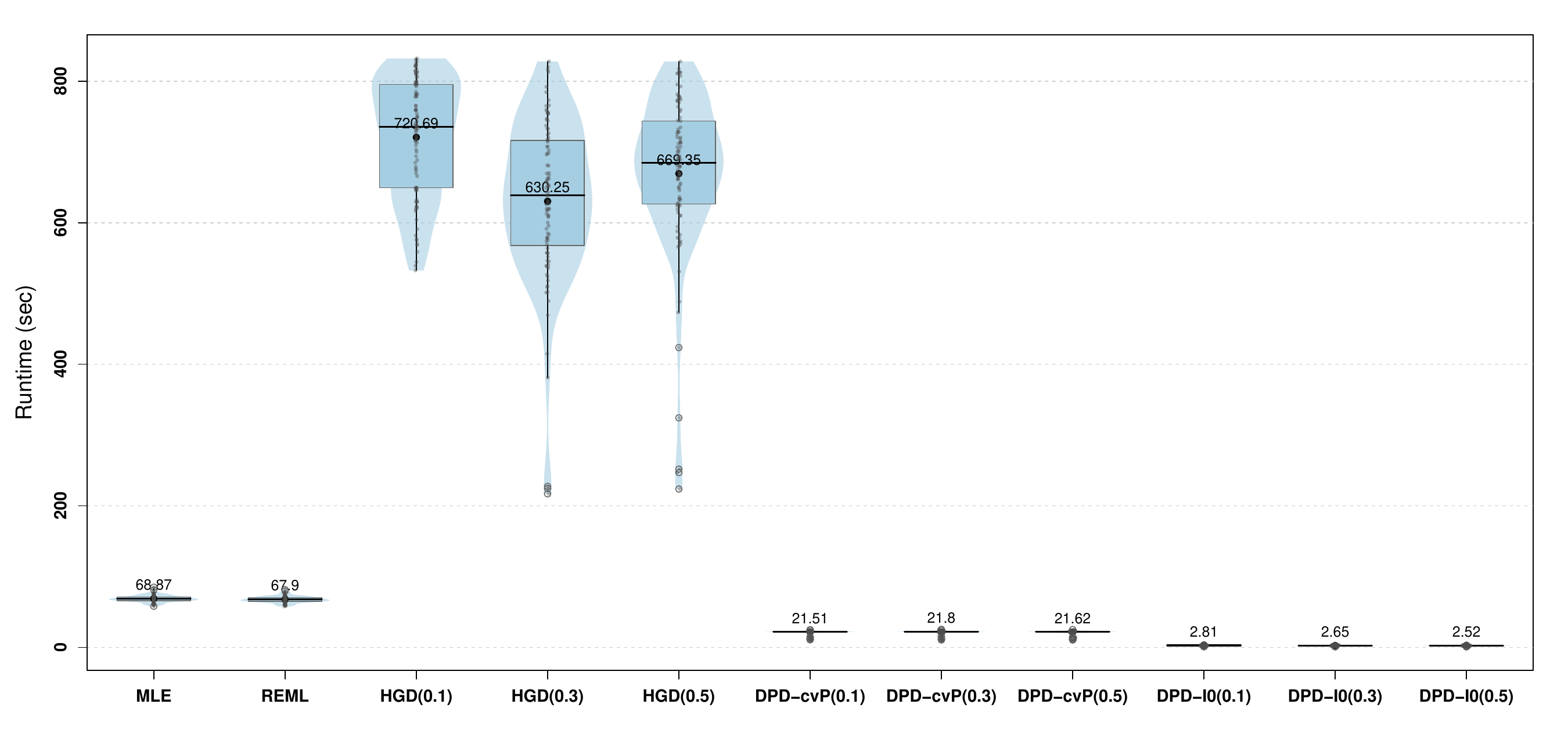}
		\label{FIG:boxplot_Y}}
	~	
	\subfloat[\tiny 10\% (C4) Contamination]{
		\includegraphics[page=3, width=0.3\textwidth]{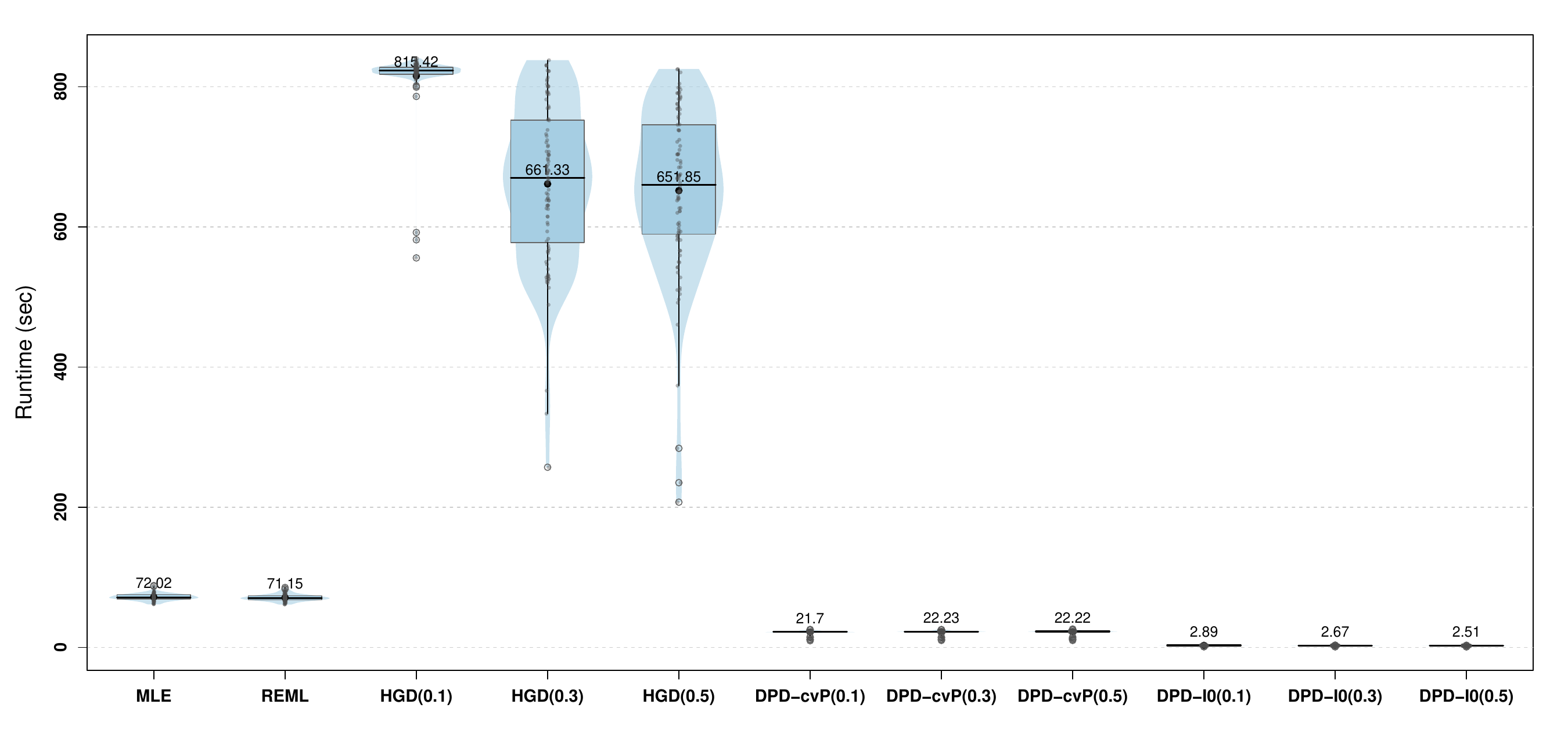}
		\label{FIG:boxplot_Y}}
	~	
	\subfloat[\tiny 20\% (C4) Contamination]{
		\includegraphics[page=3, width=0.3\textwidth]{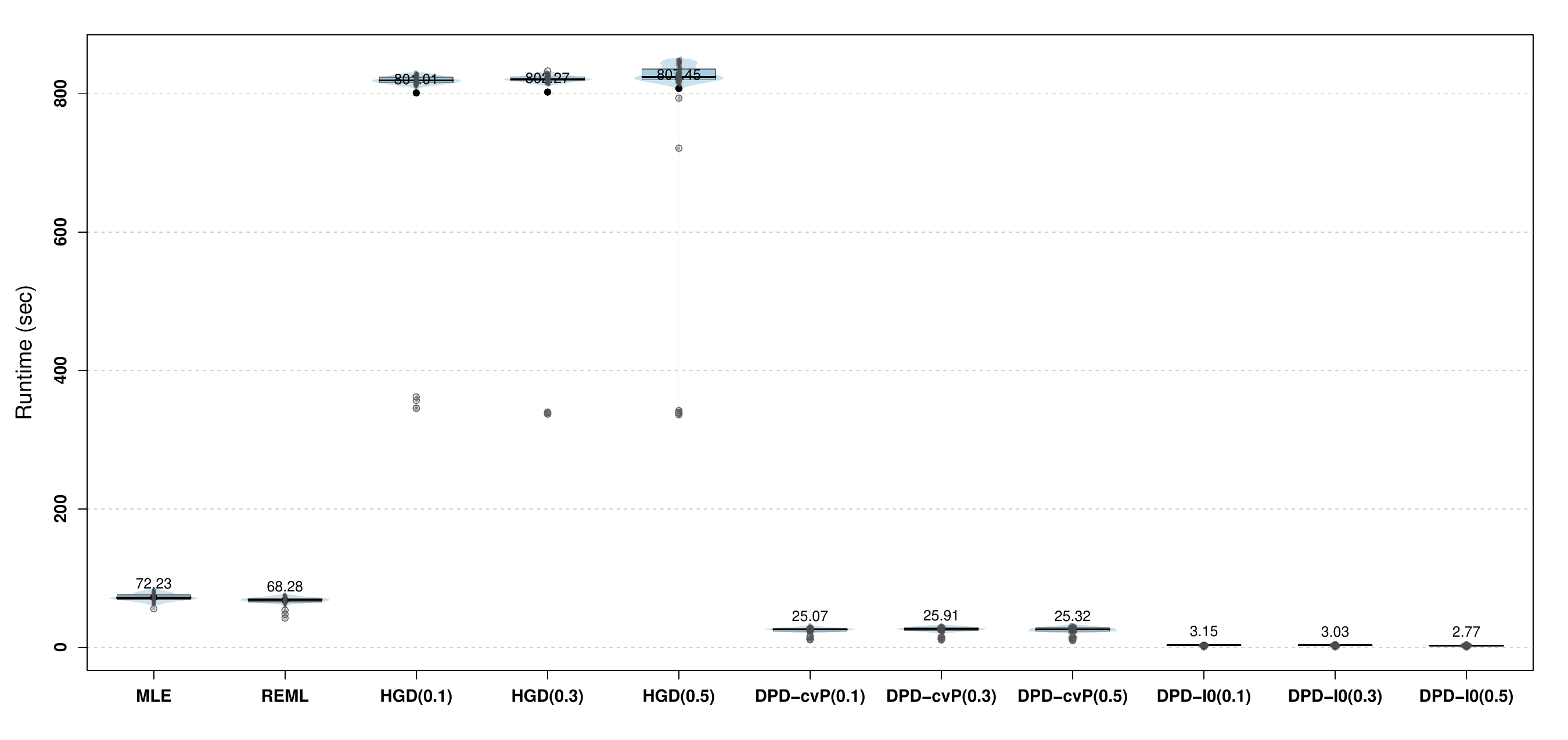}
		\label{FIG:boxplot_Y}}
	\caption{Boxplots  (with overlaid sample means) of computational times (in seconds) of the screening methods  under scenario (S2)$\times$(R0) with $\bm\Sigma_x$ as CS(0.3) and different types of contamination}
	\label{FIG:time_Csp3}
\end{figure}

\begin{figure}[!h]
	\centering
	\subfloat[\tiny 5\% (C1) Contamination]{
		\includegraphics[page=3, width=0.3\textwidth]{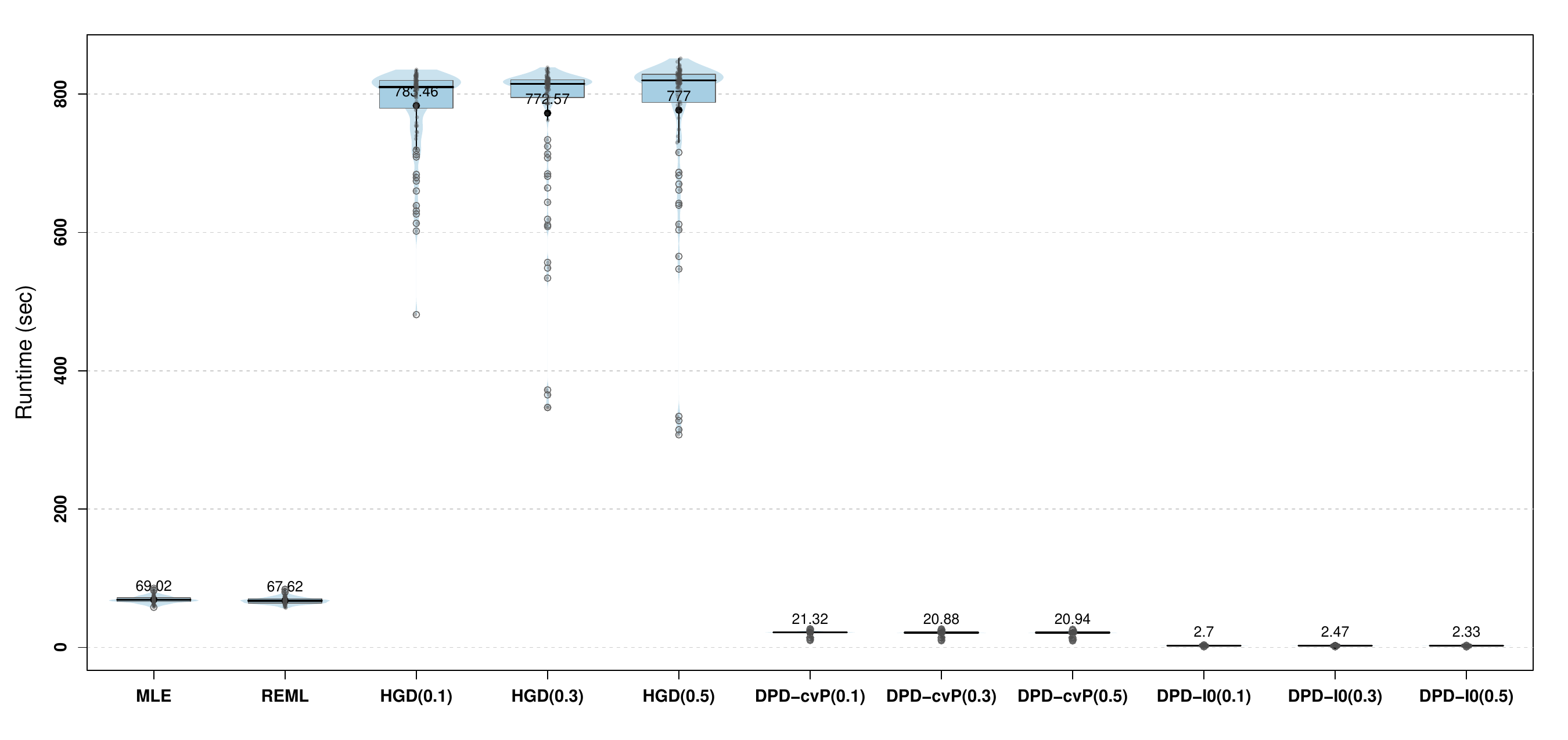}
		\label{FIG:boxplot_Y}}
	~	
	\subfloat[\tiny 10\% (C1) Contamination]{
		\includegraphics[page=3, width=0.3\textwidth]{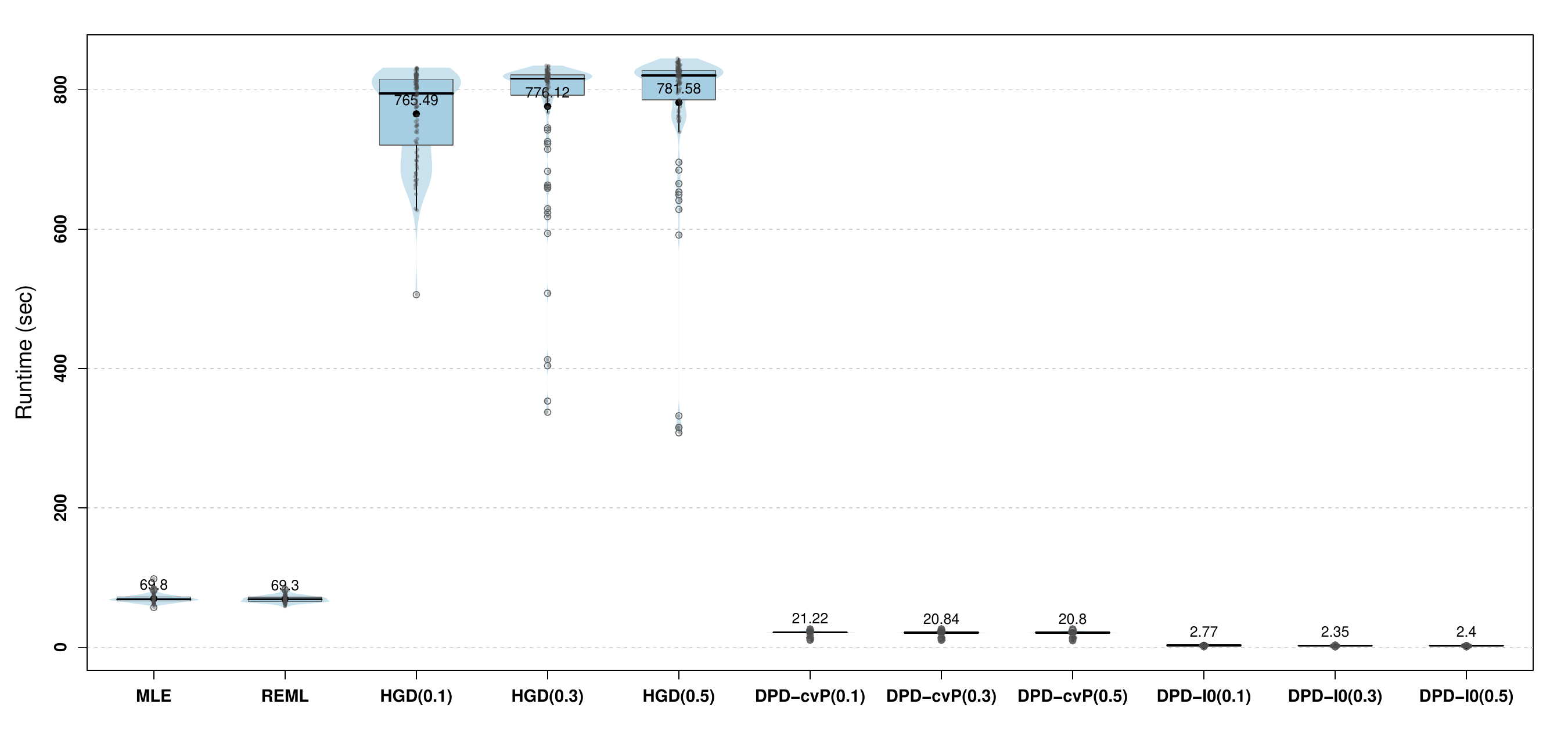}
		\label{FIG:boxplot_Y}}
	~	
	\subfloat[\tiny 20\% (C1) Contamination]{
		\includegraphics[page=3, width=0.3\textwidth]{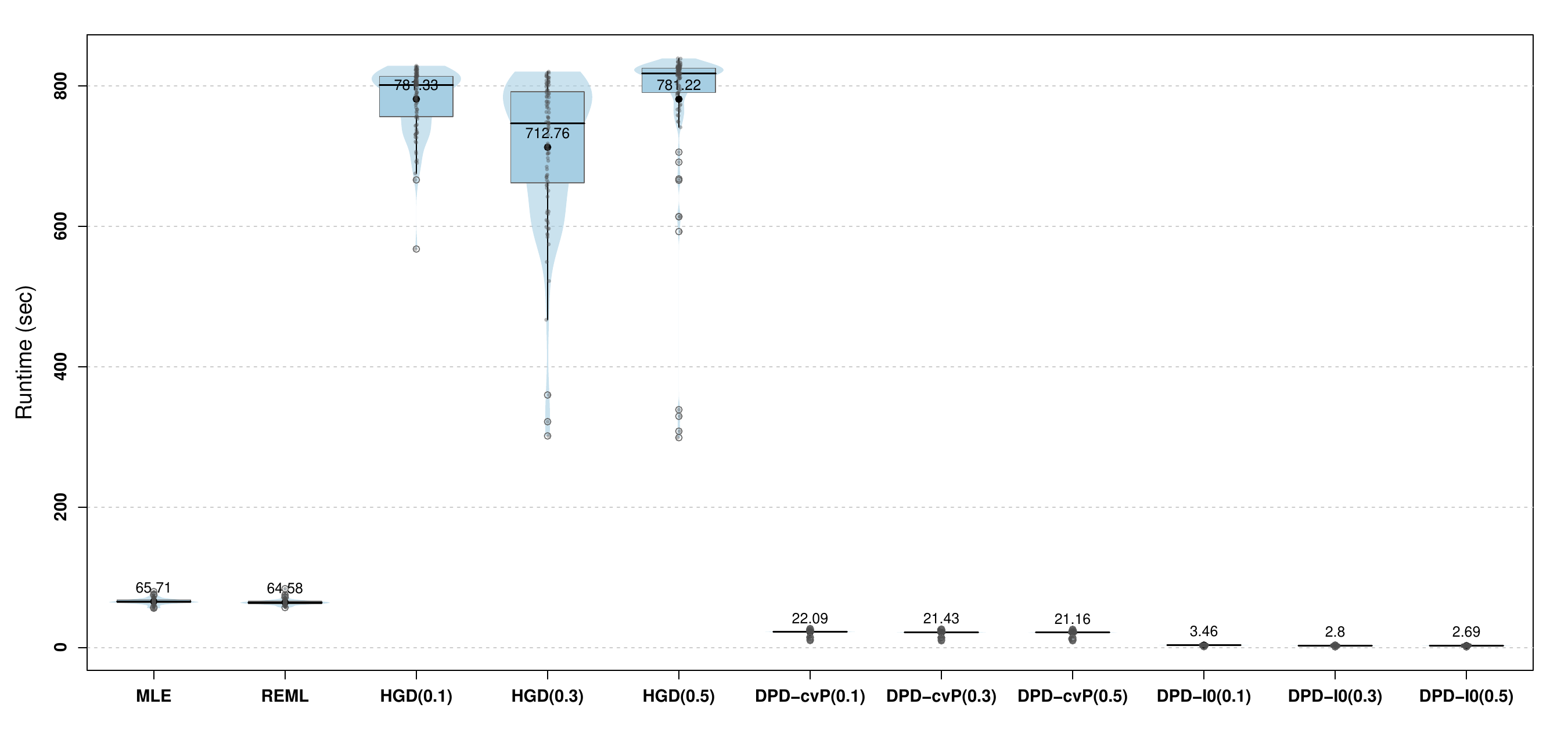}
		\label{FIG:boxplot_Y}}
	\\	
	\subfloat[\tiny 5\% (C2) Contamination]{
		\includegraphics[page=3, width=0.3\textwidth]{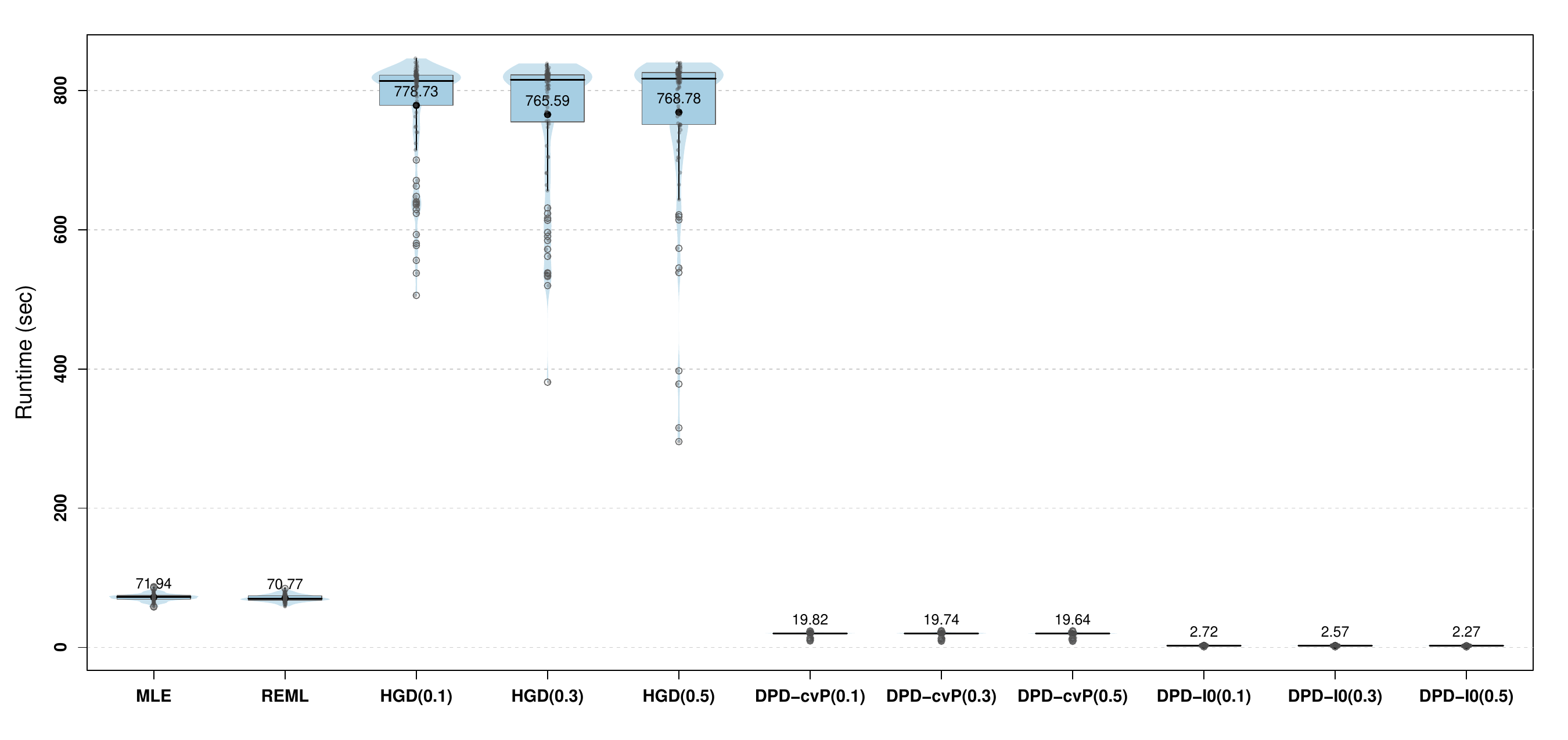}
		\label{FIG:boxplot_Y}}
	~	
	\subfloat[\tiny 10\% (C2) Contamination]{
		\includegraphics[page=3, width=0.3\textwidth]{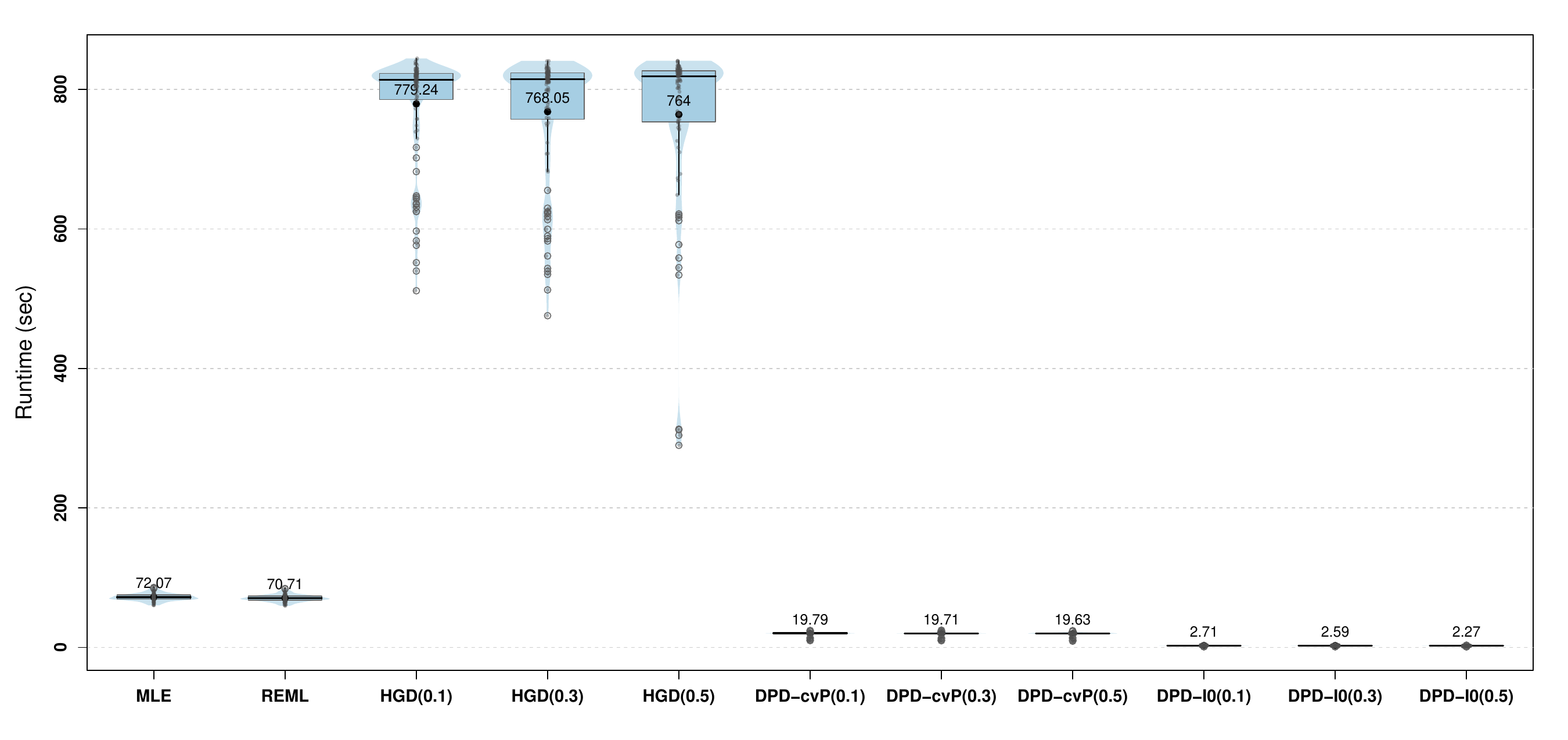}
		\label{FIG:boxplot_Y}}
	~	
	\subfloat[\tiny 20\% (C2) Contamination]{
		\includegraphics[page=3, width=0.3\textwidth]{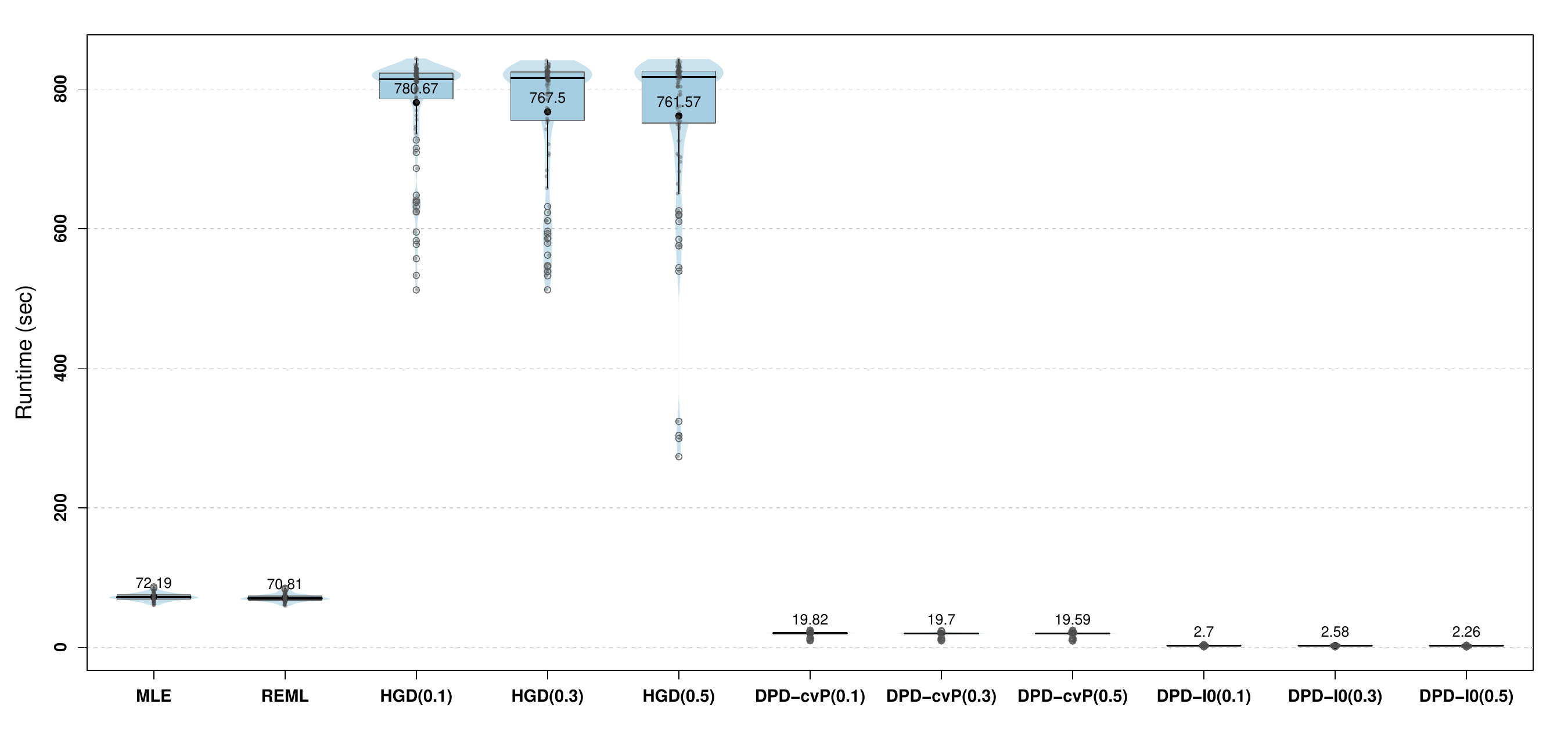}
		\label{FIG:boxplot_Y}}
	\\	
	\subfloat[\tiny 5\% (C3) Contamination]{
		\includegraphics[page=3, width=0.3\textwidth]{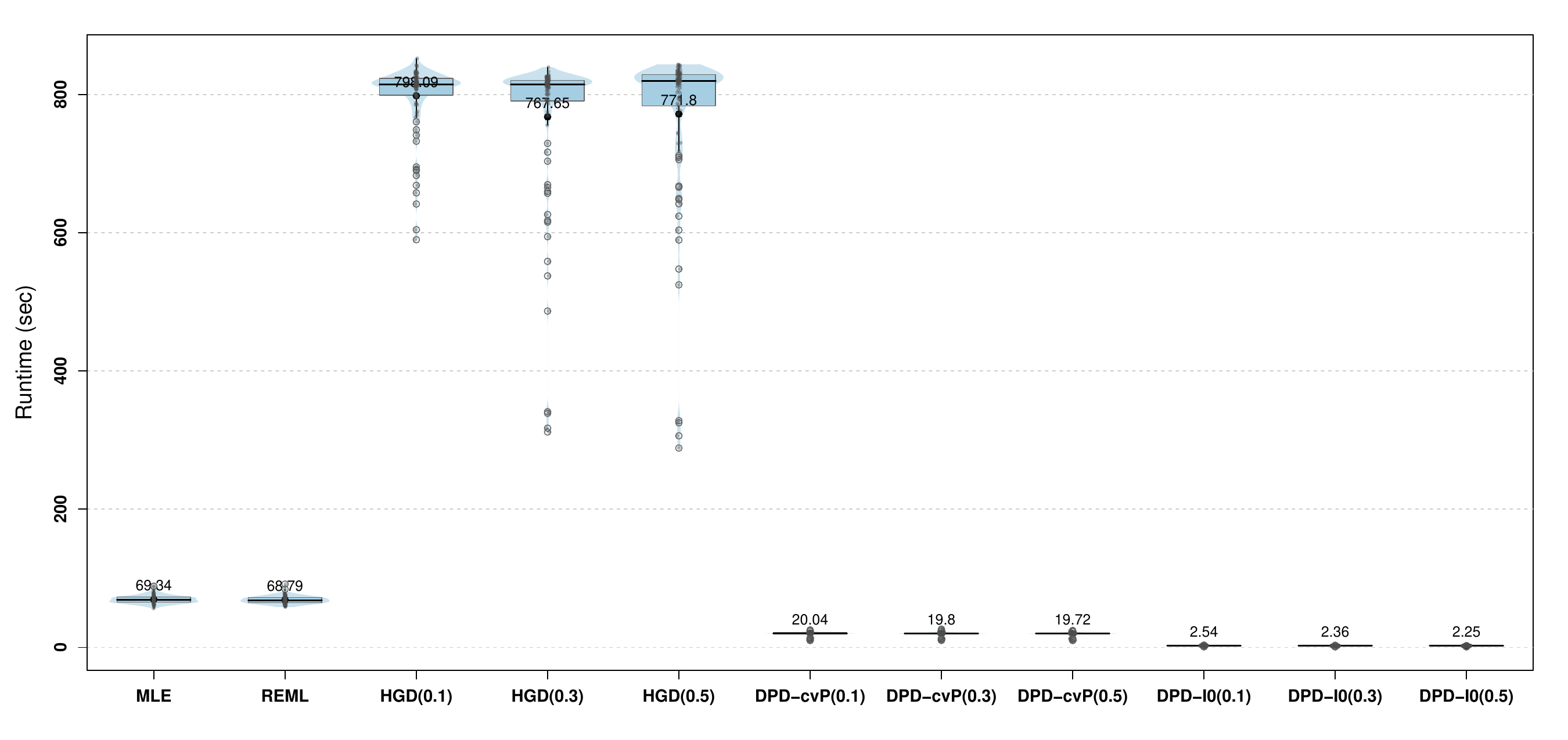}
		\label{FIG:boxplot_Y}}
	~	
	\subfloat[\tiny 10\% (C3) Contamination]{
		\includegraphics[page=3, width=0.3\textwidth]{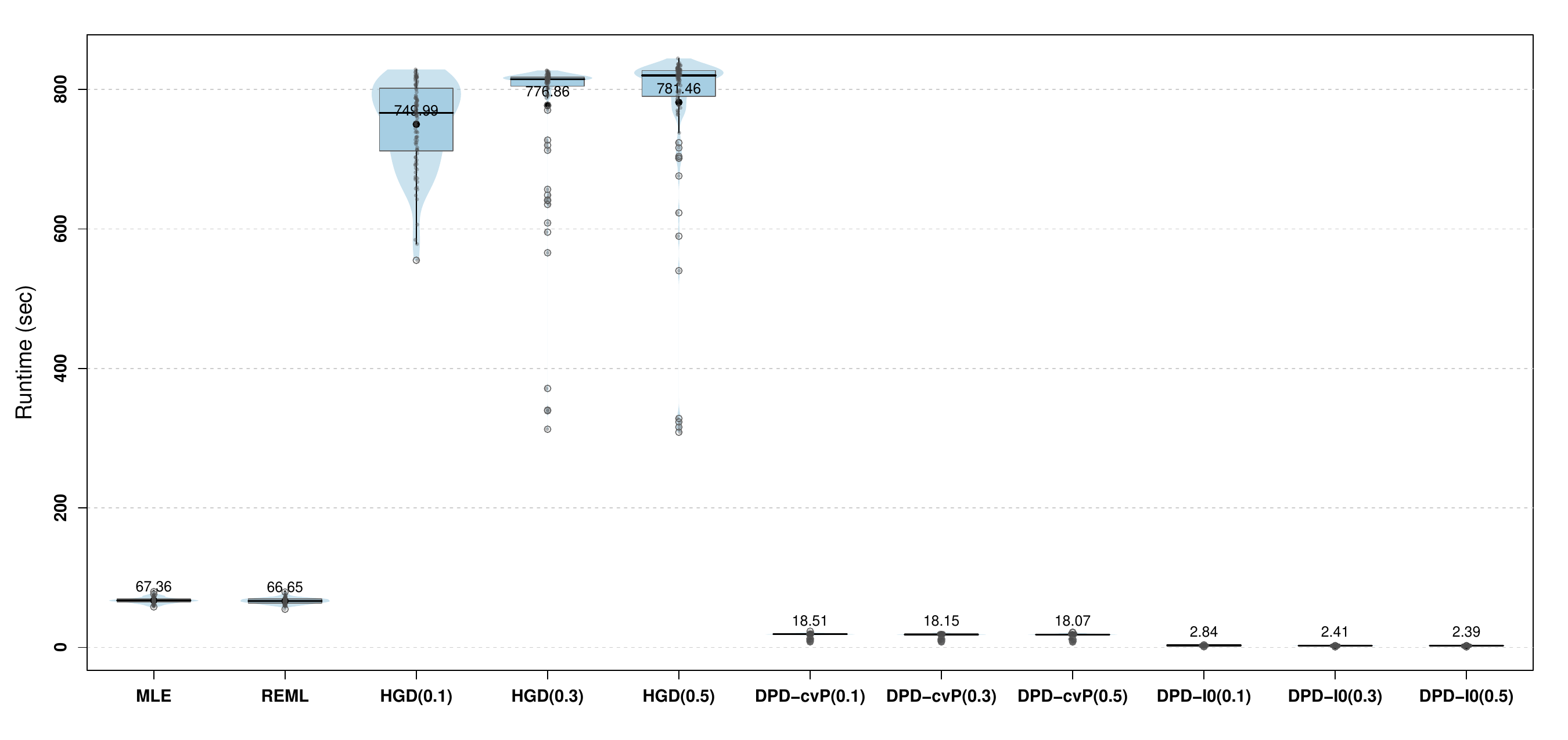}
		\label{FIG:boxplot_Y}}
	~	
	\subfloat[\tiny 20\% (C3) Contamination]{
		\includegraphics[page=3, width=0.3\textwidth]{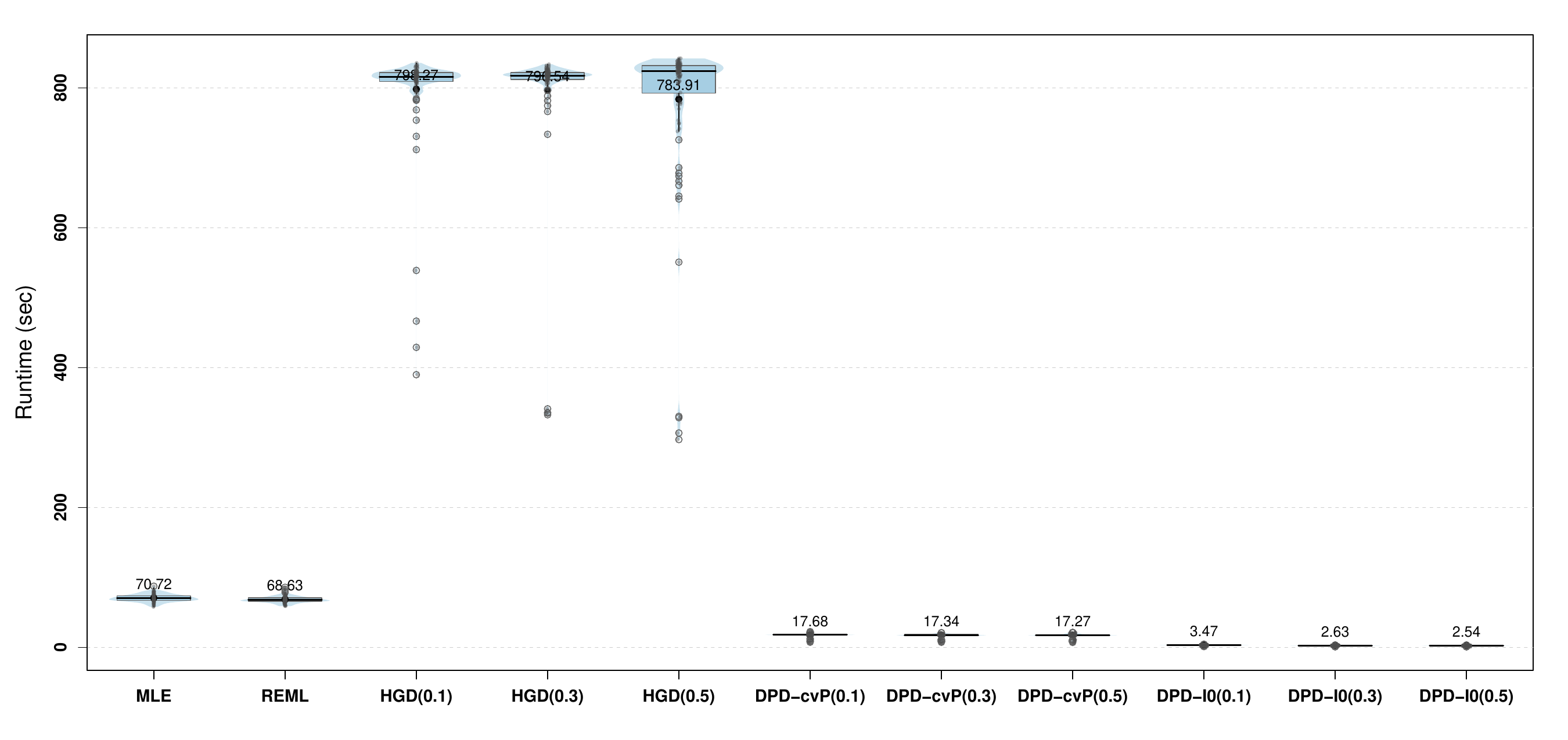}
		\label{FIG:boxplot_Y}}
	\\	
	\subfloat[\tiny 5\% (C4) Contamination]{
		\includegraphics[page=3, width=0.3\textwidth]{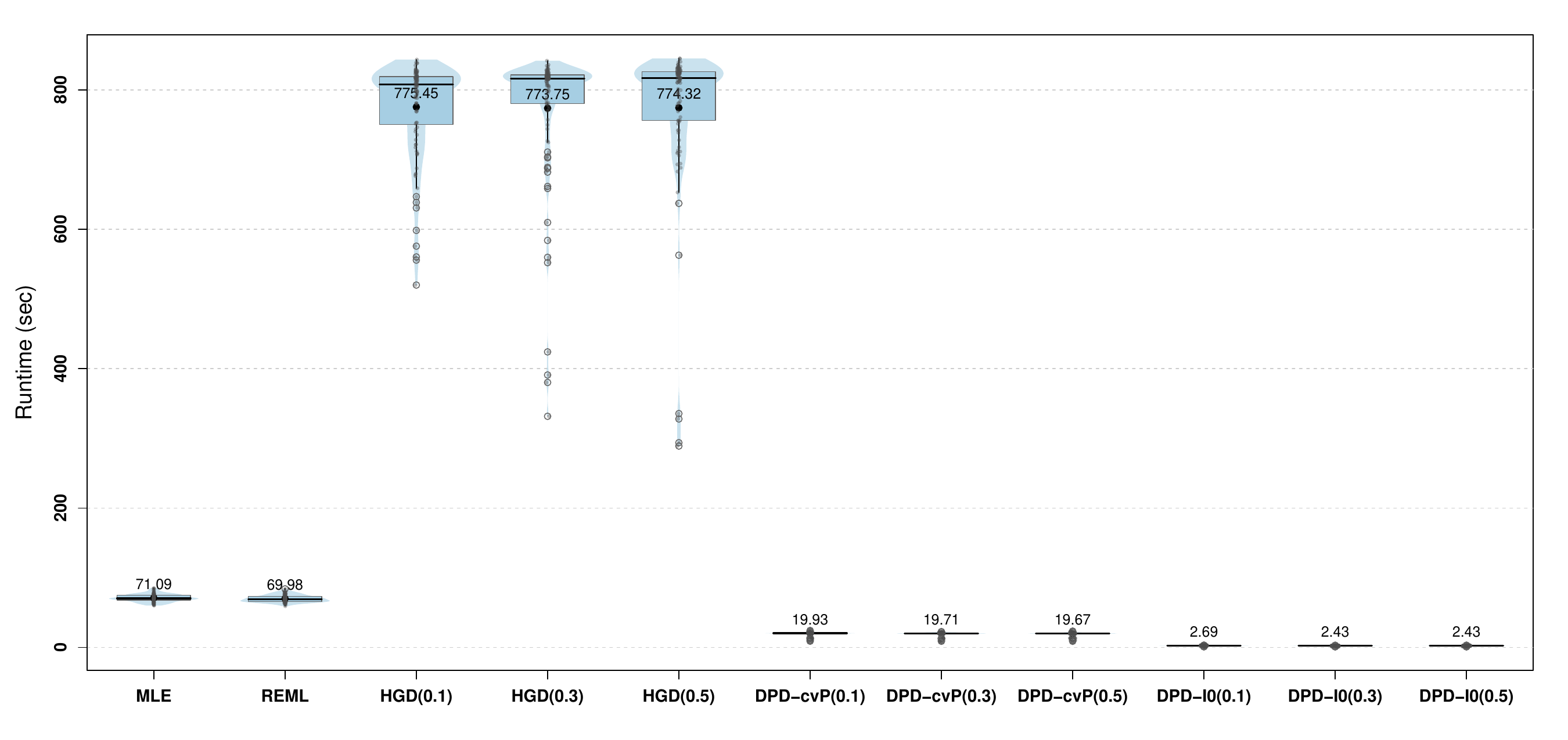}
		\label{FIG:boxplot_Y}}
	~	
	\subfloat[\tiny 10\% (C4) Contamination]{
		\includegraphics[page=3, width=0.3\textwidth]{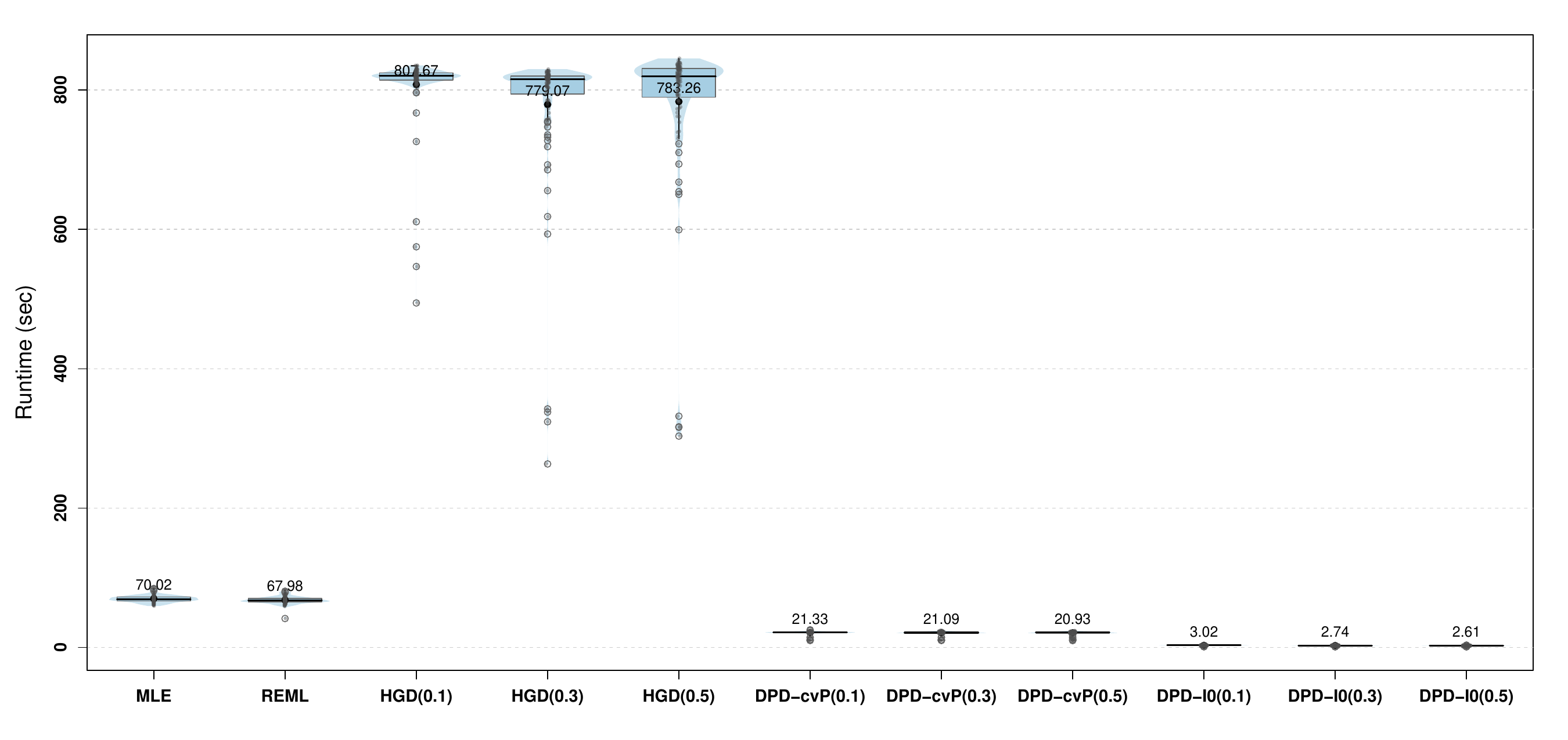}
		\label{FIG:boxplot_Y}}
	~	
	\subfloat[\tiny 20\% (C4) Contamination]{
		\includegraphics[page=3, width=0.3\textwidth]{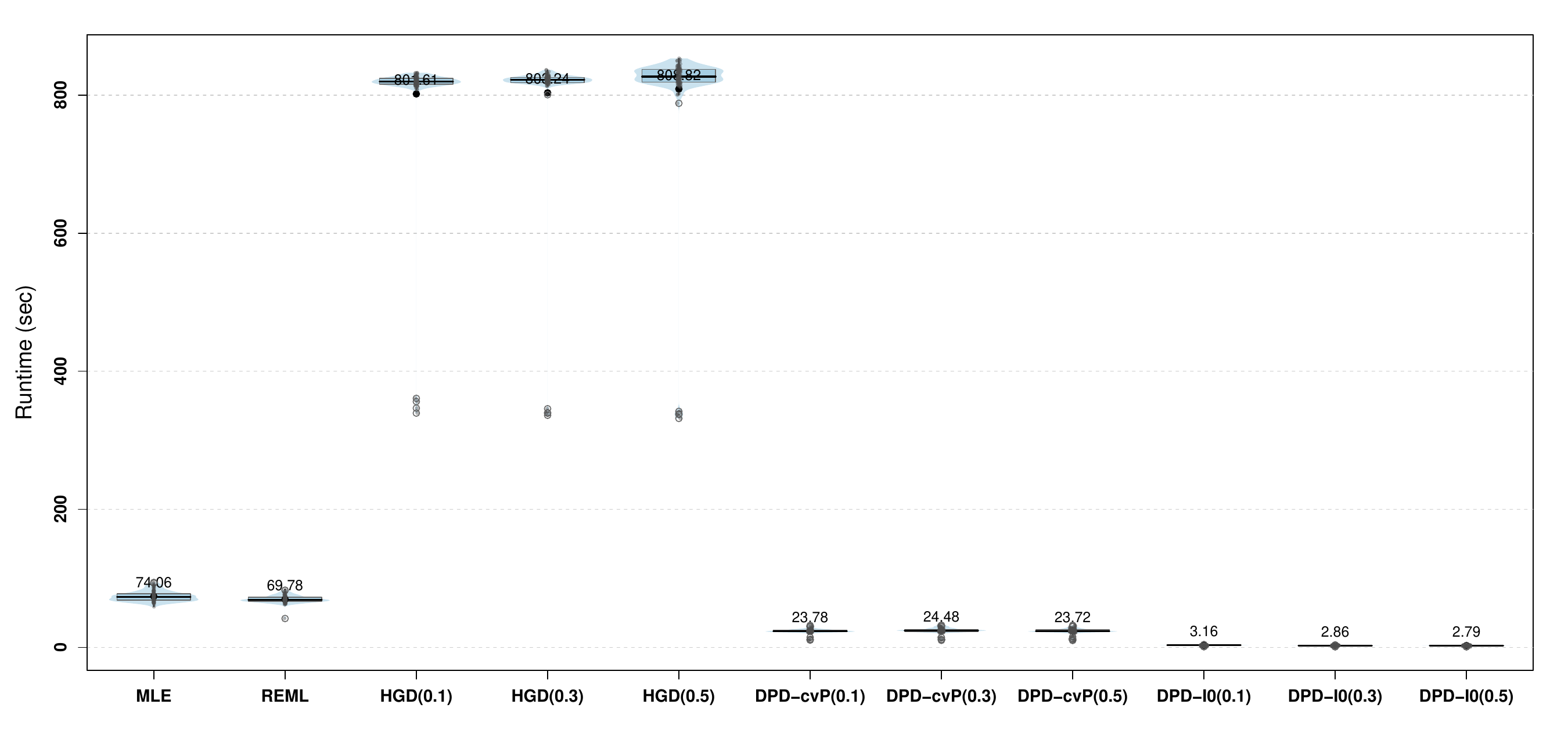}
		\label{FIG:boxplot_Y}}
	\caption{Boxplots  (with overlaid sample means) of computational times (in seconds) of the screening methods  under scenario (S2)$\times$(R1) with $\bm\Sigma_x$ as CS(0.3) and different types of contamination}
	\label{FIG:time_S2R1CS}
\end{figure}

\begin{figure}[!h]
	\centering
	\subfloat[\tiny 5\% (C1) Contamination]{
		\includegraphics[page=3, width=0.3\textwidth]{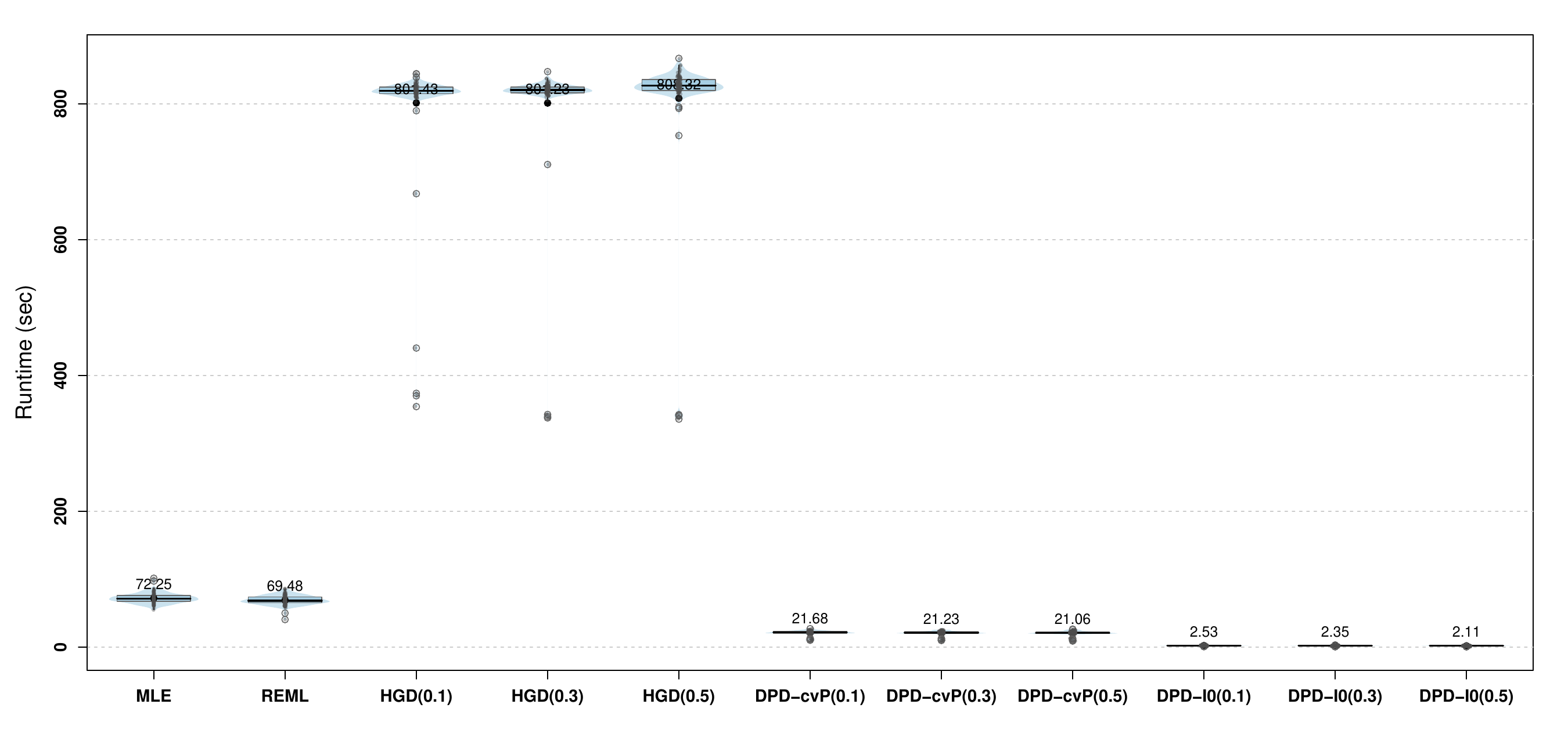}
		\label{FIG:boxplot_Y}}
	~	
	\subfloat[\tiny 10\% (C1) Contamination]{
		\includegraphics[page=3, width=0.3\textwidth]{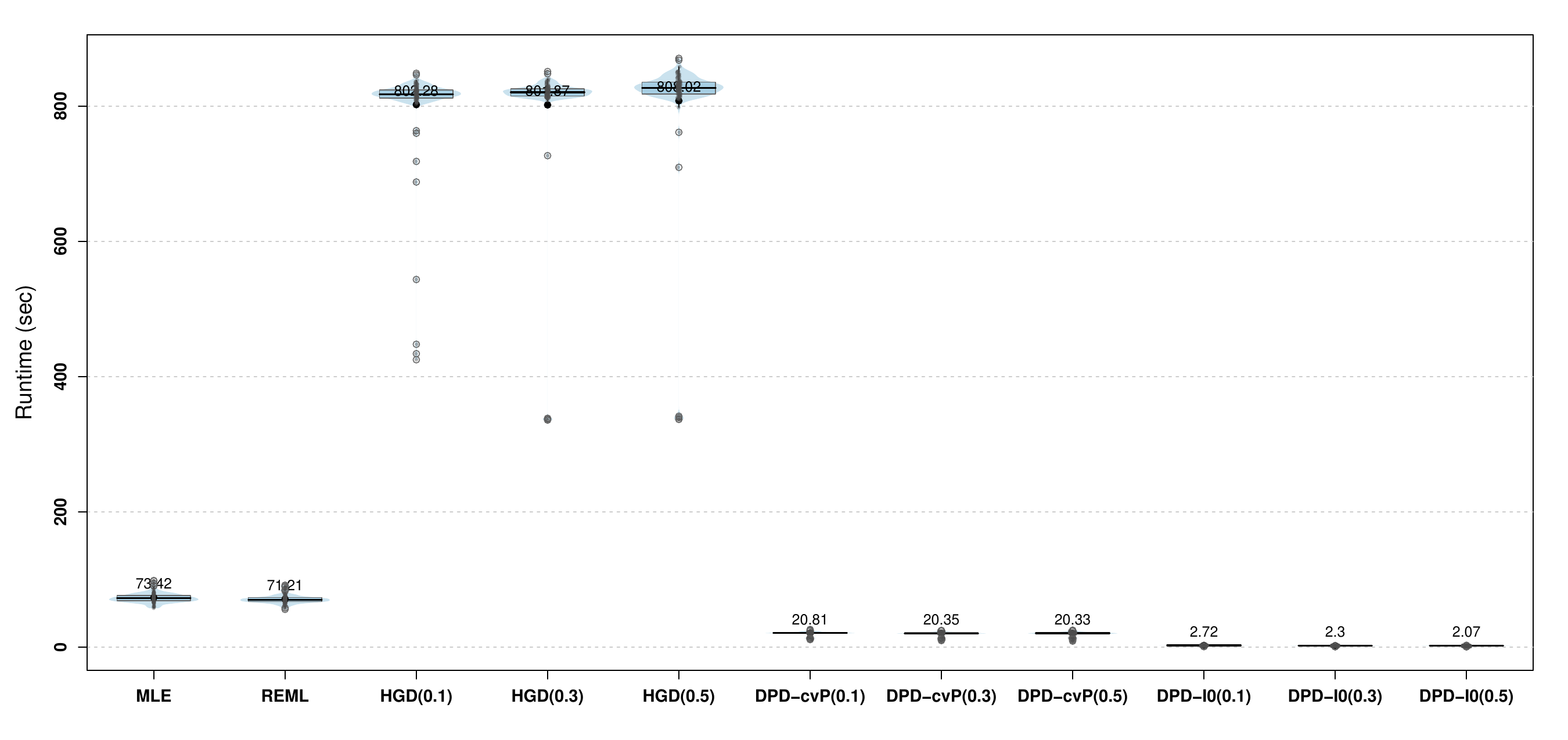}
		\label{FIG:boxplot_Y}}
	~	
	\subfloat[\tiny 20\% (C1) Contamination]{
		\includegraphics[page=3, width=0.3\textwidth]{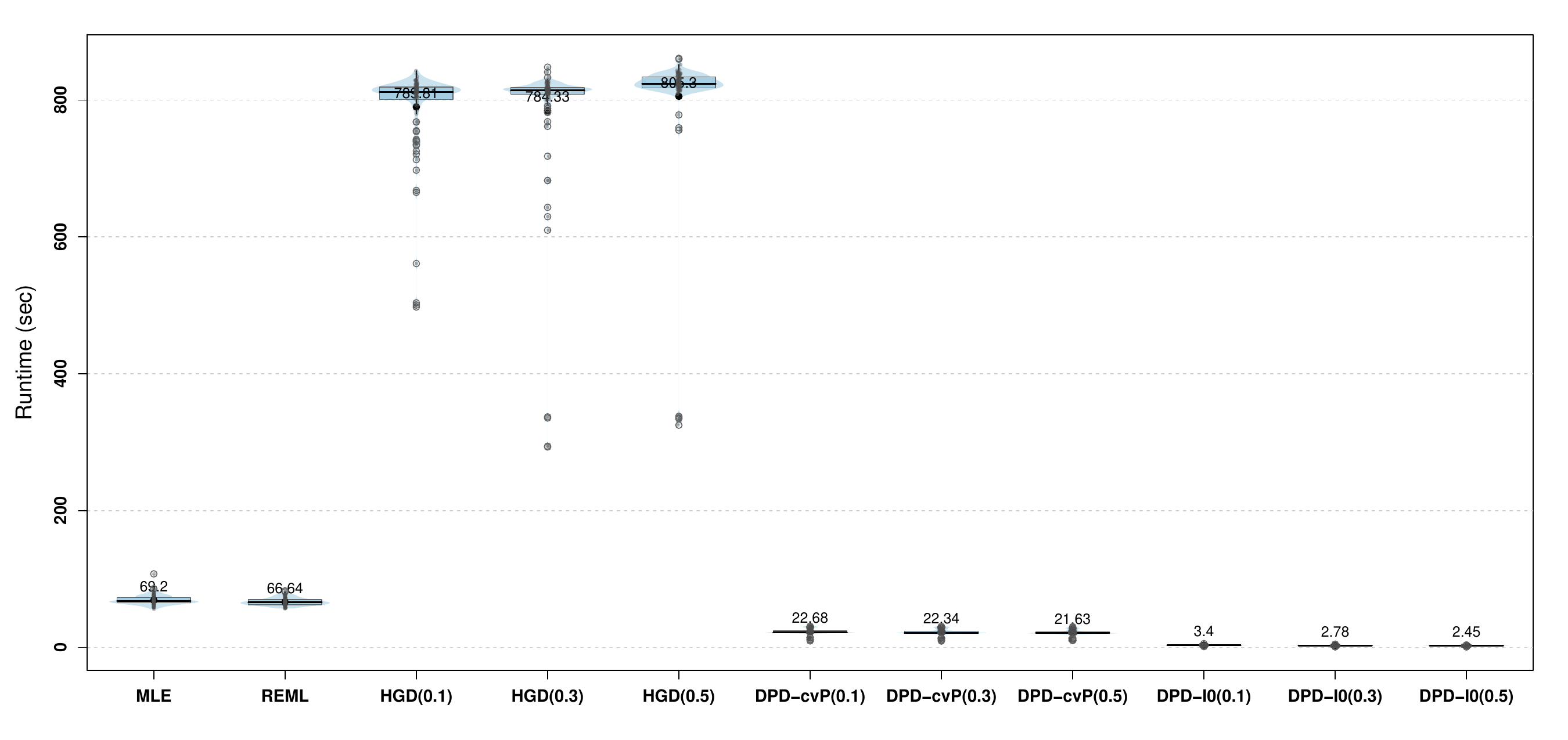}
		\label{FIG:boxplot_Y}}
	\\	
	\subfloat[\tiny 5\% (C2) Contamination]{
		\includegraphics[page=3, width=0.3\textwidth]{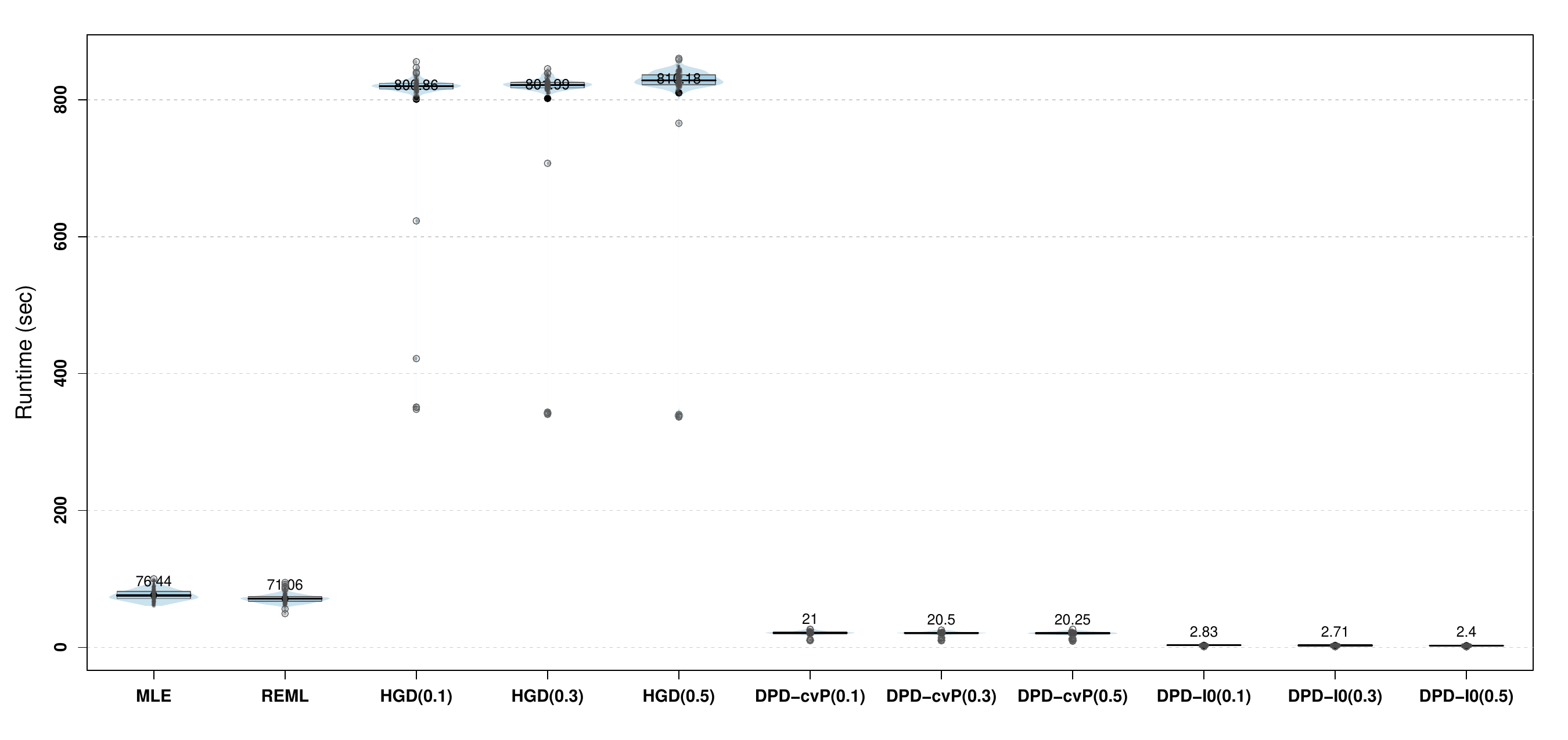}
		\label{FIG:boxplot_Y}}
	~	
	\subfloat[\tiny 10\% (C2) Contamination]{
		\includegraphics[page=3, width=0.3\textwidth]{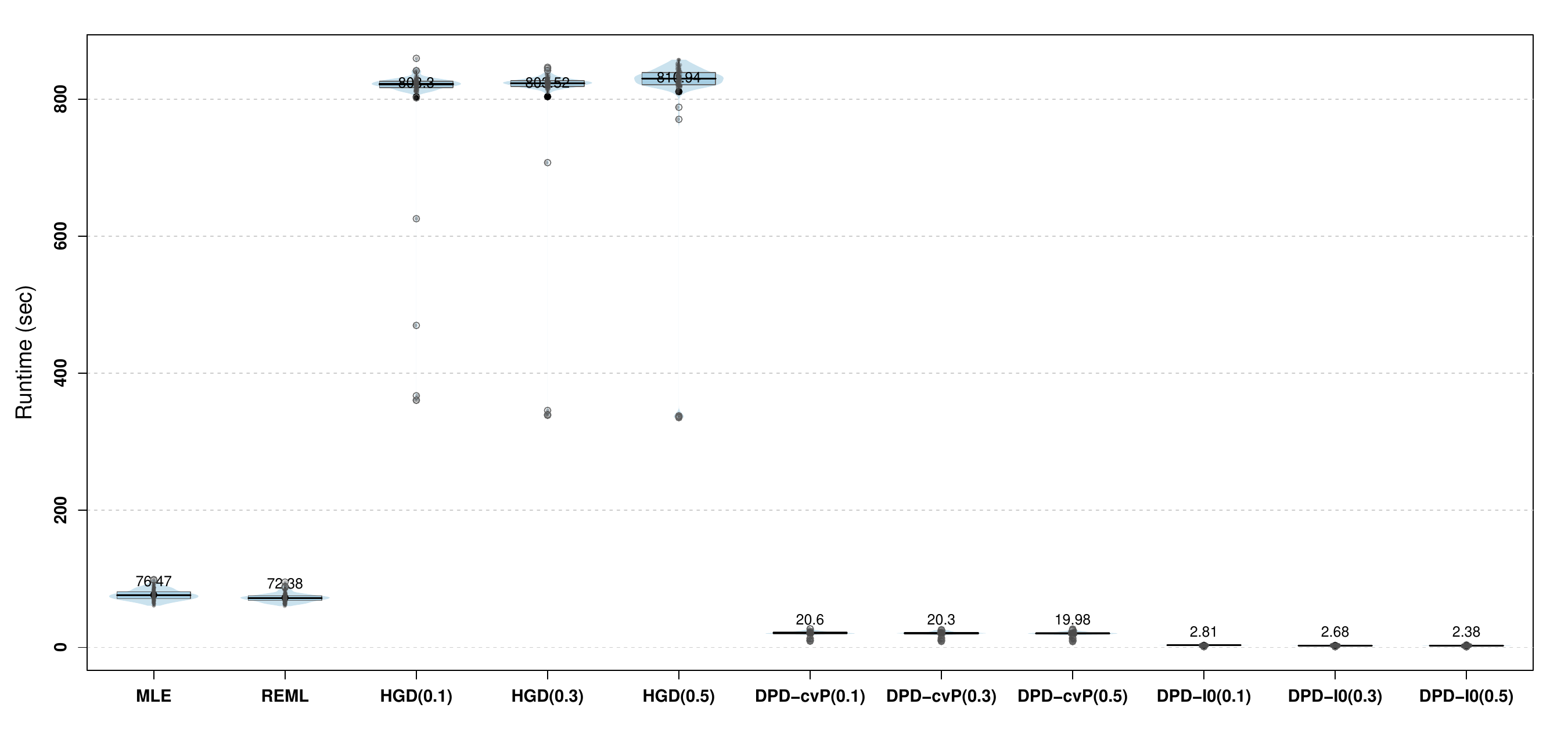}
		\label{FIG:boxplot_Y}}
	~	
	\subfloat[\tiny 20\% (C2) Contamination]{
		\includegraphics[page=3, width=0.3\textwidth]{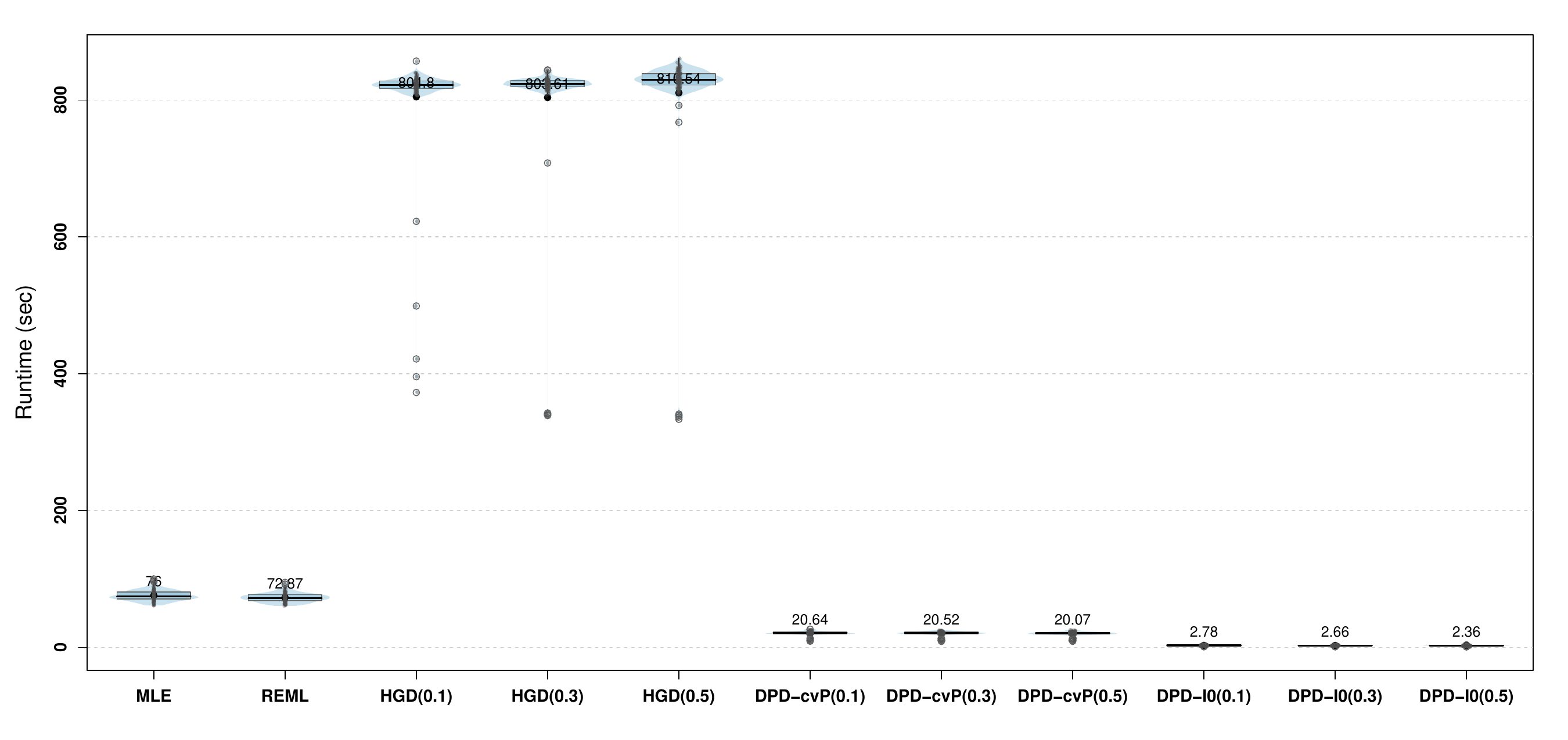}
		\label{FIG:boxplot_Y}}
	\\	
	\subfloat[\tiny 5\% (C3) Contamination]{
		\includegraphics[page=3, width=0.3\textwidth]{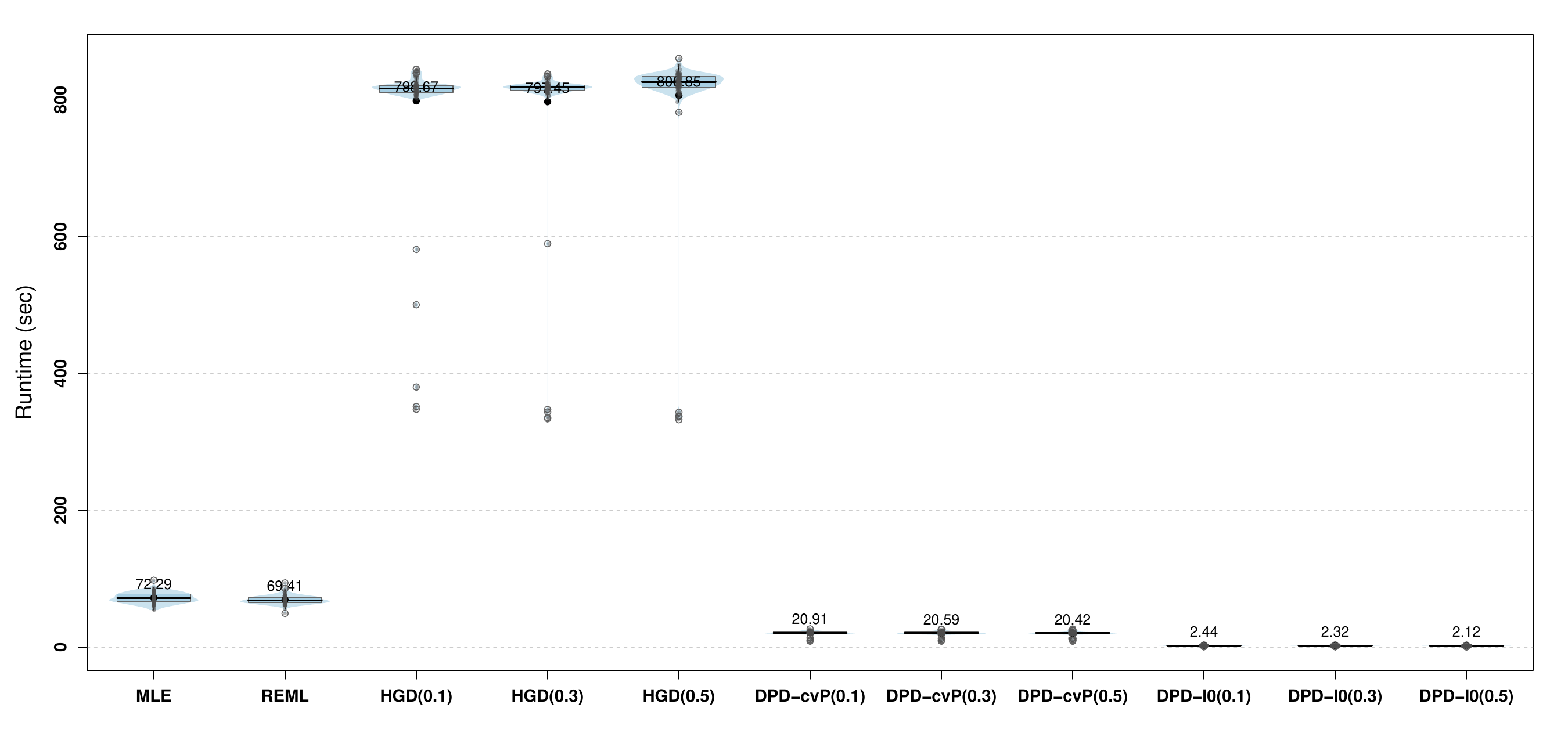}
		\label{FIG:boxplot_Y}}
	~	
	\subfloat[\tiny 10\% (C3) Contamination]{
		\includegraphics[page=3, width=0.3\textwidth]{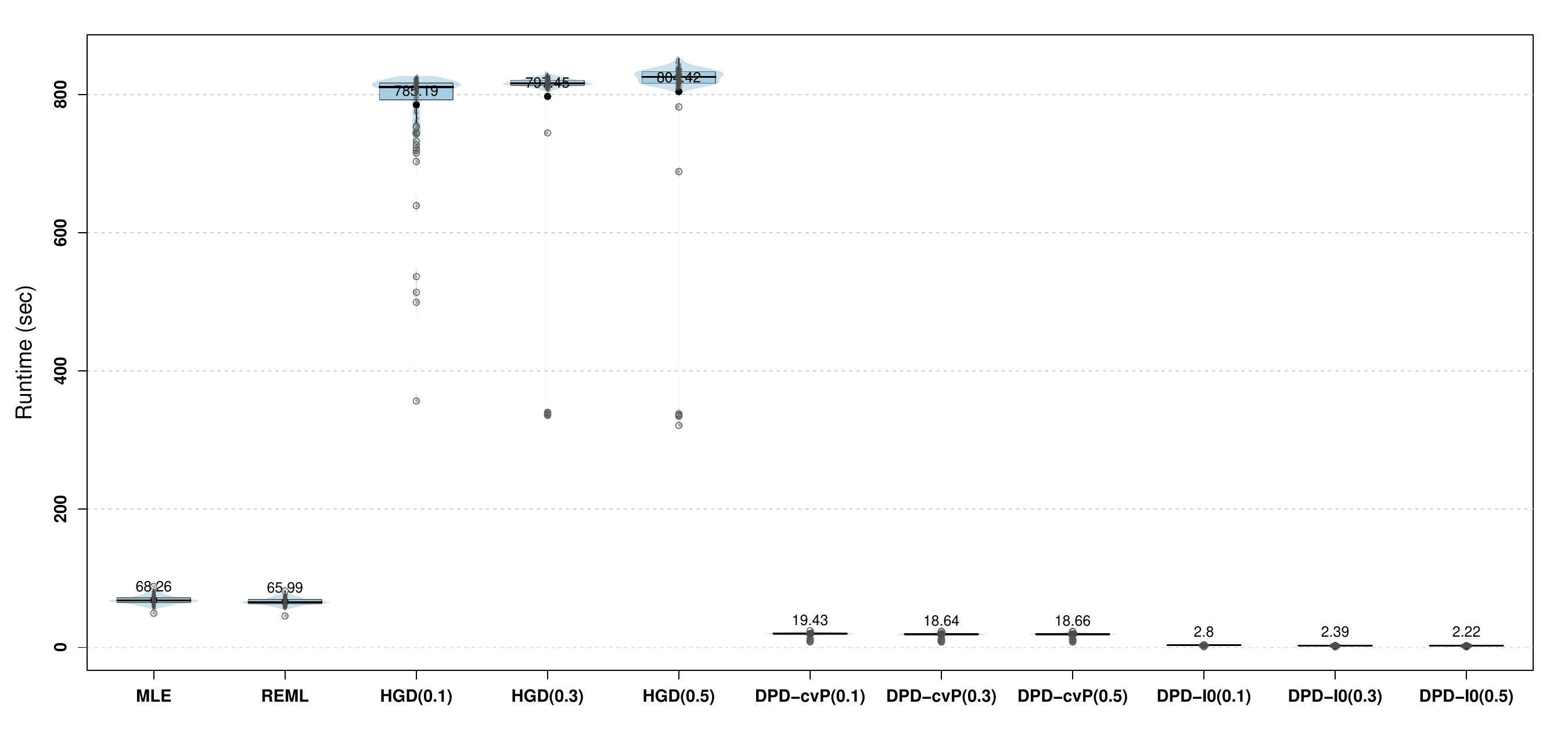}
		\label{FIG:boxplot_Y}}
	~	
	\subfloat[\tiny 20\% (C3) Contamination]{
		\includegraphics[page=3, width=0.3\textwidth]{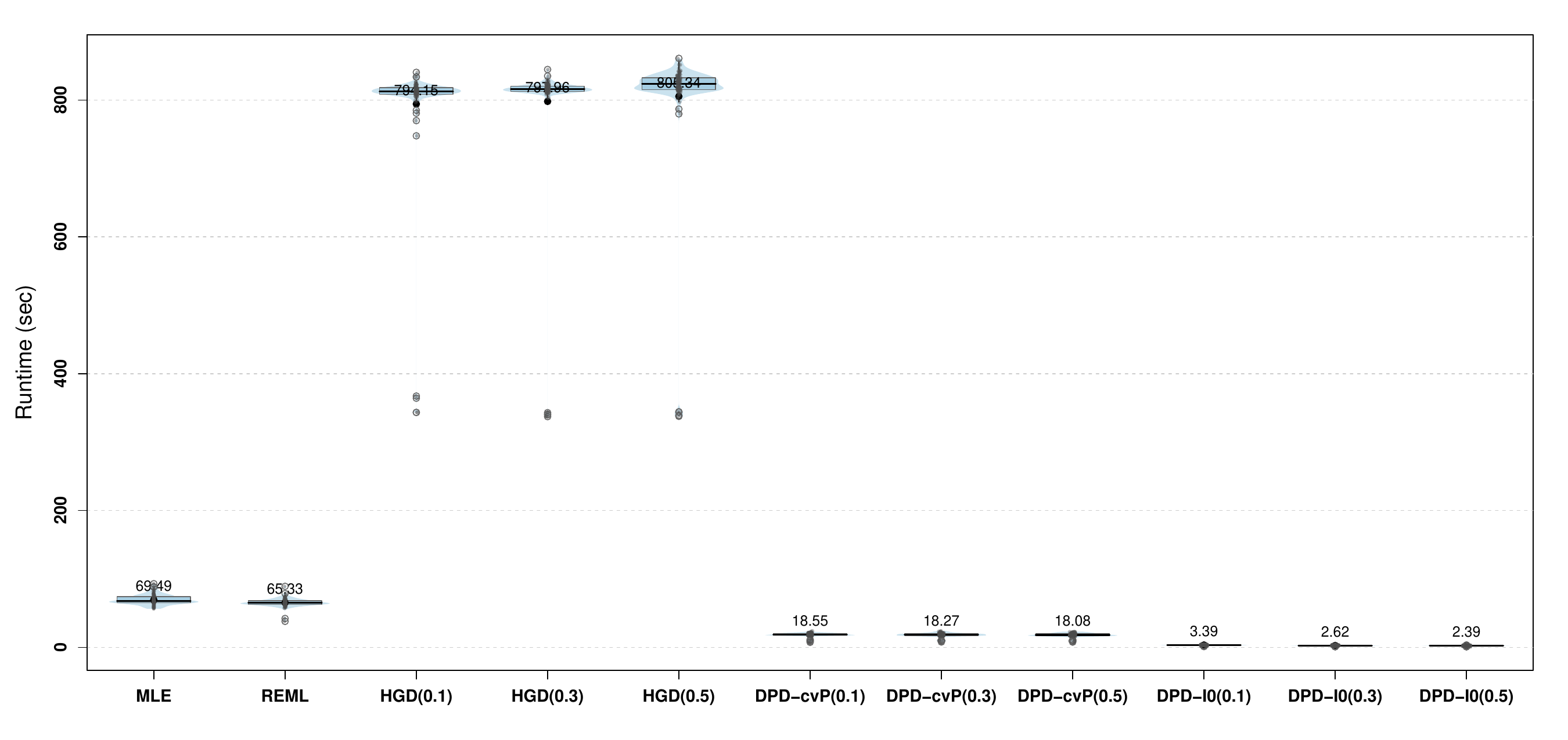}
		\label{FIG:boxplot_Y}}
	\\	
	\subfloat[\tiny 5\% (C4) Contamination]{
		\includegraphics[page=3, width=0.3\textwidth]{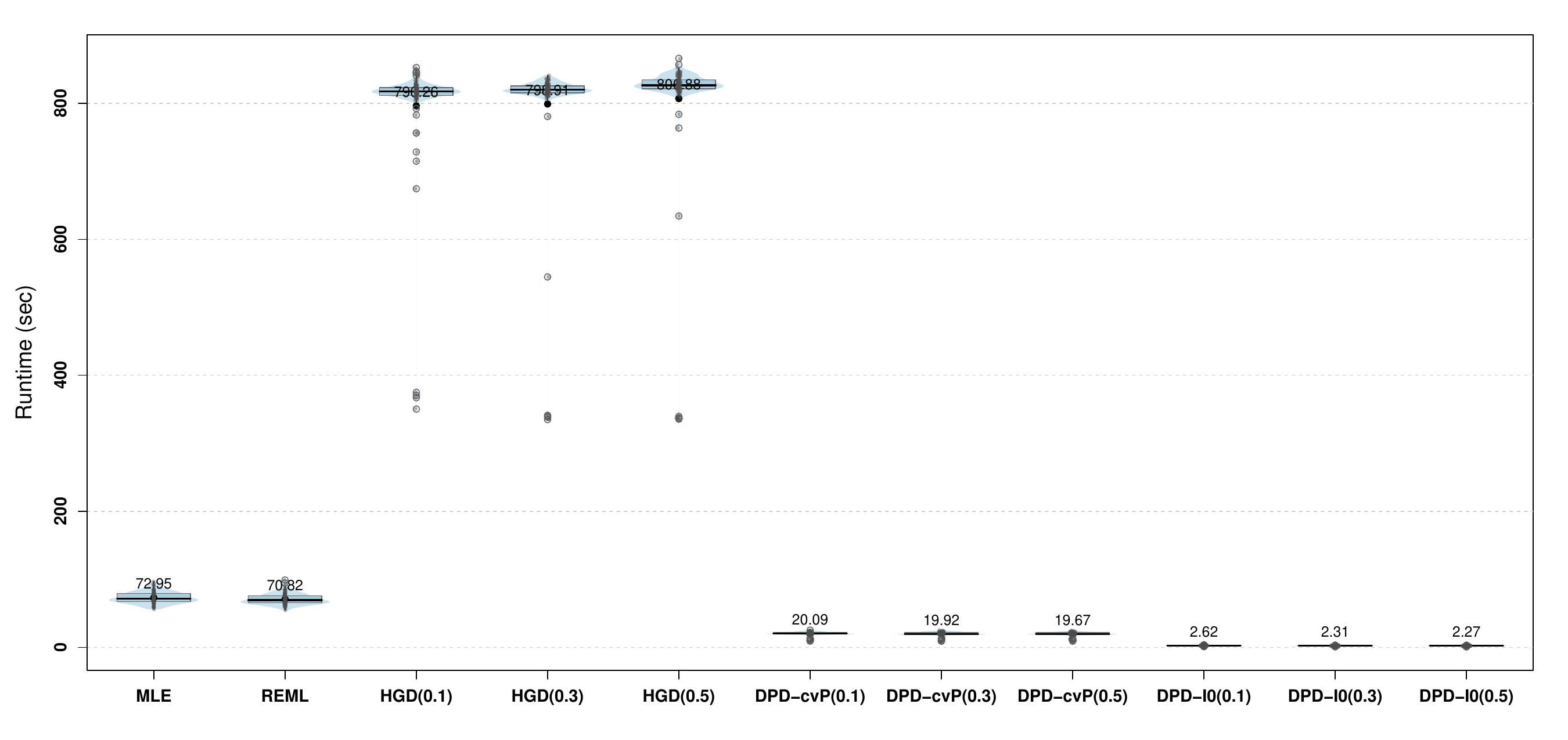}
		\label{FIG:boxplot_Y}}
	~	
	\subfloat[\tiny 10\% (C4) Contamination]{
		\includegraphics[page=3, width=0.3\textwidth]{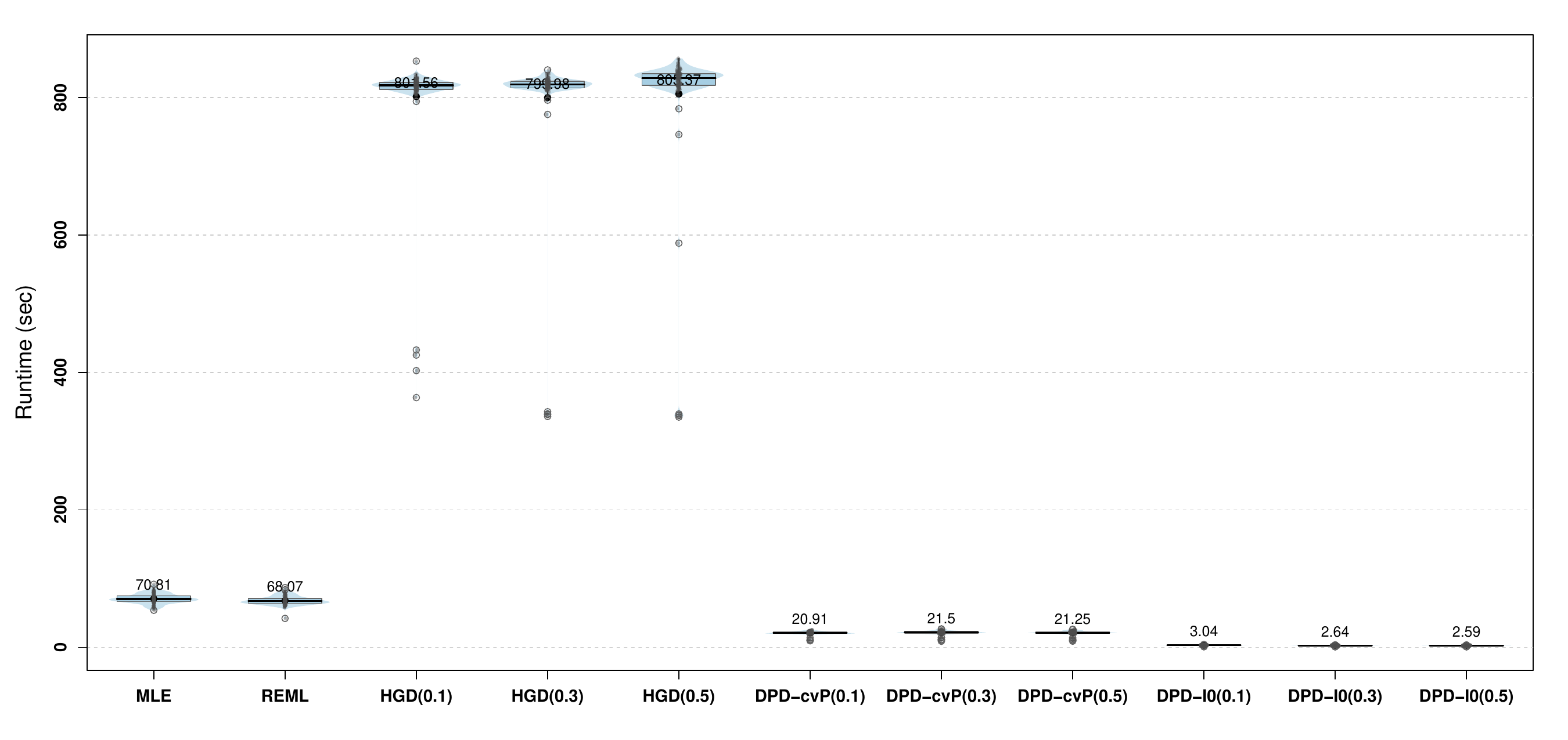}
		\label{FIG:boxplot_Y}}
	~	
	\subfloat[\tiny 20\% (C4) Contamination]{
		\includegraphics[page=3, width=0.3\textwidth]{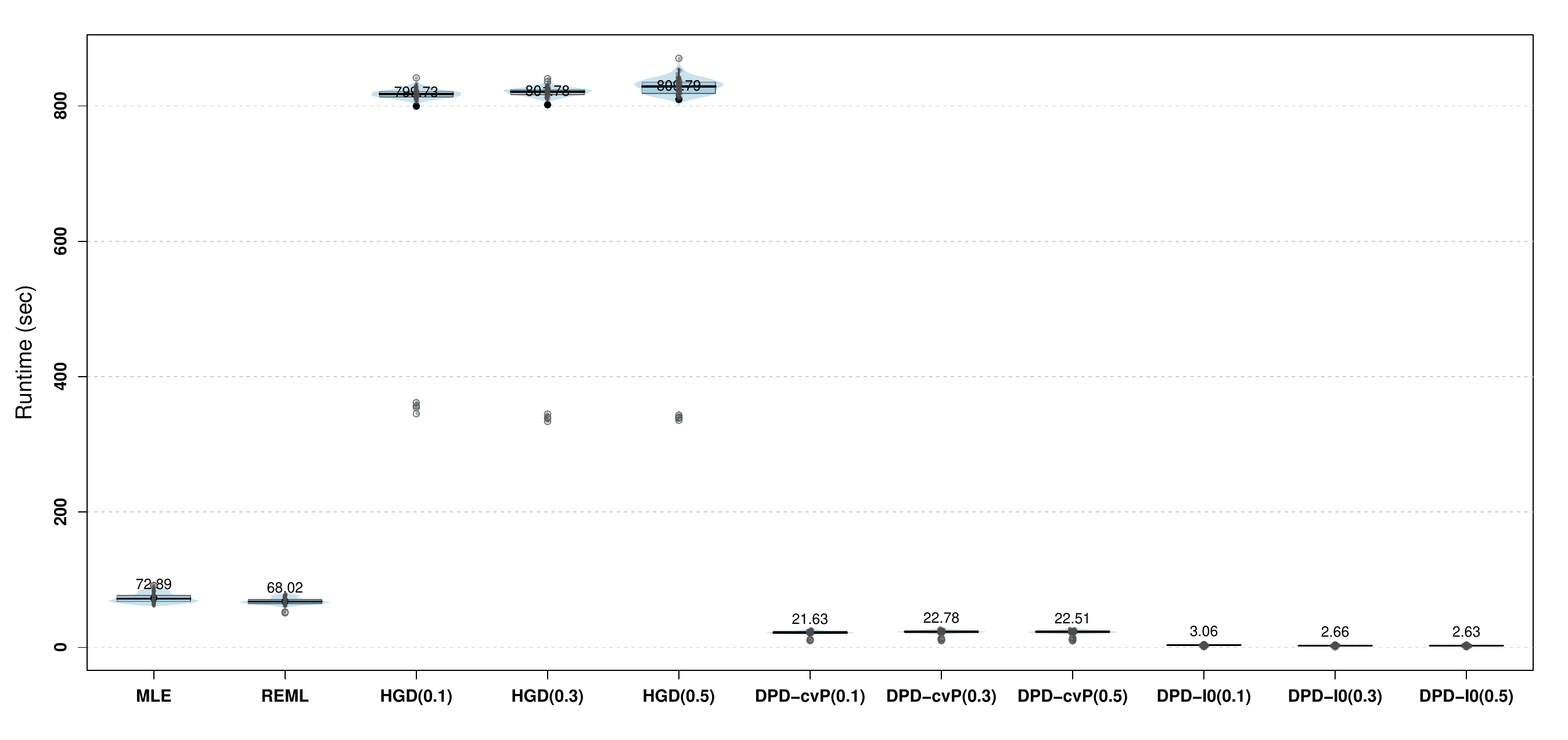}
		\label{FIG:boxplot_Y}}
	\caption{Boxplots  (with overlaid sample means) of computational times (in seconds) of the screening methods  under scenario (S1)$\times$(R1) with $\bm\Sigma_x$ as T(0.3) and different types of contamination}
	\label{FIG:time_Csp2}
\end{figure}

\begin{figure}[!h]
	\centering
	\subfloat[\tiny 5\% (C1) Contamination]{
		\includegraphics[page=3, width=0.3\textwidth]{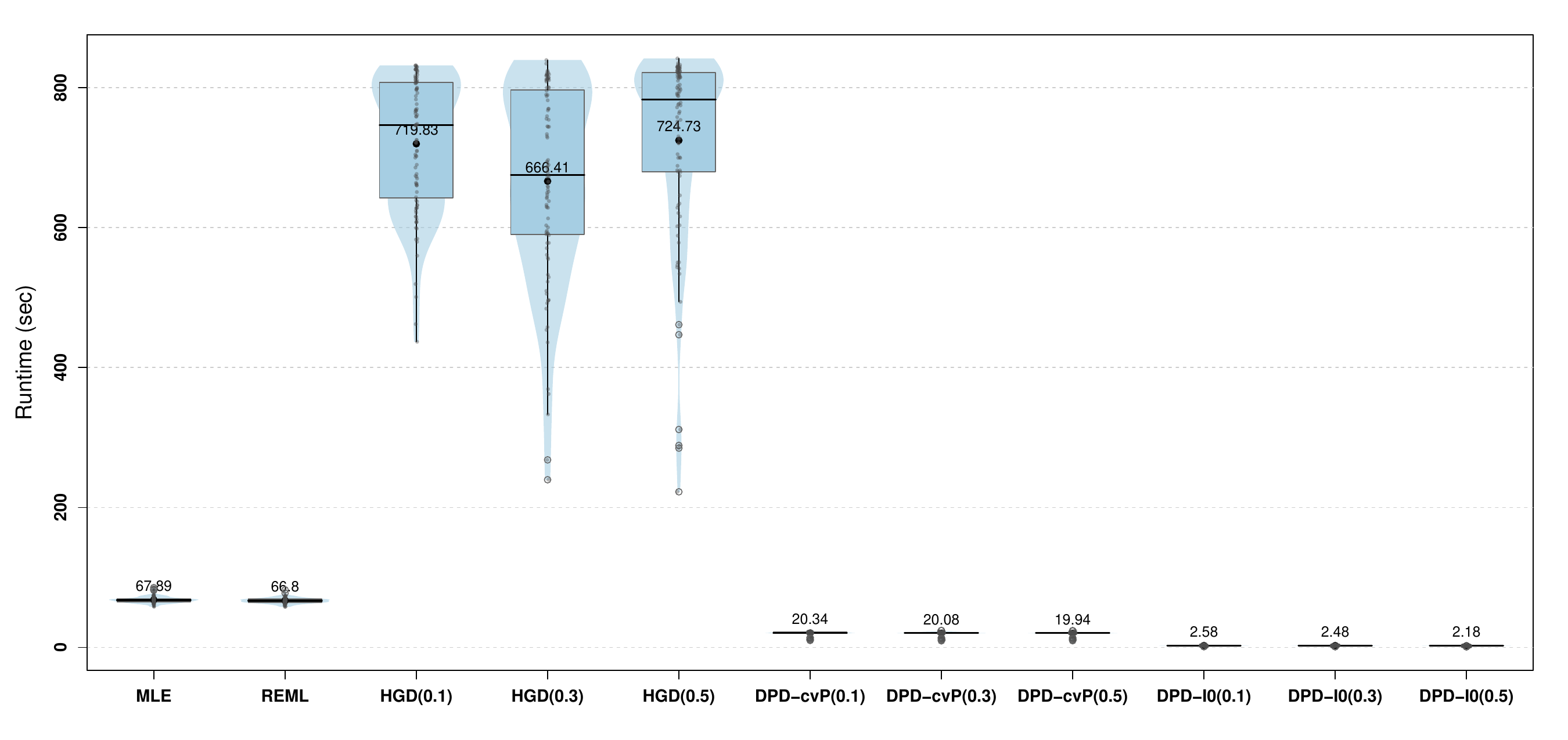}
		\label{FIG:boxplot_Y}}
	~	
	\subfloat[\tiny 10\% (C1) Contamination]{
		\includegraphics[page=3, width=0.3\textwidth]{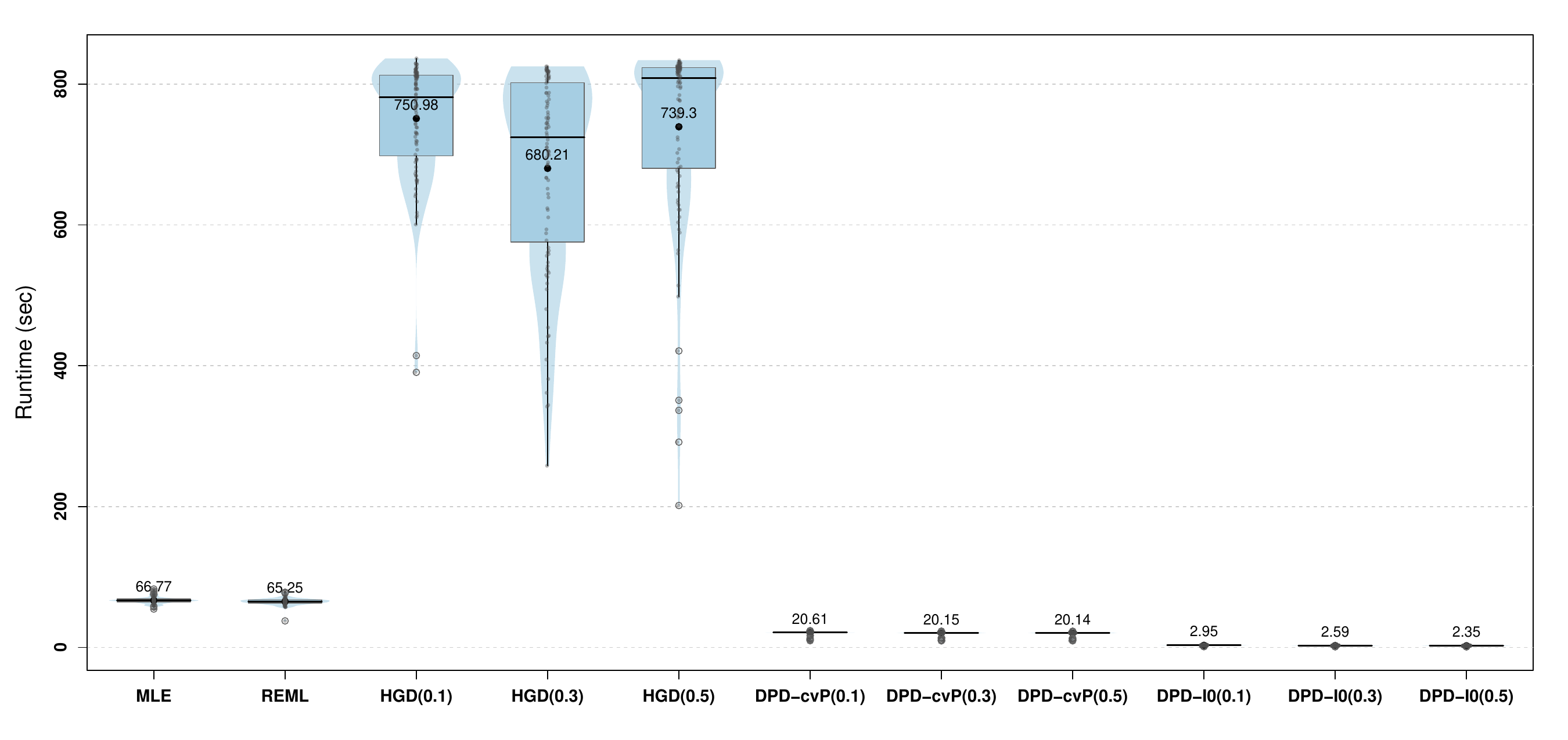}
		\label{FIG:boxplot_Y}}
	~	
	\subfloat[\tiny 20\% (C1) Contamination]{
		\includegraphics[page=3, width=0.3\textwidth]{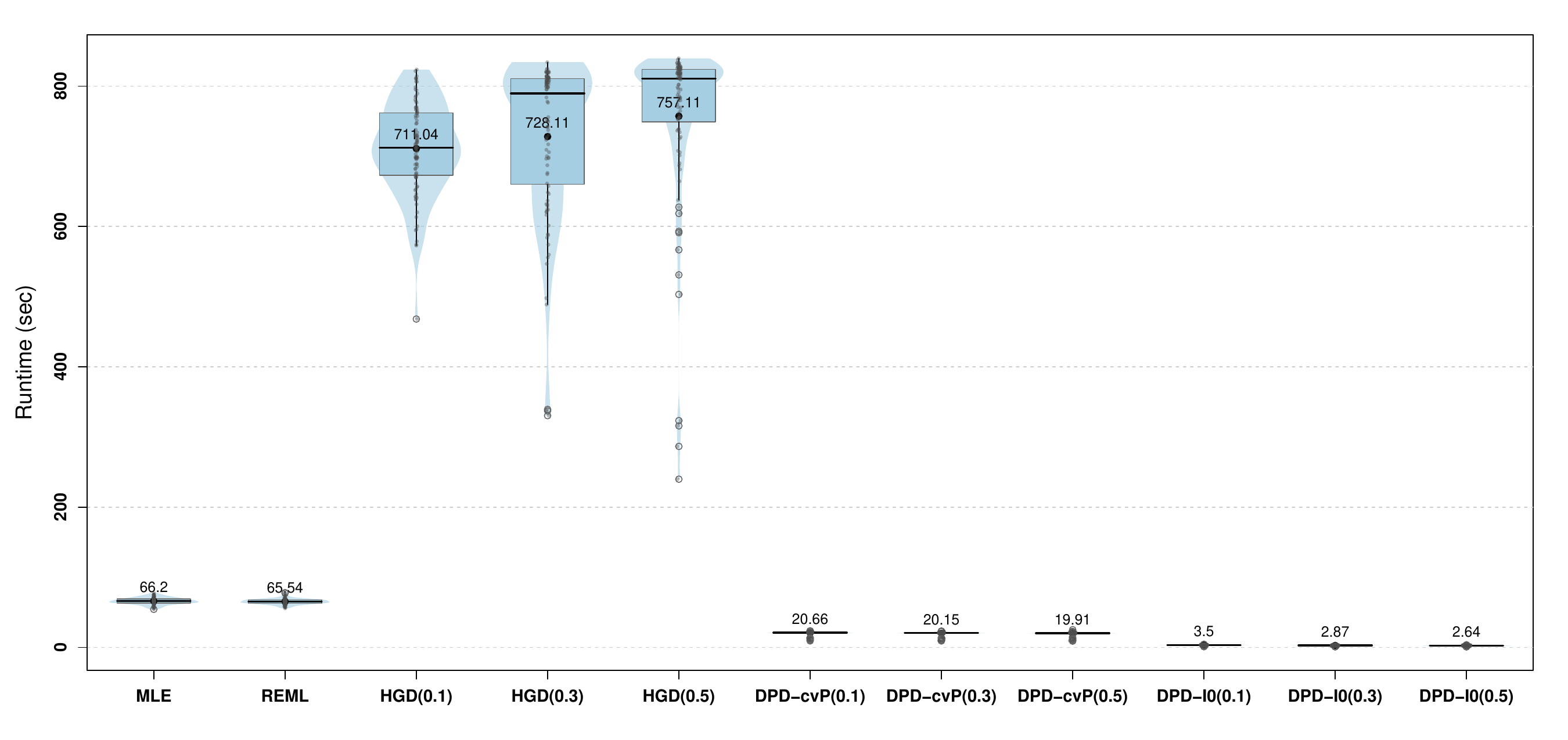}
		\label{FIG:boxplot_Y}}
	\\	
	\subfloat[\tiny 5\% (C2) Contamination]{
		\includegraphics[page=3, width=0.3\textwidth]{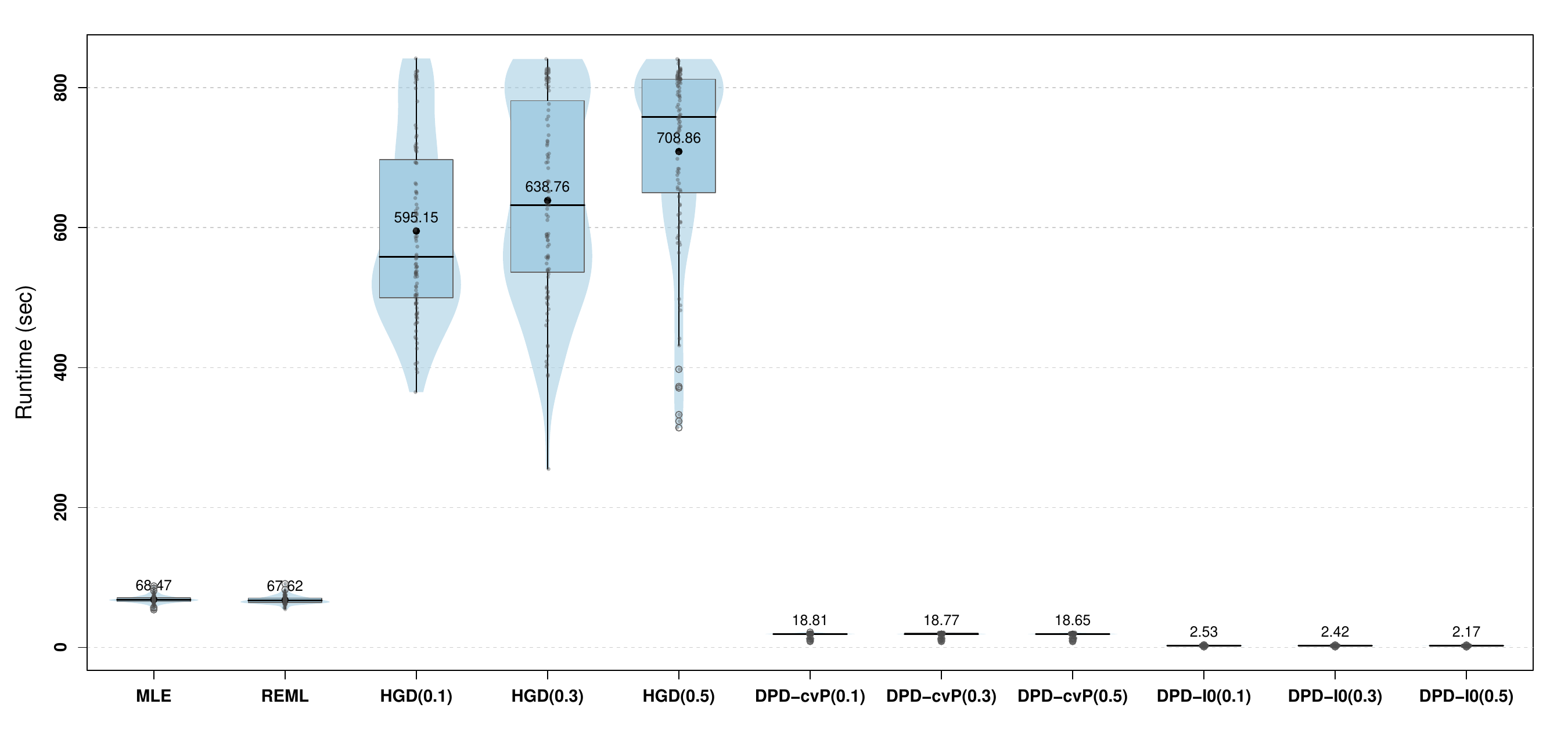}
		\label{FIG:boxplot_Y}}
	~	
	\subfloat[\tiny 10\% (C2) Contamination]{
		\includegraphics[page=3, width=0.3\textwidth]{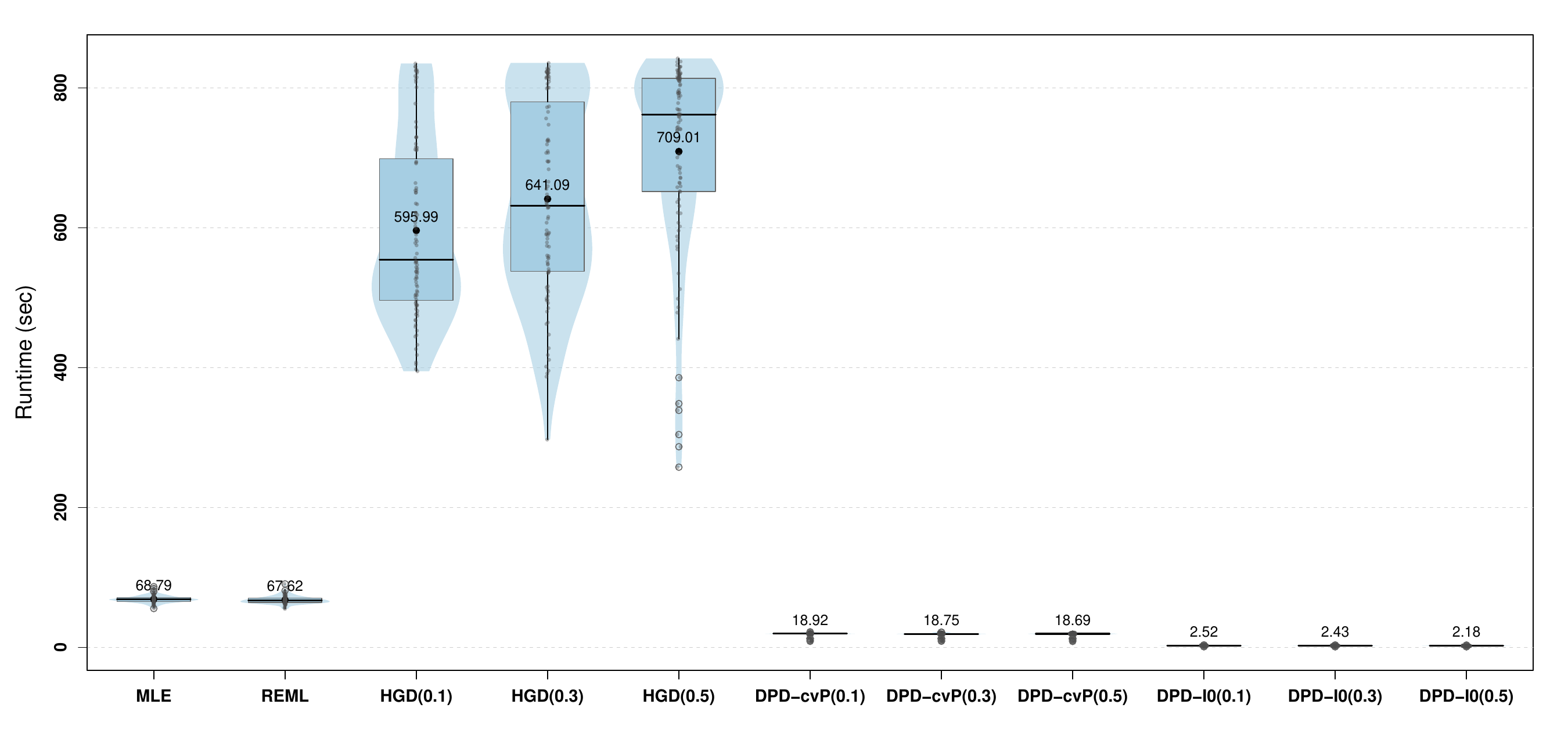}
		\label{FIG:boxplot_Y}}
	~	
	\subfloat[\tiny 20\% (C2) Contamination]{
		\includegraphics[page=3, width=0.3\textwidth]{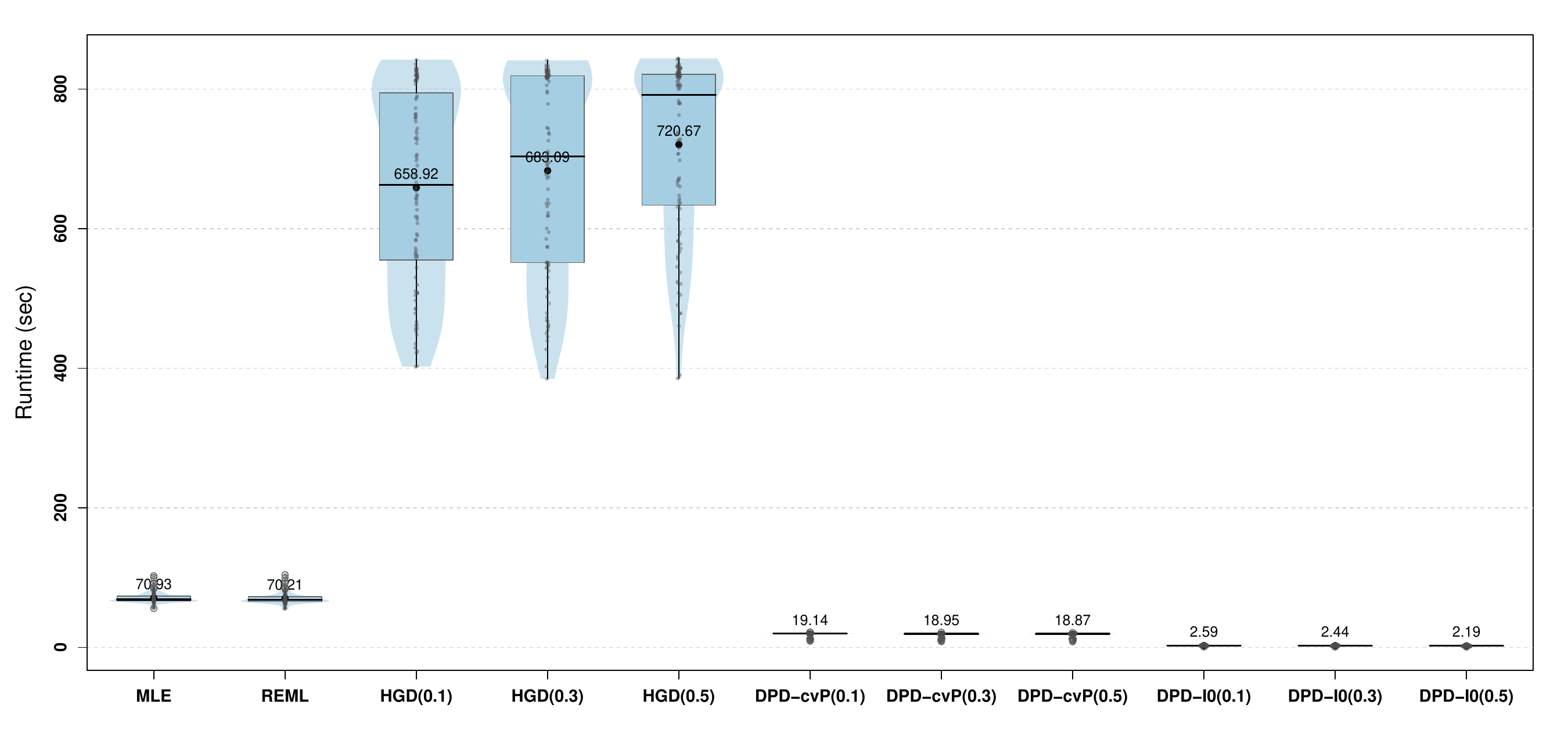}
		\label{FIG:boxplot_Y}}
	\\	
	\subfloat[\tiny 5\% (C3) Contamination]{
		\includegraphics[page=3, width=0.3\textwidth]{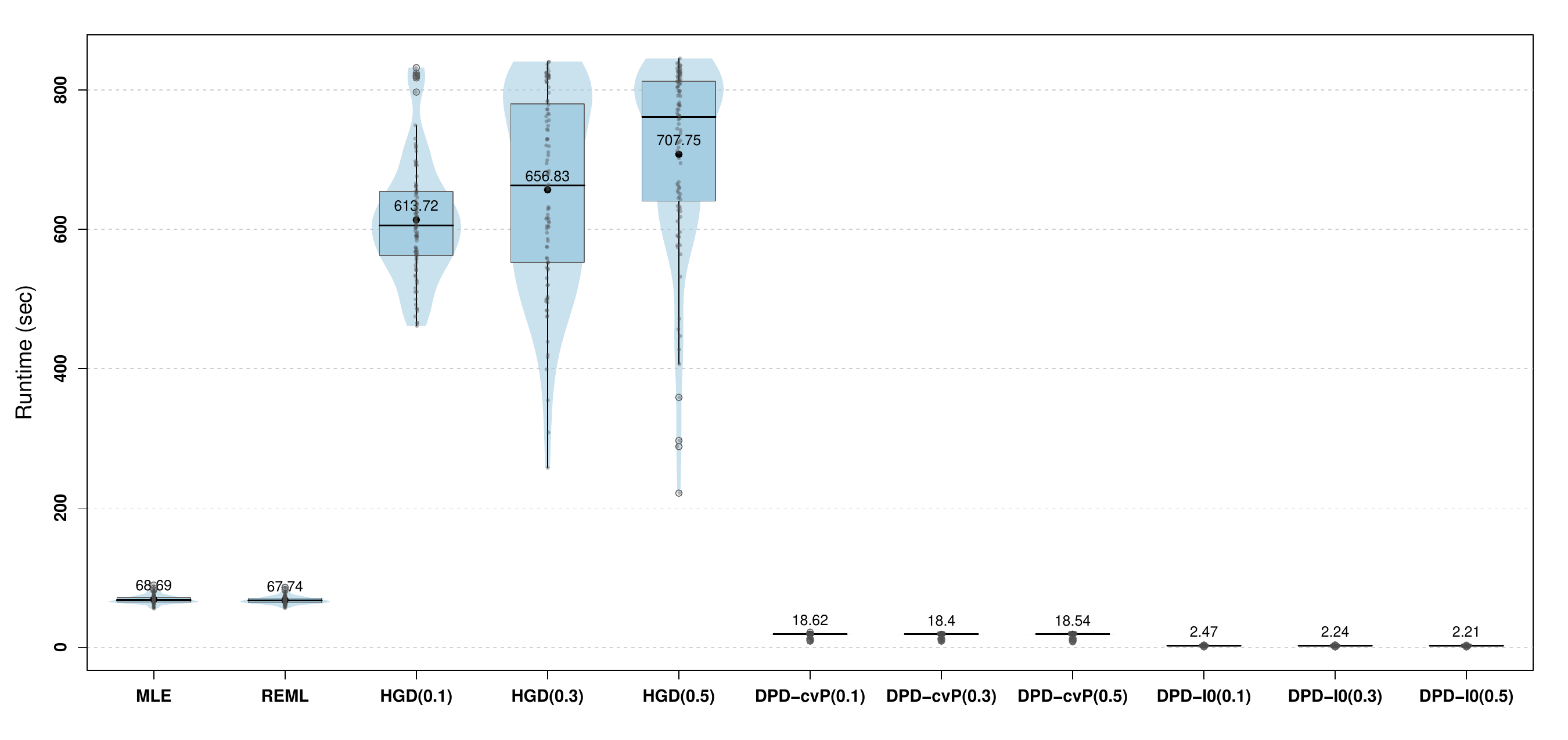}
		\label{FIG:boxplot_Y}}
	~	
	\subfloat[\tiny 10\% (C3) Contamination]{
		\includegraphics[page=3, width=0.3\textwidth]{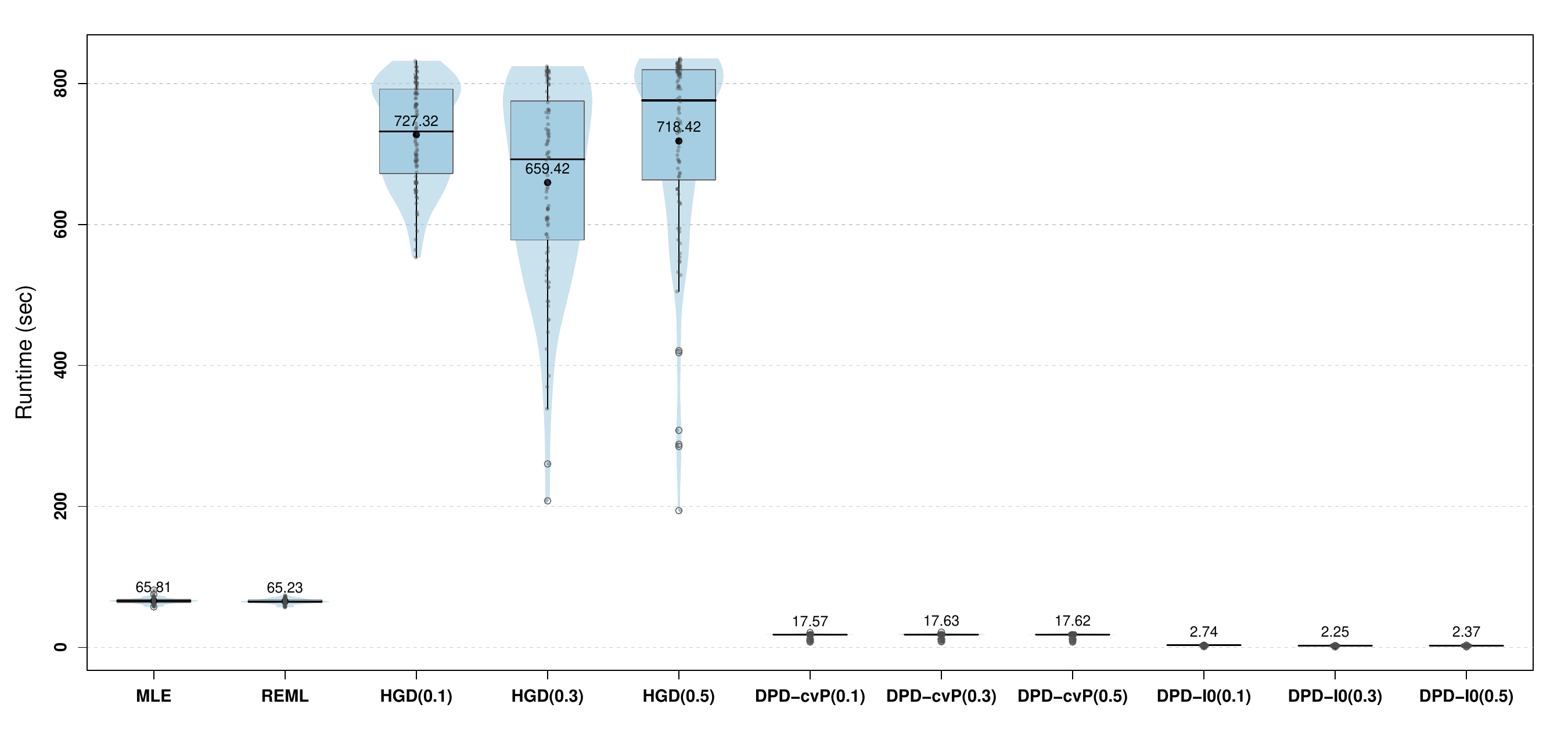}
		\label{FIG:boxplot_Y}}
	~	
	\subfloat[\tiny 20\% (C3) Contamination]{
		\includegraphics[page=3, width=0.3\textwidth]{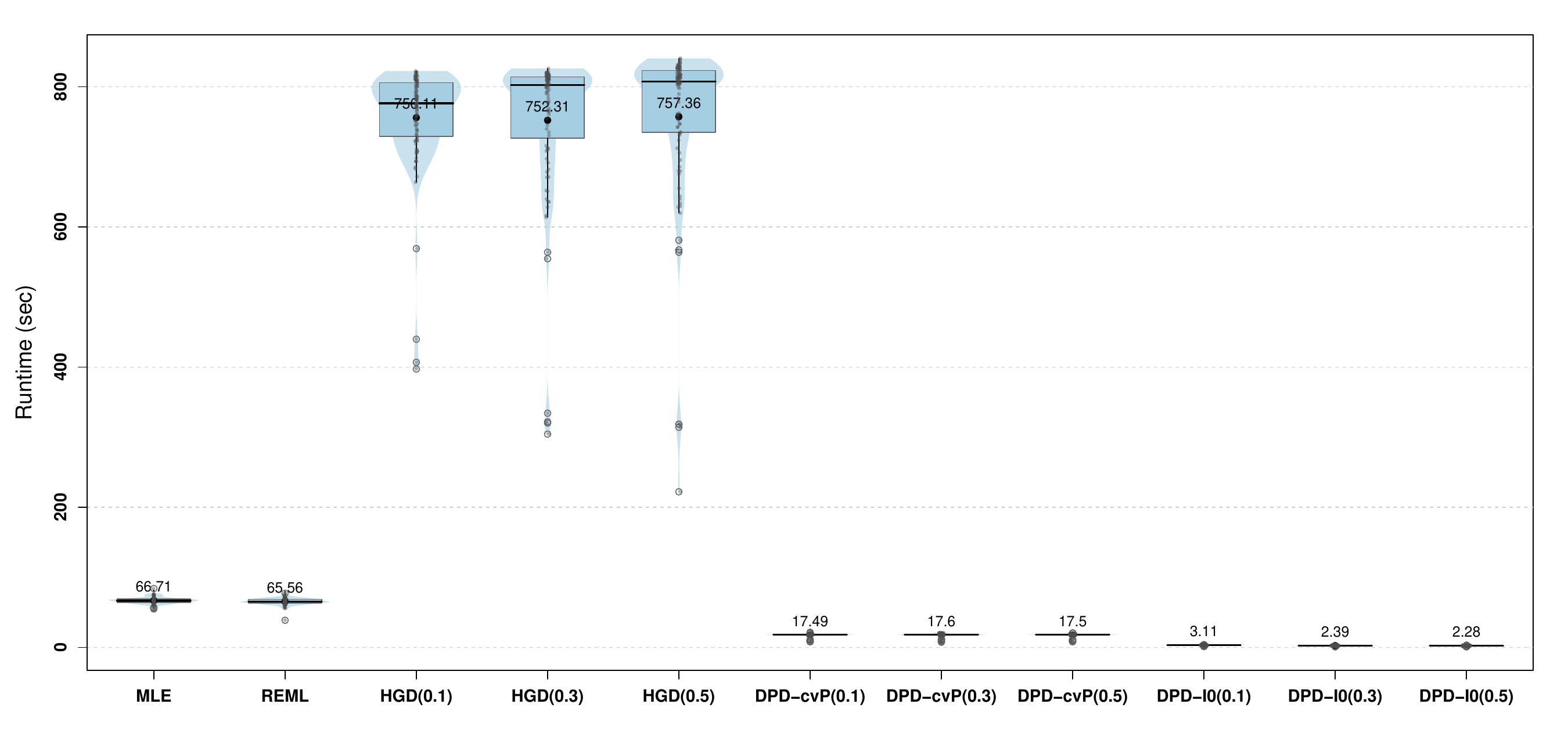}
		\label{FIG:boxplot_Y}}
	\\	
	\subfloat[\tiny 5\% (C4) Contamination]{
		\includegraphics[page=3, width=0.3\textwidth]{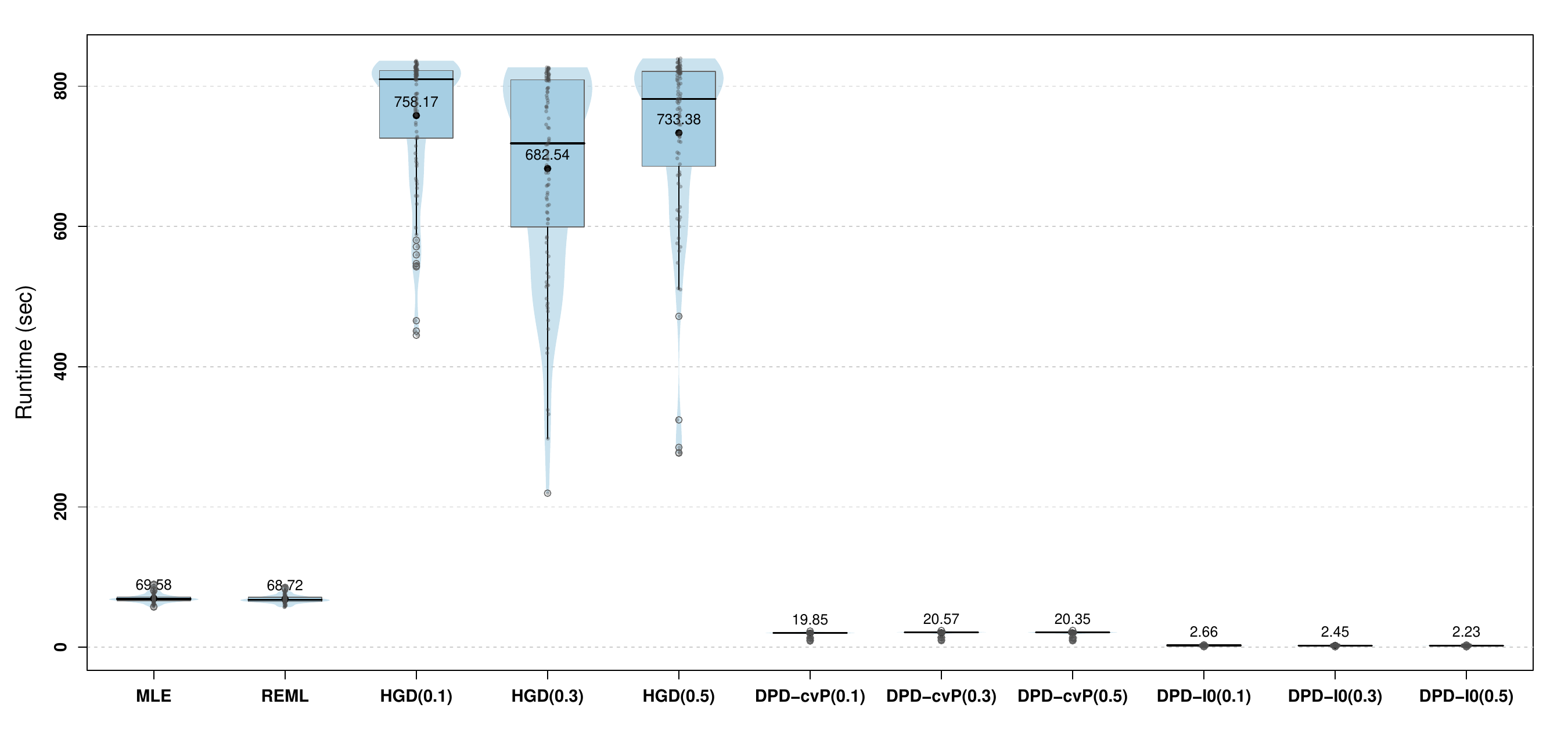}
		\label{FIG:boxplot_Y}}
	~	
	\subfloat[\tiny 10\% (C4) Contamination]{
		\includegraphics[page=3, width=0.3\textwidth]{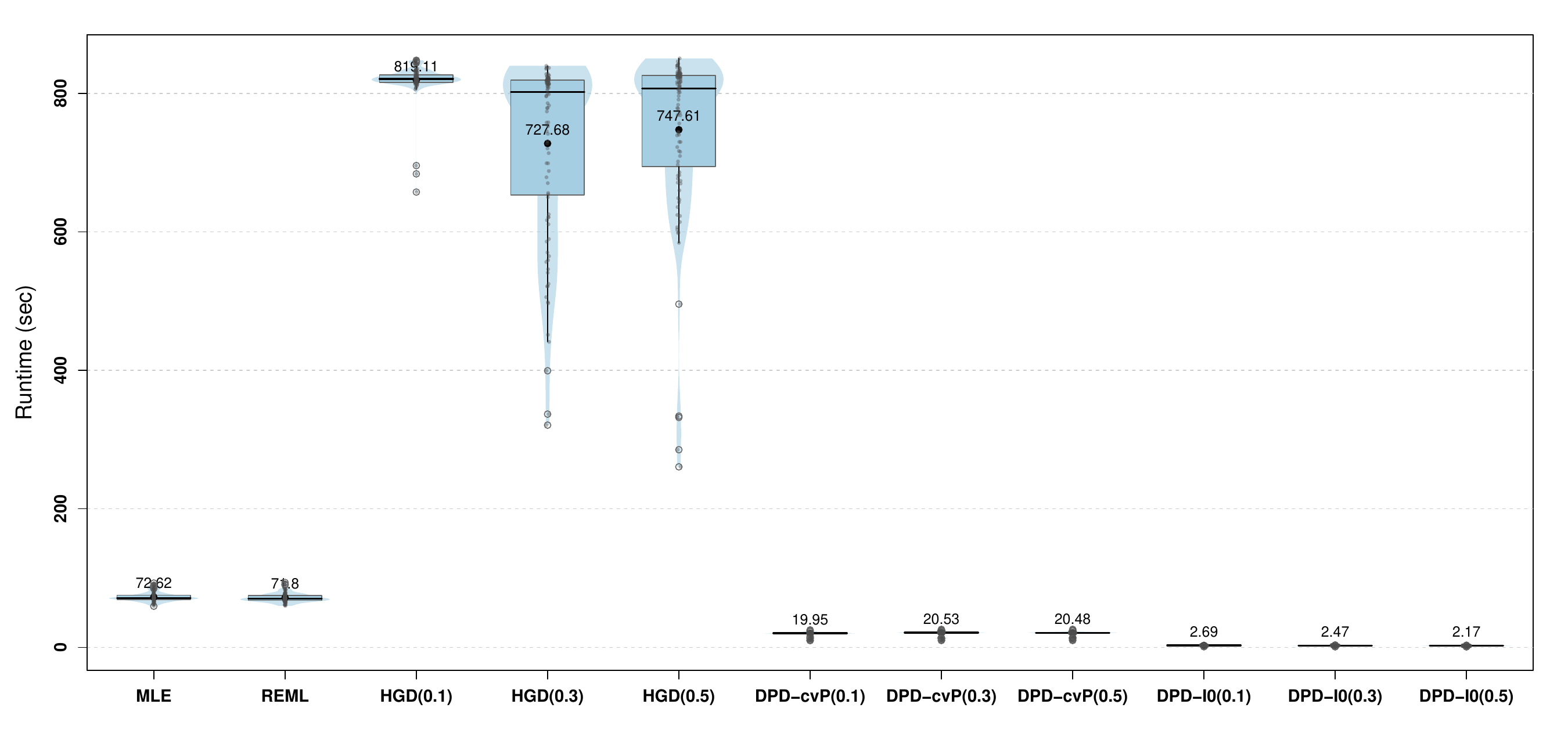}
		\label{FIG:boxplot_Y}}
	~	
	\subfloat[\tiny 20\% (C4) Contamination]{
		\includegraphics[page=3, width=0.3\textwidth]{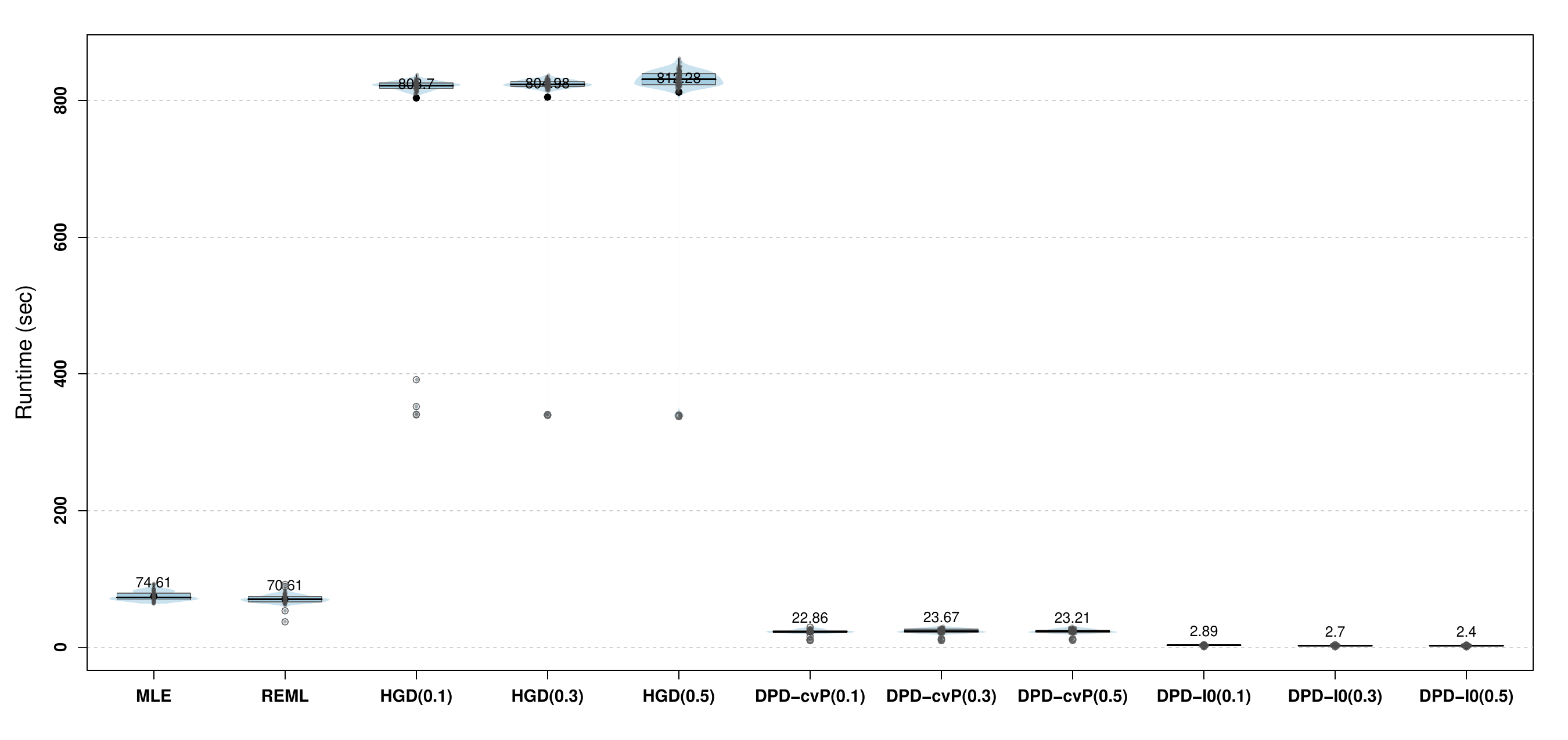}
		\label{FIG:boxplot_Y}}
	\caption{Boxplots  (with overlaid sample means) of computational times (in seconds) of the screening methods  under scenario (S2)$\times$(R0) with $\bm\Sigma_x$ as T(0.3) and different types of contamination}
	\label{FIG:time_S2R0T}
\end{figure}

\begin{figure}[!h]
	\centering
	\subfloat[\tiny 5\% (C1) Contamination]{
		\includegraphics[page=3, width=0.3\textwidth]{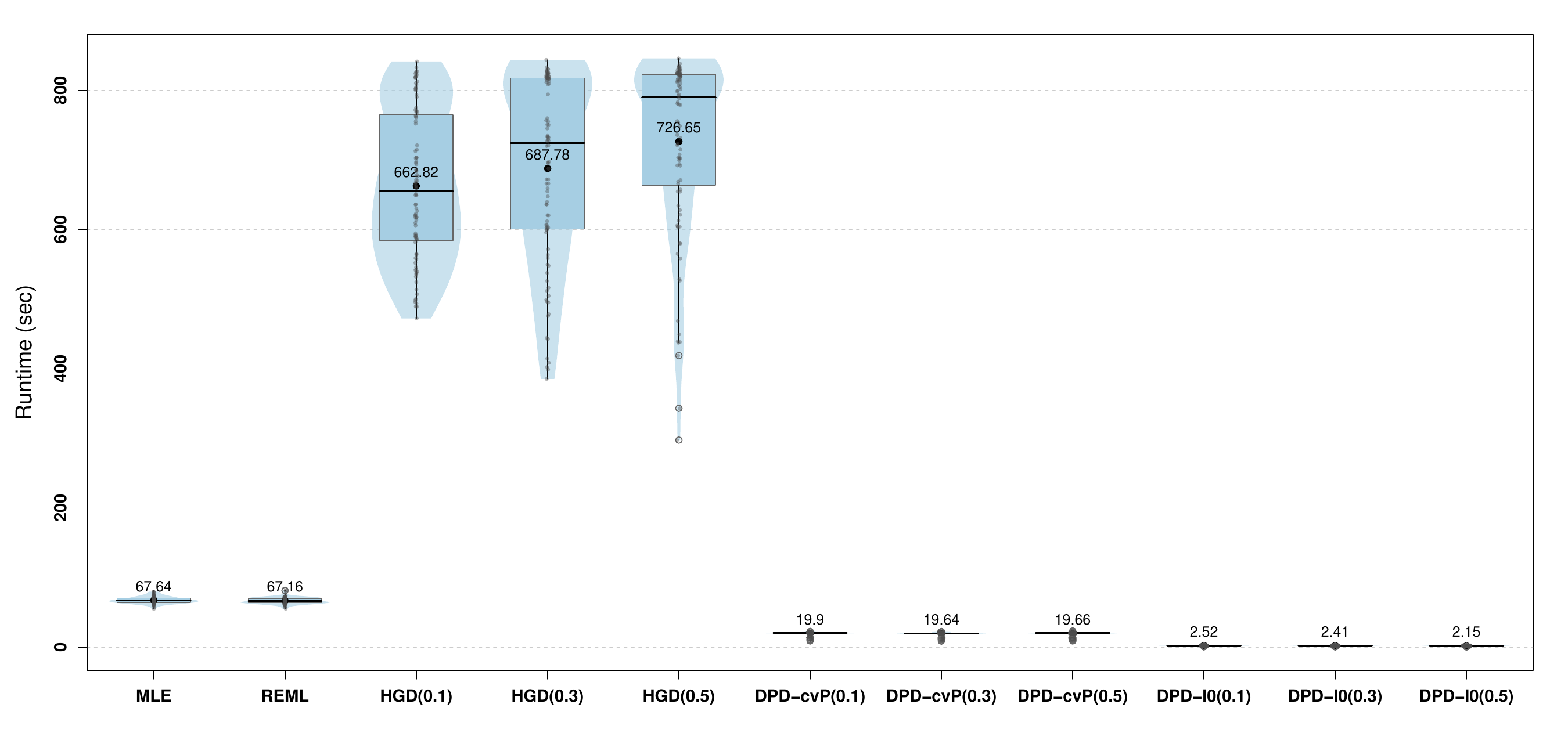}
		\label{FIG:boxplot_Y}}
	~	
	\subfloat[\tiny 10\% (C1) Contamination]{
		\includegraphics[page=3, width=0.3\textwidth]{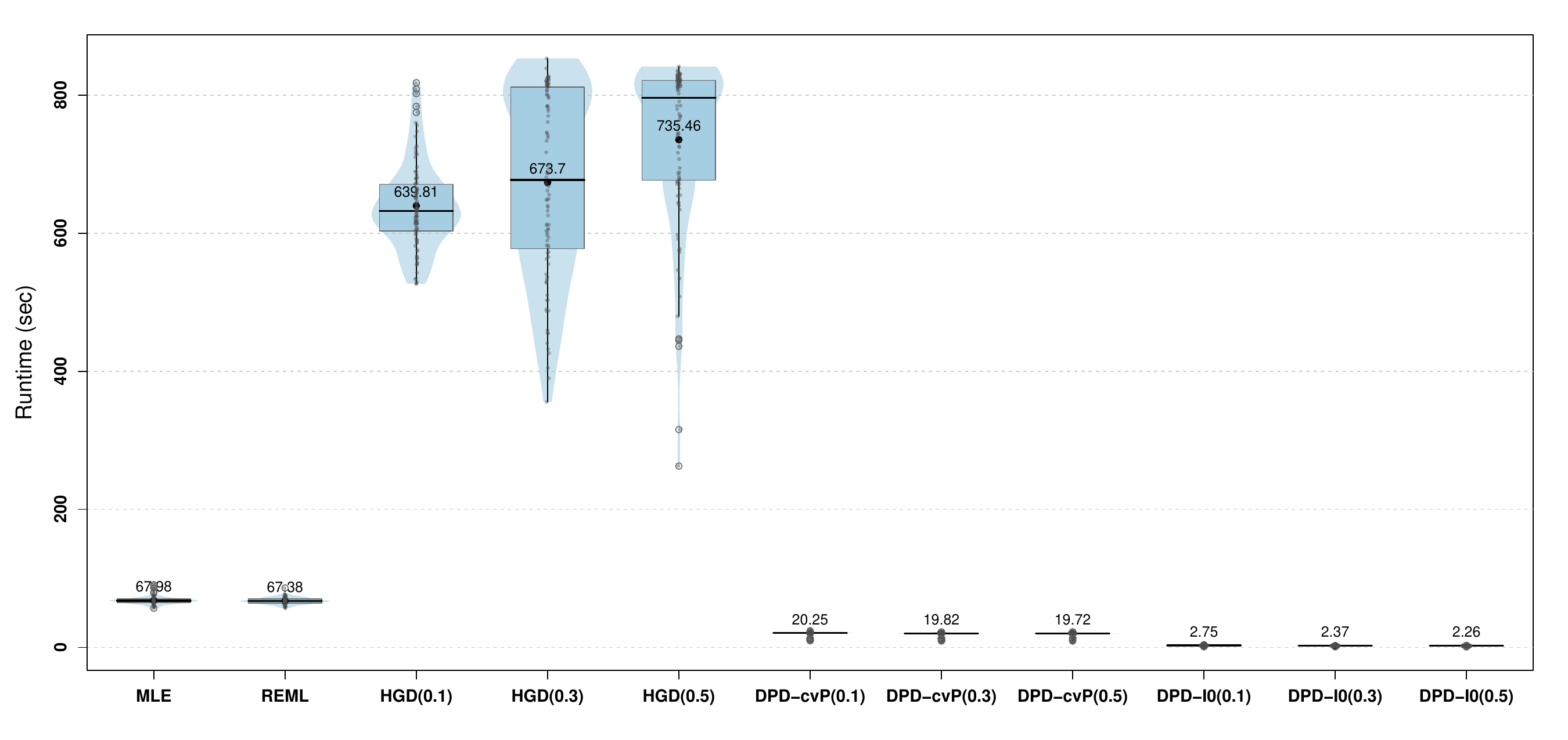}
		\label{FIG:boxplot_Y}}
	~	
	\subfloat[\tiny 20\% (C1) Contamination]{
		\includegraphics[page=3, width=0.3\textwidth]{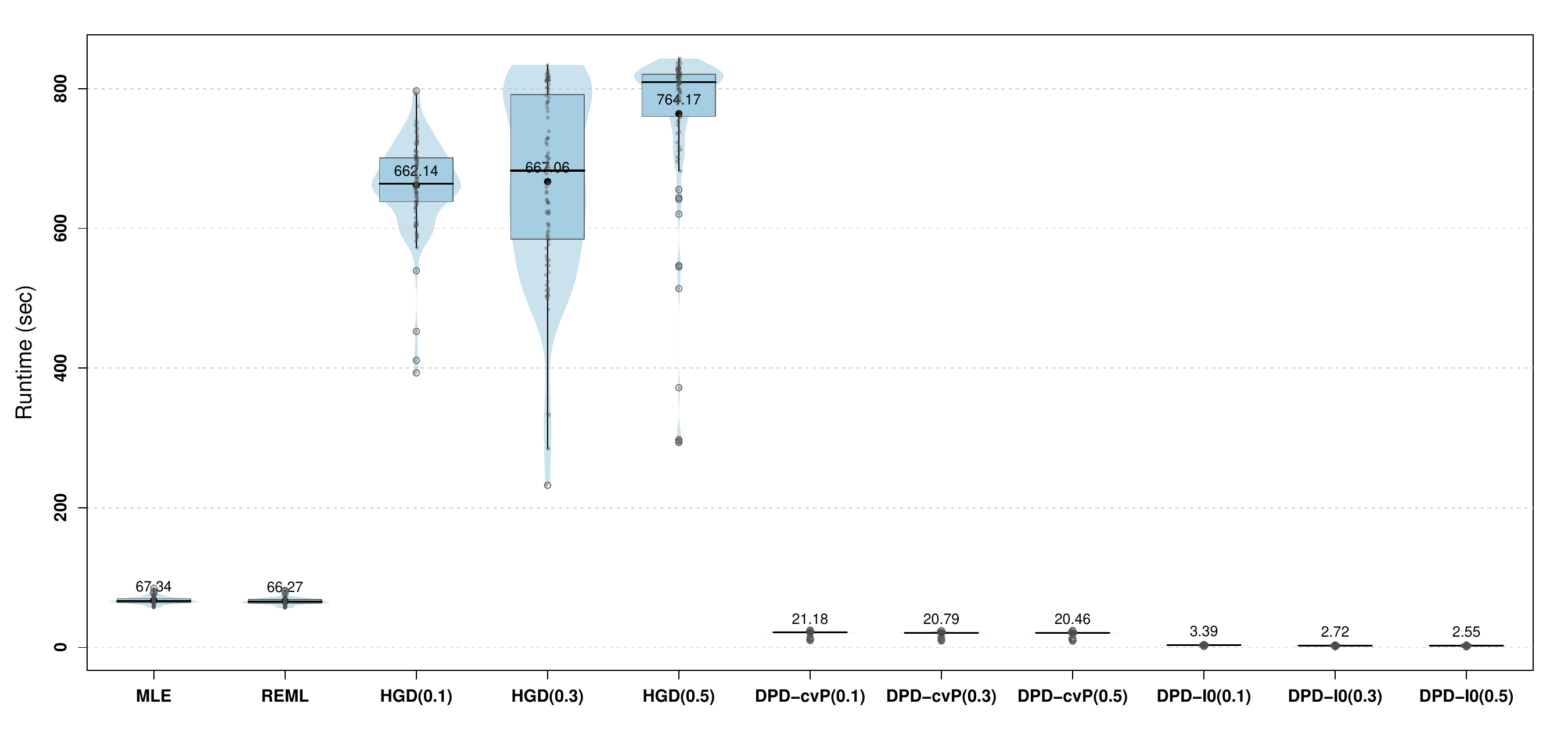}
		\label{FIG:boxplot_Y}}
	\\	
	\subfloat[\tiny 5\% (C2) Contamination]{
		\includegraphics[page=3, width=0.3\textwidth]{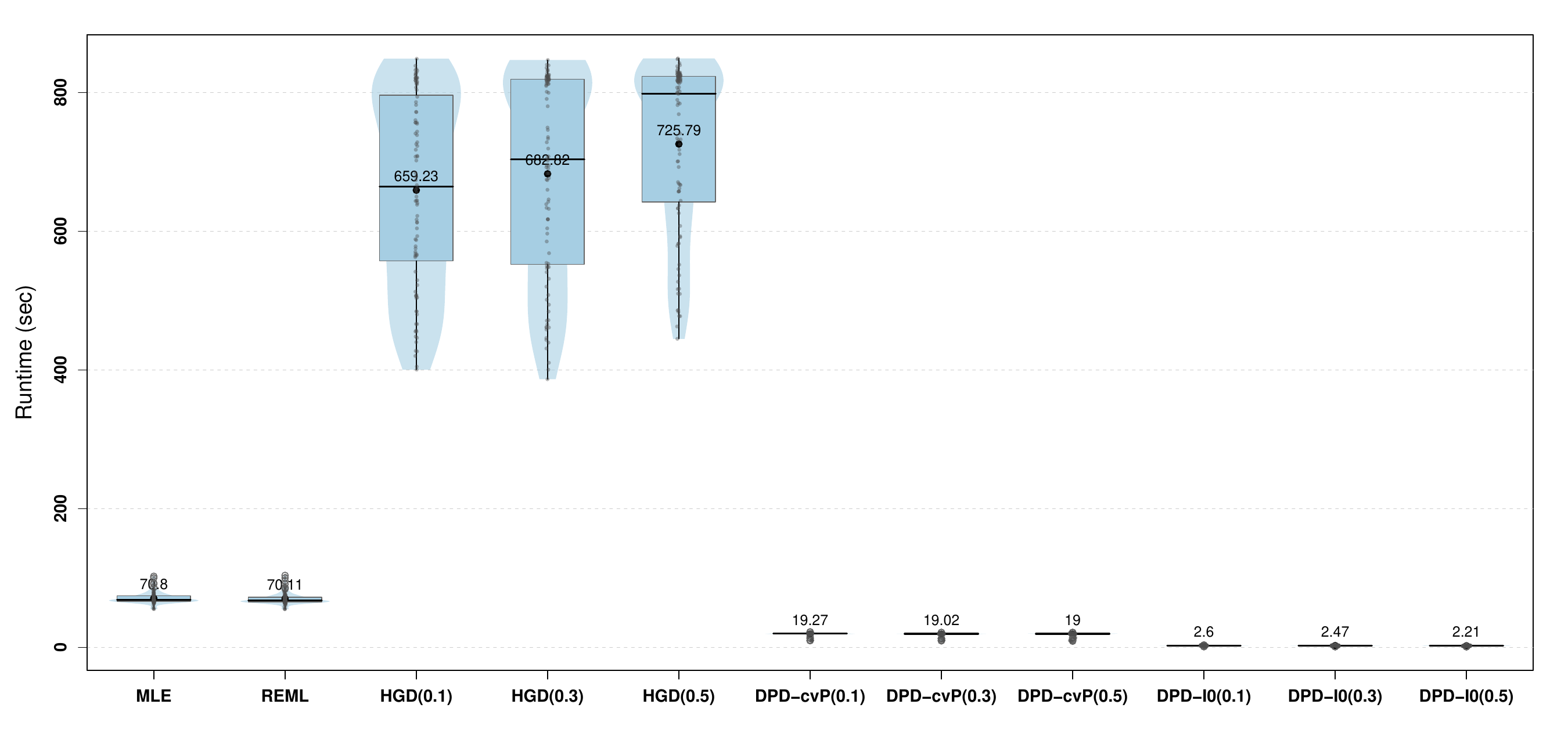}
		\label{FIG:boxplot_Y}}
	~	
	\subfloat[\tiny 10\% (C2) Contamination]{
		\includegraphics[page=3, width=0.3\textwidth]{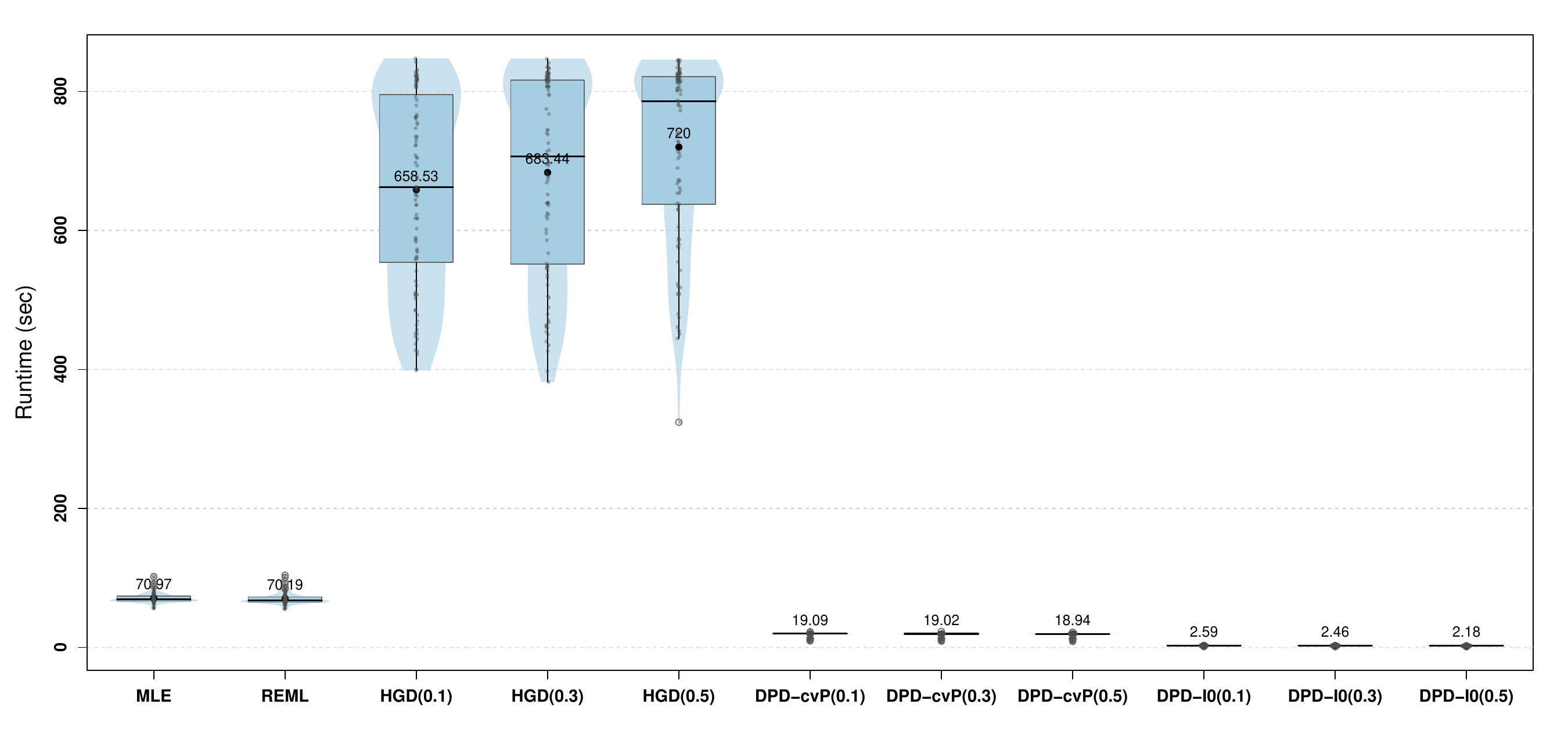}
		\label{FIG:boxplot_Y}}
	~	
	\subfloat[\tiny 20\% (C2) Contamination]{
		\includegraphics[page=3, width=0.3\textwidth]{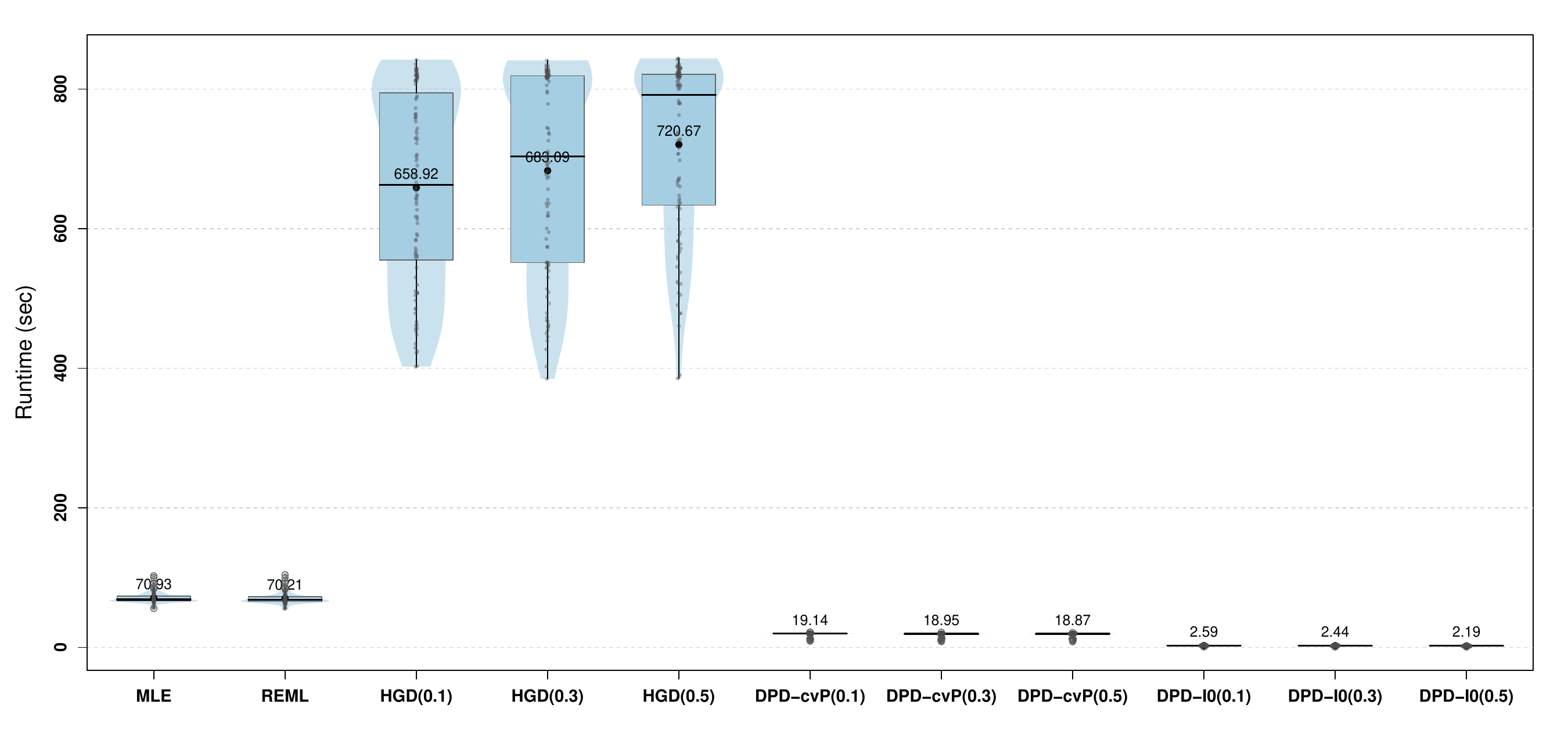}
		\label{FIG:boxplot_Y}}
	\\	
	\subfloat[\tiny 5\% (C3) Contamination]{
		\includegraphics[page=3, width=0.3\textwidth]{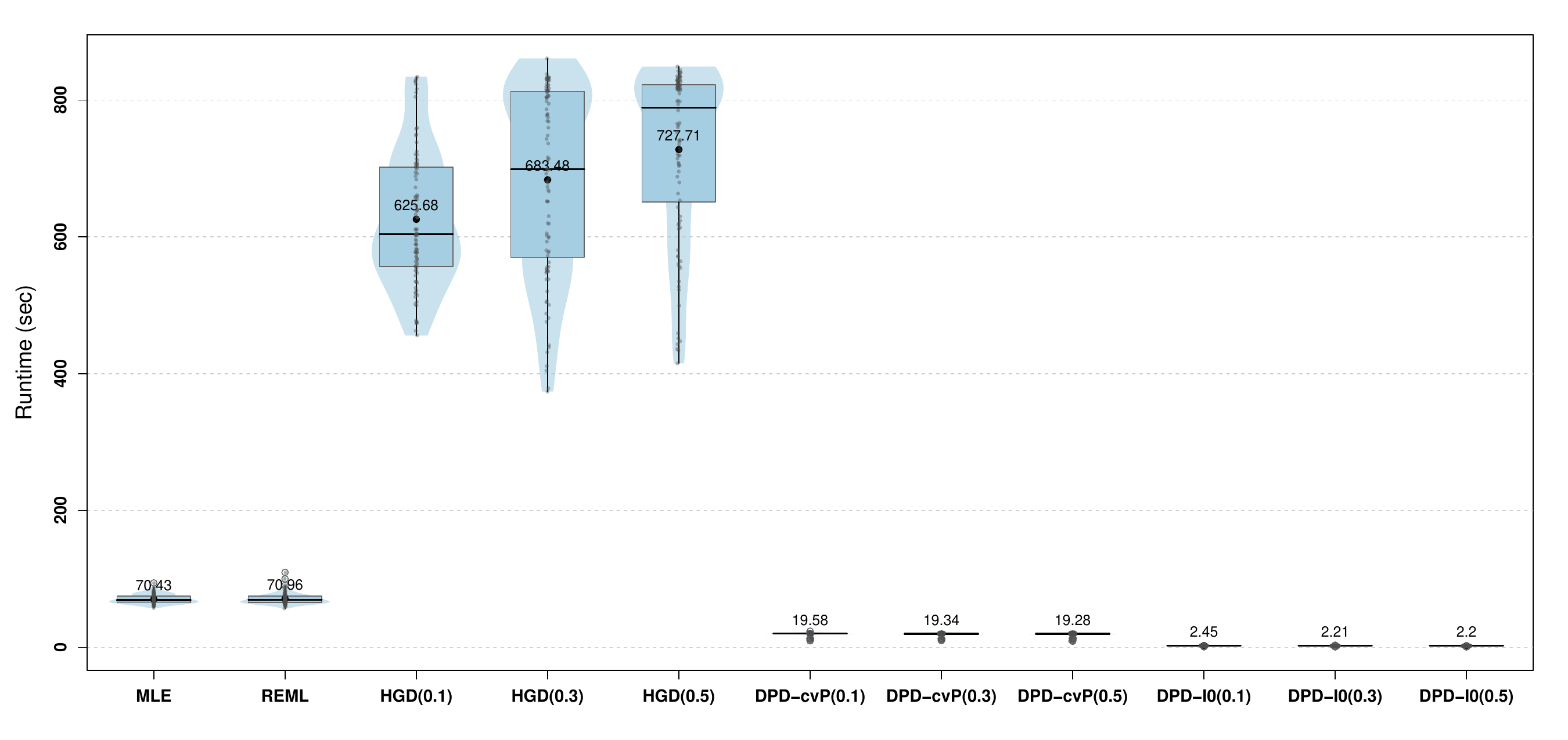}
		\label{FIG:boxplot_Y}}
	~	
	\subfloat[\tiny 10\% (C3) Contamination]{
		\includegraphics[page=3, width=0.3\textwidth]{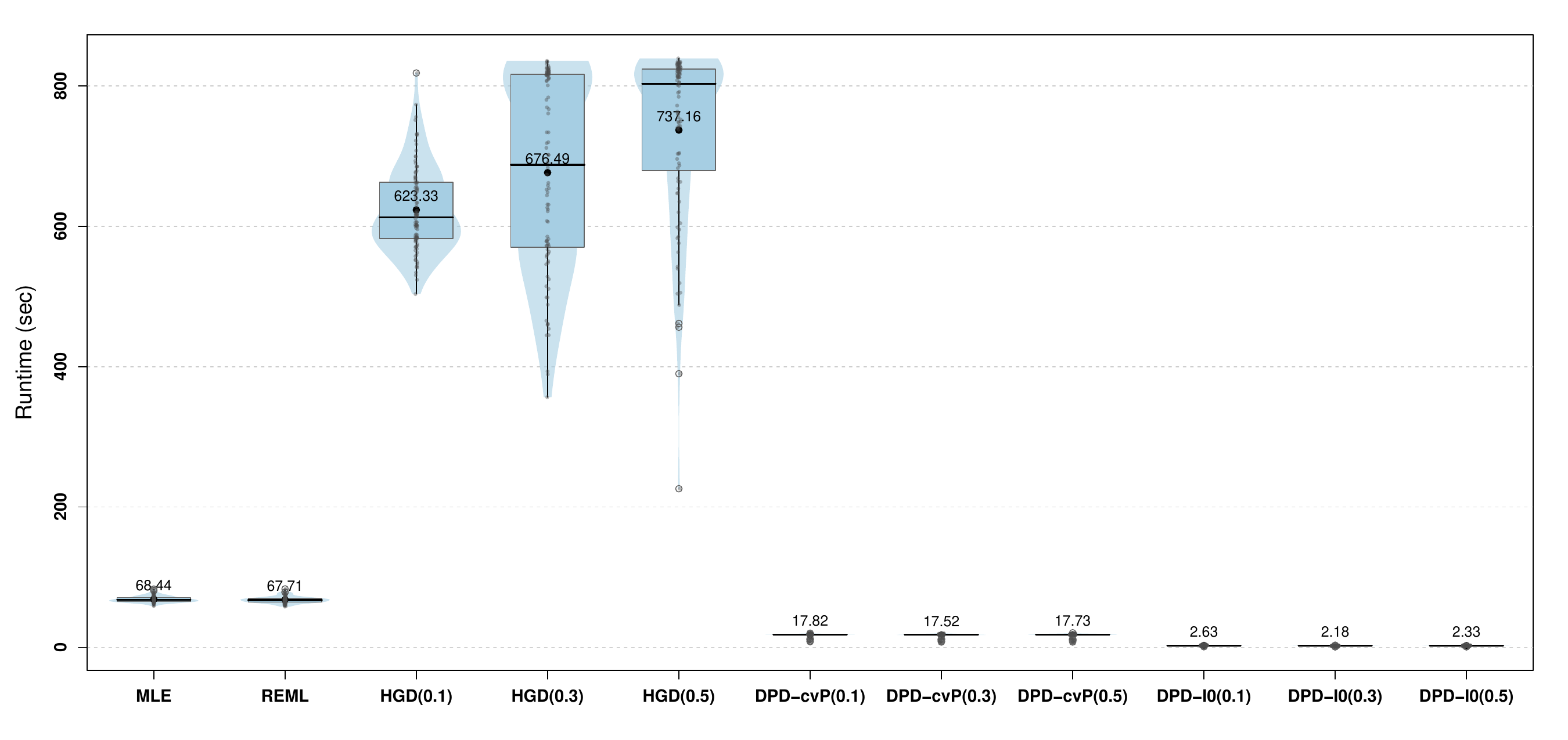}
		\label{FIG:boxplot_Y}}
	~	
	\subfloat[\tiny 20\% (C3) Contamination]{
		\includegraphics[page=3, width=0.3\textwidth]{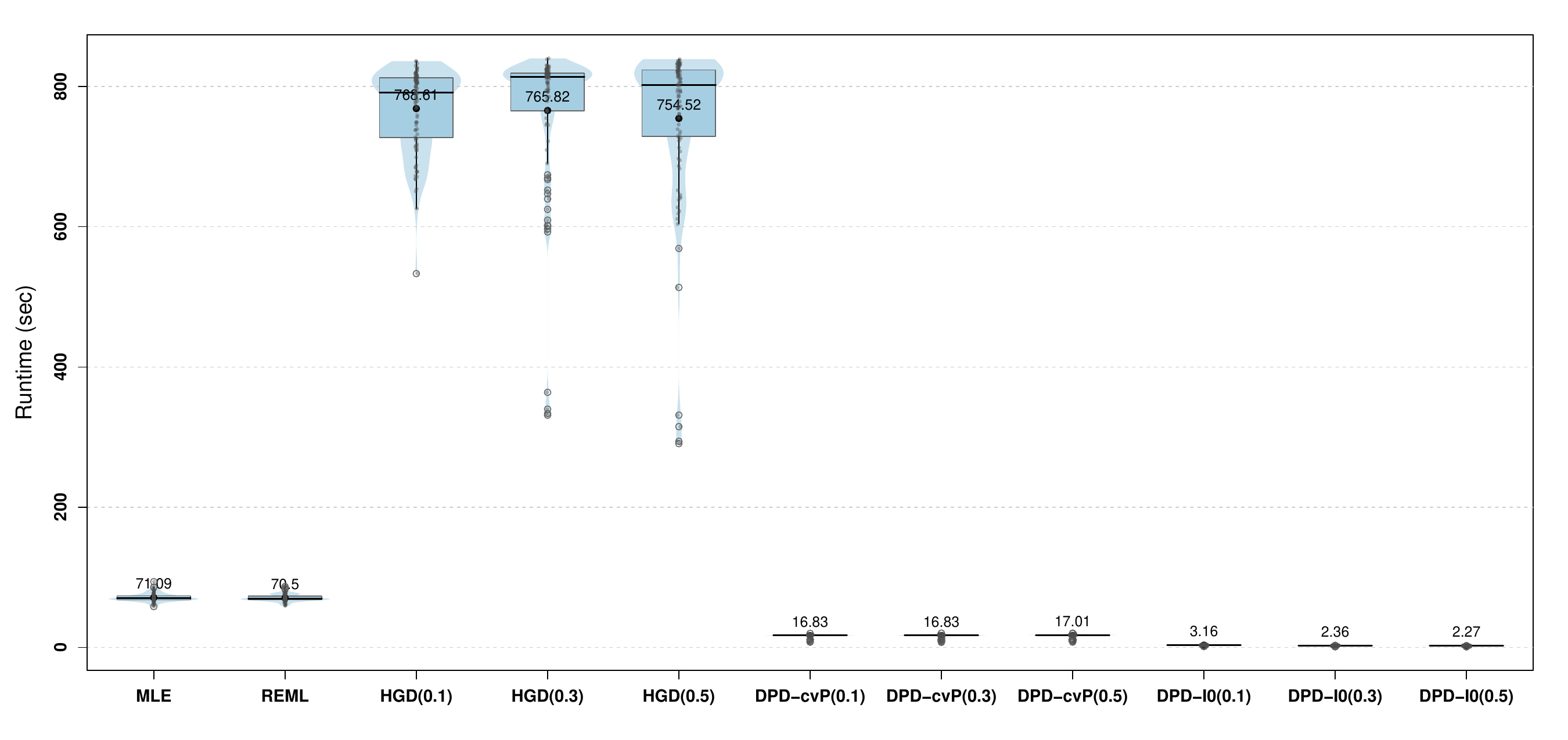}
		\label{FIG:boxplot_Y}}
	\\	
	\subfloat[\tiny 5\% (C4) Contamination]{
		\includegraphics[page=3, width=0.3\textwidth]{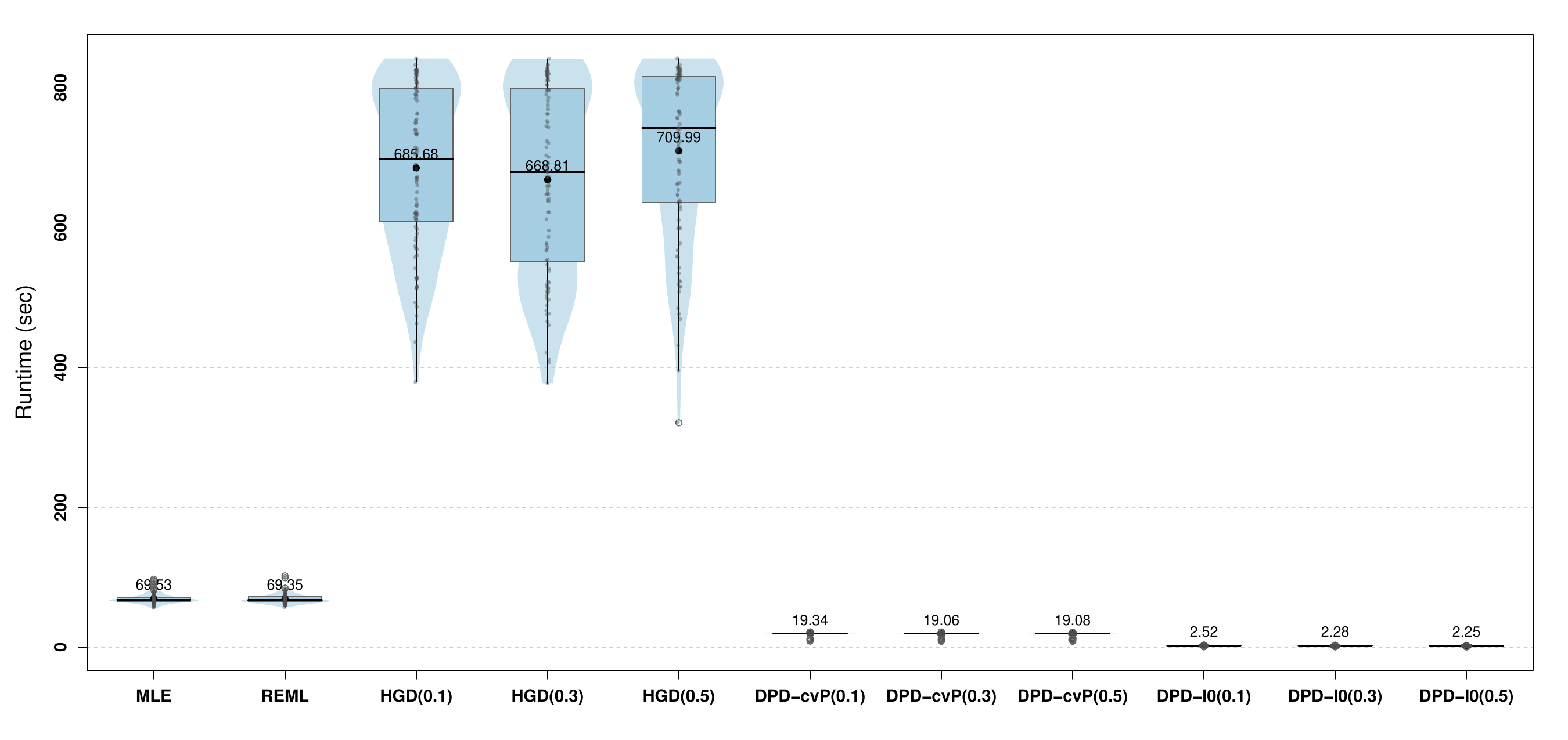}
		\label{FIG:boxplot_Y}}
	~	
	\subfloat[\tiny 10\% (C4) Contamination]{
		\includegraphics[page=3, width=0.3\textwidth]{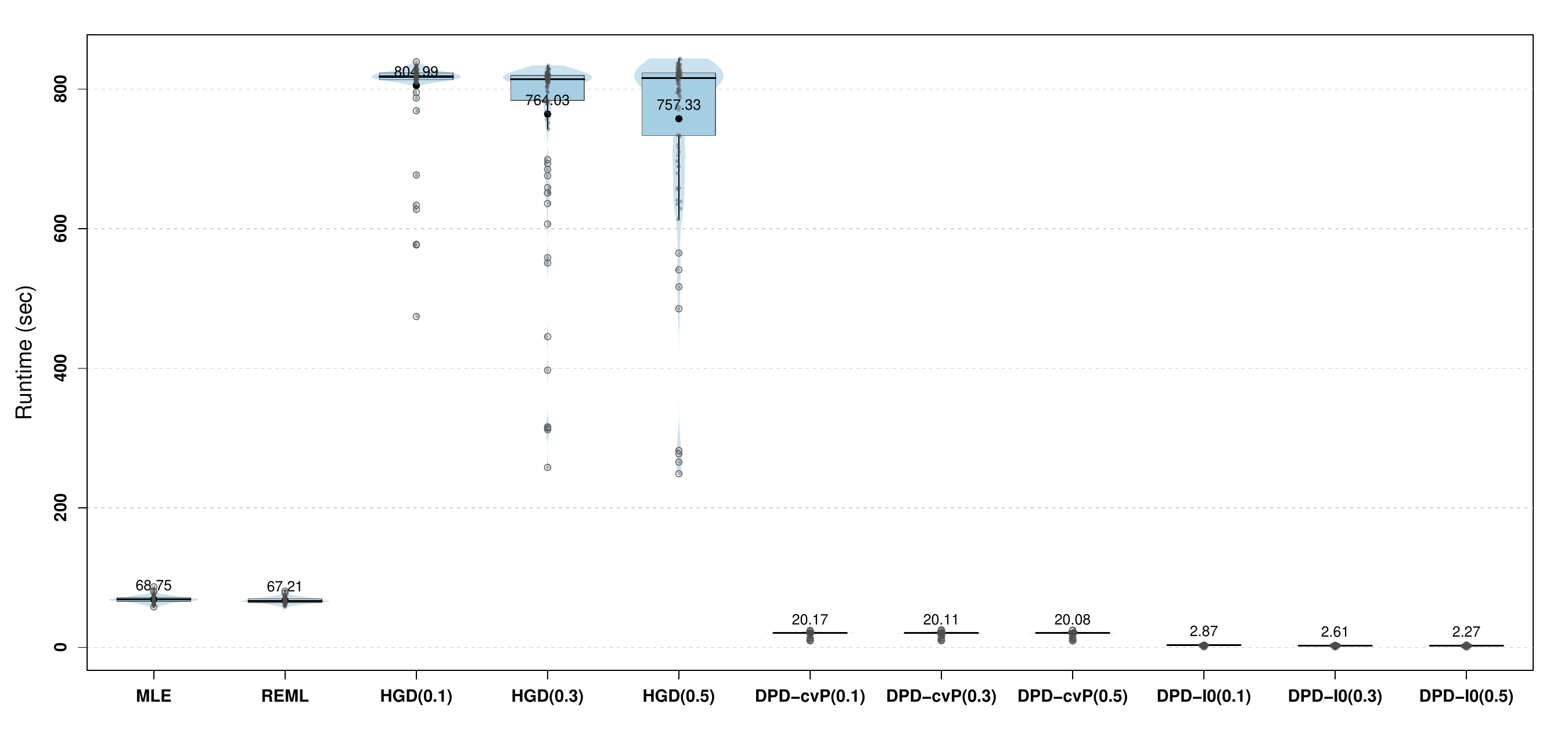}
		\label{FIG:boxplot_Y}}
	~	
	\subfloat[\tiny 20\% (C4) Contamination]{
		\includegraphics[page=3, width=0.3\textwidth]{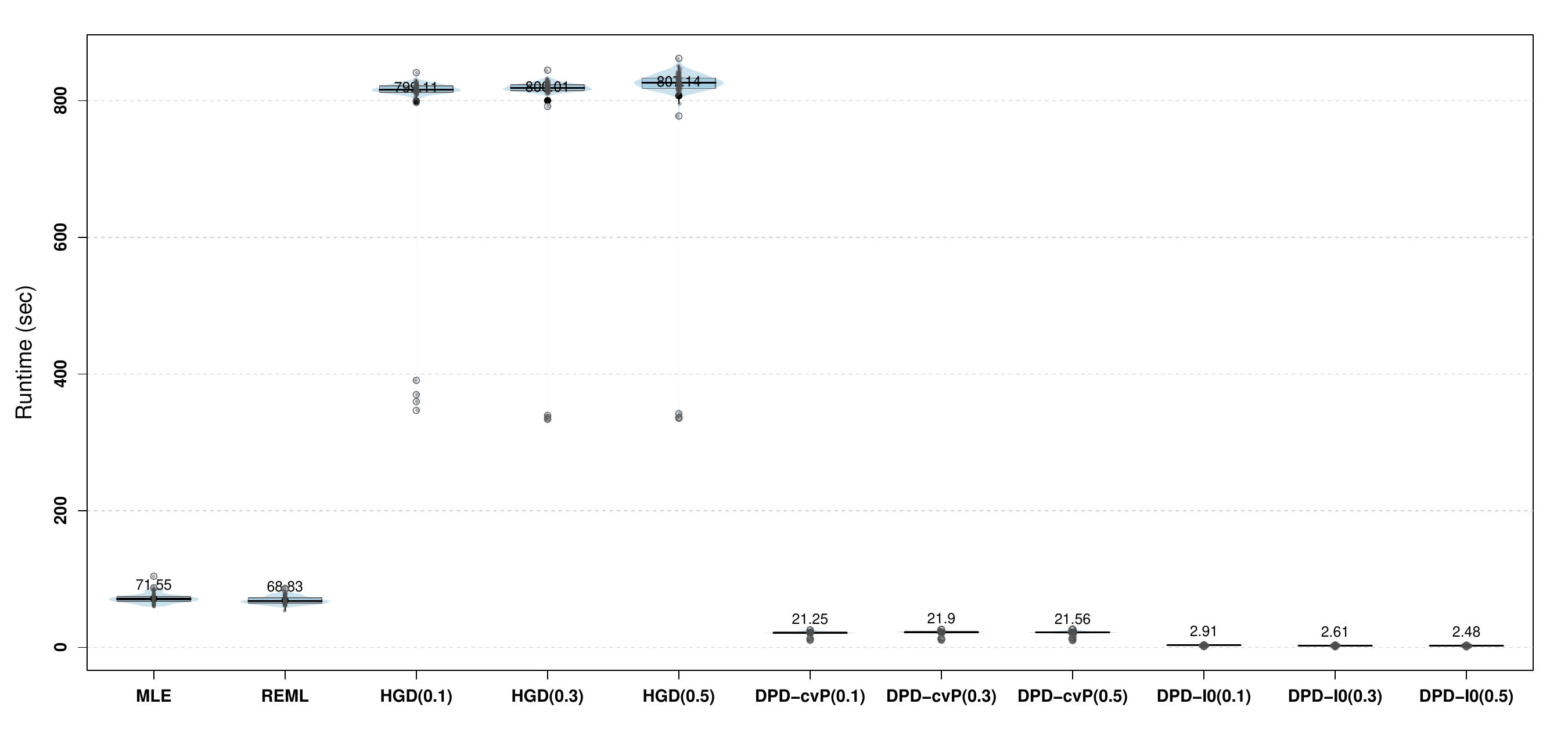}
		\label{FIG:boxplot_Y}}
	\caption{Boxplots  (with overlaid sample means) of computational times (in seconds) of the screening methods  under scenario (S2)$\times$(R1) with $\bm\Sigma_x$ as T(0.3) and different types of contamination}
	\label{FIG:tim_S2R1T}
\end{figure}

\begin{table}[!h]
	\centering
	\caption{Median TPR and EmpSSp (in parenthesis) under contaminated data with (C2)}
	\resizebox{\textwidth}{!}{
		\begin{tabular}{l|cc|ccc|ccc|ccc}\hline
			Settings	&		\multicolumn{5}{|c|}{Benchmark SIS}	&	\multicolumn{3}{c|}{DPD-SISP with cv-P}	&	\multicolumn{3}{c}{DPD-SISP with I0-P}\\	
			$\bm\Sigma_x$	&	MLE	&	REML	&	HGD(0.1) 	&	HGD(0.3)	&	HGD(0.5)	&	($0.1$)	&	($0.3$)	&	($0.5$)	&	($0.1$)	&	($0.3$)	&	($0.5$)	\\	\hline
			\hline
			\multicolumn{12}{c}{\textbf{5\% contamination}} \\																										
			(S1)$\times$(R1)			&	0.8	&	0.8	&	0.8	&	0.8	&	0.8	&	1	&	1	&	0.8	&	0.6	&	0.6	&	0.6	\\	
			Identity	&	(0)	&	(0)	&	(0)	&	(0)	&	(0)	&	(0.65)	&	(0.52)	&	(0.4)	&	(0)	&	(0)	&	(0)	\\	\hline
			(S2)$\times$(R1)			&	1	&	1	&	1	&	1	&	1	&	1	&	1	&	1	&	1	&	1	&	1	\\	
			Identity	&	(1)	&	(1)	&	(1)	&	(1)	&	(1)	&	(0.99)	&	(0.98)	&	(0.96)	&	(0.99)	&	(0.97)	&	(0.95)	\\	\hline\hline
			(S1)$\times$(R1)			&	0.8	&	0.8	&	0.8	&	0.8	&	0.8	&	0.8	&	0.8	&	0.8	&	0.4	&	0.4	&	0.6	\\	
			CS(0.3)	&	(0)	&	(0)	&	(0)	&	(0)	&	(0)	&	(0)	&	(0)	&	(0)	&	(0)	&	(0)	&	(0)	\\	\hline
			(S2)$\times$(R0)			&	1	&	1	&	1	&	1	&	1	&	1	&	1	&	1	&	1	&	1	&	1	\\	
			CS(0.3)	&	(1)	&	(1)	&	(1)	&	(0.99)	&	(1)	&	(0.89)	&	(0.9)	&	(0.85)	&	(0.88)	&	(0.89)	&	(0.88)	\\	\hline
			(S2)$\times$(R1)			&	1	&	1	&	1	&	1	&	1	&	1	&	1	&	1	&	1	&	1	&	1	\\	
			CS(0.3)	&	(1)	&	(1)	&	(1)	&	(1)	&	(1)	&	(0.85)	&	(0.88)	&	(0.84)	&	(0.86)	&	(0.86)	&	(0.84)	\\	\hline\hline
			(S1)$\times$(R1)			&	0.8	&	0.8	&	0.8	&	0.8	&	0.8	&	1	&	1	&	1	&	0.6	&	0.6	&	0.6	\\	
			T(0.3)	&	(0)	&	(0)	&	(0)	&	(0)	&	(0)	&	(0.93)	&	(0.83)	&	(0.7)	&	(0)	&	(0)	&	(0)	\\	\hline
			(S2)$\times$(R0)			&	1	&	1	&	1	&	1	&	1	&	1	&	1	&	1	&	1	&	1	&	1	\\	
			T(0.3)	&	(1)	&	(1)	&	(1)	&	(1)	&	(1)	&	(0.99)	&	(0.98)	&	(0.96)	&	(0.99)	&	(0.99)	&	(0.96)	\\	\hline
			(S2)$\times$(R1)			&	1	&	1	&	1	&	1	&	1	&	1	&	1	&	1	&	1	&	1	&	1	\\	
			T(0.3)	&	(1)	&	(1)	&	(1)	&	(1)	&	(1)	&	(0.99)	&	(0.97)	&	(0.94)	&	(0.99)	&	(0.99)	&	(0.95)	\\	\hline\hline
			\multicolumn{12}{c}{\textbf{10\% contamination}} \\																								
			(S1)$\times$(R1)			&	0.8	&	0.8	&	0.8	&	0.8	&	0.8	&	1	&	1	&	1	&	0.6	&	0.6	&	0.6	\\	
			Identity	&	(0.06)	&	(0.06)	&	(0.05)	&	(0.05)	&	(0.06)	&	(0.87)	&	(0.76)	&	(0.62)	&	(0)	&	(0.02)	&	(0.03)	\\	\hline
			(S2)$\times$(R1)			&	1	&	1	&	1	&	1	&	1	&	1	&	1	&	1	&	1	&	1	&	1	\\	
			Identity	&	(1)	&	(1)	&	(1)	&	(1)	&	(1)	&	(0.99)	&	(0.98)	&	(0.97)	&	(0.99)	&	(0.97)	&	(0.95)	\\	\hline\hline
			(S1)$\times$(R1)			&	0.8	&	0.8	&	0.8	&	0.8	&	0.8	&	0.8	&	0.8	&	0.8	&	0.4	&	0.4	&	0.6	\\	
			CS(0.3)	&	(0)	&	(0)	&	(0)	&	(0)	&	(0)	&	(0)	&	(0)	&	(0)	&	(0)	&	(0.01)	&	(0.01)	\\	\hline
			(S2)$\times$(R0)			&	1	&	1	&	1	&	1	&	1	&	1	&	1	&	1	&	1	&	1	&	1	\\	
			CS(0.3)	&	(1)	&	(1)	&	(1)	&	(0.99)	&	(1)	&	(0.88)	&	(0.89)	&	(0.86)	&	(0.88)	&	(0.89)	&	(0.88)	\\	\hline
			(S2)$\times$(R1)			&	1	&	1	&	1	&	1	&	1	&	1	&	1	&	1	&	1	&	1	&	1	\\	
			CS(0.3)	&	(1)	&	(1)	&	(1)	&	(1)	&	(1)	&	(0.86)	&	(0.88)	&	(0.83)	&	(0.86)	&	(0.86)	&	(0.84)	\\	\hline\hline
			(S1)$\times$(R1)			&	0.8	&	0.8	&	0.8	&	0.8	&	0.8	&	1	&	1	&	1	&	0.6	&	0.6	&	0.6	\\	
			T(0.3)	&	(0.06)	&	(0.06)	&	(0.05)	&	(0.06)	&	(0.05)	&	(0.99)	&	(0.98)	&	(0.93)	&	(0.01)	&	(0.04)	&	(0.04)	\\	\hline(S2)$\times$(R0)			&	1	&	1	&	1	&	1	&	1	&	1	&	1	&	1	&	1	&	1	&	1	\\	
			T(0.3)	&	(1)	&	(1)	&	(1)	&	(1)	&	(1)	&	(0.99)	&	(0.98)	&	(0.96)	&	(0.99)	&	(0.99)	&	(0.96)	\\	\hline
			(S2)$\times$(R1)			&	1	&	1	&	1	&	1	&	1	&	1	&	1	&	1	&	1	&	1	&	1	\\	
			T(0.3)	&	(1)	&	(1)	&	(1)	&	(1)	&	(1)	&	(0.99)	&	(0.97)	&	(0.94)	&	(0.99)	&	(0.99)	&	(0.95)	\\	\hline\hline
			\multicolumn{12}{c}{\textbf{20\% contamination}} \\																									
			(S1)$\times$(R1)			&	0.8	&	0.8	&	0.8	&	0.8	&	0.8	&	1	&	1	&	0.8	&	0.6	&	0.6	&	0.6	\\	
			Identity	&	(0)	&	(0)	&	(0)	&	(0)	&	(0)	&	(0.65)	&	(0.52)	&	(0.4)	&	(0)	&	(0)	&	(0)	\\	\hline
			(S2)$\times$(R1)			&	1	&	1	&	1	&	1	&	1	&	1	&	1	&	1	&	1	&	1	&	1	\\	
			Identity	&	(1)	&	(1)	&	(1)	&	(1)	&	(1)	&	(0.99)	&	(0.98)	&	(0.96)	&	(0.99)	&	(0.97)	&	(0.95)	\\	\hline\hline
			(S1)$\times$(R1)			&	0.8	&	0.8	&	0.8	&	0.8	&	0.8	&	0.8	&	0.8	&	0.8	&	0.4	&	0.4	&	0.6	\\	
			CS(0.3)	&	(0)	&	(0)	&	(0)	&	(0)	&	(0)	&	(0)	&	(0)	&	(0)	&	(0)	&	(0)	&	(0)	\\	\hline
			(S2)$\times$(R0)			&	1	&	1	&	1	&	1	&	1	&	1	&	1	&	1	&	1	&	1	&	1	\\	
			CS(0.3)	&	(1)	&	(1)	&	(1)	&	(0.99)	&	(1)	&	(0.89)	&	(0.9)	&	(0.85)	&	(0.88)	&	(0.89)	&	(0.88)	\\	\hline
			(S2)$\times$(R1)			&	1	&	1	&	1	&	1	&	1	&	1	&	1	&	1	&	1	&	1	&	1	\\	
			CS(0.3)	&	(1)	&	(1)	&	(1)	&	(1)	&	(1)	&	(0.85)	&	(0.88)	&	(0.84)	&	(0.86)	&	(0.86)	&	(0.84)	\\	\hline\hline
			(S1)$\times$(R1)			&	0.8	&	0.8	&	0.8	&	0.8	&	0.8	&	1	&	1	&	1	&	0.6	&	0.6	&	0.6	\\	
			T(0.3)	&	(0)	&	(0)	&	(0)	&	(0)	&	(0)	&	(0.93)	&	(0.83)	&	(0.7)	&	(0)	&	(0)	&	(0)	\\	\hline
			(S2)$\times$(R0)			&	1	&	1	&	1	&	1	&	1	&	1	&	1	&	1	&	1	&	1	&	1	\\	
			T(0.3)	&	(1)	&	(1)	&	(1)	&	(1)	&	(1)	&	(0.99)	&	(0.98)	&	(0.96)	&	(0.99)	&	(0.99)	&	(0.96)	\\	\hline
			(S2)$\times$(R1)			&	1	&	1	&	1	&	1	&	1	&	1	&	1	&	1	&	1	&	1	&	1	\\	
			T(0.3)	&	(1)	&	(1)	&	(1)	&	(1)	&	(1)	&	(0.99)	&	(0.97)	&	(0.94)	&	(0.99)	&	(0.99)	&	(0.95)	\\	\hline
			\hline
		\end{tabular}
	}
	\label{TAB:C2}
\end{table}


\begin{table}[!h]
	\centering
	\caption{Median TPR and EmpSSp (in parenthesis) under contaminated data with (C3)}
	\resizebox{\textwidth}{!}{
		\begin{tabular}{l|cc|ccc|ccc|ccc}\hline
			Settings	&		\multicolumn{5}{|c|}{Benchmark SIS}	&	\multicolumn{3}{c|}{DPD-SISP with cv-P}	&	\multicolumn{3}{c}{DPD-SISP with I0-P}\\	
			$\bm\Sigma_x$	&	MLE	&	REML	&	HGD(0.1) 	&	HGD(0.3)	&	HGD(0.5)	&	($0.1$)	&	($0.3$)	&	($0.5$)	&	($0.1$)	&	($0.3$)	&	($0.5$)	\\	\hline
			\hline
			\multicolumn{12}{c}{\textbf{5\% contamination}} \\																										
			(S1)$\times$(R1)			&	0.6	&	0.6	&	1	&	1	&	1	&	1	&	1	&	1	&	0.4	&	0.8	&	0.8	\\	
			Identity	&	(0.09)	&	(0.09)	&	(0.56)	&	(0.57)	&	(0.59)	&	(0.53)	&	(0.86)	&	(0.81)	&	(0.06)	&	(0.22)	&	(0.22)	\\	\hline
			(S2)$\times$(R1)			&	0.6	&	0.6	&	1	&	1	&	1	&	1	&	1	&	1	&	1	&	1	&	1	\\	
			Identity	&	(0)	&	(0)	&	(1)	&	(1)	&	(1)	&	(0.58)	&	(0.95)	&	(0.9)	&	(0.54)	&	(0.95)	&	(0.88)	\\	\hline\hline
			(S1)$\times$(R1)			&	0.8	&	0.8	&	1	&	1	&	1	&	0.6	&	0.8	&	0.8	&	0.4	&	0.4	&	0.6	\\	
			CS(0.3)	&	(0.02)	&	(0.01)	&	(0.54)	&	(0.59)	&	(0.6)	&	(0.02)	&	(0.42)	&	(0.36)	&	(0.02)	&	(0.05)	&	(0.05)	\\	\hline
			(S2)$\times$(R0)			&	0.8	&	0.8	&	1	&	1	&	1	&	1	&	1	&	1	&	1	&	1	&	1	\\	
			CS(0.3)	&	(0.02)	&	(0.01)	&	(1)	&	(0.99)	&	(0.99)	&	(0.77)	&	(0.85)	&	(0.84)	&	(0.81)	&	(0.85)	&	(0.83)	\\	\hline
			(S2)$\times$(R1)			&	0.6	&	0.6	&	1	&	1	&	1	&	0.8	&	1	&	1	&	0.8	&	1	&	1	\\	
			CS(0.3)	&	(0)	&	(0)	&	(0.87)	&	(0.93)	&	(0.89)	&	(0.31)	&	(0.78)	&	(0.73)	&	(0.36)	&	(0.76)	&	(0.71)	\\	\hline\hline
			(S1)$\times$(R1)			&	0.8	&	0.8	&	1	&	1	&	1	&	0.8	&	1	&	1	&	0.6	&	0.8	&	0.8	\\	
			T(0.3)	&	(0.46)	&	(0.46)	&	(0.69)	&	(0.63)	&	(0.62)	&	(0.31)	&	(0.99)	&	(0.97)	&	(0.14)	&	(0.33)	&	(0.33)	\\	\hline
			(S2)$\times$(R0)			&	0.8	&	0.8	&	1	&	1	&	1	&	1	&	1	&	1	&	1	&	1	&	1	\\	
			T(0.3)	&	(0.44)	&	(0.44)	&	(1)	&	(1)	&	(1)	&	(0.9)	&	(0.99)	&	(0.98)	&	(0.94)	&	(0.99)	&	(0.98)	\\	\hline
			(S2)$\times$(R1)			&	0.6	&	0.6	&	1	&	1	&	1	&	1	&	1	&	1	&	1	&	1	&	1	\\	
			T(0.3)	&	(0)	&	(0)	&	(1)	&	(1)	&	(1)	&	(0.54)	&	(0.92)	&	(0.86)	&	(0.59)	&	(0.91)	&	(0.88)	\\	\hline\hline
			\multicolumn{12}{c}{\textbf{10\% contamination}} \\																								
			(S1)$\times$(R1)			&	0.4	&	0.4	&	0.6	&	1	&	0.8	&	0.8	&	1	&	1	&	0.4	&	0.8	&	0.8	\\	
			Identity	&	(0)	&	(0)	&	(0.02)	&	(0.55)	&	(0.49)	&	(0.3)	&	(0.83)	&	(0.79)	&	(0.04)	&	(0.23)	&	(0.21)	\\	\hline
			(S2)$\times$(R1)			&	0.2	&	0.2	&	0.6	&	1	&	1	&	0.6	&	1	&	1	&	0.6	&	1	&	1	\\	
			Identity	&	(0)	&	(0)	&	(0)	&	(1)	&	(1)	&	(0.03)	&	(0.91)	&	(0.9)	&	(0.03)	&	(0.92)	&	(0.91)	\\	\hline\hline
			(S1)$\times$(R1)			&	0.4	&	0.4	&	0.7	&	1	&	1	&	0.2	&	0.8	&	0.8	&	0.4	&	0.6	&	0.6	\\	
			CS(0.3)	&	(0)	&	(0)	&	(0.08)	&	(0.61)	&	(0.62)	&	(0.02)	&	(0.36)	&	(0.37)	&	(0)	&	(0.06)	&	(0.06)	\\	\hline
			(S2)$\times$(R0)			&	0.6	&	0.6	&	0.8	&	1	&	1	&	0.8	&	1	&	1	&	0.8	&	1	&	1	\\	
			CS(0.3)	&	(0)	&	(0)	&	(0.49)	&	(0.99)	&	(0.99)	&	(0.47)	&	(0.89)	&	(0.86)	&	(0.49)	&	(0.91)	&	(0.85)	\\	\hline
			(S2)$\times$(R1)			&	0.2	&	0.2	&	0.6	&	1	&	1	&	0.6	&	1	&	1	&	0.6	&	1	&	1	\\	
			CS(0.3)	&	(0)	&	(0)	&	(0)	&	(0.92)	&	(0.84)	&	(0)	&	(0.7)	&	(0.68)	&	(0)	&	(0.79)	&	(0.72)	\\	\hline
			\hline
			(S1)$\times$(R1)			&	0.4	&	0.4	&	0.8	&	1	&	1	&	1	&	1	&	1	&	0.6	&	0.8	&	0.8	\\	
			T(0.3)	&	(0)	&	(0)	&	(0.16)	&	(0.65)	&	(0.64)	&	(0.82)	&	(0.97)	&	(0.99)	&	(0.04)	&	(0.16)	&	(0.17)	\\	\hline
			(S2)$\times$(R0)			&	0.8	&	0.8	&	1	&	1	&	1	&	1	&	1	&	1	&	1	&	1	&	1	\\	
			T(0.3)	&	(0.07)	&	(0.07)	&	(0.92)	&	(1)	&	(1)	&	(0.73)	&	(0.98)	&	(0.97)	&	(0.74)	&	(0.98)	&	(0.97)	\\	\hline
			(S2)$\times$(R1)			&	0.2	&	0.2	&	0.6	&	1	&	1	&	0.6	&	1	&	1	&	0.6	&	1	&	1	\\	
			T(0.3)	&	(0)	&	(0)	&	(0)	&	(1)	&	(1)	&	(0.02)	&	(0.92)	&	(0.88)	&	(0.02)	&	(0.9)	&	(0.86)	\\	\hline\hline
			\multicolumn{12}{c}{\textbf{20\% contamination}} \\																									
			(S1)$\times$(R1)			&	0.4	&	0.4	&	0.4	&	0.8	&	0.8	&	0.6	&	1	&	1	&	0.4	&	0.6	&	0.8	\\	
			Identity	&	(0)	&	(0)	&	(0)	&	(0.44)	&	(0.47)	&	(0.23)	&	(0.57)	&	(0.67)	&	(0)	&	(0.08)	&	(0.39)	\\	\hline
			(S2)$\times$(R1)			&	0.2	&	0.2	&	0.2	&	0.6	&	1	&	0.2	&	0.6	&	1	&	0.2	&	0.6	&	1	\\	
			Identity	&	(0)	&	(0)	&	(0)	&	(0.15)	&	(0.99)	&	(0)	&	(0.01)	&	(0.86)	&	(0)	&	(0)	&	(0.88)	\\	\hline\hline
			(S1)$\times$(R1)			&	0.4	&	0.4	&	0.6	&	0.9	&	0.8	&	0.6	&	0.6	&	0.8	&	0.4	&	0.6	&	0.7	\\	
			CS(0.3)	&	(0)	&	(0)	&	(0)	&	(0.5)	&	(0.48)	&	(0.05)	&	(0.19)	&	(0.25)	&	(0)	&	(0.08)	&	(0.12)	\\	\hline
			(S2)$\times$(R0)			&	0.6	&	0.6	&	0.6	&	1	&	1	&	0.6	&	0.8	&	1	&	0.6	&	0.8	&	1	\\	
			CS(0.3)	&	(0)	&	(0)	&	(0)	&	(0.51)	&	(0.92)	&	(0)	&	(0.03)	&	(0.72)	&	(0)	&	(0.14)	&	(0.73)	\\	\hline
			(S2)$\times$(R1)			&	0.2	&	0.2	&	0.2	&	0.8	&	1	&	0.2	&	0.6	&	0.8	&	0.2	&	0.6	&	1	\\	
			CS(0.3)	&	(0)	&	(0)	&	(0)	&	(0.29)	&	(0.67)	&	(0)	&	(0)	&	(0.38)	&	(0)	&	(0.04)	&	(0.55)	\\	\hline			\hline
			(S1)$\times$(R1)			&	0.4	&	0.4	&	0.6	&	1	&	1	&	0.8	&	1	&	1	&	0.4	&	0.6	&	0.8	\\	
			T(0.3)	&	(0)	&	(0)	&	(0)	&	(0.61)	&	(0.62)	&	(0.11)	&	(0.93)	&	(0.96)	&	(0)	&	(0.11)	&	(0.25)	\\	\hline
			(S2)$\times$(R0)			&	0.6	&	0.6	&	0.6	&	1	&	1	&	0.6	&	0.6	&	1	&	0.6	&	0.8	&	1	\\	
			T(0.3)	&	(0)	&	(0)	&	(0)	&	(0.63)	&	(1)	&	(0)	&	(0.03)	&	(0.84)	&	(0)	&	(0.06)	&	(0.86)	\\	\hline
			(S2)$\times$(R1)			&	0.2	&	0.2	&	0.2	&	0.8	&	1	&	0.2	&	0.6	&	1	&	0.2	&	0.6	&	1	\\	
			T(0.3)	&	(0)	&	(0)	&	(0)	&	(0.36)	&	(1)	&	(0)	&	(0)	&	(0.86)	&	(0)	&	(0)	&	(0.85)	\\	\hline\hline
		\end{tabular}
	}
	\label{TAB:C3}
\end{table}


\begin{table}
	\centering
	\caption{Median TPR and EmpSSp (in parenthesis) under contaminated data with (C4)}
	\resizebox{\textwidth}{!}{
		\begin{tabular}{l|cc|ccc|ccc|ccc}\hline
			Settings	&		\multicolumn{5}{|c|}{Benchmark SIS}	&	\multicolumn{3}{c|}{DPD-SISP with cv-P}	&	\multicolumn{3}{c}{DPD-SISP with I0-P}\\	
			$\bm\Sigma_x$	&	MLE	&	REML	&	HGD(0.1) 	&	HGD(0.3)	&	HGD(0.5)	&	($0.1$)	&	($0.3$)	&	($0.5$)	&	($0.1$)	&	($0.3$)	&	($0.5$)	\\	\hline
			\hline
			\multicolumn{12}{c}{\textbf{5\% contamination}} \\																										
			(S1)$\times$(R1)			&	0.8	&	0.8	&	1	&	1	&	1	&	1	&	1	&	1	&	0.8	&	0.8	&	0.8	\\	
			Identity	&	(0.36)	&	(0.38)	&	(0.63)	&	(0.69)	&	(0.68)	&	(0.79)	&	(0.79)	&	(0.75)	&	(0.14)	&	(0.29)	&	(0.25)	\\	\hline
			(S2)$\times$(R1)			&	0.8	&	0.8	&	1	&	1	&	1	&	1	&	1	&	1	&	0.8	&	1	&	1	\\	
			Identity	&	(0.14)	&	(0.15)	&	(0.92)	&	(1)	&	(1)	&	(0.56)	&	(0.94)	&	(0.92)	&	(0.39)	&	(0.89)	&	(0.84)	\\	\hline\hline
			(S1)$\times$(R1)			&	0.8	&	0.8	&	0.8	&	1	&	1	&	0.8	&	0.9	&	0.8	&	0.6	&	0.6	&	0.6	\\	
			CS(0.3)	&	(0.24)	&	(0.23)	&	(0.31)	&	(0.61)	&	(0.65)	&	(0.16)	&	(0.5)	&	(0.47)	&	(0)	&	(0)	&	(0.03)	\\	\hline
			(S2)$\times$(R0)			&	0.8	&	0.8	&	1	&	1	&	1	&	0.8	&	0.8	&	0.8	&	0.8	&	0.8	&	0.9	\\	
			CS(0.3)	&	(0.01)	&	(0.01)	&	(0.64)	&	(1)	&	(0.99)	&	(0.29)	&	(0.45)	&	(0.49)	&	(0.16)	&	(0.45)	&	(0.5)	\\	\hline
			(S2)$\times$(R1)			&	0.6	&	0.6	&	0.8	&	1	&	1	&	0.8	&	1	&	1	&	0.8	&	1	&	0.9	\\	
			CS(0.3)	&	(0)	&	(0)	&	(0.13)	&	(0.72)	&	(0.77)	&	(0.31)	&	(0.77)	&	(0.73)	&	(0.12)	&	(0.57)	&	(0.5)	\\	\hline\hline
			(S1)$\times$(R1)			&	1	&	1	&	1	&	1	&	1	&	1	&	1	&	1	&	0.6	&	0.8	&	0.8	\\	
			T(0.3)	&	(0.71)	&	(0.68)	&	(0.84)	&	(0.8)	&	(0.78)	&	(0.95)	&	(0.96)	&	(0.95)	&	(0.13)	&	(0.24)	&	(0.25)	\\	\hline
			(S2)$\times$(R0)			&	0.8	&	0.8	&	1	&	1	&	1	&	1	&	1	&	1	&	0.8	&	1	&	1	\\	
			T(0.3)	&	(0.32)	&	(0.33)	&	(0.91)	&	(1)	&	(1)	&	(0.52)	&	(0.73)	&	(0.79)	&	(0.48)	&	(0.74)	&	(0.82)	\\	\hline
			(S2)$\times$(R1)			&	0.8	&	0.8	&	1	&	1	&	1	&	1	&	1	&	1	&	0.8	&	1	&	1	\\	
			T(0.3)	&	(0.1)	&	(0.09)	&	(0.89)	&	(1)	&	(1)	&	(0.53)	&	(0.9)	&	(0.87)	&	(0.33)	&	(0.83)	&	(0.81)	\\	\hline\hline
			\multicolumn{12}{c}{\textbf{10\% contamination}} \\																								
			(S1)$\times$(R1)			&	0.6	&	0.6	&	0.8	&	0.8	&	0.8	&	0.6	&	1	&	1	&	0.4	&	0.6	&	0.8	\\	
			Identity	&	(0.01)	&	(0.01)	&	(0.17)	&	(0.48)	&	(0.48)	&	(0.02)	&	(0.63)	&	(0.64)	&	(0)	&	(0.08)	&	(0.18)	\\	\hline
			(S2)$\times$(R1)			&	0.8	&	0.8	&	0.8	&	1	&	1	&	0.6	&	1	&	1	&	0.8	&	1	&	1	\\	
			Identity	&	(0)	&	(0.01)	&	(0.06)	&	(0.91)	&	(0.96)	&	(0.05)	&	(0.82)	&	(0.87)	&	(0.14)	&	(0.61)	&	(0.74)	\\	\hline\hline
			(S1)$\times$(R1)			&	0.6	&	0.6	&	0.6	&	1	&	1	&	0.6	&	0.8	&	0.8	&	0.4	&	0.6	&	0.6	\\	
			CS(0.3)	&	(0)	&	(0)	&	(0.01)	&	(0.52)	&	(0.63)	&	(0)	&	(0.3)	&	(0.48)	&	(0)	&	(0.01)	&	(0.02)	\\	\hline
			(S2)$\times$(R0)			&	0.6	&	0.6	&	0.8	&	1	&	1	&	0.6	&	0.8	&	0.8	&	0.8	&	0.8	&	0.8	\\	
			CS(0.3)	&	(0)	&	(0)	&	(0.06)	&	(0.87)	&	(0.94)	&	(0.09)	&	(0.36)	&	(0.32)	&	(0.11)	&	(0.36)	&	(0.38)	\\	\hline
			(S2)$\times$(R1)			&	0.4	&	0.4	&	0.6	&	0.8	&	0.8	&	0.6	&	1	&	1	&	0.6	&	0.8	&	0.8	\\	
			CS(0.3)	&	(0)	&	(0)	&	(0)	&	(0.13)	&	(0.22)	&	(0)	&	(0.52)	&	(0.55)	&	(0)	&	(0.11)	&	(0.18)	\\	\hline\hline
			(S1)$\times$(R1)			&	0.8	&	0.8	&	1	&	1	&	1	&	0.8	&	1	&	1	&	0.4	&	0.8	&	0.8	\\	
			T(0.3)	&	(0.12)	&	(0.12)	&	(0.64)	&	(0.83)	&	(0.75)	&	(0.09)	&	(0.96)	&	(0.93)	&	(0)	&	(0.1)	&	(0.13)	\\	\hline
			(S2)$\times$(R0)			&	0.8	&	0.8	&	1	&	1	&	1	&	0.8	&	1	&	1	&	0.8	&	1	&	1	\\	
			T(0.3)	&	(0.08)	&	(0.09)	&	(0.54)	&	(0.99)	&	(0.99)	&	(0.43)	&	(0.59)	&	(0.57)	&	(0.38)	&	(0.56)	&	(0.58)	\\	\hline
			(S2)$\times$(R1)			&	0.6	&	0.6	&	0.8	&	1	&	1	&	0.6	&	1	&	1	&	0.6	&	0.8	&	1	\\	
			T(0.3)	&	(0.01)	&	(0.01)	&	(0.02)	&	(0.88)	&	(0.97)	&	(0.04)	&	(0.7)	&	(0.83)	&	(0.08)	&	(0.48)	&	(0.7)	\\	\hline\hline
			\multicolumn{12}{c}{\textbf{20\% contamination}} \\																									
			(S1)$\times$(R1)			&	0.8	&	0.8	&	0.8	&	0.8	&	0.8	&	0.8	&	1	&	1	&	0.6	&	0.6	&	0.8	\\	
			Identity	&	(0.18)	&	(0.15)	&	(0.17)	&	(0.3)	&	(0.29)	&	(0.34)	&	(0.77)	&	(0.85)	&	(0)	&	(0.09)	&	(0.21)	\\	\hline
			(S2)$\times$(R1)			&	0.6	&	0.6	&	0.6	&	1	&	1	&	0.6	&	0.8	&	1	&	0.8	&	0.8	&	0.8	\\
			Identity	&	(0)	&	(0)	&	(0)	&	(0.63)	&	(0.78)	&	(0)	&	(0.06)	&	(0.52)	&	(0)	&	(0.12)	&	(0.41)	\\	\hline\hline
			(S1)$\times$(R1)			&	0.8	&	0.8	&	0.8	&	1	&	0.8	&	1	&	1	&	1	&	0.6	&	0.6	&	0.8	\\	
			CS(0.3)	&	(0.34)	&	(0.34)	&	(0.27)	&	(0.51)	&	(0.35)	&	(0.51)	&	(0.83)	&	(0.68)	&	(0)	&	(0)	&	(0.01)	\\	\hline
			(S2)$\times$(R0)			&	0.6	&	0.6	&	0.6	&	0.8	&	0.8	&	0.6	&	0.6	&	0.8	&	0.6	&	0.6	&	0.8	\\	
			CS(0.3)	&	(0)	&	(0)	&	(0)	&	(0.01)	&	(0.4)	&	(0)	&	(0.05)	&	(0.29)	&	(0)	&	(0)	&	(0.05)	\\	\hline
			(S2)$\times$(R1)			&	0.6	&	0.6	&	0.6	&	0.6	&	0.6	&	0.4	&	0.6	&	0.8	&	0.6	&	0.6	&	0.6	\\	
			CS(0.3)	&	(0)	&	(0)	&	(0)	&	(0)	&	(0)	&	(0)	&	(0)	&	(0.22)	&	(0)	&	(0)	&	(0.03)	\\	\hline			\hline
			(S1)$\times$(R1)			&	1	&	1	&	1	&	1	&	1	&	1	&	1	&	1	&	0.6	&	0.8	&	0.8	\\	
			T(0.3)	&	(0.66)	&	(0.65)	&	(0.68)	&	(0.85)	&	(0.82)	&	(0.76)	&	(0.95)	&	(0.91)	&	(0)	&	(0.21)	&	(0.41)	\\	\hline
			(S2)$\times$(R0)			&	0.6	&	0.6	&	0.6	&	0.8	&	0.8	&	0.6	&	0.6	&	0.8	&	0.6	&	0.6	&	0.8	\\	
			T(0.3)	&	(0)	&	(0)	&	(0)	&	(0.01)	&	(0.34)	&	(0)	&	(0.02)	&	(0.33)	&	(0)	&	(0.01)	&	(0.04)	\\	\hline
			(S2)$\times$(R1)			&	0.6	&	0.6	&	0.6	&	0.8	&	1	&	0.6	&	0.8	&	0.8	&	0.6	&	0.8	&	0.8	\\	
			T(0.3)	&	(0)	&	(0)	&	(0)	&	(0.42)	&	(0.6)	&	(0)	&	(0.04)	&	(0.28)	&	(0)	&	(0.07)	&	(0.3)	\\	\hline
		\end{tabular}
	}
	\label{TAB:C4}
\end{table}

\begin{small}
	
	\begin{landscape}
		\begin{longtable}{l|L{0.1\textheight}|L{0.15\textheight}|L{0.1\textheight}|L{0.4\textheight}|L{0.3\textwidth}|L{0.2\textheight}		}
			\caption{List of genes and their PR scores (in parenthesis) across different groups of association, for the three screened sets}\\
			
			\hline
			Gene Set & G1 & G2 & G3 & G4 & G5 & G6\\
			\hline
			\endfirsthead
			
			\multicolumn{5}{c}{\tablename\ \thetable{} -- Continued from previous page} \\
			\hline
			Gene Set & G1 & G2 & G3 & G4 & G5 & G6\\
			\hline
			\endhead
			
			
			\hline
			\endlastfoot
			
			\hline
			\texttt{Gene\_Set\_0}		&	INSR (10.78) 	&	HMOX1 (8.25), 	~~MSR1 (6.71), 	PCSK9 (6.59), 	~~CHGB (6.29), 	GRM1 (6.16), 	GRIA3 (6.15), 	~~AKT2 (6.07), 	KCNMA1 (6.07), 	IL16 (5.24), 	TAT (5.22), 	GRK5 (4.96), 	DRD5 (4.18), 	NR4A1 (4.11), 	MEGF10 (3.71), 	PDE10A (3.64), 	PIAS2 (3.58), 	NLGN4X (3.45) 	&	TANGO2 (3.06), 	ARL13B (2.96), 	RREB1 (2.51), 	NFIX (2.35), 	SLC13A5 (1.97), 	SLC12A1 (1.97), 	RET (1.97), 	CHEK1 (1.96), 	USF2 (1.96), 	FTL (1.96), 	CLTCL1 (1.75) 	&	KIF3A (1.67), 	ARRB1 (1.65), 	RAP1GAP2 (1.54), 	CSE1L (1.44), 	HNRNPLL (1.29), 	SULF1 (1.2), 	SRD5A1 (1.18), 	DECR1 (1.17), 	PRUNE2 (1.16), 	HSPA14 (1.13), 	RPL9 (1.1), 	ZEB1 (1.09), 	MASTL (1.07), 	NEURL1 (1.07), 	NR4A3 (1.07), 	PLRG1 (0.93), 	RHOB (0.93), 	HBG2 (0.89), 	NAA25 (0.89), 	RER1 (0.85), 	SYPL2 (0.81), 	TNFRSF25 (0.8), 	RAD1 (0.78), 	BTBD1 (0.78), 	MPC2 (0.78), 	NBPF4 (0.76), 	CEP70 (0.75), 	NR2C1 (0.71), 	MCF2L2 (0.71), 	HDDC2 (0.71), 	DHX57 (0.71), 	RASAL2 (0.71), 	GFI1B (0.71), 	SMOX (0.65), 	MDK (0.58), 	ADM5 (0.53), 	SRL (0.53), 	ETF1 (0.53), 	BCCIP (0.53), 	OSBP2 (0.53), 	PARD6G (0.53), 	EXD3 (0.53), 	PLEKHG5 (0.52), 	HMGA2 (0.51), 	LMOD1 (0.47), 	SEMA6A (0.47), 	SLC4A5 (0.47), 	PRR5 (0.47), 	MAGEA4 (0.47), 	CXADR (0.44), 	ARFGAP3 (0.41), 	RNF41 (0.41), 	SNX29 (0.38), 	CDH19 (0.38), 	KISS1R (0.38), 	CFHR2 (0.36), 	ERMN (0.36), 	ZNF75A (0.36), 	HBG1 (0.36), 	VSIG4 (0.36), 	PIM1 (0.36), 	LRCH3 (0.29), 	UBAP1 (0.27), 	FANCC (0.19), 	MAP3K6 (0.13), 	DNASE1 (0.13), 	ZBTB8OS (0.09), 	BCL2L10 (0.09), 	ARHGEF26 (0.09), 	FBXO45 (0.09), 	RUNX1T1 (0.09), 	BMF (0.09), 	NXNL1 (0.09) 	&	CASP10 (4.49), 	TAL1 (3.48), 	STAP2 (2.62), 	RAB37 (2.14), 	RORC (2.1), 	COPS8 (1.86), 	TRIM29 (1.82), 	AREL1 (1.71), 	GDF6 (1.44), 	RUNDC3A (1.41), 	SHISA4 (1.4), 	ABTB2 (1.36), 	KRTAP4-2 (1.22), 	AKR1C2 (1.2), 	GPR89A (1.18), 	CLGN (1.14), 	HPDL (1.13), 	JMY (1.12), 	PABPC1L (1.09), 	OAZ2 (1.07), 	HIPK4 (1.03), 	AGXT2 (1.01), 	FAM136A (1), 	OR5M3 (0.96), 	MOV10L1 (0.89), 	ALDH3B2 (0.85), 	AK6 (0.84), 	RBMXL3 (0.78), 	PARPBP (0.73), 	DNAJB3 (0.67), 	LRRC28 (0.67), 	FAXDC2 (0.67), 	CCDC80 (0.55), 	LDLRAD1 (0.54), 	GPR62 (0.52), 	LRRIQ1 (0.51), 	GCAT (0.45), 	SMIM5 (0.44), 	SNORD97 (0.29), 	COLCA1 (0.23), 	C22orf42 (0.21), 	SDHAP2 (0.04), 	ZNF252P (0.04) 	&	C15ORF59,	RARS,	METTL7A,	C1ORF112,	C20ORF26,	C19ORF77,	RPL36AP33,	C2ORF90,	EFCAB4B,	TMEM189,	NARFL,	C11ORF45,	TARS,	HIST1H2BG,	LOC400590,	C9ORF69,	ENSG00000249882 \\\hline
			

			\texttt{Gene\_Set\_1(cv)}		&	LMNB1 (22.08), 	CBS (9.96), 	BCL2 (8.97) 	&	DNM1 (6.94), 	PDYN (6.29), 	TGIF1 (6.11), 	PCMT1 (5.61), 	TCN2 (4.76), 	C1R (4.57), 	CTNNA1 (4.47), 	CLIC1 (4.38), 	S100A12 (4.09), 	TGFBR2 (3.97), 	CYP1A1 (3.96), ~~~~	HNRNPDL (3.87), 	MEGF10 (3.71), 	ANK2 (3.49), 	ABCG2 (3.49), 	~~CASK (3.35) 	&	FMNL1 (3.29), 	RAB27A (3.16), 	LYRM7 (2.71), 	ASL (2.6), 	RREB1 (2.51), 	AFF3 (2.51), 	COX14 (2.44), 	WDR45B (2.33), 	PMM2 (2.33), 	CTBP2 (2.21), 	HMGN1 (2.2), 	LGI4 (1.97), 	ALG13 (1.97), 	EDC3 (1.97), 	RDX (1.97) 	&	TSPAN13 (1.73), 	TPD52 (1.49), 	RASSF4 (1.47), 	BLNK (1.44), 	PDE7A (1.42), 	SRPK1 (1.38), 	HVCN1 (1.3), 	CCDC50 (1.29), 	KNDC1 (1.28), 	SORCS2 (1.28), 	PTPRK (1.24), 	FXR2 (1.22), 	TAOK2 (1.22), 	POR (1.17), 	PRUNE2 (1.16), 	LATS2 (1.09), 	CERCAM (1.09), 	SWSAP1 (1.09), 	OSBPL1A (1.09), 	OLAH (1.07), 	HERC4 (1.07), 	TBX2 (1.01), 	ACAA1 (1.01), 	DUSP3 (1), 	HDAC10 (0.98), 	FUBP3 (0.89), 	FRYL (0.89), 	MMP16 (0.86), 	DDX3Y (0.86), 	GTF3C3 (0.75), 	PPFIBP1 (0.75), 	DTWD1 (0.74), 	DZIP1 (0.71), 	ZNF331 (0.71), 	WBP11 (0.71), 	UHRF1 (0.71), 	PLA2G2D (0.71), 	TSC22D1 (0.66), 	VAMP3 (0.54), 	ST6GAL2 (0.53), 	IQGAP2 (0.53), 	LIMK2 (0.53), 	NAB2 (0.51), 	RASSF5 (0.47), 	TNFSF13 (0.47), 	KCTD8 (0.47), 	SHISA9 (0.47), 	MAGEA4 (0.47), 	CXADR (0.44), 	SNX22 (0.38), 	ZCCHC7 (0.38), 	TNFSF12-TNFSF13 (0.38), 	SH3RF2 (0.36), 	HBS1L (0.36), 	KIAA0930 (0.36), 	FANCA (0.19), 	GLIS2 (0.13), 	LOXHD1 (0.13), 	NEK6 (0.09), 	STRN4 (0.09), 	SLC45A2 (0.09), 	ANKHD1 (0.09), 	IL27 (0.09), 	PSMB11 (0.09), 	UBE2J1 (0.09)	&	S100A2 (2.53), 	POMP (2.35), 	P2RY10 (2), 	LZTFL1 (1.98), 	UTY (1.77), 	FUCA2 (1.75), 	CD177 (1.73), 	NEK2 (1.71), 	PRMT3 (1.71), 	SSBP2 (1.68), 	ORC1 (1.68), 	CNTRL (1.66), 	GTF2H2 (1.57), 	CTIF (1.56), 	RAB24 (1.56), 	TRIM7 (1.51), 	STAP1 (1.51), 	NUP214 (1.5), 	COMMD9 (1.46), 	CLEC4C (1.45), 	DNAJC24 (1.4), 	DEFB135 (1.35), 	TOR2A (1.34), 	LCE2B (1.33), 	RPS4Y2 (1.32), 	LYPD8 (1.24), 	ZNF549 (1.23), 	HSD17B14 (1.1), 	NOL8 (1.07), 	SLC16A10 (1.04), 	SLC44A4 (1.03), 	OR5M9 (0.96), 	PRADC1 (0.94), 	DNAJB8 (0.86), 	GTF2H2C (0.84), 	ZNF154 (0.82), 	GPD1 (0.82), 	MXD3 (0.78), 	LBH (0.78), 	NXPH3 (0.76), 	HOXD8 (0.63), 	MCTP1 (0.6), 	WFDC1 (0.58), 	PPP1R36 (0.56), 	IFI27L1 (0.56), 	TAPT1 (0.53), 	TIMM23B (0.5), 	TIGD4 (0.49), 	SAPCD2 (0.43), 	SNORA72 (0.43), 	BEND4 (0.33), 	ATF7IP2 (0.27), 	GTF2H2B (0.13) 	&	HCG8,	FAM105B,	SQRDL,	PPAPDC1B,	LOC730441,	FAM129C,	ENSG00000225824,	KIAA0355,	LOC101060445, LOC101930418, 	LOC101927707 \\\hline
			

			\texttt{Gene\_Set\_1(I0)}	&	VCP (33.29), 	LMNB1 (22.08), 	SNAP25 (13.88), 	CBS (9.96) 	&	DNM1 (6.94), 	MSR1 (6.71), 	PDYN (6.29), 	TGIF1 (6.11), 	PCMT1 (5.61), 	TCN2 (4.76), 	C1R (4.57), 	CTNNA1 (4.47), 	CLIC1 (4.38), 	S100A12 (4.09), 	CYP1A1 (3.96), ~~~~	HNRNPDL (3.87), 	F5 (3.77), 	MEGF10 (3.71), 	EFS (3.59), 	ABCG2 (3.49), 	~~CASK (3.35) 	&	FMNL1 (3.29), 	RAB27A (3.16), 	BCL11A (2.71), 	LYRM7 (2.71), 	AFF3 (2.51), 	COX14 (2.44), 	WDR45B (2.33), 	PMM2 (2.33), 	CTBP2 (2.21), 	RAB8B (2.02), 	ALG13 (1.97), 	EDC3 (1.97), 	LGI4 (1.97), 	RDX (1.97) 	&	TSPAN13 (1.73), 	TPD52 (1.49), 	RASSF4 (1.47), 	CDC42EP3 (1.44), 	BLNK (1.44), 	PDE7A (1.42), 	TMEM163 (1.33), 	HVCN1 (1.3), 	CCDC50 (1.29), 	KNDC1 (1.28), 	SORCS2 (1.28), 	PTPRK (1.24), 	RPS14 (1.22), 	FXR2 (1.22), 	POR (1.17), 	PRUNE2 (1.16), 	INSC (1.14), 	NHLRC2 (1.13), 	OSBPL1A (1.09), 	CERCAM (1.09), 	SWSAP1 (1.09), 	LATS2 (1.09), 	HERC4 (1.07), 	OLAH (1.07), 	NR4A3 (1.07), 	TRPM2 (1.03), 	TBX2 (1.01), 	DUSP3 (1), 	HDAC10 (0.98), 	FUBP3 (0.89), 	FRYL (0.89), 	MMP16 (0.86), 	DDX3Y (0.86), 	PPFIBP1 (0.75), 	DTWD1 (0.74), 	DZIP1 (0.71), 	PLA2G2D (0.71), 	WBP11 (0.71), 	TMTC1 (0.7), 	TSC22D1 (0.66), 	VAMP3 (0.54), 	LIMK2 (0.53), 	IQGAP2 (0.53), 	NAB2 (0.51), 	TNFSF13 (0.47), 	KCTD8 (0.47), 	NRP2 (0.47), 	MAGEA4 (0.47), 	CXADR (0.44), 	FGR (0.38), 	SNX22 (0.38), 	TNFSF12-TNFSF13 (0.38), 	HBS1L (0.36), 	WLS (0.36), 	PLBD1 (0.36), 	SH3RF2 (0.36), 	IGSF6 (0.36), 	LOXHD1 (0.13), 	ANKHD1 (0.09), 	FOXR1 (0.09), 	TENM1 (0.09), 	ARHGEF26 (0.09), 	IL27 (0.09), 	NEK6 (0.09), 	UBE2J1 (0.09), 	PSMB11 (0.09), 	E2F5 (0.09) 	&	BPI (4.75), 	IRAK3 (3.03), 	S100A2 (2.53), 	POMP (2.35), 	P2RY10 (2), 	LZTFL1 (1.98), 	AGTRAP (1.79), 	UTY (1.77), 	FUCA2 (1.75), 	CD177 (1.73), 	NEK2 (1.71), 	PRMT3 (1.71), 	SSBP2 (1.68), 	ORC1 (1.68), 	GTF2H2 (1.57), 	CTIF (1.56), 	NBN (1.53), 	TRIM7 (1.51), 	CAPRIN2 (1.47), 	COMMD9 (1.46), 	CLEC4C (1.45), 	DEFB135 (1.35), 	TOR2A (1.34), 	LCE2B (1.33), 	RPS4Y2 (1.32), 	LYSMD1 (1.28), 	LYPD8 (1.24), 	ZNF549 (1.23), 	LRIG2 (1.23), 	HSD17B14 (1.1), 	SLC16A10 (1.04), 	FCRL1 (1.03), 	SLC44A4 (1.03), 	OR5M9 (0.96), 	PRADC1 (0.94), 	GTF2H2C (0.84), 	ZNF154 (0.82), 	GPD1 (0.82), 	MXD3 (0.78), 	LBH (0.78), 	NXPH3 (0.76), 	HOXD8 (0.63), 	MCTP1 (0.6), 	WFDC1 (0.58), 	TMEM64 (0.58), 	PPP1R36 (0.56), 	TIMM23B (0.5), 	TIGD4 (0.49), 	KLHL29 (0.49), 	SAPCD2 (0.43), 	SNORA72 (0.43), 	ATF7IP2 (0.27), 	GTF2H2B (0.13) 	&	HCG8,	FAM105B,	SQRDL,	LOC730441,	FAM129C,	ENSG00000225824,	KIAA0355	\\	
			\label{TAB:data_gene}
		\end{longtable}
	\end{landscape}
	
\end{small}

\end{document}